\documentclass[oneside,reqno]{amsart}
\usepackage[margin=1in]{geometry}
\usepackage{amsmath, amsfonts, amssymb}
\usepackage{amsthm}
\usepackage[dvipsnames]{xcolor}
\usepackage{algpseudocode}
\usepackage[shortlabels]{enumitem}
\usepackage[algo2e,linesnumbered,vlined,ruled]{algorithm2e}
\usepackage{siunitx}
\usepackage{graphicx}
\usepackage{cancel}
\usepackage{subfig}
\usepackage{upgreek}
\usepackage{tikz}
\usetikzlibrary{patterns}
\usepackage{tkz-euclide}
\usepackage{hyperref,cleveref} 
\numberwithin{equation}{section} 
\hypersetup{
    colorlinks=true,
    linkcolor=blue,
    citecolor=magenta,
    filecolor=magenta,      
    urlcolor=cyan,
    breaklinks=true,   
}

\allowdisplaybreaks

\addtocontents{toc}{\protect\setcounter{tocdepth}{1}}

\theoremstyle{definition}
\newtheorem{theorem}{Theorem}[section]
\newtheorem{lemma}[theorem]{Lemma}
\newtheorem{proposition}[theorem]{Proposition}
\newtheorem{definition}[theorem]{Definition}

\newtheorem{remark}{Remark}

\newcommand{\ml}{\left(}
\newcommand{\mr}{\right)}
\newcommand{\dv}{\partial_v}
\newcommand{\du}{\partial_u}
\newcommand{\p}{\partial}
\newcommand{\mlm}{\left|}
\newcommand{\mrm}{\right|}
\newcommand{\oh}{\frac{1}{2}}
\newcommand{\e}{\mathfrak{e}}
\newcommand{\ul}[1]{\underline{#1}}
\newcommand\ii{\mathrm{i}}
\newcommand{\tp}{\tilde{\phi}}
\newcommand{\tgam}{\tilde{\gamma}}
\newcommand{\te}{\tilde{\varepsilon}}

\newcommand{\td}{\tilde{\delta}}
\newcommand{\ttheta}{\tilde{\theta}}
\newcommand{\tz}{\tilde{\zeta}}
\newcommand{\ttz}{\vardbtilde{\zeta}}
\newcommand{\tr}{\tilde{\rho}}
\newcommand{\trs}{\tilde{\rho}^*}
\newcommand{\tu}{\tilde{u}}
\newcommand{\tv}{\tilde{v}}
\newcommand{\tf}{\tilde{f}}
\newcommand{\tfs}{\tilde{f}^*}
\newcommand{\tg}{\tilde{g}}
\newcommand{\tgs}{\tilde{g}^*}

\newcommand{\tmfa}{\tilde{\mathfrak{a}}}
\newcommand{\tmfb}{\tilde{\mathfrak{b}}}
\newcommand{\vardbtilde}[1]{\tilde{\raisebox{0pt}[0.85\height]{$\tilde{#1}$}}}
\newcommand{\ttmfb}{\vardbtilde{\mathfrak{b}}}
\newcommand{\tmcd}{\mathcal{\tilde{D}}}
\newcommand{\mfa}{\mathfrak{a}}
\newcommand{\mfb}{\mathfrak{b}}

\newcommand{\D}{\mathrm{d}}

\DeclareMathOperator\re{Re}
\DeclareMathOperator\im{Im}

\title{A Proof of Weak Cosmic Censorship Conjecture \\ for the Spherically Symmetric\\ Einstein-Maxwell-Charged Scalar Field System}

\author[Xinliang An]{Xinliang An}
\address{Department of Mathematics, National University of Singapore\\ Singapore}
\email{matax@nus.edu.sg}

\author[Hong Kiat Tan]{Hong Kiat Tan}
\address{Department of Mathematics, University of California, Los Angeles\\ USA}
\email{maxtanhk@math.ucla.edu}

\begin{document}

\maketitle

\begin{abstract}
Under spherical symmetry, we show that the weak cosmic censorship holds for the gravitational collapse of the Einstein-Maxwell-charged scalar field system. Namely, for this system, with generic initial data, the formed spacetime singularities are concealed inside black-hole regions. This generalizes Christodoulou’s celebrated results to the charged case. Due to the presence of charge $Q$ and the complexification of the scalar field $\phi$, multiple delicate features and miraculous monotonic properties of the Einstein-(real) scalar field system are not present. We develop a systematical approach to incorporate $Q$ and the complex-valued $\phi$ into the integrated arguments. For instance, we discover a new path, employing the reduced mass ratio, to establish the sharp trapped surface formation criterion for the charged case. Due to the complex structure and the absence of translational symmetry of $\phi$, we also carry out detailed modified scale-critical BV area estimates with renormalized quantities to deal with $Q$ and $\phi$. We present a new $C^1$ extension criterion by utilizing the Doppler exponent to elucidate the blueshift effect, analogous to the role of integrating vorticity in the Beale-Kato-Madja breakdown criterion for incompressible fluids. Furthermore, by utilizing only double-null foliations, we establish the desired first and second instability theorems for the charged scenarios and identify generic initial conditions for the non-appearance of naked singularities. Our instability argument requires intricate generalizations of the treatment for the uncharged case via analyzing the precise contribution of the charged terms and its connection to the reduced mass ratio.
\end{abstract}

{\hypersetup{linkcolor=black}
\tableofcontents
}

\section{Introduction}

In general relativity, the validity of the weak cosmic censorship is a central unresolved issue. It is conjectured that \textit{for generic initial data, the singularities formed in the gravitational collapse would be confined within black-hole regions and cannot be seen by distant observers}. In 1969 Penrose \cite{Penrose69} wrote: \textit{We are thus presented with what is perhaps the most fundamental unanswered question of general-relativistic collapse theory, namely: does there exist a ``cosmic censor" who forbids the appearance of naked singularities, clothing each one in an absolute event horizon?} More precisely, as written by Christodoulou \cite{Chr.99}, the weak cosmic censorship conjecture for Einstein's equations states:

\textit{Generic asymptotically flat initial data have a maximal future development possessing a complete future null infinity.} 

\noindent Christodoulou \cite{Chr.99} further explained:  \textit{In heuristic terms, this means that, if we disregard exceptional initial conditions, no singularities are observed from infinity, even though observations from infinity are allowed to continue indefinitely.} So far, only for the spherically symmetric Einstein-(real) scalar field system, this censorship was rigorously established and it was proved by Christodoulou in 1990s through a series of breakthrough results \cite{christ1}-\cite{christ4}. To generalize this celebrated conclusion to the more physical scenario, one would consider adding the joint evolution of the electric charge $Q$, often viewed as the poor man's angular momentum, to the Einstein field equations. This motivates the study of the spherically symmetric Einstein-Maxwell-charged scalar field system. However, proving the weak cosmic censorship for this system remains open for more than 20 years. 

Returning to Christodoulou's proof in \cite{christ1}-\cite{christ4}, integrated arguments with five indispensable steps are established. In each of these steps, he explored the delicate structures of the Einstein-(real) scalar field system and established sharp results. Based on these, he designed a spectacular blueprint for the proof. All these steps and the corresponding sharp results are later ingeniously incorporated into one piece, hence proving the final conclusion. In his proof, various monotonic formulas tied to the special structures of the Einstein-(real) scalar field system are critically employed. Taking the evolution of $Q$ into account, the subtle structure changes. The given monotonic formulas for the uncharged case are no longer present, and numerous crucial arguments in \cite{christ1}-\cite{christ4} have to be circumvented or modified. 

In this article, we prove that, \textit{within spherical symmetry, the weak cosmic censorship holds for the gravitational collapse of the Einstein-Maxwell-charged scalar field system.} Here we consider the $3+1$ dimensional spacetime  $(\mathcal{M},g,\phi,F)$ with $\mathcal{M}$ being the $3+1$ dimensional manifold, $g$ being the Lorenztian metric, $\phi$ being the complex-valued function on $\mathcal{M}$ and $F$ being a differential $2$-form on $\mathcal{M}$. The system of equations governing these quantities takes the form: 
\begin{equation}\label{Intro1}
\begin{split}
R_{\mu \nu} - \frac{1}{2} g_{\mu \nu} R &= 8 \pi T_{\mu \nu}^{SF} + 8\pi T_{\mu \nu}^{EM},\\
T_{\mu \nu}^{SF} &= \frac{1}{2}(D_{\mu} \phi)(D_{\nu} \phi)^{\dagger} + \frac{1}{2}(D_{\nu} \phi)(D_{\mu} \phi)^{\dagger} - \frac{1}{2}g_{\mu \nu}g^{\alpha \beta} D_{\alpha}\phi (D_{\beta}\phi)^{\dagger}, \\
T_{\mu \nu}^{EM} &= \frac{1}{4\pi}(g^{\alpha \beta}F_{\alpha \mu}F_{\beta \nu} - \frac{1}{4}g_{\mu \nu}F^{\alpha\beta}F_{\alpha\beta}) \\
\end{split}
\end{equation}
with $^\dagger$ standing for the complex conjugate and $D_\mu:= \partial_{\mu} + \e i A_{\mu}$ being the gauge covariant derivative, where $\e$ is the coupling constant,  and $A_\mu$ is the electromagnetic potential. Note that, via the second Bianchi identity, the divergence-free property of the Einstein tensor also implies an evolution equation for the electromagnetic field, i.e., 
\begin{equation}\label{Intro2}
\nabla^{\nu}F_{\mu\nu} = 2 \pi \e i \big(\phi(D_{\mu}\phi)^{\dagger} - \phi^{\dagger}(D_{\mu}\phi)\big).
\end{equation}

\subsection{Main Theorem}
In this study, we conduct all the proofs with double null coordinates. Under spherical symmetry, the spacetime metric $g$ now takes the form:
\begin{equation}
g =  - \Omega^2(u,v) \; \mathrm{d}u\mathrm{d}v + r^{2}(u,v) \; (\mathrm{d}\theta^2 + \sin^2 \theta \mathrm{d} \phi^2)
\end{equation}
with $\Omega(u,v)$ called the lapse function and $r(u,v)$ called the radial function. Here we employ $u$ and $v$ as coordinates and they are optical functions toward the future satisfying $g^{\mu\nu}\partial_{\mu}u\partial_{\nu}u=0$ and $g^{\mu\nu}\partial_{\mu}v\partial_{\nu}v=0$. Furthermore, constant $u$ and constant $v$ hypersurfaces correspond to the outgoing and incoming null cones, respectively. 

Our main conclusion is as follows:

\begin{theorem}\label{main theorem rough version}
(rough version) In the setting of gravitational collapse for the spherically symmetric Einstein-Maxwell-charged scalar field system, for generic\footnote{Here the term ``generic'' means that the subset of initial conditions leading to the formation of naked singularities has positive co-dimensions. } asymptotically flat initial data prescribed along $C^+_0$, the singularities formed in the evolution would be covered by a black-hole region. 
\end{theorem}

\begin{figure}[htbp]
\begin{minipage}[!t]{0.7\textwidth}
\centering
\begin{tikzpicture}
	\begin{scope}[thick]
	\draw[-] (0,0) node[anchor=north]{$\mathcal{O}$} --  (0,8) node[anchor = east]{$\mathcal{O}'$};
	\draw[-] (0,0) -- (8,8);
        \node[rotate=45] at (5.2,4.8) {$C_0^+$};
        \node[] at (-0.3,5.0) {$\Gamma$};
        \fill[gray!50] (0,8) to (1.5,9.5) to[out = 45,in=135] (6,10) to[out = -135, in = 45] (0,8);
        \draw[] (0,8) -- (4,4);
        \draw[densely dashed] (0,4) -- (6,10);
        \node[rotate=45] at (4.2,7.8) {$\mathcal{EH}$};
        \draw[] (6,10) -- (8,8);
        \node[rotate = -45] at (7.2,9.2){$\mathcal{I}^+$};
        \draw[thick, densely dotted] (0,8) to (1.5,9.5);
        \draw[thick, loosely dotted] (1.5,9.5) to [out = 45,in=135] (6,10);
        \draw[fill=white] (1.5,9.5) circle (2pt);
        \node[rotate=12] at (3.6,10.9) {$\mathcal{B} \setminus \mathcal{B}_0$};
        \node[rotate=45] at (0.55,8.95) {$\mathcal{B}_0$};
        \draw[thick, dashdotted] (0,8) to[out = 45, in = - 135] (6,10);
        \node[] at (3.2,8.8) {$\mathcal{A}$};
        \node[] at (3.7,9.8) {$\mathcal{BH}$};
        \node[anchor = west] at (9.8,10.2) {\ul{Legend:}};
        \node[anchor = west] at (11,9.5) {$\mathcal{B}_0$}; 
        \draw[thick, densely dotted] (10,9.3) to (11,9.7);
        \node[anchor = west] at (11,9) {$\mathcal{B} \setminus \mathcal{B}_0$}; 
        \draw[thick, loosely dotted] (10,8.8) to (11,9.2);
        \node[anchor = west] at (11,8.5) {$\mathcal{A}$}; 
        \draw[thick, dashdotted] (10,8.3) to (11,8.7);
        \draw[fill=white] (0,8) circle (2pt);
        \draw[fill=white] (6,10) circle (2pt);
	\end{scope}
\end{tikzpicture}
\end{minipage}
\caption{In above picture, $\Gamma$ corresponds to the set of fixed points under the group action $SO(3)$; $\mathcal{O}'$ stands for the first singularity formed along $\Gamma$ in the evolution; $\mathcal{O}'$ and the singular boundary $\mathcal{B}$ are censored by the black-hole region ($\mathcal{BH}$); $\mathcal{A}$ is the apparent horizon; $\mathcal{EH}$ is the event horizon; $\mathcal{I}^+$ is the future null infinity.}
\end{figure}
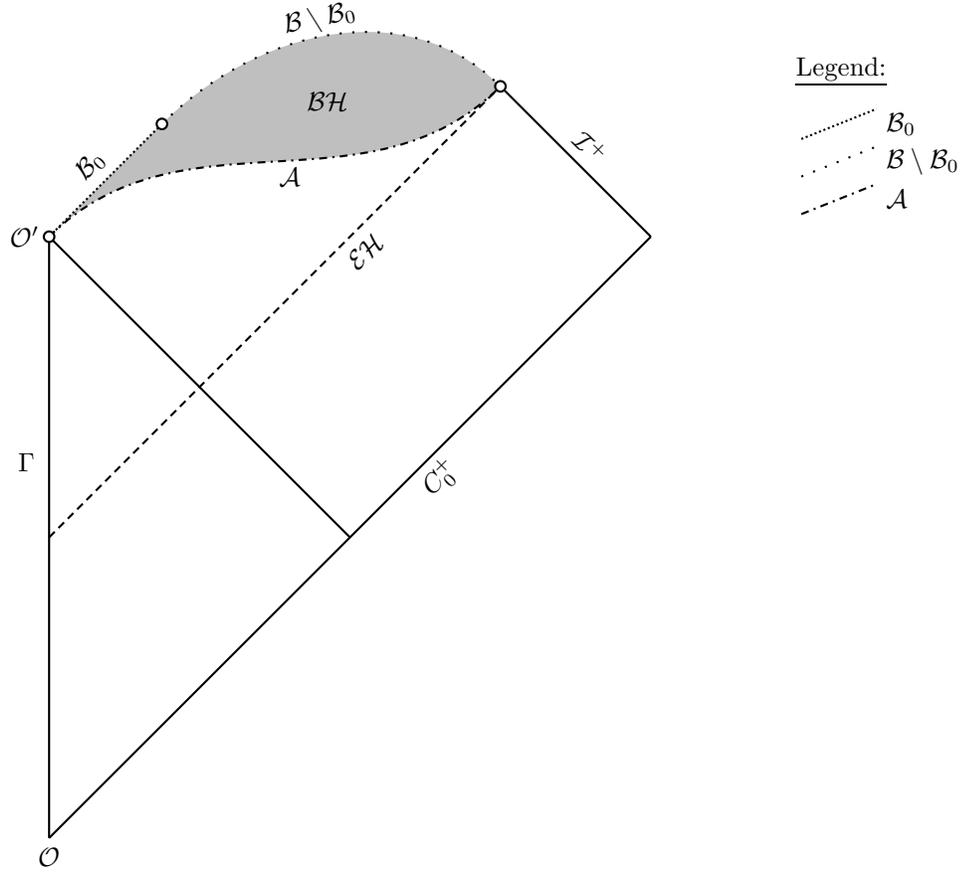

\begin{remark}
The precise version of Theorem \ref{main theorem rough version} combines our statements in Theorem \ref{Trapped}, Theorem \ref{AreaLemma3}, Theorem \ref{SET}, Theorem \ref{FITTheorem}, Theorem \ref{SITTheorem} and Theorem \ref{ExceptionalSet}. 
The conclusion in Theorem \ref{main theorem rough version} also holds for generic asymptotically flat Cauchy initial data. For the convenience of observing the instability mechanism, we state our main theorem with characteristic initial data.
\end{remark}

Our proofs in this paper are divided into six main steps: 
\begin{enumerate}
\item \textit{Find Novel Estimates for $Q$ and $\phi$}.
In Section \ref{Q and phi}, we derive new estimates for $\phi$ and $Q$ and develop a novel strategy crucially employed throughout the paper to control related charged terms. 

\item \textit{Obtain the Sharp Criterion for Formation of Trapped Surfaces via a New Approach}.
In Section \ref{TSF}, via a novel approach, we prove the desired trapped-surface-formation criterion for the charged scenario and state it as Theorem \ref{Trapped}. Specifically, restricted to the $Q=0$ case, in Appendix \ref{Appendix A}, we give a 2-page new proof of Christodoulou's sharp criterion for the Einstein-(real) scalar field system. 

\item \textit{Establish Scale-Critical $BV$ Area Estimates with Renormalizations}. In this step, the objective is to establish technical BV area estimates that are crucial for the entire regularity argument.  Section \ref{BV Area Estimates} is dedicated to proving Theorem \ref{AreaLemma3} for the charged case. Incorporating the evolution of $Q(u,v)$ into the system breaks the transitional symmetry of $\phi(u,v)$ and complicates the system. These require us to introduce new renormalizations and to conduct heavy intricate calculations. In the process, the complex structure of the system is carefully explored and utilized. 

\item \textit{Present a New $C^1$ Extension Principle related to the Blueshift Effect}. To connect with the instability arguments, we propose and certify a new $C^1$ extension principle (Theorem \ref{SET}) based on the Doppler exponent $\gamma$, which resembles the Beale-Kato-Majda type breakdown criterion, requiring the boundness of the integration of a certain physical quantity. The proof of an equivalent $C^1$ extension principle for the uncharged scalar field is also listed as Appendix \ref{Appendix B}.

\item \textit{Prove the First Instability Theorem with Double Null Foliations}. In Section \ref{First Instability Theorem}, we prove and summarize the first-instability conclusions as Theorem \ref{FITTheorem} with double null foliations. Christodoulou employed Bondi formalism in his original arguments. Our proofs for both the first and second instability theorems inject new ingredients about the charge and associated estimates that are central and delicate to the overall argument. 

\item \textit{Vindicate the Second Instability Theorem with Double Null Foliations}. In Section \ref{Second Instability Theorem}, we further confirm the second instability mechanism and prove Theorem \ref{SITTheorem}. All the aforementioned components are interconnected, constituting an integrated argument. These eventually lead to Theorem \ref{ExceptionalSet} in Section \ref{ExceptionalSetSection} describing conclusions of genericness.  

\end{enumerate}

To prove the weak cosmic censorship conjecture for the spherically symmetric Einstein-Maxwell-charged scalar field system, we incorporate these steps with the following coherent progression of logical arguments:

\begin{itemize}

\item[-] Near the center $\Gamma$ portrayed below, we start to consider the characteristic initial value problem with initial data prescribed along $C_0^+$. Applying Theorem \ref{AreaLemma3}, we extend the BV existence region toward the future, with scale-critical BV area estimates controlled.

\item[-] Assume that $O'$ is the first singular point along $\Gamma$. For any regular point $R$ along $\Gamma$ before $O'$, we consider a triangular shadowed region with the future tip at $R$. Within this region, a blueshift factor $e^\gamma$ (referred to as the Doppler factor, where $\gamma$ is termed the Doppler exponent) is bounded. Physically, the Doppler factor $e^\gamma$ measures the blue shift of a photon emitted backward from any $(u,v)$ in the shadowed region and measured by an observer located at $(0,v)$ along $C_0^+$. As $R$ approaches $O'$ along $\Gamma$, by employing Theorem \ref{SET}, along the backward cone $v=v_0$ emitting from $\mathcal{O}'$, we have $\gamma\rightarrow+\infty$ as it approaches $O'$. 
\item[-] The fact that $\gamma\rightarrow +\infty$ along $v=v_0$ triggers the blueshift-driven instability mechanism. We then employ the first and the second instability theorems (together with the trapped surface formation criterion) in the region between $v=v_0$ and $v=v_1$ to arrive at the main conclusion.

\end{itemize}

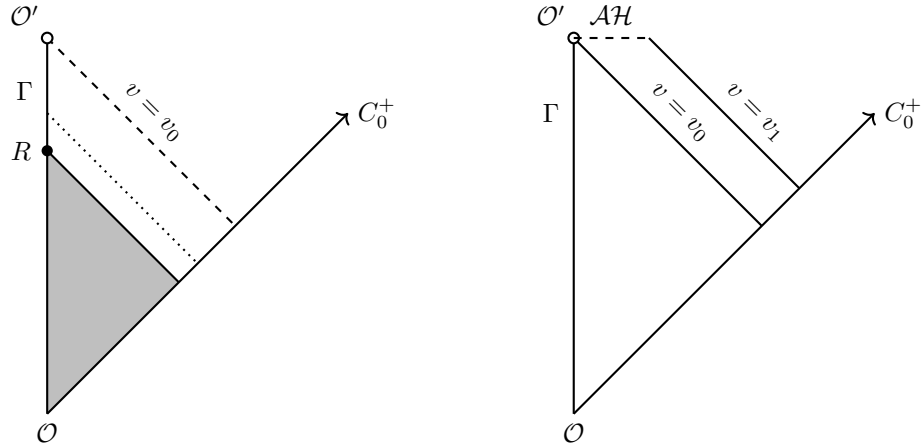
\begin{figure}[ht]
\begin{minipage}[!t]{0.7\textwidth}
	\centering
\begin{tikzpicture}
\begin{scope}[thick]
        \fill[gray!50] (0,0) to (1.75,1.75) to (0,3.5) to (0,0);
	\draw (0,0) node[anchor=north]{$\mathcal{O}$} --  (0,5);
	\draw[->] (0,0) -- (4,4) node[anchor = west]{$C_0^+$};
        \node at (-0.35,3.5) {$R$};
        \node at (-0.3,4.3) {$\Gamma$};
        \node at (-0.3,5.3) {$\mathcal{O}'$};
	\draw(0,3.5) -- (1.75,1.75);
        \draw[thick,dotted](0,4) -- (2,2);
        \draw[thick,dashed](0,5) -- (2.5,2.5);
        \node at (0,3.5)[circle,fill,inner sep=1.5pt]{};
        \node at (0,4)[circle,inner sep=1.5pt]{};
        \node[rotate=-45] at (1.4,4) {$v = v_0$};
        \draw[fill=white] (0,5) circle (2pt);
	\draw (7,0) node[anchor=north]{$\mathcal{O}$} --  (7,5);
	\draw[->] (7,0) -- (11,4) node[anchor = west]{$C_0^+$};
        \node at (6.7,4) {$\Gamma$};
        \node at (6.7,5.3) {$\mathcal{O}'$};
        \draw[thick](7,5) -- (9.5,2.5);
        \node at (7,4)[circle,inner sep=1.5pt]{};
        \node[rotate=-45] at (8.4,4) {$v = v_0$};
        \draw[fill=white] (7,5) circle (2pt);
        \draw[thick,dashed](7,5) -- (8,5);
        \draw[thick](8,5) -- (10,3);
        \node[rotate=-45] at (9.4,4) {$v = v_1$};
        \node[] at (7.5,5.3) {$\mathcal{AH}$};
\end{scope}
\end{tikzpicture}
\end{minipage}
\end{figure}

\subsection{Equations} Before we explain our proof strategy, we review several physical quantities and state their evolution equations below. The \textit{Hawking mass} $m(u,v)$ for sphere $S(u,v)$ on $\mathcal{Q}$ is defined via
\begin{equation}\label{Setup5}
1 - \frac{2m}{r} = g(\nabla r, \nabla r) = - 4\Omega^{-2}\du r \dv r.
\end{equation}
The \textit{mass ratio} $\mu(u,v)$ for sphere $S(u,v)$ is further set to be
\begin{equation}\label{Setup6}
\mu(u,v) := \frac{2m(u,v)}{r(u,v)}.
\end{equation} Given the electromagnetic $2$-form $F(u,v)$ in \eqref{Intro1}, we obtain the charge $Q(u,v)$ contained in a sphere $S(u,v)$ implicitly by solving
\begin{equation}\label{Setup7}
F(u,v) =  F_{uv}\; \mathrm{d}u \wedge \mathrm{d}v =\frac{Q(u,v)}{2r^2(u,v)}\Omega^2(u,v) \; \mathrm{d}u \wedge \mathrm{d}v.
\end{equation}
Equivalently, an explicit expression for $Q(u,v)$ is given by
\begin{equation}\label{Setup8}
Q(u,v) := 2r^2 \Omega^{-2}F_{uv}(u,v).
\end{equation} 

At the same time, notice that each electromagnetic field $F_{uv}$ is associated with a $1$-form $A_u$, called the electromagnetic potential. Furthermore, $F_{uv}$ can be rewritten as
\begin{equation}\label{Setup9}
F_{uv} = \du A_v - \dv A_u. 
\end{equation}
As explained in \cite{kommemi}, when $\e \neq 0$, the system \eqref{Intro1}-\eqref{Intro2} is invariant under the following local $U(1)$ gauge transformations:
\begin{equation}\label{Setup10}
\begin{aligned}
\phi \rightarrow e^{-\e \ii f}\phi, \quad  A_\mu \rightarrow A_\mu + \partial_\mu f
\end{aligned}
\end{equation}
with $f$ being a smooth real-valued function to be fixed. In this study, we impose the gauge condition 
$$A_v = 0$$
by choosing the appropriate $f$. With this gauge choice, we reduce \eqref{Setup8} and \eqref{Setup9} to
\begin{equation}\label{Setup11}
\begin{aligned}
F_{uv} = -\dv A_u \quad \mbox{ and } \quad Q(u,v) = -2r^2\Omega^{-2} \dv A_u.
\end{aligned}
\end{equation}
Substituting \eqref{Setup4} and \eqref{Setup11} back to the field equations \eqref{Intro1} and \eqref{Intro2}, we then derive the following evolution system:
\begin{equation}\label{Setup12}
r \dv \du r + \dv r \du r = -\frac{\Omega^2}{4}\ml 1- \frac{Q^2}{r^2}\mr,
\end{equation}
\begin{equation}\label{Setup13}
r^2 \du \dv \log(\Omega) = -2\pi r^2(D_u \phi (\dv \phi)^\dagger + \dv \phi(D_u \phi)^\dagger) - \frac{1}{2}\Omega^2 \frac{Q^2}{r^2} + \frac{1}{4}\Omega^2 + \du r \dv r,
\end{equation}
\begin{equation}\label{Setup14}
\du(\Omega^{-2}\du r) = - 4\pi r\Omega^{-2}|D_u\phi|^2,
\end{equation}
\begin{equation}\label{Setup15}
\dv(\Omega^{-2}\dv r) = - 4\pi r\Omega^{-2}|\dv \phi|^2,
\end{equation}
\begin{equation}\label{Setup16}
r\du \dv \phi + \du r \dv \phi + \dv r \du \phi + \e \ii \Psi(A) = 0,
\end{equation}
\begin{equation}\label{Setup17}
\Psi(A) = A_u \dv(r\phi) - \frac{\Omega^2}{4}\frac{Q}{r}\phi,
\end{equation}
\begin{equation}\label{Setup18}
\du Q = 4\pi \e r^2 \im{(\phi^\dagger D_u \phi)},
\end{equation}
\begin{equation}\label{Setup19}
\dv Q = -4\pi \e r^2 \im{(\phi^\dagger \dv \phi)}.
\end{equation}
Note that in the equations above we use the notation $D_u := \du + \ii \e A_u$ for the covariant derivative. Moreover,the coupling constant $\e$ here satisfies $0 \leq \e \in \mathbb{R}$. We also emphasize that, in the above system, $\phi$ is a complex-valued unknown function, while $r, A_u, \Omega^2, \im(\phi)$,  $\re(\phi)$ are real-valued unknown functions. 

In our arguments, derived from \eqref{Setup12}, \eqref{Setup14}, \eqref{Setup15}, \eqref{Setup17},  we also frequently employ the following equations:
\begin{equation}\label{Setup20}
\begin{aligned}
\dv \du r &= \frac{1}{r}\ml \frac{\mu - Q^2/r^2}{1-\mu}\mr \du r \dv r,
\end{aligned}
\end{equation}
\begin{equation}\label{Setup21}
\begin{aligned}
\du \ml \frac{1-\mu}{\dv r}\mr &= \frac{4 \pi r |D_u \phi|^2(1-\mu)}{(-\du r)(\dv r)},
\end{aligned}
\end{equation}
\begin{equation}\label{Setup22}
\begin{aligned}
\dv \ml \frac{1-\mu}{-\du r}\mr &= -\frac{4 \pi r |\dv \phi|^2(1-\mu)}{(-\du r)(\dv r)},
\end{aligned}
\end{equation}
\begin{equation}\label{Setup71}
\begin{aligned}
\Psi(A) &= A_u \dv(r\phi) + \frac{Q \phi (\du r) (\dv r)}{r(1-\mu)}.
\end{aligned}
\end{equation}
For the wave equation in \eqref{Setup16}, by \eqref{Setup17}, we also obtain and utilize the following equivalent forms:
\begin{equation}\label{Setup72}
\begin{aligned}
\du \ml r \dv \phi \mr + (\dv r)(\du \phi) &= - \ii \e A_u \dv (r\phi) - \ii \e \frac{Q \phi (\du r) (\dv r)}{r(1-\mu)}, \\
\dv \ml r \du \phi \mr + (\du r)(\dv \phi) &= - \ii \e A_u \dv (r\phi) - \ii \e \frac{Q \phi (\du r) (\dv r)}{r(1-\mu)},
\end{aligned}
\end{equation}
and 
\begin{equation}\label{Setup73}
\begin{aligned}
D_u(r\dv\phi) + \dv r(D_u \phi) &= -\ii\e \frac{Q\phi(\du r)(\dv r)}{r(1-\mu)}, \\
\dv(r D_u \phi) + \du r(\dv \phi) &= \ii\e \frac{Q\phi(\du r)(\dv r)}{r(1-\mu)}.
\end{aligned}
\end{equation}
It is also instructive to compute the evolution equations for $m$ and for $\mu$. These equations are
\begin{equation}\label{Setup24}
\begin{aligned}
 \du m 
&= -\frac{2\pi r^2 (1-\mu) |D_u \phi|^2}{(-\du r)} -  \frac{Q^2 (-\partial_u r)}{2 r^2},
\end{aligned}
\end{equation}
\begin{equation}\label{Setup25}
\begin{aligned}
\dv m &= \frac{2\pi r^2 (1-\mu)|\dv \phi|^2}{\dv r} +  \frac{Q^2 \dv r}{2 r^2},
\end{aligned}
\end{equation}
\begin{equation}\label{Setup26}
\begin{aligned}
\du \mu 
&= - \frac{4 \pi r (1-\mu) |D_u \phi|^2}{(-\du r)} + \frac{(-\du r)}{r}\ml \mu - \frac{Q^2}{r^2}\mr,
\end{aligned}
\end{equation}
\begin{equation}\label{Setup27}
\begin{aligned}
\dv \mu 
&= \frac{4 \pi r (1-\mu) |\dv \phi|^2}{\dv r} - \frac{\dv r}{r} \ml \mu - \frac{Q^2}{r^2}\mr.
\end{aligned}
\end{equation}
In addition, the terms $\dv(Q^2/r^2)$, $\du(Q^2/r^2)$, $\dv A_u$ appear fairly often in later sections. We can obtain their expressions via using \eqref{Setup18}, \eqref{Setup19},  and \eqref{Setup11} as follows:
\begin{equation}\label{Setup29}
\dv \ml \frac{Q^2}{r^2}\mr = -8 \pi \e Q \im(\phi^\dagger \dv \phi) - \frac{2 Q^2 \dv r}{r^3},
\end{equation}
\begin{equation}\label{Setup29.1}
\du \ml \frac{Q^2}{r^2}\mr = 8 \pi \e Q \im(\phi^\dagger D_u \phi) - \frac{2 Q^2 \du r}{r^3},
\end{equation}
\begin{equation}\label{Setup30}
\dv A_u = \frac{2 Q \du r \dv r}{r^2(1-\mu)}.
\end{equation}

\noindent A physical quantity of great interest in this paper is the \textit{Doppler exponent}. For any $(u,v), (\overline{u},v) \in \mathcal{Q}$, we define $\gamma(u,v;\overline{u})$ to be
\begin{equation}\label{Setup90}
\gamma(u,v;\overline{u}) := \int_{\overline{u}}^u \frac{-\du r}{r}\frac{\mu - Q^2/r^2}{1-\mu}(u',v) \D u'.
\end{equation}
In particular, if we consider initial data prescribed along the $u = 0$ hypersurface, we denote $\gamma(u,v)$ to be the associated \textit{Doppler exponent} at $(u,v)$, given by
\begin{equation}\label{gamma}
\gamma(u,v) := \int_{0}^u \frac{-\du r}{r}\frac{\mu - Q^2/r^2}{1-\mu}(u',v) \D u'.
\end{equation}
For the uncharged case, in \cite{christ4, Chr.99} Christodoulou employed this quantity to study the instability mechanism. In this paper, in addition to the instability theorems, we also use $\gamma(u,v)$ to establish the $C^1$ extension theorem. Note that the etymology of this term in this paper stems from the following \textit{Doppler identity} derived from \eqref{Setup20}:
\begin{equation}\label{Setup92}
(\dv r)(0,v) = (\dv r)(u,v) \exp \ml \int^u_0 \frac{-\du r}{r}\frac{\mu - Q^2/r^2}{1-\mu}(u',v) \D u' \mr = (\dv r)(u,v)e^{\gamma(u,v)}.
\end{equation}
Furthermore, we will refer to $e^{\gamma(u,v)}$ as the Doppler factor, and $\gamma(u,v)$ as the associated Doppler exponent. Physically, we can view the quantity $\dv r(u,v)$ as a proxy for the local scale factor, which measures the wavelength of a photon emitted (to the past) from $(u,v)$. The Doppler identity in \eqref{Setup92} implies that an observer at $(0,v)$ would observe the same photon with its wavelength scaled by a factor of $e^{\gamma(u,v)}$. On the other hand, if a photon is emitted from $(0,v)$ towards $(u,v)$, an observer at $(u,v)$ would observe the same photon with its wavelength scaled by a factor of $e^{-\gamma(u,v)}$. In both cases, if $\gamma(u,v) > 0$, the former corresponds to the \textit{redshift effect}, while the latter corresponds to the \textit{blueshift effect}. The \textit{Doppler factor} is thus a proxy for the extent in which the wavelength of a photon is being redshifted or blueshifted. A parallel argument using the notion of proper time and length to justify the redshift and blueshift effects is given in \cite{dafermos2}.

It is worth noting that along $\Gamma$ the Doppler exponent $\gamma(v,v)$ defined in \eqref{gamma} is the integral of a dimensionless quantity. As we will see in Theorem \ref{SET}, the breakdown of regularity is signaled by the divergence of $\gamma(v,v)$. This is consistent with Christodoulou's comment in \cite{Chr.99} on finding an analogous quantity to the integration of vorticity in the breakdown criterion for the incompressible 3D Euler equation as in \cite{BKM}.

We can further compute the derivatives of $\gamma$ with respect to $u$ and $v$ as follows. For the former, we have
\begin{equation}\label{Setup93}
\du \gamma(u,v) = \frac{-\du r}{r}\frac{\mu - Q^2/r^2}{1-\mu}(u,v)
\end{equation}
if $\frac{-\du r}{r}\frac{\mu - Q^2/r^2}{1-\mu}(u,v)$ is an integrable function in its corresponding domain. For the latter, we  first calculate
\begin{equation}\label{Setup94}
\begin{aligned}
&\dv \ml \frac{-\du r}{r} \frac{\mu - Q^2/r^2}{1-\mu}\mr \\
=\; & \dv \ml \frac{-\du r}{r}\mr \ml \frac{\mu - Q^2/r^2}{1-\mu}\mr + \ml \frac{-\du r}{r}\mr \dv \ml \frac{1-Q^2/r^2}{1-\mu} - 1\mr \\
=\; & \frac{(-\du \dv r) r + (\dv r)(\du r)}{r^2}\ml \frac{\mu - Q^2/r^2}{1-\mu}\mr  + \ml \frac{-\du r}{r}\mr \ml \frac{-\dv(Q^2/r^2)}{1-\mu} + \frac{1-Q^2/r^2}{(1-\mu)^2} \dv \mu \mr.
\end{aligned}
\end{equation}
With \eqref{Setup20}, we have
\begin{equation}\label{Setup95}
\frac{(-\du \dv r) r + (\dv r)(\du r)}{r^2}\ml \frac{\mu - Q^2/r^2}{1-\mu}\mr  = \frac{(\dv r)(\du r)}{r^2} \ml \ml \frac{\mu - Q^2/r^2}{1-\mu}\mr \ml 1 - \frac{\mu - Q^2/r^2}{1-\mu} \mr \mr.
\end{equation}
Employing \eqref{Setup27} and \eqref{Setup29}, we also rewrite
\begin{equation}\label{Setup96}
\begin{aligned}
& \ml \frac{-\du r}{r}\mr \ml \frac{-\dv(Q^2/r^2)}{1-\mu} + \frac{1-Q^2/r^2}{(1-\mu)^2} \dv \mu \mr \\
=\; & \ml \frac{-\du r}{r}\mr \ml \frac{8 \pi \e Q \im(\phi^\dagger \dv \phi)}{1-\mu} + \frac{2 Q^2 \dv r}{r^3(1-\mu)} + \frac{1-Q^2/r^2}{(1-\mu)^2} \ml \frac{4 \pi r (1-\mu)|\dv \phi|^2}{\dv r} + \frac{- \dv r}{r}\ml \mu - \frac{Q^2}{r^2} \mr \mr \mr \\
=\; &\ml \frac{\du r \dv r}{r^2} \mr \ml - \frac{2 Q^2}{r^2(1-\mu)} + \frac{(\mu - Q^2/r^2)(1-Q^2/r^2)}{(1-\mu)^2}\mr \\
 & \; + \ml \frac{- \du r}{r}\mr\ml  \frac{8 \pi \e Q \im(\phi^\dagger \dv \phi)}{1-\mu} \mr + \ml \frac{- \du r}{r}\mr \ml \frac{4 \pi r (1-Q^2/r^2)|\dv \phi|^2}{\dv r(1-\mu)} \mr.
\end{aligned}
\end{equation}
Combining \eqref{Setup95} and \eqref{Setup96} back to \eqref{Setup94}, we hence obtain
\begin{equation}\label{Setup97}
\begin{aligned}
\dv \ml \frac{-\du r}{r} \frac{\mu - Q^2/r^2}{1-\mu}\mr = T_1 + T_2 + T_3
\end{aligned}
\end{equation}
with
\begin{equation}\label{Setup98}
T_1 := \frac{-\du r}{r}\ml \frac{8 \pi \e Q \im(\phi^\dagger \dv \phi)}{1-\mu}\mr, \; T_2 := \frac{-\du r}{r}\ml \frac{4\pi r(1-Q^2/r^2)|\dv \phi|^2}{\dv r(1-\mu)}\mr, \;
T_3 := \frac{2 \du r \dv r}{r^2}\ml \frac{\mu - 2Q^2/r^2}{1-\mu}\mr. 
\end{equation}
Back to \eqref{gamma}, this implies 
\begin{equation}\label{Setup99}
\dv \gamma(u,v) = \int_0^u (T_1 + T_2 + T_3)(u',v) \D u'
\end{equation}
if we can justify the interchange between the derivative and the integral.

\subsection{Key Elements and New Ingredients} In this subsection, we outline some central ideas and remark on several new ingredients for each of the important steps.  

\subsubsection{Novel Estimates for $Q$ and $\phi$} In Section \ref{Q and phi}, we will show that $|Q|$ and $|Q\phi|$ enjoy enhanced estimates, which are heavily used throughout this paper. To explain this, we start by listing the estimate for $\phi$:
\begin{equation}\label{phi intro}
    \begin{split}
    |\phi|(u,v)\leq& \; |\phi|(\bar{u}, v)+\int_{\bar{u}}^u|D_u\phi|(u',v)du'\\
    \leq& \; |\phi|(\bar{u},v)+\frac{1}{(4\pi)^{{\frac12}}}\bigg(\int_{\bar{{u}}}^u\frac{4\pi r|D_u\phi|^2}{-\partial_u r}(u',v)du'\bigg)^{\frac12}\bigg(\int_{\bar{u}}^u\frac{-\partial_u r}{r}(u',v)du'\bigg)^{\frac{1}{2}}\\
    \leq& \; |\phi|(\bar{u},v)+\frac{1}{(4\pi)^{\frac12}}\bigg(\log \ml \frac{\frac{1-\mu}{\partial_v r}(u,v)}{\frac{1-\mu}{\partial_v r}(\bar{u},v)} \mr \bigg)^\frac12 \bigg(\log\ml \frac{r(\bar{u},v)}{r(u,v)} \mr \bigg)^\frac12.
    \end{split}
\end{equation}
Here, data along $u=\bar{u}$ are considered given. The first inequality in \eqref{phi intro} utilizes Gr\"onwall's inequality while  the second inequality employs $$\partial_u \log \ml \frac{1-\mu}{\partial_v r} \mr =\frac{4\pi r|D_u\phi|^2}{-\partial_u r}.$$
Note that the above identity also implies $\du (\frac{\dv r}{1-\mu})\leq 0$, and this could give an upper bound for $\frac{\dv r}{1-\mu}$. However, it fails to offer an upper bound for $\frac{1-\mu}{\dv r}$ and $\log\frac{\frac{1-\mu}{\partial_v r}(u,v)}{\frac{1-\mu}{\partial_v r}(\bar{u},v)}$. Hence, potentially utilizing \eqref{phi intro} to control $\phi$ and related terms could lead to non-uniform bounds. \\

In this paper, we overcome this problem and similar difficulties by spotting and employing a new strategy. Take the estimate for $|Q|$ as an example. We first consider the expression for $\int_u^v \partial_v Q(u, v')dv'$. Applying the above estimate \eqref{phi intro} for $\phi$ and the H\"older's inequality, as shown in \eqref{Estimates6}-\eqref{Estimates11}, for each give $\chi\in(0,1)$, we obtain
\begin{equation*}
\begin{aligned}
|Q|&(u,v)\leq \int^v_u 4\pi \e r^{2}|\phi||\dv \phi| (u,v') \mathrm{d}v' \\
\leq& \; 4\pi \e  \int^v_u \frac{r^2(u,v')}{(4\pi)^{\frac{1}{2}}}\ml  \log \ml \frac{\frac{1 - \mu}{ \dv r}(u,v')}{\frac{1 - \mu}{ \dv r}(\overline{u},v')}\mr  \mr^{\frac{1}{2}}\ml \log \ml \frac{r(\overline{u},v')}{r(u,v')}\mr \mr^{\frac{1}{2}}|\dv \phi| (u,v') \mathrm{d}v'+l.o.t.  \\
\leq & \; \frac{1}{(8\pi^2)^{\frac{1}{2}}}\ml \int^v_u \frac{2\pi r^2(1-\mu)|\dv \phi|^2}{\dv r} (u,v') \mathrm{d}v'\mr^{\frac{1}{2}} \\
&\times \ml \int^v_u \frac{ r^2 \dv r}{1 - \mu}(u,v') \cdot \ml \log \ml \frac{\frac{1 - \mu}{ \dv r}(u,v')}{\frac{1 - \mu}{ \dv r}(\overline{u},v')}\mr  \mr \cdot \ml \log \ml \frac{r(\overline{u},v')}{r(u,v')}\mr \mr \mathrm{d}v'\mr^{\frac{1}{2}}+l.o.t. \\
\leq & \; \frac{m^{\frac12}(u,v)}{(8\pi^2)^{\frac{1}{2}}} \ml \int^v_u r^2 (u,v') \cdot \ml \frac{\dv r}{1 - \mu}(\overline{u},v') \mr \cdot 
\ml \frac{\frac{\dv r}{1 - \mu}(u,v')}{\frac{\dv r}{1 - \mu}(\overline{u},v')}\mr \cdot \ml \log \ml \frac{\frac{1 - \mu}{ \dv r}(u,v')}{\frac{1 - \mu}{ \dv r}(\overline{u},v')}\mr  \mr \cdot \ml \log \ml \frac{r(\overline{u},v')}{r(u,v')}\mr \mr \mathrm{d}v'\mr^{\frac{1}{2}}+l.o.t. \\
\end{aligned}
\end{equation*}
\begin{equation*}
\begin{aligned}
\leq & \; C_1(\bar{u}, v; \chi) \mu^{\oh}(u,v)r^{\oh}(u,v) \cdot \ml \int^v_u r^{2-\chi} (u,v') \cdot r^{\chi}(\overline{u},v') \cdot \ml \frac{r^{\chi}(u,v')}{r^{\chi}(\overline{u},v')}\mr \cdot
\ml \log \ml \frac{r(\overline{u},v')}{r(u,v')}\mr \mr \mathrm{d}v'\mr^{\frac{1}{2}}+l.o.t. \\
\leq & \; C_1(\bar{u}, v; \chi) \mu^{\oh}(u,v) \cdot r^{\frac{3}{2} - \frac{\chi}{2}}(u,v) \cdot r^{\frac{\chi}{2}}(\overline{u},v) \cdot (v - u)^{\oh} \\
\end{aligned}
\end{equation*}
with constant $C_1(\bar{u}, v; \chi)$ provided in \eqref{Estimates11}. Here, to get rid of the dangerous $\bigg(\log\bigg(\frac{\frac{1-\mu}{\partial_v r}(u,v)}{\frac{1-\mu}{\partial_v r}(\bar{u},v)} \bigg) \bigg)^\frac12$ term in the second line, we crucially utilize the fact that
\textit{after applying the H\"older's inequality for $|\partial_v\phi|$, there is an additional $\frac{\dv r}{1-\mu}$ introduced in the expression (as in the third inequality).} By viewing $\frac{\partial_v r}{1-\mu}$ as $x$, we then employ the inequality that, for any $x \in (0,1)$ and $c > 0$, we have 
$$x^c \log \ml \frac{1}{x}\mr \leq \frac{1}{ce}.$$ 
Hence, a uniform upper bound is obtained, and the worrisome term in the logarithm is thus controlled. Note that this strategy is employed throughout this paper. With this strengthened bound for $|Q|$, we arrive at yet another frequently used estimate, as follows:
\begin{equation}
\bigg(\mu-\frac{Q^2}{r^2}\bigg)(u,v)\geq \mu(u,v) \ml 1-C_1^2(\bar{u}, v; \chi)r^{1-\chi}(u,v)r^{\chi}(\bar{u},v)(v-u) \mr >0.     
\end{equation}

Besides $|Q|$, we also derive an improved estimate for $|Q\phi|$. This is done by utilizing $\partial_v Q=-4\pi e r^2 \mbox{Im}(\phi^{\dagger}\partial_v\phi)$ together with \eqref{Estimates20}-\eqref{Estimates33} to obtain
\begin{equation*}
\begin{aligned}
|Q\phi|^{2}(u,v)  \leq&\; \int^v_u 2 |Q\phi| \dv(Q\phi) (u,v') \D v' \\
\leq& \; \int^v_u 2  (4\pi \e r^2 |\phi|^{2} |Q|+ |\phi||Q|^{2})|\dv \phi| (u,v') \D v' \\
\leq& \; 2  \ml \int^v_u \frac{\dv r}{2\pi r^2 (1-\mu)}(4\pi \e r^2 |\phi|^{3} |Q| + |\phi||Q|^{2})^2  \D v' \mr^{\frac{1}{2}} \ml \int^v_u  \frac{2 \pi r^2(1-\mu)|\dv \phi|^2}{\dv r} \D v'\mr^{\frac{1}{2}}\\
\leq& \; 2  \ml \int^v_u \frac{\dv r}{2\pi r^2 (1-\mu)}(4\pi \e r^2 |\phi|^{3} |Q| + |\phi||Q|^{2})^2  (u,v') \D v' \mr^{\frac{1}{2}} m(u,v)^{\frac{1}{2}}\\
\leq& \; \mu(u,v)^{\frac12} r^{2-\frac{\chi}{2}}(u,v).
\end{aligned}
\end{equation*}
For the final inequality, we use \eqref{phi intro} to control $|\phi|$ and apply our new strategy with the gained factor $\frac{\partial_v r}{1-\mu}$. Note that $|Q\phi|$ obeys a better bound than the product of individual bounds for $|Q|$ and $|\phi|$.

\subsubsection{Trapped Surface Formation Criterion} In this paper, we take a new approach to establish the desired trapped surface formation criterion for the charged case. With these new arguments, in the appendix, we also give a direct proof for the uncharged case to retrieve Christodoulou's sharp criterion in \cite{christ1}. Let us begin by first explaining some key calculations in \cite{christ1}. With the mass ratio $\eta=\frac{m_2-m_1}{r_2}$ introduced by Christodoulou, for the uncharged case, we can compute the following:
\begin{equation}\label{eta intro1}
 \begin{split}
     \frac{\D \eta}{\D x} = \frac{\frac{\mathrm{d}\eta}{\mathrm{d}u}}{\frac{\mathrm{d}x}{\mathrm{d}u}} =-\frac{\eta}{x}-\frac{16\pi\dv r_2 \Omega_2^{-2}}{x\du r_2}\bigg(r_2^2|\du \phi_2|^2-\frac{\Omega_1^{-2}\dv r_1}{\Omega_2^{-2}\dv r_2}r_1^2|\du \phi_1|^2 \bigg).
 \end{split}   
\end{equation}
Based on the particular structures of the Einstein-scalar field system, via using $\Theta:=r_2\du \phi_2-r_1\du \phi_1$ and
several ingenious lemmas, one could further prove
\begin{equation}\label{eta intro2}
 \begin{split}
     r_2^2|\du \phi_2|^2-\frac{\Omega_1^{-2}\dv r_1}{\Omega_2^{-2}\dv r_2}r_1^2|\du \phi_1|^2\leq&\; r_2^2|\du \phi_2|^2-e^{\eta}r_1^2|\du \phi_1|^2\\
     \leq&\; \Theta^2+2\Theta\du \phi_1 r_1+(1-e^{\eta})r_1^2|\du \phi_1|^2\\
     \leq&\; \bigg(1+\frac{1}{e^{\eta}-1}\bigg)\Theta^2\leq\bigg(1+\frac{1}{\eta}\bigg)\Theta^2\\
     \leq&\; \bigg(1+\frac{1}{e^{\eta}-1}\bigg)\cdot\frac{\du r_2}{8\pi\Omega_2^{-2}\dv r_2}(m_2-m_1)\bigg(\frac{1}{r_2}-\frac{1}{r_1}\bigg).
 \end{split}   
\end{equation}
Together with \eqref{eta intro1}, this gives
\begin{equation}\label{eta intro2.5}
 \begin{split}
     \frac{d\eta}{d x}\leq&-\frac{\eta}{x}-\frac{2}{x}\bigg(1+\frac{1}{\eta}\bigg)\bigg(\frac{1}{r_2}-\frac{1}{r_1}\bigg)(m_2-m_1)\\
     =&-\frac{\eta}{x}+\frac{\eta}{x}\bigg(1+\frac{1}{\eta}\bigg)\bigg(\frac{r_2}{r_1}-1\bigg).
 \end{split}   
\end{equation}
Solving this ODE inequality then leads to a desired criterion for the formation of trapped surfaces. The first author and Lim tried to generalize this approach in \cite{christ1} to the charged scenario. However, for the charged case, a few inequalities in \eqref{eta intro2} cannot be as sharp as for the uncharged case. Hence, the result obtained by An-Lim in \cite{An-Lim} is slightly weaker than the desired one, making it incompatible with the instability arguments we seek to employ. \\

We then state our new approach in this paper as follows. We begin by demonstrating it for the uncharged case.\footnote{The complete proof for the uncharged case is in the appendix.} 
With the mass ratio $\eta$, we compute and rewrite the expression of its derivative as
\begin{equation}\label{eta intro3}
\frac{\D \eta}{\D x} = \frac{\frac{\mathrm{d}\eta}{\mathrm{d}u}}{\frac{\mathrm{d}x}{\mathrm{d}u}} = \frac{\du \ml \frac{2(m_2 - m_1)}{r_2} \mr}{\frac{\du r_2}{r_2(u_0)}} = - \frac{\eta}{x} - \frac{2}{x(\du r_2)}\int^{v_2}_{v_1} -\du \dv m \; \D v'.
\end{equation}
Employing $\dv m = 2\pi r^2(1-\mu)(\dv \phi)^2/\dv r$, we further calculate
\begin{equation*}
\begin{aligned}
- \du\dv m  &= -\du \ml \frac{2\pi r^2(1-\mu)(\dv \phi)^2}{\dv r} \mr \\
&= -2 \pi \ml \du\ml \frac{(1-\mu)}{\dv r}\mr r^2(\dv \phi)^2 + \ml \frac{2(1-\mu)}{\dv r}\mr r (\dv \phi) \du (r\dv \phi) \mr \\
&= 8\pi^2 r^3 (1-\mu) \ml  (\du r)(\dv r) \frac{(\du \phi)^2 (\dv \phi)^2}{(\du r)^2(\dv r)^2} - \frac{1}{2\pi r^2} (\du r)(\dv r) \frac{(\dv \phi)(\du \phi)}{(\dv r)(\du r)} \mr.
\end{aligned}
\end{equation*}
Utilizing the elementary inequality $-x^2 + ax \leq {a^2}/{4}$ for $x, a\in \mathbb{R}$, we immediately obtain
\begin{equation*}
\begin{aligned}
\int^{v_2}_{v_1} -\du \dv m \; \D v' &= \int^{v_2}_{v_1} 8\pi^2 r^3 (1-\mu)(-\du r)(\dv r) \ml  -\frac{(\du \phi)^2 (\dv \phi)^2}{(\du r)^2(\dv r)^2} + \frac{1}{2\pi r^2} \frac{(\dv \phi)(\du \phi)}{(\dv r)(\du r)} \mr \; \D v' \\
&\leq \int^{v_2}_{v_1} 8\pi^2 r^3 (1-\mu)(-\du r)(\dv r) \ml \frac{1}{16\pi^2 r^4}\mr \; \D v' \\
&\leq \frac{(-\du r)_2}{2}\int^{v_2}_{v_1} \frac{\dv r}{r} \D v' = \frac{(-\du r)_2}{2}\log \ml 1 + \frac{r_2(u) - r_1(u)}{r_1(u)}\mr\leq \frac{(-\du r)_2}{2}\delta(u).
\end{aligned}
\end{equation*}
Here, we define $\delta(u):=(r_2(u)-r_1(u))/r_1(u)$, and we further denote $\delta_0$ to be the corresponding initial value. 

Back to \eqref{eta intro3}, this leads to an ODE inequality of the form
\begin{equation*}
\begin{aligned}
\frac{\D \eta}{\D x} &= - \frac{\eta}{x} + \frac{\delta(u)}{x} \leq - \frac{\eta}{x} + \frac{\delta_0}{x(x(1+\delta_0)-\delta_0)}.
\end{aligned}
\end{equation*}
Upon solving the ODE inequality above, we then get the desired result. In this new approach, we tend to use identities instead of inequalities, which relies on fewer requirements of the structures of equations. \\

We now generalize this new approach to the charged scenario. We first define a new quantity

\begin{equation*}
\eta^*(u) := \frac{2(m(u,v_2) - m(u,v_1))}{r(u,v_2)} - \frac{2}{r(u,v_2)}\int^{v_2}_{v_1} \frac{Q^2(\dv r)}{2r^2}(u,v')\mathrm{d}v'.
\end{equation*}
Note that outside the trapped region, this quantity is less or equal to the standard mass ratio, and we refer to it  as the \textit{reduced mass ratio}, denoted by $\eta^*(u)$. We conclude that, for the charged case, if it satisfies the condition:
\begin{equation}\label{TSF2'}
\eta^*(u_0) > C_{24} \frac{\delta(u_0)}{1+\delta(u_0)} \log \ml \frac{1}{\delta(u_0)}\mr + C_{23} \ml \frac{\delta(u_0)r_1(u_0)}{1-\mu^*(u_0)}\mr
\end{equation}
with the constants $C_{23}$, $C_{24}$ given in \eqref{TSF51}, \eqref{TSF54}, and $\mu^*(u_0)$ being the supremum of $\mu(u,v)$ along $u=u_0$, then a MOTS or a trapped surface is guaranteed to form in the evolution.  We emphasize that this is a crucial criterion. The factor $\frac{1}{1-\mu^*(u_0)}$ is related to the blueshift effect, and the additional $r_1(u_0)$ term will play a pivotal role in establishing the corresponding instability theorems. 

For the proof of this conclusion, we proceed as in \eqref{eta intro3} and obtain
\begin{equation}\label{eta intro4}
    \begin{split}
    \frac{d\eta^*}{dx}=\frac{\frac{d\eta^*}{du}}{\frac{dx}{du}}=-\frac{\eta^*}{x}-\frac{2}{x(\du r_2)}\bigg(\int_{v_1}^{v_2}[-\du \dv m+\du \ml \frac{Q^2\dv r}{2r^2}\mr] \bigg).    \end{split}
\end{equation}
Applying the equations for $m$ and for $Q$, we get
\begin{equation}\label{eta intro5}
\begin{aligned}
&- \du\dv m + \du \ml \frac{Q^2 \dv r}{2r^2} \mr \\ =& \; (\dv r)(-\du r) \ml -8\pi^2 r^3 (1-\mu) \ml \frac{|D_u \phi| |\dv \phi|}{(-\du r)(\dv r)} \mr^2  - 4\pi r (1-\mu) \frac{\re{((\dv \phi)^\dagger (D_u \phi))}}{\dv r (-\du r)}\mr\\
&\; + 4\pi \e Q (-\du r) \im{(\phi^\dagger (\dv \phi))} \\
\leq&\; 8\pi^2 r^3 (\dv r)(-\du r)(1-\mu) \ml - \ml \frac{|D_u \phi| |\dv \phi|}{(-\du r)(\dv r)} \mr^2  + \frac{1}{2\pi r^2} \frac{|D_u \phi||\dv \phi|}{(-\du r)(\dv r)}\mr + 4\pi \e |Q| (-\du r) |\phi||\dv \phi|. \\
\end{aligned}
\end{equation}
Back to \eqref{eta intro4}, we then integrate the first two terms on the right of \eqref{eta intro5} and obtain
\begin{equation*}
\begin{aligned}
&\int^{v_2}_{v_1} 8\pi^2 r^3 (\dv r)(1-\mu) \ml - \ml \frac{|D_u \phi| |\dv \phi|}{(-\du r)(\dv r)} \mr^2  + \frac{1}{2\pi r^2} \frac{|D_u \phi||\dv \phi|}{(-\du r)(\dv r)}\mr(u,v') \mathrm{d}v' \\
\leq&\; \int^{v_2}_{v_1} 2\pi^2 r^3 (\dv r)(1-\mu) \ml \frac{1}{2\pi r^2}\mr^2(u,v') \mathrm{d}v' \\
\leq&\; \frac{1}{2} \int^{v_2}_{v_1}   \frac{\dv r}{r}(u,v') \mathrm{d}v' = \frac{1}{2}\log\ml \frac{r(u,v_2)}{r(u,v_1)}\mr  = \frac{1}{2}\log \ml 1 + \frac{r(u,v_2) - r(u,v_1)}{r(u,v_1)}\mr \leq \frac{1}{2}\delta(u).
\end{aligned}
\end{equation*}
To bound the integration of the last term in \eqref{eta intro5}, by applying the obtained estimates for $|Q\phi|$, with $\chi \in \left( 0, \frac{1}{3}\right]$ we derive
\begin{equation*}
\begin{aligned}
\int^{v_2}_{v_1} 4\pi\e|Q||\phi||\dv \phi| (u,v') \mathrm{d}v' 
&\leq \sqrt{4\pi} \e {\eta^*}^{\frac{1}{2}} r_2^{\frac{1}{2}}(u) \ml \int^{v_2}_{v_1} \frac{\dv r}{1-\mu} \frac{|Q|^2|\phi|^2}{r^2}(u,v') \mathrm{d}v'\mr^{\frac{1}{2}} \\
&\leq \sqrt{4\pi} \e C_{5} {\eta^*}^{\frac{1}{2}} r_2^{\frac{1}{2}}(u)  \ml \int^{v_2}_{v_1} \frac{\dv r}{1 - \mu}(u_0,v') r^{\frac{1}{2}-\frac{3\chi}{2}}(u,v') \mathrm{d}v'\mr^{\frac{1}{2}} \\
&\leq \frac{\sqrt{4\pi} \e C_{5}}{(1-\mu_0^*)^{\frac{1}{2}}} {\eta^*}^{\frac{1}{2}} r_2^{\frac{3}{4}-\frac{3\chi}{4}}(u)  r_1^{\frac{1}{2}}(u) \delta(u_0)^{\frac{1}{2}}.
\end{aligned}
\end{equation*}

\noindent These together with $(-\du r)(u,v') \leq (-\du r)_2(u)$ for $v' \in [v_1,v_2]$ give 
\begin{equation*}
\begin{aligned}
\frac{\mathrm{d}\eta^*}{\mathrm{d}x} &\leq - \frac{\eta^*}{x} - \frac{2}{x(\du r_2)} \ml \int^{v_2}_{v_1} -\du\dv m + \du \ml \frac{Q^2 \dv r}{2r^2}\mr \mathrm{d}v'\mr \\
&\leq - \frac{\eta^*}{x} + \frac{\delta}{x} + \frac{1}{x} \frac{4\sqrt{\pi} \e C_{5}}{(1-\mu_0^*)^{\frac{1}{2}}} {\eta^*}^{\frac{1}{2}} r_2^{\frac{3}{4}-\frac{3\chi}{4}}(u)r_1^{\frac{1}{2}}(u_0) \delta^{\frac{1}{2}}(u_0) \\
&\leq - \frac{\eta^*}{x} + \frac{\delta}{x} + C_{21}x^{-\frac{1}{4}-\frac{3\chi}{4}}\ml \frac{2{\eta^*}^{\frac{1}{2}}\delta_0^{\frac{1}{2}}(r_1)_0^{\frac{1}{2}}}{(1-\mu_0^*)^{\frac{1}{2}}} \mr \\
&\leq - \frac{\eta^*}{x} + \frac{\delta}{x} + \frac{C_{21}}{x^{\frac{1}{2}}} \ml \eta^* + \frac{\delta_0\cdot(r_1)_0}{1-\mu_0^*} \mr \\
\end{aligned}
\end{equation*}
with $C_{21} = 2\sqrt{\pi} \e C_{5}r(\overline{u},v_3)^{\frac{3}{4}-\frac{3\chi}{4}}$. Solving this ODE then leads us to the desired trapped surface formation criterion for the charged case.

\subsubsection{BV Area Estimates}
{\color{black} In Section \ref{BV Area Estimates}, we establish the scale-critical BV area estimates in a region $\mathcal{D}(0,v_0)$ with $v_0$ small. We start by recalling the definition for the area of a pair of functions and it is given by
$$A[f,g] = \int_{\mathcal{D}(0,v_0)} \mlm \frac{\p f}{\p u} \frac{\p g}{\p v} - \frac{\p f}{\p v} \frac{\p g}{\p u}\mrm \D u \D v.$$
Notice that, unlike the non-charged case, our system \eqref{Setup12} - \eqref{Setup19} is \textit{not} invariant under the transformation $r \mapsto r$, $\Omega \mapsto \Omega$, $\phi \mapsto \phi + c$.\footnote{This is translation transformation for $\phi$.} We can no longer assume $\phi(0,0)=0$ without loss of generality as in \cite{christ2} for the uncharged case. Henceforth, in this paper, we conduct renormalizations for terms involving $\phi$ and meticulously define the following quantities:
\begin{equation*}
\begin{aligned}
\phi_0 := \phi(0,0), \quad \tp := \phi - \phi_0, \quad \lambda := \dv r, \quad \nu := \du r, \\ 
\tilde{\theta} := r \dv \tp, \quad \quad \quad
\tilde{\zeta} := r \du \tp, \quad \quad \quad
\ttz := r D_u \phi, \\
\mfa := \dv(r\phi), \quad \mfb := \du(r\phi), \quad \alpha := \frac{\mfa}{\lambda}, \quad \beta := \frac{\mfb}{\nu}, \\ 
\tmfa := \dv(r\tp), \quad \tmfb := \du(r\tp), \quad \tilde{\alpha} := \frac{\alpha}{\lambda}, \quad \tilde{\beta} := \frac{\beta}{\nu},  \\
A:= \sup_{\mathcal{D}(0,v_0)} |\tmfa|, \quad B:= \sup_{\mathcal{D}(0,v_0)} |\tmfb|, \quad P:= \sup_{\mathcal{D}(0,v_0)} |\tp|.
\end{aligned}
\end{equation*}
We further set
\begin{equation*}
\begin{aligned}
\ttmfb := \tmfb + \ii \e A_u r \phi, \quad X := \sup_{u\in(0,v_0)}\int_{u}^{v_0}|\dv \tp|(u,v') \D v', \quad Y := \sup_{v \in (0,v_0)} \int_{v}^{0}|\du \tp|(u',v) \D u',
\end{aligned}
\end{equation*}
and let 
$$D := TV_{\{0\} \times (0,v_0)}[\mfa] = TV_{\{0\} \times (0,v_0)}[\tmfa].$$
Note that in $\ttz, \alpha, \beta, \tilde{\alpha}, \tilde{\beta}$, we exclude the use of $\tp$, and the form $\ttmfb$ is specifically tailored for the charged scenario. Additionally, notice that the size of the initial data is measured in $D$ and is small. Our obtained area estimates are summarized in the theorem below. 

\noindent \textbf{Theorem.} ($BV$ Area Estimates.)
Assume that \begin{equation}\label{Area1'}
\inf_{\mathcal{D}(0,v_0)} \dv r \geq \frac{5}{12}, \quad \sup_{\mathcal{D}(0,v_0)} (-\du r) \leq \frac{2}{3}, \quad \sup_{\mathcal{D}(0,v_0)} \mu \leq \frac{1}{3}
\end{equation} 
hold for $v_0$ sufficiently small.\footnote{In the regular region, employing a spacetime rescaling as conducted by Christodoulou in Section 6 of \cite{christ2} can verify these assumptions.} Then there exists a constant $\varepsilon_0 \in (0,1]$, such that if the initial data satisfy
\begin{equation*}
Z = \max\{X,Y\} \leq \varepsilon_0 \quad \text{ and } \quad D := TV_{\{0\} \times (0,v_0)}[\tmfa]  \leq \varepsilon_0,
\end{equation*}
then we have $Z \lesssim D$, the total variations satisfy
\begin{equation}\label{Area40'}
\begin{aligned}
\sup_{u \in (0,v_0)}TV_{\{u\}\times(u,v_0)}[\tmfa] &\lesssim D,\\
\sup_{v \in (0,v_0)}TV_{(0,v) \times \{v\}}[\ttmfb] &\lesssim D,\\
TV_{(0,v_0)}[\tmfa|_{\Gamma}]  &\lesssim D,
\end{aligned} \quad 
\begin{aligned}
\sup_{u \in (0,v_0)} TV_{ \{u\} \times (u,v_0) } [\log \lambda] &\lesssim D^2, \\
\sup_{v\in(0,v_0)} TV_{(0,v) \times \{v\}}[\log |\nu|] &\lesssim D^2, \\
TV_{(0,v_0)}[\log(\lambda|_{\Gamma})] &\lesssim D^2, \\
\end{aligned}
\end{equation}
and the area estimates obey
\begin{equation}\label{Area41'}
\begin{aligned}
A[\tp^\dagger ,\tmfa/\lambda] &\lesssim D^2, \\
A[\tp^\dagger ,\ttmfb/\nu] &\lesssim D^2, \\
\end{aligned} \quad
\begin{aligned}
A[\lambda, \tp] &\lesssim D^3, \\
A[\nu, \tp] &\lesssim D^3. \\
\end{aligned}
\end{equation}

Prior to computing the relevant area estimates, we first establish a couple of lemmas in Section \ref{BV Area Estimates} which we will described below.

\noindent \textbf{Lemma 1.} For all $(u,v) \in \mathcal{D}(0,v_0),$ we have
$$\int^{v_0}_u |\dv \tilde{\phi}|(u,v') \D v' \leq TV_{\{u\} \times (u,v_0)} [\tilde\alpha] \quad \text{ and } \quad \int^{v}_0 |\du \tilde{\phi}|(u',v) \D u' \leq TV_{(0,v) \times \{v\} } [\tilde\beta].$$
We prove this lemma by utilizing integration by parts and leveraging the fact that the relevant boundary terms vanish along the center $\Gamma$. To proceed, in Section \ref{BV Area Estimates}, we then prove\\

\noindent \textbf{Lemma 2.} Assuming $D \leq 1$, $v_0$ being sufficiently small and \eqref{Area1'} holds, then we get
$$A \leq 4D, \quad B  \leq 8D, \quad P  \leq 12 D.$$
The proof of the lemma is via a bootstrapping argument. We first assume that for a fixed $(\tilde{u},\tilde{v}) \in \mathcal{D}(0,v_0),$
$$\sup_{\mathcal{D}(0,\tilde{u},\tilde{v})} |\tmfa| \leq 4D, \quad \sup_{\mathcal{D}(0,\tilde{u},\tilde{v})} |\tp|  \leq 12 D,$$
and attempt to improve on such an estimate. Since $\tp = \frac{1}{r(u,v)}\int^v_u \tmfa (u,v') \D v'$, we obtain
\begin{equation}\label{Area2'}
\sup_{\mathcal{D}(0,\tu,\tv)}|\tp| \leq \sup_{\mathcal{D}(0,\tu,\tv)}|\tmfa| \sup_{\mathcal{D}(0,\tu,\tv)}\mlm \frac{v - u}{r(u,v)}\mrm \leq 3 \sup_{\mathcal{D}(0,\tu,\tv)}|\tmfa|.
\end{equation}
To control $\tmfa$, we employ the following equation:
\begin{equation}\label{Area10'}
\du\ml \dv \ml r\tp \mr \mr = \tp \du \dv r- \ii \e  A_u \dv (r\tp) - \ii \e \phi_0 A_u \dv r  - \ii \e 
 \frac{\du r \dv r }{1 - \mu}\frac{Q\phi}{r}.
\end{equation}
By viewing this as a first-order ODE in $\dv(r\tp)$, with the help of an integrating factor, we deduce that
\begin{equation}\label{Area3'}
\begin{aligned}
|\dv(r\tp)|(u,v) \leq&\; |\dv (r\tp)|(0,v) + \int^u_0 |\tp|(-\du \dv r) (u',v)\D u'   \\
&+ \e \int^u_0 \frac{(-\du r)(\dv r)}{r(1-\mu)}|Q||\phi| (u',v)\D u'  + \e |\phi_0| \int^u_0 |A_u| \dv r (u',v) \D u'.
\end{aligned}
\end{equation}
To close the bootstrap, we will use $\tmfa$ to bound $|Q|$ and $|A_u|$. From \eqref{Setup19}, we have
\begin{equation}\label{Area5'}
\begin{aligned}
|Q|(u,v) &\leq 4 \pi \e \int^v_u r^2 |\phi||\dv \phi|(u,v') \D v' \leq 4 \pi \e \ml |\phi_0| + \sup_{\mathcal{D}(0,\tu,\tv)}|\tp| \mr r(u,v) \int^v_u r |\dv \phi|(u,v') \D v' \\
&\leq  4 \pi \e \ml |\phi_0| + \sup_{\mathcal{D}(0,\tu,\tv)}|\tp| \mr r(u,v) \int^v_u |\dv (r\tp)|(u,v') + |\tp|(\dv r)(u,v') \D v' \\
&\leq  4 \pi \e \ml |\phi_0| + \sup_{\mathcal{D}(0,\tu,\tv)}|\tp| \mr r(u,v) \ml \ml \sup_{\mathcal{D}(0,\tu,\tv)}|\tmfa| \mr (v-u) + \ml \sup_{\mathcal{D}(0,\tu,\tv)}|\tp| \mr r(u,v) \mr \\
&\leq 24\pi\e \ml \mlm \phi_0 \mrm + \sup_{\mathcal{D}(0,\tu,\tv)}|\tp| \mr \ml \sup_{\mathcal{D}(0,\tu,\tv)}|\tmfa| \mr  r(u,v)^2.
\end{aligned}
\end{equation}
By utilizing \eqref{Setup30} and \eqref{Area5'}, we further obtain
\begin{equation}\label{Area12'}
\begin{aligned}
|A_u|(u',v) &\leq \int^{v}_{u'} \frac{2|Q| (-\du r) \dv r}{r^2(1-\mu)} (u',v') \D v' \leq 24 \pi v_0 \e \ml \sup_{\mathcal{D}(0,\tu,\tv)}|\tmfa| \mr \ml \mlm \phi_0 \mrm + \sup_{\mathcal{D}(0,u',v)}|\tp| \mr.
\end{aligned}
\end{equation}
These two estimates will allow us to deduce that 
\begin{equation}\label{Area9'}
\ml \sup_{\mathcal{D}(0,\tu,\tv)}|\tmfa| \mr \leq 3D \quad \mbox{ and } \quad \sup_{\mathcal{D}(0,\tu,\tv)}|\tp| \leq 3\ml \sup_{\mathcal{D}(0,\tu,\tv)}|\tmfa| \mr \leq 9D,
\end{equation}
which improves the bootstrap assumption in \eqref{Area2'}. The desired estimates for $A$ and $P$ further imply the bound for $B$.\\

With these preparations, we start the proof of the theorem. Here, we generalize Christodoulou's ingenious arguments from \cite{christ2} to the complex charged scenario. The proof in Section \ref{BV Area Estimates} is relies on extensive calculations, and the estimates are intricately connected. Here, we provide an outline. For brevity, given a quantity $H$, we write $H \lesssim D$ to convey that there exists a non-negative constant $c$ (possibly depending on $\e$ and $|\phi_0|$) independent of $D$ such that $H \leq cD$. In view of Lemma 2 above, we then have
$$A \lesssim D, \quad B \lesssim D, \quad P \lesssim D.$$
Furthermore, $\dv r|_{\Gamma} = -\du r|_{\Gamma}$ and \eqref{Area1'} implies
$$\frac{5}{12} \leq \dv r(u,v) \leq \frac{1}{2}, \quad \text{ and } \quad \frac{1}{2} \leq (-\du r)(u,v) \leq \frac{2}{3}.$$
In addition, by
$$r \dv \tp = \dv \ml r \tp \mr - \tp \dv r, \quad r \du \tp = \du \ml r \tp \mr - \tp \du r, \quad \phi = \phi_0 + \tp,$$
together with estimates \eqref{Area5'} and \eqref{Area12'},
we deduce that
$$\sup_{\mathcal{D}(0,v_0)}\mlm r \dv \tp \mrm \lesssim D, \quad \sup_{\mathcal{D}(0,v_0)}\mlm r \du \tp \mrm \lesssim D, \quad \sup_{\mathcal{D}(0,v_0)}\mlm \phi \mrm \lesssim 1, \quad$$ 
and
$$|Q|(u,v) \lesssim D r^2(u,v), \quad |A_u|(u,v) \lesssim Dr(u,v), \quad |rD_u \phi|(u,v) \lesssim D.$$

We begin with deriving estimates for $\partial_v\tmfa$. This is done by computing $\du \dv \tmfa$ via using \eqref{Area42} and we get
\begin{equation}\label{Area19'}
\du \ml \dv \tmfa \mr(u',v') - \du\ml \tp \dv \lambda \mr(u',v') = G_1(u',v')  + G_2(u',v')  + \frac{\p (\lambda, \tp)}{\p(u,v)}(u',v')
\end{equation}
with 
$$\begin{aligned}
G_1 &= -\ii \e \ml \phi_0 A_u + \frac{\du r }{1-\mu}\frac{Q\phi}{r} \mr \dv \lambda, \\
G_2 &=  -\ii \e \ml (\dv A_u) \tmfa + \phi_0 (\dv A_u)\dv r + A_u \dv \tmfa + \dv r \dv \ml \frac{\du r}{1-\mu}\frac{Q\phi}{r}\mr\mr.
\end{aligned}$$
For each $(u,v) \in \mathcal{D}(0,v_0)$, by integrating with respect to $u'$ from $0$ to $v$ and $v'$ from $u$ to $v_0$, we obtain
$$\begin{aligned}
\int_u^{v_0} \mlm \dv \tmfa - \tp \dv \lambda \mrm (u,v') \D v' \leq \int^{v_0}_u \mlm \dv \tmfa \mrm(0,v') \D v' + \int_u^{v_0}\int^u_0 \ml |G_1| + |G_2| + \mlm\frac{\p(\lambda,\tp)}{\p(u,v)}\mrm \mr(u',v') \D u' \D v'.
\end{aligned}$$
Taking supremum over $u \in (0,v_0)$ and applying relevant estimates listed at the start of the proof, we derive:
\begin{equation}\label{Area13'}
\begin{aligned}
\sup_{u \in (0,v_0)}TV_{\{u\}\times(u,v_0)}[\tmfa] \leq& \; D + P \sup_{u \in (0,v_0)}TV_{\{u\}\times(u,v_0)}[\lambda] + A[\lambda,\tp] \\
&+ \sup_{u\in(0,v_0)}\int_u^{v_0}\int^u_0 |G_1| (u',v') \D u' \D v' + \sup_{u\in(0,v_0)}\int_u^{v_0}\int^u_0 |G_2| (u',v') \D u' \D v'.
\end{aligned}
\end{equation}
Furthermore, by applying the obtained estimates, we also show that
\begin{equation}\label{Area14'}
\begin{aligned}
\int_u^{v_0}\int^u_0 |G_1| (u',v') \D u' \D v' &\lesssim D \int_0^{u}\int^{v_0}_u |\dv \lambda| (u',v') \D v' \D u' \lesssim D \sup_{u \in (0,v_0)}TV_{\{u\}\times(u,v_0)}[\lambda], 
\end{aligned}
\end{equation}
and
\begin{equation}\label{Area15'}
\int_u^{v_0}\int^u_0 |G_2| (u',v') \D u' \D v' - \frac{1}{2}\sup_{u \in (0,v_0)}TV_{\{u\} \times (u,v_0) }[\tmfa] \lesssim  \int_u^{v_0}\int^u_0 D \D u' \D v' \lesssim D.
\end{equation}
These then simplify \eqref{Area13'} to 
\begin{equation}\label{Area16'}
\begin{aligned}
\frac{1}{2}\sup_{u \in (0,v_0)}TV_{\{u\}\times(u,v_0)}[\tmfa] \lesssim & \; D\ml 1 + \sup_{u \in (0,v_0)}TV_{\{u\}\times(u,v_0)}[\lambda]\mr  + A[\lambda,\tp]. \\
\end{aligned}
\end{equation}
Moving forward, we proceed to derive estimates for $\ttmfb$. By \eqref{Setup73} and taking derivatives in $u$, $v$, we arrive at 
$$\dv \du \ttmfb = \dv \ml \tp \du \nu \mr + G_3 + G_4 - \frac{\p(\nu,\tp)}{\p(u,v)}$$
with
$$\begin{aligned}
G_3 = \ii \e \du \nu \ml \frac{\dv r}{1-\mu}\frac{Q\phi}{r}\mr \quad \text{ and } \quad G_4 = \ii \e \du r \du \ml \frac{\dv r}{1-\mu}\frac{Q\phi}{r} \mr. 
\end{aligned}$$
Integrating along a future-directed outgoing null curve originating from $\Gamma$, we obtain
\begin{equation}\label{Area17'}
\begin{aligned}
\int^v_0 \mlm \du \ttmfb - \tp \du \nu\mrm(u',v) \D u' \leq&\; TV_{(0,v_0)}[\ttmfb|_{\Gamma}] + P \cdot TV_{(0,v_0)}[\nu|_{\Gamma}] + A[\nu,\tp] \\
&\; + \int^v_0 \int^v_{u'} |G_3|(u',v') \D v' \D u' + \int^v_0 \int^v_{u'} |G_4|(u',v') \D v' \D u'.
\end{aligned}
\end{equation}
By utilizing the obtained estimates, we deduce that
$$\int^v_0 \int^v_{u'} |G_3|(u',v') \D v' \D u' \lesssim D \cdot \sup_{v \in (0,v_0)}TV_{(0,v) \times \{v\}}[\nu] \quad \text{ and } \quad \int^v_0 \int^v_{u'} |G_4|(u',v') \D v' \D u' \lesssim D.$$
Furthermore, by employing relevant regularity conditions along $\Gamma$, we get
$$TV_{(0,v_0)}[\ttmfb|_\Gamma] = TV_{(0,v_0)}[\tmfa|_\Gamma], \quad TV_{(0,v_0)}[\nu|_\Gamma] = TV_{(0,v_0)}[\lambda|_\Gamma].$$
Combining these with \eqref{Area17'} then yields:
\begin{equation}\label{Area18'}
\begin{aligned}
\sup_{v \in (0,v_0)}TV_{(0,v) \times \{v\}}[\ttmfb] \lesssim & \; TV_{(0,v_0)}[\tmfa|_\Gamma] + D \ml 1 + TV_{(0,v_0)}[\lambda|_{\Gamma}] + \sup_{v \in (0,v_0)} TV_{(0,v) \times \{v\}}[\nu]\mr  + A[\nu,\tp].
\end{aligned}
\end{equation}
On the other hand, if we first integrate along a future-directed incoming null curve $v = u$ starting from $C_0^+$ to $(u,u) \in \Gamma$, and then integrate with respect to $u$ for $u \in (0,v_0)$, together with the boundary conditions along $\Gamma$, we will arrive at:
$$\begin{aligned}
TV_{(0,v_0)}[\tmfa|_{\Gamma}] &\leq  P \cdot TV_{(0,v_0)}[\lambda|_{\Gamma}] + D  + A[\lambda,\tp] + \int_0^{v_0}\int^{v_0}_{u'} |G_1|(u',v)  + |G_2|(u',v) \D v \D u'.
\end{aligned}$$
Using a similar argument for the double integrations of $|G_1|$ and $|G_2|$ as in \eqref{Area14'} and \eqref{Area15'}, we then deduce that
\begin{equation}\label{Area20'}
TV_{(0,v_0)}[\tmfa|_{\Gamma}] 
\lesssim D\ml 1 + TV_{(0,v_0)}[\lambda|_{\Gamma}] + \sup_{u\in(0,v_0)}TV_{\{u\}\times(u,v_0)}[\lambda]\mr + A[\lambda,\tp].
\end{equation}

\vspace{5mm}

Next, we will attempt to control $\sup_{u\in(0,v_0)}TV_{\{u\}\times(u,v_0)}[\lambda]$ appearing in \eqref{Area20'}. Starting from \eqref{Setup20},  by applying \eqref{Setup27}, \eqref{Setup29}, and by carefully exploiting the relevant cancellations between the terms, we can show that
\begin{equation}\label{Area23'}
\begin{aligned}
\dv \du \log \lambda =&\; \du\ml \frac{\mu - Q^2/r^2}{1-\mu}\frac{\dv r}{r}\mr + \frac{4\pi}{r^2}\ml \frac{1-Q^2/r^2}{1-\mu}\mr\ml \frac{(r|\dv \phi|)^2(\du r)}{(\dv r)} - \frac{(r|D_u \phi|)^2(\dv r)}{\du r}\mr\\
&+ \ml \frac{8\pi \e Q}{1-\mu} \mr \ml \frac{1}{1-\mu}\mr\ml (\du r)  \im(\phi^{\dagger}\dv \phi ) + (\dv r)\im(\phi^{\dagger}D_u \phi)\mr
\end{aligned}
\end{equation}
and
\begin{equation}\label{Area24'}
\begin{aligned}
\du \dv \log(-\nu) =&\; \dv\ml \frac{\mu - Q^2/r^2}{1-\mu}\frac{\du r}{r}\mr - \frac{4\pi}{r^2}\ml \frac{1-Q^2/r^2}{1-\mu}\mr\ml \frac{(r|\dv \phi|)^2(\du r)}{(\dv r)} - \frac{(r|D_u \phi|)^2(\dv r)}{\du r}\mr\\
&- \ml \frac{8\pi \e Q}{1-\mu} \mr \ml \frac{1}{1-\mu}\mr\ml (\du r)  \im(\phi^{\dagger}\dv \phi ) + (\dv r)\im(\phi^{\dagger}D_u \phi)\mr.
\end{aligned}
\end{equation}
Define 
\begin{equation*}
\tf(u,v) := 4\pi\re\ml \frac{r (\dv \tp)^\dagger}{r}\ml \frac{\tmfa}{\lambda} - \frac{\ttmfb}{\nu}\mr\mr(u,v)
\end{equation*}
and
\begin{equation*}
\tfs(u,v) := 4\pi\re\ml \frac{r (\du \tp)^{\dagger}}{r}\ml \frac{\tmfa}{\lambda} - \frac{\ttmfb}{\nu}\mr \mr(u,v).
\end{equation*}
By performing laborious computations and utilizing the fact that $( \du \tp)^{\dagger}(\dv \tp) - (\du \tp)(\dv \tp)^\dagger$ is purely imaginary, along with the identity below,
$$|D_u \phi|^2 = (\du \phi + \ii \e A_u \phi)(\du \phi - \ii \e A_u \phi)^\dagger = |\du \tp|^2 - 2 \e A_u \im(\phi \du \phi^\dagger) + \e^2 A_u^2 |\phi|^2,$$
we derive
\begin{equation}\label{Area25'}
\begin{aligned}
\du\ml \dv \log \lambda \mr = \du \ml \frac{\mu - Q^2/r^2}{1-\mu}\frac{\dv r}{r} - \tf\mr + 4\pi \re \ml \frac{\p(\tp^\dagger,\ttmfb/\nu)}{\p(u,v)}\mr + E_5 + E_6
\end{aligned}
\end{equation}
and
\begin{equation}\label{Area26'}
\begin{aligned}
\dv\ml \du \log (-\nu) \mr = \dv \ml \frac{\mu - Q^2/r^2}{1-\mu}\frac{\du r}{r} + \tfs\mr - 4\pi \re \ml \frac{\p(\tp^\dagger,\tmfa/\lambda)}{\p(u,v)}\mr - E_5 + E_6.
\end{aligned}
\end{equation}
Here, $E_5$ and $E_6$ are terms satisfying $|E_5| \lesssim D^2$ and $|E_6| \lesssim D^2$. From \eqref{Area25'}, for a fixed $(u,v) \in \mathcal{D}(0,v_0)$, by first integrating along a future-directed incoming null curve and further integrating for $v \in (u,v_0)$, we obtain
\begin{equation}\label{Area27'}
\begin{aligned}
\int^{v_0}_u \mlm \dv \log \lambda - \frac{\mu - Q^2/r^2}{1-\mu}\frac{\dv r}{r} + \tf\mrm(u,v') \D v'\lesssim \int^{v_0}_u \mlm \frac{\mu - Q^2/r^2}{1-\mu}\frac{\dv r}{r} - \tf\mrm(0,v') \D v' + A[\tp^\dagger,\ttmfb/\nu] + D^2.
\end{aligned}
\end{equation}
By employing relevant derived estimates and boundary conditions along $\Gamma$, we have
\begin{equation*}
\begin{aligned}
\int^{v_0}_u \frac{\mu - Q^2/r^2}{1-\mu}\frac{(\dv r)}{r}(u,v') \D v' \lesssim DX, \quad
&\int^{v_0}_u \frac{\mu - Q^2/r^2}{1-\mu}\frac{(\dv r)}{r}(0,v') \D v' \lesssim D^2, \\
\int^{v_0}_u |\tf|(u,v') \D v' \lesssim DX, \quad 
&\int^{v_0}_u |\tf|(0,v') \D v' \lesssim D^2.
\end{aligned}
\end{equation*}
Substituting these back to \eqref{Area27'}, we arrive at
\begin{equation}\label{Area28'}
\sup_{u \in (0,v_0)} TV_{ \{u\} \times (u,v_0) } [\log \lambda] \lesssim D^2 + DX + A[\tp^\dagger, \ttmfb/\nu].
\end{equation}
An analogous calculation for \eqref{Area26'} yields
\begin{equation}\label{Area29'}
\begin{aligned}
\int^v_0 \mlm \du \log(-\nu) - \frac{\mu - Q^2/r^2}{1-\mu}\frac{\du r}{r} - 4\pi \tfs \mrm(u',v) \D u' \lesssim  TV_{(0,v_0)}[\log(-\nu)|_{\Gamma}] +  A[\tp^\dagger,\tmfa/\lambda] + D^2.
\end{aligned}
\end{equation}
Since
$$\int^v_0 \frac{\mu - Q^2/r^2}{1-\mu}\frac{(-\du r)}{r}(u',v) \D u' \lesssim DY + D^2 \quad \text{ and } \quad \int^v_0 |\tfs|(u',v) \D u' \lesssim DY,$$
inequality \eqref{Area29'} implies
\begin{equation}\label{Area30'}
\sup_{v\in(0,v_0)} TV_{(0,v) \times \{v\}}[\log |\nu|] \lesssim D^2 + DY + TV_{(0,v_0)}[\log\lambda|_{\Gamma}] + A[\tp^\dagger,\tmfa/\lambda].
\end{equation}
By a similar approach for $\log(\lambda|_{\Gamma})$, we also get
\begin{equation}\label{Area31'}
TV_{(0,v_0)}[\log(\lambda|_{\Gamma})] \lesssim D^2 + A[\tp^\dagger,\ttmfb/\nu].
\end{equation}

\vspace{5mm}

It now remains to establish relations between the different area terms that have appeared in the estimates above. Starting from $A[\tp,\tmfa/\lambda]$, we observe that
\begin{equation}\label{Area32'}
\frac{\p(\tp^\dagger,\tmfa/\lambda)}{\p(u,v)} = \frac{(\du \tp)^\dagger}{\lambda}(\dv \tmfa - \tp \dv \lambda) + \frac{(r \dv \tp)}{\lambda^2}\frac{\p(\lambda, \tp^\dagger)}{\p(u,v)}- \frac{(\dv \tp)^\dagger E_7}{\lambda}  
\end{equation}
with $E_7$ satisfying $|E_7| \lesssim Dr$. Since $A[\lambda,\tp^\dagger] = A[\lambda,\tp]$ and 
$$\int_{D(0,v_0)} \frac{|\du \tp|}{\lambda}|\dv \tmfa - \tp \dv \lambda|(u,v) \D u \D v \lesssim DY\ml 1 + \sup_{u \in (0,v_0)}TV_{\{u\}\times(u,v_0)}[\lambda]\mr  + Y\cdot A[\lambda,\tp^\dagger],$$
identity \eqref{Area32'} then implies 
\begin{equation}\label{Area33'}
A[\tp^\dagger, \tmfa/\lambda] \lesssim DY\ml 1 + \sup_{u \in (0,v_0)}TV_{\{u\}\times(u,v_0)}[\lambda]\mr + D^2 + (D+Y)\cdot A[\lambda,\tp^\dagger].
\end{equation}
A similar strategy also gives
\begin{equation}\label{Area34'}
\begin{aligned}
A[\tp^\dagger, \ttmfb/\nu] \lesssim&\; X \cdot TV_{(0,v_0)}[\ttmfb|_{\Gamma}] + DX\ml 1 +  TV_{(0,v_0)}[\lambda|_{\Gamma}] + \sup_{v \in (0,v_0)} TV_{(0,v) \times \{v\} }[\nu] \mr \\
&+ D^2 + (D + X)\cdot A[\nu,\tp^\dagger].
\end{aligned}
\end{equation}

In view of \eqref{Area33'} and \eqref{Area34'}, we proceed to estimate $A[\lambda,\tp^\dagger]$ and $A[\nu,\tp^\dagger]$. Inspired by the definition of $\tf$ above, we further define
\begin{equation}\label{Area80'}
\tg := 4 \pi \im \ml \frac{\ttheta^\dagger}{r}\ml\frac{\tmfa}{\lambda} - \frac{\ttmfb}{\nu} \mr \mr
\end{equation}
so that
\begin{equation*}
\tf + \ii \tg = 4\pi \ml \frac{\ttheta^\dagger}{r} \mr\ml  \frac{\ttheta}{\lambda} - \frac{\tz}{\nu} + E_2\mr
\end{equation*}
with $E_2$ satisfying $|E_2| \lesssim Dr^2$. By computing $\frac{1}{\lambda}\frac{\p(\lambda,\tp)}{\p(u,v)}$, we deduce that
\begin{equation*}
\begin{aligned}
\frac{1}{\lambda}\frac{\p(\lambda,\tp)}{\p(u,v)} 
=&\;\frac{1}{r^2}\ml \frac{ \nu \ttheta - \lambda \tz }{\lambda \nu}\mr\ml \frac{\mu - Q^2/r^2}{1-\mu}\lambda \nu + 4 \pi \ttheta^\dagger \ttz\mr - \frac{\tz}{r}\ml \dv \log \lambda - \frac{\mu - Q^2/r^2}{1-\mu}\frac{\lambda}{r} + \tf \mr \\
&- \ii \frac{\tz \tg}{r} + \frac{4\pi \tz \ttheta^\dagger}{r^2}E_2 - 4\pi \ii \e \frac{\nu \ttheta - \lambda \tz}{\lambda \nu}\frac{\ttheta^\dagger A_u \phi}{r}.  \\
\end{aligned}
\end{equation*}
Note that the introduction of $\tg$ and the derivation of its desired bound are novel and crucial for the charged scenario. Next, we define
$$\tr := r \ml \frac{\mu - Q^2/r^2}{1-\mu}\lambda \nu + 4\pi \ttheta^\dagger \ttz\mr.$$
By estimating $\dv \log \lambda - \frac{\mu - Q^2/r^2}{1-\mu}\frac{\lambda}{r} + \tf$ via \eqref{Area25'} above, we obtain
\begin{equation}\label{Area35'}
\begin{aligned}
\int_{\mathcal{D}(0,v_0)}\mlm \frac{1}{\lambda}\frac{\p(\lambda,\tp)}{\p(u,v)} \mrm(u,v) \D u \D v \lesssim&\;  D \int_{\mathcal{D}(0,v_0)}\frac{|\tr|}{r^3}(u,v) \D u \D v 
+ D \int_{\mathcal{D}(0,v_0)}\frac{|r^2\tg|}{r^3}(u,v) \D u \D v \\ &+\; D^3 + D^2Y + Y\cdot A[\tp,\ttmfb/\nu].
\end{aligned}
\end{equation}
For the first integration on the right of \eqref{Area35'}, by conducting integration by parts and using the fact that $\tr/r^2|_{\Gamma} = 0$, we derive
\begin{equation*}
\begin{aligned}
\int_{\mathcal{D}(0,v_0)} \frac{|\tr|}{r^3}(u,v) \D u \D v &= 
\frac{1}{2}\int^{v_0}_0 \ml \int^v_0 |\tr/\nu| \du\ml \frac{1}{r^2}\mr(u,v) \D u \mr \D v \\
&= - \frac{1}{2}\int^{v_0}_0 \frac{|\tr/\nu|}{r^2}(0,v) \D v - \frac{1}{2}\int^{v_0}_0 \ml \int^v_0 \frac{1}{r^2}\du \mlm \frac{\tr}{\nu}\mrm(u,v) \D u\mr \D v \\
&\leq \frac{1}{2}\int_{\mathcal{D}(0,v_0)}\frac{1}{r^2}\mlm \du\ml \frac{\tr}{\nu}\mr \mrm (u,v) \D u \D v.
\end{aligned}
\end{equation*}
Employing 
\begin{equation*}
\du \ml \frac{\ttmfb}{\nu}\mr = \du \ml \frac{\nu \tp + r D_u \phi}{\nu}\mr = \frac{\ttz}{r} - \ii \e A_u \tp + \du \ml \frac{\ttz}{\nu}\mr
\end{equation*}
and
\begin{equation*}
\begin{aligned}
\ttheta^\dagger \du \ml \frac{\ttmfb}{\nu}\mr &= r\ml \frac{\p(\ttmfb/\nu,\tp^\dagger)}{\p(u,v)}\mr - \frac{\lambda(\mu - Q^2/r^2)|\ttz|^2}{r\nu(1-\mu)} + E_{10} \quad \mbox{ with } \quad |E_{10}| \lesssim D^2r,
\end{aligned}
\end{equation*}
through detailed computations, we unravel key cancellations in the expressions. Upon using them, we will arrive at 
$$\du \ml \frac{\tr}{\nu}\mr = 4\pi r^2 \frac{\p(\ttmfb/\nu,\tp^\dagger)}{\p(u,v)} + E_{11} \quad \mbox{ with } \quad |E_{11}| \lesssim D^2r^2.$$ 
Furthermore, we deduce that
\begin{equation}\label{Area36'}
\int_{\mathcal{D}(0,v_0)} \frac{|\tr|}{r^3}(u,v) \D u \D v \lesssim A[\tp^\dagger,\ttmfb/\nu] + D^2.
\end{equation}

We also employ integration by parts to treat the second integration on the right of \eqref{Area35'}. Together with $\tg|_{\Gamma}=0$, we obtain
\begin{equation}\label{Area60'}
\begin{aligned}
\int_{\mathcal{D}(0,v_0)} \frac{|r^2 \tg|}{r^3}(u,v) \D u \D v \leq \frac{1}{2}\int_{\mathcal{D}(0,v_0)}\frac{1}{r^2}\mlm \du\ml \frac{r^2\tg}{\nu}\mr \mrm (u,v) \D u \D v.
\end{aligned}
\end{equation}
Since $\tg= 4 \pi \im \ml \frac{\ttheta^\dagger}{r}\ml\frac{\tmfa}{\lambda} - \frac{\ttmfb}{\nu} \mr \mr$,
to compute $\du\ml \frac{r^2 \tg}{\nu}\mr$, we first calculate $\du\ml \frac{\ttheta^\dagger r}{\nu} \ml \frac{\tmfa}{\lambda} - \frac{\ttmfb}{\nu}\mr\mr$.Through a rigorous calculation, along with several crucial cancellations,  we derive
\begin{equation*}
\begin{aligned}
\du\ml\frac{\ttheta^\dagger r}{\nu} \ml \frac{\tmfa}{\lambda} - \frac{\ttmfb}{\nu}\mr\mr 
=&\; r^2 \frac{(\tz/\nu - E_2)}{\nu^2} \frac{\p(\nu,\tp^\dagger)}{\p(u,v)} + r^2\frac{1}{\nu}\frac{\p(\tp^\dagger,\ttmfb/\nu)}{\p(u,v)} + R_1 + E_{12},
\end{aligned}
\end{equation*}
where $R_1$ representing a term that is purely real (i.e., $\im(R_1) = 0$), and $E_{12}$ satisfying $|E_{12}|\lesssim D^2 r^2$.
This implies that
\begin{equation*}
\begin{aligned}
\du\ml \frac{r^2 \tg}{\nu}\mr &=   4\pi \im \ml \du \ml \frac{\ttheta^\dagger r}{\nu} \ml \frac{\tmfa}{\lambda} - \frac{\ttmfb}{\nu}\mr\mr \mr \\
&= 4\pi \im \ml r^2 \frac{(\tz/\nu - E_2)}{\nu^2} \frac{\p(\nu,\tp^\dagger)}{\p(u,v)} + r^2\frac{1}{\nu}\frac{\p(\tp^\dagger,\ttmfb/\nu)}{\p(u,v)}  + E_{12}\mr.
\end{aligned}
\end{equation*}
Going back to \eqref{Area60'}, we then deduce
\begin{equation}\label{Area37'}
\int_{\mathcal{D}(0,v_0)} \frac{|r^2 \tg|}{r^3}(u,v) \D u \D v \lesssim D \cdot A[\nu,\tp^\dagger] + A[\tp^\dagger,\ttmfb/\nu] + D^2.
\end{equation}
Plugging \eqref{Area36'} and \eqref{Area37'} into \eqref{Area35'}, we obtain
\begin{equation}\label{Area38'}
A[\lambda,\tp^\dagger] \lesssim (D+Y)A[\tp^\dagger,\ttmfb/\nu] + D^2\cdot A[\nu,\tp^\dagger] + D^2Y + D^3.
\end{equation}
We also introduce
$$\tgs := 4\pi \im\ml \frac{\tz^\dagger}{r}\ml \frac{\tmfa}{\lambda} - \frac{\ttmfb}{\nu}\mr\mr \quad \text{and} \quad \trs := r\ml \frac{\mu - Q^2/r^2}{1-\mu}\lambda \nu + 4\pi \ttheta \tz^\dagger\mr.$$
Here $\tgs$ is associated with the new complex charged scenario. Employing an analogous strategy as for $\tg$ and $\tr$, while giving additional attention to boundary terms, we arrive at
\begin{equation}\label{Area39'}
\begin{aligned}
A[\nu,\tp^\dagger] &\lesssim  X \cdot TV_{(0,v_0)}[\lambda|_{\Gamma}] + (X + D) \cdot A[\tp^\dagger,\tmfa/\lambda] + D^2 \cdot A[\lambda,\tp^\dagger] + D^2X + D^3.\\
\end{aligned}
\end{equation}
Combining \eqref{Area16'}, \eqref{Area18'}, \eqref{Area20'}, \eqref{Area28'}, \eqref{Area30'}, \eqref{Area31'}, \eqref{Area33'}, \eqref{Area34'}, \eqref{Area38'},  \eqref{Area39'}, together with $D$ and $v_0$ being small, we conclude that $Z \lesssim D$, and inequalities in \eqref{Area40'} and \eqref{Area41'} all hold.\\

\subsubsection{\texorpdfstring{$C^1$ Extension Theorem}{C1 Extension Theorem}} 

For spherically symmetric Einstein-scalar field system, Christodoulou proved the following:

\noindent \textbf{Theorem.} \cite{christ2} ($C^1$ Extension Theorem with Mass Ratio $\mu$.) \textit{Set $v_* > 0$ and assume that for every $\hat v \in (0,v_*)$ there is a $C^1$ solution on $\mathcal{D}(0,\hat v)$. Then, there exists a constant $\varepsilon' > 0$, such that if we have
$$ \lim_{u \rightarrow v_*} \sup_{\mathcal{D}(u,v_*)} \mu(u,v) < \varepsilon',$$
then we can find a $\tilde{v} > v_*$ and a $C^1$ solution in the extended region $\mathcal{D}(0,\tilde{v})$.}

In Section \ref{C1 Extension}, we prove two independent $C^1$ extension theorems. The first is novel, and is as follows:
\noindent \textbf{Theorem.} ($C^1$ Extension Theorem with Doppler Exponent $\gamma$.) Let $v_* > 0$ and assume for every $\hat{v} \in (0,v_*)$, a $C^1$ solution exists on $\mathcal{D}(0,\hat{v})$. If we have
\begin{equation}
\sup_{\mathcal{D}(0,v_*)} \gamma < + \infty,
\end{equation}
then a $C^1$ solution can be extended to $\mathcal{D}(0,\tilde{v})$ with $\tilde{v} > v_*$.  \\

\noindent This is our main theorem in Section \ref{C1 Extension} and is directly related to the later instability arguments. In this paper, we also generalize the aforementioned Christodoulou's $C^1$ extension theorem with mass-ratio $\mu$ to the charged case. It is stated and proved as Theorem \ref{FET}. For both proofs, we employ the quantities: 
\begin{equation*}
\alpha:=\partial_v(r\phi) \quad \mbox{ and } \quad \alpha' := \frac{1}{\dv r} \dv \left( \frac{\dv (r\phi)}{\dv r} \right) =\frac{\dv \alpha}{\dv r}.
\end{equation*}

\noindent We further define

\begin{equation}\label{GandBu introduction}
G:= \sup_{\mathcal{D}(0,v_*)} \gamma(u,v) \quad \text{ and } \quad B'(u):=\sup_{v\in[u, v_*]}|\alpha'|(u,v).
\end{equation}


With $\alpha'$, we then compute $\du \alpha'$ and obtain
\begin{equation}\label{a' intro1}
\begin{aligned}
\du \alpha'=& \; \frac{4 \pi r (Q^2/r^2-1)|\dv \phi|^2(\dv \phi)(\du r) }{(\dv r)^3 (1- \mu)} - \frac{2\du r}{r}\ml \frac{\mu - Q^2/r^2}{1-\mu}\mr \alpha' + \frac{\du r \dv \phi}{r \dv r}\ml \frac{5Q^2/r^2 - 
 3 \mu}{1-\mu}\mr\\
&\; -\frac{8 \pi \e Q (\im(\phi^\dagger \dv \phi))(\dv \phi)(\du r)}{(\dv r)^2(1-\mu)}   \\
&\;  -\ii \e \ml A_u \alpha' + \frac{3 Q \du r \dv \phi}{r(1-\mu)(\dv r)} + \frac{Q\phi \du r}{r^2 (1-\mu)} + \frac{4 \pi |\dv \phi|^2 Q\phi \du r}{(1-\mu)(\dv r)^2} - \frac{4  \pi \e \phi\im(\phi^\dagger \dv \phi) r\du r }{ \dv r (1-\mu)} \mr. \\
\end{aligned}
\end{equation}

\noindent Notice that, since  $Q=0$ and $\im{\phi}=0$ for the uncharged case, there are only three terms on the right of \eqref{a' intro1}. For the charged scenario, additional terms arise. Henceforth, we get
\begin{equation*}
\begin{aligned}
\alpha'(u,v) e^{\int^u_{0}\ii \e A_u (u',v) \mathrm{d}u'} &= \alpha'(0,v) +  \sum_{i=1}^{11} I_i(u,v),
\end{aligned}
\end{equation*}
with $I_i$ given in \eqref{FET16} via \eqref{K}. 
This implies
\begin{equation}\label{alpha' Ii introduction} 
|\alpha'|(u,v)\leq \sup_{\{0\}\times [0, v_*]}|\alpha'|+\sum_{i=1}^{11}|I_i|(u', v)du'.
\end{equation}

To illustrate our proof strategy, here we will focus on estimating $I_3$ and $I_{10}$. For $I_3$ we have

\begin{equation}\label{I_3' introduction}
|I_3|(u,v) \leq  4\pi \int^{u}_{0} (-\du r)\ml \frac{1}{r^2}\mr\ml \frac{1}{1-\mu}\mr\ml \frac{r|\dv \phi|}{\dv r}\mr^3 (u',v) \mathrm{d}u'.
\end{equation}
Observe that, since
\begin{equation*}
\begin{aligned}
\dv  \ml \frac{r \dv \phi}{\dv r}\mr^3 &= \dv \ml \frac{\dv (r\phi)}{\dv r} - \phi \mr^3 = 3 \alpha' \frac{r^2(\dv \phi)^2}{\dv r} - \frac{3 \dv r}{r}\ml \frac{r \dv \phi}{\dv r}\mr^3,
\end{aligned}
\end{equation*}
we then obtain
\begin{equation*}
\begin{aligned}
\dv  \ml r^3\ml \frac{r \dv \phi}{\dv r}\mr^3 \mr &= 3 \alpha' r^3 \ml \frac{r^2(\dv \phi)^2}{\dv r}\mr.
\end{aligned}
\end{equation*}
Together with $r^3\ml \frac{r \dv \phi}{\dv r}\mr^3 |_{\Gamma} = 0$, integrating the above equation with respect to $v$, we then derive
\begin{equation}\label{partial phi 3}
\begin{aligned}
r^3\ml \frac{r|\dv \phi|}{\dv r} \mr^3 (u,v) 
&\leq 3B'(u)r^3(u,v)\int_u^v \frac{r^2|\dv \phi|^2}{\dv r}(u,v') \D v' \leq \frac{3}{2\pi}B'(u)r^3(u,v) \int^v_u \frac{\dv m}{1-\mu}(u,v') \D v' \\
&\leq \frac{3e^G}{2\pi(1-\mu_*)} B'(u)r^3(u,v) \int^v_u \dv m(u,v') \D v' \leq \frac{3e^G}{\pi(1-\mu_*)} B'(u)r^4(u,v) \mu(u,v).
\end{aligned}
\end{equation}
To prove the third inequality, we employ the equation
$$\du\ml \frac{1}{1-\mu(u,v)}\mr  = \frac{1}{1-\mu(u,v)} \ml -\frac{4\pi r|D_u \phi|^2}{(-\du r)} + \frac{(-\du r)}{r}\frac{\mu - Q^2/r^2}{1-\mu}\mr$$
and its implication
\begin{equation*}
\begin{aligned}
\frac{1}{1-\mu(u,v)} &= \frac{1}{1-\mu(0,v)}\exp\ml {\int^u_{0} \ml -\frac{4\pi r|D_u \phi|^2}{(-\du r)} + \frac{(-\du r)}{r}\frac{\mu - Q^2/r^2}{1-\mu}\mr (u',v) \mathrm{d}u'} \mr\\
&\leq \frac{1}{1-\mu(0,v)}\exp\ml {\int^u_{0} \ml  \frac{(-\du r)}{r}\frac{\mu - Q^2/r^2}{1-\mu}\mr (u',v) \mathrm{d}u'} \mr\leq \frac{1}{1-\mu_*}e^G\\
\end{aligned}
\end{equation*}
with $G$ defined in \eqref{GandBu introduction}.
Plugging \eqref{partial phi 3} back to \eqref{I_3' introduction}, we then obtain 
\begin{equation*}
|I_3|(u,v) \leq \frac{12 e^G}{1-\mu_*} \int_0^u \frac{- \du r}{r} \frac{\mu}{1-\mu}(u',v) B(u') \D u'.
\end{equation*}

For $|I_{10}|$, with our improved estimate for $Q$, we deduce
\begin{equation*}
\begin{split}
|I_{10}|(u,v) \lesssim& \; \e  \int^{u}_{0} (-\du r) \ml \frac{|Q|}{r^2}\mr  \ml \frac{1}{1-\mu}\mr\ml \frac{r|\dv \phi|}{\dv r}\mr (u',v) \mathrm{d}u'\lesssim\frac{\e e^G r(0,v_*)^\oh}{1-\mu_*} \int_0^u \frac{- \du r}{r^{\oh}}(u',v) B(u') \D u'.
\end{split}
\end{equation*}

Together with our treatment for all other terms from $I_1$ to $I_{10}$, in Section \ref{C1 Extension}, we arrive at  
\begin{equation*}
B(u) \leq \sup_{\{0\}\times [0, v_*]}|\alpha'| + C_{11} + \int_0^u \ml C_9\frac{- \du r}{r} \frac{\mu}{1-\mu} + C_{10} \frac{- \du r}{r^\oh}\mr(u',\tilde{v}) B(u') \D u'.
\end{equation*}
Note that the exact expressions for constants $C_9, C_{10},$ and $C_{11}$ are provided in Section \ref{C1 Extension}. 
By Gr\"onwall's inequality, we then obtain
\begin{equation}\label{ExtensionGronwall}
\begin{aligned}
B(u) &\leq \ml \sup_{\{0\}\times [0, v_*]}|\alpha'|+C_{11} \mr \exp\ml \int_0^u \ml C_9\frac{- \du r}{r} \frac{\mu}{1-\mu} + C_{10} \frac{- \du r}{r^\oh}\mr(u',\tilde{v}) \D u' \mr \\
&\leq \ml \sup_{\{0\}\times [0, v_*]}|\alpha'|+C_{11} \mr \exp\ml C_9 \gamma(u,\tilde{v}) + 2C_{10}r(0,\tilde{v})^{\oh} \mr \\
&\leq \ml \sup_{\{0\}\times [0, v_*]}|\alpha'|+C_{11} \mr \exp\ml C_9 G + 2C_{10}r(0,v_*)^{\oh} \mr<\infty.\\
\end{aligned}
\end{equation}\\

We proceed to extend the $C^1$ solution beyond the triangular boundary. We consider the rectangular strip $U_{\xi}(u_1) = [0,u_1) \times [v_*,v_*+\xi]$.
We also define
$$\overline{B'_{\xi}}(u_1) := \sup_{U_{\xi}(u_1)} (|\alpha'|)\quad \quad \mbox{ and } \quad \quad \mathcal{D}(0,u_1,v_* + \xi) := \mathcal{D}(0,u_1,v_*) \cup U_{\xi}(u_1).$$
Now, fix some $\hat{\xi}>0$. Repeating the proof outlined as above, we obtain

\noindent {\bf Lemma.} For any $u_1 \in (0,v_*)$ and $\xi \in (0,\hat{\xi}]$, if $\gamma$ satisfies
\begin{equation}
\sup_{\mathcal{D}(0,u_1,v_* + \xi)} \gamma \leq 3G,
\end{equation}
then we have
\begin{equation}
\overline{B'_{\xi}}(u_1) \leq \max\{  C_{12}(G,v_*), C_{12}(3G,v_*+\hat{\xi}) \} =: B_{2}.
\end{equation}

With the above lemma, we then define
\begin{equation*}
u_2 = \sup\left\{ u_1 \in [0,v_*): \sup_{\mathcal{D}(0,u_1,v_*+\xi)} \gamma \leq 3G\right\}.
\end{equation*}
Subsequently, we prove that for $\xi$ sufficiently small, $u_2$ has to be $v_*$ and this will give uniform bound of $|\alpha'|$ in $\mathcal{D}(0,v_*,v_* + \xi)$. To achieve this, we consider the equation 
\begin{equation*}
\dv \gamma(u,v) = \int_0^u (T_1 + T_2 + T_3)(u',v) \D u'.
\end{equation*}
with
\begin{equation*}
T_1 := \frac{-\du r}{r}\ml \frac{8 \pi \e Q \im(\phi^\dagger \dv \phi)}{1-\mu}\mr, \; T_2 := \frac{-\du r}{r}\ml \frac{4\pi r(1-Q^2/r^2)|\dv \phi|^2}{\dv r(1-\mu)}\mr, \;
T_3 := \frac{2 \du r \dv r}{r^2}\ml \frac{\mu - 2Q^2/r^2}{1-\mu}\mr. 
\end{equation*}
Employing estimates previously derived, we get 
\begin{equation*}
|\dv \gamma|(u,v) \leq C_{14}
\end{equation*}
with
\begin{equation*}
C_{14} := \ml \frac{4\pi \e e^{4G} C_8' C_{13} B_2 }{6\sqrt{2}(1-\mu_*)}(v_* + \hat \xi)^{\frac{1}{2}} + \frac{\pi e^{4G}B_2^2}{2(1-\mu_*)}\ml 1+\frac{C_{13}^2(v_* + \hat \xi)^2}{4}\mr + \frac{e^{4G}}{2(1-\mu_*)}\ml \frac{\pi B_2^2}{12} + \frac{25 C_{13}^2}{12}\mr\mr (v_* + \hat \xi).
\end{equation*}
For $(u,v) \in U_\xi(u_2)$, we then derive
\begin{equation*}
\begin{aligned}
\gamma(u,v) &\leq \gamma(u,v_*) + \int_v^{v_*}|\dv \gamma|(u,v') \D v' \leq G + C_{14}\xi \leq 2G,
\end{aligned}
\end{equation*}
if we pick $\xi := \min \left\{  \hat{\xi}, \frac{G}{C_{14}}\right\}$. This concludes the proof of Theorem \ref{SET}.\\

\vspace{5mm}

To prove Theorem \ref{FET}, with $0<\delta<1$, we set
$$A'(u_1):=\sup_{\mathcal{D}(u_0, u_1, v_*)}(r^{\delta}|\alpha'|).$$ Using $\mu(u,v)$ being small in a region $\mathcal{D}(u_0, u_1, v_*)$ close to $O'$, via detailed estimate for each $I_i$ in \eqref{alpha' Ii introduction}, after conducting extensive estimates involving $Q$ and $\phi$, we prove that $A'(u_1)\leq C_{13}$ with $C_{13}$ given in \eqref{FET58}.

For all $(u,v) \in \mathcal{D}(u_0, u_1, v_*)$, using the definition of $A'(u_1)$, we then acquire
\begin{equation}\label{dv phi intro1}
\begin{aligned}
\frac{r |\dv \phi|}{\dv r}(u,v) &\leq C_{14}r^{1-\delta}(u,v), \quad  \mbox{ and } \quad C_{14}=\frac{C_{13}}{1-\delta}.
\end{aligned}
\end{equation}
Applying \eqref{Setup25}, together with $m|_{\Gamma} = 0$, we deduce
\begin{equation}\label{FET108'}
m(u,v) \leq \int^v_u \ml \frac{2\pi r^2(1-\mu)|\dv \phi|^2}{\dv r} + \frac{Q^2 \dv r}{2r^2} \mr(u,v') \D v'.
\end{equation}
Using \eqref{dv phi intro1}, we hence obtain
\begin{equation*}
\begin{aligned}
\int^v_u \frac{2\pi r^2(1-\mu)|\dv \phi|^2}{\dv r}(u,v') \D v'&\leq 2\pi \int^v_u \frac{r^2|\dv \phi|^2}{(\dv r)^2}(\dv r)(u,v') \D v'\\
&\leq 2\pi \int^v_u C_{14}^2 r^{2-2\delta}(\dv r)(u,v') \D v' = \frac{2\pi C_{14}^2}{3-2\delta}r^{3-2\delta}(u,v).
\end{aligned}
\end{equation*}
At the same time, when $0<\chi<\delta<1$, the term involving $Q$ in \eqref{FET108'} also obeys
\begin{equation*}
\begin{aligned}
\int^v_u \frac{Q^2 \dv r}{2r^2}(u,v') \D v' &\leq \frac{C_{6}^2}{2}\int^v_u r^{1-\chi}(\dv r)(u,v') \D v' = \frac{C_{6}^2}{2(2-\chi)}r^{2-\chi}(u,v)
\end{aligned}
\end{equation*}
with $C_6$ given in \eqref{FET130}.  If we further restrict our choice of $\delta$ to $\left( 0, \frac{1}{2} \right]$, we then get
\begin{equation*}
m(u,v) \leq \frac{C_{15}}{2} r^{2-\chi}(u,v) \mbox{ with } C_{15}(\delta,\chi) := \frac{4\pi C_{14}^2}{3-2\delta}r^{1-2\delta+\chi}(0,v_*) + \frac{C_6^2}{(2-\chi)}.
\end{equation*}
Consequently, this implies the following improved estimate for $\mu$:
\begin{equation*}
\mu(u,v) \leq C_{15}r^{1-\chi}(u,v).
\end{equation*}

With these enhanced bounds, we then employ the estimates for $I_i$ in 
\begin{equation*}
|\alpha'(u,v)| \leq \sup_{\{u_0\} \times [u_0,v_*]} (|\alpha'|) + \sum_{i=1}^{11}|I_i|(u,u_0,v)
\end{equation*}
and obtain a uniform bound for $|\alpha'(u,v)|$. This in turn implies that the solution extends to $\mathcal{D}(0,v_*).$ \\

In a similar fashion, we can extend the $C^1$ solution beyond the triangular boundary as follows. Consider a thin rectangle region
\begin{equation*}
U_{\xi}(u_1) = [u_0,u_1) \times [v_*,v_*+\xi].
\end{equation*}
We then derive uniform bounds for quantities
$$\overline{A'_{\xi}}(u_1) := \sup_{U_{\xi}(u_1)} (r^{\delta}|\alpha'|) \quad
\mbox{ and }
\quad A'_* := \sup_{u_1\in(u_0,v_*)}A'(u_1).$$
With $0<\delta<1$, we arrive at 
$$\sup_{\mathcal{D}(u_0,u_1,v_*+\xi)}(r^{\delta}|\alpha'|) \leq \max\{ A'(u_1), \overline{A'_\xi}(u_1)\} \leq \max\{A'_*, \overline{A'_\xi}(u_1)\}.$$
By replacing $A'(u_1)$ with $\max\{A'_*, \overline{A'_\xi}(u_1)\}$, we also obtain an analogous bound for $\overline{A'_{\xi}}(u_1)$. As explained before, we then employ improved estimates for $m(u,v)$ and $\mu(u,v)$, leading to the derivation of the uniform bound for $|\alpha'(u,v)|$ in the rectangular region $U_{\xi}(u_1)$. The desired conclusion then follows.

\subsubsection{First Instability Theorem}
In this paper, to explore the blueshift effect, we define the quantities:
\begin{equation}
\begin{split}
\gamma(u,v) :=& \int^u_{0} \frac{(- \du r)}{r} \frac{\mu - Q^2/r^2}{1 - \mu} (u',v) \mathrm{d}u',\quad \quad \gamma_0(u) := \int^u_{0} \frac{(- \du r)}{r} \frac{\mu - Q^2/r^2}{1 - \mu} (u',v_0) \mathrm{d}u',\\
\gamma_s(u,v) :=& \int^u_{0} \frac{(- \du r)}{r} \frac{\mu}{1 - \mu} (u',v) \mathrm{d}u', \quad \quad \gamma_{s, 0}(u) :=\int^u_{0} \frac{(- \du r)}{r} \frac{\mu}{1 - \mu} (u',v_0) \mathrm{d}u'.
\end{split}
\end{equation}
Note that $\gamma_s$ was introduced by Christodoulou in \cite{christ4}. For our use, with $\overline{u} \in [0,v_0)$ we also set
\begin{equation}
\gamma(u,v;\overline{u}) = \tgam(u,v) := \int^u_{\overline{u}} \frac{(- \du r)}{r} \frac{\mu - Q^2/r^2}{1 - \mu} (u',v) \mathrm{d}u'.
\end{equation}
The following two properties reveal that the quantity $\gamma$ is the source of the blueshift effect:
\begin{itemize}
    \item \mbox{Property 1} $$(\dv r) (u,v) = (\dv r)(0,v) e^{-\gamma(u,v)},$$
    \item \mbox{Property 2} $$ \frac{1}{1-\mu(u,v)} \leq \frac{1}{1-\mu(0,v)}e^{\gamma(u,v)}.$$
\end{itemize}
Hence, if $\gamma(u,v)$ is large, then $\partial_v r(u,v)$ would be close to $0$. 

The first property above is obtained from integrating
\begin{equation*}
\du \log \ml \dv r\mr = \frac{\du r}{r}\ml \frac{\mu-Q^2/r^2}{1-\mu}\mr.
\end{equation*}
The second property is a result of  considering
\begin{equation*}
\begin{aligned}
\du\ml \frac{1}{1-\mu(u,v)}\mr=& \frac{1}{1-\mu(u,v)} \ml -\frac{4\pi r|D_u \phi|^2}{(-\du r)} + \frac{(-\du r)}{r}\frac{\mu - Q^2/r^2}{1-\mu}\mr,
\end{aligned}
\end{equation*}
and
\begin{equation*}
\begin{aligned}
\frac{1}{1-\mu(u,v)} =& \frac{1}{1-\mu(0,v)}\exp\ml {\int^u_{0} \ml -\frac{4\pi r|D_u \phi|^2}{(-\du r)} + \frac{(-\du r)}{r}\frac{\mu - Q^2/r^2}{1-\mu}\mr (u',v) \mathrm{d}u'} \mr\\
\leq& \frac{1}{1-\mu(0,v)}\exp\ml {\int^u_{0} \ml  \frac{(-\du r)}{r}\frac{\mu - Q^2/r^2}{1-\mu}\mr (u',v) \mathrm{d}u'} \mr.
\end{aligned}
\end{equation*}

In the subsequent portions of that section, we will select a sequence of $u_n \in [u_2,v_0)$ and determine corresponding values of $v(u_n) \in [v_0,v_0 + \te]$, where $0<\te<1$. This selection ensures that for all $u \in [u_2,u_n]$, the following inequality holds:
\begin{equation*}
\delta(u,v(u_n)) \log \ml \frac{1}{\delta(u,v(u_n))}\mr \leq C_{33}e^{-2\gamma_0(u)}
\end{equation*} with $C_{33}$ given in \eqref{FIT44} and $\delta(u,v(u_n)) \leq 1/2$. Then, within the considered regions, we proceed to obtain 
\begin{equation*}
\begin{aligned}
\frac{1}{1-\mu(u,v)} &\leq \frac{8}{7}\frac{1}{1-\mu(u, v_0)}.
\end{aligned}
\end{equation*}

We further state an important property for the function $\delta(u,v)$ in double-null coordinates.

\noindent \textbf{Lemma.}  For a sequence of $u_n$ satisfying $u_n \in [u_2,v_0)$, suppose that there exist corresponding values of $v(u_n) \in [v_0,v_0 + \te]$ such that, for all $u \in [u_2,u_n]$, the following inequalities hold: \begin{equation}\label{FIT35c'}
\delta(u,v(u_n)) \leq \min\left\{ \frac{1}{2},\frac{2(1 - \mu_0(0))e^{-\gamma_0(u)}}{3}, \frac{e^{-\frac{3}{2}(1+2\xi)\gamma_0(u)}}{2 C_{34}} \right\}
\end{equation} 
and
\begin{equation}\label{FIT35c''}
\delta(u,v(u_n))\log \ml \frac{1}{\delta(u,v(u_n))} \mr \leq C_{33}e^{-2\gamma_0(u)}
\end{equation}
with $C_{34}$, $C_{33}$ given in \eqref{FIT71}, \eqref{FIT44} and $\xi\in(0, \frac12)$. Then, for a fixed $v \in (v_0,v(u_n)]$ and for any $u_3,u_4 \in [u_2,u_n]$ with $u_3 < u_4$, we have
\begin{equation}\label{FIT83'}
\delta(u_3,v) \leq 4 e^{-\frac{1}{2}\log\ml \frac{r_0(u_3)}{r_0(u_4)}\mr + \frac{3}{2}(1+2\xi)\gamma_{0}(u_4)}\delta(u_4,v).
\end{equation}

\begin{remark}
The conclusion in \eqref{FIT83'} means, along the null curve with the same $v$, the outer value of $\delta(u,v)$ can be controlled by its inner value, since $(u_4, v)$ is closer to the center compared with $(u_3, v)$. In double-null coordinates, this is new, constituting a non-trivial feature of the geometry that is crucially used in later arguments.
\end{remark}

To prove this lemma, we first the following employ inequality:
\begin{equation*}
    \frac{1}{2}\delta(u,v) \leq \log \ml \frac{r(u,v)}{r_0(u)}\mr \leq \delta(u,v)
\end{equation*}
with $u \in [u_3,u_4]$, $v \in (v_0,v(u_n)]$ and $\delta(u, v(u_n))\leq 1/2$. We then apply
\begin{equation*}
\begin{aligned}
\dv \ml \log \ml \frac{-\du r}{r}\mr\mr =& \; \frac{\dv r}{r}\ml \frac{\mu - Q^2/r^2}{1-\mu} - 1\mr\\
\leq& \;\frac{\dv r}{r} \ml \frac{\mu}{1-\mu} - 1\mr
= \frac{\dv r}{r} \ml \frac{1}{1-\mu} - 2\mr \leq \frac{\dv r}{r}\ml \frac{3}{2}\frac{1}{1-\mu_0}-2\mr.
\end{aligned}
\end{equation*}
Integrating from $v_0$ to $v$ implies that
\begin{equation*}
\begin{aligned}
    \log \ml \frac{-\frac{\du r}{r}(u,v)}{\ml -\frac{\du r}{r} \mr_0(u)}\mr &\leq \ml \frac{3}{2}\frac{1}{1-\mu_0(u)} - 2\mr \log \ml \frac{r(u,v)}{r_0(u)}\mr, 
\end{aligned}
\end{equation*}
and
\begin{equation*}
\begin{aligned}
    \ml - \frac{\du r}{r}\mr(u,v) &\leq \ml -\frac{\du r}{r}\mr_0(u) e^{\ml \frac{3}{2}\frac{1}{1-\mu_0(u)} - 2\mr \log \ml \frac{r(u,v)}{r_0(u)}\mr}.
\end{aligned}
\end{equation*}
This further yields
\begin{equation*}
\ml \frac{- \du r}{r}\mr (u,v) - \ml \frac{- \du r}{r}\mr_0(u) \leq \ml e^{\ml \frac{3}{2}\frac{1}{1-\mu_0(u)} - 2\mr \log \ml \frac{r(u,v)}{r_0(u)}\mr} - 1 \mr \ml \frac{-\du r}{r}\mr_0.
\end{equation*}
With these, we derive
\begin{equation*}
\begin{aligned}
-\du \log\ml \frac{r(u,v)}{r_0(u)}\mr  &= -\frac{\frac{\du r}{r_0}(u,v) - \frac{r(u,v)}{r_0(u)^2}(\du r)_0(u)}{\frac{r(u,v)}{r_0(u)}} =\frac{- \du r}{r}(u,v) - \ml\frac{- \du r}{r}\mr_0(u) \\
&\leq \ml e^{\ml \frac{3}{2}\frac{1}{1-\mu_0(u)} - 2\mr \log \ml \frac{r(u,v)}{r_0(u)}\mr} - 1 \mr \ml \frac{-\du r}{r}\mr_0.
\end{aligned}
\end{equation*}
Inspired by \cite{christ4}, we utilize the inequality $e^x - 1 \leq x + (e-2)x^2$ for $ x \leq 1$. Denoting
\begin{equation*}
    \td(u,v) := \log \ml \frac{r(u,v)}{r_0(u)}\mr,
\end{equation*}
for all $(u,v) \in [u_2,u_n] \times (v_0,v(u_n)]$, we then obtain
\begin{equation*}
\begin{aligned}
-\du \td &\leq \ml \ml \frac{3}{2(1-\mu_0(u))} - 2 \mr \td + (e-2)\ml \frac{3}{2(1-\mu_0(u))} - 2 \mr^2\td^2 \mr \ml \frac{-\du r}{r}\mr_0
\end{aligned}
\end{equation*}
and
\begin{equation*}
\begin{aligned}
\du \ml \frac{1}{\td}\mr &\leq \ml \ml \frac{3}{2(1-\mu_0(u))} - 2 \mr\ml \frac{1}{\td}\mr + (e-2)\ml \frac{3}{2(1-\mu_0(u))} - 2 \mr^2 \mr\ml \frac{-\du r}{r}\mr_0.
\end{aligned}
\end{equation*}
By employing a Gr\"onwall-like argument, we arrive at
\begin{equation*}
\begin{aligned}
\frac{e^{- \frac{3}{2}\gamma_{s,0}(u_4) + \frac{1}{2}\log\ml\frac{r_0(0)}{r_0(u_4)}\mr}}{\td(u_4,v)} - \frac{e^{- \frac{3}{2}\gamma_{s,0}(u_3)+ \frac{1}{2}\log\ml\frac{r_0(0)}{r_0(u_3)}\mr}}{\td(u_3,v)} &\leq \frac{e^{-\frac{3}{2}\gamma_{s,0}(u_4) + \frac{1}{2}\log\ml\frac{r_0(0)}{r_0(u_4)}\mr}}{2\td(u_4,v)},
\end{aligned}
\end{equation*}
which implies the above-stated lemma. \\

With the above preparations, we now delve into the physical arguments as outlined in \cite{christ4}. Employing
\begin{equation*}
\begin{aligned}
\du \ml \frac{r \dv \phi}{\dv r} \mr
= \ml -D_u \phi - \ii \e \frac{Q \phi (\du r)}{r(1-\mu)}\mr + \ml \frac{r \dv \phi}{\dv r}\mr \ml -\ii \e A_u + \frac{(-\du r)}{r} \frac{\mu - Q^2/r^2}{1 - \mu} \mr,
\end{aligned}
\end{equation*}
we deduce that
\begin{equation*}
\ml \frac{r |\dv \phi|}{|\dv r|} \mr(u,v)  = \mlm \ml \frac{r \dv \phi}{\dv r} \mr(0,v) - I(u)\mrm e^{\gamma(u,v)}
\end{equation*}
with 
\begin{equation*}
I(u)=\int^u_{0} g(u',v)e^{-f(u',v)}\mathrm{d}u' 
\end{equation*}
and
\begin{equation*}
\begin{split}
g(u,v) =& D_u \phi(u,v) +  \ii \e \frac{Q \phi (\du r)}{r(1-\mu)}(u,v),\\ 
f(u,v) =& \int^u_{0} -\ii \e A_u(u',v) + \frac{(-\du r)}{r} \frac{\mu - Q^2/r^2}{1 - \mu}(u',v) \mathrm{d}u'.
\end{split}
\end{equation*}

Our to-be-established first instability theorem considers the following case:
\begin{equation}\label{first instability condition}
\ml \frac{r \dv \phi}{\dv r} \mr_0 (0):=\ml \frac{r \dv \phi}{\dv r} \mr(0,0) \neq \lim_{u \rightarrow v_0^-} I(u).
\end{equation}

\noindent For this scenario, as $u\rightarrow v_0^-$, there are various cases for the behaviors of the complex-valued $I(u)$. Since $I(u)$ is complex-valued, we rewrite it as 
\begin{equation*}
        I(u) = X(u) + \ii Y(u)
\end{equation*} 
with $X(u) = \re{(I(u))}$ and $Y(u) = \im{(I(u))}$. Let us denote the following quantities:
\begin{equation}
\begin{aligned}
    X_- := \liminf_{u \rightarrow v_0^-} X(u), \qquad & \qquad X_+ := \limsup_{u \rightarrow v_0^-} X(u), \\ 
    Y_- := \liminf_{u \rightarrow v_0^-} Y(u), \qquad & \qquad Y_+ := \limsup_{u \rightarrow v_0^-} Y(u). \\ 
\end{aligned}
\end{equation}
\textit{A priori}, as the asymptotic property of $I$ is unknown, we have to consider all possible behaviors of $X$ and $Y$ below. For notational convenience, by using $Z$ as a placeholder for either $X$ or $Y$, we characterize their behavior as follows:
\begin{enumerate}[(1)]
    \item $-\infty < Z_- = Z_+ < \infty$ (limit exists and it is finite),
    \item $-\infty = Z_- = Z_+$ (limit exists and it is equals to $- \infty$),
    \item $Z_- = Z_+ = \infty$ (limit exists and it is equals to $+ \infty$),
    \item $-\infty < Z_- < Z_+ < \infty$ (limit does not exist but $Z_-$ and $Z_+$ are bounded),
    \item $-\infty < Z_- < Z_+ = \infty$ (limit does not exist),
    \item $-\infty = Z_- < Z_+ < \infty$ (limit does not exist),  
    \item $-\infty = Z_- < Z_+ = \infty$ (limit does not exist).
\end{enumerate}

With respect to $X$ and $Y$, there are in total $7\times 7=49$ cases, and we denote them as ordered integer pairs $(a,b)$ with $1\leq a, b \leq 7$. For simplicity of demonstration, we focus on the $(1,1)$ case here. Since both limits for $X$ and $Y$ exist, the limit for the complex-valued $I$ does exist. By the hypothesis stated in \eqref{first instability condition}, we set
\begin{equation}
    h := \mlm \ml \frac{r\dv \phi}{\dv r}\mr_0(0) - l\mrm \neq 0 \quad \mbox{ with } \quad l = \lim_{u \rightarrow v_0^-} I(u).
\end{equation}
For $u_*$ sufficiently close to $v_0$ and for all $u \in [u_*,v_0)$, we then have
\begin{equation*}
    |I(u)-l| \leq \frac{2h}{3}.
\end{equation*}
This implies that
\begin{equation*}
\begin{aligned}
\ml \frac{r |\dv \phi|}{\dv r} \mr_0(u) 
&= \mlm \ml\frac{r \dv \phi}{\dv r} \mr_0(0) - l + l - I(u) \mrm e^{\gamma_0(u)} \geq \mlm \mlm \ml\frac{r \dv \phi}{\dv r} \mr_0(0) - l\mrm - |l - I(u)|  \mrm e^{\gamma_0(u)} \geq \frac{h}{3}e^{\gamma_0(u)}.
\end{aligned}
\end{equation*}

Let us define the following quantities:
\begin{equation*}
    \Upxi(u,v) := \frac{r(u,v)}{r_0(u)}\ml \frac{r \dv \phi}{\dv r}\mr(u,v) - \ml \frac{r\dv \phi}{\dv r
    }\mr_0 (u) 
\end{equation*}
and
\begin{equation*}
    \psi(u,v) := e^{-\gamma_0(u)}\Upxi(u,v).
\end{equation*}
We first assume that, with $n \in \mathbb{N}$ being sufficiently large, we can work under the condition
\begin{equation}\label{first instability condition 2}
\sup_{v' \in [v_0,v(u_n)]}|\psi(u_n,v')| \leq \frac{h}{6}.
\end{equation}
Under this assumption, for each $v' \in [v_0,v(u_n)]$, we have
\begin{equation*}
\begin{aligned}
\frac{r(u_n,v')}{r_0(u_n)}\frac{r|\dv \phi|}{\dv r}(u_n,v') &= \mlm \frac{r\dv \phi}{\dv r}(u_n,v')\ml \frac{r(u_n,v')}{r_0(u_n)}\mr - \ml \frac{r\dv \phi}{\dv r}\mr_0(u_n) + \ml \frac{r \dv \phi}{\dv r}\mr_0(u_n) \mrm  \\
&\geq \mlm \mlm \frac{r\dv \phi}{\dv r}(u_n,v')\ml \frac{r(u_n,v')}{r_0(u_n)}\mr - \ml \frac{r\dv \phi}{\dv r}\mr_0(u_n) \mrm - \mlm \ml \frac{r \dv \phi}{\dv r}\mr_0(u_n)\mrm\mrm \\
&= \mlm \mlm \ml \frac{r \dv \phi}{\dv r}\mr_0(u_n)\mrm - |\psi|(u_n,v')e^{\gamma_0(u_n)} \mrm \\
&\geq \mlm \frac{h}{3} - \frac{h}{6}\mrm e^{\gamma_0(u_n)}= \frac{h}{6}e^{\gamma_0(u_n)}.
\end{aligned}
\end{equation*}

Recalling the definition of $\eta^*(u,v)$, this implies:
\begin{equation*}
\begin{aligned}
\eta^*(u_n,v) &= \frac{2\int^v_{v_0} \dv m(u_n,v') \mathrm{d}v'}{r(u_n,v)} - \frac{2}{r(u_n,v)} \int^v_{v_0} \frac{Q^2(\dv r)}{2r^2} (u_n,v') \mathrm{d}v' \\
&=  \frac{2}{r(u_n,v)}\ml \int^v_{v_0} \frac{2\pi r^2(1-\mu)|\dv \phi|^2}{\dv r}(u_n,v') \mathrm{d}v' \mr  \\
&\geq \frac{2\pi(1-\mu_0(0))h^2 }{27r(u_n,v)}e^{\gamma_0(u_n)}r^2_0(u_n)\ml \int^v_{v_0} \frac{(\dv r)(u_n,v')}{r^2(u_n,v')} \mathrm{d}v' \mr  \\
&= \frac{2\pi(1-\mu_0(0))h^2 }{27}e^{\gamma_0(u_n)}\ml \frac{1}{1+\delta(u_n,v)}\mr^2\ml \frac{r(u_n,v) - r_0(u_n)}{r_0(u_n)}\mr  \\
&\geq C_{37} \delta(u_n,v)e^{\gamma_0(u_n)} \quad \mbox{ with } \quad C_{37} := \frac{8\pi(1-\mu_0(0))h^2 }{243}.  \\
\end{aligned}
\end{equation*}

On the other hand, assuming that there is no trapped surface formation, from \eqref{TSF2'} (and its more detailed version \eqref{FIT34}), we obtain
\begin{equation*}
\begin{aligned}
\eta^*(u_n,v) &\leq C_{24} \frac{\delta(u_n,v)}{1+\delta(u_n,v)}\log\ml \frac{1}{\delta(u_n,v)}\mr + C_{23}\frac{\delta(u_n,v)r_0(u_n)}{1-\sup_{{v \in [v_0,v_0+\varepsilon]}}\mu(u_n,v)} \\
&\leq C_{24} \delta(u_n,v) \log\ml \frac{1}{\delta(u_n,v)}\mr + \frac{3 C_{23}}{2(1-\mu_0(0))} \delta(u_n,v) e^{\gamma_0(u_n)}r_0(u_n) \\
&\leq C_{36}\delta(u_n,v) \ml \log \ml \frac{1}{\delta(u_n,v)}\mr + e^{\gamma_0(u_n)} r_0(u_n) \mr
\end{aligned}
\end{equation*}
with 
\begin{equation*}
C_{36} = C_{24} + \frac{3 C_{23}}{2(1-\mu_0(0))}.
\end{equation*}

By comparing the above two inequalities for $\eta^*(u_n,v)$ at $v = v(u_n)$, a contradiction arises if
\begin{equation}\label{first instability condition 3}
\log \ml \frac{1}{\delta(u_n,v(u_n))}\mr e^{-\gamma_0(u_n)} + r_0(u_n)  < \frac{C_{37}}{C_{36}}.
\end{equation}
This means that if \eqref{first instability condition 3} is true, the assumption of no trapped surface is violated, and the formation of trapped surface is guaranteed.

\begin{remark}
It is crucial to have the additional $r_0(u_n)$ as in \eqref{TSF2'} in the expression of the trapped surface formation criterion. By picking $n$ sufficiently large, we can set
\begin{equation*}
r_0(u_n) \leq \frac{C_{37}}{4C_{36}}.
\end{equation*}
At the same time, for a sufficiently large $n$, we choose $v(u_n)$ such that
\begin{equation*}
\delta(u_n,v(u_n)) = e^{-\frac{C_{37}}{4C_{36}}e^{\gamma_0(u_n)}},
\end{equation*}
hence verifying \eqref{first instability condition 3},  ultimately concluding the proof of the first instability theorem. If we cannot obtain $r_0(u_n)$ in \eqref{TSF2'} and relax the factor to $1$, we would enter the realm of comparing constant with both $C_{37}$ and $C_{36}$ being of size 1. In this case, the validation of \eqref{first instability condition 3} remains open, and we would fail to provide the integrated proof as stated above.
\end{remark}

Recall the assumption \eqref{first instability condition 2} we imposed. The next lemma verifies it.

\noindent \textbf{Lemma.} Suppose that on top of having \eqref{FIT35c'} and \eqref{FIT35c''}, we demand that
\begin{equation*}
\delta(u,v(u_n)) \leq \min \left\{ \ml \frac{\log(3)}{8C_{48}}\mr^2e^{-(9+6\xi)\gamma_0(u)}, \ml \frac{h}{1296C_{45}}\mr^2 e^{-\ml \frac{15}{2} + 3\xi\mr\gamma_0(u)}, \ml \frac{h}{1296 b C_{49}}\mr^2 e^{- \frac{19}{2}\gamma_0(u)} \right\}
\end{equation*}
also holds for all $u \in [u_2,u_n]$. Then, the following requirement for initial data:
\begin{equation}\label{first instability condition 4}
\sup_{v \in [v_0,v(u_n)]}|\psi|(u_2,v) \leq \frac{h}{54}
\end{equation}
implies that 
\begin{equation*}
\sup_{v' \in [v_0,v(u_n)]}|\psi|(u_n,v') \leq \frac{h}{6}.
\end{equation*}

To prove this claim, we first compute
\begin{equation}\label{first instability condition 5}
\begin{aligned}
\du \psi &= e^{-\gamma_0} J_4 + \psi(J_5 - \ii \e A_u) + e^{-\gamma_0}\ml \frac{r \dv \phi}{\dv r}\mr_0 J_6. 
\end{aligned}
\end{equation}
Here, we have 
\begin{equation*}
\begin{aligned}
J_4 &=  \frac{r}{r_0}\ml - D_u \phi - \ii \e \frac{Q \phi (\du r)}{r(1-\mu)}\mr + \ml D_u \phi + \ii \e \frac{Q \phi (\du r)}{r(1-\mu)}\mr_0 , \\
J_5 &=  \du \gamma - \du \gamma_0 + \frac{\du r}{r} - \ml \frac{\du r}{r}\mr_0,  \\
J_6 &= J_5 - \ii \e (A_u - (A_u)_0).
\end{aligned}
\end{equation*}

To control $J_4$, we employ the following equation:
\begin{equation*}
\begin{aligned}
& \dv \ml \frac{r}{r_0} \ml D_u \phi + \ii \e \frac{Q \phi (\du r)}{r(1-\mu)} \mr \mr \\
=\; & \frac{1}{r_0}\ml (-\du r) \dv \phi + \ii \e \frac{Q (\du r)(\dv r) \phi}{r(1-\mu)} +  (4 \pi \ii \e^2)\frac{\phi(-\du r)r^2 \im{(\phi^{\dagger} \dv \phi)}}{1-\mu} +\ii \e \frac{Q (\dv \phi)(\du r)}{1-\mu} +  4 \pi \ii \e \frac{Q r |\dv \phi|^2(\du r) \phi}{\dv r(1-\mu)}\mr.
\end{aligned}
\end{equation*}
In the subsequent detailed proofs, we bound each term on the right. In one of the upcoming steps, we employ the following estimate:
\begin{equation*}
\begin{aligned}
\int^v_{v_0} r^2|\phi|^4 \frac{\dv r}{1-\mu}(u,v') \D v' \leq \; & 8C_{35}^4\int^v_{v_0} r^2 \frac{\dv r}{1-\mu}(u,v') \D v' \\
&+  \frac{1}{2\pi^2}\int^v_{v_0} r^2 \frac{\dv r}{1-\mu}(u,v') \ml  \log \ml \frac{\frac{1 - \mu}{ \dv r}(u,v)}{\frac{1 - \mu}{ \dv r}(u_2,v)}\mr  \mr^{2}\ml \log \ml \frac{r(u_2,v)}{r(u,v)}\mr \mr^{2} \D v'. \\
\end{aligned}
\end{equation*}
We then utilize $(\frac{\dv r}{1-\mu})^{\frac12}$ to balance the term in the logarithm and relate the other $(\frac{\dv r}{1-\mu})^{\frac12}$ to $e^{\frac{1}{2}\gamma_0(u)}$. Henceforth, we obtain
\begin{equation*}
\begin{aligned}
& \int^v_{v_0} r^2 \frac{\dv r}{1-\mu}(u,v') \ml  \log \ml \frac{\frac{1 - \mu}{ \dv r}(u,v)}{\frac{1 - \mu}{ \dv r}(u_2,v)}\mr  \mr^{2}\ml \log \ml \frac{r(u_2,v)}{r(u,v)}\mr \mr^{2} \D v' \\
\lesssim \; & \int^v_{v_0} \ml \frac{\dv r}{1-\mu}\mr^{\frac{1}{2}}(u,v')  \D v' \lesssim \;   e^{\frac{1}{2}\gamma_0(u)} \ml \int^v_{v_0} \dv r(u,v')  \D v' \mr^{\frac{1}{2}} \lesssim \;  e^{\frac{1}{2}\gamma_0(u)} \delta(u,v)^{\frac{1}{2}}.
\end{aligned}
\end{equation*}
This leads to 
\begin{equation*}
\int^v_{v_0} r^2|\phi|^4 \frac{\dv r}{1-\mu}(u,v') \D v' \leq C_{41}e^{\frac{1}{2}\gamma_0(u)}\delta(u,v)^{\frac{1}{2}} \quad \mbox{ with } \quad  C_{41} \mbox{ given in } \eqref{FIT247}.
\end{equation*}

To bound the contribution from $J_5$, we first rewrite $J_5$ as
\begin{equation*}
J_5(u,v) = \int^v_{v_0} \dv \ml \frac{\du r}{r} + \du \gamma \mr(u,v') \, \D v'.
\end{equation*}
Furthermore, we can show that
\begin{equation*}
\begin{aligned}
|J_5|(u,v) &\leq C_{48} \ml \frac{-\du r}{r}\mr_0(u) e^{3 \gamma_0(u)} \delta(u,v)^{\frac{1}{2}}
\end{aligned}
\end{equation*}
with $C_{48}$ given in \eqref{FIT195}. This further implies
\begin{equation*}
\begin{aligned}
\int^u_{u_2} |J_5|(u',v) \D v' &\leq C_{48} \int^u_{u_2} \ml \frac{-\du r}{r}\mr_0(u') e^{3\gamma_0(u')}\delta(u',v)^{\frac{1}{2}}\D u' \\
&\leq 2C_{48} e^{\ml \frac{15}{4} + \frac{3}{2}\xi \mr\gamma_0(u)}\delta(u,v)^{\frac{1}{2}} \int^u_{u_2} \ml \frac{-\du r}{r}\mr_0(u') e^{-\frac{1}{4}\log\ml \frac{r_0(u')}{r_0(u)}\mr} \D u' \\
&\leq 8C_{48} e^{\ml \frac{15}{4} + \frac{3}{2}\xi \mr\gamma_0(u)}\delta(u,v)^{\frac{1}{2}}.
\end{aligned}
\end{equation*}
Note that the above lemma for $\delta(u,v)$ is critically employed in the second inequality.

Back to the equation \eqref{first instability condition 5} for $\du \psi$, for each $v \in [v_0,v(u_n)]$, we then arrive at 
\begin{equation*}
|\psi|(u_n,v) \leq 3 \ml \frac{h}{54} \times 3\mr = \frac{h}{6}.
\end{equation*}

We are left with the task of verifying all the assumptions for $\delta$, as summarized and listed in \eqref{FIT139}, \eqref{FIT140}, \eqref{FIT257}, \eqref{FIT240}, \eqref{FIT256}, along with the requirement on the initial data in \eqref{first instability condition 4}. Our complete proof is presented in a later section. For instance, the verification of \eqref{FIT139} for $\delta$ can be carried out as follows:
\begin{equation*}
\begin{aligned}
\delta(u_*,v(u_n))\log \ml \frac{1}{\delta(u_*,v(u_n))}\mr &\leq \frac{2}{e}\delta(u_*,v(u_n))^{\frac{1}{2}}
\\
&\leq \frac{4}{e}\delta(u_n,v(u_n))^{\frac{1}{2}} e^{-\frac{1}{4}\log \ml \frac{r_0(u_*)}{r_0(u_n)}\mr + \frac{3}{4}\ml 1 + 2\xi \mr\gamma_0(u_n)} \\
&= 4e^{-\frac{C_{37}}{8C_{36}}e^{-\gamma_0(u_n)} + \frac{3}{4}\ml 1 + 2\xi \mr\gamma_0(u_n)-1}.
\end{aligned}
\end{equation*}

\noindent By continuity of $\delta(u_2,\cdot)$ and \eqref{FIT256}, we can further verify the requirement \eqref{first instability condition 4}: 
\begin{equation*}
\sup_{v \in [v_0,v(u_n)]}|\psi|(u_2,v) \leq \frac{h}{54}.
\end{equation*}
The above outlines the proof for the case $(1,1)$. Analogous arguments are conducted for case $(1,4), (4,1),$ and $(4,4)$.

For the remaining 45 cases, since $\limsup_{u \rightarrow v_0^-}|I(u)| = +\infty$, we define $b:[0,v_0) \rightarrow \mathbb{R}$ via
$$b(u) := \sup_{u'\in [0,u]} \ml \frac{r|\dv \phi|}{\dv r}\mr_0(u')e^{-\gamma_0(u')}$$
and deduce that there exists an increasing sequence $u_n \rightarrow v_0^-$ such that
$$b(u_n) =\ml \frac{r|\dv \phi|}{\dv r}\mr_0(u')e^{-\gamma_0(u')}$$
with $b(u_n) \rightarrow \infty$. Repeating the above arguments, similar to \eqref{first instability condition 3}, we can deduce that a contradiction arises if 
\begin{equation}
\log\ml \frac{1}{\delta(u_n,v(u_n))}\mr e^{-\gamma_0(u_n)} + r_0(u_n) < \frac{C_{50}}{C_{36}}b(u_n)^2.
\end{equation}
In these cases, we instead pick $v(u_n)$ such that
$$\delta(u_n,v(u_n)) = e^{-b(u_n)^2\frac{C_{50}}{4C_{36}}e^{\gamma_0(u_n)}}.$$
Correspondingly, some of the required assumptions would have additional linear factors in $b(u_n)$ in their exponents. Henceforth, despite the stronger requirements, these are met through picking a value of $v(u_n)$ that is even closer to $v_0$ as described above.

\subsubsection{Second Instability Theorem}

We now investigate the scenario where 
\begin{equation*}
\ml \frac{r \dv \phi}{\dv r} \mr_0 (0) = \lim_{u \rightarrow v_0^-} I(u).
\end{equation*}
For the Einstein-scalar field system, in \cite{christ4}, Christodoulou remarked that the proof for the first instability theorem still holds (and, therefore, a similar conclusion is valid) if there exists a positive constant $p > 1$ such that
\begin{equation}\label{SIT2'}
\limsup_{u \rightarrow v_0^-}\left\{ \mlm \ml \frac{r \dv \phi}{\dv r}\mr_0(0) - I(u)\mrm e^{\frac{1}{2}p\gamma_0(u)} \right\} \neq 0.
\end{equation}
For the uncharged situation, assumption \eqref{SIT2'} can be satisfied. However, for the charged case, it is not clear whether \eqref{SIT2'} can still be imposed and verified. Here we take a different path and avoid requiring \eqref{SIT2'}.

We first define
\begin{equation*}
V(u) := \mlm \ml \frac{r \dv \phi}{\dv r}\mr_0(0) - I(u)\mrm, \quad v(u) := v_0 + e^{-14\gamma_0(u) - \log \ml \frac{r_0(0)}{r_0(u)}\mr}.
\end{equation*}
and
\begin{equation*}
g(v(u)) := \sqrt{\max\left\{V^2(u),e^{-\frac{1}{2}\gamma_0(u)},e^{-\log\ml\frac{r_0(0)}{r_0(u)}\mr} \right\}}.
\end{equation*} 
We also set
\begin{equation*}
h^2(v(u)) := \frac{1}{v(u) - v_0}\int_{v_0}^{v(u)} |\psi|^2(0,v') \D v'.
\end{equation*}
Subsequently, we will observe that the second instability theorem still holds if
\begin{equation*}
\limsup_{u \rightarrow v_0^-}\frac{h(v(u))}{g(v(u))} = \infty,
\end{equation*}
indicating that an additional instability mechanism has been triggered. In particular, we allow $$\lim_{u\rightarrow v_0^-}h(v(u))=0.$$ 

We first state a useful lemma below.

\noindent \textbf{Lemma.} Let $\tu \in [u_2,v_0)$ be given. Suppose there exists a corresponding value $v(\tu) \in [v_0,v_0 + \te]$ with $0<\te<1$ such that, for all $u \in [u_2,\tu]$, we have
\begin{equation*}
 \delta(u,v(\tu)) \leq \frac{e^{-\ml \frac{15}{2} + 3 \xi \mr\gamma_0(u)}}{64C_{53}^2}
\end{equation*}
with $C_{53}$ defined in \eqref{SIT87}.
Then, for each $u \in [u_2,\tu]$ and $v' \in [v_0,v(\tu)]$, we have
\begin{equation*}
|\tgam(u,v) - \tgam_0(u)| \leq 1,
\end{equation*}
where 
$$\tgam(u,v):=\gamma(u,v;\bar{u})=\int_{\bar{u}}^u\frac{(-\partial_u r)}{r}\frac{\mu-Q^2/r^2}{1-\mu}(u', v)du' \quad \mbox{ and } \quad \tgam_0(u):=\tgam(u, v_0).$$

To prove this lemma, we compute $\dv \gamma$ and obtain
\begin{equation}
|\tgam(u,v) - \tgam_0(u)| \leq \int^u_{u_2} \int^v_{v_0} |J_9|(u',v') + |J_{10}|(u',v') + |J_{11}|(u',v') \D v' \D u'
\end{equation}
with the expressions for $J_8, J_{9}, J_{10}$ available in \eqref{FIT248}. By carefully estimating each of these integrals, we deduce that
\begin{equation}
\int^v_{v_0} |J_8|(u',v') + |J_{9}|(u',v') + |J_{10}|(u',v') \D v' \leq C_{53}\ml \frac{-\du r}{r} \mr_0(u')e^{3\gamma_0(u')}\delta(u',v)^{\frac{1}{2}}.
\end{equation}
This gives
\begin{equation*}
\begin{aligned}
|\tgam(u,v) - \tgam_0(u)| &\leq C_{53} e^{3\gamma_0(u)}\int^u_{u_2} \ml \frac{-\du r}{r}\mr_0(u') \delta(u',v)^{\frac{1}{2}}\D u' \\
&\leq 2C_{53} e^{ \ml \frac{15}{4} + \frac{3 \xi}{2} \mr \gamma_0(u)}\delta(u,v)^{\frac{1}{2}} \int^u_{u_2} \ml \frac{-\du r}{r}\mr_0(u') e^{-\frac{1}{4}\log \ml \frac{r_0(u_2)}{r_0(u')}\mr}\D u' \\
&\leq 8C_{53} e^{\ml \frac{15}{4} + \frac{3 \xi}{2} \mr\gamma_0(u)}\delta(u,v(\tu))^{\frac{1}{2}} \leq 1.
\end{aligned}
\end{equation*}

Applying the above lemma, we obtain
\begin{equation*}
\begin{aligned}
\delta(\tu,v(\tu)) &= \frac{1}{r_0(\tu)}\int^{v(\tu)}_{v_0} (\dv r)(\tu,v')\D v' = \frac{1}{r_0(\tu)} \int^{v(\tu)}_{v_0} (\dv r)(0,v')e^{-\gamma_0(\tu) + \ml \gamma_0(\tu) -\tgam(\tu,v') \mr - \gamma(u_2,v')} \D v' \\
&\quad\quad\quad\leq \frac{e C_{56}}{2r_0(\tu)} \int^{v(\tu)}_{v_0} e^{-\gamma_0(\tu) } \D v' \leq \frac{e C_{56} e^{-\gamma_0(\tu)}}{2r_0(\tu)}(v(\tu) - v_0)
\end{aligned}
\end{equation*}
with
\begin{equation*}
C_{56}(u_2) := \max_{v'\in[v_0,v_0 + \te]}e^{|\gamma|(u_2,v')} < + \infty.
\end{equation*}
Note that the direction of the above inequality for $\delta$ is reversed as compared to \eqref{FIT83'}. In \eqref{FIT83'} the outer $\delta(u_3, v')$ is bounded by the inner $\delta(u_4, v')$ with inflation factors. 

The bulk of the arguments for that section relies on a simple inequality mentioned in Lemma \ref{SITLemma1}, which states that for any arbitrary measurable subsets in $\mathbb{R}$,
\begin{equation}\label{second1}
\int_X f|g|^2 \geq 4 \int_X f|h|^2 \implies \int_X f|g+h|^2 \geq \frac{1}{4}\int_X f|g|^2 \geq \int_X f|h|^2.
\end{equation}
To proceed, we first deduce:
\begin{equation}\label{second2}
4 \int_{v_0}^{v}  \mlm \frac{r \dv \phi}{\dv r} \mrm_0^2  (\tu,v')  \D v' \leq 4( v(\tu) - v_0 ) V^2(\tu)e^{2\gamma_0(\tu)}.
\end{equation}
Subsequently, we will attempt to show that 
\begin{equation}\label{V(u) and Theta}
4( v(\tu) - v_0 ) V^2(\tu)e^{2\gamma_0(\tu)}\leq \int_{v_0}^{v(\tu)} \mlm \Upxi \mrm^2 (\tu,v')  \D v' .
\end{equation}
To prove \eqref{V(u) and Theta}, we employ
\begin{equation}\label{second3}
\int_{v_0}^{v(\tu)} \mlm \Upxi \mrm^2 (\tu,v')  \D v' 
\geq \frac{e^{-2C_{55}}}{9}e^{2\gamma_0(\tu)}\int^{v(\tu)}_{v_0} |\psi(0,v') + J_{13}(\tu,v')|^2  \D v'.
\end{equation}
In the later sections, we will verify that
\begin{equation}\label{second4}
\int_{v_0}^{v(\tu)} |\psi|(0,v') \D v' \geq 4 \int_{v_0}^{v(\tu)} |J_{13}|(\tu,v') \D v' .
\end{equation}

Next, we define
$$h^2(v):=\frac{1}{v-v_0}\int_{v_0}^v |\psi|^2(0, v')dv'$$ 
and choose
\begin{equation*}
h^2(v(\tu)) \geq 144 e^{2C_{55}} V^2(\tu).
\end{equation*}
These imply that
\begin{equation}\label{second5}
4(v(\tu)-v_0)V^2(\tu)e^{2\gamma_0(\tu)} \leq \frac{e^{-2C_{55}}}{36}e^{2\gamma_0(\tu)}\int^{v(\tu)}_{v_0} |\psi|^2(0,v') \D v'.
\end{equation}
Stitching \eqref{second1}, \eqref{second2}, \eqref{second3}, \eqref{second4}, and \eqref{second5} together, we obtain the inequality:
\begin{equation*}
\int_{v_0}^{v(\tu)}  \mlm \ml \frac{r \dv \phi}{\dv r}\mr_0 + \Upxi \mrm^2 (\tu,v') \D v' \geq \frac{e^{-2C_{55}}}{144}e^{2\gamma_0(\tu)}\int^{v(\tu)}_{v_0}|\psi|^2(0,v')\D v'.
\end{equation*}

Plugging this back to the expression for $\eta^*(\tu,v(\tu))$, we thus obtain
\begin{equation*}
\begin{aligned}
\eta^*(\tu,v(\tu)) &= \frac{4\pi}{r(\tu,v(\tu))}\int^{v(\tu)}_{v_0} (1-\mu) \ml \frac{r^2|\dv \phi|^2}{\dv r}\mr(\tu,v') \D v' \\
&\geq \frac{16\pi C_{54}(1-\mu_0(0))}{27e r(\tu,v(\tu))}e^{-2\gamma_0(\tu)}\int_{v_0}^{v(\tu)}  \mlm \ml \frac{r \dv \phi}{\dv r}\mr_0 + \Upxi \mrm^2 (\tu,v') \D v' \\
&\geq \frac{2\pi C_{54}e^{-2C_{55}}(1-\mu_0(0))}{729 e r_0(\tu)} \int^{v(\tu)}_{v_0}|\psi|^2(0,v')\D v'.
\end{aligned}
\end{equation*}
This would contradict the no trapped-surface-formation condition \eqref{TSF2'} if we choose
\begin{equation*}
\begin{aligned}
& \frac{eC_{36}C_{56}}{2r_0(\tu)}(v(\tu) - v_0)\ml r_0(\tu) + e^{-\gamma_0(\tu)}\ml \log \ml \frac{1}{v(\tu) - v_0}\mr + \log \ml \frac{2r_0(0)}{eC_{56}}\mr  + 
 \log \ml\frac{r_0(\tu)}{r_0(0)} \mr + \gamma_0(\tu)\mr\mr \\
 \leq & \; \frac{\pi C_{54}e^{-2C_{55}}(1-\mu_0(0))}{729e r_0(\tu)} \int^{v(\tu)}_{v_0}|\psi|^2(0,v')\D v'.
\end{aligned}
\end{equation*}
Therefore, to trigger the second instability mechanism, it suffices to require
\begin{equation*}
\begin{aligned}
h^2(v(\tu)) \geq  C_{61}\ml r_0(\tu) + e^{-\gamma_0(\tu)}\ml \log \ml \frac{1}{v(\tu) - v_0}\mr + \log \ml \frac{2r_0(0)}{eC_{56}}\mr  + 
 \log \ml\frac{r_0(\tu)}{r_0(0)} \mr + \gamma_0(\tu)\mr \mr
 \end{aligned}
\end{equation*}
with
\begin{equation*}
C_{61}(u_2) := \frac{729 e^2 C_{36} C_{56}}{2\pi C_{54}e^{-2C_{55}}(1-\mu_0(0))}.
\end{equation*}

For the inequality
$$h^2(v(\tu)) \geq 144 e^{2C_{55}} V^2(\tu)$$
as mentioned above, we carefully choose
\begin{equation*}
v(\tu) - v_0 = e^{-A\gamma_0(\tu) - B \log \ml \frac{r_0(0)}{r_0(\tu)}\mr} \quad \mbox{ with } \quad A=14 \mbox{ and } B=1.
\end{equation*}
By employing a small bootstrap argument involving the use of $\du \log \delta(u,v)$,
we ultimately satisfy all the imposed requirements and establish the second instability theorem. It is noteworthy that setting $B = 1$ is necessary, while $A$ can be chosen to be as large as required.

\subsection{Related Works} 
In this section, we list some references related to our approach. In \cite{EU2, EU1, EU3}, Christodoulou initiated the exploration of the gravitational collapse for the spherically symmetric Einstein-scalar field system. In a series of breakthrough works spanning four papers \cite{christ1}-\cite{christ4}, Christodoulou proved the weak cosmic censorship conjecture for this system. Subsequently, the proofs and the conclusions in these four papers have been revisited as follows:
\begin{itemize}[leftmargin=*]
\item[-] The method used to prove trapped surface formation as in \cite{christ1} is subsequently applied to the spherically symmetric Einstein-scalar field system with a positive cosmological constant in \cite{Costa} by Costa, and to the Einstein-Maxwell-charged scalar field system within spherical symmetry in \cite{An-Lim} by the first author and Lim.    

\item[-] Within spherical symmetry and self-similar ansatz, the construction of naked singular solutions in \cite{christ3} for the Einstein-scalar field system later inspires its extension to higher dimensions by the first author and Zhang in \cite{An-ZhangXF}, and by Cicortas in \cite{Cicortas}. A recent revisit of \cite{christ3} by Singh is presented in \cite{Singh}. Additionally, for the Einstein-Euler system, under the symmetry assumption, naked singularities are successfully constructed by Guo-Hadzic-Jang in \cite{GHJ}.

\item[-] In \cite{christ4}, Christodoulou presented two instability theorems, consistent with his earlier established sharp trapped surface formation criterion. Building on BV regularity estimates from \cite{christ2} and the blueshift effect discussed in \cite{christ4}, Li-Liu \cite{LL1} gave a different argument to obtain instability results with a priori estimates. Note that their instability theorem diverges from Christodoulou's two instability theorems in \cite{christ4}.

\end{itemize}

Within spherical symmetry, based on the aforementioned Christodoulou's works, various investigations have been unfolded regarding gravitational collapse. For insights into spacelike singularities in the Einstein-scalar field system, refer to \cite{An-ZhangRX} by the first author and Zhang, \cite{An-Gajic} by the first author and Gajic, and \cite{Warren Li} by Li. For the study of weak null singularity, related to the proof presented in this paper for the charged scenario, see \cite{dafermos1} by Dafermos, \cite{Luk-Oh} by Luk-Oh, \cite{Van de Moortel1}, \cite{vandemoortel2020breakdown} by Van de Moortel, \cite{Li-Moortel} by Li-Van de Moortel, \cite{CCDHJ18a} by Cardoso-Costa-Destounis-Hintz-Jansen, \cite{Rossetti} by Rossetti, and to the references therein. We also note a recent result \cite{KeUn} by Kehle and Unger investigating extremal critical collapse for the Einstein-Maxwell-Vlasov system. The approach of this current paper with double-null foliations is quite robust. The first author and collaborators are extending this approach to cosmological spacetimes and other matter models. 

Beyond spherical symmetry, in a monumental work \cite{Chr:book} Christodoulou invented the short-pulse method and proved the first trapped-surface-formation result for the Einstein vacuum equations. Subsequently, Klainerman-Rodnianski\cite{KR:Trapped} simplified and expanded this work by introducing the signature for the short pulse. In \cite{KLR}, Klainerman-Luk-Rodnianski further relaxed the initial requirements and provided a fully anisotropic mechanism for the formation of trapped surfaces. The first scale-critical result on trapped surface formation was obtained by the first author and Luk in \cite{AL}, opening a door to study the region near the central singularity. In \cite{An17} by the first author and in \cite{AH} by the first author and Han, emergences of (isotropic and anisotropic) apparent horizons from the center of gravitational collapse are further established. Dynamics of apparent horizon is also studied by the first author and He in \cite{AnHe}. In \cite{An12}, the first author introduced the signature for decay rates. Utilizing only this signature, a new concise proof for the (scale-critical) trapped surface formation criterion is provided by the first author in \cite{An19} . Later, this approach is generalized to the Einstein-Maxwell system by the first author and Athanasiou in \cite{AA}. For the Einstein-scalar field system, with singular initial data, Li-Liu proved an almost scale-critical trapped surface formation criterion in \cite{LL2}. It is worth noting that the construction of the naked-singularity solutions for the Einstein vacuum equations was recently accomplished by Rodnianski-Shlapentokh-Rothman \cite{R-S} and by Shlapentokh-Rothamn \cite{YS}. Based on Christodoulou's naked singularity constructed in \cite{christ3}, for the Einstein-scalar field system with no symmetric assumption, the proof of naked-singularity censoring with anisotropic apparent horizon was also recently established by the first author in \cite{An2023}.  

\subsection{Acknowledgement} Both authors would like to thank Demetrios Christodoulou for insightful correspondence. XA extends his thanks to Jonathan Luk for a valuable discussion of Christodoulou's work back to Princeton. Additionally, both authors would like to thank Haoyang Chen and Kai Dong for the beneficial discussions and for their careful reading of this manuscript. 
XA is supported by MOE Tier 1 grants A-0004287-00-00, A-0008492-00-00 and MOE Tier 2 grant A-8000977-00-00. HKT is partially supported under the MOE Tier 1 grant R-146-000-269-133 during his time at the National University of Singapore.

\section{Preliminaries}\label{Prelim}

In this study, we utilize the double-null coordinate system to re-express the Einstein-Maxwell-charged-scalar-field system as presented in \eqref{Intro1}.

We say that a spacetime $(\mathcal{M},g)$ is \textit{spherically symmetric} if $SO(3)$ acts on it by isometry, and the orbits of the group are (topological) $2$-dimensional spheres $\mathcal{S}$. Our spacetime $(\mathcal{M},g)$ can be represented with a $2$-dimensional Penrose's diagram, by considering the quotient manifold  $\mathcal{Q} = \mathcal{M}/\mathcal{S}$. Under spherical symmetry, we employ the following ansatz for the Lorentzian metric:
\begin{equation}\label{Setup4}
g = - \Omega^2(u,v) \; \mathrm{d}u\mathrm{d}v + r^{2}(u,v) (\mathrm{d}\theta^2 + \sin^2 \theta \mathrm{d} \phi^2),
\end{equation}
where $\Omega(u,v)$ and $r(u,v)$ are positive functions on $\mathcal{Q}$ representing lapse and area radius respectively. Here, $u$ and $v$ are optical functions satisfying $g^{\alpha\beta}\partial_{\alpha} u \partial_{\beta} u = 0$ and $g^{\alpha\beta}\partial_{\alpha} v \partial_{\beta} v = 0$, with their level sets corresponding to outgoing and incoming null hypersurfaces respectively. In Penrose's diagram, these null geodesics are $45$-degree slanted. We further denote the axis of symmetry to be $\Gamma$, where the area radius function vanishes, and we represent $\Gamma$ as the vertical axis. Mathematically, we express this as 
\begin{equation}\label{Setup1}
\Gamma := \{ (u,v) \in \mathcal{Q} \hspace{3pt}|\hspace{3pt} r(u,v) = 0\}.
\end{equation}

As shown in Figure \ref{SetupFig1}, we proceed to set up the coordinate system in $\mathcal{Q}$ as follows.

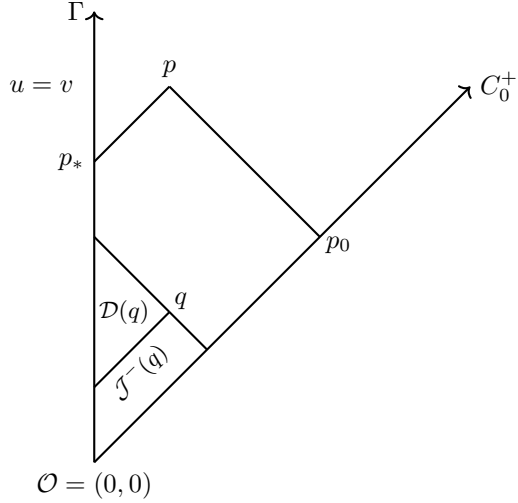
\begin{figure}[htbp]
\begin{minipage}[!t]{0.4\textwidth}
	\centering
\begin{tikzpicture}
	\begin{scope}[thick]
	\draw[->] (0,0) node[anchor=north]{$\mathcal{O} = (0,0)$} --  (0,6) node[anchor = east]{$\Gamma$};
	\draw[->] (0,0) -- (5,5) node[anchor = west]{$C_0^+$};
	\draw(3,3) -- (1,5)node[anchor = south]{$p$};
	\draw(1,5) -- (0,4)node[anchor = east]{$p_*$};
	\draw(0,3) -- (1,2);
        \draw(0,1) -- (1,2);
        \draw(1,2) -- (1.5,1.5);
        \node[] at (3.25,2.9) {$p_0$};
        \node[] at (1.15,2.15) {$q$};
	\node[] at (0.4,2) {\small $\mathcal{D}(q)$};
        \node[] at (-0.7,5) {$u = v$};
        \node[rotate=45] at (0.6,1.2) {\small $\mathcal{J}^-(q)$};
	\end{scope}
\end{tikzpicture}
\end{minipage}
\caption{The coordinate system in $\mathcal{Q}$}
\label{SetupFig1}
\end{figure}

\begin{enumerate}[leftmargin=*]
    \item Let $\Gamma$ denote the axis of symmetry as described in \eqref{Setup1}. 
      
    \item For any point $p$ in the causal future of $C_0^+$, there is a past-directed incoming null curve emanating from it and intersecting $C_0^+$ at $p_0$. We then assign the $v$-coordinate of $p$ to be the same as that of $p_0$.

   \item For the points along $\Gamma$, by the aforementioned instruction, each point about $\mathcal{O}$ is associated with a $v$-coordinate. Furthermore, we choose its $u$-coordinate along $\Gamma$ to satisfy
    \begin{equation}\label{Setup2}
        \Gamma:\; u = v.
    \end{equation}
    
    \item For an outgoing null curve emitting from a point $\mathcal{O}\in \Gamma$, we denote it $C_0^+$. We fix $u=0$ and prescribe the coordinate parameter $v$ along it. We require the point $\mathcal{O}$ to have the coordinate $(0,0)$.
     
    \item Back to the point $p$, there is a past-directed {\color{black}outgoing} null curve emanating from $p$ and intersecting $\Gamma$ at $p_*$. We proceed to assign the $u$-coordinate of $p$ to be equal to that of $p_*$.
\end{enumerate}

With the above-chosen coordinate system, we define some commonly used terminologies below.
\begin{itemize}[leftmargin = *]
    \item With $q = (u',v')$, we denote $\mathcal{D}(u',v')$ to be its {\color{black}\textit{domain of dependence determined by the point $q=(u', v')$}} and we have
    \begin{equation}\label{Setup3}
    \mathcal{D}(u',v') = \{(u,v) \in \mathcal{Q}: u \in [u',v'), v \in (u,v']\}.
    \end{equation}
    \item With $q = (u',v')$, we set $\mathcal{J}^-(u',v')$ to be its \textit{causal past} and the following holds:
    \begin{equation}\label{Setup28}
    \mathcal{J}^-(u',v') = \{(u,v) \in \mathcal{Q}: u \in [0,u'], v \in (u,v')\}.
    \end{equation}
\end{itemize}

We further eliminate the gauge freedom of $v$ for our double-null coordinates. Along $C_0^+$, we require our coordinate parameter $v$ to satisfy
\begin{equation}\label{Setup33}
C_0^+:\; v = 2r
\end{equation}
and we also impose
\begin{equation}\label{Setup34}
\Omega(0,0) = 1.
\end{equation}

Next, we remark that the system considered above is subjected to the following boundary conditions along $\Gamma$: 
\begin{enumerate}
\item For every $(u,u) \in \Gamma$ with $u \geq 0$ and $F \in C^1(\mathcal{Q})$, we first define
\begin{equation}\label{Setup74}
F|_{\Gamma}(u) := \lim_{v \rightarrow u^+}F(u,v).
\end{equation}
Under spherical symmetry, when we consider points infinitesimally close to the center, we note that the incoming null geodesics would become outgoing in the opposite direction (and vice versa) after bypassing $\Gamma$. Mathematically, for $u \geq 0$ and $F \in C^1(\mathcal{Q})$, we hence have
\begin{equation}\label{Setup75}
F|_{\Gamma}(u) = \lim_{\tilde{u} \rightarrow u^-}F(\tilde{u},u) \quad \mbox{ and } \quad \dv F|_{\Gamma}(u) = - \du F|_{\Gamma}(u).
\end{equation}
\item As employed in \cite{christ2}, at regular points along $\Gamma$, we impose the following boundary conditions:
\begin{equation}\label{Setup35}
|g_{\mathcal{Q}}(\nabla r,\nabla r)| < +\infty, \quad |\phi|_\Gamma < +\infty, \quad |F_{uv}|_{\Gamma} < + \infty. 
\end{equation}
These imply 
\begin{equation}\label{Setup36}
m|_{\Gamma} = 0, \quad r\phi|_{\Gamma} = 0, \quad Q|_{\Gamma} = 0.
\end{equation}
An alternative approach to observe that \eqref{Setup36} must hold at regular points along $\Gamma$ is to consider the Kretschmann scalar of this system. To avoid being a singular point, \eqref{Setup36} must hold. 
\item In addition, it is worth noting that the electromagnetic potential $A_u$ enjoys a gauge freedom as described in \eqref{Setup10}. Considering \eqref{Setup11}, without loss of generality, we set
\begin{equation}\label{Setup37}
A_u|_{\Gamma} = 0.
\end{equation}
\item {\color{black} By the spherical symmetry assumption, for points infinitesimally close to the center, its incoming null geodesics essentially become outgoing in the opposite direction. Mathematically, for all $u \geq 0$, we have
\begin{equation}\label{Setup38}
\dv r|_{\Gamma} = -\du r|_{\Gamma}.
\end{equation}}
\end{enumerate}

For the study of gravitational collapse, it is natural to assume that there is no anti-trapped surface along $C_0^+$, that is, $\du r(0,v) < 0$ for all $v > 0$. By \eqref{Setup14}, this in turn implies that 
\begin{equation}\label{Setup40}
\du r(u,v) < 0
\end{equation}
for all $(u,v) \in \mathcal{Q}$. 

Next, we define the regular region $\mathcal{R}$, the apparent horizon $\mathcal{A}$, and the trapped region $\mathcal{T}$ below.
\begin{itemize}[leftmargin=*]
\item The \textit{regular region} $\mathcal{R}$ refers to the set of points on $\mathcal{Q}$ in which $\dv r > 0$;
\item The \textit{apparent horizon} $\mathcal{A}$ refers to the set of points  on $\mathcal{Q}$, where $\dv r = 0$.
\item The \textit{trapped region} $\mathcal{T}$ refers to the set of points  on $\mathcal{Q}$ in which $\dv r < 0$.
\end{itemize}

A $2$-sphere is termed a \textit{trapped surface} if $\dv r(u,v) < 0$. If $\dv r(u,v) = 0$, we labelthe $2$-sphere at $(u,v)$ as a \textit{marginally outer trapped surface} (MOTS).

In view of the definitions above, using \eqref{Setup5}, we have an equivalent characterization as shown below:
\begin{center}
\begin{tabular}{|c|c|c|c|}
\hline
    Regions & Sign of $\dv r$ & Sign of $1- \mu$ & Value of $\mu$ \\
\hline
    $(u,v) \in \mathcal{R}$ & $\dv r > 0$ & $1 - \mu > 0$ & $\mu < 1$ \\
    $(u,v) \in \mathcal{A}$ & $\dv r = 0$ & $1 - \mu = 0$ & $\mu = 1$ \\
    $(u,v) \in \mathcal{T}$ & $\dv r < 0$ & $1 - \mu < 0$ & $\mu > 1$ \\
\hline
\end{tabular}
\end{center}

\noindent Related to the gauge fixing along $C_0^+$, since we set $v=2r$ along $C_0^+$, we have $\partial_v r|_{C_0^+}=1/2$. Referring to the table above, we note that neither the trapped surface nor the MOTS appears along $C_0^+$. In addition, we introduce the term
\begin{equation}\label{Setup42}
\mu_* := \sup_{v > 0} \mu(0,v) 
\end{equation}
which lies in the range $[0,1)$.

We further record a proposition that will be commonly referenced in subsequent sections.
\begin{proposition}\label{SetupProp3}
Along any future-directed outgoing null curve with $u > 0$, if $(u,v_1) \in \mathcal{A}$, then for all $(u,v)$ with $v \geq v_1$, we have $(u,v) \in \mathcal{A} \cup \mathcal{T}$. Moreover, if $(u,v_1) \in \mathcal{T}$, then for all $(u,v)$ with $v \geq v_1$, we have $(u,v) \in \mathcal{T}$. 
\end{proposition}
\begin{proof}
From \eqref{Setup15}, we see that $\dv(\Omega^{-2}\dv r) \leq 0$. For all $v \geq v_1$, this implies
\begin{equation}\label{Setup44}
\Omega^{-2}\dv r(u,v) \leq \Omega^{-2}\dv r(u,v_1).
\end{equation}
Thus, if $(u,v_1) \in \mathcal{A}$ with $\dv r(u,v_1) = 0$, the above yields that $\Omega^{-2}\dv r(u,v) \leq 0$. In particular, we have $\dv r(u,v) \leq 0$ since $\Omega^2 > 0$ for all $(u,v) \in \mathcal{Q}$, and thus $(u,v) \in \mathcal{A} \cup \mathcal{T}$. Similarly, if $(u,v_1) \in \mathcal{T}$ with $\dv r(u,v_1) < 0$, for $v\geq v_1$, we then conclude $\dv r(u,v) < 0$ and thus $(u,v) \in \mathcal{T}$. 
\end{proof}
Intuitively, this implies that once we encounter a MOTS along a fixed future-directed outgoing null curve at some $(u,v_1)$, then the future of this curve would either be part of the apparent horizon or be within the trapped region.

Last but not least, we review some preliminaries for the (future) singular boundary of the spacetime.

In view of Theorem \ref{FET}, there exists a point $P = (0,v_0)$ on $C_0^+$, such that the incoming null curve $C_0^-$
\begin{equation}\label{Setup45}
\mbox{with } C_0^- := \{ (u,v) \in \mathcal{Q}: v = v_0 \text{  and  } 0 \leq u \leq v_0\}
\end{equation}
cannot terminate at a regular vertex on $\Gamma$. Denote this vertex to be $\mathcal{O}' = (v_0,v_0)$. In addition, we call $P$ as the \textit{maximal sphere of development}, and we say that the past of $C_0^-$ is a \textit{terminal indecomposable past set} in the terminology of \cite{Penrose1}. The vertex $\mathcal{O}'$, also known as the unique future limit point of $\Gamma$ in $\overline{\mathcal{Q}} \setminus \mathcal{Q}$, marks the start of the \textit{singular boundary} $\mathcal{B}$, which consists of boundary points of $\mathcal{Q}$ lying in the ambient manifold $\mathbb{R}^{1+1}$. Note that $\mathcal{B}$ contains two components, the central component $\mathcal{B}_0$ (composed by possibly empty null segment and its past endpoint $\mathcal{O}'$) and the noncentral component $\mathcal{B} \setminus \mathcal{B}_0$ (being achronal, with $r$ vanishing on it). According to \cite{kommemi}, we can decompose the null segment $\mathcal{B}_0$ as $\mathcal{B}_0 = \mathcal{S}^1_{\Gamma} \cup \mathcal{CH}_\Gamma \cup \mathcal{S}^2_{\Gamma}$. Moreover, the radius function $r$ extends continuously to zero on $\mathcal{S}^1_{\Gamma}$ and $\mathcal{S}^2_{\Gamma}$, but extends continuously to possibly non-zero values on $\mathcal{CH}_\Gamma$. Furthermore, $\mathcal{A}$ necessarily issues from $\mathcal{CH}_\Gamma$, and in particular, from $\mathcal{B}_0$. The relative positions of the aforementioned points are illustrated in Figure \ref{SetupFig5}. \\

\begin{figure}[htbp]
\begin{minipage}[!t]{0.7\textwidth}
\centering
\begin{tikzpicture}
	\begin{scope}[thick]
	\draw[-] (0,0) node[anchor=north]{$\mathcal{O}$} --  (0,8) node[anchor = east]{$\mathcal{O}'$};
	\draw[-] (0,0) -- (8,8);
        \node[rotate=45] at (5.2,4.8) {$C_0^+$};
        \node[] at (-0.3,5.0) {$\Gamma$};
        \fill[gray!50] (1,9) to (2,10) to[out = 45,in=135] (6,10) to[out = -135, in = 45] (1,9);
        \draw[] (0,8) -- (4,4);
        \draw[densely dashed] (0,4) -- (6,10);
        \node[rotate=45] at (4.2,7.8) {$\mathcal{EH}$};
        \draw[] (6,10) -- (8,8);
        \node[rotate = -45] at (7.2,9.2){$\mathcal{I}^+$};
        \draw[thick, densely dotted] (0,8) to (2,10);
        \draw[thick, loosely dotted] (2,10) to [out = 45,in=135] (6,10);
        \draw[fill=white] (2,10) circle (2pt);
        \draw[fill=white] (0.7,8.7) circle (2pt);
        \node[rotate=45] at (0.1,8.5) {$\mathcal{S}^1_\Gamma$};
        \draw[fill=white] (1.4,9.4) circle (2pt);
        \node[rotate=45] at (0.8,9.2) {$\mathcal{CH}_\Gamma$};
        \node[rotate=45] at (1.5,9.9) {$\mathcal{S}^2_\Gamma$};
        \draw[thick, dashdotted] (1,9) to[out = 45, in = - 135] (6,10);
        \node[] at (4.1,9.7) {$\mathcal{A}$};
        \node[] at (4.2,10.3) {$\mathcal{BH}$};
        \node[anchor = west] at (9.8,10.2) {\ul{Legend:}};
        \node[anchor = west] at (11,9.5) {$\mathcal{B}_0$}; 
        \draw[thick, densely dotted] (10,9.3) to (11,9.7);
        \node[anchor = west] at (11,9) {$\mathcal{B} \setminus \mathcal{B}_0$}; 
        \draw[thick, loosely dotted] (10,8.8) to (11,9.2);
        \node[anchor = west] at (11,8.5) {$\mathcal{A}$}; 
        \draw[thick, dashdotted] (10,8.3) to (11,8.7);
        \draw[fill=white] (0,8) circle (2pt);
        \draw[fill=white] (6,10) circle (2pt);
	\end{scope}
\end{tikzpicture}
\end{minipage}
\caption{}
\label{SetupFig5}
\end{figure}

Let $\mathcal{U}$ be a subset of the regular region $\mathcal{R}$. With regards to the regularity of solutions to the spherically symmetric Einstein-Maxwell-charged scalar field system, we first define the notion of a $C^1$ solution as follows:

\begin{definition}\label{C1solution}
A solution to equations \eqref{Setup12}-\eqref{Setup19} is called a $C^1$ solution in $\mathcal{U}$ if $$\inf_{\mathcal{U}} \du r > - \infty$$ 
and $\dv r, \du r, \phi, r \dv \phi, r \du \phi, Q/r$ are $C^1$ functions on $\overline{\mathcal{U}}$.
\end{definition}

\begin{remark}
Contrasting with the corresponding definitions in \cite{christ2}, here we also require $Q/r$ being $C^1$ on $\overline{\mathcal{U}}$. This follows from considering \eqref{Setup18}, \eqref{Setup19} and the fact that $\dv r, \du r, \phi, r \dv \phi, r \du \phi$ being $C^1$ on $\overline{\mathcal{U}}$. Furthermore, note that $Q/r$ is a dimensionless quantity, which is consistent with the other terms mentioned above. 
\end{remark} 

\begin{remark}
In this paper, when we refer to $C^1$ initial data, we mean
$\alpha:=\frac{\partial_v(r\phi)}{\partial_v r}$ is $C^1$ along $C_0^+$. 
\end{remark}

We proceed to define the notion of a solution of bounded variation in $\mathcal{U}$ as in \cite{christ2} below:
\begin{definition}\label{BVsolution}
A solution to equations  \eqref{Setup12}-\eqref{Setup19} is called a solution of bounded variation in $\mathcal{U}$ if 
\begin{equation}\label{Setup77}
\inf_{\mathcal{U}} \du r > - \infty,
\end{equation}
and 
\begin{enumerate}[(i)]
\item $\dv r$ is of bounded variation along each outgoing null curve and $\du r$ is of bounded variation along each incoming null curve;
\item For each $p$ with $(p,p) \in \Gamma$, we have
\begin{equation}\label{Setup78}
\lim_{\varepsilon \rightarrow 0^+}(\dv r + \du r)(a,a+\varepsilon) = 0;
\end{equation}
\item $\phi$ is an absolutely continuous function on each outgoing null curve and on each incoming null curve;
\item For each $p$ with $(p,p) \in \Gamma$, the following conditions hold:
\begin{equation}\label{Setup79}
\begin{aligned}
\lim_{\varepsilon \rightarrow 0} \sup_{0 < \delta \leq \varepsilon} TV_{\{p - \delta\} \times (p-\delta,p)}[\phi] = 0, \qquad \lim_{\varepsilon \rightarrow 0} \sup_{0 \leq \delta < \varepsilon} TV_{ (p-\varepsilon,p-\delta) \times \{p-\delta\} }[\phi] = 0, \\
\lim_{\varepsilon \rightarrow 0} \sup_{0 < \delta \leq \varepsilon} TV_{ (p,p+\delta) \times \{p+\delta\} }[\phi] = 0, \qquad \lim_{\varepsilon \rightarrow 0} \sup_{0 \leq \delta < \varepsilon} TV_{\{p+\delta\} \times (p+\delta,p+\varepsilon) }[\phi] = 0;
\end{aligned}
\end{equation}
\item $r\dv \phi$ is of bounded variation on each outgoing null curve, and $r \du \phi$ is of bounded variation on each incoming null curve;
\item For each $p$ with $(p,p) \in \Gamma$, we have
\begin{equation}\label{Setup80}
\lim_{\varepsilon \rightarrow 0^+} \ml \dv(r\phi) + \du(r\phi) \mr (p,p+\varepsilon) = 0.
\end{equation}
\end{enumerate}
\end{definition}

\begin{remark}
Here, we retain the definition of a solution of  bounded variation solution as in \cite{christ2}. In Section \ref{BV Area Estimates}, we will deduce the additional total-variation estimates for the charged terms and observe that the above is a proper definition.
\end{remark}

\section{Novel Estimates for \texorpdfstring{$|Q|$}{Q} and  \texorpdfstring{$|\phi|$}{Phi}}\label{Q and phi} 
We now prove several crucial estimates for $Q$ and $\phi$ which are employed throughout the paper. We start off with a useful lemma for $\dv r/(1-\mu)$ and $\du r/(1-\mu)$ below.
\begin{lemma}
\label{SetupLemma1}
In the regular region $\mathcal{R}$, the function $\frac{\dv r}{1 - \mu}(u,v)$ is decreasing along any future-directed incoming null curve and the function $\frac{-\du r}{1-\mu}$ is increasing along any future-directed outgoing null curve. Furthermore, if $(u,v) \in \mathcal{R}$ and there exists $0 \leq \overline{u} < u$ such that $(u',v) \in \mathcal{R}$ for all $u' \in [\overline{u},u]$, then we have
\begin{equation}\label{Estimates1}
\frac{\dv r}{1-\mu} (u,v) \leq \frac{\bar{B}}{(1-\overline{\mu_*})}(\overline{u},v)
\end{equation}
with 
\begin{equation}\label{Bandmu*}
\bar{B}(\overline{u},v) := \sup_{v'\in(\overline{u},v]} \dv r(\overline{u},v') \quad \text{ and } \quad \overline{\mu_*}(\overline{u},v) := \sup_{v'\in (\overline{u},v]}\mu(\overline{u},v') \in [0,1).
\end{equation}
\end{lemma}
\begin{figure}[htbp]
\begin{tikzpicture}
\begin{scope}[thick]
\draw[->] (0,0) node[anchor=north]{$\mathcal{O}$} --  (0,6) node[anchor = east]{$\Gamma$};
\draw[->] (0,0) -- (5,5) node[anchor = west]{$C_0^+$};
\draw[thick, dashdotted]  (4.9,5.1) to[in=-120,out= -140] (2.4,2.9) to[out=60,in=0] (2,4) to[out=180,in=-45] (0,4.6);
\draw[thick] (1.2,3.8) node[anchor = east]{$p$} -- (2,3)node[anchor = south]{$\overline{p}$};
\draw[thick] (0,1) -- (2,3);
\draw[dotted] (2,3) -- (2.5,2.5);
\node[] at (0.5,2.5) {$\mathcal{R}$};
\node[] at (3.5,4.5) {$\mathcal{T}$};
\node[] at (2,4.3) {$\mathcal{A}$};
\node at (1.2,3.8)[circle,fill,inner sep=1.5pt]{};
\node at (2,3)[circle,fill,inner sep=1.5pt]{};
\end{scope}
\end{tikzpicture}
\caption{For a given point $p = (u,v) \in \mathcal{R}$, suppose that we can find some $\overline{p} = (\overline{u},v) \in \mathcal{R}$ with $0 \leq \overline{u} < u$ such that the incoming null curve $\overline{p}p$ is in $\mathcal{R}$, then the estimates \eqref{Estimates1} holds. Here the supremums in \eqref{Bandmu*} are obtained along the outgoing null curve originating from $\Gamma$ which intersects with $\overline{p}$. Even though the dotted line from $\overline{p}$ extends to the past along an incoming null curve may intersect $\mathcal{A}$ and transits into $\mathcal{T}$ before intersecting $C_0^+$, this does not affect our lemma here as it suffices to only find a point $\overline{p}$ along a past-directed incoming null curve from $p$ such that $\overline{p}p$ is in $\mathcal{R}$. If the dotted line from $\overline{p}$ does not intersect $\mathcal{A}\cup\mathcal{T}$, we can set $\overline{u} = 0$. }
\label{SetupFig2}
\end{figure}
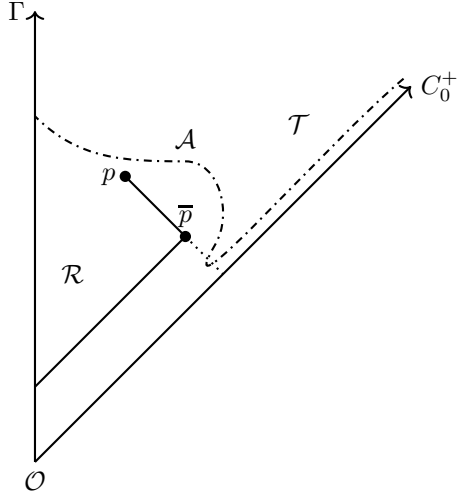
\begin{proof}
We use Figure \ref{SetupFig2} to illustrate this lemma. For $(u,v) \in \mathcal{R}$, we have $1-\mu > 0$, $\dv r > 0$ and $\du r < 0$. By \eqref{Setup21}, it implies that $\du \ml \frac{1-\mu}{\dv r}\mr \geq 0$ and $\du \ml \frac{\dv r}{1 - \mu}\mr \leq 0$. Hence, $\frac{\dv r}{1-\mu}$ is decreasing along future-directed incoming null curves, which lies in $\mathcal{R}$. Similarly, via \eqref{Setup22}, we deduce that $\frac{-\du r}{1 - \mu}$ is increasing along future-directed outgoing null directions. The bound \eqref{Estimates1} then follows from the monotonicity of $\frac{\dv r}{1-\mu}$ along incoming null directions.
\end{proof}
Next, we record a calculus lemma that would be used not only for obtaining estimates of $|\phi|$ and $|Q|$, but also heavily employed throughout the paper. 
\begin{lemma}\label{xlogx}
For any $x \in (0,1)$ and $c > 0$, we have
\begin{equation}\label{xlogx1}
x^c \log \ml \frac{1}{x}\mr \leq \frac{1}{ce}.
\end{equation}
\end{lemma}

\begin{proof}
Let $f(x) = x^c \log \ml \frac{1}{x}\mr$. Observe that $f'(x_*) = x_*^{-1+c}\ml -1 + c\log (\frac{1}{x_*})\mr = 0$ when $x_* = e^{-1/c}$. Furthermore, we have $f''(x_*) = x_*^{-2+c}\ml 1-2c-(1-c)c\log(\frac{1}{x_*})\mr = - c e^{-1+2/c} < 0$ for all $c > 0$. This implies that $f$ attains a maximum at $x_* = e^{-1/c} \in (0,1)$, and thus 
\begin{equation}\label{xlogx2}
\begin{aligned}
f(x) &\leq \ml e^{-1/c}\mr^c \log \ml e^{1/c}\mr = \frac{1}{ce}
\end{aligned}
\end{equation}
all $x \in (0,1)$.
\end{proof}

\begin{remark}
Lemma \ref{xlogx} allows us to absorb any logarithmic term $\log \ml \frac{1}{x} \mr$ by $x^c$ for an arbitrarily small $c > 0$ with an explicit upper bound. 
\end{remark}

With the above preparations, we proceed to derive some novel estimates for $|\phi|$ and $|Q|$, which will be heavily used in subsequent sections.  For notational simplicity, for each $(\overline{u},v) \in \mathcal{R}$, we define
\begin{equation}\label{overlineP}
\overline{P}(\overline{u},v) := \sup_{v'\in(\overline{u},v]} |\phi|(\overline{u},v') < +\infty.
\end{equation}
And we have

\begin{proposition}\label{PropEstimates} For $(u,v) \in \mathcal{R}$, if there exists a $0 \leq \overline{u} < u$ such that $(u',v') \in \mathcal{R}$ for all  $u' \in [\overline{u},u]$ and $v' \in (u',v]$, then for each given $\chi \in (0,1)$,  we have 
\begin{equation}\label{Estimates13}
\begin{aligned}
|\phi|(u,v) \leq \overline{P}(\overline{u},v) + \frac{1}{(4\pi)^{\frac{1}{2}}}\ml  \log \ml \frac{\frac{1 - \mu}{ \dv r}(u,v)}{\frac{1 - \mu}{ \dv r}(\overline{u},v)}\mr  \mr^{\frac{1}{2}}\ml \log \ml \frac{r(\overline{u},v)}{r(u,v)}\mr \mr^{\frac{1}{2}}
\end{aligned}
\end{equation}
and
\begin{equation}\label{Estimates14}
\begin{aligned}
|Q|(u,v) 
&\leq C_1(\overline{u},v;\chi) \cdot \mu^{\oh}(u,v) \cdot r^{\frac{3}{2} - \frac{\chi}{2}}(u,v) \cdot r^{\frac{\chi}{2}}(\overline{u},v) \cdot (v - u)^{\oh}
\end{aligned}
\end{equation}
with constant $C_1(\overline{u},v;\chi)$ given in \eqref{Estimates11}. In particular, we have
\begin{equation}\label{Estimates17}
\frac{Q^2}{r^2}(u,v) \leq C_1^2(\overline{u},v;\chi) \cdot \mu(u,v) \cdot r^{1 - \chi}(u,v) \cdot r^{\chi}(\overline{u},v) \cdot (v - u).
\end{equation}
Furthermore, if $Q$ is paired with $\phi$, we obtain an improved estimate
\begin{equation}\label{Estimates19}
|Q\phi|(u,v) \leq C_5(\overline{u},v;\chi,\xi_1,\xi_2) \cdot \mu(u,v)^{\frac{1}{2\xi_1}} \cdot r^{\frac{3}{2}- \frac{\chi}{2} -\frac{1}{2\xi_1}(1+\xi_2)}(u,v),
\end{equation}
where the constant $C_5$ (depending on $\overline{u},v,\chi,\xi_1$, $\xi_2$) is given in \eqref{Estimates33},
with any $\xi_1 \geq 1$ and $\xi_2 \in (0,1-\chi)$.
\end{proposition}

\begin{figure}[htbp]
\begin{tikzpicture}
\begin{scope}[thick]
\fill[gray!50] (0,1) to (2,3) to (1.2,3.8) to (0,2.6) to (0,1);
\draw[->] (0,0) node[anchor=north]{$\mathcal{O}$} --  (0,6) node[anchor = east]{$\Gamma$};
\draw[->] (0,0) -- (5,5) node[anchor = west]{$C_0^+$};
\draw[thick, dashdotted]  (4.9,5.1) to[in=-120,out= -140] (2.4,2.9) to[out=60,in=0] (2,4) to[out=180,in=-45] (0,4.6);
\draw[thick] (1.2,3.8) node[anchor = east]{$p$} -- (2,3)node[anchor = south]{$\overline{p}$};
\node[] at (0.5,1) {$\mathcal{R}$};
\node[] at (3.5,4.5) {$\mathcal{T}$};
\node[] at (2,4.3) {$\mathcal{A}$};
\node[] at (-0.3,1) {$\overline{p}'$};
\node[] at (-0.3,3) {$p'$};
\draw[thick] (0,1) -- (2,3);
\draw[thick] (0,2.6) -- (1.2,3.8);
\node at (1.2,3.8)[circle,fill,inner sep=1.5pt]{};
\node at (2,3)[circle,fill,inner sep=1.5pt]{};
\end{scope}
\end{tikzpicture}
\caption{For a given point $p = (u,v) \in \mathcal{R}$, let $\overline{p} = (\overline{u},v)$ with $0 \leq \overline{u} < u$ be also a point in $\mathcal{R}$. Denote $p'p$ and $\overline{p}'\overline{p}$ as outgoing null curves emitting from $p'$ and $\overline{p}'$ along $\Gamma$ respectively. We then have that the shaded region indicated above must be in $\mathcal{R}$. To derive pointwise estimates at $p = (u,v)$, we consider the data along $\overline{p}'\overline{p}$ as the initial data.}
\label{SetupFig3}
\end{figure}
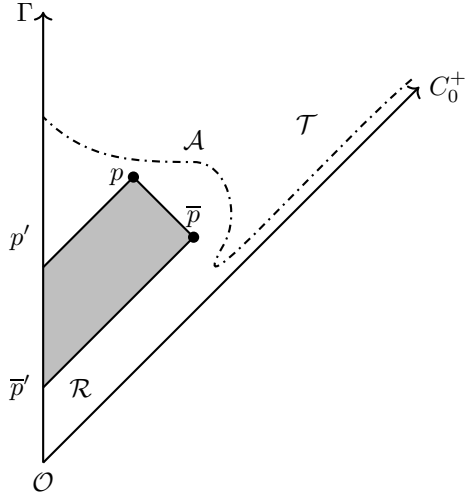

\begin{remark}\label{Remark mu Q}
A direct consequence of \eqref{Estimates17} would be that, for $\mu - Q^2/r^2$, it yields
\begin{equation}\label{Estimates18}
\ml \mu - \frac{Q^2}{r^2} \mr(u,v) \geq \mu(u,v) \ml 1 - C_1^2(\overline{u},v;\chi) \cdot r^{1 - \chi}(u,v) \cdot r^{\chi}(\overline{u},v) \cdot (v - u) \mr.
\end{equation}
Hence, we have $\mu - \frac{Q^2}{r^2} \geq 0$ if $C_1^2(\overline{u},v;\chi) \cdot r^{1 - \chi}(u,v) \cdot r^{\chi}(\overline{u},v) \cdot (v - u) \leq 1$, which naturally holds for regions sufficiently close to $\Gamma$. The sign of $\mu - Q^2/r^2$ is essential for this study as it is employed repeatedly throughout the entire paper. Furthermore, from the explicit expression of $C_1$ provided later in \eqref{Estimates11}, we can see that although $C_1(\overline{u},v;\chi)$ is a constant depending on $\overline{u}$ and $v$, it is possible to obtain an upper bound for $C_1$ that is independent of $\overline{u}$ and $v$. This will be addressed in the individual sections on a case-by-case basis.
\end{remark}

\begin{proof}
We begin with estimating $\phi$ as follows. Fix any $(u,v) \in \mathcal{R}$. By the hypothesis, there exists a $\overline{u} \in [0,u)$ such that $(u',v') \in \mathcal{R}$ for all $u' \in [\overline{u},u]$ and $v' \in [u',v']$. (This is illustrated as the shaded region in Figure \ref{SetupFig3}.) By the definition of $D_u,$ we have
$\du \phi + \ii \e A_u \phi = D_u \phi$. This implies that
\begin{equation}\label{Estimates16}
\du \ml \phi e^{\int^{u}_{\overline{u}} \ii \e A_u (u',v) \; \D u'}\mr = e^{\int^{u}_{\overline{u}} \ii \e A_u (u',v) \; \D u'} D_u \phi
\end{equation}
for any $\overline{u}$. By conducting integration with respect to $u$ from $\overline{u}$, we obtain a gauge invariant estimate for $\phi$
\begin{equation}\label{Estimates2}
|\phi|(u,v) \leq |\phi| (\overline{u},v) + \int^u_{\overline{u}} |D_u \phi| (u',v) \mathrm{d}u',
\end{equation}
where we have used the fact that $\int^{u}_{\overline{u}} \ii \e A_u (u',v) \; \D u'$ is purely imaginary. Next, notice that \eqref{Setup21} is equivalent to the following equation:
\begin{equation}\label{Estimates3}
\begin{aligned}
\du \ml \log \ml \frac{1-\mu}{\dv r}\mr\mr &= \frac{4\pi r |D_u \phi|^2}{-\du r}.
\end{aligned}
\end{equation}
Using \eqref{Estimates2} and the Cauchy-Schwarz inequality, we deduce  
\begin{equation}\label{Estimates5}
\begin{aligned}
|\phi|(u,v) &\leq  |\phi| (\overline{u},v) + \frac{1}{(4\pi)^{\frac{1}{2}}}\ml \int^u_{\overline{u}} \frac{4 \pi r|D_u \phi|^2}{-\du r} (u',v) \mathrm{d}u' \mr^{\frac{1}{2}} \ml \int^u_{\overline{u}} \frac{-\du r}{r} (u',v) \mathrm{d}u' \mr^{\frac{1}{2}} \\
&\leq \overline{P}(\overline{u},v) + \frac{1}{(4\pi)^{\frac{1}{2}}}\ml  \log \ml \frac{\frac{1 - \mu}{ \dv r}(u,v)}{\frac{1 - \mu}{ \dv r}(\overline{u},v)}\mr  \mr^{\frac{1}{2}}\ml \log \ml \frac{r(\overline{u},v)}{r(u,v)}\mr \mr^{\frac{1}{2}}.
\end{aligned}
\end{equation}

To estimate $Q$, we shall employ \eqref{Setup19}. Given any point $(u,v) \in \mathcal{R}$, we perform an integration along an outgoing null curve from $\Gamma$ to $(u,v)$. Together with the fact that $Q|_{\Gamma} = 0$ from \eqref{Setup36} and utilizing our bounds in \eqref{Estimates5}, we then derive 
\begin{equation}\label{Estimates6}
\begin{aligned}
|Q|(u,v) &\leq \int^v_u 4\pi \e r^{2}|\phi||\dv \phi| (u,v') \mathrm{d}v' \\
&\leq 4\pi \e  \int^v_u r^{2}\ml \overline{P}(\overline{u},v') + \frac{1}{(4\pi)^{\frac{1}{2}}}\ml  \log \ml \frac{\frac{1 - \mu}{ \dv r}(u,v')}{\frac{1 - \mu}{ \dv r}(\overline{u},v')}\mr  \mr^{\frac{1}{2}}\ml \log \ml \frac{r(\overline{u},v')}{r(u,v')}\mr \mr^{\frac{1}{2}}\mr|\dv \phi| (u,v') \mathrm{d}v'  \\
&= 4\pi\e (Q_1(u,v) + Q_2(u,v))
\end{aligned}
\end{equation}
with
\begin{equation}\label{Estimates7}
\begin{aligned}
Q_1(u,v) &:= \overline{P}(\overline{u},v) \int^v_u r^2 |\dv \phi|(u,v') \mathrm{d}v',  \\
Q_2(u,v) &:= \frac{1}{(4\pi)^{\frac{1}{2}}} \int^v_u \ml \log \ml \frac{\frac{1 - \mu}{ \dv r}(u,v')}{\frac{1 - \mu}{\dv r}(\overline{u},v')}\mr  \mr^{\frac{1}{2}} \ml \log \ml \frac{r(\overline{u},v')}{r(u,v')}\mr \mr^{\frac{1}{2}} r^2|\dv \phi| (u,v') \mathrm{d}v'.
\end{aligned}
\end{equation}
For $Q_1$, using \eqref{Estimates1}, \eqref{Setup25} and $\dv r(u,v) > 0$, we obtain
\begin{equation}\label{Estimates8}
\begin{aligned}
Q_1(u,v) &= \overline{P}(\overline{u},v) \int^v_u r^2 |\dv \phi|(u,v') \mathrm{d}v'  \\
&= \frac{\overline{P}(\overline{u},v)}{(2\pi)^{\frac{1}{2}}}\ml \int^v_u \frac{2\pi r^2(1-\mu)|\dv \phi|^2}{\dv r} (u,v') \mathrm{d}v'\mr^{\frac{1}{2}} \ml \int^v_u \frac{ r^2 \dv r}{1 - \mu} (u,v') \mathrm{d}v'\mr^{\frac{1}{2}} \\
&\leq \frac{\overline{P} \bar{B}^{\frac{1}{2}}}{{(2\pi (1-\overline{\mu_*}))^{\oh}}} (\overline{u},v) \cdot r(u,v) \cdot (v-u)^{\oh}  \cdot \ml \int^v_u \dv m (u,v') \mathrm{d}v'\mr^{\oh} \\
&\leq  \frac{\overline{P} \bar{B}^{\frac{1}{2}}}{{(4\pi (1-\overline{\mu_*}))^{\oh}}} (\overline{u},v) \cdot  r^{\frac{3}{2}}(u,v) \cdot \mu^{\oh}(u,v) \cdot (v-u)^{\oh}.
\end{aligned}
\end{equation}
To bound $Q_2$, we shall rely on Lemma \ref{xlogx} and we will crucially employ the structure of $\frac{\partial_v r}{1-\mu}$-terms after using the equation of $\partial_v m$ and a H\"older's inequality. Also noting that the terms in the arguments of the logarithm functions in \eqref{Estimates7} are greater than or equal to $1$, we then derive
\begin{equation}\label{Estimates9}
\begin{aligned}
Q_2(u,v) = &\; \frac{1}{(4\pi)^{\frac{1}{2}}} \int^v_u \ml \log \ml \frac{\frac{1 - \mu}{\dv r}(u,v')}{\frac{1 - \mu}{\dv r}(\overline{u},v')}\mr  \mr^{\frac{1}{2}} \ml \log \ml \frac{r(\overline{u},v')}{r(u,v')}\mr \mr^{\frac{1}{2}}r^2|\dv \phi| (u,v') \mathrm{d}v' \\
\leq &\; \frac{1}{(8\pi^2)^{\frac{1}{2}}}\ml \int^v_u \frac{2\pi r^2(1-\mu)|\dv \phi|^2}{\dv r} (u,v') \mathrm{d}v'\mr^{\frac{1}{2}} \\
&\times \ml \int^v_u \frac{ r^2 \dv r}{1 - \mu}(u,v') \cdot \ml \log \ml \frac{\frac{1 - \mu}{ \dv r}(u,v')}{\frac{1 - \mu}{ \dv r}(\overline{u},v')}\mr  \mr \cdot \ml \log \ml \frac{r(\overline{u},v')}{r(u,v')}\mr \mr \mathrm{d}v'\mr^{\frac{1}{2}} \\
\leq &\; \frac{1}{(8\pi^2)^{\frac{1}{2}}} m^{\oh}(u,v) \\
&\times \ml \int^v_u r^2 (u,v') \cdot \ml \frac{\dv r}{1 - \mu}(\overline{u},v') \mr \cdot 
\ml \frac{\frac{\dv r}{1 - \mu}(u,v')}{\frac{\dv r}{1 - \mu}(\overline{u},v')}\mr \cdot \ml \log \ml \frac{\frac{1 - \mu}{ \dv r}(u,v')}{\frac{1 - \mu}{ \dv r}(\overline{u},v')}\mr  \mr \cdot \ml \log \ml \frac{r(\overline{u},v')}{r(u,v')}\mr \mr \mathrm{d}v'\mr^{\frac{1}{2}} \\
\leq &\; \frac{\bar{B}^{\frac{1}{2}}}{(16\pi^2 e (1-\overline{\mu_*}))^{\frac{1}{2}}}(\overline{u},v) \cdot \mu^{\oh}(u,v)r^{\oh}(u,v) \cdot \ml \int^v_u r^{2-\chi} (u,v') \cdot r^{\chi}(\overline{u},v') \cdot \ml \frac{r^{\chi}(u,v')}{r^{\chi}(\overline{u},v')}\mr \cdot
\ml \log \ml \frac{r(\overline{u},v')}{r(u,v')}\mr \mr \mathrm{d}v'\mr^{\frac{1}{2}} \\
\leq &\; \frac{\bar{B}^{\frac{1}{2}}}{(16\pi^2 \chi e^2 (1-\overline{\mu_*}))^{\frac{1}{2}}}(\overline{u},v) \cdot \mu^{\oh}(u,v) \cdot r^{\frac{3}{2} - \frac{\chi}{2}}(u,v) \cdot r^{\frac{\chi}{2}}(\overline{u},v) \cdot (v - u)^{\oh}. \\
\end{aligned}
\end{equation}
Here, $\chi \in (0,2)$ is a fixed constant to be chosen.
Combining \eqref{Estimates8} and \eqref{Estimates9}, we simplify \eqref{Estimates6} as
\begin{equation}\label{Estimates10}
\begin{aligned}
|Q|(u,v) &\leq 4\pi\e \ml \frac{\overline{P}\bar{B}^{\frac{1}{2}}}{(4\pi (1-\overline{\mu_*}))^{\oh}} + \frac{\bar{B}^{\frac{1}{2}}}{(16\pi^2 \chi e^2 (1-\overline{\mu_*}))^{\oh}}\mr (\overline{u},v) \cdot \mu^{\oh}(u,v) \cdot r^{\frac{3}{2} - \frac{\chi}{2}}(u,v) \cdot r^{\frac{\chi}{2}}(\overline{u},v) \cdot (v - u)^{\oh}   \\
&= C_1(\overline{u},v;\chi) \cdot \mu^{\oh}(u,v) \cdot r^{\frac{3}{2} - \frac{\chi}{2}}(u,v) \cdot r^{\frac{\chi}{2}}(\overline{u},v) \cdot (v - u)^{\oh} 
\end{aligned}
\end{equation}
where
\begin{equation}\label{Estimates11}
C_1(\overline{u},v;\chi) := \frac{4\pi\e \bar{B}^{\oh}}{(1-\overline{\mu_*})^{\oh}} (\overline{u},v) \cdot \ml \frac{\overline{P}(\overline{u},v)}{(4\pi)^{\oh}} + \frac{1}{(16\pi^2 \chi e^2)^{\oh}}\mr.
\end{equation}
In particular, \eqref{Estimates10} is equivalent to  \eqref{Estimates14}, and \eqref{Estimates17} follows immediately.

To derive \eqref{Estimates19}, we employ a trick via considering $(Q\phi)^{\xi_1}$ for some $\xi_1 \geq 1$. By the fact that $Q\phi|_\Gamma = 0$ and $m|_{\Gamma} = 0$, via fundamental theorem of calculus and \eqref{Setup25}, we have
\begin{equation}\label{Estimates20}
\begin{aligned}
|Q\phi|^{\xi_1}(u,v)  &\leq \int^v_u \xi_1 |Q\phi|^{\xi_1 - 1} \dv(Q\phi) (u,v') \D v' \\
&\leq \int^v_u \xi_1  (4\pi \e r^2 |\phi|^{\xi_1 + 1} |Q|^{\xi_1 - 1} + |\phi|^{\xi_1 - 1}|Q|^{\xi_1})|\dv \phi| (u,v') \D v' \\
&\leq \xi_1  \ml \int^v_u \frac{\dv r}{2\pi r^2 (1-\mu)}(4\pi \e r^2 |\phi|^{\xi_1 + 1} |Q|^{\xi_1 - 1} + |\phi|^{\xi_1 - 1}|Q|^{\xi_1})^2  \D v' \mr^{\frac{1}{2}} \ml \int^v_u  \frac{2 \pi r^2(1-\mu)|\dv \phi|^2}{\dv r} \D v'\mr^{\frac{1}{2}}\\
&\leq \xi_1  \ml \int^v_u \frac{\dv r}{2\pi r^2 (1-\mu)}(4\pi \e r^2 |\phi|^{\xi_1 + 1} |Q|^{\xi_1 - 1} + |\phi|^{\xi_1 - 1}|Q|^{\xi_1})^2  (u,v') \D v' \mr^{\frac{1}{2}} m(u,v)^{\frac{1}{2}}\\
&\leq \frac{\xi_1}{\sqrt{2}} \cdot \mu(u,v)^{\frac{1}{2}} \cdot r(u,v)^{\frac{1}{2}} \cdot \ml \int^v_u \frac{\dv r}{2\pi r^2 (1-\mu)}(4\pi \e r^2 |\phi|^{\xi_1 + 1} |Q|^{\xi_1 - 1} + |\phi|^{\xi_1 - 1}|Q|^{\xi_1})^2  (u,v') \D v' \mr^{\frac{1}{2}}.
\end{aligned}
\end{equation}
The first integral on the right can be bounded as follows
\begin{equation}\label{Estimates21}
\begin{aligned}
& \int^v_u \frac{\dv r}{2\pi r^2 (1-\mu)}(4\pi \e r^2 |\phi|^{\xi_1 + 1} |Q|^{\xi_1 - 1} + |\phi|^{\xi_1 - 1}|Q|^{\xi_1})^2  \D v' \\
\leq &\; 16 \pi \e^2 \int^v_u \frac{\dv r}{1-\mu}r^2 |\phi|^{2\xi_1 + 2} |Q|^{2\xi_1 - 2} (u,v') \D v' + \frac{1}{\pi} \int^v_u \frac{\dv r}{1-\mu} \frac{1}{r^2} |\phi|^{2\xi_1 -2} |Q|^{2\xi_1 } (u,v') \D v'.
\end{aligned}
\end{equation}
The first term in the upper bound in \eqref{Estimates21} can be estimated in the same fashion as in \eqref{Estimates9} via using \eqref{Estimates13} and \eqref{Estimates14}. Indeed, for $v' \in (u',v]$, using the fact that $\overline{P}(\overline{u},\cdot)$ and $C_1(\overline{u},\cdot;\chi)$ are non-decreasing, $\dv r > 0$ and $\mu < 1$, we get
\begin{equation}\label{Estimates22}
|\phi|(u,v') \leq \overline{P}(\overline{u},v) + \frac{1}{(4\pi)^{\frac{1}{2}}} \ml  \log \ml \frac{\frac{1 - \mu}{ \dv r}(u,v')}{\frac{1 - \mu}{ \dv r}(\overline{u},v')}\mr  \mr^{\frac{1}{2}} \cdot \ml \log \ml \frac{r(\overline{u},v')}{r(u,v')}\mr \mr^{\frac{1}{2}},
\end{equation}
and
\begin{equation}\label{Estimates23}
|Q|(u,v') \leq C_2(\overline{u},v;\chi) \cdot r^{\frac{3}{2}-\frac{\chi}{2}}(u,v'),
\end{equation}
with
\begin{equation}\label{Estimates24}
\begin{aligned}
C_{2}(\overline{u},v) &= C_1(\overline{u},v;\chi) \cdot r^{\frac{\chi}{2}}(\overline{u},v) \cdot (v-\overline{u})^{\frac{1}{2}}.
\end{aligned}
\end{equation}
These in turn imply that
\begin{equation}\label{Estimates34}
\begin{aligned}
|\phi|^{2\xi_1 -2}(u,v') &\leq 2^{2\xi_1 -3}(1+\overline{P}(\overline{u},v))^{2\xi_1 -2} + \frac{1}{2\pi^{\xi_1 - 1}}\ml  \log \ml \frac{\frac{1 - \mu}{ \dv r}(u,v')}{\frac{1 - \mu}{ \dv r}(\overline{u},v')}\mr  \mr^{\xi_1 - 1} \cdot \ml \log \ml \frac{r(\overline{u},v')}{r(u,v')}\mr \mr^{\xi_1 - 1}, \\
|\phi|^{2\xi_1 +2}(u,v') &\leq 2^{2\xi_1 +1}(1+\overline{P}(\overline{u},v))^{2\xi_1 +2} + \frac{1}{2\pi^{\xi_1 + 1}}\ml  \log \ml \frac{\frac{1 - \mu}{ \dv r}(u,v')}{\frac{1 - \mu}{ \dv r}(\overline{u},v')}\mr  \mr^{\xi_1 + 1} \cdot \ml \log \ml \frac{r(\overline{u},v')}{r(u,v')}\mr \mr^{\xi_1 + 1}
\end{aligned}
\end{equation}
by appealing to the inequality $(a+b)^p \leq 2^{p-1}(a^p + b^p)$ for $a,b \geq 0$ and $p \geq 1$.

We start by dealing with the first term on the right of \eqref{Estimates21} as follows
\begin{equation}\label{Estimates25}
\begin{aligned}
& \int^v_u \frac{\dv r}{1-\mu}r^2 |\phi|^{2\xi_1 + 2} |Q|^{2\xi_1 - 2} (u,v') \D v' \\
\leq &\; C_2^{2\xi_1 - 2} 2^{2\xi_1 + 1}(1+ \overline{P}(\overline{u},v))^{2\xi_1 + 2 }\int^v_u \frac{\dv r}{1-\mu}(u,v')r^{(\xi_1 - 1) \cdot \ml 3 - \chi\mr + 2}(u,v')   \D v' \\
&+ \;  \frac{C_2^{2\xi_1 - 2}}{2 \pi^{\xi_1 + 1}}\int^v_u \frac{\dv r}{1-\mu}(u,v') \cdot r^{(\xi_1 - 1)\ml 3 - \chi\mr + 2}(u,v') \cdot \ml  \log \ml \frac{\frac{1 - \mu}{ \dv r}(u,v')}{\frac{1 - \mu}{ \dv r}(\overline{u},v')}\mr  \mr^{\xi_1 + 1} \cdot \ml \log \ml \frac{r(\overline{u},v')}{r(u,v')}\mr \mr^{\xi_1 + 1}  \D v'. \\
\end{aligned}
\end{equation}
By the monotonicity of $\frac{\dv r}{1-\mu}$ along future-directed incoming null directions, for any $\xi_2 \in (0,2)$, we have
\begin{equation}\label{Estimates35}
\begin{aligned}
& C_2^{2\xi_1 - 2} 2^{2\xi_1 + 1}(1+ \overline{P}(\overline{u},v))^{2\xi_1 + 2 }\int^v_u \frac{\dv r}{1-\mu}(u,v') \cdot r^{(\xi_1 - 1)\ml 3 - \chi\mr + 2}(u,v')   \D v' \\
\leq &\; C_2^{2\xi_1 - 2} 2^{2\xi_1 + 1}(1+ \overline{P}(\overline{u},v))^{2\xi_1 + 2 } \frac{\bar{B}(\overline{u},v) r^{\xi_2}(\overline{u},v)}{(1-\overline{\mu_*})(\overline{u},v)} \cdot (v - u) \cdot r^{(\xi_1 - 1)(3-\chi) + 2 - \xi_2}(u,v).
\end{aligned}
\end{equation}
By Lemma \ref{xlogx}, we further obtain
\begin{equation}\label{Estimates26}
\frac{\dv r}{1-\mu}(u,v') \ml  \log \ml \frac{\frac{1 - \mu}{ \dv r}(u,v')}{\frac{1 - \mu}{ \dv r}(\overline{u},v')}\mr  \mr^{\xi_1 + 1} \leq \frac{\dv r}{1-\mu}(\overline{u},v') \ml \frac{\xi_1 + 1}{e}\mr^{\xi_1 + 1},
\end{equation}
and
\begin{equation}\label{Estimates27}
\ml \log \ml \frac{r(\overline{u},v')}{r(u,v')}\mr \mr^{\xi_1 + 1}  \leq  r^{-\xi_2}(u,v') r^{\xi_2}(\overline{u},v') \ml \frac{\xi_1 + 1}{\xi_2 e}\mr^{\xi_1 + 1}.
\end{equation}
Consequently, the second term in  \eqref{Estimates25} simplifies to 
\begin{equation}\label{Estimates28}
\begin{aligned}
& \frac{C_2^{2\xi_1 - 2}}{2 \pi^{\xi_1 + 1}}\int^v_u \frac{\dv r}{1-\mu}(u,v')r^{(\xi_1 - 1)\ml 3 - \chi\mr + 2}(u,v') \ml  \log \ml \frac{\frac{1 - \mu}{ \dv r}(u,v')}{\frac{1 - \mu}{ \dv r}(\overline{u},v')}\mr  \mr^{\xi_1 + 1}\ml \log \ml \frac{r(\overline{u},v')}{r(u,v')}\mr \mr^{\xi_1 + 1}  \D v' \\
\leq &\; \frac{C_2^{2\xi_1 - 2}}{2\pi^{\xi_1 + 1}}\frac{\bar{B}(\overline{u},v)r^{\xi_2}(\overline{u},v)  }{(1-\overline{\mu_*})(\overline{u},v)} \ml \frac{\xi_1 + 1}{ e}\mr^{\xi_1 + 1}  
 \ml \frac{\xi_1 + 1}{\xi_2 e}\mr^{\xi_1 + 1} (v-u) r^{(\xi_1-1)(3-\chi) + 2 - \xi_2}(u,v)
\end{aligned}
\end{equation}
for all possible $\xi_1 \geq 1$ and $\chi \in (0,1)$. Thus, we arrive at 
\begin{equation}\label{Estimates36}
\begin{aligned}
\int^v_u \frac{\dv r}{1-\mu}r^2 |\phi|^{2\xi_1 + 2} |Q|^{2\xi_1 - 2} (u,v') \D v' 
\leq C_3(\overline{u},v;\chi,\xi_1,\xi_2) \cdot r^{(\xi_1-1)(3-\chi) + 2 - \xi_2}(u,v)
\end{aligned}
\end{equation}
with
\begin{equation}\label{Estimates37}
C_3 :=   C_2^{2\xi_1 - 2}\cdot 
 (v - \overline{u}) \frac{\bar{B}(\overline{u},v)r^{\xi_2}(\overline{u},v)  }{(1-\overline{\mu_*})(\overline{u},v)}\ml 2^{2\xi_1 + 1}(1+ \overline{P}(\overline{u},v))^{2\xi_1 + 2 } 
 + \frac{1}{2\pi^{\xi_1 + 1}}\ml \frac{\xi_1 + 1}{ e}\mr^{\xi_1 + 1}  
 \ml \frac{\xi_1 + 1}{\xi_2 e}\mr^{\xi_1 + 1} \mr.
\end{equation}
Back to the second term on the right of \eqref{Estimates21}, by an analogous argument, one can show that if $\xi_2 \in (0,1-\chi)$, it then holds
\begin{equation}\label{Estimates29}
\begin{aligned}
\int^v_u \frac{\dv r}{1-\mu} \frac{1}{r^2} |\phi|^{2\xi_1 -2} |Q|^{2\xi_1 } (u,v') \D v' 
\leq  C_4(\overline{u},v;\chi,\xi_1,\xi_2) \cdot r^{\xi_1(3-\chi) - 2 - \xi_2}(u,v),
\end{aligned}
\end{equation}
with
\begin{equation}\label{Estimates38}
C_4 :=  C_2^{2\xi_1}\cdot 
 (v - \overline{u}) \frac{\bar{B}(\overline{u},v)r^{\xi_2}(\overline{u},v)  }{(1-\overline{\mu_*})(\overline{u},v)}\ml 2^{2\xi_1 - 3}(1+ \overline{P}(\overline{u},v))^{2\xi_1 - 2 } 
 + \frac{1}{2\pi^{\xi_1 - 1}}\ml \frac{\xi_1 - 1}{ e}\mr^{\xi_1 - 1}  
 \ml \frac{\xi_1 - 1}{\xi_2 e}\mr^{\xi_1 - 1} \mr.
\end{equation}
Going back to \eqref{Estimates21}, by applying \eqref{Estimates36} and \eqref{Estimates29}, we also obtain
\begin{equation}\label{Estimates30}
\begin{aligned}
& \int^v_u \frac{\dv r}{2\pi r^2 (1-\mu)}(4\pi \e r^2 |\phi|^{\xi_1 + 1} |Q|^{\xi_1 - 1} + |\phi|^{\xi_1 - 1}|Q|^{\xi_1})^2  \D v' \\
\leq &\; \ml 16 \pi \e^2 C_3 r^{1+\chi}(\overline{u},v) + \frac{C_4}{\pi} \mr r^{\xi_1(3-\chi)-2-\xi_2}(u,v).
\end{aligned}
\end{equation}
We then apply \eqref{Estimates30} to \eqref{Estimates20} and get
\begin{equation}\label{Estimates32}
\begin{aligned}
|Q\phi|^{\xi_1}(u,v) &\leq \frac{\xi_1}{\sqrt{2}} 
 \ml 16 \pi \e^2 C_3 r^{1+\chi}(\overline{u},v) + \frac{C_4}{\pi} \mr^{\frac{1}{2}}
\mu(u,v)^{\frac{1}{2}} \cdot r^{\xi_1 \ml \frac{3}{2}- \frac{\chi}{2} \mr -\frac{1}{2}(1+\xi_2)}(u,v), \\
\end{aligned}
\end{equation}
and
\begin{equation}\label{Estimates32a}
\begin{aligned}
|Q\phi|(u,v) &\leq C_5(\overline{u},v;\chi,\xi_1,\xi_2) \cdot \mu(u,v)^{\frac{1}{2\xi_1}} \cdot 
r^{\frac{3}{2}- \frac{\chi}{2} -\frac{1}{2\xi_1}(1+\xi_2)}(u,v),
\end{aligned}
\end{equation}
with
\begin{equation}\label{Estimates33}
C_{5}(\overline{u},v;\chi,\xi_1,\xi_2) :=  \ml \frac{\xi_1}{\sqrt{2}} 
 \ml 16 \pi \e^2 C_3(\overline{u},v;\chi,\xi_1,\xi_2) \cdot r^{1+\chi}(\overline{u},v) + \frac{C_4 (\overline{u},v;\chi,\xi_1,\xi_2)}{\pi} \mr^{\frac{1}{2}} \mr^{\frac{1}{\xi_1}}.
\end{equation}
\end{proof}

\section{Trapped Surface Formation}\label{TSF}

In this section, we shall provide a novel proof and establish a new trapped surface formation criterion for the Einstein-matter fields. In particular, we prove a desired result for the charged case. This will play a crucial role in subsequent sections. If we set $Q\equiv 0$, in the appendix, we also give an approximately 1.5-page new proof of the sharp criterion obtained by Christodoulou for the Einstein-(real) scalar field system. Inspired by \cite{christ1}, for a given coordinate patch $[u_1,u_2]\times[v_1,v_2]$ with $u_1 < u_2$ and $v_1 < v_2$, we denote the dimensionless radius ratio $\delta(u)$ by
\begin{equation}\label{RadiusRatio}
\delta(u) = \frac{r(u,v_2) - r(u,v_1)}{r(u,v_1)}.
\end{equation}
We also define the dimensionless  
\textit{reduced mass ratio} as 
\begin{equation}\label{ReducedMassRatio}
\eta^*(u) := \frac{2 \ml m(u,v_2) - m(u,v_1) \mr}{r(u,v_2)} - \frac{2}{r(u,v_2)}\int^{v_2}_{v_1} \frac{Q^2(\dv r)}{2r^2}(u,v')\mathrm{d}v'.
\end{equation}
Note that the additional term containing $Q^2$ is of a special form and is not from the modified Hawking mass. This additional term will play an important role in the later-established instability theorems.

\begin{figure}[htbp]
\begin{tikzpicture}[
    dot/.style = {draw, fill = white, circle, inner sep = 0pt, minimum size = 4pt}]
\begin{scope}[thick]
\fill[gray!50] (1,4) to (2.5,5.5) to (3,5) to (1.5,3.5) to (1,4);
\fill[gray!80] (0.2,4.8) to (1.2,5.8) to (2,5) to (1,4) to (0.2,4.8);
\draw[thick] (0,0) node[anchor=north]{$\mathcal{O}$} --  (0,5);
\node[] at (-0.3,2.5) {$\Gamma$};
\node[] at (-0.3,5.0) {$\mathcal{O}'$};
\draw[->] (0,0) -- (6,6) node[anchor = west]{$C_0^+$};
\draw[thick, densely dotted] (0,5) to (1.4,6.4);
\draw[thick, loosely dotted] (1.4,6.4) to [out=45,in=175] (4.7,7.0);
\node[rotate=45] at (0.6,5.9) {$\mathcal{B}_0$};
\node[rotate=0] at (4.3,7.3) {$\mathcal{B}\setminus\mathcal{B}_0$};
\draw[fill=white] (1.4,6.4) circle (2pt);
\draw[thick] (0,5) --  (2.5,2.5);
\node[] at (2.7,2.3){$v_1$};
\draw[thick,dashed] (4,4) --  (1.5,6.5);
\node[] at (4.2,3.8){$v_3$};
\draw[thick, dashdotted] (0.2,5.2) to [out=30,in=150] (3.5,5.5) to [out=-30,in=150] (3.2,4.2) [out=10] to [in=215] (5,5.1);
\node[] at (4.15,4.85){$\mathcal{A}$};
\draw[thick] (0.2,4.8) to (1.2,5.8);
\draw[thick] (1.2,5.8) to (2,5);
\draw[thick,dashed] (1,6) to (3.5,3.5);
\node[] at (3.7,3.3){$v_2$};
\draw[thick] (1,4) to (2.5,5.5) to (3,5) to (1.5,3.5);
\node[rotate = 45] at (0.75,5.05) {\small{$u = u_*$}};
\node[rotate = 45] at (1.4,4.6) {\small{$u = u_0$}};
\node[rotate = 45] at (1.9,4.1) {\small{$u = \overline{u}$}};
\node[rotate=45] at (1.1,4.9) {$R'$};
\node[rotate=45] at (1.6,4.2) {$R$};
\draw[fill=white] (0,5) circle (2pt);
\end{scope}
\end{tikzpicture}
\begin{tikzpicture}[
    dot/.style = {draw, fill = white, circle, inner sep = 0pt, minimum size = 4pt}]
\begin{scope}[thick]
\fill[gray!50] (1,4) to (2.5,5.5) to (3,5) to (1.5,3.5) to (1,4);
\fill[gray!80] (0.2,4.8) to (1.2,5.8) to (2,5) to (1,4) to (0.2,4.8);
\draw[thick] (0,0) node[anchor=north]{$\mathcal{O}$} --  (0,5);
\node[] at (-0.3,2.5) {$\Gamma$};
\node[] at (-0.3,5.0) {$\mathcal{O}'$};
\draw[->] (0,0) -- (6,6) node[anchor = west]{$C_0^+$};
\draw[thick, densely dotted] (0,5) to (0.8,5.8);
\draw[thick, loosely dotted] (0.8,5.8) to [out=45,in=175] (4.7,7.0);
\node[rotate=45] at (0.4,5.7) {$\mathcal{B}_0$};
\node[rotate=0] at (4.3,7.3) {$\mathcal{B}\setminus\mathcal{B}_0$};
\draw[fill=white] (0.8,5.8) circle (2pt);
\draw[thick] (0,5) --  (2.5,2.5);
\node[] at (2.7,2.3){$v_1$};
\draw[thick,dashed] (4,4) --  (1.5,6.5);
\node[] at (4.2,3.8){$v_3$};
\draw[thick, dashdotted] (0.2,5.2) to [out=30,in=150] (3.5,5.5) to [out=-30,in=150] (3.2,4.2) [out=10] to [in=215] (5,5.1);
\node[] at (4.15,4.85){$\mathcal{A}$};
\draw[thick] (0.2,4.8) to (1.2,5.8);
\draw[thick] (1.2,5.8) to (2,5);
\draw[thick,dashed] (1,6) to (3.5,3.5);
\node[] at (3.7,3.3){$v_2$};
\draw[thick] (1,4) to (2.5,5.5) to (3,5) to (1.5,3.5);
\node[rotate = 45] at (0.75,5.05) {\small{$u = u_*$}};
\node[rotate = 45] at (1.4,4.6) {\small{$u = u_0$}};
\node[rotate = 45] at (1.9,4.1) {\small{$u = \overline{u}$}};
\node[rotate=45] at (1.1,4.9) {$R'$};
\node[rotate=45] at (1.6,4.2) {$R$};
\draw[fill=white] (0,5) circle (2pt);
\end{scope}
\end{tikzpicture}
\caption{Diagram for the proof of Theorem \ref{Trapped}. Rectangles $R$ and $R'$ are shaded in a lighter and a darker tone respectively and the apparent horizon $\mathcal{A}$ is indicated as a dotted curve. The scenario on the left corresponds to $|\mathcal{B}_0| \geq v_2 - v_1$, while the scenario on the right corresponds to $|\mathcal{B}_0| < v_2 - v_1$, and $\mathcal{B}_0$ is allowed to be the empty set. }
\label{TrappedFig1}
\end{figure}
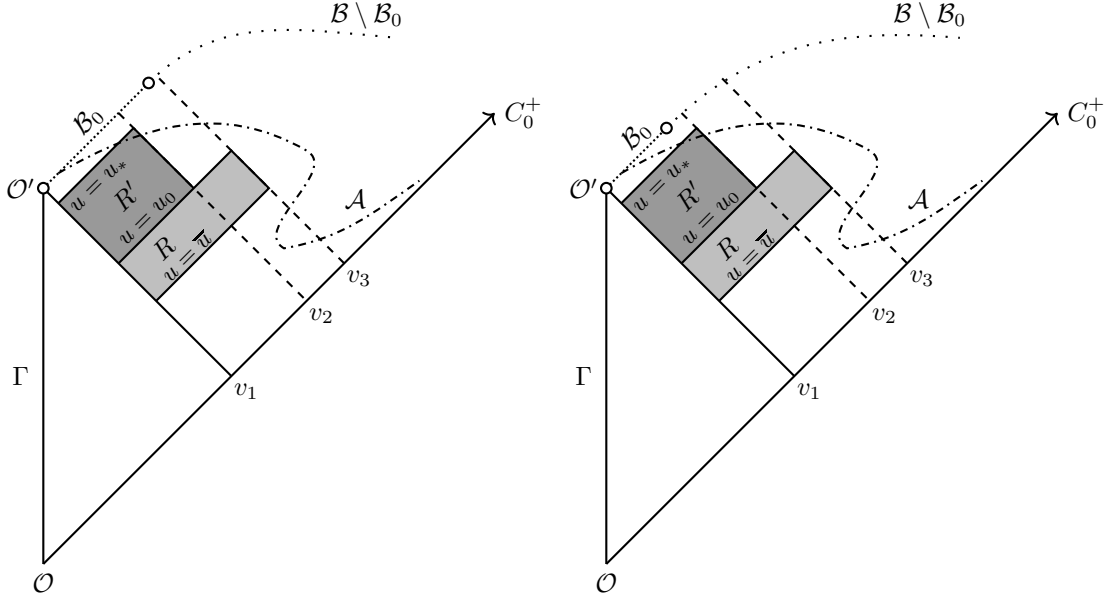

For notational convenience, here we define some commonly used abbreviations. A subscript with $1$ indicates that the quantity is evaluated at $v = v_1$. Similarly, a subscript with $2$ indicates that the quantity is evaluated at $v = v_2$. For a quantity with subscript $0$, this means that it is evaluated at $u = u_0$. For instance, we have $\delta_0 = \delta(u_0)$, $m_1(u) = m(u,v_1)$ and $r_2(u) = r(u,v_2)$. In addition, denote
\begin{equation}\label{TSF43}
\mu^*(u_0) := \sup_{v \in [v_1,v_2]} \mu(u_0,v).
\end{equation}

We are ready to state the below criterion, which is tailored for application in the upcoming instability arguments.
{\color{black}
\begin{theorem}\label{Trapped} (Trapped Surface Formation Criteria.)
Let $\mathcal{O}' = (v_1,v_1)$ be the future ending point of $\Gamma$ as portrayed in Figure \ref{TrappedFig1} and set $v_2 >  v_1$ and $u_0 < v_1$ to be sufficiently close to $v_1$. Suppose that there exist $\overline{u} \in [0,u_0)$ and $v_3 > v_2$ such that the coordinate patch $R:= [\overline{u},u_0] \times [v_1,v_3]$ lies in the regular region $\mathcal{R}$. If the reduced mass ratio $\eta^*(u_0)$ satisfies the lower-bound condition
\begin{equation}\label{TSF2}
\eta^*(u_0) > C_{24} \frac{\delta(u_0)}{1+\delta(u_0)} \log \ml \frac{1}{\delta(u_0)}\mr + C_{23} \ml \frac{\delta(u_0)r_1(u_0)}{1-\mu^*(u_0)}\mr,
\end{equation}
with $\delta(u_0) \in (0,1)$ and with constants $C_{23}$ and $C_{24}$ explicitly given in \eqref{TSF51} and \eqref{TSF54}, then a marginally outer trapped surface (MOTS) or a trapped surface is guaranteed to form along the incoming null curve $[u_0, u_*]\times \{v_2\}$ with $u_* \in (u_0,v_0)$.

\end{theorem}

\remark{In the proof of the later instability theorems, by taking into account the blueshift mechanism, for $u_0$ sufficiently close to $v_1$, we will see that $\frac{1}{1-\mu^*(u_0)}$} is of order $\log(\frac{1}{\delta(u_0)})$ and the gained factor $r_1(u_0)$ will play a crucial role.

\proof 
With $v_2 > v_1$, $u_0 < v_1$ and $u_* \in [u_0,v_1)$ being an arbitrary constant such that $(u_*,v_2) \in \mathcal{Q}$, we denote $R' := [u_0,u_*] \times [v_1,v_2]$. Under condition \eqref{TSF2}, we explain a mechanism in the regular (untrapped) region, that would lead to MOTS or trapped surface formation.

If there exists a MOTS or a trapped surface within $R'$, then the conclusion of Theorem \ref{Trapped} is achieved. We then proceed to assume the absence of a MOTS or a trapped surface in $R'$. By Remark \ref{Remark mu Q}, if $v_2$ and $u_0$ are picked sufficiently close to $v_1$, for each point $(u,v) \in R'$, the hypothesis of Proposition \ref{PropEstimates} is satisfied. Consequently, we have that \eqref{Estimates17} holds for any $\chi \in (0,1)$ and we obtain
\begin{equation}\label{TSF65}
\mu -\frac{Q^2}{r^2}(u,v) \geq \ml 1 - C_1^{2}(\overline{u},v;\chi)r^{1-\chi}(u,v)r^{\chi}(\overline{u},v)(v-u)\mr \mu(u,v).
\end{equation}
Here, $C_1(\overline{u},v;\chi)$ is defined in \eqref{Estimates11}. By the expression there, we can also see that $C_1(\overline{u},v;\chi)$ is monotonically non-decreasing in $v$ for $v \in (u,v_3]$. Hence, we obtain
\begin{equation}\label{TSF66}
\mu -\frac{Q^2}{r^2}(u,v) \geq \ml 1 - C_1^{2}(\overline{u},v_3;\chi)r(\overline{u},v_3)(v-u)\mr \mu(u,v).
\end{equation}
Using $v - u = v - v_1 + v_1 - u \leq (v_2 - v_1) + (v_1 - u_0)$ for all $u \in [u_0,v_1]$ and $v \in [v_1,v_2]$, if we pick $v_2$ and $u_0$ sufficiently close to $v_1$ such that
\begin{equation}\label{TSF67}
(v_2 - v_1) + (v_1 - u_0) \leq \frac{1}{C_1^{2}(\overline{u},v_3;\chi)r(\overline{u},v_3)},
\end{equation}
then we deduce that $\mu - \frac{Q^2}{r^2} \geq 0$ on $R' \subset \mathcal{R}$. 

With these preparations, we are ready to prove the main theorem. Introducing a scale-free length $x$ by
\begin{equation}\label{TSF3}
x(u) = \frac{r_2(u)}{r_2(u_0)},
\end{equation}
we then consider the evolution of $\eta^*(x)$ along incoming null curves lying within $R'$. For $x \in (0,1]$, we have 
\begin{equation}\label{TSF4}
\begin{aligned}
\frac{\mathrm{d}\eta^*}{\mathrm{d}x} =& \;\frac{\frac{\mathrm{d}\eta^*}{\mathrm{d}u}}{\frac{\mathrm{d}x}{\mathrm{d}u}} = \frac{\du \ml \frac{2(m_2 - m_1)}{r_2} - \frac{2}{r_2}\int^{v_2}_{v_1} \frac{Q^2 \dv r}{r^2} (u,v') \mathrm{d}v'\mr }{\frac{\du r_2}{r_2(u_0)}} \\
&\,\,= \; \frac{r_2(u_0)}{\du r_2} \ml 
-\frac{\du r_2}{r_2}\ml\frac{2(m_2 - m_1)}{r_2} +\frac{2}{r_2} \int^{v_2}_{v_1} \frac{Q^2 \dv r}{2r^2} \mathrm{d}v'
\mr  + \frac{2}{r_2}\du(m_2 - m_1)  - \frac{2}{r_2} \int^{v_2}_{v_1} \du \ml \frac{Q^2 \dv r}{2r^2}\mr\mathrm{d}v'\mr \\
&\,\,= \;- \frac{\eta^*}{x} - \frac{2}{x(\du r_2)} \ml \int^{v_2}_{v_1} -\du\dv m + \du \ml \frac{Q^2 \dv r}{2r^2}\mr \mathrm{d}v'\mr.
\end{aligned}
\end{equation}
Note that the term $\du\ml \frac{Q^2 \du r}{2r^2}\mr$ comes from the definition of $\eta^*(u)$ in \eqref{TSF2}, and it is designed to have a precise cancellation with terms in $-\du \dv m$. This is illustrated using \eqref{Setup25}, where we see that
\begin{equation}\label{TSF5}
- \du\dv m + \du \ml \frac{Q^2 \dv r}{2r^2} \mr = -\du \ml \frac{2\pi r^2(1-\mu)|\dv \phi|^2}{\dv r} \mr.
\end{equation}
Hence, with the help of \eqref{Setup21}, the integrand in the last line of \eqref{TSF4} can be simplified as
\begin{equation}\label{TSF6}
\begin{aligned}
- \du\dv m + \du \ml \frac{Q^2 \dv r}{2r^2} \mr &= -\du \ml \frac{2\pi r^2(1-\mu)|\dv \phi|^2}{\dv r} \mr \\
&= - 2 \pi \ml \du\ml \frac{1-\mu}{\dv r}\mr r^2|\dv \phi|^2 + \ml \frac{1-\mu}{\dv r}\mr \du (r^2|\dv \phi|^2) \mr \\
&= 8\pi^2 r^3 (1-\mu)(\du r)(\dv r) \frac{|D_u \phi|^2 |\dv \phi|^2}{(\du r)^2(\dv r)^2} - 2\pi \ml \frac{1-\mu}{\dv r}\mr \du\ml r^2|\dv \phi|^2\mr.
\end{aligned}
\end{equation}
For $\du(r^2|\dv \phi|^2)$, from \eqref{Setup73} and the definition of $D_u$, we have
\begin{equation}\label{TSF7}
\du(r\dv \phi) = -\ii \e A_u r \dv \phi - \dv r D_u \phi + \ii\e \frac{Q\phi(-\du r)(\dv r)}{r(1-\mu)}
\end{equation}
and thus
\begin{equation}\label{TSF8}
\begin{aligned}
\du(r^2|\dv \phi|^2) =&\; \du(r \dv \phi) (r\dv \phi)^{\dagger} + \du( (r \dv \phi)^\dagger) (r\dv \phi)
\\ =&\; r (\dv \phi)^\dagger \ml -\ii \e A_u r \dv \phi - \dv r D_u \phi + \ii\e \frac{Q\phi(-\du r)(\dv r)}{r(1-\mu)} \mr \\
&+ r (\dv \phi) \ml -\ii \e A_u r \dv \phi - \dv r D_u \phi + \ii\e \frac{Q\phi(-\du r)(\dv r)}{r(1-\mu)} \mr^{\dagger} \\
=& -2 r \dv r \re{\ml (\dv \phi)^\dagger D_u \phi  \mr} + 2\e Q \frac{(\du r) (\dv r)}{1-\mu} \im{(\phi^\dagger \dv \phi)}.
\end{aligned}
\end{equation}
Combining these calculations, we obtain
\begin{equation}\label{TSF9}
\begin{aligned}
&- \du\dv m + \du \ml \frac{Q^2 \dv r}{2r^2} \mr \\ =& \; (\dv r)(-\du r) \ml -8\pi^2 r^3 (1-\mu) \ml \frac{|D_u \phi| |\dv \phi|}{(-\du r)(\dv r)} \mr^2  + 4\pi r (1-\mu) \frac{\re{\ml (\dv \phi)^\dagger (D_u \phi) \mr}}{\dv r (-\du r)}\mr\\
&\; + 4\pi \e Q (-\du r) \im{(\phi^\dagger (\dv \phi))} \\
\leq&\; 8\pi^2 r^3 (\dv r)(-\du r)(1-\mu) \ml - \ml \frac{|D_u \phi| |\dv \phi|}{(-\du r)(\dv r)} \mr^2  + \frac{1}{2\pi r^2} \frac{|D_u \phi||\dv \phi|}{(-\du r)(\dv r)}\mr + 4\pi \e |Q| (-\du r) |\phi||\dv \phi|. \\
\end{aligned}
\end{equation}
Using the fact that $\du \dv r \leq 0$ in $R'$ and equation \eqref{TSF9}, the integral in \eqref{TSF4} can be simplified to
\begin{equation}\label{TSF10}
\begin{aligned}
-\frac{1}{\du r_2}\int^{v_2}_{v_1} -\du\dv m + \du \ml \frac{Q^2 \dv r}{2r^2}\mr \mathrm{d}v' &\leq \int^{v_2}_{v_1} J_4 (u,v') + J_5 (u,v') \mathrm{d}v'
\end{aligned}
\end{equation}
with
\begin{equation}\label{TSF11}
J_4 = 8\pi^2 r^3 (\dv r)(1-\mu) \ml - \ml \frac{|D_u \phi| |\dv \phi|}{(-\du r)(\dv r)} \mr^2  + \frac{1}{2\pi r^2} \frac{|D_u \phi||\dv \phi|}{(-\du r)(\dv r)}\mr
\end{equation}
and
\begin{equation}\label{TSF12}
J_5 = 4 \pi \e |Q||\phi||\dv \phi|.
\end{equation}
For $J_4$, by using an elementary inequality  $-x^2 + ax \leq \frac{a^2}{4}$, we have
\begin{equation}\label{TSF13}
\begin{aligned}
&\int^{v_2}_{v_1} 8\pi^2 r^3 (\dv r)(1-\mu) \ml - \ml \frac{|D_u \phi| |\dv \phi|}{(-\du r)(\dv r)} \mr^2  + \frac{1}{2\pi r^2} \frac{|D_u \phi||\dv \phi|}{(-\du r)(\dv r)}\mr(u,v') \mathrm{d}v' \\
\leq&\; \int^{v_2}_{v_1} 2\pi^2 r^3 (\dv r)(1-\mu) \ml \frac{1}{2\pi r^2}\mr^2(u,v') \mathrm{d}v' \leq\; \frac{1}{2} \int^{v_2}_{v_1}   \frac{\dv r}{r}(u,v') \mathrm{d}v' \\
&\quad\quad=\; \frac{1}{2}\log\ml \frac{r(u,v_2)}{r(u,v_1)}\mr  = \frac{1}{2}\log \ml 1 + \frac{r(u,v_2) - r(u,v_1)}{r(u,v_1)}\mr \leq \frac{1}{2}\delta(u).
\end{aligned}
\end{equation} 
Note that via our new approach, $J_4$ is bounded by $\frac{1}{2}\delta(u)$ independent of $\phi$ through a simple computation. 

For $J_5$, with \eqref{Setup25} and the definition in \eqref{ReducedMassRatio}, we proceed to deduce that
\begin{equation}\label{TSF21}
\begin{aligned}
\int^{v_2}_{v_1} J_5 (u,v') \mathrm{d}v'&= \int^{v_2}_{v_1} 4 \pi \e |Q||\phi||\dv \phi| (u,v') \mathrm{d}v' \\
&\leq \sqrt{8\pi} \e \ml \int^{v_2}_{v_1} \frac{2\pi r^2 (1-\mu)|\dv \phi|^2}{\dv r} (u,v') \mathrm{d}v' \mr^{\frac{1}{2}} \ml \int^{v_2}_{v_1} \frac{\dv r}{1-\mu} \frac{|Q|^2|\phi|^2}{r^2}(u,v') \mathrm{d}v'\mr^{\frac{1}{2}} \\
&\leq \sqrt{8\pi} \e \ml \int^{v_2}_{v_1} \dv m - \frac{Q^2 \dv r}{2r^2} (u,v') \mathrm{d}v' \mr^{\frac{1}{2}} \ml \int^{v_2}_{v_1} \frac{\dv r}{1-\mu} \frac{|Q|^2|\phi|^2}{r^2}(u,v') \mathrm{d}v'\mr^{\frac{1}{2}} \\
&\leq \sqrt{4\pi} \e {\eta^*}^{\frac{1}{2}} r_2^{\frac{1}{2}}(u) \ml \int^{v_2}_{v_1} \frac{\dv r}{1-\mu} \frac{|Q|^2|\phi|^2}{r^2}(u,v') \mathrm{d}v'\mr^{\frac{1}{2}}.
\end{aligned}
\end{equation}
Employing \eqref{Estimates19} of Proposition \ref{PropEstimates} with $\xi_1 = 2$ and $\xi_2 = \chi$ for $\chi \in \left(0,\frac{1}{2}\right)$, we obtain the bound
\begin{equation}\label{TSF35}
|Q \phi|^2(u,v) \leq C_{5}(\overline{u},v_2;\chi,2,\chi)^2 r^{\frac{5}{2} - \frac{3\chi}{2}}(u,v),
\end{equation}
where we have used the fact that $\mu < 1$. Returning to \eqref{TSF21}, with \eqref{TSF35} and $\chi \in \left(0, \frac{1}{3}\right]$, we thus have
\begin{equation}\label{TSF37}
\begin{aligned}
\int^{v_2}_{v_1} J_5 (u,v') \mathrm{d}v' 
&\leq \sqrt{4\pi} \e {\eta^*}^{\frac{1}{2}} r_2^{\frac{1}{2}}(u) \ml \int^{v_2}_{v_1} \frac{\dv r}{1-\mu} \frac{|Q|^2|\phi|^2}{r^2}(u,v') \mathrm{d}v'\mr^{\frac{1}{2}} \\
&\leq \sqrt{4\pi} \e C_{5} {\eta^*}^{\frac{1}{2}} r_2^{\frac{1}{2}}(u)  \ml \int^{v_2}_{v_1} \frac{\dv r}{1 - \mu}(u_0,v') r^{\frac{1}{2}-\frac{3\chi}{2}}(u,v') \mathrm{d}v'\mr^{\frac{1}{2}} \\
&\leq \sqrt{4\pi} \e C_{5} {\eta^*}^{\frac{1}{2}} r_2^{\frac{3}{4} - \frac{3\chi}{4}}(u) \ml \frac{1}{1 - \mu_0^*}\mr^{\frac{1}{2}} \ml \int^{v_2}_{v_1} \dv r(u_0,v') \mathrm{d}v'\mr^{\frac{1}{2}} \\
&\leq \frac{\sqrt{4\pi} \e C_{5}}{(1-\mu_0^*)^{\frac{1}{2}}} {\eta^*}^{\frac{1}{2}} r_2^{\frac{3}{4}-\frac{3\chi}{4}}(u)  r_1^{\frac{1}{2}}(u_0) \delta(u_0)^{\frac{1}{2}}.
\end{aligned}
\end{equation}

Utilizing \eqref{TSF10}, \eqref{TSF13}, \eqref{TSF37} and the fact that $x \in (0,1]$, the main inequality \eqref{TSF4} simplifies to 
\begin{equation}\label{TSF38}
\begin{aligned}
\frac{\mathrm{d}\eta^*}{\mathrm{d}x} &\leq - \frac{\eta^*}{x} - \frac{2}{x(\du r_2)} \ml \int^{v_2}_{v_1} -\du\dv m + \du \ml \frac{Q^2 \dv r}{2r^2}\mr \mathrm{d}v'\mr \\
&\leq - \frac{\eta^*}{x} + \frac{\delta}{x} + \frac{1}{x} \frac{4\sqrt{\pi} \e C_{5}}{(1-\mu_0^*)^{\frac{1}{2}}} {\eta^*}^{\frac{1}{2}} r_2^{\frac{3}{4}-\frac{3\chi}{4}}(u)r_1^{\frac{1}{2}}(u_0) \delta^{\frac{1}{2}}(u_0) \\
&\leq - \frac{\eta^*}{x} + \frac{\delta}{x} + C_{21}x^{-\frac{1}{4}-\frac{3\chi}{4}}\ml \frac{2{\eta^*}^{\frac{1}{2}}\delta_0^{\frac{1}{2}}(r_1)_0^{\frac{1}{2}}}{(1-\mu_0^*)^{\frac{1}{2}}} \mr \\
&\leq - \frac{\eta^*}{x} + \frac{\delta}{x} + \frac{C_{21}}{x^{\frac{1}{2}}} \ml \eta^* + \frac{\delta_0(r_1)_0}{1-\mu_0^*} \mr. \\
\end{aligned}
\end{equation}
Furthermore, the constant $C_{21}$ is defined as
\begin{equation}\label{TSF39}
C_{21} = 2\sqrt{\pi} \e C_{5}r(\overline{u},v_3)^{\frac{3}{4}-\frac{3\chi}{4}}
\end{equation}
with $\chi \in \left(0,\frac{1}{3}\right]$. Meanwhile, for all $(u,v) \in R'$, using $\du \dv r \leq 0$ in $R'$, we get
\begin{equation}\label{TSF40}
\dv r (u,v) \leq \dv r(u_0,v).
\end{equation}
For $u \in [u_0,v_1)$, this implies
\begin{equation}\label{TSF41}
r_2(u) - r_1(u) \leq r_2(u_0) - r_1(u_0).
\end{equation}
Thus, one can establish the following inequality
\begin{equation}\label{TSF42}
\begin{aligned}
\delta(u) = \frac{r_2}{r_1} - 1 &= \frac{r_2 - r_1}{r_2 - (r_2 - r_1)} \leq \frac{r_2(u_0)-r_1(u_0)}{r_2(u) - (r_2(u_0) - r_1(u_0))} \\
&\quad\quad\quad\leq \frac{\delta(u_0)}{\frac{r_2(u)}{r_1(u_0)} - \delta(u_0)} = \frac{\delta_0}{x(u)(1+\delta_0) - \delta_0}.
\end{aligned}
\end{equation}
{\color{black}
Rewriting \eqref{TSF38}, we hence arrive at
\begin{equation}\label{TSF44}
\begin{aligned}
\frac{\mathrm{d}\eta^*}{\mathrm{d}x} +  \ml \frac{1}{x} - \frac{C_{21}}{x^{\frac{1}{2}}}\mr\eta^*&\leq  \delta_0\ml \frac{1}{x(x(1+\delta_0) - \delta_0)} +  \frac{C_{21}(r_1)_0}{x^{\frac{1}{2}}(1-\mu_0^*)} \mr.\\
\end{aligned}
\end{equation}
This can be rewritten as
\begin{equation}\label{TSF44'}
\frac{\mathrm{d}\eta^*}{\mathrm{d}x} + f(x)\eta^* \leq \delta_0 g(x)
\end{equation}
with
\begin{equation}\label{TSF45}
f(x) = \frac{1}{x} - \frac{C_{21}}{x^{\frac{1}{2}}} \quad \text{ and } \quad g(x) = \frac{1}{x(x(1+\delta_0) - \delta_0)} +  \frac{C_{21}(r_1)_0}{x^{\frac{1}{2}}(1-\mu_0^*)}.
\end{equation}
We then solve the differential inequality \eqref{TSF44'} and get 
$$\frac{\mathrm{d}}{\mathrm{d}x} \ml e^{-\int^1_x f(s) \mathrm{d}s} \eta^* \mr \leq \delta_0 g(x) e^{-\int^1_x f(s) \mathrm{d}s},$$
which implies that, for any $x \in (x(u_*),1)$, we obtain
$$\eta^*_0  - e^{-\int^1_x f(s) \mathrm{d}s} \eta^*(x) \leq \delta_0 \int^1_x g(s) e^{-\int^1_s f(s') \mathrm{d}s'}\mathrm{d}s.$$
Recall that in the regular region definition \eqref{ReducedMassRatio} implies $\eta_*(x) \leq \mu_2(x)$. Rearranging the above inequality, we get
\begin{equation}\label{TSF46}
\mu_2(x) \geq \eta^*(x) \geq e^{\int_x^1 f(s) \D s}\ml \eta_0^* - \delta_0 \int_x^1 g(s) e^{-\int_s^1 f(s') \D s'}\mr.
\end{equation}
We further evaluate the integrals in \eqref{TSF46} as follows. We first calculate
\begin{equation}\label{TSF47}
\begin{aligned}
\int^1_x f(s) \mathrm{d}s &= \int^1_x \frac{1}{s} - \frac{C_{21}}{s^{\frac{1}{2}}} \mathrm{d}s = -\log(x) - 2C_{21}\ml 1-x^{\frac{1}{2}} \mr,
\end{aligned}
\end{equation}
which implies
\begin{equation}\label{TSF48}
\begin{aligned}
e^{-\int^1_x f(s) \mathrm{d}s} = xe^{2C_{21}(1-x^{\frac{1}{2}})} \leq xe^{2C_{21}}.
\end{aligned}
\end{equation}
The second integral in \eqref{TSF46} obeys
\begin{equation}\label{TSF49}
\begin{aligned}
\int^1_x g(s) e^{-\int^1_s f(s') \mathrm{d}s'}\mathrm{d}s &\leq \int^1_x \ml \frac{1}{s(s(1+\delta_0) - \delta_0)} +  \frac{C_{21}(r_1)_0}{s^{\frac{1}{2}}(1-\mu_0^*)}\mr s e^{2C_{21}} \mathrm{d}s \\
&\leq e^{2C_{21}} \ml \frac{\log\ml \frac{1}{x(1+\delta_0)-\delta_0}\mr}{1+
\delta_0} + \frac{2C_{21} (r_1)_0}{3(1-\mu_0^*)}\ml 1 - x^{\frac{3}{2}}\mr   \mr
\end{aligned}
\end{equation}
if we impose $x > \frac{\delta_0}{1+\delta_0}$. For our following arguments, we restrict ourselves to $x \in \left( \frac{\delta_0}{1+\delta_0}, 1\right]$. Back to \eqref{TSF46}, with \eqref{TSF48} and \eqref{TSF49}, we now obtain
\begin{equation}\label{TSF50}
\mu_2(x) \geq  \frac{1}{C_{22}x}\ml \eta_0^* - C_{22}\frac{\delta_0}{1+\delta_0} \log\ml\frac{1}{x(1+\delta_0) - \delta_0}\mr - C_{23}\frac{\delta_0 (r_1)_0}{1-\mu_0^*} \mr
\end{equation}
with
\begin{equation}\label{TSF51}
C_{22} = e^{2C_{21}} \quad \text{ and } \quad C_{23} = \frac{2C_{21}e^{2C_{21}}}{3}.
\end{equation}
Next, we consider the following two cases:
\begin{itemize}[leftmargin=*]
\item Suppose that $|\mathcal{B}_0| \geq v_2 - v_1$. This scenario corresponds to the picture on the left of Figure \ref{TrappedFig1}. As seen from the picture, this allows us to set $u_*$ to be arbitrarily close to $v_1$. Furthermore, observe that at $(u_*, v_2)$, we have
\begin{equation}\label{Scalar21}
x(u_*) \leq \frac{r_1(u_*)}{r_2(u_0)} + \frac{\delta_0}{1 + \frac{r_2(u_0) - r_1(u_0)}{r_1(u_0)}} =  \frac{r_1(u_*)}{r_2(u_0)} + \frac{\delta_0}{1+\delta_0}.
\end{equation}
Henceforth, by \eqref{Scalar21} and the fact 
 that $r_1(u_*) \rightarrow 0$ as $u_* \rightarrow v_1^-$, we can have $x(u_*) \leq \frac{p\delta_0}{1+\delta_0}$ for any $p > 1$.
\item Assume that $|\mathcal{B}_0| < v_2 - v_1$. We encounter the scenario portrayed in the right picture of Figure \ref{TrappedFig1}. As observed from the picture, we might not be able to pick $u_*$ to be arbitrarily close to $v_1$ due to the possibly spacelike nature of $\mathcal{B}\setminus \mathcal{B}_0$. Nonetheless,  $r_2(u)\rightarrow 0$ as $(u,v_2)$ approaches any point on $\mathcal{B} \setminus \mathcal{B}_0$. This implies that it suffices to pick $u_*$ along $v = v_2$ to be sufficiently close to $\mathcal{B} \setminus \mathcal{B}_0$ and we conclude $x(u_*) \leq \frac{p\delta_0}{1+\delta_0}$ for any $p > 1$ as in the previous case. 
\end{itemize}

\noindent For both scenarios, since $x(u)$ is a decreasing continuous function on $[u_0,u_*]$, there exists $\tilde{u}(k)$ such that $x(\tilde{u}(k)) = \frac{k\delta_0}{1+\delta_0}$ for any $k > 1$. We then evaluate \eqref{TSF50} at the point $(\tilde{u}(k), v_2)$. Further applying \eqref{TSF2} we then obtain
\begin{equation}\label{TSF60}
\mu(\tilde{u}(k),v_2) \geq \frac{1}{C_{22}k}\ml (C_{24}-C_{22})\log\ml \frac{1}{\delta_0}\mr + C_{22} \log \ml k-1\mr\mr.
\end{equation}
Hence, we would have $\mu(\tilde{u}(k),v_2) \geq 1$ if 
\begin{equation}\label{TSF61}
(C_{24}-C_{22})\log\ml \frac{1}{\delta_0}\mr \geq C_{22}\ml k - \log(k-1) \mr.
\end{equation}
Note that with $k>1$ the function $k-\log(k-1)$ is minimized at $k = 2$. Henceforth, we get $\mu(\tilde{u}(2),v_2) \geq 1$ if
\begin{equation}\label{TSF62}
(C_{24}-C_{22})\log\ml \frac{1}{\delta_0}\mr \geq 2C_{22}.
\end{equation}
This requirement is easily satisfied if we set 
\begin{equation}\label{TSF54}
C_{24} = 3C_{22},
\end{equation}
since we have $\log\ml\frac{1}{\delta_0}\mr \geq 1$ for $\delta_0 \in (0,1)$. With this choice of $C_{24}$, we then obtain $\mu(\tilde{u}(2),v_2) \geq 1$, which guarantees that a MOTS or a trapped surface is formed in $R'$. This concludes the proof of Theorem \ref{Trapped}. 
\qed
}

\section{BV Area Estimates}\label{BV Area Estimates}

In this section, we establish scale-critical BV estimates for $C^1$ solutions with BV initial data. To begin, we first define the concept of area for maps of two Lipschitz functions. Let $\mathcal{U}$ be a domain in the $(u,v)$ plane and let $f,g$ be a pair of Lipschitz functions defined on it. The area $A(\mathcal{U})[f,g]$ for the image of $\mathcal{U}$ lying in a plane by the mapping $(u,v) \rightarrow (f(u),g(v))$ is given by
\begin{equation}\label{Area2}
A(\mathcal{U})[f,g] := \int_{\mathcal{U}} \mlm \frac{\p(f,g)}{\p(u,v)}\mrm \D u \D v
\end{equation}
with
\begin{equation}\label{Area3}
\frac{\p(f,g)}{\p(u,v)} = \frac{\p f}{\p u}\frac{\p g}{\p v} - \frac{\p f}{\p v}\frac{\p g}{\p u}
\end{equation}
being the Jacobian of the mapping $(u,v) \rightarrow (f(u,v),g(u,v))$.

A major difference between our charged system \eqref{Setup12}-\eqref{Setup19} and the uncharged scalar-field system investigated in \cite{christ1} - \cite{christ4} is that the charged system is \textit{not} invariant under translations of $\phi$. In other words, for any constant $c$, the following transformation
$$ r\rightarrow r, \quad \Omega \rightarrow \Omega, \quad \phi \rightarrow \phi + c,$$
does not necessarily leaves all equations from \eqref{Setup12} - \eqref{Setup19} invariant. This means that we will have to modify the area estimates obtained in \cite{christ2}, in which it was assumed without loss of generality that $\phi(0,0) = 0$. In our case, we denote
\begin{equation}\label{tildephi}
\phi_0 := \phi(0,0), \quad \text{ and } \quad  \tp := \phi - \phi_0.
\end{equation}
Correspondingly, we define
\begin{equation}\label{mathfraka}
\mathfrak{a} := \dv(r\phi), \quad \quad \alpha := \frac{\dv(r\phi)}{\dv r} = \frac{\mfa}{\dv r},
\end{equation}
and
\begin{equation}\label{tildemathfraka}
\tilde{\mathfrak{a}} := \dv(r\tp), \quad  \quad \tilde{\alpha} := \frac{\dv(r\tp)}{\dv r} = \frac{\tmfa}{\dv r}.
\end{equation}
Analogously, we set
\begin{equation}\label{mathfrakb}
\mathfrak{b} := \du(r\phi), \quad \quad \beta := \frac{\du(r\phi)}{\du r} = \frac{\mfb}{\du r}, 
\end{equation}
and
\begin{equation}\label{tildemathfrakb}
\tilde{\mathfrak{b}} := \du(r\tp) \quad \quad \tilde{\beta} := \frac{\du(r\tp)}{\du r} = \frac{\tmfb}{\du r}.
\end{equation}
For convenience, we will also denote
\begin{equation}\label{lambda}
\lambda := \dv r, \quad \text{ and } \nu := \du r.
\end{equation}
\begin{remark}
The above definitions of $\alpha$ and $\beta$ are slightly different from \cite{christ2}, where $\alpha$ and $\beta$ are defined to be just $\dv(r\phi)$ and $\du(r\phi)$. To be consistent with the notations adopted in \cite{christ4}, we will adopt the definitions presented in this paper. 
\end{remark}

From the above definitions, it is not difficult to see that
\begin{equation}\label{Area39}
\tmfa|_{C_0^+} = \mfa|_{C_0^+} - \frac{1}{2}\phi_0.
\end{equation}
Therefore, the total variations of $\tmfa$ and $\mfa$ are the same. Hence, we define
\begin{equation}\label{Area40}
D := TV_{\{0\} \times (0,v_0)}[\tmfa] = TV_{\{0\} \times (0,v_0)}[\mfa],
\end{equation}
which measures the size of the initial data in the TV (total variation) norm. In addition, we define the following quantities
\begin{equation}\label{Area15}
\begin{aligned}
X(v_0) &= \sup_{u \in (0,v_0)} \int^{v_0}_u 
 |\dv \phi|(u,v') \D v' = \sup_{u \in (0,v_0)} \int^{v_0}_u 
 |\dv \tp|(u,v') \D v', \\
Y(v_0) &= \sup_{v \in (0,v_0)} \int^{v}_0 
 |\du \phi|(u',v) \D u' = \sup_{v \in (0,v_0)} \int^{v}_0 
 |\du \tp|(u',v) \D u'.
\end{aligned}
\end{equation}

For ease of computation, we pick $v_0$ sufficiently small, such that $\mu - Q^2/r^2 \geq 0$ in $\mathcal{D}(0,v_0)$. We assume that we have a $C^1$ solution on $\mathcal{D}(0,v_0)$ and $\mathcal{D}(0,v_0) \subset \mathcal{R}.$ This means that we can apply Proposition \ref{PropEstimates} with $\overline{u} = 0$. In particular, by \eqref{Estimates17}, we have that $\mu - Q^2/r^2 \geq 0$ only if 
\begin{equation}\label{Area11}
C_1^2(0,v;\chi)r^{1-\chi}(u,v)r^{\chi}(0,v)(v-u) \leq 1
\end{equation}
for all $(u,v) \in \mathcal{D}(0,v_0).$ By restricting our choice of $v_0$ to stay in $(0,1)$, it suffices to demand that
\begin{equation}\label{Area12}
C_1^2(0,1;\chi)r^{1-\chi}(u,v)r^{\chi}(0,v)(v-u) \leq 1
\end{equation}
since by \eqref{Estimates11}, we know that $C_1(0,\cdot;\chi)$ is a non-decreasing function. In particular, from \eqref{Setup33} and the fact that $\du r < 0$, $\dv r > 0$, and $u \geq 0$, if we pick $\chi = \frac{1}{2}$, it then suffices to demand that 
\begin{equation}\label{Area13}
v_0^2 \leq \frac{2}{C_1^2\ml 0,1;\frac{1}{2}\mr}.
\end{equation}
For instance, by choosing
\begin{equation}\label{Area14}
v_0 = \min\left\{\oh, \frac{\sqrt{2}}{C_1 \ml 0, 1; \frac{1}{2}\mr}\right\},
\end{equation} 
we then have $v_0 < 1$ and \eqref{Area13} are satisfied. Note that such a choice is \textit{independent} of $D$, which is important in this section. In the arguments that follow, we will refine $v_0$ to be even smaller if necessary, as long as it remains to be independent of $D$.

Last but not least, we work with the following assumptions for solutions in $\mathcal{D}(0,v_0)$:
\begin{equation}\label{Area27}
\inf_{\mathcal{D}(0,v_0)} \dv r \geq \frac{5}{12}, \quad \sup_{\mathcal{D}(0,v_0)} (-\du r) \leq \frac{2}{3}, \quad \sup_{\mathcal{D}(0,v_0)} \mu \leq \frac{1}{3}.
\end{equation}
These assumptions can be verified directly with the same type of arguments as in Section 6 of \cite{christ2}.

With the above preparation, we now state the main theorem of this section.

\begin{theorem}\label{AreaLemma3} ($BV$ Area Estimates.)
Assume that \eqref{Area27} holds with our choice of $v_0$ to be given in \eqref{Area97}. Then there exists a constant $\varepsilon_0 \in (0,1]$ such that if the initial data satisfy
\begin{equation}\label{Area75}
Z(v_0) = \max\{X(v_0),Y(v_0)\} \leq \varepsilon_0, \quad \text{ and } \quad D := TV_{\{0\} \times (0,v_0)}[\tmfa]  \leq \varepsilon_0,
\end{equation}
then we have 
\begin{equation}\label{Area76}
Z(v_0) \lesssim D,
\end{equation}
\begin{equation}\label{Area301}
\begin{aligned}
\sup_{u \in (0,v_0)}TV_{\{u\}\times(u,v_0)}[\tmfa] &\lesssim D,\\
\sup_{v \in (0,v_0)}TV_{(0,v) \times \{v\}}[\ttmfb] &\lesssim D,\\
TV_{(0,v_0)}[\tmfa|_{\Gamma}]  &\lesssim D,
\end{aligned} \quad 
\begin{aligned}
\sup_{u \in (0,v_0)} TV_{ \{u\} \times (u,v_0) } [\log \lambda] &\lesssim D^2, \\
\sup_{v\in(0,v_0)} TV_{(0,v) \times \{v\}}[\log |\nu|] &\lesssim D^2, \\
TV_{(0,v_0)}[\log(\lambda|_{\Gamma})] &\lesssim D^2, \\
\end{aligned}
\end{equation}
and
\begin{equation}\label{Area302}
\begin{aligned}
A[\tp^\dagger ,\tmfa/\lambda] &\lesssim D^2, \\
A[\tp^\dagger ,\ttmfb/\nu] &\lesssim D^2, \\
\end{aligned} \quad
\begin{aligned}
A[\lambda, \tp] &\lesssim D^3, \\
A[\nu, \tp] &\lesssim D^3. \\
\end{aligned}
\end{equation}
\end{theorem}

\begin{remark}
For a quantity $H$, we write $H \lesssim D$ to mean that there exists a non-negative constant $c$ (possibly depending on $\e$ and $|\phi_0|$) independent of $D$ such that $H \leq cD$. For brevity, we also write $A[f,g]$ to mean $A(\mathcal{D}(0,v_0))[f,g]$ for any pair of Lipschitz functions $f$ and $g$.
\end{remark}

To prove this theorem, we first state and prove several lemmas.
\begin{lemma}\label{AreaLemma1}
For any $(u,v) \in \mathcal{D}(0,v_0)$, we have
\begin{equation}\label{Area16}
\begin{aligned}
\int^{v_0}_u |\dv \tilde{\phi}|(u,v') \D v' &\leq TV_{\{u\} \times (u,v_0)} [\tilde\alpha], \quad \int^{v}_0 |\du \tilde{\phi}|(u',v) \D u' \leq TV_{(0,v) \times \{v\} } [\tilde\beta].
\end{aligned}
\end{equation}
\end{lemma}
The proof of this lemma is similar to the corresponding one in \cite{christ2} and we outline the approach.
\begin{proof} Using the fact that $r\phi|_{\Gamma} = 0$ and $r\phi \in C^1$, for each $(u,v) \in \mathcal{D}(0,v_0)$, we have
\begin{equation}\label{Area17}
\phi(u,v) = \frac{1}{r(u,v)}
 \int^v_u \dv(r\phi) (u,v') \D v'.
\end{equation}
This implies 
\begin{equation}\label{Area18}
\begin{aligned}
(\dv \phi)(u,v) &= - \frac{\dv r(u,v)}{r^2(u,v)}\int^{v}_u \dv (r\phi)(u,v') \D v' + \frac{\dv (r\phi)(u,v)}{r(u,v)} \\
&= \frac{(\dv r)(u,v)}{r^2(u,v)} \int^v_u \ml \alpha(u,v) - \alpha(u,v')\mr(\dv r)(u,v') \D v'.
\end{aligned}
\end{equation}
Next, since $\phi$, $r \dv \phi$, $\dv r$ are $C^1$ functions on $\mathcal{D}(0,v_0)$,  this implies that $\alpha \in C^1$ on $\mathcal{D}(0,v_0)$. Consequently, we have
\begin{equation}\label{Area19}
\lim_{v \rightarrow u+}\frac{ \int^v_u \ml \alpha(u,v) - \alpha(u,v')\mr(\dv r)(u,v') \D v'}{r(u,v)} = 0.
\end{equation}
We further denote
\begin{equation}\label{Ialpha}
I_{\alpha}(u,v) := \int^v_u \ml \alpha(u,v) - \alpha(u,v')\mr(\dv r)(u,v') \D v'
\end{equation}
and it satisfies 
\begin{equation}\label{Area20}
\dv I_{\alpha} = r \dv \alpha.
\end{equation}
From \eqref{Area18}, integrating along $v \in (u,v_0)$ and applying integration by parts, we obtain
\begin{equation}\label{Area21}
\begin{aligned}
\int_u^{v_0} |\dv \phi|(u,v') \D v' 
&= - \int_u^{v_0} |I_{\alpha}(u,v')|\D\ml \frac{1}{r(u,v')}\mr = -\frac{|I_{\alpha}|(u,v_0)}{r(u,v_0)} + \int^{v_0}_u \frac{\p |I_{\alpha}(u,v')|}{\p v}\frac{\D v}{r(u,v')} \\
&\quad\quad\leq \int^{v_0}_u \mlm \dv I_{\alpha} (u,v')\mrm \frac{1}{r(u,v')} \D v' = \int^{v_0}_u |\dv \alpha|
(u,v') \D v' = TV_{\{u\} \times (u,v_0)}[\alpha].
\end{aligned}
\end{equation}
Together with the fact that $\dv \phi = \dv \tp$ and $TV_{\{u\} \times (u,v_0)}[\alpha] = TV_{\{u\} \times (u,v_0)}[\tilde{\alpha}]$, inequality \eqref{Area21} implies the first half of \eqref{Area16}.\\

For the second half of \eqref{Area16}, we begin by writing
\begin{equation}\label{Area22}
\phi(u,v) = - \frac{1}{r(u,v)}
 \int^v_u \du(r\phi) (u',v) \D u'.
\end{equation}
We proceed as above and compute
\begin{equation}\label{Area23}
(\du \phi)(u,v) = \frac{(-\du r)(u,v)}{r^2(u,v)} \int^v_u \ml \beta(u,v) - \beta(u',v)\mr(\du r)(u',v) \D u'.
\end{equation}
Next, we denote
\begin{equation}\label{Ibeta}
I_\beta(u,v) := \int^v_u \ml \beta(u,v) - \beta(u',v)\mr(\du r)(u',v) \D u'
\end{equation}
and observe that
\begin{equation}\label{Area24}
\lim_{u \rightarrow v^-}\frac{I_{\beta}(u,v)}{r(u,v)} = 0 \quad \text{ and } \quad
\du I_\beta = - r \du \beta.
\end{equation}
Integrating \eqref{Area23} along $u \in (0,v)$, we obtain
\begin{equation}\label{Area26}
\begin{aligned}
\int^v_0 |\du \phi|(u',v) \D u' &= -\int^{v}_0 |I_\beta (u',v)| \D \ml \frac{1}{r(u',v)}\mr \\
&\leq \int^{v}_0 |\du \beta|(u',v) \D u' = TV_{(0,v) \times \{v\}}[\beta].
\end{aligned}
\end{equation}
Since $TV_{(0,v) \times \{v\}}[\beta] = TV_{(0,v) \times \{v\}}[\tilde\beta]$, we then prove the second inequality in \eqref{Area16}.
\end{proof}

We now attempt to derive estimates for the following quantities:
\begin{equation}\label{Area28}
A := \sup_{\mathcal{D}(0,v_0)}|\tmfa|, \quad B := \sup_{\mathcal{D}(0,v_0)}|\tmfb|, \quad P := \sup_{\mathcal{D}(0,v_0)}|\tp|.
\end{equation}
We summarize the estimates for these three quantities in a lemma below.
\begin{lemma}\label{AreaLemma2}
Suppose that $D \leq 1$. By refining of $v_0$ to be as described in \eqref{Area64}, we have
\begin{equation}\label{Area29}
A \leq 4D, \quad B \leq 8D, \quad P \leq 12D.
\end{equation}
\end{lemma}
\begin{proof}
We begin by considering the set
\begin{equation}\label{Area70}
\mathcal{\tilde{D}} = \left\{ (\tu,\tv) \in \mathcal{D}(0,v_0): \sup_{\mathcal{D}(0,\tu,\tv)} \mlm \tmfa \mrm \leq 4 D, \quad \sup_{\mathcal{D}(0,\tu,\tv)} \mlm \tp \mrm \leq 12 D \right \}.
\end{equation}
Note that $(0,0) \in \mathcal{\tilde{D}}$ and $\mathcal{\tilde{D}}$ is a closed subset of $\mathcal{D}(0,v_0)$. It remains to show that $\mathcal{\tilde{D}}$ is open.

First, observe that for each $(u,v) \in \mathcal{D}(0,v_0)$, we have
\begin{equation}\label{Area30}
r(u,v) = \int^v_u (\dv r)(u,v') \D v' \geq (v-u) \inf_{\mathcal{D}(0,v_0)} (\dv r) \geq \frac{5}{12}(v-u) \geq \frac{1}{3}(v-u).
\end{equation}
Furthermore, analogous to \eqref{Area17}, we get
\begin{equation}\label{Area41}
\tp = \frac{1}{r(u,v)}\int^v_u \tmfa (u,v') \D v'.
\end{equation}
Applying \eqref{Area30} to \eqref{Area41}, we hence obtain
\begin{equation}\label{Area31}
\sup_{\mathcal{D}(0,\tu,\tv)}|\tp| \leq \sup_{\mathcal{D}(0,\tu,\tv)}|\tmfa| \sup_{\mathcal{D}(0,\tu,\tv)}\mlm \frac{v - u}{r(u,v)}\mrm \leq 3 \sup_{\mathcal{D}(0,\tu,\tv)}|\tmfa|.
\end{equation}

Next, within $\mathcal{D}(0,\tu,\tv)$, we then compute $\du(\dv (r\phi))$ using \eqref{Setup5}, \eqref{Setup16}, \eqref{Setup17} as follows
\begin{equation}\label{Area32}
\begin{aligned}
\du \ml \dv(r\phi) \mr &= \du r \dv \phi + r \du \dv \phi + \du \phi \dv r + \phi \du \dv r \\
&= \phi \du \dv r - \ii \e  A_u \dv (r\phi) - \ii \e 
 \frac{\du r \dv r }{1 - \mu}\frac{Q\phi}{r}.
\end{aligned}
\end{equation}
The above implies 
\begin{equation}\label{Area42}
\du\ml \dv \ml r\tp \mr \mr = \tp \du \dv r- \ii \e  A_u \dv (r\tp) - \ii \e \phi_0 A_u \dv r  - \ii \e 
 \frac{\du r \dv r }{1 - \mu}\frac{Q\phi}{r}.
\end{equation}
By viewing \eqref{Area42} as a first order ODE for $\dv(r\tp)$ in the $u$-direction and noting that the integrating factor $e^{\int^{u}_0 \ii \e A_u(u',v) \D u'}$ only contributes to a phase change, we henceforth derive
\begin{equation}\label{Area33}
\begin{aligned}
& e^{\ii \e \int^{u}_0 A_u(u',v) \D u' }\dv(r\tp)(u,v) \\
 = &\; \dv(r\tp)(0,v) + \int^u_0 e^{\ii \e \int^{u'}_0 A_u(u'',v) \D u'' }\ml \tp \du \dv r - \ii \e \phi_0 A_u \dv r - \ii \e 
 \frac{\du r \dv r }{1 - \mu}\frac{Q\phi}{r} \mr(u',v) \D u',
\end{aligned}
\end{equation}
which implies that
\begin{equation}\label{Area34}
\begin{aligned}
|\dv(r\tp)|(u,v) \leq&\; |\dv (r\tp)|(0,v) + \int^u_0 |\tp|(-\du \dv r) (u',v)\D u'   \\
&+ \e \int^u_0 \frac{(-\du r)(\dv r)}{r(1-\mu)}|Q||\phi| (u',v)\D u'  + \e |\phi_0| \int^u_0 |A_u| \dv r (u',v) \D u'.
\end{aligned}
\end{equation}
In the above estimate, we have used the fact that $\mu - Q^2/r^2 \geq 0$ from our choice of $v_0$ in \eqref{Area64} and its implication $-\du \dv r \geq 0$ by \eqref{Setup20}. \\

It remains to estimate the integrals in \eqref{Area34}. By employing $\dv r\geq \frac{5}{12}$ and \eqref{Setup33}, the first integral obeys
\begin{equation}\label{Area35}
\begin{aligned}
\int^u_0 |\tp|(-\du \dv r) (u',v)\D u' \leq \ml \sup_{\mathcal{D}(0,\tu,\tv)}|\tp| \mr\ml (\dv r)(0,v) - (\dv r)(u,v) \mr \leq \frac{1}{12}\ml \sup_{\mathcal{D}(0,\tu,\tv)}|\tp| \mr.
\end{aligned}
\end{equation}
For the second integral in \eqref{Area34}, since $|\phi| \leq |\tp| + |\phi_0| \leq \ml \sup_{\mathcal{D}(0,\tu,\tv)}|\tp| \mr + |\phi_0|$ for every $(u',v) \in \mathcal{D}(0,\tu,\tv)$, we have
\begin{equation}\label{Area43}
\begin{aligned}
\e \int^u_0 \frac{(-\du r)(\dv r)}{r(1-\mu)}|Q||\phi| (u',v)\D u' \leq & \;  \e \ml \sup_{\mathcal{D}(0,\tu,\tv)}|\tp| \mr \int^u_0 \frac{(-\du r)(\dv r)}{r(1-\mu)}|Q| (u',v)\D u' \\
&+  \e |\phi_0| \int^u_0 \frac{(-\du r)(\dv r)}{r(1-\mu)}|Q| (u',v)\D u'.
\end{aligned}
\end{equation}
For the two integrals in \eqref{Area43}, we apply two different estimates for $|Q|$. For the first integral, since $\mathcal{D}(0,v_0) \subset \mathcal{R}$, we have $\mu \leq 1$ and $\mu - Q^2/r^2 \geq 0$, which in turn imply $|Q| \leq \mu^{\frac{1}{2}} r \leq r$. In light of Lemma \ref{SetupLemma1} and our assumptions in \eqref{Area27}, we thus have
\begin{equation}\label{Area36}
\begin{aligned}
 \e \ml \sup_{\mathcal{D}(0,\tu,\tv)}|\tp| \mr& \int^u_0 \frac{(-\du r)(\dv r)}{r(1-\mu)}|Q| (u',v)\D u' 
\leq \e \ml \sup_{\mathcal{D}(0,\tu,\tv)}|\tp| \mr \int^u_0 \frac{(-\du r)(\dv r)}{1- \mu}(u',v) \D u' \\
&\quad\quad \leq \e \ml \sup_{\mathcal{D}(0,\tu,\tv)}|\tp| \mr \frac{\dv r}{1 - \mu}(0,v)\int^u_0 (-\du r)(u',v) \D u'\\
&\quad\quad \leq \ml \e \frac{3v}{8}\mr \ml \sup_{\mathcal{D}(0,\tu,\tv)}|\tp| \mr \leq \ml \e \frac{3v_0}{8}\mr \ml \sup_{\mathcal{D}(0,\tu,\tv)}|\tp| \mr \leq \frac{1}{12}\ml \sup_{\mathcal{D}(0,\tu,\tv)}|\tp| \mr
\end{aligned}
\end{equation}
as long as we refine $v_0$ from \eqref{Area14} to be 
\begin{equation}\label{Area37}
v_0 = \min\left\{\oh, \frac{\sqrt{2}}{C_1 \ml 0, 1; \frac{1}{2}\mr}, \frac{2}{9\e} \right\}.
\end{equation}
For the second integral in \eqref{Area43}, we use \eqref{Setup19}, $Q|_{\Gamma} = 0$ and their implication 
\begin{equation}\label{Area44}
\begin{aligned}
|Q|(u,v) \leq 4 \pi \e \int^v_u r^2 |\phi||\dv \phi|(u,v') \D v' \leq 4 \pi \e \ml |\phi_0| + \sup_{\mathcal{D}(0,\tu,\tv)}|\tp| \mr r(u,v) \int^v_u r |\dv \phi|(u,v') \D v'. \\
\end{aligned}
\end{equation}
We further employ \eqref{Area30} and obtain
\begin{equation}\label{Area45}
\begin{aligned}
\int^v_u r|\dv \phi|(u,v') \D v' &= \int^v_u r|\dv \tp|(u,v') \D v' \leq \int^v_u |\dv (r\tp)|(u,v') + |\tp|(\dv r)(u,v') \D v' \\
&\quad \leq \ml \sup_{\mathcal{D}(0,\tu,\tv)}|\tmfa| \mr (v-u) + \ml \sup_{\mathcal{D}(0,\tu,\tv)}|\tp| \mr r(u,v) \leq 6\ml \sup_{\mathcal{D}(0,\tu,\tv)}|\tmfa| \mr r(u,v). \\
\end{aligned}
\end{equation}
Plugging \eqref{Area45} into \eqref{Area44}, we deduce that
\begin{equation}\label{Area47} 
|Q|(u,v) \leq 24\pi\e \ml \mlm \phi_0 \mrm + \sup_{\mathcal{D}(0,\tu,\tv)}|\tp| \mr \ml \sup_{\mathcal{D}(0,\tu,\tv)}|\tmfa| \mr  r(u,v)^2.
\end{equation}
Analogous to \eqref{Area36}, by using \eqref{Area47}, we estimate the second integral in \eqref{Area43} and get
\begin{equation}\label{Area48}
\begin{aligned}
\e |\phi_0| \int^u_0 \frac{(-\du r)(\dv r)}{r(1-\mu)}|Q| (u',v)\D u' 
&\leq 12 \pi v_0 \e^2 |\phi_0| \ml \sup_{\mathcal{D}(0,\tu,\tv)}|\tmfa| \mr \ml \mlm \phi_0 \mrm + \sup_{\mathcal{D}(0,\tu,\tv)}|\tp| \mr  \int^u_0 \frac{(-\du r)(\dv r)}{1- \mu}(u',v) \D u' \\
&\leq \frac{9}{2} \pi v_0^2 \e^2 |\phi_0| \ml \sup_{\mathcal{D}(0,\tu,\tv)}|\tmfa| \mr \ml \mlm \phi_0 \mrm + \sup_{\mathcal{D}(0,\tu,\tv)}|\tp| \mr. \\
\end{aligned}
\end{equation}
For the remaining integral in \eqref{Area34}, we begin by estimating $|A_u|$. For $(u',v) \in \tmcd$, employing \eqref{Setup30}, \eqref{Setup37}, \eqref{Area47}, we deduce
\begin{equation}\label{Area68}
\begin{aligned}
|A_u|(u',v) &\leq \int^{v}_{u'} \frac{2|Q| (-\du r) \dv r}{r^2(1-\mu)} (u',v') \D v' \\
&\leq 48 \pi \e \ml \mlm \phi_0 \mrm + \sup_{\mathcal{D}(0,u',v)}|\tp| \mr  \ml \sup_{\mathcal{D}(0,\tu,\tv)}|\tmfa| \mr \int^v_{u'} \frac{(-\du r)(\dv r)}{1-\mu}(u',v') \D v'  \\
&\leq 48 \pi \e \ml \mlm \phi_0 \mrm + \sup_{\mathcal{D}(0,u',v)}|\tp| \mr  \ml \sup_{\mathcal{D}(0,\tu,\tv)}|\tmfa| \mr r(u,v)  \\
&\leq 24 \pi v_0 \e \ml \sup_{\mathcal{D}(0,\tu,\tv)}|\tmfa| \mr \ml \mlm \phi_0 \mrm + \sup_{\mathcal{D}(0,u',v)}|\tp| \mr.
\end{aligned}
\end{equation}
The above further implies
\begin{equation}\label{Area69}
\e |\phi_0| \int^u_0 |A_u|(\dv r)(u',v) \D u' \leq 12 \pi v_0^2 \e^2 |\phi_0| \ml \sup_{\mathcal{D}(0,\tu,\tv)}|\tmfa| \mr \ml \mlm \phi_0 \mrm + \sup_{\mathcal{D}(0,\tu,\tv)}|\tp| \mr.
\end{equation}

\vspace{5mm}

We now plug \eqref{Area35}, \eqref{Area36}, \eqref{Area48} into \eqref{Area34}. Together with \eqref{Area31}, we obtain
\begin{equation}\label{Area50}
\begin{aligned}
|\dv(r\tp)|(u,v) =& \; |\dv (r\tp)|(0,v) + \int^u_0 |\tp|(-\du \dv r) (u',v)\D u' + \e \int^u_0 \frac{(-\du r)(\dv r)}{r(1-\mu)}|Q||\phi| (u',v)\D u' \\
&+ \e |\phi_0| \int^u_0 |A_u| \dv r (u',v) \D u'\\ 
\leq& \; D + \frac{1}{6}\ml \sup_{\mathcal{D}(0,\tu,\tv)}|\tp| \mr + \frac{33}{2} \pi \e^2 |\phi_0| v_0^2 \ml \mlm \phi_0 \mrm + \sup_{\mathcal{D}(0,\tu,\tv)}|\tp|  \mr \ml \sup_{\mathcal{D}(0,\tu,\tv)}|\tmfa| \mr.
\end{aligned}
\end{equation}
In \eqref{Area50}, we use $D := TV_{\{0\} \times (0,v_0)}[\tmfa]$, $\tmfa(0,0) = 0$ and the implication $|\dv(r\tp)|(0,v) \leq D$. Taking supremum over $\mathcal{D}(0,\tu,\tv)$ and applying \eqref{Area70} and \eqref{Area31}, we deduce that
\begin{equation}\label{Area71}
\begin{aligned}
\ml \sup_{\mathcal{D}(0,\tu,\tv)}|\tmfa| \mr &\leq D + \frac{1}{6}\ml \sup_{\mathcal{D}(0,\tu,\tv)}|\tp| \mr + \frac{33}{2} \pi \e^2 |\phi_0| v_0^2 \ml \mlm \phi_0 \mrm + \sup_{\mathcal{D}(0,\tu,\tv)}|\tp|  \mr \ml \sup_{\mathcal{D}(0,\tu,\tv)}|\tmfa| \mr \\
&\leq D +  \frac{1}{2}\ml \sup_{\mathcal{D}(0,\tu,\tv)}|\tmfa| \mr + \frac{33}{2} \pi \e^2 |\phi_0| v_0^2 \ml \mlm \phi_0 \mrm + 12 \mr \ml \sup_{\mathcal{D}(0,\tu,\tv)}|\tmfa| \mr \\
&\leq D +  \frac{2}{3}\ml \sup_{\mathcal{D}(0,\tu,\tv)}|\tmfa| \mr,
\end{aligned}
\end{equation}
which gives
\begin{equation}\label{Area52}
\ml \sup_{\mathcal{D}(0,\tu,\tv)}|\tmfa| \mr \leq 3D.
\end{equation}
Note that in \eqref{Area71} we further restrict our choice of $v_0$ as in \eqref{Area37} to 
\begin{equation}\label{Area72}
v_0 = \min\left\{\oh, \frac{\sqrt{2}}{C_1 \ml 0, 1; \frac{1}{2}\mr}, \frac{2}{9\e}, \frac{1}{\sqrt{99\pi\e^2|\phi_0|\ml |\phi_0| + 12 \mr}} \right\}.
\end{equation}
Furthermore, by \eqref{Area31}, inequality \eqref{Area52} also leads to 
\begin{equation}\label{Area53}
\sup_{\mathcal{D}(0,\tu,\tv)}|\tp| \leq 3\ml \sup_{\mathcal{D}(0,\tu,\tv)}|\tmfa| \mr \leq 9D.
\end{equation}
We thus retrieve our bootstrap estimate in \eqref{Area70} and deduce that $\tmcd$ is an open subset of $\mathcal{D}(0,v_0)$. Moreover, the derived estimates also extend to the entire region $\mathcal{D}(0,v_0)$. 

\vspace{5mm}

We proceed to estimate $B$. {\color{black}In the following arguments, we take $(u,v) \in \mathcal{D}(0,v_0)$ and set $\tu = \tv = v_0$.} By viewing \eqref{Area42} as an ODE in the $v$-direction for $\du \ml r \tp \mr$ and by applying \eqref{Setup38}, we deduce
\begin{equation}\label{Area54}
\begin{aligned}
\mlm \du \ml r\tp \mr \mrm (u,v)\leq & \; \mlm \dv \ml r\tp \mr \mrm (u,u) +   \int^v_u |\tp|(- \du \dv r)(u,v') \D v'  + \e \int^v_u |A_u| \mlm \dv (r \tp)\mrm  (u,v') \D v' \\
&+ \e |\phi_0|  \int^v_u |A_u|(\dv r) (u,v') \D v' + \e \int^v_u \frac{(-\du r) \dv r}{1 - \mu}\frac{|Q| |\phi|}{r} (u,v') \D v'. \\
\end{aligned} 
\end{equation}
Analogous to the previous case, since $\du\dv r \leq 0$ in $\mathcal{D}(0,v_0)$, we bound the first integral in \eqref{Area54} by
\begin{equation}\label{Area60}
\begin{aligned}
\int^v_u |\tp|(-\du \dv r) (u,v')\D v' &\leq P\ml (-\du r)(u,v) - (\dv r)(u,u) \mr \leq \frac{P}{4} \leq 3D.
\end{aligned}
\end{equation}
Furthermore, by \eqref{Area68}, we get
\begin{equation}\label{Area55}
\begin{aligned}
|A_u|(u,v') \leq 24 \pi \e (|\phi_0| + P)A v_0 \leq 96 \pi \e (|\phi_0| + 12) v_0 D.
\end{aligned}
\end{equation}
With these, we can see that the second and the third integral in \eqref{Area54} obey
\begin{equation}\label{Area58}
\begin{aligned}
\e \int^v_u |A_u| \mlm \dv (r \tp)\mrm  (u,v') \D v' &\leq 96 \pi\e^2 (|\phi_0| + 12) v_0^2 A D \leq 384 \pi\e^2 (|\phi_0| + 12) v_0^2 D \leq \frac{D}{4}
\end{aligned}
\end{equation}
and 
\begin{equation}\label{Area73}
\begin{aligned}
\e |\phi_0| \int^v_u |A_u| (\dv r) (u,v') \D v' &\leq 96 \pi\e^2 |\phi_0|(|\phi_0| + 12) v_0^2 D \leq \frac{D}{4},
\end{aligned}
\end{equation}
if we refine our choice of $v_0$ from \eqref{Area72} to 
\begin{equation}\label{Area59}
v_0 = \min\left\{\oh, \frac{\sqrt{2}}{C_1 \ml 0, 1; \frac{1}{2}\mr}, \frac{2}{9\e}, \frac{1}{\sqrt{384\pi\e^2|\phi_0|\ml |\phi_0| + 12 \mr}}, \frac{1}{\sqrt{1536\pi}\e\sqrt{|\phi_0| + 12}}\right\}.
\end{equation}
To estimate the last integral in \eqref{Area54}, by  \eqref{Area47}, we first note
\begin{equation}\label{Area62}
|Q\phi| \leq 12 \pi \e \ml \mlm \phi_0 \mrm + 12 \mr^2 A  v_0 r(u,v) \leq 48\pi \e \ml \mlm \phi_0 \mrm + 12 \mr^2 D  v_0 r(u,v).
\end{equation}
Henceforth, the last integral in \eqref{Area54} satisfies
\begin{equation}\label{Area63}
\begin{aligned}
\e \int^v_u \frac{(-\du r) \dv r}{1 - \mu}\frac{|Q| |\phi|}{r} (u,v') \D v' &\leq 48 \pi \e^2 \ml |\phi_0| + 12\mr^2 v_0 D \int^v_u \frac{(-\du r) \dv r}{1 - \mu} (u,v') \D v' \\
&\leq 24 \pi \e^2 \ml |\phi_0| + 12\mr^2 v_0^2 D \leq \frac{D}{2}
\end{aligned}
\end{equation}
if we restrict our choice of $v_0$ from \eqref{Area59} to 
\begin{equation}\label{Area64}
v_0 = \min\left\{\oh, \frac{\sqrt{2}}{C_1 \ml 0, 1; \frac{1}{2}\mr}, \frac{2}{9\e}, \frac{1}{\sqrt{384\pi\e^2|\phi_0|\ml |\phi_0| + 12 \mr}}, \frac{1}{\sqrt{1536\pi}\e\sqrt{|\phi_0| + 12}}, \frac{1}{\sqrt{48\pi}\e(|\phi_0| + 12)}\right\}.
\end{equation}
Plugging our bounds from \eqref{Area60}, \eqref{Area58}, \eqref{Area63} into \eqref{Area54}, we deduce that
\begin{equation}\label{Area65}
B \leq A + 3D + \frac{D}{2} + \frac{D}{2} \leq 8D.
\end{equation}
This then completes the proof of Lemma \ref{AreaLemma2}.
\end{proof}

With the above lemma proven, we return to the proof of Theorem \ref{AreaLemma3}.
\begin{proof}By Lemma \ref{AreaLemma2} we can now write 
\begin{equation}\label{Area81}
A \lesssim D, \quad B \lesssim D, \quad P \lesssim D.
\end{equation}
Furthermore, recall that in the proof of Lemma \ref{AreaLemma2}, we chose $v_0$ such that $\du \dv r \leq 0$. Combined with the fact that $\mathcal{D}(0,v_0) \subset \mathcal{R}$ and $\dv r|_{\Gamma} = -\du r|_{\Gamma}$, for all $(u,v) \in \mathcal{D}(0,v_0)$, \eqref{Area27} then implies
\begin{equation}\label{Area84}
\frac{5}{12} \leq \dv r(u,v) \leq \frac{1}{2}, \quad \text{ and } \quad \frac{1}{2} \leq (-\du r)(u,v) \leq \frac{2}{3}.
\end{equation}
Henceforth, applying
\begin{equation}\label{Area66}
r \dv \tp = \dv \ml r \tp \mr - \tp \dv r, \quad r \du \tp = \du \ml r \tp \mr - \tp \du r, \quad \phi = \phi_0 + \tp,
\end{equation}
and $D \leq \varepsilon_0 \leq 1$, we then get
\begin{equation}\label{Area82}
\sup_{\mathcal{D}(0,v_0)}\mlm r \dv \tp \mrm \lesssim D, \quad \sup_{\mathcal{D}(0,v_0)}\mlm r \du \tp \mrm \lesssim D, \quad \sup_{\mathcal{D}(0,v_0)}\mlm \phi \mrm \lesssim 1.
\end{equation}
Plugging these estimates into \eqref{Area47} and \eqref{Area68}, for each $(u,v) \in \mathcal{D}(0,v_0)$, we further have
\begin{equation}\label{Area83}
|Q|(u,v) \lesssim D r^2(u,v) \quad \text{ and } \quad |A_u|(u,v) \lesssim Dr(u,v).
\end{equation}
Note that if needed, we can also use $r(u,v) \leq r(0,v_0) = \frac{1}{2}v_0 \leq 1$. As a consequence, we obtain
\begin{equation}\label{Area114}
|rD_u \phi|(u,v) = |r\du \phi + \ii \e A_u r \phi|(u,v) \lesssim D + Dr(u,v)^2 \lesssim D.
\end{equation} 

On top of the estimates in the region $\mathcal{D}(0,v_0)$, we also derive some asymptotics (in $r$) for $C^1$ solutions at points sufficiently close to $\Gamma$. This is to justify that various terms vanish as we take the limit along incoming or (backward) outgoing null curves towards points on $\Gamma$. For a $C^1$ solution, we have that $\phi, r \dv \phi, \text{ and }r \du \phi$ are $C^1$ functions on $\overline{\mathcal{D}(0,v_0)}$. Since $\dv \phi$ and $\du \phi$ are continuous on $\overline{\mathcal{D}(0,v_0)}$, these imply that $r \dv \phi|_{\Gamma} = r \du \phi|_{\Gamma} = 0$. 
By the fundamental theorem of calculus and the extreme value theorem, there exists a constant $M(v_0)$ depending on $v_0$ such that
\begin{equation}\label{Area310}
r|\dv \phi|(u,v) \leq M(v_0) (v-u) \quad \text{ and } \quad r|\du \phi| \leq M(v_0)(v-u).
\end{equation}
Under the bootstrap assumption, similar to the proof in \eqref{Area30}, we now have $v-u \leq 3r(u,v)$. Hence, we get
\begin{equation}\label{Area311} 
|\dv \tilde{\phi}|(u,v) = |\dv \phi|(u,v) \leq 3M(v_0) \quad \text{ and } \quad |\du \tilde{\phi}|(u,v) = |\du \phi| \leq 3M(v_0).
\end{equation}
In addition, since $\dv r$ and $\du r$ are also $C^1$ functions, a similar argument allows us to deduce that there exists a constant $M_1(v_0)$ such that
\begin{equation}\label{Area521}
|\tmfa| (u,v) \leq M_1(v_0)r(u,v) \quad \text{ and } \quad |\tmfb|(u,v) \leq M_1(v_0)r(u,v).
\end{equation}
Additionally, as $\phi$, $\dv r$, and $\du r$ are $C^1$ functions on $\overline{\mathcal{D}(0,v_0)}$, these are bounded functions on $\overline{\mathcal{D}(0,v_0)}$.
Furthermore, since $Q/r$ is yet another $C^1$ function, there exists a constant $M_2(v_0)$ such that
\begin{equation}\label{Area312}
|Q|(u,v) \leq M_2(v_0)r^2(u,v) \quad \text{ and } \quad |A_u|(u,v) \leq M_2(v_0)r(u,v).
\end{equation}
Using \eqref{Setup25}, together with \eqref{Area84} and $m|_{\Gamma} = 0$, we derive
\begin{equation}\label{Area313}
\begin{aligned}
m(u,v) &\leq \int_u^v \frac{2\pi r^2(1-\mu)|\dv \phi|^2}{\dv r}(u,v') + \frac{Q^2}{r^4}\cdot(r^2 \dv r)(u,v') \D v' \\
&\leq \frac{24}{5}(3M(v_0))^2 \pi r^2 (v - u) + \frac{M_2(v_0)^2}{2} r^2(u,v) (v - u) \\
&\leq \ml \frac{648}{5}\pi^2 M^2(v_0) + \frac{3M_2(v_0)^2}{2}\mr r^3(u,v).
\end{aligned}
\end{equation}
This implies
\begin{equation}\label{Area314}
\mu(u,v) \leq \ml \frac{1296}{5}\pi^2 M^2(v_0) + 3M_2(v_0)^2\mr r^2(u,v).
\end{equation}

\vspace{5mm}

We are back to derive estimates for $\tmfa$. Differentiating both sides of \eqref{Area42} with respect to $v$, we get
\begin{equation}
\begin{aligned}
\du \dv \tmfa &= \dv\ml \tp \du \dv r\mr - \ii \e A_u \dv \tmfa - \ii \e \phi_0 (\dv A_u) \dv r - \ii \e \phi_0 A_u \dv \lambda  - \ii \e \dv \ml \frac{\du r \dv r}{1-\mu}\frac{Q \phi}{r}\mr. \\
\end{aligned}
\end{equation}
This can be rewritten as 
\begin{equation}\label{Area67}
\begin{aligned}
\du \ml \dv \tmfa \mr &= \du\ml \tp \dv \lambda \mr + G_1 + G_2 + \frac{\p (\lambda, \tp)}{\p(u,v)}
\end{aligned}
\end{equation}
with
\begin{equation}\label{Area77}
\begin{aligned}
G_1 &= -\ii \e \ml \phi_0 A_u + \frac{\du r }{1-\mu}\frac{Q\phi}{r} \mr \dv \lambda, \\
G_2 &=  -\ii \e \ml (\dv A_u) \tmfa + \phi_0 (\dv A_u)\dv r + A_u \dv \tmfa + \dv r \dv \ml \frac{\du r}{1-\mu}\frac{Q\phi}{r}\mr\mr.
\end{aligned}
\end{equation}
For each $(u,v) \in \mathcal{D}(0,v_0)$, we then have
\begin{equation}\label{Area78}
\begin{aligned}
\mlm \dv \tmfa - \tp \dv \lambda \mrm (u,v) \leq \mlm \dv \tmfa - \tp \dv \lambda \mrm(0,v) + \int^u_0 \mlm G_1 + G_2 + \frac{\p(\lambda,\tp)}{\p(u,v)}\mrm(u',v) \D u'.
\end{aligned}
\end{equation}
Using the fact $\lambda(0,v) = \frac{1}{2}$ and integrating \eqref{Area78} with respect to $v$ from $u$ to $v_0$, we obtain
\begin{equation}\label{Area79}
\begin{aligned}
\int_u^{v_0} \mlm \dv \tmfa - \tp \dv \lambda \mrm (u,v') \D v' \leq \int^{v_0}_u \mlm \dv \tmfa \mrm(0,v') \D v' + \int_u^{v_0}\int^u_0 \ml |G_1| + |G_2| + \mlm\frac{\p(\lambda,\tp)}{\p(u,v)}\mrm \mr(u',v') \D u' \D v'.
\end{aligned}
\end{equation}
By taking the supremum over $u \in (0,v_0)$ in \eqref{Area70} and using \eqref{Area39}, \eqref{Area29}, we arrive at 
\begin{equation}\label{Area80}
\begin{aligned}
\sup_{u \in (0,v_0)}TV_{\{u\}\times(u,v_0)}[\tmfa] \leq& \; D + P \sup_{u \in (0,v_0)}TV_{\{u\}\times(u,v_0)}[\lambda] + A[\lambda,\tp] \\
&+ \sup_{u\in(0,v_0)}\int_u^{v_0}\int^u_0 |G_1| (u',v') \D u' \D v' + \sup_{u\in(0,v_0)}\int_u^{v_0}\int^u_0 |G_2| (u',v') \D u' \D v'.
\end{aligned}
\end{equation}
Next, we turn to estimate the integrals containing $G_1$ and $G_2$ in \eqref{Area80}. Observe that, for each $(u,v) \in \mathcal{D}(0,v_0)$, the following inequality holds uniformly in $(u,v)$ 
\begin{equation}\label{Area85}
\mlm \phi_0 A_u + \frac{\du r}{1- \mu}\frac{Q\phi}{r}\mrm(u,v) \lesssim D r(u,v) \lesssim D.
\end{equation} 
This implies that for every $u \in (0,v_0)$, we have
\begin{equation}\label{Area86}
\begin{aligned}
\int_u^{v_0}\int^u_0 |G_1| (u',v') \D u' \D v' &\lesssim D \int_0^{u}\int^{v_0}_u |\dv \lambda| (u',v') \D v' \D u' \lesssim D \sup_{u \in (0,v_0)}TV_{\{u\}\times(u,v_0)}[\lambda],  
\end{aligned}
\end{equation}
where we employ the fact $u \leq v_0 \leq 1$ and apply Fubini's theorem to interchange the order of integrations. 

To estimate the integrand $G_2$, we start by writing
\begin{equation}\label{Area87}
\begin{aligned}
\dv \ml \frac{\du r}{1-\mu}\frac{Q \phi}{r}\mr &= \frac{\dv \du r}{1-\mu}\frac{Q\phi}{r} + \frac{\du r}{1-\mu}\frac{(\dv Q)\phi}{r} + \frac{\du r}{1-\mu}\frac{Q \dv \phi}{r} - \frac{\du r}{1-\mu}\frac{Q\phi}{r^2}\dv r + \frac{\du r}{(1-\mu)^2}\frac{Q\phi}{r}\dv \mu
\end{aligned}
\end{equation}
and proceed to estimate each of these terms. For the first term, by \eqref{Setup20}, we have
\begin{equation}\label{Area88}
\begin{aligned}
\mlm \frac{\dv \du r}{1-\mu}\frac{Q\phi}{r} \mrm(u,v) &= \mlm \frac{(\dv r)(\du r)}{(1-\mu)^2}\frac{Q\phi}{r^2}\ml \mu - \frac{Q^2}{r^2}\mr \mrm(u,v) \lesssim D,
\end{aligned}
\end{equation}
where we have used \eqref{Area29}, \eqref{Area81}, \eqref{Area82}, \eqref{Area83}, \eqref{Area84} and the fact that $|\mu - Q^2/r^2| = \mu - Q^2/r^2 \leq \mu \leq 1/3$. In the remaining parts of this section, we will not explicitly mention that these are used unless we really need. In a similar fashion, for the second term in \eqref{Area87}, by \eqref{Setup19}, we obtain
\begin{equation}\label{Area89}
\begin{aligned}
\mlm \frac{\du r}{1-\mu}\frac{(\dv Q)\phi}{r}\mrm(u,v) = \mlm \frac{\du r}{1-\mu}\frac{\phi}{r}(-4\pi \e r^2\im(\phi^\dagger \dv \phi))\mrm(u,v) \lesssim D.
\end{aligned}
\end{equation}
The third term in \eqref{Area87} can be estimated directly as
\begin{equation}\label{Area90}
\begin{aligned}
\mlm \frac{\du r}{1-\mu}\frac{Q \dv\phi}{r}\mrm(u,v) \lesssim D^2 \lesssim D.
\end{aligned}
\end{equation}
The fourth term in \eqref{Area87} obeys
\begin{equation}\label{Area91}
\begin{aligned}
\mlm \frac{\du r}{1-\mu}\frac{Q \phi}{r^2}\dv r\mrm(u,v) \lesssim D.
\end{aligned}
\end{equation}
For the fifth term in \eqref{Area87}, we appeal to \eqref{Setup27} and obtain
\begin{equation}\label{Area92}
\begin{aligned}
\mlm \frac{\du r}{(1-\mu)^2}\frac{Q \phi}{r}\dv \mu\mrm(u,v) &= \mlm \frac{\du r}{(1-\mu)^2}\frac{Q \phi}{r}\ml \frac{4\pi r(1-\mu)|\dv \phi|^2}{\dv r} - \frac{\dv r}{r}\ml \mu - \frac{Q^2}{r^2}\mr \mr\mrm(u,v)  \lesssim D^3 + D \lesssim D.
\end{aligned}
\end{equation}
Plugging \eqref{Area88}, \eqref{Area89}, \eqref{Area90}, \eqref{Area91}, \eqref{Area92} into \eqref{Area87}, we thus deduce that
\begin{equation}\label{Area93}
\mlm \dv \ml \frac{\du r}{1-\mu}\frac{Q \phi}{r}\mr\mrm (u,v) \lesssim D.
\end{equation}
For the other terms of $G_2$ in \eqref{Area77}, by \eqref{Setup30} and $A \lesssim D$, we get
\begin{equation}\label{Area94}
\mlm (\dv A_u)\tmfa\mrm(u,v) = \mlm \frac{2Q\du r \dv r}{r^2(1-\mu)}\tmfa\mrm(u,v) \lesssim D^2 \lesssim D.
\end{equation} 
Similarly, we deduce that
\begin{equation}\label{Area95}
\mlm \phi_0(\dv A_u)\dv r\mrm(u,v) = \mlm \frac{2Q\du r \dv r}{r^2(1-\mu)}\phi_0 \dv r\mrm(u,v) \lesssim D.
\end{equation}
For the term $A_u \dv \tmfa$, we first observe from \eqref{Area55} that
\begin{equation}\label{Area96}
|A_u|(u,v) \leq 96v_0 \e (|\phi_0| + 12)D.
\end{equation}
By modifying our choice of $v_0$ from \eqref{Area64} to 
\begin{equation}\label{Area97}
v_0 = \min\left\{\oh, \frac{\sqrt{2}}{C_1 \ml 0, 1; \frac{1}{2}\mr}, \frac{2}{9\e}, \frac{1}{\sqrt{384\e^2|\phi_0|\ml |\phi_0| + 12 \mr}}, \frac{1}{\sqrt{1536}\e\sqrt{|\phi_0| + 12}}, \frac{3}{128\e(|\phi_0| + 12)}\right\},
\end{equation}
we obtain
\begin{equation}\label{Area98}
|A_u|(u,v) \leq \frac{9}{4}D \leq \frac{9}{4},
\end{equation}
which in turn implies that
\begin{equation}\label{Area99}
\mlm A_u \dv \tmfa \mrm(u,v) \leq  \frac{9}{4}\mlm \dv \tmfa \mrm(u,v).
\end{equation}
Consequently, we prove
\begin{equation}\label{Area100}
\begin{aligned}
\int^{v_0}_u \int^u_0 |A_u \dv \tmfa|(u',v') \D u' \D v' &= \int^u_0  \int^{v_0}_u |A_u \dv \tmfa|(u',v') \D v' \D u' \leq \frac{1}{2\e} \sup_{u \in (0,v_0)}TV_{\{u\} \times (u,v_0) }[\tmfa].
\end{aligned}
\end{equation}
Combining \eqref{Area93}, \eqref{Area94}, \eqref{Area95}, \eqref{Area100}, together with \eqref{Area77}, we thus arrive at 
\begin{equation}\label{Area101}
\int_u^{v_0}\int^u_0 |G_2| (u',v') \D u' \D v' - \frac{1}{2}\sup_{u \in (0,v_0)}TV_{\{u\} \times (u,v_0) }[\tmfa] \lesssim  \int_u^{v_0}\int^u_0 D \D u' \D v' \lesssim D.
\end{equation}
Since the estimates in \eqref{Area86} and \eqref{Area101} are uniform in $u \in (0,v_0)$, taking supremum over $u \in (0,v_0)$ and substituting them into \eqref{Area80}, it yields
\begin{equation}\label{Area103}
\begin{aligned}
\frac{1}{2}\sup_{u \in (0,v_0)}TV_{\{u\}\times(u,v_0)}[\tmfa] \lesssim & \; D\ml 1 + \sup_{u \in (0,v_0)}TV_{\{u\}\times(u,v_0)}[\lambda]\mr  + A[\lambda,\tp]. \\
\end{aligned}
\end{equation}
{\color{black}Note that in \eqref{Area103}, we are allowed to shift the term $\frac{1}{2}\sup_{u \in (0,v_0)}TV_{\{u\}\times(u,v_0)}[\tmfa]$ in the inequality as it is a finite term for a $C^1$ solution.}

\vspace{5mm}

Next, we move to derive estimates for $\tmfb$. In view of the definition of $\tmfb$ in \eqref{tildemathfrakb}, we define a renormalized quantity
\begin{equation}\label{tildetildemathfrakb}
\ttmfb := \tmfb + \ii \e A_u r \phi = \tp \du r + r \du \tp + \ii \e A_u r \phi = \tp \du r + r D_u \phi.
\end{equation}
By \eqref{Setup73}, taking derivative of it with respect to $v$, it yields
\begin{equation}\label{Area104}
\begin{aligned}
\dv \ttmfb &= \tp \du \dv r + \ii \e \frac{\du r \dv r}{1 - \mu}\frac{Q\phi}{r}.
\end{aligned}
\end{equation}
Upon taking an additional derivative with respect to $u$, we arrive at 
\begin{equation}
\begin{aligned}
\du \dv \ttmfb &= \du \ml \tp \du \dv r \mr + \ii \e \du \ml \frac{\du r \dv r}{1 - \mu}\frac{Q\phi}{r} \mr.
\end{aligned}
\end{equation}
This can be rewritten as
\begin{equation}\label{Area105}
\dv \du \ttmfb = \dv \ml \tp \du \nu \mr + G_3 + G_4 - \frac{\p(\nu,\tp)}{\p(u,v)}
\end{equation}
with
\begin{equation}\label{Area106}
\begin{aligned}
G_3 = \ii \e \du \nu \ml \frac{\dv r}{1-\mu}\frac{Q\phi}{r}\mr \quad \text{ and } \quad G_4 = \ii \e \du r \du \ml \frac{\dv r}{1-\mu}\frac{Q\phi}{r} \mr. 
\end{aligned}
\end{equation}
For each $(u,v) \in \mathcal{D}(0,v_0)$, by integrating along the future-directed outgoing null curve originating from $\Gamma$, we obtain
\begin{equation}\label{Area107}
\begin{aligned}
\mlm \du \ttmfb - \tp \du \nu\mrm(u,v) &\leq \mlm \du \ttmfb - \tp \du \nu\mrm_{\Gamma}(u) + \int^v_u \mlm G_3 + G_4 - \frac{\p(\nu,\tp)}{\p(u,v)} \mrm (u,v') \D v',
\end{aligned}
\end{equation}
where we use \eqref{Setup74}, \eqref{Setup75}, \eqref{Setup38}. 
{\color{black}Applying \eqref{Setup20}, together with $0 \leq \mu - Q^2/r^2 \leq \mu$, \eqref{Area311}, \eqref{Area312}, \eqref{Area314}, we derive
\begin{equation}\label{Area316}
\begin{aligned}
|\dv \nu|(u,v) = \frac{\mu - Q^2/r^2}{1-\mu}\frac{\lambda |\nu|}{r} \lesssim r(u,v).
\end{aligned}
\end{equation}
Employing \eqref{Area104}, we also get
\begin{equation}\label{Area317}
\begin{aligned}
|\dv \ttmfb|(u,v) \lesssim D r(u,v) + Dr^2(u,v) \lesssim r(u,v).
\end{aligned}
\end{equation}
Henceforth, we deduce that
\begin{equation}\label{Area122}
\lim_{v \rightarrow u^+} \dv \nu(u,v) = \lim_{v \rightarrow u^+} \dv \ttmfb(u,v) = 0.
\end{equation}}
Since by \eqref{Setup37} and \eqref{Setup38}, we also have
\begin{equation}\label{Area116}
\ttmfb|_{\Gamma} = \tmfb|_{\Gamma} = -\tmfa|_{\Gamma} \quad \mbox{ and } \quad \lambda|_{\Gamma} = -\nu|_{\Gamma}.
\end{equation}
By the chain rule, we deduce
\begin{equation}\label{Area120}
\lim_{v \rightarrow u^+} \du \nu(u,v) = \lim_{v \rightarrow u^+} \frac{\D \nu|_{\Gamma}(u)}{\D u} = -\frac{\D \lambda|_{\Gamma}(u)}{\D u}
\end{equation}
and
\begin{equation}\label{Area121}
\lim_{v \rightarrow u^+} \du \tmfb(u,v) = \lim_{v \rightarrow u^+} \frac{\D \tmfb|_{\Gamma}(u)}{\D u} = -\frac{\D \tmfa|_{\Gamma}(u)}{\D u}.
\end{equation}
In addition, \eqref{Area116} also implies that
\begin{equation}\label{Area118}
\begin{aligned}
TV_{(0,v_0)}[\ttmfb|_\Gamma] = TV_{(0,v_0)}[\tmfa|_\Gamma], \quad TV_{(0,v_0)}[\nu|_\Gamma] = TV_{(0,v_0)}[\lambda|_\Gamma].
\end{aligned}
\end{equation}
Hence, by integrating \eqref{Area107} along a future-directed incoming null curve originating from $C_0^+$ and applying \eqref{Area122} - \eqref{Area113}, we obtain
\begin{equation}\label{Area108a}
\begin{aligned}
\int^v_0 \mlm \du \ttmfb - \tp \du \nu\mrm(u',v) \D u' \leq&\; TV_{(0,v_0)}[\ttmfb|_{\Gamma}] + P \cdot TV_{(0,v_0)}[\nu|_{\Gamma}] + A[\nu,\tp] \\
&\; + \int^v_0 \int^v_{u'} |G_3|(u',v') \D v' \D u' + \int^v_0 \int^v_{u'} |G_4|(u',v') \D v' \D u'.
\end{aligned}
\end{equation}
We proceed to estimate the integrals containing $G_3$ and $G_4$ in \eqref{Area108a}. For the former, for each $(u,v) \in \mathcal{D}(0,v_0)$, we have
\begin{equation}\label{Area109}
|G_3|(u,v) \lesssim D r(u,v) |\du \nu|(u,v) \lesssim D|\du \nu|(u,v)
\end{equation} and hence
\begin{equation}\label{Area110}
\begin{aligned}
\int^v_0 \int^v_{u'} |G_3|(u',v') \D v' \D u' &= \int^v_0 \int^{v'}_{0} |G_3|(u',v') \D u' \D v' \lesssim D \cdot \sup_{v \in (0,v_0)}TV_{(0,v) \times \{v\}}[\nu].
\end{aligned}
\end{equation}
For the integrand $G_4$, similar to equation \eqref{Area87}, we now have
\begin{equation}\label{Area111}
\du \ml \frac{\dv r}{1-\mu}\frac{Q \phi}{r}\mr = \frac{\dv \du r}{1-\mu}\frac{Q\phi}{r} + \frac{\dv r}{1-\mu}\frac{(\du Q)\phi}{r} + \frac{\dv r}{1-\mu}\frac{Q \du \phi}{r} - \frac{\dv r}{1-\mu}\frac{Q\phi}{r^2}\du r + \frac{\dv r}{(1-\mu)^2}\frac{Q\phi}{r}\du \mu.
\end{equation}
Notice that each of the terms above is analogous to the corresponding ones in \eqref{Area87}. Their estimates are exactly symmetric (with the use of \eqref{Setup18} and \eqref{Setup26} instead of \eqref{Setup19} and \eqref{Setup27} for the second and the fifth terms). Hence, by utilizing \eqref{Area114}, we then deduce that
\begin{equation}\label{Area112}
\mlm \du \ml \frac{\dv r}{1-\mu}\frac{Q \phi}{r}\mr \mrm (u,v)  \lesssim D,
\end{equation}
which implies that
\begin{equation}\label{Area113}
\int^v_0 \int^v_{u'} |G_4|(u',v') \D v' \D u' \lesssim D.
\end{equation}
Plugging \eqref{Area110} and \eqref{Area113} into \eqref{Area108a}, we hence derive
\begin{equation}\label{Area108}
\begin{aligned}
\int^v_0 \mlm \du \ttmfb - \tp \du \nu\mrm(u',v) \D u' 
\lesssim & \; TV_{(0,v_0)}[\tmfa|_\Gamma] + P \cdot TV_{(0,v_0)}[\lambda|_\Gamma] + A[\nu,\tp] + D \cdot \sup_{v \in (0,v_0)}TV_{(0,v) \times \{v\}}[\nu] + D.\\
\end{aligned}
\end{equation}
This implies that 
\begin{equation}\label{Area318}
\int^v_0 \mlm \du \ttmfb  \mrm (u',v) \D u' \lesssim TV_{(0,v_0)}[\tmfa|_\Gamma] + P \cdot TV_{(0,v_0)}[\lambda|_\Gamma] + A[\nu,\tp] + D \cdot \sup_{v \in (0,v_0)}TV_{(0,v) \times \{v\}}[\nu] + D.
\end{equation}
Taking supremum over $v \in (0,v_0)$, we then conclude that
\begin{equation}\label{Area115}
\begin{aligned}
\sup_{v \in (0,v_0)}TV_{(0,v) \times \{v\}}[\ttmfb] \lesssim & \; TV_{(0,v_0)}[\tmfa|_\Gamma] + D \ml 1 + TV_{(0,v_0)}[\lambda|_{\Gamma}] + \sup_{v \in (0,v_0)} TV_{(0,v) \times \{v\}}[\nu]\mr  + A[\nu,\tp].
\end{aligned}
\end{equation}

By scrutinizing the terms in the upper bound in \eqref{Area115}, it necessitates estimating $TV_{(0,v_0)}[\tmfa|_\Gamma]$, $TV_{(0,v_0)}[\lambda|_{\Gamma}]$ and $\sup_{v \in (0,v_0)} TV_{(0,v) \times \{v\}}[\nu]$. We now estimate each of them. Starting from $TV_{(0,v_0)}[\tmfa|_\Gamma]$, for a fixed $u \in (0,v_0)$, we integrate \eqref{Area67} along the future-directed incoming null curve $v = u$ starting from $C_0^+$ to $(u,u) \in \Gamma$ and obtain
\begin{equation}\label{Area119}
\begin{aligned}
\dv \tmfa|_{\Gamma}(u) =&\; \tp|_{\Gamma}(u) \dv \lambda|_{\Gamma}(u) + (\dv \tmfa)(0,u) - \ml \tp \dv \lambda \mr (0,u) \\
&+ \int^u_0 G_1(u',u) + G_2(u',u) + \frac{\p(\lambda,\tp)}{\p(u,v)}(u',u) \D u'.
\end{aligned}
\end{equation}
Using similar estimates as in \eqref{Area316} and \eqref{Area317}, it is not difficult to see that we have 
\begin{equation}\label{Area123}
\lim_{v \rightarrow u^+} \du \lambda(u,v) = \lim_{v \rightarrow u^+} \du \tmfa(u,v) = 0.
\end{equation}
Together with the fact that $\lambda(0,v) = \frac{1}{2}$, this implies that
\begin{equation}\label{Area124}
\begin{aligned}
\frac{\D}{\D u} \tmfa|_{\Gamma}(u) =&\; \tp|_{\Gamma}(u) \frac{\D}{\D u} \lambda|_{\Gamma}(u) + (\dv \tmfa)(0,u)  + \int^u_0 G_1(u',u) + G_2(u',u) + \frac{\p(\lambda,\tp)}{\p(u,v)}(u',u) \D u'.
\end{aligned}
\end{equation}
Integrating \eqref{Area124} with variable $u \in (0,v_0)$, we deduce
\begin{equation}\label{Area125}
\begin{aligned}
TV_{(0,v_0)}[\tmfa|_{\Gamma}] &\leq  P \cdot TV_{(0,v_0)}[\lambda|_{\Gamma}] + D  + A[\lambda,\tp] + \int_0^{v_0}\int^v_0 |G_1|(u',v)  + |G_2|(u',v) \D u' \D v \\
&= P \cdot TV_{(0,v_0)}[\lambda|_{\Gamma}] + D  + A[\lambda,\tp] + \int_0^{v_0}\int^{v_0}_{u'} |G_1|(u',v)  + |G_2|(u',v) \D v \D u'.
\end{aligned}
\end{equation}
For the integral with $|G_1|$, analogous to \eqref{Area86}, we deduce
\begin{equation}\label{Area126}
\begin{aligned}
\int^{v_0}_0 \int^{v_0}_{u'} |G_1| (u',v) \D v \D u' &\lesssim D \int^{v_0}_0 \int^{v_0}_{u'} |\dv \lambda| (u',v) \D v \D u'\lesssim D \sup_{u \in (0,v_0)}TV_{\{u\}\times(u,v_0)}[\lambda].  \\
\end{aligned}
\end{equation}
To estimate the integral with $|G_2|$, in a similar fashion as for \eqref{Area100}, we obtain
\begin{equation}\label{Area127}
\begin{aligned}\int^{v_0}_0  \int^{v_0}_{u'} |A_u \dv \tmfa|(u',v) \D v \D u' \leq \frac{1}{2\e} \sup_{u \in (0,v_0)}TV_{\{u\} \times (u,v_0) }[\tmfa].
\end{aligned}
\end{equation}
Consequently, by utilizing identical pointwise estimates \eqref{Area93}, \eqref{Area94}, \eqref{Area95}, we derive the below estimate analogous to \eqref{Area101} 
\begin{equation}\label{Area128}
\int^{v_0}_0 \int^{v_0}_{u'} |G_2| (u',v) \D v \D u' - \frac{1}{2}\sup_{u \in (0,v_0)}TV_{\{u\} \times (u,v_0) }[\tmfa] \lesssim D.
\end{equation}
Applying \eqref{Area126} and \eqref{Area128} to \eqref{Area125}, we then conclude that
\begin{equation}\label{Area129} 
TV_{(0,v_0)}[\tmfa|_{\Gamma}] 
\lesssim D\ml 1 + TV_{(0,v_0)}[\lambda|_{\Gamma}] + \sup_{u\in(0,v_0)}TV_{\{u\}\times(u,v_0)}[\lambda]\mr + A[\lambda,\tp].
\end{equation}

Observe that the upper bound in \eqref{Area129} involves the term $\sup_{v\in(0,v_0)}TV_{\{u\}\times(u,v_0)}[\lambda]$. We now move to bound it. Recall that from \eqref{Setup20} we have
\begin{equation}\label{Area130}
\du \log \lambda = \frac{\mu - Q^2/r^2}{1-\mu}\frac{\du r}{r}.
\end{equation}
Upon taking an additional derivative with respect to $v$ and using \eqref{Setup20}, \eqref{Setup27}, \eqref{Setup29}, we arrive at 
\begin{equation}\label{Area131}
\begin{aligned}
\dv \du \log \lambda =&\; \frac{\dv\ml \mu -Q^2/r^2\mr}{1-\mu}\frac{\du r}{r} + (\dv \mu) \frac{\mu - Q^2/r^2}{(1-\mu)^2}\frac{\du r}{r} + \frac{\mu - Q^2/r^2}{1-\mu}\frac{\du \dv r}{r} - \frac{\mu - Q^2/r^2}{1-\mu}\frac{(\du r)(\dv r)}{r^2} \\
=&\; \frac{\dv \mu}{1-\mu}\ml 1 + \frac{\mu-Q^2/r^2}{1-\mu}\mr\ml \frac{\du r}{r}\mr - \frac{\dv \ml \frac{Q^2}{r^2}\mr}{1-\mu}\frac{\du r}{r} + \frac{\mu - Q^2/r^2}{1-\mu}\ml \frac{\mu - Q^2/r^2}{1-\mu} - 1\mr\frac{(\du r)(\dv r)}{r^2} \\
=&\; \ml \frac{4\pi r|\dv \phi|^2}{\dv r} - \frac{\dv r}{r}\frac{\mu - Q^2/r^2}{1-\mu}\mr\ml \frac{1-Q^2/r^2}{1-\mu}\mr\ml \frac{\du r}{r}\mr + \ml \frac{8\pi \e Q \im(\phi^{\dagger}\dv \phi) + \frac{2Q^2 \dv r}{r^3}  }{1-\mu}\mr \ml \frac{\du r}{r}\mr \\
&+ \frac{\mu - Q^2/r^2}{1-\mu}\ml \frac{\mu - Q^2/r^2}{1-\mu} - 1\mr\frac{(\du r)(\dv r)}{r^2} \\
=&\; \ml \frac{4\pi r|\dv \phi|^2}{\dv r} \mr \ml \frac{1-Q^2/r^2}{1-\mu}\mr\ml \frac{\du r}{r}\mr + \ml \frac{8\pi \e Q \im(\phi^{\dagger}\dv \phi) + \frac{2Q^2 \dv r}{r^3}  }{1-\mu}\mr \ml \frac{\du r}{r}\mr \\
&-2 \ml \frac{\mu-Q^2/r^2}{1-\mu}\mr \ml \frac{\du r \dv r}{r^2}\mr.
\end{aligned}
\end{equation}
We can also rewrite \eqref{Setup20} as
\begin{equation}\label{Area132}
\dv \log (-\nu) = \frac{\mu - Q^2/r^2}{1-\mu}\frac{\dv r}{r}.
\end{equation}
By an analogous computation, with the help of \eqref{Setup26} and \eqref{Setup29.1}, we obtain
\begin{equation}\label{Area133}
\begin{aligned}
\du \dv \log (-\nu) =&\; \ml \frac{4\pi r|D_u \phi|^2}{\du r} \mr \ml \frac{1-Q^2/r^2}{1-\mu}\mr\ml \frac{\dv r}{r}\mr + \ml \frac{-8\pi \e Q \im(\phi^{\dagger}D_u \phi) + \frac{2Q^2 \du r}{r^3}  }{1-\mu}\mr \ml \frac{\dv r}{r}\mr \\
&-2 \ml \frac{\mu-Q^2/r^2}{1-\mu}\mr \ml \frac{\du r \dv r}{r^2}\mr.
\end{aligned}
\end{equation}
Observe that \eqref{Area131} and \eqref{Area133} together imply
\begin{equation}\label{Area134}
\begin{aligned}
\dv \du \log \lambda - \du \dv \log (-\nu) =&\; \frac{4\pi}{r^2}\ml \frac{1-Q^2/r^2}{1-\mu}\mr\ml \frac{(r|\dv \phi|)^2(\du r)}{(\dv r)} - \frac{(r|D_u \phi|)^2(\dv r)}{\du r}\mr \\
&+ \ml \frac{8\pi \e Q}{1-\mu} \mr \ml \frac{1}{1-\mu}\mr\ml (\du r)  \im(\phi^{\dagger}\dv \phi ) + (\dv r)\im(\phi^{\dagger}D_u \phi)\mr.
\end{aligned}
\end{equation}
We then exploit the symmetry in \eqref{Area134}. Notice that $\dv \du \log \lambda$ and $\du \dv \log (-\nu)$ have alternative expressions given below
\begin{equation}\label{Area135}
\begin{aligned}
\dv \du \log \lambda =&\; \du\ml \frac{\mu - Q^2/r^2}{1-\mu}\frac{\dv r}{r}\mr + \frac{4\pi}{r^2}\ml \frac{1-Q^2/r^2}{1-\mu}\mr\ml \frac{(r|\dv \phi|)^2(\du r)}{(\dv r)} - \frac{(r|D_u \phi|)^2(\dv r)}{\du r}\mr\\
&+ \ml \frac{8\pi \e Q}{1-\mu} \mr \ml \frac{1}{1-\mu}\mr\ml (\du r)  \im(\phi^{\dagger}\dv \phi ) + (\dv r)\im(\phi^{\dagger}D_u \phi)\mr
\end{aligned}
\end{equation}
and
\begin{equation}\label{Area136}
\begin{aligned}
\du \dv \log(-\nu) =&\; \dv\ml \frac{\mu - Q^2/r^2}{1-\mu}\frac{\du r}{r}\mr - \frac{4\pi}{r^2}\ml \frac{1-Q^2/r^2}{1-\mu}\mr\ml \frac{(r|\dv \phi|)^2(\du r)}{(\dv r)} - \frac{(r|D_u \phi|)^2(\dv r)}{\du r}\mr\\
&- \ml \frac{8\pi \e Q}{1-\mu} \mr \ml \frac{1}{1-\mu}\mr\ml (\du r)  \im(\phi^{\dagger}\dv \phi ) + (\dv r)\im(\phi^{\dagger}D_u \phi)\mr.
\end{aligned}
\end{equation}
Next, we define
\begin{equation}\label{Area137}
\tf(u,v) := 4\pi\re\ml \frac{r (\dv \tp)^\dagger}{r}\ml \frac{\tmfa}{\lambda} - \frac{\ttmfb}{\nu}\mr\mr(u,v)
\end{equation}
and
\begin{equation}\label{Area138}
\tfs(u,v) := 4\pi\re\ml \frac{r (\du \tp)^{\dagger}}{r}\ml \frac{\tmfa}{\lambda} - \frac{\ttmfb}{\nu}\mr \mr(u,v).
\end{equation}
To compute the partial derivative of $f$ with respect to $u$, we first observe from \eqref{Setup72} that
\begin{equation}\label{Area139}
\begin{aligned}
\du(r\dv \phi) = -(\dv r)(\du \phi) + E_1 \quad \text{ and } \quad \dv(r\du \phi) = -(\du r)(\dv \phi) + E_1
\end{aligned}
\end{equation}
with
\begin{equation}\label{Area140}
E_1 := -\ii\e A_u(r\dv \phi + \phi \dv r) - \ii \e \frac{Q\phi}{1-\mu}\frac{(\du r)(\dv r)}{r}
\end{equation}
and $|E_1|(u,v) \lesssim Dr(u,v)$. In addition, we observe that
\begin{equation}\label{Area141}
\frac{\tmfa}{\lambda} - \frac{\ttmfb}{\nu} = \frac{r \dv \tp}{\dv r} - \frac{r \du \tp}{\du r} + E_2 \quad \text{ with } \quad E_2 := -\ii \e \ml \frac{A_u r \phi}{\du r}\mr
\end{equation}
and $|E_2|(u,v) \lesssim Dr^2(u,v)$. We then proceed to take the partial derivative with respect to $u$ on both sides of \eqref{Area137}. Together with \eqref{Area139} and \eqref{Area141}, we obtain
\begin{equation}\label{Area143}
\begin{aligned}
\frac{\p \tf}{\p u} =&\; 4\pi \re \ml -\ml \frac{\du r \dv \tp}{r} + \frac{\dv r\du \tp }{r} - \frac{E_1}{r}\mr^{\dagger} \ml \frac{r\dv \tp}{\dv r} - \frac{r \du \tp}{\du r} + E_2\mr + \ml \frac{r \dv \tp}{r}\mr^{\dagger} \du\ml \frac{\tmfa}{\lambda} - \frac{\ttmfb}{\nu}\mr \mr\\
=&\; \frac{4\pi}{r^2}\ml \frac{\dv r}{\du r}(r|\du \tp|)^2 - \frac{\du r}{\dv r}(r|\dv \tp|)^2\mr + 4\pi\re \ml \ml \frac{r \dv \tp}{r}\mr^\dagger \du \ml \frac{\tmfa}{\lambda} - \frac{\ttmfb}{\nu}\mr \mr - 4\pi\re\ml (\du \tp)^{\dagger}(\dv \tp) - (\du \tp)(\dv \tp)^\dagger \mr \\
&+ 4\pi\re \ml \ml \frac{E_1^\dagger}{r} \mr \ml \frac{r \dv \tp}{\dv r} - \frac{r \du \tp}{\du r}\mr \mr - 4\pi \re \ml \ml \frac{E_2}{r^2} \mr \ml \du r(r\dv \tp) + \dv r (r \du \tp)\mr \mr + 4\pi\re \ml \frac{E_1^\dagger E_2}{r}\mr \\
=&\; \frac{4\pi}{r^2}\ml \frac{(r|\du \tp|)^2(\dv r)}{\du r} - \frac{ (r|\dv \tp|)^2(\du r)}{\dv r}\mr + 4\pi\re \ml \ml \frac{r \dv \tp}{r}\mr^\dagger \du \ml \frac{\tmfa}{\lambda} - \frac{\ttmfb}{\nu}\mr \mr + E_3
\end{aligned}
\end{equation}
with 
\begin{equation}\label{Area144}
E_3 := 4\pi\re \ml \ml \frac{E_1^\dagger}{r} \mr \ml \frac{r \dv \tp}{\dv r} - \frac{r \du \tp}{\du r}\mr \mr - 4\pi \re \ml \ml \frac{E_2}{r^2} \mr \ml \du r(r\dv \tp) + \dv r (r \du \tp)\mr \mr + 4\pi \re \ml \frac{E_1^\dagger E_2}{r}\mr.
\end{equation}
Note that we use the fact that $( \du \tp)^{\dagger}(\dv \tp) - (\du \tp)(\dv \tp)^\dagger$ is purely imaginary. And we have $|E_3|(u,v) \lesssim D^2$. Similarly, we compute the partial derivative of $\tfs$ with respect to $v$ and get
\begin{equation}\label{Area145}
\begin{aligned}
\frac{\p \tfs}{\p v} =&\; 4\pi\re \ml -\ml \frac{\du r \dv \tp}{r} + \frac{\dv r\du \tp }{r} - \frac{E_1}{r}\mr^{\dagger} \ml \frac{r\dv \tp}{\dv r} - \frac{r \du \tp}{\du r} + E_2\mr + \ml \frac{r \du \tp}{r}\mr^{\dagger} \dv\ml \frac{\tmfa}{\lambda} - \frac{\ttmfb}{\nu}\mr \mr\\
=&\; \frac{4\pi}{r^2}\ml \frac{(r|\du \tp|)^2(\dv r)}{\du r} - \frac{ (r|\dv \tp|)^2(\du r)}{\dv r}\mr + 4\pi\re \ml \ml \frac{r \du \tp}{r}\mr^\dagger \dv \ml \frac{\tmfa}{\lambda} - \frac{\ttmfb}{\nu}\mr \mr + E_3.
\end{aligned}
\end{equation}

\vspace{5mm}

The equations for $\frac{\p \tf}{\p u}$ and $\frac{\p \tfs}{\p v}$ will be used in the estimates for $\du \ml \frac{\tmfa}{\lambda} - \frac{\ttmfb}{\nu}\mr$ and $\dv \ml \frac{\tmfa}{\lambda} - \frac{\ttmfb}{\nu}\mr$. To evaluate them, we start by considering \eqref{Area139} and rewrite it as
\begin{equation}\label{Area146}
\begin{aligned}
\du \dv (r\tp) = \dv \du (r\tp) = \tp \du \dv r + E_1.
\end{aligned}
\end{equation}
By the definition of $\tmfa$ from \eqref{tildemathfraka} and employing \eqref{Setup20}, the above equation implies that
\begin{equation}\label{Area147}
\begin{aligned}
\du\ml \frac{\tmfa}{\lambda} \mr = \frac{1}{\dv r}\ml E_1 - \frac{r \dv \tp}{\dv r}\du \dv r\mr = - \frac{\mu - Q^2/r^2}{1-\mu}\ml \frac{r \dv \tp}{\dv r}\mr\ml \frac{\du r}{r}\mr + \frac{E_1}{\dv r}.
\end{aligned}
\end{equation}
Similarly, we also obtain
\begin{equation}\label{Area148}
\begin{aligned}
\dv\ml \frac{\tmfb}{\nu} \mr = - \frac{\mu - Q^2/r^2}{1-\mu}\ml \frac{r \du \tp}{\du r} \mr \ml \frac{\dv r}{r} \mr + \frac{E_1}{\du r}.
\end{aligned}
\end{equation}
By using the definition of $\ttmfb$ in \eqref{tildetildemathfrakb}, the above further simplifies to
\begin{equation}\label{Area149}
\begin{aligned}
\dv\ml \frac{\ttmfb}{\nu} \mr = - \frac{\mu - Q^2/r^2}{1-\mu}\ml \frac{r \du \tp}{\du r} \mr \ml \frac{\dv r}{r} \mr + \frac{E_4}{\du r}
\end{aligned}
\end{equation}
with
\begin{equation}\label{Area150}
E_4 := E_1 + (\du r)\dv \ml \frac{\ii \e A_u r \phi}{\du r}\mr.
\end{equation}
Observe that by applying \eqref{Setup20} and \eqref{Setup30}, we get
\begin{equation}\label{Area151}
\begin{aligned}
\dv \ml \frac{A_u r \phi}{\du r}\mr &= - A_u r \phi \ml\frac{\du \dv r}{(\du r)^2}\mr + \frac{1}{\du r}\ml (\dv A_u)(r\phi) + A_u \phi (\dv r) + A_u r(\dv \tp)\mr \\
&= - A_u \phi \ml\frac{(\mu - Q^2/r^2)(\dv r)}{(1-\mu)(\du r)}\mr + \frac{1}{\du r}\ml \ml \frac{2Q(\du r)(\dv r)}{r^2(1-\mu)}\mr(r\phi) + A_u \phi (\dv r) + A_u r(\dv \tp)\mr.
\end{aligned}
\end{equation}
Using the fact that $|E_1|(u,v) \lesssim Dr(u,v)$, we derive the below estimate
\begin{equation}\label{Area152}
\begin{aligned}
|E_4|(u,v) &\leq \mlm E_1\mrm(u,v) + \mlm (\du r)\dv \ml \frac{\ii \e A_u r \phi}{\du r}\mr\mrm \\
&\lesssim D r(u,v) + Dr(u,v) + Dr(u,v) + Dr(u,v) + Dr^2(u,v) \lesssim Dr(u,v).
\end{aligned}
\end{equation}
At the same time, by the expressions in \eqref{Area147} and \eqref{Area149}, we compute the following terms
\begin{equation}\label{Area153}
\begin{aligned}
(\dv \tp)^\dagger \du\ml \frac{\tmfa}{\lambda} - \frac{\ttmfb}{\nu} \mr =&\; (\dv \tp)^\dagger \du \ml \frac{\tmfa}{\lambda}\mr - (\du \tp)^\dagger \dv \ml \frac{\ttmfb}{\nu}\mr + \ml (\du \tp)^\dagger \dv \ml \frac{\ttmfb}{\nu}\mr - (\dv \tp)^\dagger \du \ml \frac{\ttmfb}{\nu}\mr \mr \\
=&\; \frac{1}{r^2}\ml \frac{\mu - Q^2/r^2}{1-\mu} \mr\ml \frac{(r|\du \tp|)^2(\dv r)}{\du r} - \frac{ (r|\dv \tp|)^2(\du r)}{\dv r}\mr + \frac{\p(\tp^\dagger,\ttmfb/\nu)}{\p(u,v)}\\
&+ \ml \frac{(\dv \tp)^\dagger}{\dv r} \mr E_1 - \ml \frac{(\du \tp)^\dagger}{\du r} \mr E_4,
\end{aligned}
\end{equation}
and
\begin{equation}\label{Area154}
\begin{aligned}
(\du \tp)^\dagger \dv\ml \frac{\tmfa}{\lambda} - \frac{\ttmfb}{\nu} \mr =&\; (\dv \tp)^\dagger \du \ml \frac{\tmfa}{\lambda}\mr -(\du \tp)^\dagger \dv \ml \frac{\ttmfb}{\nu} \mr + \ml (\du \tp)^\dagger \dv \ml \frac{\tmfa}{\lambda}\mr - (\dv \tp)^\dagger \du \ml \frac{\tmfa}{\lambda}\mr\mr \\
=&\; \frac{1}{r^2}\ml \frac{\mu - Q^2/r^2}{1-\mu} \mr\ml \frac{(r|\du \tp|)^2(\dv r)}{\du r} - \frac{ (r|\dv \tp|)^2(\du r)}{\dv r}\mr + \frac{\p(\tp^\dagger,\tmfa/\lambda)}{\p(u,v)}\\
&+ \ml \frac{(\dv \tp)^\dagger}{\dv r} \mr E_1 - \ml \frac{(\du \tp)^\dagger}{\du r} \mr E_4.
\end{aligned}
\end{equation}
By substituting \eqref{Area153} into the expression for $\du \tf$ in \eqref{Area143}, we obtain
\begin{equation}\label{Area155}
\begin{aligned}
\frac{\p \tf}{\p u} =&\; \frac{4\pi}{r^2}\ml 1 + \frac{\mu - Q^2/r^2}{1-\mu}\mr\ml \frac{(r|\du \tp|)^2(\dv r)}{\du r} - \frac{ (r|\dv \tp|)^2(\du r)}{\dv r}\mr + 4\pi \re \ml \frac{\p(\tp^\dagger,\ttmfb/\nu)}{\p(u,v)} \mr\\
&+ \re \ml 4\pi\ml \frac{(\dv \tp)^\dagger}{\dv r} \mr E_1 - 4\pi\ml \frac{(\du \tp)^\dagger}{\du r} \mr E_4\mr + E_3 \\
=&\; \frac{4\pi}{r^2}\ml \frac{1 - Q^2/r^2}{1-\mu}\mr\ml \frac{(r|\du \tp|)^2(\dv r)}{\du r} - \frac{ (r|\dv \tp|)^2(\du r)}{\dv r}\mr + 4\pi \re \ml \frac{\p(\tp^\dagger,\ttmfb/\nu)}{\p(u,v)} \mr + E_5
\end{aligned}
\end{equation}
in which
\begin{equation}\label{Area156}
E_5 := \ml 4\pi\ml \frac{(\dv \tp)^\dagger}{\dv r} \mr E_1 - 4\pi\ml \frac{(\du \tp)^\dagger}{\du r} \mr E_4 \mr + E_3.
\end{equation}
Note that since $|E_1| \lesssim Dr$, $|E_3| \lesssim D^2$, $|E_4| \lesssim Dr$, $|r\dv \tp| \lesssim D$, $|r\du \tp| \lesssim D$, we can deduce that $|E_5| \lesssim D^2$.\\

By a parallel argument, for $\dv \tfs$, we can utilize \eqref{Area154} to \eqref{Area145} to deduce that
\begin{equation}\label{Area157}
\begin{aligned}
\frac{\p \tfs}{\p v} 
=&\; \frac{4\pi}{r^2}\ml \frac{1 - Q^2/r^2}{1-\mu}\mr\ml \frac{(r|\du \tp|)^2(\dv r)}{\du r} - \frac{ (r|\dv \tp|)^2(\du r)}{\dv r}\mr + 4\pi \re \ml \frac{\p(\tp^\dagger,\tmfa/\lambda)}{\p(u,v)} \mr + E_5.
\end{aligned}
\end{equation}
Comparing terms in the expressions for $\frac{\p \tf}{\p u}$ and for $\frac{\p \tfs}{\p v}$ above with that for $\dv \du \log \lambda$ and for $\du \dv \log \ml - \nu \mr$ in \eqref{Area135} and \eqref{Area136}, it is instructional to consider following decomposition of the term $|D_u \phi|^2$:
\begin{equation}\label{Area158}
\begin{aligned}
|D_u \phi|^2 &= (\du \phi + \ii \e A_u \phi)(\du \phi - \ii \e A_u \phi)^\dagger = |\du \tp|^2 - 2 \e A_u \im(\phi \du \phi^\dagger) + \e^2 A_u^2 |\phi|^2.
\end{aligned}
\end{equation}
This implies that \eqref{Area135} and \eqref{Area136} can be simplified as
\begin{equation}\label{Area159}
\begin{aligned}
\dv \du \log \lambda =&\; \du\ml \frac{\mu - Q^2/r^2}{1-\mu}\frac{\dv r}{r}\mr + E_6 \\
&- \frac{4\pi}{r^2}\ml \frac{1-Q^2/r^2}{1-\mu}\mr\ml \frac{(r|\du \tp|)^2(\dv r)}{\du r} - \frac{ (r|\dv \tp|)^2(\du r)}{\dv r}\mr\\
\end{aligned}
\end{equation}
and
\begin{equation}\label{Area160}
\begin{aligned}
\du \dv \log(-\nu) =&\; \dv\ml \frac{\mu - Q^2/r^2}{1-\mu}\frac{\du r}{r}\mr + E_6 \\
&+ \frac{4\pi}{r^2}\ml \frac{1-Q^2/r^2}{1-\mu}\mr\ml \frac{(r|\du \tp|)^2(\dv r)}{\du r} - \frac{ (r|\dv \tp|)^2(\du r)}{\dv r}\mr
\end{aligned}
\end{equation}
with
\begin{equation}\label{Area161}
\begin{aligned}
E_6 :=&\; \ml \frac{8\pi \e Q}{1-\mu} \mr \ml \frac{1}{1-\mu}\mr\ml (\du r)  \im(\phi^{\dagger}\dv \phi ) + (\dv r)\im(\phi^{\dagger}D_u \phi)\mr \\
&+ 4\pi \ml \frac{1-Q^2/r^2}{1-\mu}\mr\ml \frac{\dv r}{\du r}\mr\ml 2\e A_u \im(\phi \du \phi^\dagger) - \e^2 A_u^2|\phi|^2 \mr
\end{aligned}
\end{equation}
satisfying $|E_6|(u,v) \lesssim D^2$. Substituting the expressions of $\frac{\p \tf}{\p u}$ and $\frac{\p \tfs}{\p v}$ in \eqref{Area155} and in \eqref{Area157} to \eqref{Area159} and to \eqref{Area160}, we obtain
\begin{equation}\label{Area162}
\begin{aligned}
\du\ml \dv \log \lambda \mr = \du \ml \frac{\mu - Q^2/r^2}{1-\mu}\frac{\dv r}{r} - \tf\mr + 4\pi \re \ml \frac{\p(\tp^\dagger,\ttmfb/\nu)}{\p(u,v)}\mr + E_5 + E_6
\end{aligned}
\end{equation}
and
\begin{equation}\label{Area163}
\begin{aligned}
\dv\ml \du \log (-\nu) \mr = \dv \ml \frac{\mu - Q^2/r^2}{1-\mu}\frac{\du r}{r} + \tfs\mr - 4\pi \re \ml \frac{\p(\tp^\dagger,\tmfa/\lambda)}{\p(u,v)}\mr - E_5 + E_6.
\end{aligned}
\end{equation}

\vspace{5mm}

With the above preparations, we start to estimate the total variations of $\log(\lambda), \log(-\nu)$ and $\log(\lambda|_{\Gamma})$.  Employing \eqref{Area162}, for a fixed $(u,v) \in \mathcal{D}(0,v_0)$, we first integrate along a future-directed incoming null curve. Together with the fact that $\dv \lambda (0,v) = 0$, we obtain
\begin{equation}\label{Area164}
\begin{aligned}
\dv \log \lambda (u,v) =&\; \ml \frac{\mu - Q^2/r^2}{1-\mu}\frac{\dv r}{r} - \tf \mr(u,v) - \ml \frac{\mu - Q^2/r^2}{1-\mu}\frac{\dv r}{r} - \tf \mr(0,v) \\
&+ 4\pi \int^u_0 \re \ml \frac{\p(\tp^\dagger,\ttmfb/\nu)}{\p(u,v)} \mr (u',v) \D u' + \int^u_0 (E_5 + E_6)(u',v) \D u'.
\end{aligned}
\end{equation}
Further integrating \eqref{Area164} for $v \in (u,v_0)$, we derive
\begin{equation}\label{Area165}
\begin{aligned}
& \int^{v_0}_u \mlm \dv \log \lambda - \frac{\mu - Q^2/r^2}{1-\mu}\frac{\dv r}{r} + \tf\mrm(u,v') \D v'\\
\lesssim &\int^{v_0}_u \mlm \frac{\mu - Q^2/r^2}{1-\mu}\frac{\dv r}{r} - \tf\mrm(0,v') \D v' + A[\tp^\dagger,\ttmfb/\nu] + D^2
\end{aligned}
\end{equation}
To obtain the total variation of $\lambda$, it suffices to estimate the integral in the upper bound of \eqref{Area165}. To estimate the term $\frac{\mu - Q^2/r^2}{1-\mu}\frac{\dv r}{r}$, we appeal to \eqref{Setup22} and see that
\begin{equation}\label{Area169}
\dv \log \ml \frac{-\du r}{1-\mu} \mr = \frac{4 \pi r|\dv \phi|^2}{\dv r}. 
\end{equation}
By \eqref{Area132}, we can rewrite this as
\begin{equation}\label{Area170}
\frac{\mu - Q^2/r^2}{1-\mu}\frac{\dv r}{r} - \dv \log \ml 1-\mu \mr = \frac{4 \pi r|\dv \phi|^2}{\dv r}.
\end{equation}
Integrating both sides of \eqref{Area170} with respect to $v'$, we obtain
\begin{equation}\label{Area171}
\begin{aligned}
\int^{v_0}_u \frac{\mu - Q^2/r^2}{1-\mu}\frac{(\dv r)}{r}(u,v') \D v' &= \log(1-\mu(u,v_0)) + 4\pi \int^{v_0}_u \frac{r|\dv \phi|^2}{\dv r}(u,v') \D v' \\
&\lesssim D \int^{v_0}_u |\dv \phi|(u,v') \D v' \lesssim DX
\end{aligned}
\end{equation}
and
\begin{equation}\label{Area172}
\begin{aligned}
\int^{v_0}_u \frac{\mu - Q^2/r^2}{1-\mu}\frac{(\dv r)}{r}(0,v') \D v' &= \log(1-\mu(u,v_0)) + 4\pi \int^{v_0}_u \frac{r|\dv \phi|^2}{\dv r}(0,v') \D v' \\
&\lesssim D \int^{v_0}_u |\dv \phi|(0,v') \D v' \lesssim D \cdot TV_{\{0\} \times (0,v_0)}[\tilde{\alpha}] \\
&\quad\quad\lesssim D \cdot \ml 2 \cdot TV_{\{0\} \times (0,v_0)}[\tmfa]\mr \lesssim D^2. \\
\end{aligned}
\end{equation}
In the above computations, we have used $\mu|_{\Gamma} = 0$, $\mu < 1$, $\frac{r|\dv \phi|}{\dv r} \lesssim D$, $\lambda(0,v) = \frac{1}{2}$ and Lemma \ref{AreaLemma1}. Furthermore, from the definition of $\tf$ in \eqref{Area137}, we have
\begin{equation}\label{Area173}
|\tf|(u,v) \lesssim D |\dv \phi|(u,v)
\end{equation}
and hence
\begin{equation}\label{Area174}
\begin{aligned}
\int^{v_0}_u |\tf|(u,v') \D v' &\lesssim D \int^{v_0}_u |\dv \phi|(u,v') \D v' \lesssim DX, \\
\int^{v_0}_u |\tf|(0,v') \D v' &\lesssim D \int^{v_0}_u |\dv \phi|(0,v') \D v' \lesssim D^2. \\
\end{aligned}
\end{equation}
Gathering the estimates in \eqref{Area171}, \eqref{Area172}, \eqref{Area174}, by \eqref{Area165}, we conclude that
\begin{equation}\label{Area175}
\sup_{u \in (0,v_0)} TV_{ \{u\} \times (u,v_0) } [\log \lambda] \lesssim D^2 + DX + A[\tp^\dagger, \ttmfb/\nu].
\end{equation}

{\color{black}Next, we move on to estimate the total variation of $\log(-\nu)$. Observe that from \eqref{Area311}, \eqref{Area521}, and \eqref{Area138} that we have
\begin{equation}\label{Area340}
|\tilde{f}^*|(u,v) \lesssim r(u,v).
\end{equation}
Together with \eqref{Area314}, this implies
\begin{equation}\label{Area341}
\tfs|_{\Gamma} = \frac{\mu}{r}|_{\Gamma} = 0.
\end{equation}}
Using a symmetric argument as for \eqref{Area163}, by \eqref{Area120}, we get
\begin{equation}\label{Area176}
\begin{aligned}
\du \log(-\nu)(u,v) =&\; \frac{\D }{\D u}\log(-\nu)(u) + \ml \frac{\mu - Q^2/r^2}{1-\mu}\frac{\du r}{r} + 4 \pi \tfs\mr (u,v)  \\
&- 4\pi \int^v_u \re \ml \frac{\p(\tp^\dagger,\tmfa/\lambda)}{\p(u,v)} \mr (u,v') \D v' + \int^v_u (E_6 - E_5)(u,v') \D v'
\end{aligned}
\end{equation}
and hence
\begin{equation}\label{Area177}
\begin{aligned}
&\int^v_0 \mlm \du \log(-\nu) - \frac{\mu - Q^2/r^2}{1-\mu}\frac{\du r}{r} - 4\pi \tfs \mrm(u',v) \D u' \\
\lesssim & \;  TV_{(0,v_0)}[\log(-\nu)|_{\Gamma}] +  A[\tp^\dagger,\tmfa/\lambda] + D^2.
\end{aligned}
\end{equation}
Similar to the arguments made above, we start from \eqref{Setup21} and rewrite \eqref{Area130} as 
\begin{equation}\label{Area167}
\begin{aligned}
\frac{\mu - Q^2/r^2}{1-\mu}\frac{\du r}{r} - \du \log(1-\mu) &= \frac{4\pi r|D_u \phi|^2}{\du r}.
\end{aligned}
\end{equation}
Integrating with respect to $u'$, recalling the definition of $Y$ in \eqref{Area15} and applying $|D_u \phi| = |\du \phi + \ii \e A_u r \phi| \lesssim |\du \phi| + D$, we deduce that 
\begin{equation}\label{Area168}
\begin{aligned}
\int^v_0 \frac{\mu - Q^2/r^2}{1-\mu}\frac{(-\du r)}{r}(u',v) \D u' &\lesssim D \int^v_u |D_u \phi|(u',v) \D u' \lesssim D \int^v_u |\du \tp|(u',v) \D u' + D^2 \lesssim DY + D^2.
\end{aligned}
\end{equation}
In addition, from the expression of $\tfs$ in \eqref{Area138}, we see that 
\begin{equation}\label{Area178}
|\tfs|(u,v)\lesssim D|\du \phi|(u,v)
\end{equation}
and thus obtain
\begin{equation}\label{Area179}
\begin{aligned}
\int^v_0 |\tfs|(u',v) \D u' &\lesssim D \cdot T.V_{(0,v) \times \{v\}}[\tp]  \lesssim D Y.
\end{aligned}
\end{equation}
By applying \eqref{Area116}, \eqref{Area168}, \eqref{Area179} to \eqref{Area177}, we hence derive
\begin{equation}\label{Area180}
\sup_{v\in(0,v_0)} TV_{(0,v) \times \{v\}}[\log |\nu|] \lesssim D^2 + DY + TV_{(0,v_0)}[\log\lambda|_{\Gamma}] + A[\tp^\dagger,\tmfa/\lambda].
\end{equation}

To estimate $TV_{(0,v_0)}[\log \lambda|_{\Gamma}]$, we appeal to \eqref{Area123}, \eqref{Area162}, \eqref{Area341}. Integrating along future-directed incoming null curve from $C_0^+$ to $(u,u) \in \Gamma$ yields
\begin{equation}\label{Area181}
\begin{aligned}
\frac{\D}{\D u}\log \ml \lambda|_{\Gamma}\mr (u) =&\; - \ml \frac{\mu - Q^2/r^2}{1-\mu}\frac{\dv r}{r} - 4\pi \tf\mr(0,u) \\
&+ 4\pi \int^u_0 \re \ml \frac{\p(\tp^\dagger,\ttmfb/\nu)}{\p(u,v)}\mr(u',u) \D u' + \int^u_0 (E_5 + E_6) (u',u) \D u'.
\end{aligned}
\end{equation}
By utilizing estimates analogous to \eqref{Area172} and \eqref{Area174}, we deduce that 
\begin{equation}\label{Area182}
TV_{(0,v_0)}[\log(\lambda|_{\Gamma})] \lesssim D^2 + A[\tp^\dagger,\ttmfb/\nu].
\end{equation}
\hspace{5pt}

As the next course of action, we proceed to establish relations between the different area terms that have appeared in the estimates above. Starting from $A[\tp^\dagger,\tmfa/\lambda]$, by \eqref{Area42}, we observe that
\begin{equation}\label{Area187}
\du \tmfa = \tp \du \dv r + E_7 \quad \text{ with } \quad E_7 := -\ii \e \ml A_u \dv(r\tp) + \phi_0 A_u \dv r + \frac{\du r \dv r }{1-\mu}\frac{Q \phi}{r}\mr
\end{equation}
and $|E_7|(u,v) \lesssim Dr(u,v)$. This then implies that
\begin{equation}\label{Area183}
\begin{aligned}
\frac{\p(\tp^\dagger,\tmfa/\lambda)}{\p(u,v)} &= \du \tp^\dagger \dv \ml \frac{\tmfa}{\lambda}\mr - \dv \tp^\dagger \du \ml \frac{\tmfa}{\lambda}\mr \\
&= \frac{(\du \tp)^\dagger(\dv \tmfa) - (\dv \tp)^\dagger(\tp \du \dv r + E_7 )}{\lambda} + \frac{(r \dv \tp) + \tp \lambda}{\lambda^2}\frac{\p(\lambda, \tp^\dagger)}{\p(u,v)} \\
&= \frac{(\du \tp)^\dagger}{\lambda}(\dv \tmfa - \tp \dv \lambda) + \frac{(r \dv \tp)}{\lambda^2}\frac{\p(\lambda, \tp^\dagger)}{\p(u,v)}- \frac{(\dv \tp)^\dagger E_7}{\lambda}. \\
\end{aligned}
\end{equation}
Furthermore, applying \eqref{Area79}, \eqref{Area86}, \eqref{Area101}, \eqref{Area103}, we get
\begin{equation}\label{Area184}
\mlm \dv \tmfa - \tp \dv \lambda\mrm(u,v) \leq F(v)
\end{equation}
with
\begin{equation}\label{Area185}
F(v) =  \mlm \dv \tmfa - \tp \dv \lambda \mrm(0,v) + \int^u_0 \mlm G_1 + G_2 + \frac{\p(\lambda,\tp)}{\p(u,v)}\mrm(u',v) \D u'
\end{equation}
and 
\begin{equation}\label{Area186}
\int^{v_0}_0 F(v) \D v \lesssim D\ml 1 + \sup_{u \in (0,v_0)}TV_{\{u\}\times(u,v_0)}[\lambda]\mr  + A[\lambda,\tp^\dagger].
\end{equation}
Note that in \eqref{Area186}, we have used $A[\lambda,\tp^\dagger] = A[\lambda,\tp]$. This follows from the fact that, for any complex number $z$, we have $|z| = |z^\dagger|$. Back to \eqref{Area184}, we then derive 
\begin{equation}\label{Area189}
\begin{aligned}
\int_{D(0,v_0)} \frac{|\du \tp|}{\lambda}|\dv \tmfa - \tp \dv \lambda|(u,v) \D u \D v &\lesssim \ml \sup_{v \in (0,v_0)}\int^v_0 |\du \tp|(u,v) \D u \mr \ml \int^{v_0}_0 F(v) \D v\mr \\
&\lesssim DY\ml 1 + \sup_{u \in (0,v_0)}TV_{\{u\}\times(u,v_0)}[\lambda]\mr  + Y\cdot A[\lambda,\tp^\dagger].
\end{aligned}
\end{equation}
Therefore, equation \eqref{Area183} allows us to deduce that
\begin{equation}\label{Area190}
A[\tp^\dagger, \tmfa/\lambda] \lesssim DY\ml 1 + \sup_{u \in (0,v_0)}TV_{\{u\}\times(u,v_0)}[\lambda]\mr + D^2 + (D+Y)\cdot A[\lambda,\tp^\dagger].
\end{equation}
We conduct a similar strategy to control $A[\tp^\dagger, \ttmfb/\nu]$ as follows. Observe that from \eqref{Area104}, we have
\begin{equation}\label{Area191}
\begin{aligned}
\dv \ttmfb &= \tp \du \dv r + E_8 \quad \text{ with } \quad 
E_8 := \ii \e \frac{\du r \dv r}{1-\mu}\frac{Q \phi}{r}
\end{aligned}
\end{equation}
and $|E_8|(u,v) \lesssim Dr(u,v)$. Together with \eqref{tildetildemathfrakb}, we compute $\frac{\p(\tp,\ttmfb/\nu)}{\p(u,v)}$ and obtain
\begin{equation}\label{Area193}
\begin{aligned}
\frac{\p(\tp,\ttmfb/\nu)}{\p(u,v)} &= (\du \tp)^\dagger \dv \ml\frac{\ttmfb}{\nu} \mr - (\dv \tp)^\dagger \du \ml \frac{\ttmfb}{\nu}\mr \\
&= \frac{(\du \tp)^\dagger (\tp \du \dv r + E_8) - (\dv \tp)^\dagger(\du \ttmfb) }{\nu} + \frac{(r \du \tp) + \tp \nu + \ii \e A_u r \phi}{\nu^2}\frac{\p(\nu,\tp^\dagger)}{\p(u,v)} \\
&= -\frac{(\dv \tp)^\dagger}{\nu}\ml \du \ttmfb - \tp \du \nu \mr + \frac{(r \du \tp) + \ii \e A_u r \phi}{\nu^2}\frac{\p(\nu,\tp^\dagger)}{\p(u,v)} + \frac{(\du \tp)^\dagger E_8}{\nu}.
\end{aligned}
\end{equation}
Employing \eqref{Area107}, \eqref{Area108}, \eqref{Area110}, \eqref{Area113}, we see that
\begin{equation}\label{Area194}
\mlm \du \ttmfb - \tp \du \nu \mrm(u,v) \leq F^*(u)
\end{equation}
with
\begin{equation}\label{Area195}
F^*(u) =  \mlm \du \ttmfb - \tp \du \nu\mrm_{\Gamma}(u) + \int^v_u \mlm G_3 + G_4 - \frac{\p(\nu,\tp)}{\p(u,v)} \mrm (u,v') \D v'
\end{equation}
and
\begin{equation}\label{Area196}
\int_0^{v_0} F^*(u) \D u \lesssim TV_{(0,v_0)}[\ttmfb|_{\Gamma}] + D\ml 1 +  TV_{(0,v_0)}[\lambda|_{\Gamma}] + \sup_{v \in (0,v_0)} TV_{(0,v) \times \{v\} }[\nu] \mr + A[\nu,\tp^\dagger].
\end{equation}
Note that in the expression of \eqref{Area195}, the definitions of $G_3$ and $G_4$ were given \eqref{Area106}. Consequently, we deduce that 
\begin{equation}\label{Area197}
\begin{aligned}
&\int_{D(0,v_0)} \frac{| \dv \tp|}{|\nu|}|\du \ttmfb - \tp \du \nu|(u,v) \D u \D v \lesssim \ml \sup_{u \in (0,v_0)}\int^{v_0}_u |\dv \tp|(u,v) \D v\mr \ml \int^{v_0}_0 F^*(u) \D u\mr \\
\lesssim&\; X \cdot TV_{(0,v_0)}[\ttmfb|_{\Gamma}] + DX\ml 1 +  TV_{(0,v_0)}[\lambda|_{\Gamma}] + \sup_{v \in (0,v_0)} TV_{(0,v) \times \{v\} }[\nu] \mr + X\cdot A[\nu,\tp^\dagger].
\end{aligned}
\end{equation}
Combining \eqref{Area193} and  \eqref{Area197}, we now arrive at the following estimate
\begin{equation}\label{Area198}
\begin{aligned}
A[\tp^\dagger, \ttmfb/\nu] \lesssim&\; X \cdot TV_{(0,v_0)}[\ttmfb|_{\Gamma}] + DX\ml 1 +  TV_{(0,v_0)}[\lambda|_{\Gamma}] + \sup_{v \in (0,v_0)} TV_{(0,v) \times \{v\} }[\nu] \mr \\
&+ D^2 + (D + X)\cdot A[\nu,\tp^\dagger].
\end{aligned}
\end{equation}

In view of \eqref{Area190} and \eqref{Area198}, we proceed to estimate $A[\lambda,\tp^\dagger]$ and $A[\nu,\tp^\dagger]$. For convenience, we denote
\begin{equation}\label{ttzeta}
\tilde{\theta} := r \dv \tp = r \dv \phi, \quad \quad 
\tilde{\zeta} := r \du \tp = r \du \phi, \quad \quad
\ttz := r D_u \phi.
\end{equation}
It is worth noting that, contrary to \eqref{ttzeta}, here $\ttz \neq r D_u \tp$ but rather $\ttz = \tz + \ii \e A_u r \phi.$

Inspired by the definition of $\tf$ in \eqref{Area137}, we also define
\begin{equation}\label{Area199}
\tg := 4 \pi \im \ml \frac{\ttheta^\dagger}{r}\ml\frac{\tmfa}{\lambda} - \frac{\ttmfb}{\nu} \mr \mr.
\end{equation}
By \eqref{Area141}, we have
\begin{equation}\label{Area200}
\tf + \ii \tg = 4\pi \ml \frac{\ttheta^\dagger}{r} \mr\ml  \frac{\ttheta}{\lambda} - \frac{\tz}{\nu} + E_2\mr
\end{equation}
with $E_2$ given in \eqref{Area141}. This further implies that
\begin{equation}\label{Area201}
-\frac{4\pi}{r^2}\frac{\tz \ttheta^\dagger}{\lambda \nu}(\nu \ttheta - \lambda \tz) = -\frac{\tf \tz}{r} - \ii \frac{\tz \tg}{r} + \frac{4\pi \tz \ttheta^\dagger}{r^2}E_2.
\end{equation}
Together with \eqref{Area130}, we now compute $\frac{1}{\lambda}\frac{\p(\lambda,\tp)}{\p(u,v)}$ as follows
\begin{equation}\label{Area202}
\begin{aligned}
\frac{1}{\lambda}\frac{\p(\lambda,\tp)}{\p(u,v)} =&\; \frac{\p \log \lambda}{\p u}\frac{\p \tp}{\p v} - \frac{\p \log \lambda}{\p v}\frac{\p \tp}{\p u} = \frac{\mu - Q^2/r^2}{1-\mu}\frac{\nu \ttheta}{r^2} - \dv \log \lambda \frac{\tz}{r} \\
=&\; \frac{1}{r^2}(\nu \ttheta - \lambda \tz)\ml \frac{\mu - Q^2/r^2}{1-\mu} + 4 \pi \frac{\ttheta^\dagger \tz}{\lambda \nu}\mr - \frac{\tz}{r}\ml \dv \log \lambda - \frac{\mu - Q^2/r^2}{1-\mu}\frac{\lambda}{r}\mr - \frac{4\pi}{r^2}\frac{\ttheta^\dagger \tz}{\lambda \nu}(\nu\ttheta - \lambda \tz) \\
=&\; \frac{1}{r^2}\ml \frac{ \nu \ttheta - \lambda \tz }{\lambda \nu}\mr\ml \frac{\mu - Q^2/r^2}{1-\mu}\lambda \nu + 4 \pi \ttheta^\dagger \ttz\mr - \frac{\tz}{r}\ml \dv \log \lambda - \frac{\mu - Q^2/r^2}{1-\mu}\frac{\lambda}{r} + \tf \mr \\
&- \ii \frac{\tz \tg}{r} + \frac{4\pi \tz \ttheta^\dagger}{r^2}E_2 - 4\pi \ii \e \frac{\nu \ttheta - \lambda \tz}{\lambda \nu}\frac{\ttheta^\dagger A_u \phi}{r}.  \\
\end{aligned}
\end{equation}
By recalling and re-writing \eqref{Area162}, we have
\begin{equation}\label{Area203}
\du \ml \dv \log \lambda - \frac{\mu - Q^2/r^2}{1-\mu}\frac{\lambda}{r} + \tf\mr =  4\pi \re \ml \frac{\p(\tp^\dagger,\ttmfb/\nu)}{\p(u,v)}\mr + E_5 + E_6,
\end{equation}
with $|E_5|(u,v) \lesssim D^2$ and $|E_6|(u,v) \lesssim D^2$. Henceforth, we obtain
\begin{equation}\label{Area204}
\mlm \dv \log \lambda - \frac{\mu - Q^2/r^2}{1-\mu}\frac{\lambda}{r} + \tf\mrm(u,v) \leq G(v)
\end{equation}
with
\begin{equation}\label{Area205}
G(v) := \mlm \dv \log \lambda - \frac{\mu - Q^2/r^2}{1-\mu}\frac{\lambda}{r} + \tf\mrm(0,v) + 4\pi \int^v_0 \mlm \frac{\p(\tp^\dagger,\ttmfb/\nu)}{\p(u,v)}\mrm (u',v) \D u' + \int_0^v |E_5| + |E_6| (u',v) \D u'.
\end{equation}
Together with \eqref{Area172} and \eqref{Area174}, we have
\begin{equation}\label{Area206}
\int_0^{v_0} G(v) \D v \lesssim D^2 + A[\tp^\dagger, \ttmfb/\nu].
\end{equation}
Next, we define
\begin{equation}\label{Area207}
\tr := r\ml \frac{\mu - Q^2/r^2}{1-\mu}\lambda \nu + 4 \pi \ttheta^\dagger \ttz \mr.
\end{equation}
Back to \eqref{Area202}, with $|E_2|(u,v) \lesssim Dr(u,v)^2$, we now deduce that
\begin{equation}\label{Area208a}
\mlm \frac{1}{\lambda}\frac{\p(\lambda,\tp)}{\p(u,v)} \mrm \lesssim D\frac{|\tr|}{r^3} + |\du \tp|G(v) + D \frac{|r^2\tg|}{r^3}  + D^3.
\end{equation}{\color{black}
Integrating on the domain $\mathcal{D}(0,v_0)$, we then have
\begin{equation}\label{Area208}
\begin{aligned}
\int_{\mathcal{D}(0,v_0)}\mlm \frac{1}{\lambda}\frac{\p(\lambda,\tp)}{\p(u,v)} \mrm(u,v) \D u \D v \lesssim&\;  D \int_{\mathcal{D}(0,v_0)}\frac{|\tr|}{r^3}(u,v) \D u \D v + \int_{\mathcal{D}(0,v_0)}|\du \tp|(u,v)G(v)\D u \D v \\
&+\; D \int_{\mathcal{D}(0,v_0)}\frac{|r^2\tg|}{r^3}(u,v) \D u \D v  + D^3.
\end{aligned}
\end{equation}}
By \eqref{Area15}, we note that
\begin{equation}\label{Area209}
\begin{aligned}
\int_{\mathcal{D}(0,v_0)} |\du \tp|(u,v) G(v) \D u \D v &\lesssim Y \int^{v_0}_0 G(v) \D v \lesssim Y (D^2 + A[\tp^\dagger,\ttmfb/\nu]).
\end{aligned}
\end{equation}
{\color{black}Meanwhile, by \eqref{Area84}, \eqref{Area83}, \eqref{Area310}, we have
\begin{equation}\label{Area320}
|\ttheta^\dagger|(u,v) \lesssim r(u,v), \quad |\ttz|(u,v) \lesssim r(u,v) + Dr^3(u,v) \lesssim r(u,v).
\end{equation}
Together with \eqref{Area314} and $0 \leq \mu - Q^2/r^2 \leq \mu$, we prove
\begin{equation}\label{Area321}
|\tr|(u,v) \lesssim r^3(u,v).
\end{equation}
Thus, we deduce that
\begin{equation}\label{Area322}
\frac{\tr}{r^2}|_{\Gamma} = 0.
\end{equation}
Similarly, one can also deduce that $|\tg|(u,v) \lesssim r(u,v)$ and hence it implies
\begin{equation}\label{Area323}
\tg|_{\Gamma} = 0.
\end{equation}}
Using \eqref{Area322} and conducting integration by parts, we obtain
\begin{equation}\label{Area210}
\begin{aligned}
\int_{\mathcal{D}(0,v_0)} \frac{|\tr|}{r^3}(u,v) \D u \D v &= 
\frac{1}{2}\int^{v_0}_0 \ml \int^v_0 |\tr/\nu| \du\ml \frac{1}{r^2}\mr(u,v) \D u \mr \D v \\
&= - \frac{1}{2}\int^{v_0}_0 \frac{|\tr/\nu|}{r^2}(0,v) \D v - \frac{1}{2}\int^{v_0}_0 \ml \int^v_0 \frac{1}{r^2}\du \mlm \frac{\tr}{\nu}\mrm(u,v) \D u\mr \D v \\
&\leq \frac{1}{2}\int_{\mathcal{D}(0,v_0)}\frac{1}{r^2}\mlm \du\ml \frac{\tr}{\nu}\mr \mrm (u,v) \D u \D v.
\end{aligned}
\end{equation}
In a similar fashion, by \eqref{Area323}, we also get
\begin{equation}\label{Area211}
\begin{aligned}
\int_{\mathcal{D}(0,v_0)} \frac{|r^2 \tg|}{r^3}(u,v) \D u \D v \leq \frac{1}{2}\int_{\mathcal{D}(0,v_0)}\frac{1}{r^2}\mlm \du\ml \frac{r^2\tg}{\nu}\mr \mrm (u,v) \D u \D v.
\end{aligned}
\end{equation}
It remains to compute $\du\ml \frac{\tr}{\nu}\mr$ and $\du \ml \frac{r^2 \tg}{\nu}\mr$. For the former, we have
\begin{equation}\label{Area212}
\begin{aligned}
\du \ml \frac{\tr}{\nu}\mr =&\; \du \ml \frac{r}{\nu} \ml \frac{\mu - Q^2/r^2}{1-\mu}\lambda \nu + 4\pi \ttheta^\dagger \ttz\mr\mr \\
=&\; \frac{\mu - Q^2/r^2}{1-\mu}\lambda \nu + 4\pi \ttheta^\dagger \ttz + r \du \ml \frac{\mu - Q^2/r^2}{1-\mu}\lambda + 4 \pi \frac{\ttheta^\dagger \ttz}{\nu} \mr \\
=&\; \frac{\mu - Q^2/r^2}{1-\mu}\lambda \nu + 4\pi \ttheta^\dagger \ttz + r\lambda \du \ml \frac{\mu - Q^2/r^2}{1-\mu}\mr + \frac{\mu-Q^2/r^2}{1-\mu}r \du \lambda \\
&+ 4\pi r \ml \frac{\ttz}{\nu}\mr (D_u \ttheta - \ii \e A_u\ttheta)^\dagger  + 4\pi r \ttheta^\dagger \du \ml \frac{\ttz}{\nu} \mr.
\end{aligned}
\end{equation}
Subsequently, observe that from \eqref{Setup26} and  \eqref{Setup29.1}, we can write the term $\du \ml \frac{\mu - Q^2/r^2}{1-\mu}\mr$ as 
\begin{equation}\label{Area214}
\begin{aligned}
\du \ml \frac{\mu - Q^2/r^2}{1-\mu}\mr =&\; \du \ml \frac{1 -  Q^2/r^2}{1-\mu}\mr = \frac{1-Q^2/r^2}{(1-\mu)^2}\du \mu + \frac{1}{1-\mu}\du \ml -\frac{Q^2}{r^2}\mr \\
&\quad=  \frac{1-Q^2/r^2}{(1-\mu)^2}\ml \frac{4\pi (1-\mu)|\ttz|^2}{r\nu} - \frac{\nu}{r}(\mu - Q^2/r^2)  \mr - \frac{1}{1-\mu}\ml 8\pi \e Q \im(\phi^\dagger D_u \phi) - \frac{2Q^2 \nu}{r^3}\mr.
\end{aligned}
\end{equation}
Furthermore, observe that from \eqref{Setup73}, we can express the term $D_u \ttheta$ as
\begin{equation}\label{Area215}
D_u \ttheta = D_u(r\dv \phi) = - \lambda \frac{\ttz}{r} - \ii \e \frac{Q \phi \lambda\nu}{r(1-\mu)}.
\end{equation}
In addition, from \eqref{tildetildemathfrakb}, we have
\begin{equation}\label{Area216}
\du \ml \frac{\ttmfb}{\nu}\mr = \du \ml \frac{\nu \tp + r D_u \phi}{\nu}\mr = \frac{\ttz}{r} - \ii \e A_u \tp + \du \ml \frac{\ttz}{\nu}\mr.
\end{equation}
Collecting \eqref{Setup20} and \eqref{Area214} - \eqref{Area216}, we can reorganize and simplify \eqref{Area212} as
\begin{equation}\label{Area217}
\begin{aligned}
\du\ml \frac{\tr}{\nu}\mr  =&\; \frac{\mu - Q^2/r^2}{1-\mu}\lambda \nu + 4\pi \ttheta^\dagger \ttz + \ml \frac{\mu-Q^2/r^2}{1-\mu} \mr^2 \lambda \nu \\
&+ r\lambda \ml \frac{1-Q^2/r^2}{(1-\mu)^2}\mr\ml \frac{4\pi (1-\mu)|\ttz|^2}{r\nu} - \frac{\nu}{r}(\mu - Q^2/r^2)  \mr - \frac{r\lambda}{1-\mu}\ml 8\pi \e Q \im(\phi^\dagger D_u \phi) - \frac{2Q^2 \nu}{r^3} \mr  \\
&+ 4\pi r \ml \frac{\ttz}{\nu}\mr \ml - \lambda \frac{\ttz}{r} - \ii \e \frac{Q \phi \lambda\nu}{r(1-\mu)} - \ii \e A_u\ttheta \mr^\dagger + 4\pi r \ttheta^\dagger\ml \du\ml \frac{\ttmfb}{\nu}\mr  - \frac{\ttz}{r} + \ii \e A_u \tp \mr \\
=&\; \frac{4 \pi \lambda(\mu - Q^2/r^2)|\ttz|^2}{\nu(1-\mu)} + 4\pi r \ttheta^\dagger \du \ml \frac{\ttmfb}{\nu}\mr + E_9
\end{aligned}
\end{equation}
with
\begin{equation}\label{Area218}
E_9 := -\frac{r\lambda}{1-\mu}\ml 8\pi \e Q \im(\tp^\dagger D_u \phi) - \frac{2Q^2\nu}{r^3}\mr + 4 \pi \ii \e r \ml \frac{\ttz}{\nu}\mr\ml \frac{Q\phi \lambda \nu}{r(1-\mu)} + A_u \ttheta\mr^\dagger + 4\pi \ii \e r \ttheta^\dagger A_u \tp.
\end{equation}
{\color{black}Note that in the process of arriving at the last line in \eqref{Area217}, numerous cancellations happen.}
On top of that, since $|\ttz| \lesssim D$, we then have $|E_9|(u,v) \lesssim D^2 r^2(u,v).$ Moreover, with the help of \eqref{Area149} and \eqref{Area158}, we can also write
\begin{equation}\label{Area219}
\begin{aligned}
\ttheta^\dagger \du \ml \frac{\ttmfb}{\nu}\mr =&\; r  \du \ml \frac{\ttmfb}{\nu}\mr \dv \tp^\dagger = r\ml \frac{\p(\ttmfb/\nu,\tp^\dagger)}{\p(u,v)} + \dv \ml \frac{\ttmfb}{\nu}\mr\du \tp^\dagger\mr \\
=&\; r\ml \frac{\p(\ttmfb/\nu,\tp^\dagger)}{\p(u,v)} + \frac{\tz^\dagger}{r} \ml -\frac{\mu - Q^2/r^2}{1-\mu}\frac{\tz}{\nu}\frac{\lambda}{r} + \frac{E_4}{\nu} \mr   \mr \\
=&\; r\ml \frac{\p(\ttmfb/\nu,\tp^\dagger)}{\p(u,v)}\mr - \frac{\lambda(\mu - Q^2/r^2)|\tz|^2}{r\nu(1-\mu)} + \frac{\tz^\dagger}{\nu}E_4 \\
=&\; r\ml \frac{\p(\ttmfb/\nu,\tp^\dagger)}{\p(u,v)}\mr - \frac{\lambda(\mu - Q^2/r^2)|\ttz|^2}{r\nu(1-\mu)} + E_{10}
\end{aligned}
\end{equation}
with
\begin{equation}\label{Area220}
E_{10} := -\frac{\lambda(\mu-Q^2/r^2)}{r\nu(1-\mu)}(2 \e r^2 A_u \im (\phi \du \tp^\dagger) - r^2 \e^2 A_u^2 |\phi|^2) + \frac{\tz^\dagger}{\nu}E_4 
\end{equation}
and $|E_{10}|(u,v) \lesssim D^2 r(u,v)$ since $|E_4|(u,v) \lesssim Dr(u,v)$. Combining \eqref{Area217} and \eqref{Area219}, we can now write
\begin{equation}\label{Area221}
\begin{aligned}
\du \ml \frac{\tr}{\nu}\mr &= 4\pi r^2 \frac{\p(\ttmfb/\nu,\tp^\dagger)}{\p(u,v)} + E_{11}
\end{aligned}
\end{equation}
with
\begin{equation}\label{Area222}
E_{11} := E_9 + 4\pi r E_{10}
\end{equation}
and $|E_{11}|(u,v) \lesssim D^2 r^2(u,v)$. Applying \eqref{Area210}, this implies that
\begin{equation}\label{Area223}
\int_{\mathcal{D}(0,v_0)} \frac{|\tr|}{r^3}(u,v) \D u \D v \lesssim A[\tp^\dagger,\ttmfb/\nu] + D^2.
\end{equation}
Next, we continue to estimate $\du\ml \frac{r^2 \tg}{\nu}\mr$ and we first compute
\begin{equation}\label{Area224}
\begin{aligned}
\du\ml \frac{\ttheta^\dagger r}{\nu} \ml \frac{\tmfa}{\lambda} - \frac{\ttmfb}{\nu}\mr\mr &= \ml \ttheta^\dagger 
+ r  \du\ml \frac{\ttheta^\dagger}{\nu}\mr    \mr \ml \frac{\tmfa}{\lambda} - \frac{\ttmfb}{\nu}\mr + \frac{r}{\nu}\ttheta^\dagger \du\ml \frac{\tmfa}{\lambda} - \frac{\ttmfb}{\nu}\mr.
\end{aligned}
\end{equation}
Using \eqref{Area139}, we can rewrite $\du \ml \frac{\ttheta^\dagger}{\nu}\mr$ as
\begin{equation}\label{Area225}
\begin{aligned}
\du\ml \frac{\ttheta^\dagger}{\nu}\mr &= \frac{\nu \du \ttheta^\dagger - \ttheta^\dagger \du \nu}{\nu^2} = -\frac{\lambda}{r\nu}\tz^\dagger + \frac{E_1^\dagger}{\nu} - \frac{\ttheta^\dagger \du \nu}{\nu^2}.
\end{aligned}
\end{equation}
Together with \eqref{Area153}, we then simplify \eqref{Area224} to
\begin{equation}\label{Area226}
\begin{aligned}
\du \ml \frac{\ttheta^\dagger r}{\nu}\ml \frac{\tmfa}{\lambda} - \frac{\ttmfb}{\nu}\mr\mr =&\; \ml \ttheta^\dagger 
-\frac{\lambda \tz^\dagger}{\nu} + \frac{rE_1^\dagger}{\nu} - \frac{r\ttheta^\dagger \du \nu}{\nu^2}   \mr \ml \frac{\tmfa}{\lambda} - \frac{\ttmfb}{\nu}\mr \\
&+ \frac{r^2}{\nu}\ml \frac{1}{r^2}\ml \frac{\mu - Q^2/r^2}{1-\mu} \mr\ml \frac{|\tz|^2\lambda}{\nu} - \frac{ |\ttheta|^2\nu}{\lambda}\mr + \frac{\p(\tp^\dagger,\ttmfb/\nu)}{\p(u,v)} +\frac{\ttheta^\dagger E_1}{r\lambda} - \frac{\tz^\dagger E_4}{r\nu} \mr.
\end{aligned}
\end{equation}
Appealing to \eqref{Area141}, we have
\begin{equation}\label{Area227}
\begin{aligned}
\ml \nu\ttheta^\dagger 
-\lambda \tz^\dagger - \frac{r\ttheta^\dagger \du \nu}{\nu}   \mr \ml \frac{\tmfa}{\lambda} - \frac{\ttmfb}{\nu}\mr =&\; \ml \nu\ttheta^\dagger 
-\lambda \tz^\dagger - \frac{r\ttheta^\dagger \du \nu}{\nu}   \mr \ml \frac{\ttheta}{\lambda} - \frac{\tz}{\nu} + E_2\mr \\
=&\; \frac{\nu|\ttheta|^2}{\lambda} - (\ttheta^\dagger \tz + \tz^\dagger \ttheta) + \frac{\lambda}{\nu}|\tz|^2 + E_2(\nu\ttheta^\dagger - \lambda \tz^\dagger) \\
&- \frac{r}{\nu \lambda}\du \nu|\ttheta|^2 + (\du \nu) \ttheta^\dagger \ml \frac{r}{\nu} \mr\ml \frac{\tz}{\nu} - E_2 \mr.
\end{aligned}
\end{equation}
On the other hand, from \eqref{Setup20}, we have
\begin{equation}\label{Area228}
\begin{aligned}
\du \nu \ttheta^\dagger &= r(\du \nu \dv \tp^\dagger) = r \ml \frac{\p(\nu,\tp^\dagger)}{\p(u,v)} + \dv\nu \du \tp^\dagger \mr = r \ml \frac{\p(\nu,\tp^\dagger)}{\p(u,v)} + \frac{\lambda \nu\tz^\dagger}{r^2} \frac{\mu - Q^2/r^2}{1-\mu}\mr. \\
\end{aligned}
\end{equation}
Hence, we obtain
\begin{equation}\label{Area229}
\begin{aligned}
\ml \nu\ttheta^\dagger 
-\lambda \tz^\dagger - \frac{r\ttheta^\dagger \du \nu}{\nu}   \mr \ml \frac{\tmfa}{\lambda} - \frac{\ttmfb}{\nu}\mr 
=&\; \frac{\nu|\ttheta|^2}{\lambda} - (\ttheta^\dagger \tz + \tz^\dagger \ttheta) + \frac{\lambda}{\nu}|\tz|^2 - \frac{r}{\nu \lambda}\du \nu|\ttheta|^2 + \frac{\lambda |\tz|^2}{\nu} \frac{\mu - Q^2/r^2}{1-\mu} \\
&+  \frac{r^2 \tz}{\nu^2}\frac{\p(\nu,\tp^\dagger)}{\p(u,v)}  + E_2(\nu\ttheta^\dagger - \lambda \tz^\dagger) - E_2 \ml \frac{r^2}{\nu} \mr\ml \frac{\p(\nu,\tp^\dagger)}{\p(u,v)} + \frac{\lambda \nu \tz^\dagger}{r^2} \frac{\mu - Q^2/r^2}{1-\mu}\mr.
\end{aligned}
\end{equation}
Plugging this back to \eqref{Area226}, we arrive at
\begin{equation}\label{Area230}
\begin{aligned}
\du\ml\frac{\ttheta^\dagger r}{\nu} \ml \frac{\tmfa}{\lambda} - \frac{\ttmfb}{\nu}\mr\mr =&\; \frac{1}{\nu}\ml \frac{\nu|\ttheta|^2}{\lambda} - (\ttheta^\dagger \tz + \tz^\dagger \ttheta) + \frac{\lambda}{\nu}|\tz|^2 - \frac{r}{\nu \lambda}\du \nu|\ttheta|^2 + \frac{\lambda |\tz|^2}{\nu} \frac{\mu - Q^2/r^2}{1-\mu}\mr +  \frac{r^2 \tz}{\nu^3}\frac{\p(\nu,\tp^\dagger)}{\p(u,v)}  \\
&+ \frac{E_2}{\nu}(\nu\ttheta^\dagger - \lambda \tz^\dagger) - E_2 \ml \frac{r^2}{\nu^2} \mr\ml \frac{\p(\nu,\tp^\dagger)}{\p(u,v)} + \frac{\lambda \nu\tz^\dagger}{r^2} \frac{\mu - Q^2/r^2}{1-\mu}\mr+ \frac{rE_1^\dagger}{\nu}\ml \frac{\tmfa}{\lambda} - \frac{\ttmfb}{\nu}\mr\\
&+ \frac{1}{\nu}\ml \frac{\mu-Q^2/r^2}{1-\mu}\mr\ml \frac{|\tz|^2\lambda}{\nu} - \frac{|\ttheta|^2\nu}{\lambda}\mr + \frac{r^2}{\nu}\frac{\p(\tp^\dagger,\ttmfb/\nu)}{\p(u,v)} + \frac{r^2}{\nu}\ml \frac{\ttheta^\dagger E_1}{r\lambda} - \frac{\tz^\dagger E_4}{r\nu}\mr \\
=&\; r^2 \frac{(\tz/\nu - E_2)}{\nu^2} \frac{\p(\nu,\tp^\dagger)}{\p(u,v)} + r^2\frac{1}{\nu}\frac{\p(\tp^\dagger,\ttmfb/\nu)}{\p(u,v)} + R_1 + E_{12},
\end{aligned}
\end{equation}
with
\begin{equation}\label{Area231}
\begin{aligned}
R_1 :=&\;  \frac{1}{\nu}\ml \frac{\nu|\ttheta|^2}{\lambda} - (\ttheta^\dagger \tz + \tz^\dagger \ttheta) + \frac{\lambda}{\nu}|\tz|^2 - \frac{r}{\nu \lambda}\du \nu|\ttheta|^2 + \frac{\lambda |\tz|^2}{\nu} \frac{\mu - Q^2/r^2}{1-\mu}\mr\\
&+ \frac{1}{\nu}\ml \frac{\mu-Q^2/r^2}{1-\mu}\mr\ml \frac{|\tz|^2\lambda}{\nu} - \frac{|\ttheta|^2\nu}{\lambda}\mr, \\
E_{12}:= &\; \frac{E_2}{\nu}(\nu \ttheta^\dagger - \lambda \tz^\dagger) - \frac{E_2 \nu \lambda \tz^\dagger}{\nu^2}\frac{\mu - Q^2/r^2}{1-\mu} + \frac{rE_1^\dagger}{\nu}\ml \frac{\tmfa}{\lambda} - \frac{\ttmfb}{\nu}\mr + \frac{r^2}{\nu}\ml \frac{\ttheta^\dagger E_1}{r\lambda} - \frac{\tz^\dagger E_4}{r\nu}\mr.
\end{aligned}
\end{equation}
{\color{black}Note that to arrive at the last equality in \eqref{Area230}, numerous cancellations were explored.} By recognizing that $\im(R_1) = 0$ and $|E_{12}|(u,v) \lesssim D^2r^2(u,v)$, we can take the imaginary part of \eqref{Area230} and deduce that 
\begin{equation}\label{Area232}
\begin{aligned}
\du\ml \frac{r^2 \tg}{\nu}\mr &=   4\pi \im \ml \du \ml \frac{\ttheta^\dagger r}{\nu} \ml \frac{\tmfa}{\lambda} - \frac{\ttmfb}{\nu}\mr\mr \mr \\
&= 4\pi \im \ml r^2 \frac{(\tz/\nu - E_2)}{\nu^2} \frac{\p(\nu,\tp^\dagger)}{\p(u,v)} + r^2\frac{1}{\nu}\frac{\p(\tp^\dagger,\ttmfb/\nu)}{\p(u,v)}  + E_{12}\mr.
\end{aligned}
\end{equation}
This implies
\begin{equation}\label{Area233}
\mlm \du \ml \frac{r^2 \tg}{\nu}\mr\mrm \lesssim r^2 \ml D\mlm \frac{\p(\nu,\tp^\dagger)}{\p(u,v)}\mrm + \mlm \frac{\p(\tp^\dagger,\ttmfb/\nu)}{\p(u,v)}\mrm + D^2\mr.
\end{equation}
Plugging \eqref{Area233} back to \eqref{Area211}, we thus derive
\begin{equation}\label{Area234}
\int_{\mathcal{D}(0,v_0)} \frac{|r^2 \tg|}{r^3}(u,v) \D u \D v \lesssim D \cdot A[\nu,\tp^\dagger] + A[\tp^\dagger,\ttmfb/\nu] + D^2.
\end{equation}
Henceforth, by utilizing \eqref{Area208}, \eqref{Area209}, \eqref{Area223}, \eqref{Area234} and the fact that $A[\lambda,\tp^\dagger] = A[\lambda,\tp]$, we arrive at
\begin{equation}\label{Area235}
A[\lambda,\tp^\dagger] \lesssim (D+Y)A[\tp^\dagger,\ttmfb/\nu] + D^2\cdot A[\nu,\tp^\dagger] + D^2Y + D^3.
\end{equation}

It remains to estimate $A[\nu,\tp^\dagger]$. Parallel to the above arguments, we start off by defining 
\begin{equation}\label{Area236}
\tgs := 4\pi \im\ml \frac{\tz^\dagger}{r}\ml \frac{\tmfa}{\lambda} - \frac{\ttmfb}{\nu}\mr\mr
\end{equation}
and we have
\begin{equation}\label{Area237}
\tfs + \ii \tgs = 4\pi \frac{\tz^\dagger}{r}\ml \frac{\ttheta}{\lambda} - \frac{\tz}{\nu} + E_2\mr
\end{equation}
with $E_2$ given in \eqref{Area141}. This implies that
\begin{equation}\label{Area238}
\frac{4\pi}{r^2}\frac{\ttheta \tz^\dagger}{\lambda \nu}\ml \nu \ttheta - \lambda \tz\mr = \frac{\tfs \ttheta}{r} + \ii \frac{\tgs \ttheta}{r} - \frac{4\pi \ttheta \tz^\dagger}{r^2}E_2.
\end{equation}
Together with \eqref{Area132}, we compute
\begin{equation}\label{Area239}
\begin{aligned}
&\frac{1}{\nu}\frac{\p(\nu,\tp)}{\p(u,v)} = \du \log |\nu| \dv \tp - \dv \log |\nu| \du \tp = \du\log|\nu| \frac{\ttheta}{r} -\frac{\mu - Q^2/r^2}{1-\mu} \frac{\lambda \tz}{r^2} \\
&\quad= \frac{1}{r^2}(\nu \ttheta - \lambda \tz)\ml \frac{\mu - Q^2/r^2}{1-\mu} + 4\pi \frac{\ttheta \tz^\dagger}{\lambda \nu} \mr  + \frac{\ttheta}{r}\ml \du \log|\nu| - \frac{\mu - Q^2/r^2}{1-\mu}\frac{\nu}{r}\mr - 4\pi \frac{\ttheta \tz^\dagger}{r^2}(\nu \ttheta - \lambda \tz) \\
&\quad= \frac{1}{r^2}\ml \frac{\nu \ttheta - \lambda \tz}{\lambda \nu} \mr\ml \frac{\mu - Q^2/r^2}{1-\mu}\lambda \nu + 4\pi \ttheta \tz^\dagger \mr  + \frac{\ttheta}{r}\ml \du \log|\nu| - \frac{\mu - Q^2/r^2}{1-\mu}\frac{\nu}{r} - \tfs\mr - \ii \frac{\ttheta\tgs}{r} + 4\pi \frac{\ttheta\tz^\dagger}{r^2}E_2. \\
\end{aligned}
\end{equation}
Note that from \eqref{Area163}, we have
\begin{equation}\label{Area240}
\dv \ml \du \log|\nu| - \frac{\mu - Q^2/r^2}{1-\mu}\frac{\nu}{r} - \tfs\mr = -4\pi \re \ml \frac{\p(\tp^\dagger,\tmfa/\lambda)}{\p(u,v)}\mr - E_5 + E_6.
\end{equation}
This implies that 
\begin{equation}\label{Area241}
\mlm \du \log|\nu| - \frac{\mu - Q^2/r^2}{1-\mu}\frac{\nu}{r} - \tfs\mrm(u,v) \leq G^*(u)
\end{equation}
with
\begin{equation}\label{Area242}
G^*(u) := \mlm \frac{\D}{\D u}\log \mlm \nu|_{\Gamma}\mrm \mrm(u) + 4\pi \int^{v_0}_u \mlm \frac{\p(\tp^\dagger,\tmfa/\lambda)}{\p(u,v)}\mrm(u,v) \D v + \int_u^{v_0}|E_5| + |E_6|(u,v) \D v.
\end{equation}
In \eqref{Area242}, we have used the fact that both $\frac{\mu - Q^2/r^2}{1-\mu}\frac{\nu}{r}$ and $\tfs$ vanishes at $\Gamma$.
By \eqref{Area118} and \eqref{Area242}, we then deduce that
\begin{equation}\label{Area243}
\int_0^{v_0}G^*(u) \D u \lesssim TV_{(0,v_0)}[\lambda|_{\Gamma}] + A[\tp^\dagger,\tmfa/\lambda] + D^2.
\end{equation}
Parallel to \eqref{Area207}, we set
\begin{equation}\label{Area244}
\trs := r\ml \frac{\mu - Q^2/r^2}{1-\mu}\lambda \nu + 4\pi \ttheta \tz^\dagger\mr.
\end{equation}
Back to \eqref{Area239}, we now obtain
\begin{equation}\label{Area245a}
\mlm \frac{1}{\nu}\frac{\p(\nu,\tp)}{\p(u,v)}\mrm \lesssim D\frac{|\trs|}{r^3} + |\dv \tp|G^*(u) + D\frac{|r^2\tgs|}{r^3} + D^3.
\end{equation}
{\color{black}Integrating on the domain $\mathcal{D}(0,v_0)$, we hence derive
\begin{equation}\label{Area245}
\begin{aligned}
\int_{\mathcal{D}(0,v_0)} \mlm \frac{1}{\nu}\frac{\p(\nu,\tp)}{\p(u,v)}\mrm \D u \D v &\lesssim D \int_{\mathcal{D}(0,v_0)}\frac{|\trs|}{r^3}(u,v) \D u \D v + \int_{\mathcal{D}(0,v_0)}|\dv \tp|(u,v)G^*(u) \D u \D v \\
&+ D \int_{\mathcal{D}(0,v_0)}\frac{|r^2\tgs|}{r^3}(u,v) \D u \D v + D^3.
\end{aligned}
\end{equation}
We then estimate the integrals in the right of \eqref{Area245}.} First, note that by \eqref{Area15}, we have 
\begin{equation}\label{Area246}
\begin{aligned}
\int_{\mathcal{D}(0,v_0)}|\dv \tp|(u,v) G^*(u) \D u &\lesssim  X \int_0^{v_0}G^*(u) \D u \lesssim X(TV_{(0,v_0)}[\lambda|_{\Gamma}] + A[\tp^\dagger,\tmfa/\lambda] + D^2).
\end{aligned}
\end{equation}
On the other hand, analogously to the above, one can show that 
\begin{equation}\label{Area325}
|\trs|(u,v) \lesssim r^3(u,v) \quad \text{ and } \quad  |\tgs|(u,v) \lesssim r(u,v)
\end{equation}
and thus
\begin{equation}\label{Area326}
\frac{\trs}{r^2}|_{\Gamma} = 0 \quad \text{ and } \quad \tgs|_{\Gamma} = 0.
\end{equation}
By employing \eqref{Area326}, we perform integration by parts and obtain
\begin{equation}\label{Area247}
\begin{aligned}
\int_{\mathcal{D}(0,v_0)}\frac{|\trs|}{r^3}(u,v) \D u \D v &= \frac{1}{2}\int^{v_0}_0 \ml \int^{v_0}_u \frac{|\trs|}{-
\lambda}\dv \ml \frac{1}{r^2}\mr(u,v) \D v\mr \D u \\
&= - \frac{1}{2}\int^{v_0}_0 \frac{|\trs|/\lambda}{r^2}(u,v_0) \D u + \frac{1}{2}\int^{v_0}_0 \int^{v_0}_u \dv \mlm \frac{\trs}{\lambda}\mrm \cdot \frac{1}{r^2}(u,v) \D v \D u \\
&\leq \frac{1}{2}\int_{\mathcal{D}(0,v_0)}\frac{1}{r^2}\mlm \dv \ml \frac{\trs}{\lambda}\mr\mrm(u,v) \D u \D v.
\end{aligned}
\end{equation}
Meanwhile, a symmetric argument with \eqref{Area326} yields
\begin{equation}\label{Area248}
\begin{aligned}
\int_{\mathcal{D}(0,v_0)}\frac{|r^2 \tgs|}{r^3}(u,v) \D u \D v 
&\leq \frac{1}{2}\int_{\mathcal{D}(0,v_0)}\frac{1}{r^2}\mlm \dv \ml \frac{r^2 \tgs}{\lambda}\mr\mrm(u,v) \D u \D v.
\end{aligned}
\end{equation}
It remains to compute $\dv \ml \frac{\trs}{\lambda}\mr$ and $\dv \ml \frac{r^2 \tgs}{\lambda}\mr$. For the former, analogously to \eqref{Area212}, we have
\begin{equation}\label{Area249}
\begin{aligned}
\dv \ml \frac{\trs}{\lambda} \mr =&\; \dv \ml \frac{r}{\lambda}\ml \frac{\mu - Q^2/r^2}{1-\mu}\lambda \nu + 4\pi \ttheta \tz^\dagger\mr \mr \\
=&\; \frac{\mu- Q^2/r^2}{1-\mu}\lambda \nu + 4\pi \ttheta \tz^\dagger +  r \dv \ml \frac{\mu - Q^2/r^2}{1-\mu}\nu + \frac{4\pi \ttheta \tz^\dagger}{\lambda}\mr \\
=&\; \frac{\mu- Q^2/r^2}{1-\mu}\lambda \nu + 4\pi \ttheta \tz^\dagger +  r\nu \dv \ml \frac{\mu - Q^2/r^2}{1-\mu}\mr + \ml \frac{\mu - Q^2/r^2}{1-\mu}\mr r\dv \nu  \\
&+ 4\pi r \ml \frac{\ttheta}{\lambda}\mr \dv \tz^\dagger + 4\pi r \tz^\dagger\dv \ml\frac{\ttheta}{\lambda} \mr.
\end{aligned}
\end{equation}
By \eqref{Setup27} and \eqref{Setup29}, we can write $\dv \ml \frac{\mu - Q^2/r^2}{1-\mu}\mr$ as
\begin{equation}\label{Area250}
\begin{aligned}
\dv \ml \frac{\mu - Q^2/r^2}{1-\mu}\mr =&\; \frac{1-Q^2/r^2}{(1-\mu)^2}\dv \mu + \frac{1}{1-\mu}\dv\ml -\frac{Q^2}{r^2}\mr \\
=&\; \frac{1-Q^2/r^2}{(1-\mu)^2}\ml \frac{4\pi (1-\mu)|\ttheta|^2}{r\lambda} - \frac{\lambda}{r}(\mu - Q^2/r^2)\mr + \frac{1}{1-\mu}\ml 8\pi \e Q \im(\phi^\dagger \dv \tp) + \frac{2Q^2\lambda}{r^3}\mr.
\end{aligned}
\end{equation}
In addition, to treat the last term in \eqref{Area249}, we observe that
\begin{equation}\label{Area251}
\begin{aligned}
\dv \ml \frac{\tmfa}{\lambda}\mr = \dv\ml \frac{\dv r \tp + r \dv \tp}{\dv r}\mr = \frac{\ttheta}{r} + \dv \ml \frac{\ttheta}{\lambda}\mr.
\end{aligned}
\end{equation}
In conjunction with \eqref{Setup20} and \eqref{Area139}, the above two expressions imply that \eqref{Area249} simplifies to
\begin{equation}\label{Area252} 
\begin{aligned}
\dv \ml \frac{\trs}{\lambda}\mr =&\;\frac{\mu- Q^2/r^2}{1-\mu}\lambda \nu + 4\pi \ttheta \tz^\dagger + \ml \frac{\mu - Q^2/r^2}{1-\mu}\mr^2 \lambda \nu  + r\nu\ml \frac{1-Q^2/r^2}{(1-\mu)^2}\ml \frac{4\pi (1-\mu)|\ttheta|^2}{r\lambda} - \frac{\lambda}{r}(\mu - Q^2/r^2)\mr\mr \\
&+ \frac{r\nu}{1-\mu}\ml 8\pi \e Q \im(\phi^\dagger \dv \tp) + \frac{2Q^2\lambda}{r^3}\mr + 4\pi r \ml \frac{\ttheta}{\lambda}\mr \ml -\frac{\nu \ttheta^\dagger}{r} + E_1^\dagger\mr  + 4\pi r \tz^\dagger\ml \dv\ml\frac{\tmfa}{\lambda}\mr - \frac{\ttheta}{r} \mr\\
=&\; \frac{4\pi \nu(\mu - Q^2/r^2)|\ttheta|^2}{\lambda(1-\mu)} + 4\pi r \tz^\dagger \dv \ml \frac{\tmfa}{\lambda} \mr + E_{13}
\end{aligned}
\end{equation}
with
\begin{equation}\label{Area253}
E_{13} := \frac{r\nu}{1-\mu}\ml 8\pi \e Q \im(\phi^\dagger \dv \tp) + \frac{2Q^2\lambda}{r^3}\mr + \frac{4\pi r \ttheta E_1^\dagger}{\lambda}
\end{equation}
and $|E_{13}|(u,v) \lesssim D^2r^2(u,v).$ {\color{black}Note that to arrive at the last equality in \eqref{Area252}, various cancellations were explored.} With the help of \eqref{Area147} and \eqref{Area149}, we also have
\begin{equation}\label{Area254}
\begin{aligned}
\tz^\dagger \dv \ml \frac{\tmfa}{\lambda}\mr = r\ml \du \tp^\dagger \dv \ml \frac{\tmfa}{\lambda}\mr\mr &= r\ml \frac{\p(\tp^\dagger,\tmfa/\lambda)}{\p(u,v)} + \dv \tp^\dagger \du \ml \frac{\tmfa}{\lambda}\mr \mr\\
&= r\frac{\p(\tp^\dagger,\tmfa/\lambda)}{\p(u,v)} - \frac{\nu(\mu-Q^2/r^2)|\ttheta|^2}{r\lambda(1-\mu)} + \frac{\ttheta^\dagger E_1}{\lambda}.
\end{aligned}
\end{equation}
Plugging \eqref{Area254} into \eqref{Area252}, we deduce that
\begin{equation}\label{Area255}
\dv \ml \frac{\trs}{\lambda}\mr = 4\pi r^2 \frac{\p(\tp^\dagger,\tmfa/\lambda)}{\p(u,v)} + E_{14} \quad \text{ with } \quad E_{14} = E_{13} + \frac{4\pi r\ttheta^\dagger E_1}{\lambda}
\end{equation}
and $|E_{14}|\lesssim D^2r^2(u,v)$. Together with \eqref{Area247}, this implies that
\begin{equation}\label{Area257}
\int_{\mathcal{D}(0,v_0)}\frac{|\trs|}{r^3}(u,v) \D u \D v \lesssim A[\tp^\dagger,\tmfa/\lambda] + D^2.
\end{equation}
Next, we proceed to compute $\dv \ml \frac{r^2 \tgs}{\nu}\mr$.  First, observe that
\begin{equation}\label{Area258}
\dv \ml \frac{\tz^\dagger r}{\lambda}\ml \frac{\tmfa}{\lambda} - \frac{\ttmfb}{\nu}\mr \mr = \ml \tz^\dagger + r \dv\ml \frac{\tz^\dagger}{\lambda}\mr\mr\ml \frac{\tmfa}{\lambda} - \frac{\ttmfb}{\nu}\mr + \frac{r}{\lambda}\tz^\dagger \dv \ml \frac{\tmfa}{\lambda} - \frac{\ttmfb}{\nu}\mr.
\end{equation}
By \eqref{Area139}, we can write
\begin{equation}\label{Area259}
\begin{aligned}
\dv \ml \frac{\tz^\dagger}{\lambda}\mr &= \frac{\lambda \dv \tz^\dagger - \tz^\dagger \dv \lambda}{\lambda^2} = -\frac{\nu}{\lambda r}\ttheta^\dagger + \frac{E_1^\dagger}{\lambda} - \frac{\tz^\dagger \dv \lambda }{\lambda^2}.
\end{aligned}
\end{equation}
Combining with \eqref{Area154}, we then simplify \eqref{Area258} to
\begin{equation}\label{Area260}
\begin{aligned}
\dv \ml \frac{\tz^\dagger r}{\lambda}\ml \frac{\tmfa}{\lambda} - \frac{\ttmfb}{\nu}\mr \mr =&\; \ml \tz^\dagger 
- \frac{\nu \ttheta^\dagger}{\lambda} + \frac{rE_1^\dagger}{\lambda} - \frac{r\tz^\dagger \dv \lambda}{\lambda^2}  \mr\ml \frac{\tmfa}{\lambda} - \frac{\ttmfb}{\nu}\mr \\
&+ \frac{r^2}{\lambda}\ml  \frac{1}{r^2}\ml \frac{\mu - Q^2/r^2}{1-\mu} \mr\ml \frac{|\tz|^2\lambda}{\nu} - \frac{ |\ttheta|^2\nu}{\lambda}\mr + \frac{\p(\tp^\dagger,\tmfa/\lambda)}{\p(u,v)} + \frac{\ttheta^\dagger E_1}{r\lambda} - \frac{\tz^\dagger E_4}{r\nu} \mr.
\end{aligned}
\end{equation}
Observe that by \eqref{Area141}, we have
\begin{equation}\label{Area261}
\begin{aligned}
\ml \lambda \tz^\dagger - \nu \ttheta^\dagger - \frac{r\tz^\dagger \dv \lambda}{\lambda} \mr\ml \frac{\tmfa}{\lambda} - \frac{\ttmfb}{\nu} \mr =&\; \ml \lambda \tz^\dagger - \nu \ttheta^\dagger - \frac{r\tz^\dagger \dv \lambda}{\lambda} \mr\ml \frac{\ttheta}{\lambda} - \frac{\tz}{\nu} + E_2 \mr \\
=&\; -\frac{\nu|\ttheta|^2}{\lambda} + (\tz^\dagger \ttheta + \tz \ttheta^\dagger) - \frac{\lambda |\tz|^2}{\nu} + E_2(\lambda \tz^\dagger - \nu \ttheta^\dagger) \\
&+ \frac{r|\tz|^2 \dv \lambda}{\lambda \nu} - (\dv \lambda) \tz^\dagger \ml \frac{r}{\lambda} \mr \ml \frac{\ttheta}{\lambda} + E_2\mr. \\
\end{aligned}
\end{equation}
On the other hand, from \eqref{Setup20}, we have
\begin{equation}\label{Area262}
\begin{aligned}
(\dv \lambda) \tz^\dagger = r ((\dv \lambda) \du \tp^\dagger) = r\ml \frac{\p(\tp^\dagger,\lambda)}{\p(u,v)} + (\du\lambda) \dv \tp^\dagger\mr = r\ml \frac{\p(\tp^\dagger,\lambda)}{\p(u,v)} + \frac{\lambda \nu \ttheta^\dagger}{r^2}\frac{\mu - Q^2/r^2}{1-\mu}\mr,
\end{aligned}
\end{equation}
and hence the first term in \eqref{Area260} has the following expression
\begin{equation}\label{Area263}
\begin{aligned}
\ml \lambda \tz^\dagger - \nu \ttheta^\dagger - \frac{r\tz^\dagger \dv \lambda}{\lambda} \mr\ml \frac{\tmfa}{\lambda} - \frac{\ttmfb}{\nu} \mr 
=&\; -\frac{\nu|\ttheta|^2}{\lambda} + (\tz^\dagger \ttheta + \tz \ttheta^\dagger) - \frac{\lambda |\tz|^2}{\nu} + \frac{r|\tz|^2 \dv \lambda}{\lambda \nu} - \frac{\nu|\ttheta|^2}{\lambda}\frac{\mu - Q^2/r^2}{1-\mu}  \\
&\; - \frac{r^2 \ttheta}{\lambda^2}\frac{\p(\tp^\dagger,\lambda)}{\p(u,v)} + E_2(\lambda \tz^\dagger - \nu \ttheta^\dagger) - E_2\ml\frac{r^2}{\lambda} \mr \ml \frac{\p(\tp^\dagger,\lambda)}{\p(u,v)} + \frac{\lambda \nu \ttheta^\dagger}{r^2}\frac{\mu - Q^2/r^2}{1-\mu}\mr. \\
\end{aligned}
\end{equation}
Plugging this back to \eqref{Area260}, we have therefore obtained
\begin{equation}\label{Area264}
\begin{aligned}
\dv \ml \frac{\tz^\dagger r}{\lambda}\ml \frac{\tmfa}{\lambda} - \frac{\ttmfb}{\nu}\mr \mr =&\; \frac{1}{\lambda} \ml -\frac{\nu|\ttheta|^2}{\lambda} + (\tz^\dagger \ttheta + \tz \ttheta^\dagger) - \frac{\lambda |\tz|^2}{\nu} + \frac{r|\tz|^2 \dv \lambda}{\lambda \nu} - \frac{\nu|\ttheta|^2}{\lambda}\frac{\mu - Q^2/r^2}{1-\mu} \mr  - \frac{r^2 \ttheta}{\lambda^3}\frac{\p(\tp^\dagger,\lambda)}{\p(u,v)}\\
&+ \frac{E_2}{\lambda}(\lambda \tz^\dagger - \nu \ttheta^\dagger) - E_2\ml\frac{r^2}{\lambda^2} \mr \ml \frac{\p(\tp^\dagger,\lambda)}{\p(u,v)} + \frac{\lambda \nu \ttheta^\dagger}{r^2}\frac{\mu - Q^2/r^2}{1-\mu}\mr + \frac{rE_1^\dagger}{\lambda}\ml \frac{\tmfa}{\lambda} - \frac{\ttmfb}{\nu}\mr \\
&+  \frac{1}{\lambda}\ml \frac{\mu - Q^2/r^2}{1-\mu} \mr\ml \frac{|\tz|^2\lambda}{\nu} - \frac{ |\ttheta|^2\nu}{\lambda}\mr + \frac{r^2}{\lambda}\frac{\p(\tp^\dagger,\tmfa/\lambda)}{\p(u,v)} + \frac{r^2}{\lambda}\ml \frac{\ttheta^\dagger E_1}{r\lambda} - \frac{\tz^\dagger E_4}{r\nu} \mr  \\
=&\; r^2 \ml \frac{E_2 - \ttheta/\lambda}{\lambda^2} \mr \frac{\p(\tp^\dagger,\lambda)}{\p(u,v)} + r^2 \ml \frac{1}{\lambda} \mr \frac{\p(\tp^\dagger,\tmfa/\lambda)}{\p(u,v)} + R_2 + E_{15},
\end{aligned}
\end{equation}
with
\begin{equation}\label{Area265}
\begin{aligned}
R_2 :=&\;  \frac{1}{\lambda} \ml -\frac{\nu|\ttheta|^2}{\lambda} + (\tz^\dagger \ttheta + \tz \ttheta^\dagger) - \frac{\lambda |\tz|^2}{\nu} + \frac{r|\tz|^2 \dv \lambda}{\lambda \nu} - \frac{\nu|\ttheta|^2}{\lambda}\frac{\mu - Q^2/r^2}{1-\mu} \mr \\
&+ \frac{1}{\lambda}\ml \frac{\mu - Q^2/r^2}{1-\mu} \mr\ml \frac{|\tz|^2\lambda}{\nu} - \frac{ |\ttheta|^2\nu}{\lambda}\mr, \\
E_{15} := &\; \frac{E_2}{\lambda}(\lambda \tz^\dagger - \nu \ttheta^\dagger) - \frac{E_2\nu \ttheta^\dagger}{\lambda}\frac{\mu - Q^2/r^2}{1-\mu}+ \frac{rE_1^\dagger}{\lambda}\ml \frac{\tmfa}{\lambda} - \frac{\ttmfb}{\nu}\mr + \frac{r^2}{\lambda}\ml \frac{\ttheta^\dagger E_1}{r\lambda} - \frac{\tz^\dagger E_4}{r\nu} \mr.
\end{aligned}
\end{equation}
{\color{black}Similar to \eqref{Area252}, to arrive at the last line in \eqref{Area264}, multiple cancellations have to be explored.}
Employing $\im(R_2) = 0$ and $|E_{15}|\lesssim D^2r^2(u,v)$, we arrive at
\begin{equation}\label{Area266a}
\begin{aligned}
\dv \ml \frac{r^2 \tgs}{\lambda}\mr &= 4\pi \im \ml \dv \ml \frac{\tz^\dagger r}{\lambda}\ml \frac{\tmfa}{\lambda}- \frac{\ttmfb}{\nu}\mr\mr\mr\\
&= 4\pi \im\ml r^2 \ml \frac{E_2 - \ttheta/\lambda}{\lambda^2} \mr \frac{\p(\tp^\dagger,\lambda)}{\p(u,v)} + r^2 \ml \frac{1}{\lambda} \mr \frac{\p(\tp^\dagger,\tmfa/\lambda)}{\p(u,v)} + E_{15} \mr \\
\end{aligned}
\end{equation}
and thus 
\begin{equation}\label{Area266}
\mlm \dv \ml \frac{r^2\tgs}{\lambda }\mr\mrm \lesssim r^2\ml 
 D \mlm \frac{\p(\tp^\dagger,\lambda)}{\p(u,v)}\mrm + \mlm \frac{\p(\tp^\dagger,\tmfa/\lambda)}{\p(u,v)}\mrm + D^2\mr.
\end{equation}
Plugging these back to \eqref{Area248}, we then deduce that
\begin{equation}\label{Area267}
\int_{\mathcal{D}(0,v_0)}\frac{|r^2 \tgs|}{r^3}(u,v) \D u \D v \lesssim D\cdot A[\tp^\dagger,\lambda] + A[\tp^\dagger,\tmfa/\lambda] + D^2.
\end{equation}
By incorporating \eqref{Area246}, \eqref{Area257}, \eqref{Area266} to \eqref{Area245}, we arrive at
\begin{equation}\label{Area268}
\begin{aligned}
A[\nu,\tp^\dagger] &\lesssim  X \cdot TV_{(0,v_0)}[\lambda|_{\Gamma}] + (X + D) \cdot A[\tp^\dagger,\tmfa/\lambda] + D^2 \cdot A[\lambda,\tp^\dagger] + D^2X + D^3.\\
\end{aligned}
\end{equation}

\vspace{5mm}

Summarizing the area estimates above, we have proven the following:
\begin{enumerate}[(1)]
\item \eqref{Area103}: 
\begin{equation}\label{Area269}
\sup_{u \in (0,v_0)}TV_{\{u\}\times(u,v_0)}[\tmfa] \lesssim  \; D\ml 1 + \sup_{u \in (0,v_0)}TV_{\{u\}\times(u,v_0)}[\lambda]\mr  + A[\lambda,\tp],
\end{equation}
\item \eqref{Area115}: 
\begin{equation}\label{Area270}
\sup_{v \in (0,v_0)}TV_{(0,v) \times \{v\}}[\ttmfb] \lesssim \; TV_{(0,v_0)}[\tmfa|_\Gamma] + D \ml 1 + TV_{(0,v_0)}[\lambda|_{\Gamma}] + \sup_{v \in (0,v_0)} TV_{(0,v) \times \{v\}}[\nu]\mr  + A[\nu,\tp],
\end{equation}
\item \eqref{Area129}
\begin{equation}\label{Area271}
TV_{(0,v_0)}[\tmfa|_{\Gamma}] 
\lesssim D\ml 1 + TV_{(0,v_0)}[\lambda|_{\Gamma}] + \sup_{u\in(0,v_0)}TV_{\{u\}\times(u,v_0)}[\lambda]\mr + A[\lambda,\tp],
\end{equation}
\item \eqref{Area175}:
\begin{equation}\label{Area272}
\sup_{u \in (0,v_0)} TV_{ \{u\} \times (u,v_0) } [\log \lambda] \lesssim D^2 + DX + A[\tp^\dagger, \ttmfb/\nu],
\end{equation}
\item \eqref{Area180}:
\begin{equation}\label{Area273}
\sup_{v\in(0,v_0)} TV_{(0,v) \times \{v\}}[\log |\nu|] \lesssim D^2 + DY + TV_{(0,v_0)}[\log\lambda|_{\Gamma}] + A[\tp^\dagger,\tmfa/\lambda],
\end{equation}
\item \eqref{Area182}:
\begin{equation}\label{Area274}
TV_{(0,1)}[\log(\lambda|_{\Gamma})] \lesssim D^2 + A[\tp^\dagger,\ttmfb/\nu],
\end{equation}
\item \eqref{Area190}:
\begin{equation}\label{Area275}
A[\tp^\dagger, \tmfa/\lambda] \lesssim DY\ml 1 + \sup_{u \in (0,v_0)}TV_{\{u\}\times(u,v_0)}[\lambda]\mr + D^2 + (D+Y)\cdot A[\lambda,\tp],
\end{equation}
\item \eqref{Area198}:
\begin{equation}\label{Area276}
\begin{aligned}
A[\tp^\dagger, \ttmfb/\nu] \lesssim&\; X \cdot TV_{(0,v_0)}[\ttmfb|_{\Gamma}] + DX\ml 1 +  TV_{(0,v_0)}[\lambda|_{\Gamma}] + \sup_{v \in (0,v_0)} TV_{(0,v) \times \{v\} }[\nu] \mr \\
&+ D^2 + (D + X)\cdot A[\nu,\tp],
\end{aligned}
\end{equation}
\item \eqref{Area235}:
\begin{equation}\label{Area277}
A[\lambda,\tp] \lesssim (D+Y)A[\tp^\dagger,\ttmfb/\nu] + D^2\cdot A[\nu,\tp] + D^2Y + D^3,
\end{equation}
\item \eqref{Area268}:
\begin{equation}\label{Area278}
A[\nu,\tp] \lesssim  X \cdot TV_{(0,v_0)}[\lambda|_{\Gamma}] + (X + D) \cdot A[\tp^\dagger,\tmfa/\lambda] + D^2 \cdot A[\lambda,\tp] + D^2X + D^3.
\end{equation}
\end{enumerate}
Note that we have used $A[\nu,\tp^\dagger] = A[\nu,\tp]$ and $A[\lambda,\tp^\dagger] = A[\lambda,\tp]$ to simplify and to rewrite \eqref{Area275} - \eqref{Area278}.\\

Recall that by our hypothesis, we have that \eqref{Area75} holds. Plugging (1) to (6) into (7) to (10), keeping the dominant terms, and employing $TV_{(0,v_0)}[\ttmfb|_{\Gamma}] = TV_{(0,v_0)}[\tmfa|_{\Gamma}]$, we obtain
\begin{itemize}
\item[(7'):] \begin{equation}\label{Area279}
\begin{aligned}
A[\tp^\dagger,\tmfa/\lambda] &\lesssim DY\ml 1 + D^2 + DX + A[\tp^\dagger,\ttmfb/\nu]\mr + D^2 + (D+Y)\cdot A[\lambda,\tp] \\
&\lesssim DZ + D^2 + DZ\cdot A[\tp^\dagger,\ttmfb/\nu] + (D+Z)A[\lambda,\tp].
\end{aligned}
\end{equation}
\item[(8'):] \begin{equation}\label{Area280}
\begin{aligned}
A[\tp^\dagger,\ttmfb/\nu] \lesssim&\; X\ml D\ml 1 + TV_{(0,v_0)}[\lambda|_{\Gamma}] + \sup_{v\in(0,v_0)}TV_{(0,v) \times \{v\}}[\lambda]\mr + A[\lambda,\tp] \mr + DX  \\
&+ DX\ml TV_{(0,v_0)}[\lambda|_{\Gamma}]\mr + DX \ml \sup_{v \in (0,v_0)} TV_{(0,v) \times \{v\} }[\nu] \mr + D^2 + (D + X)\cdot A[\nu,\tp], \\
\lesssim&\; DX \ml TV_{(0,v_0)}[\lambda|_{\Gamma}] + \sup_{v\in(0,v_0)}TV_{(0,v) \times \{v\} }[\lambda] +\sup_{v\in(0,v_0)}TV_{(0,v) \times \{v\} }[\nu] \mr \\
&+ X \cdot A[\lambda,\tp] + (D+X) \cdot A[\nu,\tp] + DX + D^2 \\
\lesssim&\; DX \ml A[\tp^\dagger, \ttmfb/\nu] + A[\tp^\dagger,\tmfa/\lambda] \mr + X \cdot A[\lambda,\tp] + (D+X) \cdot A[\nu,\tp] + DX + D^2.
\end{aligned}
\end{equation}
\end{itemize}
For $\varepsilon_0 \in (0,1)$ being sufficiently small, inequality (8') implies
\begin{equation}\label{Area281}
\begin{aligned}
A[\tp^\dagger,\ttmfb/\nu]
\lesssim&\; DZ \cdot A[\tp^\dagger,\tmfa/\lambda] + Z \cdot A[\lambda,\tp] + (D+Z) \cdot A[\nu,\tp] + DZ + D^2.
\end{aligned}
\end{equation}
To further simplify (9') and (10'), we consider the following simplifications to (9) and (10)
\begin{equation}\label{Area282}
A[\lambda,\tp] \lesssim (D+Z)A[\tp^\dagger,\ttmfb/\nu] + D^2\cdot A[\nu,\tp] + D^2Z + D^3
\end{equation} 
and
\begin{equation}\label{Area283}
A[\nu,\tp] \lesssim  Z \cdot A[\tp^\dagger,\ttmfb/\nu] + (D + Z) \cdot A[\tp^\dagger,\tmfa/\lambda] + D^2 \cdot A[\lambda,\tp] + D^2Z + D^3.
\end{equation}
Plugging \eqref{Area283} into \eqref{Area282}, we get
\begin{equation}\label{Area284}
A[\lambda,\tp] \lesssim (D+Z)A[\tp^\dagger,\ttmfb/\lambda] + D^2(D+Z)A[\tp^\dagger,\tmfa/\nu] + D^4 \cdot A[\lambda,\tp] + D^2Z + D^3.
\end{equation} 
For $\varepsilon_0$ sufficiently small, we thus obtain
\begin{itemize}
\item[(9')]: \begin{equation}\label{Area285}
A[\lambda,\tp] \lesssim (D+Z)A[\tp^\dagger,\ttmfb/\nu] + (D+Z)A[\tp^\dagger,\tmfa/\lambda] + D^2Z + D^3.
\end{equation}
\end{itemize}
We then plug \eqref{Area285} to \eqref{Area283} and obtain
\begin{itemize}
\item[(10')]: \begin{equation}\label{Area286}
A[\nu,\tp] \lesssim (D+Z)A[\tp^\dagger,\ttmfb/\nu] + (D+Z)A[\tp^\dagger,\tmfa/\lambda] + D^2Z + D^3.
\end{equation}
\end{itemize}
Substituting (9') and (10') into (7') and (8'), we then derive
\begin{equation}\label{Area287}
\begin{aligned}
A[\tp^\dagger,\tmfa/\lambda] 
&\lesssim DZ + D^2 + (D+Z)^2 A[\tp^\dagger,\ttmfb/\nu] + (D+Z)^2A[\tp^\dagger,\tmfa/\lambda]
\end{aligned}
\end{equation}
and
\begin{equation}\label{Area288}
\begin{aligned}
A[\tp^\dagger,\ttmfb/\nu] 
&\lesssim DZ + D^2 + (D+Z)^2 A[\tp^\dagger,\ttmfb/\nu] + (D+Z)^2A[\tp^\dagger,\tmfa/\lambda].
\end{aligned}
\end{equation}
With a sufficiently small $\varepsilon_0$, \eqref{Area287} and \eqref{Area288} imply
\begin{equation}\label{Area294}
\begin{aligned}
A[\tp^\dagger,\tmfa/\lambda] 
&\lesssim DZ + D^2 + (D+Z)^2 A[\tp^\dagger,\ttmfb/\nu]
\end{aligned}
\end{equation}
and
\begin{equation}\label{Area295}
\begin{aligned}
A[\tp^\dagger,\ttmfb/\nu] 
&\lesssim DZ + D^2  + (D+Z)^2A[\tp^\dagger,\tmfa/\lambda].
\end{aligned}
\end{equation}
By employing \eqref{Area288} into \eqref{Area287}, we thus obtain
\begin{equation}\label{Area289}
\begin{aligned}
A[\tp^\dagger,\tmfa/\lambda] 
&\lesssim DZ + D^2 + (D+Z)^4 A[\tp^\dagger,\tmfa/\lambda].
\end{aligned}
\end{equation}
By shrinking $\varepsilon_0$ if necessary, inequality \eqref{Area289} gives
\begin{equation}\label{Area290}
\begin{aligned}
A[\tp^\dagger,\tmfa/\lambda] 
&\lesssim DZ + D^2.
\end{aligned}
\end{equation}
Together with \eqref{Area288}, this implies that
\begin{equation}\label{Area291}
\begin{aligned}
A[\tp^\dagger,\ttmfb/\nu] 
&\lesssim DZ + D^2.
\end{aligned}
\end{equation}
Substituting \eqref{Area290} and \eqref{Area291} back to \eqref{Area285} and \eqref{Area286}, we hence get
\begin{equation}\label{Area293}
A[\lambda,\tp] \lesssim DZ^2 + D^2Z+ D^3 \quad \text{ and } \quad 
A[\nu,\tp] \lesssim DZ^2 + D^2Z+ D^3.
\end{equation}
Plugging \eqref{Area290}, \eqref{Area291}, \eqref{Area293} into (4), (5), (6), we thus prove
\begin{equation}\label{Area296}
\begin{aligned}
\sup_{u \in (0,v_0)} TV_{ \{u\} \times (u,v_0) } [\log \lambda] &\lesssim DZ + D^2, \\
\sup_{v\in(0,v_0)} TV_{(0,v) \times \{v\}}[\log |\nu|] &\lesssim DZ + D^2, \\
TV_{(0,1)}[\log(\lambda|_{\Gamma})] &\lesssim DZ + D^2. \\
\end{aligned}
\end{equation}
By applying all the above estimates to (1), (2), (3), we now conclude
\begin{equation}\label{Area297}
\begin{aligned}
\sup_{u \in (0,v_0)}TV_{\{u\}\times(u,v_0)}[\tmfa] \lesssim D, \quad \quad 
\sup_{v \in (0,v_0)}TV_{(0,v) \times \{v\}}[\ttmfb] \lesssim D, \quad \text{ and } \quad
TV_{(0,v_0)}[\tmfa|_{\Gamma}]  \lesssim D.
\end{aligned}
\end{equation}
In view of Lemma \ref{AreaLemma1} and \eqref{Area15}, we observe that
\begin{equation}\label{Area298}
\begin{aligned}
X &\lesssim \sup_{u \in (0,v_0)}TV_{\{u\} \times (u,v_0)}\left[ \frac{\tmfa}{\lambda}\right] \\
&\lesssim \sup_{u \in (0,v_0)} TV_{\{u\}\times (u,v_0)}[\lambda] + \sup_{u \in (0,v_0)} TV_{\{u\}\times (u,v_0)}[\tmfa] \lesssim D.
\end{aligned}
\end{equation}
A similar argument yields
\begin{equation}\label{Area299}
\begin{aligned}
Y &\lesssim \sup_{v \in (0,v_0)}TV_{(0,v) \times \{v\} }[\tmfb] + \sup_{v \in (0,v_0)}TV_{(0,v) \times \{v\} }[\nu] \\
&\lesssim \sup_{v \in (0,v_0)}TV_{(0,v) \times \{v\} }[\ttmfb] + \sup_{v \in (0,v_0)}TV_{(0,v) \times \{v\} }[\nu] + D \lesssim D.
\end{aligned}
\end{equation}
Henceforth, by employing \eqref{Area298} and \eqref{Area299}, we arrive at
\begin{equation}\label{Area300}
Z \lesssim D.
\end{equation}
Substituting \eqref{Area300} back into \eqref{Area290}, \eqref{Area291}, \eqref{Area293}, \eqref{Area296}, we thus verify that \eqref{Area301} and \eqref{Area302} holds. \\
\end{proof}

\section{\texorpdfstring{$C^1$ Extension Theorems}{C1 Extension Theorems}}\label{C1 Extension} 

We now establish the $C^1$ extension theorem in this section. We begin by defining
\begin{equation}\label{alpha'}
\alpha' := \frac{1}{\dv r} \dv \left( \frac{\dv (r\phi)}{\dv r} \right) =\frac{\dv \alpha}{\dv r}.
\end{equation}
Recall that the $C^1$-regularity of the solution to the system is equivalent to the boundedness of the term $|\alpha'|$ defined in \eqref{alpha'}. Furthermore, for the charged scalar field system, the \textit{Doppler exponent} $\gamma$ is defined in \eqref{gamma} and given by
$$\gamma(u,v) := \int_{0}^u \frac{-\du r}{r}\frac{\mu - Q^2/r^2}{1-\mu}(u',v) \D u'.$$
Our results for this section are summarized in the two theorems below.
\begin{theorem}\label{SET} ($C^1$ Extension Theorem with Doppler Exponent $\gamma$.) Suppose that we have $C^1$ initial data $\alpha$ along $C_0^+$. Let $v_* > 0$ and assume for every $\hat{v} \in (0,v_*)$, a $C^1$ solution exists on $\mathcal{D}(0,\hat{v})$. If we have
\begin{equation}\label{SET4}
\sup_{\mathcal{D}(0,v_*)} \gamma < + \infty,
\end{equation}
then a $C^1$ solution can be extended to $\mathcal{D}(0,\tilde{v})$ with $\tilde{v} > v_*$.
\end{theorem} 
As mentioned in the introduction, analogous to \cite{BKM}, the above theorem implies that the inability to extend the solution past some $v_* > 0$ necessitates the ``finite-time" blowup of the Doppler exponent, corresponding to an integration of some physical quantity.

A related but independent $C^1$ extension theorem with mass ratio for the charged system, motivated by  \cite{christ2}, is given as follows.
\begin{theorem}\label{FET}
($C^1$ Extension Theorem with Mass Ratio $\mu$.) Suppose that we have $C^1$ initial data $\alpha$ along $C_0^+$. Let $v_* > 0$ and assume for every $\hat v \in (0,v_*)$, a $C^1$ solution exists on $\mathcal{D}(0,\hat v)$. Then, there exists a constant $\varepsilon' > 0$, such that if we have
\begin{equation}\label{FETsup}
\lim_{u \rightarrow v_*} \sup_{\mathcal{D}(u,v_*)} \mu < \varepsilon',
\end{equation}
then a $C^1$ solution can be extended to $\mathcal{D}(0,\tilde{v})$ for some $\tilde{v} > v_*$.
\end{theorem}

In this section, we will first prove Theorem \ref{SET}, followed by Theorem \ref{FET}. For both theorems, we will obtain the corresponding extension theorems via two steps. In the first step, under the assumptions in Theorem \ref{SET} or Theorem \ref{FET}, we show that the solution can be extended as a $C^1$ solution up to $\mathcal{D}(0,v_*).$ In the second step, we then extend such a $C^1$ solution to $\mathcal{D}(0,\tilde{v})$ with some $\tilde{v} > v_*$, which completes the proofs.

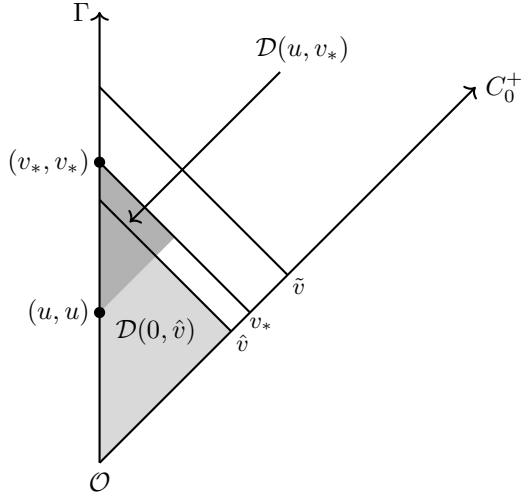
\begin{figure}[ht]
\begin{minipage}[!t]{0.4\textwidth}
	\centering
\begin{tikzpicture}
	\begin{scope}[thick]
    \fill[gray!30] (0,0) to (1.75,1.75) to (0,3.5) to (0,0);
    \fill[gray!60] (0,2) to (1,3) to (0,4) to (0,2);
	\draw[->] (0,0) node[anchor=north]{$\mathcal{O}$} --  (0,6) node[anchor = east]{$\Gamma$};
	\draw[->] (0,0) -- (5,5) node[anchor = west]{$C_0^+$};
	\draw(0,3.5) -- (1.75,1.75);
        \draw(0,4) -- (2,2);
        \draw(0,5) -- (2.5,2.5);
        \node[] at (0.75,1.75) {$\mathcal{D}(0,\hat v)$};
        \node[] at (1.9,1.6) {\small $\hat v$};
        \node[] at (2.15,1.85) {\small $v_*$};
        \node[] at (2.65,2.35) {\small $\tilde{v}$};
        \node[anchor = east] at (0,2) {$(u,u)$};
        \node at (0,2)[circle,fill,inner sep=1.5pt]{};
        \node at (0,4)[circle,fill,inner sep=1.5pt]{};
        \node[anchor = east] at (0,4) {$(v_*,v_*)$};
        \draw[<-](0.4,3.2) -- (2.4,5.2);
        \node[] at (2.7,5.5) {$\mathcal{D}(u,v_*)$};
	\end{scope}
\end{tikzpicture}
\end{minipage}
\caption{Penrose diagram representing the relative positions of $\hat v, v_*, \tilde{v}$ and the regions $\mathcal{D}(0,\hat{v})$ and $\mathcal{D}(u,v_*)$.}
\label{SetupFig4}
\end{figure}

Before we begin with the proof, we list several computations that are commonly used in both proofs. 
\begin{lemma}\label{C1complemma}
For any $v_* > 0$ and $u_1 \in (0,v_*)$, assume that a $C^1$ solution exists in $\mathcal{D}(0,v_*)\setminus \mathcal{D}(u_1,v_*)$. For each  $(u,v) \in \mathcal{D}(0,v_*)\setminus \mathcal{D}(u_1,v_*)$, we have
\begin{equation}\label{FET14}
\begin{aligned}
\du \ml \alpha' e^{\int^u_{u_0} \ii \e A_u(u',v) \D u'} \mr &= e^{\int^u_{u_0} \ii \e A_u(u',v) \D u'} \sum_{i=1}^{11}K_i(u,v) 
\end{aligned}
\end{equation}
with
\begin{equation}\label{K}
\begin{aligned}
K_1(u,v) &= -\frac{8 \pi \e Q \ml \im(\phi^\dagger \dv \phi) \mr (\dv \phi)(\du r)}{(\dv r)^2(1-\mu)}, && K_2(u,v)= \frac{4 \pi Q^2|\dv \phi|^2(\dv \phi)(\du r) }{r(\dv r)^3 (1- \mu)}, \\
K_3(u,v) &= -\frac{4 \pi r|\dv \phi|^2(\dv \phi)(\du r) }{(\dv r)^3 (1- \mu)},&& K_4(u,v)= - \frac{2\du r \mu}{r(1-\mu)}\alpha', \\
K_5(u,v) &= \frac{2Q^2 \du r}{r^3(1-\mu)}\alpha', && K_6(u,v)= \frac{5 Q^2\du r \dv \phi}{r^3 \dv r(1-\mu)}, \\
K_7(u,v) &= -\frac{3\du r \dv \phi \mu}{r \dv r(1-\mu)}, && K_8(u,v) = -\frac{4 \pi \ii \e |\dv \phi|^2 Q\phi \du r}{(1-\mu)(\dv r)^2}, \\
K_9(u,v) &= \frac{4 \ii \pi \e^2 \phi\im(\phi^\dagger \dv \phi) r\du r }{ \dv r (1-\mu)}, && K_{10}(u,v) = -\frac{3 \ii\e  Q \du r \dv \phi}{r(1-\mu)(\dv r)}, \\
K_{11}(u,v) &= -\frac{\ii \e Q\phi \du r}{r^2 (1-\mu)}. && &&&\\
\end{aligned}
\end{equation}
Furthermore, if $\gamma(u,v) \leq G$ and $\frac{1}{1-\mu(u,v)} \leq H$ uniformly in $(u,v) \in \mathcal{D}(0,v_*)\setminus \mathcal{D}(u_1,v_*)$ for some $G,H > 0$, then for each of these $(u,v)$ and any $\chi,\iota \in (0,1)$, we have the following estimates:
\begin{equation}\label{C1Lemma2}
\begin{aligned}
|Q|(u,v) &\leq  C_6'(v_*) \mu^{\frac{1}{2}}(u,v)r(u,v), \\
|Q\phi|(u,v) &\leq C_7''(v_*)\mu^{\oh}(u,v)r^{\frac{3}{2}-\frac{\chi}{2}}(u,v) \leq C_7'(v_*)\mu^{\oh}(u,v)r(u,v), \\
|\phi|^2(u,v) &\leq C_8'(v_*)r^{-\frac{3}{2}}(u,v), \\
\end{aligned}
\end{equation}
with
\begin{equation}\label{SET22}
C_6'(v_*;\chi) := C_6(v_*;\chi)r^{\frac{1}{2}-\frac{\chi}{2}}(0,v_*),
\end{equation}
\begin{equation}\label{SET24}
C_7'(v_*;\chi) := C_7''(v_*;\chi)r^{\frac{1}{2}-\frac{\chi}{2}}(0,v_*),  \quad C_7''(v_*;\chi):= C_7(v_*;\chi)\frac{e^{\oh G}}{(1-\mu_*)^\oh},
\end{equation}
\begin{equation}\label{SET26}
C_8'(v_*;\iota) := \frac{H}{2\pi}r^{\frac{3}{2}}(0,v_*) + 2C_8^2\ml \oh \mr,
\end{equation}
and
\begin{equation}\label{FET130}
C_6(v_*;\chi) := \frac{2 \pi \e}{1 - \mu^*}\ml \frac{\overline{P}(0,v_*)}{\sqrt{4\pi}} + \frac{1}{\sqrt{16\pi^2 \chi \e^2}} \mr r^{\frac{\chi}{2}}(0,v_*) v_*^{\oh},
\end{equation}
\begin{equation}\label{FET243}
C_7(v_*;\chi) = \ml 8 \pi \e^2 C_{3}(0,v_*;\chi,1,\chi) + \frac{C_6^2}{2(2-\chi)} \mr^{\frac{1}{2}},
\end{equation}
\begin{equation}\label{FET300}
C_{8}(\iota):= \overline{P}(0,v_*)r(0,v_*)^{\oh + \frac{\iota}{2}} +  \frac{v_*^\oh r(0,v_*)^{\frac{\iota}{2}}}{(4\pi e^2 \iota (1-\iota))^{\frac{1}{2}}} \ml \frac{\bar{B}(0,v_*)}{(1-\overline{\mu_*})(0,v_*)}\mr^{\frac{1}{2}}.
\end{equation}

On the other hand, if $\mu(u,v) \leq \varepsilon$ for some $\varepsilon < 1$ for each $(u,v) \in \mathcal{D}(0,v_*)\setminus \mathcal{D}(u_1,v_*)$, we can deduce
\begin{equation}\label{C1Lemma1}
\begin{aligned}
|Q|(u,v) &\leq C_6 \varepsilon^{\frac{1}{2}}r^{\frac{3}{2} - \frac{\chi}{2}}(u,v), \\
|Q\phi|(u,v) &\leq C_7 \frac{\varepsilon^\oh}{(1-\varepsilon)^\oh}r^{\frac{3}{2}-\frac{\chi}{2}}(u,v), \\
|\phi|^2 &\leq \frac{1}{2\pi}\frac{\varepsilon}{1-\varepsilon} + 2C_8^2r^{-1-\iota}(u,v).
\end{aligned}
\end{equation}

\end{lemma}

\remark{Note that $|Q|, |Q\phi|, |\phi|^2$ controlled in Lemma \ref{C1complemma} represent additional terms unique to our charged system. The methodologies employed to derive bounds for these terms are therefore new.}

\noindent \textit{Proof of Lemma \ref{C1complemma}.}  Starting from the definition of $\alpha'$ in \eqref{alpha'}, we compute
\begin{equation}\label{FET2}
\begin{aligned}
\du \alpha' =  \du \ml \frac{1}{\dv r} \dv \ml \frac{\dv (r\phi)}{\dv r}\mr\mr &= -\frac{\du\dv r }{(\dv r)^2}\dv \ml \frac{\dv (r \phi)}{\dv r}\mr + \frac{1}{\dv r}\du \dv \ml \frac{\dv (r\phi)}{\dv r}\mr \\
&= -\frac{\du\dv r }{\dv r}\alpha' + \frac{1}{\dv r}\du \dv \ml \frac{\dv (r\phi)}{\dv r}\mr.
\end{aligned}
\end{equation}
Observe that by \eqref{Setup20} we have
\begin{equation}\label{FET3}
 -\frac{\du\dv r }{\dv r}\alpha' = -\frac{\du r}{r}\ml \frac{\mu - Q^2/r^2}{1-\mu}\mr \alpha'.
\end{equation}
Employing \eqref{Setup16} and \eqref{Setup20}, we deduce that
\begin{equation}\label{FET4}
\begin{aligned}
\du \ml \frac{\dv (r\phi)}{\dv r}\mr = \du \ml \phi + \frac{r\dv \phi}{\dv r} \mr &= \du \phi + \frac{\dv r \du (r \dv \phi) - (r \dv \phi)\du\dv r }{(\dv r)^2} \\
&= \frac{\du \phi \dv r + \du r \dv \phi + r \du \dv \phi}{\dv r} - \frac{r \dv \phi}{\dv r} \ml \frac{\du \dv r}{\dv r} \mr \\
&= -\frac{\ii \e \Psi(A)}{\dv r} - \frac{ \dv \phi \du r}{ \dv r} \ml \frac{\mu - Q^2/r^2}{1-\mu}\mr
\end{aligned}
\end{equation}
and this implies
\begin{equation}\label{FET5}
\begin{aligned}
\frac{1}{\dv r}\du \dv \ml \frac{\dv (r\phi)}{\dv r}\mr &= - \frac{1}{\dv r}\dv \ml \frac{\ii \e \Psi(A)}{\dv r}\mr - \frac{1}{\dv r} \dv \ml \frac{ \dv \phi \du r}{ \dv r} \ml \frac{\mu - Q^2/r^2}{1-\mu}\mr \mr.
\end{aligned}
\end{equation}
Plugging \eqref{FET3} and \eqref{FET5} into \eqref{FET2}, we thus obtain
\begin{equation}\label{FET6}
\begin{aligned}
\du \alpha' &= -\frac{\du r}{r}\ml \frac{\mu - Q^2/r^2}{1-\mu}\mr \alpha' - \frac{1}{\dv r}\dv \ml \frac{\ii \e \Psi(A)}{\dv r}\mr - \frac{1}{\dv r} \dv \ml \frac{ \dv \phi \du r}{ \dv r} \ml \frac{\mu - Q^2/r^2}{1-\mu}\mr \mr \\
&=: J_1(u,v) + J_2(u,v) + J_3(u,v).
\end{aligned}
\end{equation}
For $J_3$, we compute the derivative and get
\begin{equation}\label{FET7}
\begin{aligned}
J_3(u,v) &=  - \frac{1}{\dv r} \dv \ml \frac{ \dv \phi \du r}{ \dv r} \ml \frac{\mu - Q^2/r^2}{1-\mu}\mr \mr \\
&= - \frac{\dv \phi \du r}{(\dv r)^2}\dv \ml \frac{\mu - Q^2/r^2}{1-\mu} \mr - \frac{\mu - Q^2/r^2}{(1-\mu)(\dv r)} \dv \ml \frac{\du r}{r}\frac{r \dv \phi }{\dv r}\mr.
\end{aligned}
\end{equation}
For the term $- \dv \ml \frac{\mu - Q^2/r^2}{1-\mu}\mr$, we have the expression
\begin{equation}\label{FET8.1}
\begin{aligned}
- \dv &\ml \frac{\mu - Q^2/r^2}{1-\mu}\mr = \; \dv \ml 1 + \frac{Q^2/r^2 - 1}{1-\mu}\mr = \; \frac{(1-\mu)\dv (Q^2/r^2) + (Q^2/r^2 -1)\dv \mu}{(1-\mu)^2} \\
=&  -\frac{8 \pi \e Q \im(\phi^\dagger \dv \phi)}{1-\mu}- \frac{2 \dv r Q^2}{r^3(1-\mu)} + \frac{4 \pi r|\dv \phi|^2 (Q^2/r^2-1)}{\dv r (1- \mu)} - \frac{\dv r}{r}\frac{\ml \mu - Q^2/r^2\mr(Q^2/r^2 - 1)}{(1-\mu)^2}.
\end{aligned}
\end{equation}
In the last equality above, we have applied \eqref{Setup27} and \eqref{Setup29}. 

\noindent For the term $\dv \ml \frac{\du r}{r}\frac{r \dv \phi}{\dv r}\mr$, with the help of \eqref{alpha'} and \eqref{Setup20}, we obtain
\begin{equation}\label{FET8.2}
\begin{aligned}
\dv \ml \frac{\du r}{r}\frac{r \dv \phi }{\dv r}\mr &= \frac{\du r}{r}\dv \ml  \frac{\dv(r\phi)}{\dv r} - \phi\mr + \ml \frac{r \dv \phi}{\dv r} \mr \dv \ml \frac{\du r}{r} \mr \\
&= \frac{\du r \dv r}{r}\alpha' - \frac{\du r \dv\phi}{r} + \ml \frac{r \dv \phi}{\dv r} \mr\ml \frac{r \du \dv r - \dv r \du r}{r^2} \mr \\
&= \frac{\du r \dv r}{r}\alpha' - \frac{\du r \dv\phi}{r} + \frac{\du r \dv \phi}{r}\ml \frac{\mu - Q^2/r^2}{1-\mu} - 1\mr \\
&= \frac{\du r \dv r}{r}\alpha' - \frac{2\du r \dv\phi}{r}  + \frac{\du r \dv \phi}{r}\ml \frac{\mu - Q^2/r^2}{1-\mu}\mr.\\
\end{aligned}
\end{equation}
Substituting both \eqref{FET8.1} and \eqref{FET8.2} into \eqref{FET7}, $J_3$ can be expanded as
\begin{equation}\label{FET9}
\begin{aligned}
J_3(u,v) =&\; - \frac{\dv \phi \du r}{(\dv r)^2}\dv \ml \frac{\mu - Q^2/r^2}{1-\mu} \mr - \frac{\mu - Q^2/r^2}{(1-\mu)(\dv r)} \dv \ml \frac{\du r}{r}\frac{r \dv \phi }{\dv r}\mr \\
=& \; -\frac{8 \pi \e Q (\im(\phi^\dagger \dv \phi))(\dv \phi)(\du r)}{(\dv r)^2(1-\mu)} - \frac{2 Q^2 (\dv \phi)(\du r)}{r^3 (\dv r)(1-\mu)} + \frac{4 \pi r (Q^2/r^2-1)|\dv \phi|^2(\dv \phi)(\du r) }{(\dv r)^3 (1- \mu)} \\
& \; - \frac{(\dv \phi) (\du r)}{r(\dv r)}\frac{\ml \mu - Q^2/r^2\mr(Q^2/r^2 - 1)}{(1-\mu)^2} -\frac{\du r}{r}\ml \frac{\mu - Q^2/r^2}{1-\mu}\mr \alpha'  \\
&\; + \frac{2 \du r \dv \phi}{r (\dv r)}\ml \frac{\mu - Q^2/r^2}{1-\mu} \mr - \frac{\du r \dv \phi}{r \dv \phi}\ml \frac{\mu - Q^2/r^2}{1-\mu}\mr^2 \\
:= & \sum_{i=1}^7 J_{3,i}(u.v).
\end{aligned}
\end{equation}
Observing that $J_{3,5} \equiv J_1$ and $J_{3,2} + J_{3,4} + J_{3,6} + J_{3,7}$ forms a simpler expression, we can rewrite $J_1 + J_3$ as
\begin{equation}\label{FET10}
\begin{aligned}
J_{1} + J_3 =&\; J_{3,1} + J_{3,3} + 2 J_{3,5} + (J_{3,2} + J_{3,4} + J_{3,6} + J_{3,7}) \\
=&\; -\frac{8 \pi \e Q (\im(\phi^\dagger \dv \phi))(\dv \phi)(\du r)}{(\dv r)^2(1-\mu)}  + \frac{4 \pi r (Q^2/r^2-1)|\dv \phi|^2(\dv \phi)(\du r) }{(\dv r)^3 (1- \mu)} - \frac{2\du r}{r}\ml \frac{\mu - Q^2/r^2}{1-\mu}\mr \alpha' \\
&\; + \frac{\du r \dv \phi}{r \dv r}\ml \frac{5Q^2/r^2 - 
 3 \mu}{1-\mu}\mr.
\end{aligned}
\end{equation}
With the help of \eqref{alpha'}, \eqref{Setup11}, \eqref{Setup17}, we express the terms in $J_2$ as follows:
\begin{equation}\label{FET11}
\begin{aligned}
J_2(u,v) &=  -\frac{\ii \e}{\dv r}  \dv \ml \frac{A_u \dv (r\phi) - \frac{\Omega^2 Q \phi}{4r}}{\dv r}\mr \\
&= -\frac{\ii \e}{\dv r}  \ml \dv A_u \frac{\dv (r\phi)}{\dv r} + A_u \dv \ml \frac{\dv(r\phi)}{\dv r}\mr - \dv \ml \frac{Q\phi(-\du r)}{r(1-\mu)}\mr\mr \\
&= -\ii \e \ml A_u \alpha' + \dv A_u \frac{\dv (r\phi)}{(\dv r)^2} - \frac{1}{\dv r}\dv \ml \frac{-\du r}{1-\mu}\frac{Q\phi}{r}\mr\mr.
\end{aligned}
\end{equation}
Together with \eqref{Setup19}, \eqref{Setup22}, \eqref{Setup30},  we arrive at
\begin{equation}\label{FET12}
\begin{aligned}\frac{J_2(u,v)}{-\ii \e} =& \; A_u \alpha' + \ml \frac{2 Q \du r \dv r}{r^2 (1-\mu)}\mr \frac{r \dv \phi + \phi \dv r}{(\dv r)^2} - \frac{1}{\dv r}\ml \frac{Q\phi}{r}\dv \ml \frac{-\du r}{1-\mu}\mr + \frac{-\du r}{1-\mu}\dv \ml \frac{Q\phi}{r}\mr\mr \\
=&\; A_u \alpha' + \frac{2 Q \du r \dv \phi}{r(1-\mu)(\dv r)} + \frac{2Q\phi \du r}{r^2 (1-\mu)} + \frac{1}{\dv r}\frac{Q\phi}{r}\ml \frac{1-\mu}{-\du r}\mr^{-2} \dv \ml \frac{1-\mu}{-\du r}\mr \\
&\; + \frac{\du r}{\dv r(1-\mu)}\ml \frac{\phi \dv Q}{r} + \frac{Q \dv \phi}{r} - \frac{Q\phi}{r^2} \dv r\mr \\
=&\; A_u \alpha' + \frac{2 Q \du r \dv \phi}{r(1-\mu)(\dv r)} + \frac{2Q\phi \du r}{r^2 (1-\mu)} + \frac{4 \pi |\dv \phi|^2 Q\phi \du r}{(1-\mu)(\dv r)^2} \\
&\; - \frac{4  \pi \e \phi\im(\phi^\dagger \dv \phi) r\du r }{ \dv r (1-\mu)} + \frac{Q \du r \dv \phi}{r(1-\mu)(\dv r)} - \frac{Q \phi \du r}{r^2(1-\mu)} \\
=&\; A_u \alpha' + \frac{3 Q \du r \dv \phi}{r(1-\mu)(\dv r)} + \frac{Q\phi \du r}{r^2 (1-\mu)} + \frac{4 \pi |\dv \phi|^2 Q\phi \du r}{(1-\mu)(\dv r)^2} - \frac{4  \pi \e \phi\im(\phi^\dagger \dv \phi) r\du r }{ \dv r (1-\mu)}.
\end{aligned}
\end{equation}
Combining \eqref{FET10} and \eqref{FET12}, we hence obtain
\begin{equation}\label{FET13}
\begin{aligned}
\du \alpha' =& \; J_1 + J_3 + J_2 \\
=& \; -\frac{8 \pi \e Q (\im(\phi^\dagger \dv \phi))(\dv \phi)(\du r)}{(\dv r)^2(1-\mu)}  + \frac{4 \pi r (Q^2/r^2-1)|\dv \phi|^2(\dv \phi)(\du r) }{(\dv r)^3 (1- \mu)} \\
&\;- \frac{2\du r}{r}\ml \frac{\mu - Q^2/r^2}{1-\mu}\mr \alpha' + \frac{\du r \dv \phi}{r \dv r}\ml \frac{5Q^2/r^2 - 
 3 \mu}{1-\mu}\mr \\
&\;  -\ii \e \ml A_u \alpha' + \frac{3 Q \du r \dv \phi}{r(1-\mu)(\dv r)} + \frac{Q\phi \du r}{r^2 (1-\mu)} + \frac{4 \pi |\dv \phi|^2 Q\phi \du r}{(1-\mu)(\dv r)^2} - \frac{4  \pi \e \phi\im(\phi^\dagger \dv \phi) r\du r }{ \dv r (1-\mu)} \mr. \\
\end{aligned}
\end{equation}
Treating $\ii \e A_u$ as a purely imaginary integrating factor, we thus arrive at \eqref{FET14}.

\vspace{5mm}

We now proceed to obtain estimates for $|Q|, |Q\phi|,$ and $|\phi|^2$. We first prove \eqref{C1Lemma1} under the assumption $\mu \leq \varepsilon < 1$ for all $(u,v) \in \mathcal{D}(0,v_*) \setminus \mathcal{D}(u_1,v_*)$. 

For $|Q|$, from \eqref{Estimates14} of Proposition \ref{PropEstimates}, by setting $\overline{u} = 0$ and recalling the definitions of $\bar{B}, \overline{\mu_*}, \overline{P}$ in \eqref{Bandmu*} and \eqref{overlineP}, we have that the $C_1(0,v;\chi)$ in \eqref{Estimates11} obeys the following upper bound
\begin{equation}\label{FET29}
\begin{aligned}
C_1(0,v;\chi) &\leq \frac{4 \pi \e \ml  \sup_{v' \in [0,v*]} \dv r(0,v') \mr^{\oh}}{(1 - \sup_{v' \in [0,v_*]}\mu(0,v'))^{\frac{1}{2}}}\ml \frac{\overline{P}(0,v_*)}{\sqrt{4\pi}} + \frac{1}{\sqrt{16\pi^2 \chi \e^2}} \mr \\
&\leq \frac{2 \pi \e}{1 - \mu^*}\ml \frac{\overline{P}(0,v_*)}{\sqrt{4\pi}} + \frac{1}{\sqrt{16\pi^2 \chi \e^2}} \mr < + \infty.
\end{aligned}
\end{equation}
Here we use \eqref{Setup33} and \eqref{Setup42} to bound $\sup_{v' \in [0,v_*]} \dv r(0,v')$ and $\sup_{v' \in [0,v_*]}\mu(0,v')^{\frac{1}{2}}$ respectively. We also appeal to the fact that $\overline{P}(0,v_*) = \sup_{v'\in [0,v_*]} |\phi| < + \infty$ since $\phi \in C^1$ along $C_0^+$ and $\phi|_{\Gamma} < +\infty$ from \eqref{Setup35}. Plugging these back into \eqref{Estimates14}, we thus obtain
\begin{equation}\label{FET28}
\begin{aligned}
|Q|(u,v) &\leq \frac{2 \pi \e}{1 - \mu^*}\ml \frac{\overline{P}(0,v_*)}{\sqrt{4\pi}} + \frac{1}{\sqrt{16\pi^2 \chi \e^2}} \mr r^{\frac{\chi}{2}}(0,v_*) v_*^{\oh}\varepsilon^{\frac{1}{2}} r^{\frac{3}{2} - \frac{\chi}{2}}(u,v) = C_6 \varepsilon^{\frac{1}{2}}r^{\frac{3}{2} - \frac{\chi}{2}}(u,v)
\end{aligned}
\end{equation}
with $C_6$ defined in \eqref{FET130}. Note that the constant $C_6$ defined above depends only on $\chi$ and $v_*$. In addition, for brevity, since $\chi$ and $v_*$ were fixed from the beginning, we can then drop the dependency of $C_6$ and subsequent constants on $v_*$ and $\chi$, if there is no danger of confusion.

For $|Q\phi|$, following a similar strategy as in the proof of Proposition \ref{PropEstimates} but with the assumption that $\mu$ is small, we can obtain a better estimate as follows. Since $Q|_\Gamma = 0$ and $|\phi||_{\Gamma} < + \infty$, by \eqref{Setup19}, we have
\begin{equation}\label{FET64}
\begin{aligned}
Q\phi(u,v) &= \int^v_u (\dv Q) \phi + Q \dv \phi (u,v') \mathrm{d}v', 
\end{aligned}
\end{equation}
and
$$|Q\phi|(u,v) \leq \int^v_u  (4 \pi \e r^2 |\phi|^2 + |Q|) |\dv \phi| (u,v') \mathrm{d}v'.$$

With the help of \eqref{Setup25}, this implies that
\begin{equation}\label{FET140}
\begin{aligned}
|Q\phi|(u,v) &\leq \ml \int^v_u \frac{\dv r}{2\pi r^2 (1-\mu)}(4 \pi \e r^2 |\phi|^2 + |Q|)^2(u,v')\D v'\mr^{\frac{1}{2}} \ml \int^v_u \frac{2\pi r^2(1-\mu)|\dv \phi|^2}{\dv r} (u,v') \D v'\mr^{\frac{1}{2}} \\
&\leq \ml \int^v_u \frac{\dv r}{\pi r^2 (1-\mu)}(16 \pi^2 \e^2 r^4 |\phi|^4 + |Q|^2)(u,v')\D v'\mr^{\frac{1}{2}} \frac{1}{\sqrt{2}}\mu(u,v)^{\frac{1}{2}}r(u,v)^{\frac{1}{2}} \\
&\leq \ml 16 \pi \e^2 \int^v_u \frac{\dv r}{(1-\mu)}r^2 |\phi|^4 (u,v')\D v' + \frac{1}{\pi} \int^v_u \frac{\dv r}{r^2 (1-\mu)}|Q|^2(u,v')\D v'\mr^{\frac{1}{2}} \frac{1}{\sqrt{2}}\varepsilon^{\frac{1}{2}}r(u,v)^{\frac{1}{2}}.
\end{aligned}
\end{equation}
Our strategy to estimate the first integral remains unchanged as in Proposition \ref{PropEstimates}. Hence, by setting $\overline{u} = 0$, $\xi_1 = 1$ and $\xi_2 = \chi$ (this is valid since we did not impose $\xi_2 \in (0,1-\chi)$ for this estimate) in \eqref{Estimates36}, we derive the below bound
\begin{equation}\label{FET241}
16 \pi \e^2 \int^v_u \frac{\dv r}{(1-\mu)}r^2 |\phi|^4 (u,v')\D v' \leq 16 \pi \e^2 C_{3}(0,v_*;\chi,1,\chi) r^{2-\chi}(u,v).
\end{equation}
On the other hand, the second integral can be estimated with the help of \eqref{FET28} as follows
\begin{equation}\label{FET242}
\begin{aligned}
\int^v_u \frac{\dv r}{1-\mu} |Q|^2 r^{-2}(u,v') \D v' \leq \int^v_u \frac{\dv r}{1-\varepsilon} C_6^2 \varepsilon r^{1-\chi}(u,v') \D v' = \frac{C_6^2 \varepsilon}{(2-\chi)(1-\varepsilon)}r^{2-\chi}(u,v).
\end{aligned}
\end{equation}
Summarizing \eqref{FET241} and \eqref{FET242}, we obtain
\begin{equation}\label{FET69}
\begin{aligned}
|Q\phi|(u,v) &\leq C_7(v_*) \frac{\varepsilon^{\frac{1}{2}}}{(1-\varepsilon)^{\frac{1}{2}}}r^{\frac
{3}{2}-\frac{\chi}{2}} (u,v) \end{aligned}
\end{equation}
with $C_7$ defined in \eqref{FET243}.

For $|\phi|^2$, analogous to the proof of Proposition \ref{PropEstimates} and \eqref{FET64}, we instead consider $\dv(r\phi)$ in an attempt to obtain pointwise estimates for $|\phi|^2$. With $r\phi|_{\Gamma} = 0$, we have
\begin{equation}\label{FET59}
|r\phi|(u,v) \leq \int^v_u r |\dv \phi|(u,v')\D v' + \int^v_u |\phi| \dv r(u,v') \D v'.
\end{equation}
The first term in \eqref{FET59} can be estimated via 
\begin{equation}\label{FET70}
\begin{aligned}
\int^v_u r|\dv \phi|(u,v') \D v' &\leq \frac{1}{\sqrt{2\pi}}\ml \int^v_u\frac{2\pi r^2(1-\mu)|\dv \phi|^2}{\dv r}(u,v')\D v'\mr^{\frac{1}{2}}\ml \int^v_u \frac{\dv r}{1-\mu}(u,v')\D v'\mr^{\frac{1}{2}} \\
&\leq \frac{1}{\sqrt{4\pi}}\frac{\varepsilon^{\frac{1}{2}}}{(1-\varepsilon)^{\frac{1}{2}}}r(u,v).
\end{aligned}
\end{equation}
For the second term in \eqref{FET59}, by Proposition \ref{PropEstimates}, Lemma \ref{xlogx} and inequality \eqref{Estimates13}, we can exploit its structure and obtain
\begin{equation}\label{FET63}
\begin{aligned}
&\int^v_u|\phi| \dv r(u,v') \D v' \\
\leq &\; \int^v_u (\overline{P}(0,v_*)) \dv r(u,v') \D v' + \int^v_u \frac{1}{(4\pi)^{\frac{1}{2}}} \ml \log \ml\frac{\frac{1-\mu}{\dv r}(u,v')}{\frac{1-\mu}{\dv r}(0,v')} \mr \mr^{\frac{1}{2}}\ml \log \ml \frac{r(0,v')}{r(u,v')} \mr \mr^{\frac{1}{2}}(\dv r)(u,v') \D v' \\
\leq &\;  \overline{P}(0,v_*)r(u,v) + \frac{1}{(4\pi e)^{\frac{1}{2}}} \ml \frac{\bar{B}(0,v_*)}{(1-\overline{\mu_*})(0,v_*)}\mr^{\frac{1}{2}}\int^v_u   (\dv r)^{\frac{1}{2}}(u,v') (1-\mu)^{\frac{1}{2}}(u,v')\ml \log \ml \frac{r(0,v')}{r(u,v')} \mr \mr^{\frac{1}{2}}\D v' \\
\leq &\;  \overline{P}(0,v_*)r(u,v) + \frac{v_*^\oh}{(4\pi e)^{\frac{1}{2}}} \ml \frac{\bar{B}(0,v_*)}{(1-\overline{\mu_*})(0,v_*)}\mr^{\frac{1}{2}} \ml \int^v_u   (\dv r)(u,v') \log \ml \frac{r(0,v')}{r(u,v')} \mr \D v' \mr^{\frac{1}{2}} \\
\leq &\;  \overline{P}(0,v_*) r(u,v) + \frac{v_*^\oh r(0,v_*)^{\frac{\iota}{2}}}{(4\pi e \iota (1-\iota))^{\frac{1}{2}}} \ml \frac{\bar{B}(0,v_*)}{(1-\overline{\mu_*})(0,v_*)}\mr^{\frac{1}{2}} r(u,v)^{\frac{1}{2} -\frac{\iota}{2}} \\
\leq &\; C_8(\iota) r^{\frac{1}{2} - \frac{\iota}{2}}.
\end{aligned}
\end{equation}
In \eqref{FET63} above, the constant $C_8$ is defined in \eqref{FET300} and depends on a parameter $\iota \in (0,1)$.
Combining the two estimates above and using the fact that $(a+b)^2 \leq 2(a^2+b^2)$ for all $a,b \in \mathbb{R}$, we derive
$$r^2|\phi|^2(u,v) \leq \frac{1}{2\pi}\frac{\varepsilon}{1-\varepsilon}r^2(u,v) + 2C_{8}^2(\iota) r^{1-\iota}(u,v),$$
and hence
\begin{equation}\label{FET71}
\begin{aligned}
|\phi|^2(u,v) &\leq \frac{1}{2\pi}\frac{\varepsilon}{1-\varepsilon} + 2C_{8}^2(\iota) r^{-1-\iota}(u,v). \\
\end{aligned}
\end{equation}

\vspace{5mm}

We continue to estimate these terms under the assumption that $\gamma(u,v) \leq G$ and $\frac{1}{1-\mu(u,v)} \leq H$ uniformly for $(u,v) \in \mathcal{D}(0,v_*) \setminus \mathcal{D}(u_1,v_*)$. 

For $|Q|$, using a similar argument as above with a minor exception of not estimating $\mu(u,v)$ by $\varepsilon$, analogous to \eqref{FET28}, we obtain the estimate for $|Q|$ in \eqref{C1Lemma2} with $C_6'(v_*;\chi)$ defined in \eqref{SET22}.

For $|Q\phi|$, by our assumption $\frac{1}{1-\mu} \leq H$, we have the estimate for $|Q\phi|$ in \eqref{C1Lemma2} with $C_7'(v_*;\chi)$ and $C_7''(v_*;\chi)$ defined in \eqref{SET24}.

For $|\phi|^2$, employing a similar argument, by \eqref{FET71}, $\mu < 1$, and \eqref{SET13}, we deduce the estimate for $|\phi|$ in \eqref{C1Lemma2} and \eqref{SET26} with  $C_8$ defined in \eqref{FET300} and $\iota$ chosen to be $\oh$.

\qed 

With Lemma \ref{C1complemma} established, we now start to prove Theorem \ref{SET}.

\noindent \textit{Proof of Theorem \ref{SET}.} By assumption, there is a positive constant $G$ such that
\begin{equation}\label{SET5}
\sup_{\mathcal{D}(0,v_*)}\gamma(u,v) \leq G.
\end{equation}
For each $u \in [0,v_*)$, we define the quantity
\begin{equation}\label{SET6}
B'(u) :=  \sup_{v \in [u,v_*]} |\alpha'|(u,v).
\end{equation}
We will attempt to show that this quantity remains bounded uniformly in $u$, which in turn implies that the constant
\begin{equation}\label{SET7}
B' := \sup_{u \in [0,v_*)} \sup_{v \in [u,v_*]} |\alpha'|(u,v) = \sup_{\mathcal{D}(0,v_*)}|\alpha'| < + \infty
\end{equation}
and the solution extends as a $C^1$ solution up to $\mathcal{D}(0,v_*)$.

\ul{Estimates for $B'(u)$.} To do so, we proceed as follows. By \eqref{FET14} of Lemma \ref{C1complemma}, we have
\begin{equation}\label{SET8}
\alpha'(u,v)e^{\int_0^u \ii \e A_u(u',v) \D u'} = \alpha'(0,v) + \sum_{i=1}^{11} I_i(u,v)
\end{equation}
with $K_i$ defined in \eqref{K} and $I_i := \int_0^u e^{\int^{u'}_{0} \ii \e A_u(u'',v) \D u''} K_i (u',v) \D u'$.
Hence, we obtain
\begin{equation}\label{SET9}
|\alpha'|(u,v) \leq E'(v_*) + \sum_{i=1}^{11} |I_i|(u',v)
\end{equation}
with
\begin{equation}\label{SET10}
E'(v_*) := \sup_{\{0\} \times [0,v_*]}|\alpha'|.
\end{equation}
and $|I_i|$ satisfying the following estimates:
\begin{equation}\label{I_1'}
|I_1|(u,v) \leq  8\pi \e \int^{u}_{0} (-\du r) \ml \frac{|Q||\phi|}{r^2} \mr  \ml \frac{1}{1-\mu}\mr\ml \frac{r|\dv \phi|}{\dv r}\mr^2 (u',v) \mathrm{d}u',
\end{equation}
\begin{equation}\label{I_2'}
|I_2|(u,v) \leq 4\pi \int^{u}_{0} (-\du r) \ml\frac{|Q|^2}{r^4}\mr  \ml \frac{1}{1-\mu}\mr \ml \frac{r|\dv \phi|}{\dv r}\mr^3 (u',v) \mathrm{d}u',
\end{equation}
\begin{equation}\label{I_3'}
|I_3|(u,v) \leq  4\pi \int^{u}_{0} (-\du r)\ml \frac{1}{r^2}\mr\ml \frac{1}{1-\mu}\mr\ml \frac{r|\dv \phi|}{\dv r}\mr^3 (u',v) \mathrm{d}u',
\end{equation}
\begin{equation}\label{I_4'}
|I_4|(u,v) \leq  2 \int^{u}_{0} (-\du r)|\alpha'|  \ml \frac{1}{r}\mr \ml \frac{\mu}{1-\mu}\mr(u',v) \mathrm{d}u' ,
\end{equation}
\begin{equation}\label{I_5'}
|I_5|(u,v) \leq  2 \int^{u}_{0} (-\du r)|\alpha'| \ml\frac{|Q|^2}{r^3}\mr \ml \frac{1}{1-\mu}\mr(u',v) \mathrm{d}u',
\end{equation}
\begin{equation}\label{I_6'}
|I_6|(u,v) \leq  5 \int^{u}_{0}(-\du r)\ml \frac{|Q|^2}{r^4}\mr\ml \frac{1}{1-\mu}\mr\ml \frac{r|\dv \phi|}{\dv r}\mr (u',v) \mathrm{d}u',
\end{equation}
\begin{equation}\label{I_7'}
|I_7|(u,v) \leq 3 \int^{u}_{0} (-\du r)\ml \frac{1}{r^2}\mr\ml \frac{\mu}{1-\mu}\mr\ml \frac{r|\dv \phi|}{\dv r}\mr (u',v) \mathrm{d}u',
\end{equation}
\begin{equation}\label{I_8'}
|I_8|(u,v)  \leq  4\pi \e \int^{u}_{0}(-\du r) \ml\frac{|Q||\phi|}{r^2} \mr \ml \frac{1}{1-\mu}\mr\ml \frac{r|\dv \phi|}{\dv r}\mr^2 (u',v) \mathrm{d}u',
\end{equation}
\begin{equation}\label{I_9'}
|I_9|(u,v) \leq  4 \pi \e^2  \int^{u}_{0} (-\du r) \ml |\phi|^2 \mr  \ml \frac{1}{1-\mu}\mr\ml \frac{r|\dv \phi|}{\dv r}\mr (u',v) \mathrm{d}u',
\end{equation}
\begin{equation}\label{I_10'}
|I_{10}|(u,v) \leq  3 \e  \int^{u}_{0} (-\du r) \ml \frac{|Q|}{r^2}\mr  \ml \frac{1}{1-\mu}\mr\ml \frac{r|\dv \phi|}{\dv r}\mr (u',v) \mathrm{d}u',
\end{equation}
\begin{equation}\label{I_11'}
|I_{11}|(u,v) \leq  \e  \int^{u}_{0} (-\du r) \ml \frac{|Q||\phi|}{r^2}\mr  \ml \frac{1}{1-\mu}\mr (u',v) \mathrm{d}u'.
\end{equation}
Next, we proceed to estimate the following quantities:
$$\frac{1}{1-\mu},\, |\alpha'|,\, \frac{r|\dv \phi|}{\dv r},\, \ml \frac{r|\dv \phi|}{\dv r}\mr^2,\, \ml \frac{r|\dv \phi|}{\dv r}\mr^3,\, |Q|,\, |Q\phi|,\, |\phi|^2 \mbox{ for all } (u,v) \in \mathcal{D}(0,v_*).$$ 
For these terms listed above, we carry out the following strategy:
\begin{itemize}
\item[$\frac{1}{1-\mu}$:] With the help of \eqref{Setup26}, we compute
$$\du\ml \frac{1}{1-\mu(u,v)}\mr = \frac{\du \mu}{(1 - \mu)^2} = \frac{1}{1-\mu(u,v)} \ml -\frac{4\pi r|D_u \phi|^2}{(-\du r)} + \frac{(-\du r)}{r}\frac{\mu - Q^2/r^2}{1-\mu}\mr.$$
This gives
\begin{equation}\label{SET11}
\begin{aligned}
\frac{1}{1-\mu(u,v)} &= \frac{1}{1-\mu(0,v)}\exp\ml {\int^u_{0} \ml -\frac{4\pi r|D_u \phi|^2}{(-\du r)} + \frac{(-\du r)}{r}\frac{\mu - Q^2/r^2}{1-\mu}\mr (u',v) \mathrm{d}u'} \mr.
\end{aligned}
\end{equation}
Since $\du r < 0$, this implies
\begin{equation}\label{SET12}
\begin{aligned}
\frac{1}{1-\mu(u,v)} &\leq \frac{1}{1-\mu(0,v)}\exp\ml {\int^u_{0} \ml  \frac{(-\du r)}{r}\frac{\mu - Q^2/r^2}{1-\mu}\mr (u',v) \mathrm{d}u'} \mr. \\
\end{aligned}
\end{equation}
Employing the definition of $\gamma$ in \eqref{gamma} and $\mu_*$ in \eqref{Setup42}, we obtain
\begin{equation}\label{SET13}
\frac{1}{1-\mu}(u,v) \leq \frac{1}{1-\mu_*}e^G.
\end{equation}
\item[$|\alpha'|$:] By the definition of $B'(u)$ in \eqref{SET6}, for each $(u,v) \in \mathcal{D}(0,v_*)$, we have
\begin{equation}\label{SET14}
|\alpha'|(u,v) \leq B'(u).
\end{equation}
\item[$\frac{r|\dv \phi|}{\dv r}$:] Utilizing the definition of $\alpha$ in \eqref{mathfraka} and $r\phi|_{\Gamma} = 0$, we can rewrite $\phi$ as
\begin{equation}
\phi(u,v) =  \frac{1}{r}\int^v_u (\alpha \dv r)(u,v') \mathrm{d}v',
\end{equation}
which in turn implies that
\begin{equation}
\frac{r \dv \phi}{\dv r}(u,v) = \alpha(u,v) - \frac{1}{r(u,v)}\int^v_u (\alpha \dv r)(u,v') \mathrm{d}v'.
\end{equation}
Using the fact that $r(u,u) = r|_{\Gamma} = 0$ and inequality \eqref{SET14}, we obtain
\begin{equation}\label{SET15}
\begin{aligned}
\mlm\frac{r \dv \phi}{\dv r}\mrm(u,v) &\leq 
\mlm \alpha(u,v) - \frac{1}{r(u,v)} \int^v_u (\alpha \dv r)(u,v') \mathrm{d}v'\mrm \\
&\leq \frac{1}{r(u,v)}\int^v_u \mlm \alpha(u,v) - \alpha(u,v') \mrm \dv r(u,v') \mathrm{d}v' \\
&\leq \frac{1}{r(u,v)}\int^v_u \ml \int^v_{v'} \mlm \dv \alpha (u,v'')\mrm \mathrm{d}v'' \mr \dv r (u,v') \mathrm{d}v' \\
&\leq \frac{1}{r(u,v)}\int^v_u \ml \int^v_{v'} B'(u) (\dv r)(u,v'') \D v'' \mr \D v'= \frac{B'(u)r(u,v)}{2}.
\end{aligned}
\end{equation}
\item[$\ml \frac{r|\dv \phi|}{\dv r}\mr^3:$] For this term, we proceed with a different strategy. Observe that since
\begin{equation}\label{SET16}
\begin{aligned}
\dv  \ml \frac{r \dv \phi}{\dv r}\mr^3 &= \dv \ml \frac{\dv (r\phi)}{\dv r} - \phi \mr^3 = 3\ml \frac{r \dv \phi}{\dv r}\mr^2 \dv \ml \frac{\dv (r\phi)}{\dv r} - \phi\mr \\
&= 3\ml \frac{r \dv \phi}{\dv r}\mr^2 \ml \alpha ' \dv r - \dv \phi\mr = 3 \alpha' \frac{r^2(\dv \phi)^2}{\dv r} - \frac{3 \dv r}{r}\ml \frac{r \dv \phi}{\dv r}\mr^3,
\end{aligned}
\end{equation}
we then derive
\begin{equation}\label{SET17}
\begin{aligned}
\dv  &\ml r^3\ml \frac{r \dv \phi}{\dv r}\mr^3 \mr = 3r^2 \dv r \ml \frac{r \dv \phi}{\dv r}\mr^3 + r^3\dv \ml \frac{r\dv \phi}{\dv r}\mr^3 \\
&= 3r^2 \dv r \ml \frac{r \dv \phi}{\dv r}\mr^3 + r^3\ml3 \alpha' \frac{r^2(\dv \phi)^2}{\dv r} - \frac{3 \dv r}{r}\ml \frac{r \dv \phi}{\dv r}\mr^3 \mr = 3 \alpha' r^3 \ml \frac{r^2(\dv \phi)^2}{\dv r}\mr.
\end{aligned}
\end{equation}
Since $r^3\ml \frac{r \dv \phi}{\dv r}\mr^3 |_{\Gamma} = 0$, by \eqref{Setup25}, \eqref{SET13}, \eqref{SET14}, $\dv r \geq 0, 1-\mu \geq 0, \dv m \geq 0$, we deduce
\begin{equation}\label{SET18}
\begin{aligned}
r^3\ml \frac{r|\dv \phi|}{\dv r} \mr^3& (u,v) \leq \int_u^v \mlm \dv \ml r^3 \ml \frac{r\dv \phi}{\dv r} \mr^3 \mr \mrm(u,v') \D v' \leq 3B'(u)r^3(u,v)\int_u^v \frac{r^2|\dv \phi|^2}{\dv r}(u,v') \D v'\\
&\leq \frac{3e^G}{2\pi(1-\mu_*)} B'(u)r^3(u,v) \int^v_u \dv m(u,v') \D v' \leq \frac{3e^G}{\pi(1-\mu_*)} B'(u)r^4(u,v) \mu(u,v), \\
\end{aligned}
\end{equation}
which in turn implies
\begin{equation}\label{SET19}
\ml \frac{r|\dv \phi|}{\dv r} \mr^3 (u,v) \leq C_3' B'(u)r(u,v) \mu(u,v)
\end{equation}
with the constant $C_1'$ given by
\begin{equation}\label{SET27}
C_3' := \frac{3e^G}{\pi(1-\mu_*)}.
\end{equation}
\item[$\ml \frac{r|\dv\phi|}{\dv r}\mr^2$:] Here we apply an interpolation estimate between \eqref{SET15} and \eqref{SET19} to obtain
\begin{equation}\label{SET20}
\ml\frac{r|\partial_v \phi|}{\partial_v r} \mr^2(u,v) = \sqrt{\ml\frac{r|\partial_v \phi|}{\partial_v r} \mr\ml\frac{r|\partial_v \phi|}{\partial_v r} \mr^3} \leq C_2'B'(u)r(u,v) \mu(u,v)^{\oh}
\end{equation}
with
\begin{equation}\label{SET41}
C_2' := \sqrt{\frac{3e^G}{2\pi(1-\mu_*)}}
\end{equation}
\end{itemize} 
In addition, recall that estimates for $|Q|, |Q\phi|, |\phi|^2$ have been established in Lemma \ref{C1complemma}, with $H = \frac{e^G}{1-\mu_*}$.

\vspace{5mm}

With these estimates, we now obtain an upper bound for each of the terms $|I_i|$ as follows. By plugging \eqref{SET13}, \eqref{SET14}, \eqref{SET15}, \eqref{SET18}, \eqref{SET20}, and \eqref{C1Lemma2}, using $r \leq r(0,v_*)$ whenever necessary, we have the following estimates:
\begin{equation}\label{SET28}
|I_1|(u,v) \leq 8\pi \e C_7' C_2' r(0,v_*) \int_0^u \frac{- \du r}{r} \frac{\mu}{1-\mu}(u',v) B'(u') \D u',  
\end{equation}
\begin{equation}\label{SET29}
|I_2|(u,v) \leq 4\pi (C_6')^2 C_3' \int_0^u \frac{- \du r}{r} \frac{\mu}{1-\mu}(u',v) B'(u') \D u',  
\end{equation}
\begin{equation}\label{SET30}
|I_3|(u,v) \leq 4 \pi C_3' \int_0^u \frac{- \du r}{r} \frac{\mu}{1-\mu}(u',v) B'(u') \D u'.
\end{equation}
The next two terms $|I_4|$ and $|I_5|$ involve $|\alpha'|$. By employing \eqref{SET14}, we proceed and obtain
\begin{equation}\label{SET31}
|I_4|(u,v) \leq 2 \int_0^u \frac{- \du r}{r} \frac{\mu}{1-\mu}(u',v) B'(u') \D u',
\end{equation}
\begin{equation}\label{SET32}
|I_5|(u,v) \leq 2 (C_6')^2 \int_0^u \frac{- \du r}{r} \frac{\mu}{1-\mu}(u',v) B'(u') \D u'.
\end{equation}
For the next three terms $|I_6|, |I_7|, |I_8|$, we have
\begin{equation}\label{SET33}
|I_6|(u,v) \leq \frac{5 (C_6')^2}{2} \int_0^u \frac{- \du r}{r} \frac{\mu}{1-\mu}(u',v) B'(u') \D u',
\end{equation}
\begin{equation}\label{SET34}
|I_7|(u,v) \leq \frac{3}{2} \int_0^u \frac{- \du r}{r} \frac{\mu}{1-\mu}(u',v) B'(u') \D u',
\end{equation}
\begin{equation}\label{SET35}
|I_8|(u,v) \leq 4\pi \e C_7' C_2' r(0,v_*) \int_0^u \frac{- \du r}{r} \frac{\mu}{1-\mu}(u',v) B'(u') \D u'.
\end{equation}
The integrands in the terms $|I_9|$ and $|I_{10}|$ lack a natural upper bound in $\mu$ to construct the integrand of $\gamma$. Instead, we proceed to show that these integrands have additional powers of $r$ as follows:
\begin{equation}\label{SET36}
|I_9|(u,v) \leq \frac{2 \pi \e^2 C_8' e^G}{1-\mu_*} \int_0^u \frac{- \du r}{r^{\oh}}(u',v) B'(u') \D u',
\end{equation}

\begin{equation}\label{SET37}
|I_{10}|(u,v) \leq \frac{3 \e C_6' e^G r(0,v_*)^\oh}{1-\mu_*} \int_0^u \frac{- \du r}{r^{\oh}}(u',v) B'(u') \D u'.
\end{equation}
For the final term $|I_{11}|$, its integrand lacks a natural upper bound in terms of $B'(u)$. We instead treat this term as an additional forcing term in a Gr\"onwall-type argument that we are about to employ. This implies that we have the following estimate for $|I_{11}|$:
\begin{equation}\label{SET38}
|I_{11}|(u,v) \leq \frac{C_7'' \e e^G}{1-\mu_*} \int_0^u \frac{- \du r}{r^{\oh+\frac{\chi}{2}}}(u',v) \D u' \leq \frac{2C_7'' \e e^G r(0,v_*)^{\oh - \frac{\chi}{2}}}{(1-\mu_*)(1-\chi)}.
\end{equation}
Plugging these estimates for $|I_1|$ to $|I_{11}|$ back into \eqref{SET9}, we obtain
\begin{equation}\label{SET39}
|\alpha'|(u,v)\leq E'(v_*) + C_{11} + \int_0^u \ml C_9\frac{- \du r}{r} \frac{\mu}{1-\mu} + C_{10} \frac{- \du r}{r^\oh}\mr(u',v) B'(u') \D u'
\end{equation}
with the constants $C_9$, $C_{10}, C_{11}$ given by
\begin{equation}\label{SET40}
\begin{aligned}
C_9 &:= 12\pi \e C_7' C_2' r(0,v_*) + 4\pi (C_6')^2 C_3' + \frac{9}{2}(C_6')^2 + 4\pi C_3' + \frac{7}{2}, \\
\quad C_{10} &:= \frac{(2\pi \e^2 C_8' + 3\e C_6' r(0,v_*)^\oh)e^G}{1-\mu_*}, \quad C_{11} := \frac{2C_7'' \e e^G r(0,v_*)^{\oh - \frac{\chi}{2}}}{(1-\mu_*)(1-\chi)}. \\
\end{aligned}
\end{equation}
For each $u$, we pick $\tilde{v} \in [u,v_*]$ such that the supremum for \eqref{SET6} is attained. Since \eqref{SET40} is true for all $v \in [u,v_*]$, we thus have
\begin{equation}\label{SET42}
B'(u) \leq E'(v_*) + C_{11} + \int_0^u \ml C_9\frac{- \du r}{r} \frac{\mu}{1-\mu} + C_{10} \frac{- \du r}{r^\oh}\mr(u',\tilde{v}) B'(u') \D u'.
\end{equation}
Applying Gr\"onwall's inequality, we deduce
\begin{equation}\label{SET43}
\begin{aligned}
B'(u) &\leq (E'(v_*)+C_{11})\exp\ml \int_0^u \ml C_9\frac{- \du r}{r} \frac{\mu}{1-\mu} + C_{10} \frac{- \du r}{r^\oh}\mr(u',\tilde{v}) \D u' \mr \\
&\leq (E'(v_*)+C_{11})\exp\ml C_9 \gamma(u,\tilde{v}) + 2C_{10}r(0,\tilde{v})^{\oh} \mr \\
&\leq (E'(v_*)+C_{11})\exp\ml C_9 G + 2C_{10}r(0,v_*)^{\oh} \mr.\\
\end{aligned}
\end{equation}
Observe that the upper bound obtained in \eqref{SET43} is uniform in $u$ and $\tilde{v}$. Hence, from \eqref{SET7}, we deduce that
\begin{equation}\label{SET44}
B' \leq (E'(v_*)+C_{11})\exp\ml C_9 G + 2C_{10}r(0,v_*)^{\oh} \mr < + \infty,
\end{equation}
which in turn implies that the solution extends as a $C^1$ solution up to $\mathcal{D}(0,v_*)$.

\vspace{5mm}

\ul{Extension Beyound Boundary.} Next, we will show that there exists some $\xi' > 0$, such that the solution extends as a $C^1$ solution to $\mathcal{D}(0,v_* +  \xi')$. This is done by estimating the change in $\gamma$ across a rectangular region $\{(u,v): u \in [0,v_*], v\in [v_*,v_*+\xi]\}$ for some $\xi > 0$.

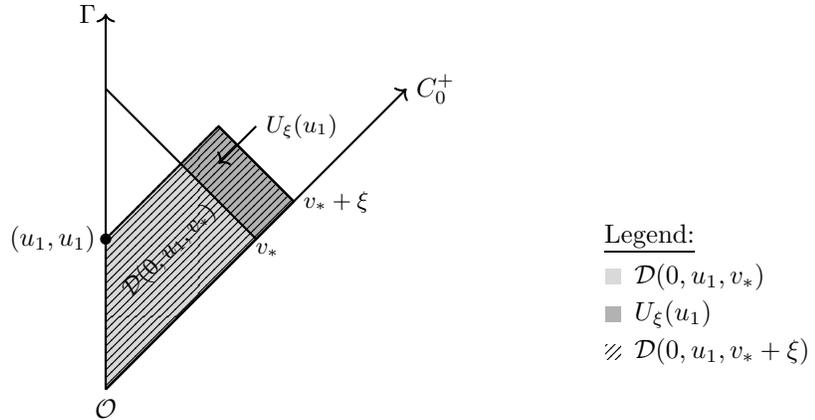
\begin{figure}[htbp!]
\begin{minipage}[!t]{0.4\textwidth}
	\centering
\begin{tikzpicture}[scale=1]
	\begin{scope}[thick]
        \fill[gray!30] (0,0) to (2,2) to (1,3) to (0,2) to (0,1);
        \fill[gray!60] (1,3) to (1.5,3.5) to (2.5,2.5) to (2,2) to (1,3);
	\draw[->] (0,0) node[anchor=north]{$\mathcal{O}$} --  (0,5) node[anchor = east]{$\Gamma$};
	\draw[->] (0,0) -- (4,4) node[anchor = west]{$C_0^+$};
        \draw(0,4) -- (2,2);
        \node[] at (2.15,1.85) {\small $v_*$};
        \node[anchor = east] at (0,2) {$(u_1,u_1)$};
        \node at (0,2)[circle,fill,inner sep=1.5pt]{};
        \node[rotate=45] at (0.85,1.85) {\small $\mathcal{D}(0,u_1,v_*)$};
        \draw[dotted] (0,2) -- (1.5,3.5);
        \draw[] (1.5,3.5) -- (2.5,2.5);
        \node[anchor = west] at (2.5,2.5) {\small $v_* + \xi$};
        \draw[->] (2,3.5) -- (1.5,3);
        \node[anchor = west] at (2,3.5){\small $U_{\xi}(u_1)$};
        \draw[pattern=north east lines, pattern color=black] (0,0) to (2.5,2.5) to (1.5,3.5) to (0,2) to (0,0);
        \node[anchor = west] at (6.5,2) {\ul{Legend:}};
        \node[anchor = west] at (6.9,1.5) {$\mathcal{D}(0,u_1,v_*)$}; 
        \fill[gray!30] (6.65,1.6) to (6.85,1.6) to (6.85,1.4) to (6.65,1.4) to (6.65,1.6);
        \node[anchor = west] at (6.9,1) {$U_\xi(u_1)$};
        \fill[gray!60] (6.65,1.1) to (6.85,1.1) to (6.85,0.9) to (6.65,0.9) to (6.65,1.1);
        \node[anchor = west] at (6.9,0.5) {$\mathcal{D}(0,u_1,v_*+\xi)$}; 
        \path[pattern=north east lines, pattern color=black] (6.65,0.6) to (6.85,0.6) to (6.85,0.4) to (6.65,0.4) to (6.65,0.6);
	\end{scope}
\end{tikzpicture}
\end{minipage}
\caption{Illustration of $U_\xi(u_1)$ and $\mathcal{D}(0,v_1,v_* + \xi)$.}
\label{SETFig1}
\end{figure}
For any given $u_1 \in [0,v_*)$, we consider the corresponding rectangular strip given by
\begin{equation}\label{SET45}
U_{\xi}(u_1) = [0,u_1) \times [v_*,v_*+\xi].
\end{equation}
In the spirit of \eqref{SET6}, we define
\begin{equation}\label{SET46}
\overline{B'_{\xi}}(u_1) := \sup_{U_{\xi}(u_1)} (|\alpha'|).
\end{equation}
We also note that, under the convention in the previous subsections, we can write
\begin{equation}\label{SET47}
\mathcal{D}(0,u_1,v_* + \xi) = \mathcal{D}(0,u_1,v_*) \cup U_{\xi}(u_1).
\end{equation}

Recall from the first part of the proof that the quantity $B'$ in \eqref{SET6} is bounded, independent of $u_1$, and obeys the following estimate derived in \eqref{SET44}:
\begin{equation}\label{SET48}
B'(u) \leq (E'(v_*)+C_{11}(G))\exp\ml C_9(G) G + 2C_{10}(G)r(0,v_*)^{\oh} \mr =: C_{12}(G,v_*).
\end{equation}
Note that in \eqref{SET48}, the constants $C_9, C_{10}, C_{11}$ depends on $G$, and we define $C_{12}$ to be a constant depending only on $v_*$ and $G$. 

Now, fix some $\hat{\xi} > 0$. We then have the following lemma:
\begin{lemma}\label{SETLemma}
For any $u_1 \in (0,v_*)$ and $\xi \in (0,\hat{\xi}]$, if $\gamma$ satisfies
\begin{equation}\label{SET49}
\sup_{\mathcal{D}(0,u_1,v_* + \xi)} \gamma \leq 3G,
\end{equation}
then we have
\begin{equation}\label{SET50}
\overline{B'_{\xi}}(u_1) \leq \max\{  C_{12}(G,v_*), C_{12}(3G,v_*+\hat{\xi}) \} =: B_{2}.
\end{equation}
\end{lemma}
The proof of this lemma follows from the fact that in the region $U_\xi(u_1)$ estimates derived as above now rely on the initial data along $\{0\} \times [v_*,v_* + \xi]$ and rely on a new uniform upper bound for $\gamma$ to be $3G$. The remaining details are similar to those outlined in the estimates for $B'(u)$ as above. Note that the upper bound in \eqref{SET50} is independent of $u_1$ and $\xi$.

With the above lemma, we proceed back on track to prove that the $C^1$ solution extends beyond the domain of dependence $\mathcal{D}(0,u_1,v_*)$. Next, we define
\begin{equation}\label{SET51}
u_2 = \sup\left\{ u_1 \in [0,v_*): \sup_{\mathcal{D}(0,u_1,v_*+\xi)} \gamma \leq 3G\right\}.
\end{equation}
By the definition of $u_2$, we then have either $u_2 = v_*$ or
\begin{equation}\label{SET52}
\sup_{\mathcal{D}(0,u_1,v_* + \xi)} \gamma = 3G
\end{equation}
is attained at some $u_2 \in [u_0,v_*)$. We now show that for $\xi$ sufficiently small, the latter will not happen. Suppose for a contradiction that the latter happens. Then, we set $u_1 = u_2$ as in the hypothesis of Lemma \ref{SETLemma}. Consequently, as described in the proof of Lemma \ref{SETLemma}, in the region $U_\xi(u_2)$, we have
\begin{equation}\label{SET54}
\frac{|\dv \phi|}{\dv r}(u,v) \leq \frac{B_2}{2}.
\end{equation}
Note that we can obtain a different estimate for $|Q|$ with a higher power of $r$ albeit without an upper bound in $\mu$. To do so, we leverage on \eqref{Setup19}, $Q|_\Gamma = 0$ and \eqref{SET54} to obtain
\begin{equation}\label{SET55}
\begin{aligned}
|Q|(u,v) &\leq \int_u^v |\dv Q|(u,v') \D v' \leq 2 \pi \e B_2 C_8' \int_u^v r^{\frac{5}{4}} \dv r(u,v') \D v' \leq \frac{8 \pi \e B_2 \sqrt{C_8'}}{9} r^{\frac{9}{4}}(u,v) \leq C_{13} r^2(u,v),
\end{aligned}
\end{equation}
with
\begin{equation}\label{SET56}
C_{13} := \frac{8 \pi \e B_2 \sqrt{C_8'} r^{\frac{1}{4}}(0,v_*)}{9}.
\end{equation}
Here, the dependence of $C_8'$ on $v_* + \hat\xi$ is dropped for brevity. 

We are now able to estimate $m(u,v)$ as follows. By \eqref{Setup25} and $m|_{\Gamma} = 0$, these imply
\begin{equation}\label{SET57}
m(u,v) \leq \int^v_u \ml \frac{2\pi r^2(1-\mu)|\dv \phi|^2}{\dv r} + \frac{Q^2 \dv r}{2r^2}\mr(u,v') \D v'.
\end{equation}
Employing \eqref{SET54}, we get
\begin{equation}\label{SET58}
\begin{aligned}
\int^v_u \frac{2\pi r^2(1-\mu)|\dv \phi|^2}{\dv r}(u,v') \D v'&\leq 2\pi \int^v_u \frac{r^2|\dv \phi|^2}{(\dv r)^2}(\dv r)(u,v') \D v'\\
&\leq \frac{\pi B_2^2 }{2} \int^v_u r^{2}(\dv r)(u,v') \D v' = \frac{\pi B_2^2}{6}r^3(u,v).
\end{aligned}
\end{equation}
Furthermore, by \eqref{SET55}, we obtain
\begin{equation}\label{SET59}
\begin{aligned}
\int^v_u \frac{Q^2 \dv r}{2r^2}(u,v') \D v' &\leq \frac{C_{13}^2}{2}\int^v_u r^2(\dv r)(u,v') \D v' = \frac{C_{13}^2}{6}r^3(u,v).
\end{aligned}
\end{equation}
This in turn implies that
\begin{equation}\label{SET60}
\mu(u,v) \leq \ml \frac{\pi B_2^2}{3} + \frac{C_{13}^2}{3}\mr r^2(u,v).
\end{equation}
Meanwhile, from \eqref{SET13} but with $G$ replaced by $3G$, we have
\begin{equation}\label{SET61}
\frac{1}{1-\mu}(u,v) \leq \frac{1}{1-\mu_*}e^{3G}.
\end{equation}
With the above estimates, we can now estimate the change in $\gamma$. Recall the expressions for the derivative of $\gamma$ with respect to $v$ from \eqref{Setup98} and \eqref{Setup99} as below:
\begin{equation}\label{SET64}
\dv \gamma(u,v) = \int_0^u (T_1 + T_2 + T_3)(u',v) \D u'.
\end{equation}
By utilizing the estimates on relevant quantities above, together with \eqref{Setup92}, we derive 
\begin{equation}\label{SET65}
\begin{aligned}
\int_0^u |T_1|(u',v) \D u' &\leq \frac{8\pi \e e^{3G} }{1-\mu_*}\int_0^u \frac{|Q||\phi||\dv \phi|}{r}(-\du r)(u',v) \D u' \leq \frac{4\pi \e e^{3G} \sqrt{C_8'} C_{13} B_2 }{1-\mu_*}\int_0^u r^{-\frac{1}{4}}(\dv r)(-\du r)(u',v) \D u' \\
&\leq \frac{2\pi \e e^{4G} \sqrt{C_8'} C_{13} B_2 }{1-\mu_*}\int_0^u r^{-\frac{1}{4}}(-\du r)\D u' \leq \frac{(8)^{\frac{3}{4}}\pi \e e^{4G} \sqrt{C_8'} C_{13} B_2 }{3(1-\mu_*)}(v_* + \hat \xi)^{\frac{3}{4}},\\
\end{aligned}
\end{equation}
\begin{equation}\label{SET66}
\begin{aligned}
\int_0^u |T_2|(u',v) \D u' &\leq \frac{4\pi e^{3G}}{1-\mu_*}\ml 1+\frac{C_{13}^2(v_* + \hat \xi)^2}{4}\mr \int_0^u (-\du r)\frac{|\dv \phi|^2}{\dv r}(u',v) \D u' \\
&\leq \frac{\pi e^{3G}B_2^2}{1-\mu_*}\ml 1+\frac{C_{13}^2(v_* + \hat \xi)^2}{4}\mr \int_0^u (-\du r)(\dv r)(u',v) \D u' \\
&\leq \frac{\pi e^{4G}B_2^2}{1-\mu_*}\ml 1+\frac{C_{13}^2(v_* + \hat \xi)^2}{4}\mr \int_0^u (-\du r)(u',v) \D u' \\
&\leq \frac{\pi e^{4G}B_2^2}{2(1-\mu_*)}\ml 1+\frac{C_{13}^2(v_* + \hat \xi)^2}{4}\mr (v_* + \hat \xi), \\
\end{aligned}
\end{equation}

\begin{equation}\label{SET67}
\begin{aligned}
\int_0^u |T_3|(u',v) \D u' &\leq \frac{2e^{3G}}{1-\mu_*}\int_0^u \frac{(-\du r)(\dv r)}{r^2}\ml |\mu| + \frac{2Q^2}{r^2} \mr(u',v) \D u' \\
&\leq \frac{2e^{3G}}{1-\mu_*}\ml \frac{\pi B_2^2}{3} + \frac{7C_{13}^2}{3}\mr\int_0^u (-\du r)(\dv r)(u',v) \D u' \\
&\leq \frac{e^{4G}}{1-\mu_*}\ml \frac{\pi B_2^2}{3} + \frac{7 C_{13}^2}{3}\mr\int_0^u (-\du r)(u',v) \D u' \\
&\leq \frac{e^{4G}}{2(1-\mu_*)}\ml \frac{\pi B_2^2}{3} + \frac{7 C_{13}^2}{3}\mr (v_* + \hat \xi). \\
\end{aligned}
\end{equation}
Note that \eqref{SET65}, \eqref{SET66}, \eqref{SET67} imply
$\int^u_0 |T_1 + T_2 + T_3|(u',v) \D u' < + \infty$, and hence the expression in \eqref{SET64} is valid by the dominated convergence theorem. Furthermore, for all $(u,v) \in U_\xi(u_2)$, we also obtain 
\begin{equation}\label{SET68}
|\dv \gamma|(u,v) \leq C_{14}
\end{equation}
with
\begin{equation}\label{SET69}
\begin{aligned}
C_{14} :=&\; \frac{(8)^{\frac{3}{4}}\pi \e e^{4G} \sqrt{C_8'} C_{13} B_2 }{3(1-\mu_*)}(v_* + \hat \xi)^{\frac{3}{4}} + \frac{\pi e^{4G}B_2^2}{2(1-\mu_*)}\ml 1+\frac{C_{13}^2(v_* + \hat \xi)^2}{4}\mr(v_* + \hat \xi)\\
&+ \frac{e^{4G}}{2(1-\mu_*)}\ml \frac{\pi B_2^2}{3} + \frac{7 C_{13}^2}{3}\mr(v_* + \hat \xi).
\end{aligned}
\end{equation}
For each $(u,v) \in U_\xi(u_2)$, we thus have
\begin{equation}\label{SET70}
\begin{aligned}
\gamma(u,v) \leq \gamma(u,v_*) + \int_v^{v_*}|\dv \gamma|(u,v') \D v' \leq G + C_{14}\xi.
\end{aligned}
\end{equation}
If we pick 
\begin{equation}\label{SET71}
\xi := \min \left\{  \hat{\xi}, \frac{G}{C_{14}}\right\},
\end{equation}
we then deduce that
$\sup_{\mathcal{D}(0,u_1,v_* + \xi)} \gamma \leq 2G$, contradicting \eqref{SET52}.

Henceforth, for each $u_1 \in (u_0,v_*)$, the hypothesis of Lemma \ref{SETLemma} is verified. This then implies that \eqref{SET50} holds and hence  $\sup_{\mathcal{D}(u_0,u_1,v_*+\Delta v)}|\alpha'|$ remains bounded as $u_1 \rightarrow v_*$. Therefore, we conclude that the solution extends as a $C^1$ solution up to the closure of $\mathcal{D}(u_0,v_*,v_*+\xi)$. In particular, with $C^1$ initial data prescribed along $\{v_*\} \times [v_*,v_*+\xi]$, we show that the solution further extends as a $C^1$ solution to $\mathcal{D}(v_*,v_*+\xi')$ for some $\xi' \in (0,\xi]$. This concludes the proof of Theorem \ref{SET}. \qed

\vspace{5mm}

We continue to prove Theorem \ref{FET}. 

\hspace{-10pt}\textit{Proof of Theorem \ref{FET}.} By the assumption, we can find a $u_0 < v_*$ such that 
\begin{equation}\label{FET1}
\varepsilon  = \sup_{\mathcal{D}(u_0,v_*)} \mu < \varepsilon'.
\end{equation}
Given any $u_1 \in (u_0,v_*)$, we further define 
\begin{equation}\label{D(u0,u1,v*)}
\mathcal{D}(u_0,u_1,v_*) := \mathcal{D}(u_0,v_*)\setminus \overline{\mathcal{D}}(u_1,v_*).
\end{equation}
This region is demonstrated in Figure \ref{SetupFig4}. For any given $\delta \in (0,1)$, we set
\begin{equation}\label{Au1}
A'(u_1) := \sup_{\mathcal{D}(u_0,u_1,v_*)} (r^{\delta}|\alpha'|).
\end{equation}
We will show that this quantity remains bounded as $u_1 \rightarrow v_*$. Based on it, we will further prove that
\begin{equation}\label{Bu1}
B'(u_1) := \sup_{\mathcal{D}(u_0,u_1,v_*)}(|\alpha'|)
\end{equation}
remains bounded as $u_1 \rightarrow v_*$. The fact that the solution extends to $\mathcal{D}(0,v_*)$ as a $C^1$ solution then follows. 
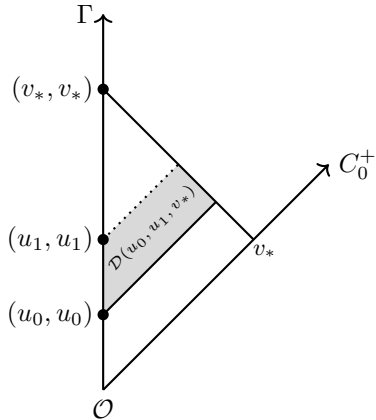
\begin{figure}[htbp!]
\begin{minipage}[!t]{0.4\textwidth}
	\centering
\begin{tikzpicture}
	\begin{scope}[thick]
         \fill[gray!30] (0,1) to (1.5,2.5) to (1,3) to (0,2) to (0,1);
	\draw[->] (0,0) node[anchor=north]{$\mathcal{O}$} --  (0,5) node[anchor = east]{$\Gamma$};
	\draw[->] (0,0) -- (3,3) node[anchor = west]{$C_0^+$};
        \draw(0,4) -- (2,2);
        \node[] at (2.15,1.85) {\small $v_*$};
        \node[anchor = east] at (0,1) {$(u_0,u_0)$};
        \node at (0,1)[circle,fill,inner sep=1.5pt]{};
        \node[anchor = east] at (0,2) {$(u_1,u_1)$};
        \node at (0,2)[circle,fill,inner sep=1.5pt]{};
        \node at (0,4)[circle,fill,inner sep=1.5pt]{};
        \node[anchor = east] at (0,4) {$(v_*,v_*)$};
        \node[rotate=45] at (0.65,2.15) {\tiny $\mathcal{D}(u_0,u_1,v_*)$};
        \draw[dotted] (0,2) -- (1,3);
        \draw[] (0,1) -- (1.5,2.5);
	\end{scope}
\end{tikzpicture}
\end{minipage}
\caption{Illustration of $\mathcal{D}(u_0,v_1,v_*)$.}
\label{FETFig1}
\end{figure}

\ul{Estimates for $A'(u_1)$.}
To bound $A'(u_1)$, we begin by applying Lemma \ref{C1complemma} to obtain \eqref{FET14}. From \eqref{FET14}, for every $(u,v) \in \mathcal{D}(u_0,u_1,v_*)$, via integration along the $u$ direction, we obtain
\begin{equation}\label{FET114}
\begin{aligned}
\alpha'(u,v) e^{\int^u_{u_0}\ii \e A_u (u',v) \mathrm{d}u'} &= \alpha'(u_0,v) +  \sum_{i=1}^{11} I_i(u,u_0,v). \\
\end{aligned}
\end{equation}
We now define
\begin{equation}\label{D'}
D'(u_0) := \sup_{\{u_0\}\times (u_0,v_*)}(r^{\delta}|\alpha'|)
\end{equation}
and observe that this term is finite as we have $\alpha \in C^1(\{u_0\} \times [u_0,v_*])$. 
This implies
\begin{equation}\label{FET15}
\begin{aligned}
|\alpha'(u,v)| &\leq D'(u_0)\ml r(u_0,v) \mr^{-\delta}  + \sum_{i=1}^{11} |I_i|(u,u_0,v)
\end{aligned}
\end{equation}
with 
\begin{equation}\label{FET16‘}
I_i(u,u_0,v) = \int^u_{u_0} e^{\int^{u'}_{u_0}\ii \e A_u (u'',v) \mathrm{d}u''} K_i(u',v) \mathrm{d}u'.
\end{equation} 
Since $A_u$ is real-valued, we also obtain
\begin{equation}\label{FET16}
|I_i|(u,u_0,v) \leq \int^u_{u_0} |K_i|(u',v) \mathrm{d}u'.
\end{equation}
Next, we move to derive an upper bound for each of the $I_i$'s. We begin by observing that
\begin{equation}\label{I_1}
|I_1|(u,u_0,v) \leq  8\pi \e \int^{u}_{u_0} (-\du r) \ml \frac{|Q||\phi|}{r^2} \mr  \ml \frac{1}{1-\mu}\mr\ml \frac{r|\dv \phi|}{\dv r}\mr^2 (u',v) \mathrm{d}u',
\end{equation}
\begin{equation}\label{I_2}
|I_2|(u,u_0,v) \leq 4\pi \int^{u}_{u_0} (-\du r) \ml\frac{|Q|^2}{r^4}\mr  \ml \frac{1}{1-\mu}\mr \ml \frac{r|\dv \phi|}{\dv r}\mr^3 (u',v) \mathrm{d}u',
\end{equation}
\begin{equation}\label{I_3}
|I_3|(u,u_0,v) \leq  4\pi \int^{u}_{u_0} (-\du r)\ml \frac{1}{r^2}\mr\ml \frac{1}{1-\mu}\mr\ml \frac{r|\dv \phi|}{\dv r}\mr^3 (u',v) \mathrm{d}u',
\end{equation}
\begin{equation}\label{I_4}
|I_4|(u,u_0,v) \leq  2 \int^{u}_{u_0} (-\du r)|\alpha'|  \ml \frac{1}{r}\mr \ml \frac{\mu}{1-\mu}\mr(u',v) \mathrm{d}u' ,
\end{equation}
\begin{equation}\label{I_5}
|I_5|(u,u_0,v) \leq  2 \int^{u}_{u_0} (-\du r)|\alpha'| \ml\frac{|Q|^2}{r^3}\mr \ml \frac{1}{1-\mu}\mr(u',v) \mathrm{d}u',
\end{equation}
\begin{equation}\label{I_6}
|I_6|(u,u_0,v) \leq  5 \int^{u}_{u_0}(-\du r)\ml \frac{|Q|^2}{r^4}\mr\ml \frac{1}{1-\mu}\mr\ml \frac{r|\dv \phi|}{\dv r}\mr (u',v) \mathrm{d}u',
\end{equation}
\begin{equation}\label{I_7}
|I_7|(u,u_0,v) \leq 3 \int^{u}_{u_0} (-\du r)\ml \frac{1}{r^2}\mr\ml \frac{\mu}{1-\mu}\mr\ml \frac{r|\dv \phi|}{\dv r}\mr (u',v) \mathrm{d}u',
\end{equation}
\begin{equation}\label{I_8}
|I_8|(u,u_0,v)  \leq  4\pi \e \int^{u}_{u_0}(-\du r) \ml\frac{|Q||\phi|}{r^2} \mr \ml \frac{1}{1-\mu}\mr\ml \frac{r|\dv \phi|}{\dv r}\mr^2 (u',v) \mathrm{d}u',
\end{equation}
\begin{equation}\label{I_9}
|I_9|(u,u_0,v) \leq  4 \pi \e^2  \int^{u}_{u_0} (-\du r) \ml |\phi|^2 \mr  \ml \frac{1}{1-\mu}\mr\ml \frac{r|\dv \phi|}{\dv r}\mr (u',v) \mathrm{d}u',
\end{equation}
\begin{equation}\label{I_10}
|I_{10}|(u,u_0,v) \leq  3 \e  \int^{u}_{u_0} (-\du r) \ml \frac{|Q|}{r^2}\mr  \ml \frac{1}{1-\mu}\mr\ml \frac{r|\dv \phi|}{\dv r}\mr (u',v) \mathrm{d}u',
\end{equation}
\begin{equation}\label{I_11}
|I_{11}|(u,u_0,v) \leq  \e  \int^{u}_{u_0} (-\du r) \ml \frac{|Q||\phi|}{r^2}\mr  \ml \frac{1}{1-\mu}\mr (u',v) \mathrm{d}u'.
\end{equation}
In order for us to prove desired upper bounds for these $11$ terms, we derive useful estimates for 
$$\frac{1}{1-\mu},\, |\alpha'|,\, \frac{r|\dv \phi|}{\dv r},\, \ml \frac{r|\dv \phi|}{\dv r}\mr^2,\, \ml \frac{r|\dv \phi|}{\dv r}\mr^3,\, |Q|,\, |Q\phi|,\, |\phi|^2 \mbox{ for all } (u,v) \in \mathcal{D}(u_0,u_1,v_*).$$ 
For these listed terms, we carry out the following strategy 
\begin{itemize}
\item[$\frac{1}{1-\mu}$:] Using the definition of $\varepsilon$ in \eqref{FET1}, we have
\begin{equation}\label{FET17}
\frac{1}{1-\mu}(u,v) \leq \frac{1}{1-\varepsilon}.
\end{equation}
\item[$|\alpha'|$:] Appealing to the definition of $A(u_1)$ in \eqref{Au1}, we get
\begin{equation}\label{FET21}
|\alpha'|(u,v) \leq A'(u_1)r^{-\delta}(u,v).
\end{equation}
\item[$\frac{r|\dv \phi|}{\dv r}$:] Recall from the proof of Theorem \ref{SET} that we have
\begin{equation}\label{FET19}
\frac{r \dv \phi}{\dv r}(u,v) = \alpha(u,v) - \frac{1}{r(u,v)}\int^v_u (\alpha \dv r)(u,v') \mathrm{d}v'.
\end{equation}
Using the fact that $r(u,u) = r|_{\Gamma} = 0$, inequality \eqref{FET21} and the definition of $A'(u_1)$ in \eqref{Au1}, we obtain
\begin{equation}\label{FET20}
\begin{aligned}
\mlm\frac{r \dv \phi}{\dv r}\mrm(u,v) &\leq 
\mlm \alpha(u,v) - \frac{1}{r(u,v)} \int^v_u (\alpha \dv r)(u,v') \mathrm{d}v'\mrm \\
&= \frac{1}{r(u,v)}\int^v_u \mlm \alpha(u,v) - \alpha(u,v') \mrm \dv r(u,v') \mathrm{d}v' \\
&\leq \frac{1}{r(u,v)}\int^v_u \ml \int^v_{v'} \mlm \dv \alpha (u,v'')\mrm \mathrm{d}v'' \mr \dv r (u,v') \mathrm{d}v' \\
&\leq \frac{1}{r(u,v)}\int^v_u \ml \int^v_{v'} A'(u_1) r^{-\delta}(\dv r)(u,v'')
\mathrm{d}v'' \mr \dv r (u,v') \mathrm{d}v' \\
&= \frac{A'(u_1)}{r(u,v)}\ml \frac{r^{2-\delta}}{(1-\delta)} -\frac{r^{2-\delta}}{(1-\delta)(2-\delta)} \mr(u,v)= \frac{A'(u_1)r^{1-\delta}(u,v)}{2-\delta}.
\end{aligned}
\end{equation}
\item[$\ml \frac{r|\dv \phi|}{\dv r}\mr^3$:] Here we will obtain an upper bound, which is linear in $A'(u_1)$. We recall a lemma established by Christodoulou in \cite{christ2} as follows:
\begin{lemma}\label{L2}
Let $f$ be a $C^1$ function on $[0,r_*]$ with $f(0) = 0$ and let 
\begin{equation}\label{Lemma2,1}
a := \sup_{r\in(0,r_*)} \ml \frac{g}{r} \mr, \quad b := \sup_{r\in(0,r_*)} \ml r^{\delta} \left| \frac{\mathrm{d}f}{\mathrm{d}r} \right| \mr
\end{equation}
with
\begin{equation}\label{Lemma2,2}
    g(r) = \int^r_0 f^2 \mathrm{d}r.
\end{equation}
Then, for any $\delta \in (0,1)$, we have
\begin{equation}\label{Lemma2,3}
    |f(r)|^3 \leq \frac{3ab}{1-\delta}r^{1-\delta}.
\end{equation}
\end{lemma}
\begin{proof}
Observe that
\begin{equation}
f(r)^3 = \int^r_0 \frac{\mathrm{d}}{\mathrm{d}r} f(r)^3 \mathrm{d}r 
= \int^r_0 3f^2 \ml \frac{\mathrm{d}f}{\mathrm{d}r} \mr \mathrm{d}r.
\end{equation}
This implies
\begin{equation}\label{Lemma2Proof}
\begin{aligned}
|f(r)|^3 &\leq \int^r_0 3f^2 \left| \frac{\mathrm{d}f}{\mathrm{d}r} \right| \mathrm{d}r \leq 3b \int^r_0 r^{-\delta}f^2 \mathrm{d}r = 3b\ml \int^r_0 r^{-\delta} \frac{dg}{dr} \mathrm{d}r \mr \\
&\quad= 3b \ml r^{-\delta}g(r) - r^{-\delta}g(r)|_{r=0} + \delta \int^r_0 gr^{-1-\delta} \mathrm{d}r \mr\\
&\quad\leq 3b \ml ar^{1-\delta} + a\delta \int^r_0 r^{-\delta} \mathrm{d}r \mr = \frac{3ab}{1-\delta}r^{1-\delta}.
\end{aligned}
\end{equation}
Here we use the fact that
\begin{equation}
r^{-\delta}g(r)|_{r=0} = \lim_{r\rightarrow 0}\frac{g(r)}{r^{\delta}} = \lim_{r\rightarrow 0}\frac{f(r)^2}{\delta r^{\delta-1}} = \lim_{r\rightarrow 0} \frac{f(r)^2 r^{1-\delta}}{\delta} = 0.
\end{equation}
\end{proof}
We then set $f = \frac{r|\dv \phi|}{\dv r}$ and apply the above lemma. To obtain a bound for $\ml\frac{r|\dv \phi|}{\dv r}\mr^3$, we need to compute $a$ and $b$ as in \eqref{Lemma2,1}. For each fixed $u \in (u_0,u_1)$ with $r_* = r(u,v_*)$, we have
\begin{equation}\label{FET22}
\begin{aligned}
a(u) &= \sup_{r\in(0,r_*)} \ml \frac{1}{r} \int^r_0 f(r')^2 \mathrm{d}r' \mr = \sup_{v\in(u,v_*)} \ml \frac{1}{r} \int^v_u \ml\frac{(r|\partial_v \phi|)^2}{\partial_v r}\mr(u,v') \mathrm{d}v' \mr.
\end{aligned}
\end{equation}
Noting that \eqref{Setup25} implies $
\frac{(r|\partial_v \phi|)^2}{\partial_v r} \leq \frac{\partial_v m}{2\pi (1-\mu)}$, we therefore obtain
\begin{equation}\label{FET24}
\begin{aligned}
a \leq \sup_{v\in(u,v_*)} \ml \frac{1}{2\pi r} \int^v_u \ml\frac{\partial_v m}{1-\mu}\mr \mathrm{d}v' \mr &\leq \frac{1}{4\pi (1-\varepsilon)} \sup_{v\in(u,v_*)}\ml \mu(u,v) - \frac{2m(u,u)}{r(u,v)}\mr \\
&\leq \frac{\varepsilon}{4\pi (1-\varepsilon)}.
\end{aligned}
\end{equation}
For $b$, by applying \eqref{FET21} and \eqref{FET20}, we have
\begin{equation}\label{FET25}
\begin{aligned}
b &= \sup_{r\in(0,r_*)} \ml r^{\delta} \left| \frac{\mathrm{d}f}{\mathrm{d}r} \right| \mr =\sup_{v\in(u,v_*)} \ml r^{\delta} \mlm \frac{\dv f}{\dv r} \mrm \mr = \sup_{v\in(u,v_*)} \ml \frac{r^\delta}{\partial_v r} \mlm \dv \ml \frac{\partial_v (r\phi)}{\partial_v r} - \phi\mr \mrm \mr \\
&\quad\quad\leq \sup_{\mathcal{D}(u_0,u_1,v_*)} \ml \frac{r^\delta}{\partial_v r} \mlm \dv \ml \frac{\partial_v (r\phi)}{\partial_v r} - \phi\mr \mrm \mr \\
&\quad\quad\leq \sup_{\mathcal{D}(u_0,u_1,v_*)} \ml r^{\delta}|\alpha'| \mr +
\sup_{\mathcal{D}(u_0,u_1,v_*)} \ml r^{\delta - 1}\ml \frac{r|\partial_v \phi|}{\partial_v r} \mr\mr\\
&\quad\quad\leq A'(u_1) + \frac{1}{2-\delta}A'(u_1) = \frac{3-\delta}{2-\delta}A'(u_1).
\end{aligned}
\end{equation}
Henceforth, combining the lemma together with \eqref{FET24} and \eqref{FET25}, we obtain
\begin{equation}\label{FET26}
\begin{aligned}
\ml\frac{r|\partial_v \phi|}{\partial_v r}\mr ^3(u,v) \leq \frac{3ab}{1-\delta}r^{1-\delta} \leq \frac{3(3-\delta)}{4\pi(1-\delta)(2-\delta)}\frac{\varepsilon}{1-\varepsilon}A'(u_1)r^{1-\delta}.
\end{aligned}
\end{equation}
\item[$\ml \frac{r|\dv \phi|}{\dv r}\mr^2:$] A direct interpolation between the upper bounds in \eqref{FET20} and \eqref{FET26} implies that
\begin{equation}\label{FET27}
\begin{aligned}
\ml\frac{r|\partial_v \phi|}{\partial_v r} \mr^2(u,v) = \sqrt{\ml\frac{r|\partial_v \phi|}{\partial_v r} \mr\ml\frac{r|\partial_v \phi|}{\partial_v r} \mr^3} \leq \sqrt{\frac{3(3-\delta)}{4\pi(1-\delta)(2-\delta)^2}\frac{\varepsilon}{1-\varepsilon}}A'(u_1)r^{1-\delta}.
\end{aligned}
\end{equation}
\end{itemize}
Note that the estimates for $|Q|, |Q\phi|,$ and $|\phi|^2$ have been established in Lemma \ref{C1complemma}.

\vspace{5mm}

With these estimates, we now obtain an estimate for each of the $|I_i|$'s as follows. For $I_1$, we have
\begin{equation}\label{FET30}
\begin{aligned}
|I_1|(u,u_0,v) &\leq  8\pi \e \int^{u}_{u_0} (-\du r) \ml \frac{|Q||\phi|}{r^2} \mr  \ml \frac{1}{1-\mu}\mr\ml \frac{r|\dv \phi|}{\dv r}\mr^2 (u',v) \mathrm{d}u' \\
&\leq 8 \pi \e C_{7} A'(u_1) \sqrt{\frac{3(3-\delta)}{4\pi(1-\delta)(2-\delta)^2}} \frac{\varepsilon}{(1-\varepsilon)^{2}}  \int^{u}_{u_0}  (-\du r) r^{\frac{1}{2}-\frac{\chi}{2}-\delta}  (u',v) \mathrm{d}u' \\
&\leq 8 \pi \e C_{7} A'(u_1) \sqrt{\frac{3(3-\delta)}{4\pi(1-\delta)(2-\delta)^2}} \frac{\varepsilon}{(1-\varepsilon)^{2}}  \frac{r^{\frac{3}{2}-\frac{\chi}{2}}(u_0,v)}{\frac{3}{2}-\frac{\chi}{2}} \cdot \frac{1}{r^\delta(u,v)} \\
&\leq \frac{16 \pi \e C_{7}r^{\frac{3}{2}-\frac{\chi}{2}}(0,v_*)}{3-\chi}  \sqrt{\frac{3(3-\delta)}{4\pi(1-\delta)(2-\delta)^2}} \frac{\varepsilon}{(1-\varepsilon)^{2}}   \frac{A'(u_1)}{r^{\delta}(u,v)}.\\
\end{aligned}
\end{equation}
Here we use the fact $\chi\in(0,1)$, $\delta\in(0,1)$ and $\frac12-\frac{\chi}{2}-\delta>-1$. 

\noindent Following a similar idea, the next term can be bounded as follows:
\begin{equation}\label{FET31}
\begin{aligned}
|I_2|(u,u_0,v) &\leq 4\pi \int^{u}_{u_0} (-\du r) \ml\frac{|Q|^2}{r^4}\mr  \ml \frac{1}{1-\mu}\mr \ml \frac{r|\dv \phi|}{\dv r}\mr^3 (u',v) \mathrm{d}u' \\
&\leq  C_6^2  A'(u_1) \frac{3(3-\delta)}{(1-\delta)(2-\delta)}\frac{\varepsilon^2}{(1-\varepsilon)^2}\int^{u}_{u_0} (-\du r) r^{-\delta - \chi} (u',v) \mathrm{d}u' \\
&\leq  \frac{C_6^2r(0,v_*)^{1-\chi}}{1-\chi} \frac{3(3-\delta)}{(1-\delta)(2-\delta)}\frac{\varepsilon^2}{(1-\varepsilon)^2} \frac{A'(u_1)}{r^{\delta}(u,v)}.
\end{aligned}
\end{equation}
We proceed to derive
\begin{equation}\label{FET33}
\begin{aligned}
|I_3|(u,u_0,v) &\leq 4\pi \int^{u}_{u_0} (-\du r)\ml \frac{1}{r^2}\mr\ml \frac{1}{1-\mu}\mr\ml \frac{r|\dv \phi|}{\dv r}\mr^3 (u',v) \mathrm{d}u'\\
&\leq  A'(u_1) \frac{3(3-\delta)}{(1-\delta)(2-\delta)}\frac{\varepsilon}{(1-\varepsilon)^2}\int^{u}_{u_0} (-\du r)r^{-1-\delta} (u',v) \mathrm{d}u'  \\
&\leq \frac{3(3-\delta)}{\delta(1-\delta)(2-\delta)}\frac{\varepsilon}{(1-\varepsilon)^2} \frac{A'(u_1)}{r^{\delta}(u,v)}
\end{aligned}
\end{equation}
and
\begin{equation}\label{FET32}
\begin{aligned}
|I_4|(u,u_0,v) &\leq   2 \int^{u}_{u_0} (-\du r)|\alpha'|  \ml \frac{1}{r}\mr \ml \frac{\mu}{1-\mu}\mr(u',v) \mathrm{d}u'  \\
&\leq 2 A'(u_1) \frac{\varepsilon}{1-\varepsilon}  \int^{u}_{u_0} (-\du r) r^{-1-\delta} (u',v) \mathrm{d}u' \leq \frac{2}{\delta} \frac{\varepsilon}{1-\varepsilon}  \frac{A'(u_1)}{r^{\delta}(u,v)}.
\end{aligned}
\end{equation} 
The subsequent three terms obey the following bounds
\begin{equation}\label{FET34}
\begin{aligned}
|I_5|(u,u_0,v) &\leq   2 \int^{u}_{u_0} (-\du r)|\alpha'| \ml\frac{|Q|^2}{r^3}\mr \ml \frac{1}{1-\mu}\mr(u',v) \mathrm{d}u' \\
&\leq 2 C^2_6 A'(u_1) \frac{\varepsilon}{1-\varepsilon} \int^{u}_{u_0} (-\du r) r^{-\delta -\chi} (u',v) \mathrm{d}u' \leq  \frac{2 C^2_6 r^{1-\chi}(0,v_*)}{1-\chi} \frac{\varepsilon}{1-\varepsilon}\frac{A'(u_1)}{r^{\delta}(u,v)},
\end{aligned}
\end{equation}
\begin{equation}\label{FET35}
\begin{aligned}
|I_6|(u,u_0,v) &\leq   5 \int^{u}_{u_0}(-\du r)\ml \frac{|Q|^2}{r^4}\mr\ml \frac{1}{1-\mu}\mr\ml \frac{r|\dv \phi|}{\dv r}\mr (u',v) \mathrm{d}u' \\
&\leq  5 C_6^2 A'(u_1) \frac{1}{2-\delta}  \frac{\varepsilon}{1-\varepsilon} \int^{u}_{u_0} (-\du r) r^{-\delta - \chi}(u',v) \mathrm{d}u' \\
&\leq \frac{5 C_6^2 r^{1-\chi}(0,v_*)}{(1-\chi)(2-\delta)} \frac{\varepsilon}{1-\varepsilon} \frac{A'(u_1)}{r^{\delta}(u,v)},
\end{aligned}
\end{equation}
\begin{equation}\label{FET36}
\begin{aligned}
|I_7|(u,u_0,v) &\leq 3 \int^{u}_{u_0} (-\du r)\ml \frac{1}{r^2}\mr\ml \frac{\mu}{1-\mu}\mr\ml \frac{r|\dv \phi|}{\dv r}\mr (u',v) \mathrm{d}u' \\
&\leq  \frac{3A'(u_1)}{2-\delta} \frac{\varepsilon}{1-\varepsilon} \int^{u}_{u_0} (-\du r) r^{-1 -\delta}(u',v) \mathrm{d}u'\leq \frac{3}{\delta(2-\delta)} \frac{\varepsilon}{1-\varepsilon} \frac{A'(u_1)}{r^{\delta}(u,v)}.
\end{aligned}
\end{equation}
For the next term $I_8$, we note that its upper bound is exactly half of that of $I_1$. This implies 
\begin{equation}\label{FET38}
\begin{aligned}
|I_8|(u,u_0,v)  &\leq  4\pi \e \int^{u}_{u_0}(-\du r) \ml\frac{|Q||\phi|}{r^2} \mr \ml \frac{1}{1-\mu}\mr\ml \frac{r|\dv \phi|}{\dv r}\mr^2 (u',v) \mathrm{d}u', \\
&\leq \frac{8 \pi \e C_{7}r^{\frac{3}{2}-\frac{\chi}{2}}(0,v_*)}{3-\chi}  \sqrt{\frac{3(3-\delta)}{4\pi(1-\delta)(2-\delta)^2}} \frac{\varepsilon}{(1-\varepsilon)^{2}}   \frac{A'(u_1)}{r^{\delta}(u,v)}. \\
\end{aligned}
\end{equation}
To estimate $|I_9|$, we employ \eqref{FET71} and get
\begin{equation}\label{FET40}
\begin{aligned}
|I_9|(u,u_0,v) &\leq  4 \pi \e^2  \int^{u}_{u_0} (-\du r) \ml |\phi|^2 \mr  \ml \frac{1}{1-\mu}\mr\ml \frac{r|\dv \phi|}{\dv r}\mr (u',v) \mathrm{d}u' \\
&\leq 4 \pi \e ^2  \frac{A'(u_1)}{2-\delta} \frac{1}{1-\varepsilon} \int^u_{u_0} (-\du r)r^{1-\delta}|\phi|^2(u',v)\mathrm{d}u
'\\
&\leq 4 \pi \e ^2  \frac{A'(u_1)}{2-\delta} \frac{1}{1-\varepsilon} \int^u_{u_0} (-\du r)r^{1-\delta} \ml \frac{1}{2\pi}\frac{\varepsilon
}{1-\varepsilon} + 2C_{8}^2(\iota)r^{-1-\iota} \mr (u',v)\mathrm{d}u
'\\
&\leq \frac{4 \pi \e ^2 }{2-\delta} \frac{1}{1-\varepsilon} \frac{A'(u_1)}{r^{\delta}(u,v)}  \ml \frac{r^2(0,v_*)}{4\pi}\frac{\varepsilon
}{1-\varepsilon} + \frac{2C_{8}^2(\iota) r^{1-\iota}(u_0,v_*)}{1-\iota}\mr.
\end{aligned}
\end{equation}
The upper bounds for the last two terms $I_{10}$ and $I_{11}$ can be derived as
\begin{equation}\label{FET42}
\begin{aligned}
|I_{10}|(u,u_0,v) &\leq  3 \e  \int^{u}_{u_0} (-\du r) \ml \frac{|Q|}{r^2}\mr  \ml \frac{1}{1-\mu}\mr\ml \frac{r|\dv \phi|}{\dv r}\mr (u',v) \mathrm{d}u' \\
&\leq 3 \e C_6 A'(u_1) \frac{1}{2-\delta}\frac{\varepsilon^{\frac{1}{2}}}{1-\varepsilon}\int^{u}_{u_0} (-\du r) r^{\frac{1}{2} -\frac{\chi}{2}-\delta}(u',v) \mathrm{d}u' \\
&\leq \frac{6 \e C_6 r^{\frac{3}{2}-\frac{\chi}{2}}(u_0,v_*)}{(3 -  \chi)(2-\delta)} \frac{\varepsilon^{\frac{1}{2}}}{1-\varepsilon} \frac{A'(u_1)}{r^{\delta}(u,v)}
\end{aligned}
\end{equation}
and
\begin{equation}\label{FET43}
\begin{aligned}
|I_{11}|(u,u_0,v) &\leq \e  \int^{u}_{u_0} (-\du r) \ml \frac{|Q||\phi|}{r^2}\mr  \ml \frac{1}{1-\mu}\mr (u',v) \mathrm{d}u'\\
&\leq \e C_{7} \frac{\varepsilon^{\frac{1}{2}}}{(1-\varepsilon)^{\frac{3}{2}}} \int^u_{u_0} (-\du r)r^{-\oh-\frac{\chi}{2}}(u',v)\mathrm{d}u' \\
&\leq \e C_{7} r^{\oh + \delta - \frac{\chi}{2}}(0,v_*) \frac{\varepsilon^{\frac{1}{2}}}{(1-\varepsilon)^{\frac{3}{2}}} \int^u_{u_0} (-\du r)r^{-1-\delta}(u',v)\mathrm{d}u' \\
&\leq \frac{2\e C_{7}r^{\oh + \delta - \frac{\chi}{2}}(0,v_*)}{\delta} \frac{\varepsilon^{\frac{1}{2}}}{(1-\varepsilon)^{\frac{3}{2}}} \frac{1}{r^{\delta}(u,v)}.
\end{aligned}
\end{equation}
Note that in \eqref{FET43}, we have restricted our choice for $\chi$ to $\chi < 2\delta$.

Back to \eqref{FET15}, multiplying $r^{\delta}(u,v)$ on both sides and taking the supremum over $\mathcal{D}(u_0,u_1,v_*)$, we therefore obtain
\begin{equation}\label{FET45}
A'(u_1) \leq D'(u_0) + \sup_{\mathcal{D}(u_0,u_1,v_*)} \ml \sum_{i=1}^{11} |I_i|r^{\delta}(u,v)\mr
\end{equation}
with $D'(u_0)$ defined in \eqref{D'}. For convenience, we impose the requirement
\begin{equation}\label{FET47}
 \varepsilon<\frac12.   
\end{equation}
With this restriction, we employ obvious bounds (such as $0 < \iota,\delta,\chi < 1$ and $0 < \varepsilon < \frac{1}{2}$) to simplify the constants appearing in the upper bounds for $|I_1|$ to $|I_{11}|$. We will also apply $1 \leq \frac{1}{(1-\delta)^{\frac{1}{2}}} \leq \frac{1}{1-\delta}\leq \frac{1}{\delta(1-\delta)}$, $\frac{1}{\delta} \leq \frac{1}{\delta(1-\delta)}$, $\frac{1}{(1-\varepsilon)^b} \leq 4$ for any $b \in [1,2]$. Furthermore, for the ease of impending computation, we will lower all powers of $\varepsilon$ that are greater than $\oh$ in the numerator to $\varepsilon^{\frac{1}{2}}$ since $\varepsilon < \frac{1}{2}$. For brevity, we will also define the unified coupling constant $E$ as
\begin{equation}\label{FET49}
E := \max\{1,\e,\e^2\}.
\end{equation} 
With the above, we can simplify the constants in our estimates for $|I_i|$. We start by analyzing $|I_1|$ and get
\begin{equation}\label{FET50}
\begin{aligned}
|I_1|(u,u_0,v) &\leq \frac{16 \pi \e C_{7}r^{\frac{3}{2}-\frac{\chi}{2}}(0,v_*)}{3-\chi}  \sqrt{\frac{3(3-\delta)}{4\pi(1-\delta)(2-\delta)^2}} \frac{\varepsilon}{(1-\varepsilon)^{2}}   \frac{A'(u_1)}{r^{\delta}(u,v)} \leq \frac{C_{9,1}\varepsilon^{\frac{1}{2}}}{\delta(1 - \delta)} \frac{A'(u_1)}{r(u,v)^{\delta}}
\end{aligned}
\end{equation}
with constant $C_{9,1} > 0$ depending on $v_*$, $\chi$, $E$. Explicitly for $|I_1|$, this is given by
\begin{equation}
C_{9,1} = 16 \sqrt{6\pi} E C_{7} r^{\frac{3}{2}-\frac{\chi}{2}}(0,v_*).
\end{equation}
One can repeat a similar argument for $|I_i|$ with $i = 2,3,4,5,6,7,8,10$ and get constants $C_{9,i} > 0$ depending on $v_*$, $\chi$, and $E$.\footnote{Here, observe that $C_7$ also depends on $v_*, \chi, E$ by its explicit expression in \eqref{FET243}.} Note that each of these constants does not depend on $u_0$. For $i = 11$, we employ a similar bound without $A'(u_1)$ and obtain
\begin{equation}\label{FET51}
\begin{aligned}
|I_{11}|(u,u_0,v) &\leq \frac{2\e C_{7}r^{\delta - \frac{\chi}{2}}(0,v_*)}{\delta} \frac{\varepsilon^{\frac{1}{2}}}{1-\varepsilon} \frac{1}{r^{\delta}(u,v)} \leq \frac{C_{10,11}\varepsilon^{\frac{1}{2}}}{\delta(1 - \delta)} \frac{1}{r(u,v)^{\delta}}
\end{aligned}
\end{equation}
with constant $C_{10,11} > 0$ depending on $v_*$, $\chi$, $E$. For $i = 9$, there is a term without $\varepsilon$. Thus, we use the following upper bound 
\begin{equation}\label{FET52}
\begin{aligned}
|I_9|(u,u_0,v) 
&\leq \ml C_{9,9} \frac{\varepsilon^{\frac{1}{2}}}{\delta(1 - \delta)} + C_{11,9}(u_0;\iota) \mr \frac{A'(u_1)}{r(u,v)^{\delta}}
\end{aligned}
\end{equation}
with $C_{9,9} > 0$ and $C_{11,9}(u_0) > 0$. Note that $C_{9,9}$ depends only on $v_*, \chi, E$, while $C_{11,9}$ depends additionally on $\iota$, $u_0$ and can be expressed as
\begin{equation}\label{FET72}
C_{11,9}(u_0;\iota) = \frac{32 \pi E C_{8}^2(\iota)}{1-\iota}r^{1-\iota}(u_0,v_*).
\end{equation}
Combining these estimates together with \eqref{FET45}, we derive
\begin{equation}\label{FET53}
A'(u_1) \leq D'(u_0) + \ml \frac{\sum_{i=1}^{10}C_{9,i}\varepsilon^{\frac{1}{2}} }{\delta(1-\delta)} + C_{11,9}(u_0;\iota)\mr A'(u_1) + \frac{ C_{10,11}}{\delta(1-\delta)}\varepsilon^{\frac{1}{2}}.
\end{equation}
Next, we demand the following constraints on the constants
\begin{equation}\label{FET54}
\frac{\sum_{i=1}^{10}C_{9,i}\varepsilon^{\frac{1}{2}} }{\delta(1-\delta)} < \frac{1}{4},
\end{equation}
and
\begin{equation}\label{FET73}
C_{11,9}(u_0;\iota) < \frac{1}{8}.\footnote{The choice of $\frac{1}{8}$ here is for simplifying parallel arguments in the next subsection for extension beyond $\mathcal{D}(0,v_*)$.}
\end{equation}
Note that \eqref{FET54} can be rewritten as
\begin{equation}\label{FET74}
\varepsilon < \ml \frac{\delta(1-\delta)}{4\sum_{i=1}^{10}C_{9,i}}\mr^2.
\end{equation}
For \eqref{FET73}, this can be attained by picking $u_0$ sufficiently close to $v_*$ with using the explicit expression of $C_{11,9}(u_0)$ in \eqref{FET72}, in which for a fixed $\iota \in (0,1)$, the term $r^{1-\iota}(u_0,v_*)$ serves as a small factor which goes to $0$ as $u_0 \rightarrow v_*$. Henceforth, \eqref{FET54} implies
\begin{equation}\label{FET56}
\frac{C_{10,11}}{\delta(1-\delta)}\varepsilon^{\frac{1}{2}} < \frac{C_{10,11}}{4\sum_{i=1}^{10}C_{9,i}} =: C_{12} < + \infty.
\end{equation}
From \eqref{FET47} and \eqref{FET74}, if we pick
\begin{equation}\label{FET57}
\varepsilon' = \min\left\{\frac{1}{2},\ml \frac{\delta(1-\delta)}{4\sum_{i=1}^{10}C_{9,i}}\mr^2 \right\},
\end{equation}
then for all $\varepsilon < \varepsilon'$, we have that \eqref{FET47} and \eqref{FET74} hold. Plugging \eqref{FET54}, \eqref{FET73}, and \eqref{FET56} back into \eqref{FET53}, we deduce that
\begin{equation}\label{FET58}
A'(u_1) \leq 2(D'(u_0) + C_{12}) =: C_{13}.
\end{equation}
In other words, $A'(u_1)$ is bounded by a constant depending only on $\delta$, $\iota$, $\chi$. Note that this bound is valid as long as we have $\chi$, $\delta$, $\iota$ satisfying $0< \chi < \delta < 1$ and $0 < \iota < 1$. 

\ul{Estimates for $B'(u_1)$.}
Next, we will further show that $B'(u_1)$ as defined in \eqref{Bu1} remains bounded. This is important for concluding the extension of the $C^1$ solution up to the boundary as this in turn implies that the solution in $\mathcal{D}(u_0,u_1,v_*)$ is indeed $C^1$. We start off by observing from \eqref{FET20} and \eqref{FET58} that
\begin{equation}\label{FET107}
\begin{aligned}
\frac{r |\dv \phi|}{\dv r}(u,v) &\leq C_{14}(\delta) r^{1-\delta}(u,v) \mbox{ with } C_{14}(\delta) := \frac{C_{13}}{1-\delta}.
\end{aligned}
\end{equation}
By \eqref{Setup25}, from $m|_{\Gamma} = 0$, this implies
\begin{equation}\label{FET108}
m(u,v) \leq \int^v_u \ml \frac{2\pi r^2(1-\mu)|\dv \phi|^2}{\dv r} + \frac{Q^2 \dv r}{2r^2} \mr(u,v') \D v'.
\end{equation}
The first term in \eqref{FET108} can be estimated with the help of \eqref{FET107} as
\begin{equation}\label{FET109}
\begin{aligned}
\int^v_u \frac{2\pi r^2(1-\mu)|\dv \phi|^2}{\dv r}(u,v') \D v'&\leq 2\pi \int^v_u \frac{r^2|\dv \phi|^2}{(\dv r)^2}(\dv r)(u,v') \D v'\\
&\leq 2\pi \int^v_u C_{14}^2 r^{2-2\delta}(\dv r)(u,v') \D v' 
= \frac{2\pi C_{14}^2}{3-2\delta}r^{3-2\delta}(u,v).
\end{aligned}
\end{equation}
For the second term, from \eqref{FET28}, we have
\begin{equation}\label{FET110}
\begin{aligned}
\int^v_u \frac{Q^2 \dv r}{2r^2}(u,v') \D v' \leq \frac{C_{6}^2}{2}\int^v_u r^{1-\chi}(\dv r)(u,v') \D v' = \frac{C_{6}^2}{2(2-\chi)}r^{2-\chi}(u,v).
\end{aligned}
\end{equation}
If we restrict our choice of $\delta$ to $\left( 0, \frac{1}{2} \right]$, from \eqref{FET108}, \eqref{FET109}, \eqref{FET110}, we observe that
\begin{equation}\label{FET111}
m(u,v) \leq \frac{C_{15}}{2} r^{2-\chi}(u,v)
\end{equation}
with 
\begin{equation}\label{FET112}
C_{15}(\delta,\chi) := \frac{4\pi C_{14}^2}{3-2\delta}r^{1-2\delta+\chi}(0,v_*) + \frac{C_6^2}{(2-\chi)}.
\end{equation}
Consequently, we obtain
\begin{equation}\label{FET113}
\mu(u,v) \leq C_{15}r^{1-\chi}(u,v).
\end{equation}

\vspace{5mm}

To see how \eqref{FET113} can help us to obtain a bound for $B'(u_1)$, we first revisit \eqref{FET114} and observe that if we define
\begin{equation}\label{FET115}
D''(u_0) := \sup_{\{u_0\} \times [u_0,v_*]} (|\alpha'|) < +\infty,
\end{equation}
we then have
\begin{equation}\label{FET116}
|\alpha'(u,v)| \leq D''(u_0) + \sum_{i=1}^{11}|I_i|(u,u_0,v)
\end{equation}
with $I_i$ defined in \eqref{FET16} remains unchanged. Analogous to \eqref{FET45}, taking supremum over $\mathcal{D}(u_0,u_1,v_*)$ yields
\begin{equation}\label{FET117}
B'(u_1) \leq D''(u_0) + \sup_{\mathcal{D}(u_0,u_1,v_*)} \ml \sum_{i=1}^{11}|I_i|\mr.
\end{equation}
By comparing \eqref{FET117} with  \eqref{FET45}, we observe that the supremum is taken without an additional factor of $r^{\delta}(u,v)$ paired with each $|I_i|$.  This is where \eqref{FET113} comes in. By viewing an upper bound for $\mu$ as that for $\varepsilon$, we are able to pick appropriate values of $\delta, \chi, \iota$ such that the term $\varepsilon^{\frac{1}{2}}$ could be replaced by sufficient powers of $r(u,v)$ to counteract the term $r^{-\delta}(u,v)$ in each of the estimates for $|I_{i}|$. The details for relevant terms are provided as follows.
\begin{itemize}
\item[$\ml\frac{r|\dv \phi|}{\dv r}\mr^2$:] By \eqref{FET107}, we have
\begin{equation}\label{FET119}
\ml \frac{r|\dv \phi|}{\dv r}\mr^2 \leq C_{14}^2 r^{2-2\delta}(u,v).
\end{equation}
\item[$\ml\frac{r|\dv \phi|}{\dv r}\mr^3$:] By \eqref{FET107}, we obtain
\begin{equation}\label{FET118}
\ml \frac{r|\dv \phi|}{\dv r}\mr^3 \leq C_{14}^3 r^{3-3\delta}(u,v).
\end{equation}
\item[$|Q|$:] As \eqref{FET28} was derived from a more general expression in Proposition \ref{PropEstimates}, if we instead apply \eqref{FET113}, we would obtain
\begin{equation}\label{FET120}
|Q|(u,v) \leq C_6 C_{15}^{\frac{1}{2}} r^{2-\chi}(u,v).
\end{equation}
\item[$|Q\phi|:$] Similar to $|Q|$, if we work through the corresponding arguments and apply \eqref{FET113} instead, we get
\begin{equation}\label{FET121}
|Q\phi| \leq \frac{C_{7}C_{15}^{\frac{1}{2}}}{(1-\varepsilon)^{\frac{1}{2}}}r^{2-\chi}(u,v).
\end{equation}
\item[$|\phi|^2:$] In the same fashion, we deduce
\begin{equation}\label{FET122}
\begin{aligned}
|\phi|^2(u,v) &\leq \frac{C_{15}}{2\pi}\frac{1}{1-\varepsilon} r^{1-\chi}(u,v) + 2C_{8}^2(\iota)r^{-1-\iota}(u,v) \leq \frac{C_{16}}{1-\varepsilon}r^{-1-\iota}(u,v)
\end{aligned}
\end{equation}
with $C_{16}(\chi,\iota,\delta) := \frac{C_{15}(\delta,\chi)}{2\pi}r^{2-\chi+\iota}(0,v_*) + 2 C_{8}^{2}(\iota).$ 
\end{itemize}
For the remaining terms with $A'(u_1)$, a direct application of \eqref{FET58} suffices.

With the above preparation, we are now ready to derive improved estimates for $|I_{i}|$ with $i = 1$ to $11$. We summarize the bounds obtained for each $|K_i|(u,v)$ defined in \eqref{FET16} and thus list the bounds for each $|I_i|(u,u_0,v)$ in Table \ref{FETTable1}. Here we use $\lesssim$ to refer to inequalities with constants depending only on $v_*$ and parameters $\delta, \chi, \iota$. Meanwhile, due to \eqref{FET58}, the dependence on $A'(u_1)$ will be encoded under $\lesssim$ as $C_{13}$ depends on $\delta, \chi, \iota, v_*$, while being independent of $u_1$. Furthermore, we recall that we require the parameters to satisfy $0 < \chi < \delta \leq \frac{1}{2}$ and $0 < \iota < 1$.

\begin{table}[htbp]
\centering
\begin{tabular}{|c|c|c|}
\hline
    Index $i$ & $\frac{|K_i|}{-\du r}(u,v) \lesssim$ & $|I_i|(u,u_0,v) \lesssim$  \\
\hline
1 & $\frac{|Q||\phi|}{r^2} \ml \frac{r|\dv \phi|}{\dv r}\mr \lesssim r^{1-\chi-\delta}$ & $r^{2-\chi - \delta}$  \\
2 & $\frac{|Q|^2}{r^4}\ml \frac{r|\dv \phi|}{\dv r}\mr^3 \lesssim r^{3-2\chi-3\delta}$&  $r^{4-2\chi-3\delta}$\\
3 & $\frac{1}{r^2}\ml \frac{r|\dv \phi|}{\dv r}\mr^3 \lesssim r^{1-3\delta}$ & $r^{2-3\delta}$ \\
4 & $\frac{|\alpha'|}{r}\mu \lesssim r^{-\chi - \delta}$ & $r^{1 - \chi - \delta}$ \\
5 & $\frac{|Q|^2}{r^3}|\alpha'|\lesssim r^{1-2\chi - \delta}$ & $r^{2-2\chi-\delta}$ \\
6 & $\frac{|Q|^2}{r^4}\ml \frac{r|\dv \phi|}{\dv r}\mr \lesssim r^{1-2\chi-\delta}$& $r^{2-2\chi - \delta}$ \\
7 & $\frac{\mu}{r^2}\ml \frac{r|\dv \phi|}{\dv r}\mr \lesssim r^{-\chi - \delta}$& $r^{1-\chi-\delta}$ \\
8 & $\frac{|Q||\phi|}{r^2}\ml \frac{r|\dv \phi|}{\dv r}\mr^2 \lesssim r^{2-\chi-2\delta}$& $r^{3-\chi -2\delta}$ \\
9 & $|\phi|^2\ml \frac{r|\dv \phi|}{\dv r}\mr \lesssim r^{-\delta - \eta}$ & $r^{1-\delta - \eta}$ \\
10 & $\frac{|Q|^2}{r^2}\ml \frac{r|\dv \phi|}{\dv r}\mr \lesssim r^{3-2\chi-\delta}$& $r^{4-2\chi-\delta}$ \\
11 & $\frac{|Q||\phi|}{r^2} \lesssim r^{-\chi}$ & $r^{1-\chi}$ \\
    \hline
    \end{tabular}
    \caption{Table of improved estimates for $|I_i|$ in terms of powers of $r(u,v)$.}
    \label{FETTable1}
\end{table}
From Table \ref{FETTable1}, we observe the relevant improvement in the estimates for $|I_3|$, $|I_4|$, $|I_7|$, $|I_{11}|$. In addition, we see that 
\begin{equation}
\sum_{i=1}^{11}|I_i|(u,u_0,v) \lesssim r^{1-\chi-\delta}(u,v) = r^{\frac{1}{4}}(u,v) \leq r^{\frac{1}{4}}(0,v_*) < +\infty
\end{equation}
if we pick $\delta = \frac{1}{2}$ and $\iota = \chi = \frac{1}{4}$. Henceforth, from \eqref{FET117}, we deduce that $B'(u_1)$ uniformly bounded in $u_1$. We then conclude that the solution extends up to $\mathcal{D}(0,v_*)$ as a $C^1$ solution. 

\ul{Extension Beyond Boundary.} Last but not least, we will show that there exists some $\xi' > 0$ such that the solution extends as a $C^1$ solution to $\mathcal{D}(0,v_* +  \xi')$. This is done by estimating the change in $\mu$ across a rectangular region $\{(u,v): u \in [0,v_*], v\in [v_*,v_*+\xi]\}$ for some $\xi > 0$ and through a contradiction argument. Furthermore, to maintain the generality of the arguments and to draw parallel conclusions made in the previous subsections, we will keep $\delta, \chi, \iota$ as parameters to be chosen, and it is implicitly assumed that all constants depend on $\delta, \chi, \iota$ unless stated otherwise. One can observe that the argument below holds if we retain our choice of $\delta = \frac{1}{2}$ and $\chi = \iota = \frac{1}{4}$ as in the previous subsection.

\begin{figure}[htbp!]
\begin{minipage}[!t]{0.4\textwidth}
	\centering
\begin{tikzpicture}[scale=1]
	\begin{scope}[thick]
        \fill[gray!30] (0,1) to (1.5,2.5) to (1,3) to (0,2) to (0,1);
        \fill[gray!60] (1,3) to (1.5,3.5) to (2,3) to (1.5,2.5) to (1,3);
	\draw[->] (0,0) node[anchor=north]{$\mathcal{O}$} --  (0,5) node[anchor = east]{$\Gamma$};
	\draw[->] (0,0) -- (4,4) node[anchor = west]{$C_0^+$};
        \draw(0,4) -- (2,2);
        \node[] at (2.15,1.85) {\small $v_*$};
        \node[anchor = east] at (0,1) {$(u_0,u_0)$};
        \node at (0,1)[circle,fill,inner sep=1.5pt]{};
        \node[anchor = east] at (0,2) {$(u_1,u_1)$};
        \node at (0,2)[circle,fill,inner sep=1.5pt]{};
        \node[rotate=45] at (0.65,2.15) {\tiny $\mathcal{D}(u_0,u_1,v_*)$};
        \draw[dotted] (0,2) -- (1.5,3.5);
        \draw[] (0,1) -- (2,3);
        \draw[] (1.5,3.5) -- (2.5,2.5);
        \node[anchor = west] at (2.5,2.5) {\small $v_* + \xi$};
        \draw[->] (2,3.5) -- (1.5,3);
        \node[anchor = west] at (2,3.5){\small $U_{\xi}(u_1)$};
        \draw[pattern=north east lines, pattern color=black] (0,1) to (2,3) to (1.5,3.5) to (0,2) to (0,1);
        \node[anchor = west] at (6.5,2) {\ul{Legend:}};
        \node[anchor = west] at (6.9,1.5) {$\mathcal{D}(u_0,u_1,v_*)$}; 
        \fill[gray!30] (6.65,1.6) to (6.85,1.6) to (6.85,1.4) to (6.65,1.4) to (6.65,1.6);
        \node[anchor = west] at (6.9,1) {$U_\xi(u_1)$};
        \fill[gray!60] (6.65,1.1) to (6.85,1.1) to (6.85,0.9) to (6.65,0.9) to (6.65,1.1);
        \node[anchor = west] at (6.9,0.5) {$\mathcal{D}(u_0,u_1,v_*+\xi)$}; 
        \path[pattern=north east lines, pattern color=black] (6.65,0.6) to (6.85,0.6) to (6.85,0.4) to (6.65,0.4) to (6.65,0.6);
	\end{scope}
\end{tikzpicture} 
\end{minipage}
\caption{Illustration of $U_\xi(u_1)$ and $\mathcal{D}(u_0,v_1,v_* + \xi)$.}
\label{FETFig2}
\end{figure}
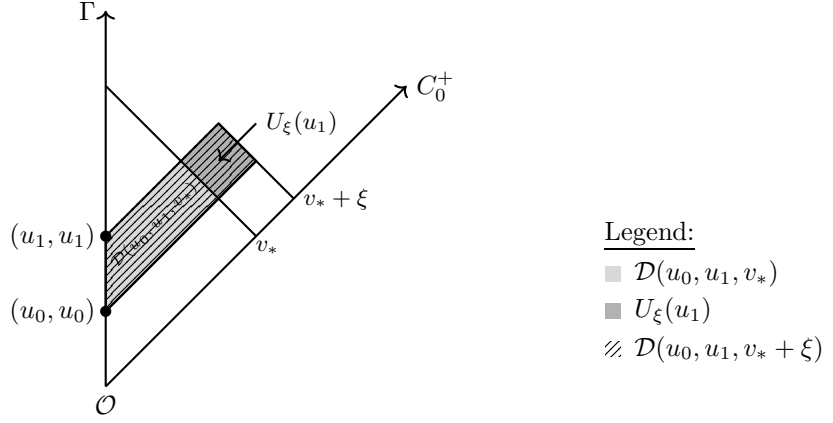

First, we consider a $u_0 \in [0,v_*)$ sufficiently close to $v_*$ so that \eqref{FET1} holds. Next, we use the continuity of $\dv r$ to deduce that $\dv r \geq 0$ on $\{u_0\} \times [v_*, v_* + \Delta v]$ for some fixed $\Delta v > 0$. Then, for any $0<\xi < \Delta v$, we consider the corresponding rectangular strip given by
\begin{equation}\label{FET75}
U_{\xi}(u_1) = [u_0,u_1) \times [v_*,v_*+\xi].
\end{equation}
In the spirit of \eqref{Au1}, we define
\begin{equation}\label{FET76}
\overline{A'_{\xi}}(u_1) := \sup_{U_{\xi}(u_1)} (r^{\delta}|\alpha'|).
\end{equation}
We also note that, under the convention in the previous subsections, we can write
\begin{equation}\label{FET77}
\mathcal{D}(u_0,u_1,v_* + \xi) = \mathcal{D}(u_0,u_1,v_*) \cup U_{\xi}(u_1).
\end{equation}

In below we will show that $\overline{A'_{\xi}}(u_1)$ remains bounded for some $u_1 \in (u_0,v_*)$ sufficiently close to $u_0$. This can be proved by utilizing the continuity of $\mu$ on $\mathcal{Q}$. With $\varepsilon$ and $\varepsilon'$ as in the previous subsections, if we pick $v_* - u_0$ and $\xi$ to be sufficiently small, then we have
\begin{equation}\label{FET78}
\varepsilon''(u_1) := \sup_{\mathcal{D}(u_0,u_1,v_*+\xi)} \mu \leq \frac{1}{2}(\varepsilon + \varepsilon') < \varepsilon'.
\end{equation}
Later, we will show that $\varepsilon''(u_1)$ defined above obeys the same upper bound as in \eqref{FET78} for all $u_1 \in (u_0,v_*)$. 

Before we begin detailing the arguments, we will first state and prove two lemmas that would be useful. We first define
\begin{equation}\label{FET83}
A'_* := \sup_{u_1\in(u_0,v_*)}A'(u_1) < +\infty.
\end{equation}
Here $A'_*$ is independent of $u_1$ as we have showed in \eqref{FET58} that $A'(u_1)$ is uniformly bounded in $u_1$ by $C_{13}$. 

\begin{lemma}\label{FETLemma3}
For any $u_1 \in (u_0,v_*)$, if $\mu$ satisfies
\begin{equation}\label{FET79}
\sup_{\mathcal{D}(u_0,u_1,v_* + \xi)} \mu \leq \frac{1}{2}(\varepsilon + \varepsilon'),
\end{equation}
then we have
\begin{equation}\label{FET80}
\overline{A'_{\xi}}(u_1) \leq \max\left\{A'_*, C_{17}(u_0,v_* + \Delta v) \right\}
\end{equation}
with $C_{17}$ defined in \eqref{FET124}. 
\end{lemma}

\proof $ $  Similar to \eqref{D'}, we start by defining
\begin{equation}\label{FET81}
\overline{D'_\xi}(u_0) := \sup_{\{u_0\}\times[v_*,v_*+\xi]}(r^{\delta}|\alpha'|).
\end{equation}
Then, repeating the same derivations for \eqref{FET14} and  \eqref{FET15}, but instead with $u_1 \in (u_0,v_*)$ and $v \in [v_*,v_*+\xi]$, we get 
\begin{equation}\label{FET82}
|\alpha'|(u,v) \leq \overline{D'_\xi}(u_0)r(u_0,v)^{-\delta} + \sum_{i=1}^{11}|I_i|(u,u_0,v).
\end{equation} 
In \eqref{FET82}, the expressions for $I_i$ remains unchanged as in \eqref{I_1} to \eqref{I_11}. Now, we will obtain an upper bound for each of the terms $|I_1|$ to $|I_{11}|$ as in the previous subsections. A few key differences are listed below.

For $\mu$, we apply \eqref{FET78} and replace each $\varepsilon$ in the estimates obtained above by 
\begin{equation}\label{tildeepsilon}
\tilde{\varepsilon} := \frac{1}{2}(\varepsilon + \varepsilon').
\end{equation}
This is justified since instead of $\mu \leq \varepsilon$, we now have $\mu \leq \frac{1}{2}(\varepsilon + \varepsilon')$. Hence, in the region $\mathcal{D}(u_0,u_1,v_*+\xi)$ as the supremum of the quantity $r^{\delta}|\alpha'|$ is bounded above by the maximum of its supremum over $\mathcal{D}(u_0,u_1,v_*)$ and its supremum over $U_\xi(u_1)$ (see Figure \ref{FETFig2}), we obtain
\begin{equation}\label{FET84}
\sup_{\mathcal{D}(u_0,u_1,v_*+\xi)}(r^{\delta}|\alpha'|) \leq \max\{ A'(u_1), \overline{A'_\xi}(u_1)\} \leq \max\{A'_*, \overline{A'_\xi}(u_1)\}.
\end{equation}
For the remaining terms, an analogous logic to the previous sections applies with a slight modification demonstrated as follows. For $(u,v) \in U_\xi(u_1)$, we estimate $\alpha'$ over the entire region $\mathcal{D}(u_0,u_1,v_*+\xi)$ to derive relevant estimates for $\ml \frac{r|\dv \phi|}{\dv r}\mr$, $|Q|$, $|Q\phi|$, and consequently $\ml \frac{r|\dv \phi|}{\dv r}\mr^2$, $\ml \frac{r|\dv \phi|}{\dv r}\mr^3$, $|\phi|^2$. Thus, by \eqref{FET84}, we can replace $A'(u_1)$ from \eqref{FET19} to \eqref{FET71} by $\max\{A'_*, \overline{A'_\xi}(u_1)\}$.

With above pointers, we then multiply $r^\delta$ on both sides of \eqref{FET82} and take supremum over $U_{\xi}(u_1)$ to obtain
\begin{equation}\label{FET85}
\begin{aligned}
\overline{A'_\xi}(u_1)  &\leq \overline{D'_\xi}(u_0) + \ml \frac{ \sum_{i=1}^{10}C_{9,i}(v_*+\xi)}{\delta(1-\delta)}\tilde{\varepsilon}^{\frac{1}{2}} + C_{11,9}(u_0;v_*+\xi)\mr \max\{A'_*, \overline{A'_\xi}(u_1)\} + \frac{C_{10,11}(v_*+\xi)}{\delta(1-\delta)} \tilde{\varepsilon}^{\frac{1}{2}} \\
&\leq \overline{D'_\xi}(u_0) + \ml \frac{ \sum_{i=1}^{10}C_{9,i}(v_* + \Delta v)}{\delta(1-\delta)}\tilde{\varepsilon}^{\frac{1}{2}} + C_{11,9}(u_0;v_* + \Delta v)\mr \max\{A'_*, \overline{A'_\xi}(u_1)\} + \frac{C_{10,11}(v_* + \Delta v)}{\delta(1-\delta)} \tilde{\varepsilon}^{\frac{1}{2}},
\end{aligned}
\end{equation}
where we use relevant estimates in the previous subsections, subject to the remarks above. Moreover, as explained in the previous subsections, the constants derived are non-decreasing with increasing $v_*$, which thus allows us to obtain the $\xi-$independent estimates for the constants $C_{9,i}, C_{10,11}, C_{11,9}$. By further refining $\Delta v > 0$ from the start, we could also assume that
\begin{equation}\label{FET89}
C_{11,9}(u_0, v_* + \Delta v) < \frac{1}{4},
\end{equation}
which is analogous to \eqref{FET73}. Note that inequality \eqref{FET85} bears resemblance to \eqref{FET53}. 

With a similar choice of $\varepsilon'$ as in \eqref{FET57} but with constants evaluated at $(u_0,v_*+\Delta v)$, we then deduce that
\begin{equation}\label{FET86}
\varepsilon' \leq \ml \frac{\delta(1-\delta)}{ \sum_{i=1}^{10}C_{9,i}(u_0,v_*+\Delta v)} \mr^2.
\end{equation}
This implies that the counterpart of \eqref{FET54} is still true, if we replace $\varepsilon$ by $\tilde{\varepsilon}$ and it satisfies
\begin{equation}\label{FET87}
\tilde{\varepsilon} = \frac{1}{2}(\varepsilon + \varepsilon') < \varepsilon' \leq \ml \frac{\delta(1-\delta)}{ \sum_{i=1}^{10}C_{9,i}(u_0,v_*+\xi)} \mr^2.
\end{equation}
This implies
\begin{equation}\label{FET88}
\frac{\ml \sum_{i=1}^{10}C_{9,i}(u_0,v_*+\xi)\mr^2}{\delta^2(1-\delta)^2} < \frac{1}{4}.
\end{equation}
Parallel to \eqref{FET58}, we also define
\begin{equation}\label{FET124}
C_{17}(u_0,v_*+\Delta v) := 2 \ml \overline{D'_\xi}(u_0) + C_{12}(v_* + \Delta v) \mr.
\end{equation}
By considering two separate cases $A_*' > \overline{A'_{\xi}(u_1)}$ and $A_*' \leq \overline{A'_{\xi}(u_1)}$, following a similar argument as in the previous subsections, analogous to \eqref{FET58}, we then deduce that \eqref{FET80} holds. This concludes the proof of this lemma. \qed \\

We will also equip ourselves with an elementary lemma as follows:
\begin{lemma}\label{FETLemma4}
For any $0 < y \leq x$ and $0 \leq p \leq 1$, we have
\begin{equation}\label{FET125}
x^p - y^p \leq (x-y)^p.
\end{equation}
\end{lemma}
\begin{proof}
Dividing $y^p$ throughout \eqref{FET125}, it is equivalent to show that, for any $r \geq 1$, we have
\begin{equation}\label{FET126}
r^p - 1 \leq (r - 1)^p.
\end{equation}
Let $f(r) = r^p - 1 - (r-1)^p$. For $r \geq 1$, we see that $f(1) = 0$ and $f'(r) = pr^{p-1}\ml 1-\ml 1-\frac{1}{r}\mr^{p-1} \mr \leq 0$. This implies that $f$ is a non-increasing function of $r$ since $1 \leq \ml 1 - \frac{1}{r}\mr^{p-1}$ for $p \in [0,1]$. Thus, we prove $f(r) \leq f(1) = 0$ for all $r \geq 1$ and conclude.
\end{proof}
With the above two lemmas, we proceed back on track to prove that the $C^1$ solution extends beyond the domain of dependence $\mathcal{D}(u_0,u_1,v_*)$. Similar to \cite{christ2}, we first define
\begin{equation}\label{FET90}
u_2 = \sup\left\{ u_1 \in (u_0,v_*): \sup_{\mathcal{D}(u_0,u_1,v_*+\xi)} \mu \leq \frac{1}{2}(\varepsilon + \varepsilon')\right\}.
\end{equation}
By the definition of $u_2$, we then have either $u_2 = v_*$ or
\begin{equation}\label{FET91}
\sup_{U_\xi (u_2)} \mu = \frac{1}{2}(\varepsilon + \varepsilon')
\end{equation}
is attained at some $u_2 \in [u_0,v_*)$. We then show that for $\xi$ sufficiently small, the latter will not happen. 

Suppose for a contradiction that the latter happens. Then, we set $u_1 = u_2$ as in the hypothesis of Lemma \ref{FETLemma3}. Consequently,  by \eqref{FET91}, $\sup_{\mathcal{D}(u_0,u_2,v_*+\xi)} \mu$ satisfies \eqref{FET79} and hence by Lemma \ref{FETLemma3}, we have that \eqref{FET80} holds. We then estimate $\mu$ in the region $U_\xi(u_2)$ to obtain a contradiction with sufficiently small $\xi$. As described in the proof of Lemma \ref{FETLemma3}, in the region $U_\xi(u_2)$, we have
\begin{equation}\label{FET92}
\frac{r|\dv \phi|}{\dv r}(u,v) \leq \frac{r^{1-\delta}(u,v)}{2-\delta} \max\{A'_*, C_{12}\}
\end{equation}
with the argument $(u_0,v_*+\Delta v)$ suppressed when the context is clear. We can then estimate $\mu$ in $U_\xi(u_2)$ as follows. From \eqref{FET28} we first observe that
\begin{equation}\label{FET127}
\frac{Q^2}{r^2}(u,v) \leq \mu(u,v) C_6^2(v_*+\Delta v)r^{1-\chi}(u,v).
\end{equation}
By further requiring $u_0$ to be sufficiently close to $v_*$, we can then demand that
\begin{equation}\label{FET128}
C_6^2(v_*)r^{\frac{1}{2}-\frac{\chi}{2}}(u_0,v_*) \leq \frac{1}{2}.
\end{equation}
Hence, from the expression of $C_6$ in \eqref{FET130}, we see that the mapping $v \mapsto C_6^2(v)r^{1-\chi}(u_0,v)$ is continuous. This allows us to refine $\Delta v$ to be sufficiently small (depending on the given $u_0$) such that
\begin{equation}\label{FET132}
C_6(v_* + \Delta)r^{1-\chi}(u_0,v_* + \Delta) \leq \frac{7}{8},
\end{equation}
and consequently,
\begin{equation}\label{FET129}
\frac{Q^2}{r^2}(u,v) \leq \frac{7}{8}\mu(u,v).
\end{equation}
This in turn implies that
\begin{equation}\label{FET131}
\ml \mu - \frac{Q^2}{r^2}\mr(u,v) \geq \frac{\mu}{8}(u,v) \geq 0
\end{equation}
for all $(u,v) \in U_\xi(u_2)$.
With this and the help of \eqref{Setup27}, we move on to estimate $\mu$ directly via calculating
\begin{equation}\label{FET93}
\begin{aligned}
\dv \mu &= \frac{4\pi r (1-\mu)|\dv \phi|^2}{\dv r} - \frac{\dv r}{r}\ml \mu - \frac{Q^2}{r^2}\mr \leq \frac{4\pi r (1-\mu)|\dv \phi|^2}{\dv r}.
\end{aligned}
\end{equation}
In \eqref{FET93}, the inequality follows from the fact that $\dv r > 0$ while $\mu - Q^2/r^2 \geq 0$. 
For $(u,v) \in U_\xi(u_2)$, we further perform integration and apply \eqref{FET92} and \eqref{FET93} to obtain
\begin{equation}\label{FET94}
\begin{aligned}
\mu(u,v) &\leq \mu(u,v_*) + 4 \pi \int^v_{v_*} \frac{(1-\mu)}{r} \ml\frac{r|\dv \phi|}{\dv r}\mr^2 (\dv r) (u,v') \; \D v' \\
&\leq \mu(u,v_*) + \frac{4 \pi (\max\{A'_*,C_{17}  \})^2}{(2-\delta)^2} \int^v_{v_*} r^{1 - 2 \delta} (\dv r) (u,v') \; \D v' \\
&= \mu(u,v_*) + \frac{4 \pi (\max\{A'_*,C_{17}\})^2}{2(2-\delta)^2(1-\delta)} \ml r^{2-2\delta}(u,v) - r^{2-2\delta}(u,v_*) \mr. \\
\end{aligned}
\end{equation}
On the other hand, since $\dv r$ is $C^1$ for $(u,v) \in U_\xi(u_2)$, we have $\dv r(u_0,v)$ is $C^1$ on $[v_*,v_*+ \Delta v]$. Thus, there exists $C_{18}(u_0) > 0$, such that $\dv r(u_0,v) \leq C_{18}(u_0)$ for all $v \in [v_*,v_*+\Delta v]$. Hence, for $(u,v) \in U_\xi(u_2)$, using the fact that $\dv \du r \leq 0$ as described above, we obtain
\begin{equation}\label{FET96}
\begin{aligned}
r(u,v) - r(u,v_*) &= r(u_0,v) - r(u_0,v_*) + \int^u_{u_0} \int^v_{v_*} \dv \du r(u',v') \; \D u' \D v' \\
&\leq r(u_0,v) - r(u_0,v_*) = \int_{v_*}^v \dv r (u_0,v') \; \D v' \\
&\leq C_{18}(u_0)(v - v_*) \leq C_{18}\xi.
\end{aligned}
\end{equation}
Now, for $0 < \delta < \frac{1}{2}$ and hence $2 - 2\delta > 1$, we employ \eqref{FET94} and \eqref{FET96} to obtain
\begin{equation}\label{FET95}
\begin{aligned}
\mu(u,v) &\leq \mu(u,v_*) + \frac{4 \pi (\max\{A'_*, C_{17}  \})^2}{2(2-\delta)^2(1-\delta)} (r^{2-2\delta}(u,v) - r^{2-2\delta}(u,v_*)) \\
&\leq \mu(u,v_*) + \frac{4 \pi (\max\{A'_*,C_{17}\})^2}{2(2-\delta)^2(1-\delta)} r^{2-2\delta-1}(u,v_* + \Delta v)(r(u,v) - r(u,v_*)) \\
&\leq \mu(u,v_*) + \frac{4 \pi (\max\{A'_*,C_{17}\})^2 C_{18}}{2(2-\delta)^2(1-\delta)} r^{2-2\delta-1}(u_0,v_*+ \Delta v)\xi.\\
\end{aligned}
\end{equation}
Since $\mu(u,v_*) \leq \varepsilon$, to obtain $\mu(u,v) < \frac{1}{2}(\varepsilon + \varepsilon')$, it suffices to pick $\xi > 0$ sufficiently small, such that 
\begin{equation}\label{FET97}
\frac{4 \pi (\max\{A'_*,C_{17}  \})^2 C_{18}}{2(2-\delta)^2(1-\delta)} r^{2-2\delta-1}(u_0,v_*+\Delta v)\xi < \frac{1}{2}(\varepsilon' - \varepsilon).
\end{equation}
In contrast, for the case $\delta \geq \frac{1}{2}$ and hence $2 - 2 \delta \leq 1$, we can apply Lemma \ref{FETLemma4} and \eqref{FET94} to obtain
\begin{equation}\label{FET100}
\begin{aligned}
\mu(u,v) &\leq \mu(u,v_*) + \frac{4 \pi (\max\{A'_*,C_{17}  \})^2}{2(2-\delta)^2(1-\delta)} \ml r^{2-2\delta}(u,v) - r^{2-2\delta}(u,v_*) \mr\\
&\leq  \mu(u,v_*) + \frac{4 \pi (\max\{A'_*, C_{17} \})^2}{2(2-\delta)^2(1-\delta)} \ml r(u,v) - r(u,v_*) \mr^{2-2\delta} \\
&\leq  \mu(u,v_*) + \frac{4 \pi (\max\{A'_*, C_{17} \})^2 C_{18}^{2-2\delta}}{2(2-\delta)^2(1-\delta)} \xi^{2-2\delta}.
\end{aligned}
\end{equation}
Similarly to the former case, it suffices to pick $\xi > 0$ sufficiently small such that 
\begin{equation}\label{FET101}
\frac{4 \pi (\max\{A'_*,C_{17}  \})^2 C_{18}^{2-2\delta}}{2(2-\delta)^2(1-\delta)} \xi^{2-2\delta} < \frac{1}{2}(\varepsilon' - \varepsilon).
\end{equation}
Combining \eqref{FET97} and \eqref{FET101}, while requiring $\xi$  to be less than or equals to $\Delta v$, a valid choice of $\xi$ would be 
\begin{equation}\label{FET102}
\xi = \min\left\{ \Delta v, \frac{(2-\delta)^2(1-\delta)(\varepsilon'-\varepsilon)}{8\pi(\max\{A'_*,C_{17} \})^2 C_{18}r^{2-2\delta-1}(u_0,v_*+\xi_1)}, \ml \frac{(2-\delta)^2(1-\delta)(\varepsilon'-\varepsilon)}{8\pi(\max\{A'_*,C_{17} 
 \})^2 C_{18}^{2-2\delta}}\mr^{\frac{1}{2-2\delta}} \right\}.
\end{equation}
Therefore, for any $\delta \in (0,1)$, such a choice $\xi$ would imply that
\begin{equation}\label{FET103}
\sup_{U_\xi (u_2)} \mu < \frac{1}{2}(\varepsilon + \varepsilon'),
\end{equation}
contradicting \eqref{FET91}. \\

In summary, the above arguments show that $u_2$ in \eqref{FET90} is achieved at $u_2 = v_*$, and hence the bound in \eqref{FET90} continues to hold. This implies that
\begin{equation}\label{FET104}
\sup_{\mathcal{D}(u_0,v_*,v_*+\xi)} \mu \leq \frac{1}{2}(\varepsilon
+ \varepsilon').
\end{equation}
Thus, for each $u_1 \in (u_0,v_*)$, the hypothesis of Lemma \ref{FETLemma3} is verified. This then implies that \eqref{FET80} holds. Observe that \eqref{FET80} further implies that  $\overline{A'}(u_1) := \sup_{\mathcal{D}(u_0,u_1,v_*+\Delta v)}(r^{\delta}|\alpha'|)$ remains bounded as $u_1 \rightarrow v_*$. By conducting a similar argument as in previous subsections, we can also obtain improved estimates for $|I_1|$ to $|I_{11}|$. Henceforth, we can deduce that the solution extends as a $C^1$ solution to the closure of $\mathcal{D}(u_0,v_*,v_*+\xi)$. In particular, with $C^1$ initial data along $\{v_*\} \times [v_*,v_*+\xi]$, we show that the solution further extends as a $C^1$ solution to $\mathcal{D}(v_*,v_*+\xi')$ for some $\xi' \in (0,\xi]$. This concludes the proof of Theorem \ref{FET}. \qed

\section{First Instability Theorem}\label{First Instability Theorem}

As previously described, we consider the following picture. We mark $\mathcal{O}' = (v_0,v_0)$ along $\Gamma$ as the start of the singular boundary $\mathcal{B}$ and $\mathcal{J}^-(\mathcal{O}') \subset \mathcal{R}$. Denote $C_0^-$ to be the past-directed incoming null curve emitting from $\mathcal{O}'$. For the remaining parts of this section, we shall focus on $\mathcal{Q}\setminus \mathcal{J}^-(\mathcal{O}')$ lying in the future of both $C_0^+$ and $C_0^-$. Figure \ref{FITFig3} illustrates the above set-up.

\begin{figure}[htbp]
\begin{tikzpicture}[
    dot/.style = {draw, fill = white, circle, inner sep = 0pt, minimum size = 4pt}, scale=0.8]
\begin{scope}[thick]
\fill[gray!50]  (2.5,2.5) to (0,5) to (1.7,6.7) node[dot]{} to [out=45,in=175] (4.7,7.0) to (5,5) to (2.5,2.5);
\draw[thick] (0,0) node[anchor=north]{$\mathcal{O}$} --  (0,5);
\node[] at (-0.3,2.5) {$\Gamma$};
\node[] at (-0.3,5.0) {$\mathcal{O}'$};
\node[] at (1,2.5) {$\mathcal{J}^-(\mathcal{O}')$};
\draw[->] (0,0) -- (5,5) node[anchor = west]{$C_0^+$};
\draw[thick, densely dotted] (0,5) to (1.7,6.7);
\draw[thick, loosely dotted](1.7,6.7) to [out=45,in=175] (4.7,7.0);
\node[] at (0.5,6) {$\mathcal{B}_0$};
\node[rotate=-10] at (4.3,7.3) {$\mathcal{B}\setminus\mathcal{B}_0$};
\draw[thick,dashed] (0,5) --  (2.5,2.5);
\node[] at (2.7,2.3){$v_0$};
\node[] at (2.7,4.8){$\mathcal{Q} \setminus \mathcal{J}^{-}(\mathcal{O}')$};
\draw[fill=white] (1.7,6.7) circle (2pt);
\draw[fill=white] (0,5) circle (2pt);
\end{scope}
\end{tikzpicture}
\caption{Diagram for this section. The shaded region corresponds to the region $\mathcal{Q} \setminus \mathcal{J}^-(\mathcal{O}')$, while the unshaded region corresponds to the region $\mathcal{J}^-(\mathcal{O}')$.}
\label{FITFig3}
\end{figure}
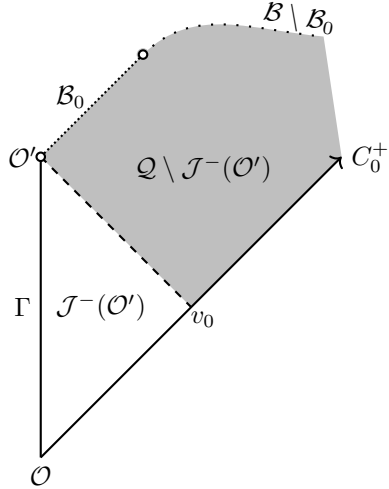

Recall the quantity $\gamma$ from \eqref{gamma} below:
\begin{equation}\label{FIT3}
\gamma(u,v) := \int^u_{0} \frac{(- \du r)}{r} \frac{\mu - Q^2/r^2}{1 - \mu} (u',v) \mathrm{d}u'.
\end{equation}
For notational simplicity, we denote
\begin{equation}\label{FIT4}
    \gamma_0(u) := \gamma(u,v_0).
\end{equation} 
In lieu of estimates that we will establish, we also define the truncated gamma $\tgam$, with $\overline{u} \in [0,v_0)$ as a parameter, as below
\begin{equation}\label{FIT4a}
\tgam(u,v) := \gamma(u,v;\overline{u}) = \int^u_{\overline{u}} \frac{(- \du r)}{r} \frac{\mu - Q^2/r^2}{1 - \mu} (u',v)\mathrm{d}u'.
\end{equation}
For brevity, the dependence on $\overline{u}$ will be dropped if the given context is clear. In addition, similar to Section \ref{TSF}, we use the subscript $()_1$ to indicate that a given quantity is evaluated along $C_0^-$.

We also use $g(u,v)$ and $f(u,v)$ to denote the following functions
\begin{equation}\label{FIT10}
g(u,v) := D_u \phi(u,v) +  \ii \e \frac{Q \phi (\du r)}{r(1-\mu)}(u,v),
\end{equation}
\begin{equation}\label{FIT11}
f(u,v) := \int^u_{0} -\ii \e A_u(u',v) + \frac{(-\du r)}{r} \frac{\mu - Q^2/r^2}{1 - \mu}(u',v) \mathrm{d}u'.
\end{equation}
With them, we define  
\begin{equation}\label{FIT14}
I(u) := \int^u_{0} g_0(u')e^{-f_0(u')}\mathrm{d}u'.
\end{equation}

We are now ready to state the main theorem as follows. 

\begin{theorem}\label{FITTheorem} (First Instability Theorem.) 
Let $\gamma_0(u)$ defined in \eqref{FIT4} be unbounded as $u \rightarrow v_0^-$. Suppose that $I$ defined in \eqref{FIT14} does not converge to a finite limit as $u \rightarrow v_0^-$ or otherwise
\begin{equation}\label{FITTheorem1}
\ml \frac{r \dv \phi}{\dv r} \mr_0 (0) \neq \lim_{u \rightarrow v_0^-} I(u).
\end{equation}
Then, we can conclude that
\begin{enumerate}[leftmargin=*]
\item $\mathcal{A}$ is non-empty, \\
and
\item Either one of the following two scenarios holds:
\begin{enumerate}[(i)]
\item$\mathcal{O}' = \mathcal{B}_0$ (and thus both $\mathcal{B}\setminus \mathcal{B}_0$ and $\mathcal{A}$ issue from $\mathcal{O}'$), or
\item $\mathcal{O}' \subsetneq \mathcal{B}_0$ and for any sufficiently small open neighborhood $U$ of $\mathcal{O}'$ with $U \subset \mathcal{Q} \setminus \mathcal{J}^-(\mathcal{O}')$, we have $U \cap \mathcal{A} \neq  \emptyset$.
\end{enumerate}\end{enumerate}
\end{theorem}

Before we begin with the proof of the theorem, we first establish several useful lemmas.

\begin{lemma}\label{FITLemma2}
For the spherically symmetric Einstein-Maxwell-charged scalar field system, if $\mathcal{A} = \emptyset$, then we have $\mathcal{B} \setminus \mathcal{B}_0 = \emptyset$ and $\mathcal{B}_0$ is an outgoing null segment originating from $\mathcal{O}'$ and extends all the way to $v \rightarrow \infty$.
\end{lemma}
\begin{proof}
By Proposition \ref{SetupProp3}, under the assumption of this lemma, we have that $\mathcal{A} \cup \mathcal{T} = \emptyset$. 
For any potential point $p \in \mathcal{B} \setminus \mathcal{B}_0$, we have $r(p) = 0$. From $\mathcal{A} \cup \mathcal{T} = \emptyset$, we also get $\dv r > 0$ everywhere. Henceforth, any outgoing null curve originating from $\Gamma$ will never hit $\mathcal{B} \setminus \mathcal{B}_0$.
\end{proof}

\begin{lemma}\label{FITLemma3}
Along any incoming null curve with $v \geq v_0$, for all $0 \leq u < v_0$, if $(u',v) \in \mathcal{R}$ for all $u' \in [0,u]$, then we have
\begin{equation}\label{FIT18a}
    (\dv r) (u,v) = (\dv r)(0,v) e^{-\gamma(u,v)}.
\end{equation}
In particular, along $v = v_0$, we deduce
\begin{equation}\label{FIT18}
    (\dv r)_0 (u) = (\dv r)_0(0) e^{-\gamma_0(u)}
\end{equation}
Furthermore, for any $\overline{u} \in [0,v_0)$, with $\tgam$ defined in \eqref{FIT4a}, we also get
\begin{equation}\label{FIT18b}
(\dv r)(u,v) = (\dv r)(\overline{u},v) e^{-\tgam(u.v)}.
\end{equation}
\end{lemma}
\begin{proof}
Using \eqref{Setup20}, we can rearrange its expression to obtain
\begin{equation}\label{FIT19}
\du \log \ml \dv r\mr = \frac{\du r}{r}\ml \frac{\mu-Q^2/r^2}{1-\mu}\mr.
\end{equation}
Integrating \eqref{FIT19}, we retrieve \eqref{FIT18a}, \eqref{FIT18} and \eqref{FIT18b}.
\end{proof}

\begin{lemma}\label{FITEstLemma1} Along any incoming null curve with $v \geq v_0$, for all $0 \leq \overline{u} \leq u < v_0$, if $(u',v) \in \mathcal{R}$ for all $u' \in [\overline{u},u]$, we have 
\begin{equation}\label{FIT5a}
    \frac{1}{1-\mu(u,v)} \leq \frac{1}{1-\mu(\overline{u},v)}e^{\tgam(u,v)}.
\end{equation}
In particular, we also get
\begin{equation}\label{FIT5}
    \frac{1}{1-\mu(u,v)} \leq \frac{1}{1-\mu(0,v)}e^{\gamma(u,v)}.
\end{equation}
\end{lemma}

\begin{proof}
With the help of \eqref{Setup26}, we compute
$$\du\ml \frac{1}{1-\mu(u,v)}\mr = \frac{\du \mu}{(1 - \mu)^2} = \frac{1}{1-\mu(u,v)} \ml -\frac{4\pi r|D_u \phi|^2}{(-\du r)} + \frac{(-\du r)}{r}\frac{\mu - Q^2/r^2}{1-\mu}\mr.$$
This gives
\begin{equation}\label{FIT6}
\begin{aligned}
\frac{1}{1-\mu(u,v)} &= \frac{1}{1-\mu(\overline{u},v)}\exp\ml {\int^u_{\overline{u}} \ml -\frac{4\pi r|D_u \phi|^2}{(-\du r)} + \frac{(-\du r)}{r}\frac{\mu - Q^2/r^2}{1-\mu}\mr (u',v) \mathrm{d}u'} \mr.
\end{aligned}
\end{equation}
Since $\du r < 0$, this implies
\begin{equation}\label{FIT7}
\begin{aligned}
\frac{1}{1-\mu(u,v)} &\leq \frac{1}{1-\mu(\overline{u},v)}\exp\ml {\int^u_{\overline{u}} \ml  \frac{(-\du r)}{r}\frac{\mu - Q^2/r^2}{1-\mu}\mr (u',v) \mathrm{d}u'} \mr. \\
\end{aligned}
\end{equation}
Employing the definition of $\gamma$ in \eqref{FIT3} yields \eqref{FIT5a} and \eqref{FIT5}.
\end{proof}
\noindent Observe that these two lemmas were established when we were proving Theorem \ref{SET}, albeit only for $\overline{u} = 0$.

Next, we remark that the definition of $\gamma$ in \eqref{FIT3} naturally arises from the evolution equation for $\frac{r \dv \phi}{\dv r}$ along the $u$-direction. To see this, by \eqref{Setup20} and \eqref{Setup73}, we get
\begin{equation}\label{FIT8}
\begin{aligned}
\du \ml \frac{r \dv \phi}{\dv r} \mr &= \frac{\du (r\dv \phi)}{\dv r} + \frac{r \dv \phi (- \du \dv r)}{(\dv r)^2} \\
&= \frac{- \dv r D_u \phi + \ii \e \frac{Q\phi (-\du r)(\dv r)}{r(1-\mu)}}{\dv r} - \ii \e A_u \frac{r \dv \phi}{\dv r} + \frac{\dv \phi (- \du r)}{(\dv r)} \frac{\mu - Q^2/r^2}{1-\mu} \\
&= -D_u \phi - \ii \e \frac{Q \phi (\du r)}{r(1-\mu)} + \ml \frac{r \dv \phi}{\dv r}\mr \ml -\ii \e A_u + \frac{(-\du r)}{r} \frac{\mu - Q^2/r^2}{1 - \mu} \mr. \\
\end{aligned}
\end{equation}
Using the definition of $f$ and $g$ in \eqref{FIT10} and \eqref{FIT11}, we obtain
\begin{equation}\label{FIT9}
\ml \frac{r \dv \phi}{\dv r} \mr(u,v)  = \ml\ml \frac{r \dv \phi}{\dv r} \mr(0,v) - \int^u_{0} g(u',v)e^{-f(u',v)}\mathrm{d}u' \mr e^{f(u,v)}
\end{equation}
Taking absolute values, we arrive at
\begin{equation}\label{FIT12}
\ml \frac{r |\dv \phi|}{|\dv r|} \mr(u,v)  = \mlm \ml \frac{r \dv \phi}{\dv r} \mr(0,v) - \int^u_{0} g(u',v)e^{-f(u',v)}\mathrm{d}u' \mrm e^{\gamma(u,v)}.
\end{equation}
Employing $(\dv r)_0 > 0$ in lieu of Lemma \ref{FITLemma3}, we then derive
\begin{equation}\label{FIT13}
\ml \frac{r |\dv \phi|}{\dv r} \mr_0(u)  = \mlm \ml \frac{r \dv \phi}{\dv r} \mr_0(0) - I(u) \mrm e^{\gamma_0(u)}
\end{equation}
with $I(u)$ defined in \eqref{FIT14}.

Last but not least, we have the following lemma
\begin{lemma}\label{FITEstLemma6} 
Along $v = v_0$, as $u \rightarrow v_0^-$, we have that
\begin{equation}\label{FIT57}
\gamma_0(u) < + \infty \quad \mbox{ implies } \quad \mu_0(u) \rightarrow 0.
\end{equation}
\end{lemma}
\begin{proof}   
By Proposition \ref{PropEstimates}, with $\overline{u} = 0$ and $\chi = \frac{1}{2}$, inequality \eqref{Estimates17} implies that
\begin{equation}\label{FIT220}
\ml \frac{Q^2}{r^2} \mr_0 (u) \leq C_1^2 \ml 0,v_0;\frac{1}{2}\mr r_0^{\frac{1}{2}}(0)v_0 r_0^{\frac{1}{2}}(u) \mu_0(u).
\end{equation}
We then pick $u_1 \in [0,v_0)$ with $u_1$ sufficiently close to $v_0$, such that $C_1^2 \ml 0,v_0;\frac{1}{2}\mr r_0^{\frac{1}{2}}(0)v_0 r_0^{\frac{1}{2}}(u) \leq \frac{1}{2}$ for all $u \in [u_1,v_0)$. This simplifies \eqref{FIT220} to
\begin{equation}\label{FIT221}
0 \leq \ml \frac{Q^2}{r^2} \mr_0 (u) \leq \frac{\mu_0(u)}{2}.
\end{equation}
Recalling the definition of $\gamma_0$ in \eqref{FIT4}, with a fixed $u_2 \in [u_1,v_0)$, for all $u \in [u_2,v_0)$, we have
\begin{equation}\label{FIT58}
\begin{aligned}
\gamma_0(u) - \gamma_0(u_2) &= \int^u_{u_2} \frac{(-\du r)_0}{r_0} \frac{(\mu - Q^2/r^2)_0}{1-\mu_0}(u')\;\D u' \geq \frac{1}{2}\int^u_{u_2} \frac{(-\du r)_0}{r_0} \frac{\mu_0}{1-\mu_0}(u')\;\D u' \\
&\quad\quad\geq \frac{1}{2}\int^u_{u_2} \frac{(-\du r)_0(u')}{r_0(u')} \frac{\frac{2m_0(u)}{r_0(u')}}{1-\frac{2m_0(u)}{r_0(u')}}\;\D u' = \frac{1}{2}\log \ml \frac{r_0(u')}{r_0(u')-2m_0(u)}\mr|^{u' = u}_{u' = u_2} \\
&\quad\quad\quad\quad= \frac{1}{2}\log \ml \frac{1-\mu_0(u)\frac{r_0(u)}{r_0(u_2)}}{1-\mu_0(u)}\mr.
\end{aligned}
\end{equation}
This in turn implies that
\begin{equation}\label{FIT222a}
\begin{aligned}
e^{2(\gamma_0(u)-\gamma_0(u_2))} - 1 &\geq \frac{\mu_0(u)}{1-\mu_0(u)}\ml 1 - \frac{r_0(u)}{r_0(u_2)} \mr \\
\end{aligned}
\end{equation}
and hence
\begin{equation}\label{FIT222}
0 \leq \mu_0(u) \leq \frac{e^{2 \ml \gamma_0(u) - \gamma_0(u_2)\mr}-1}{1-\frac{r_0(u)}{r_0(u_2)}}.
\end{equation}
Since $r_0(u) \rightarrow 0$ as $u \rightarrow v_0^-$, for each $u_2$, we can select $u(u_2) > u_2$ such that $r_0(u)/r_0(u_2) = 1/2$. Furthermore, by the continuity of $r_0$, we have that $u_2 \rightarrow v_0^-$ implies $u(u_2) \rightarrow v_0^-$. In addition, as $\gamma_0$ is non-decreasing with respect to $u$ and is bounded by hypothesis, this implies that 
$\lim_{u \rightarrow v_0^-} \gamma_0(u)$ exists and thus 
\begin{equation}
\lim_{u_2 \rightarrow v_0^-} \ml \gamma_0(u(u_2)) - \gamma_0(u_2) \mr = 0.
\end{equation}
By \eqref{FIT222}, it in turn yields that 
\begin{equation}
\lim_{u \rightarrow v_0^-}\mu_0(u)=0.
\end{equation}
\end{proof}

 In below we will see that the blueshift effect in the regular region would trigger our established trapped surface formation criterion. We now start to prove Theorem \ref{FITTheorem}.
\begin{proof}
Suppose that the conclusion in the statement of Theorem \ref{FITTheorem} is not true. Then either $\mathcal{A}$ is empty or there exists a sufficiently small open neighborhood $U$ of $\mathcal{O}'$ (with $\mathcal{O}' \subsetneq \mathcal{B}_0$) such that $U \cap \mathcal{A} = \emptyset$.

If the former holds, by Lemma \ref{FITLemma2}, we deduce that $\mathcal{B}_0$ is an outgoing null segment originating from $\mathcal{O}'$ and extends all the way to $v \rightarrow \infty$. In particular, there exists an $\varepsilon > 0$ and a corresponding outgoing null segment $\{v_0\} \times [v_0,v_0+\varepsilon]$ in $\mathcal{B}_0$. This enables us to find $u_1 \in (0,v_0)$ sufficiently close to $v_0$, such that the coordinate patch $([u_1,v_0] \times [v_0,v_0+\varepsilon])$ satisfies $([u_1,v_0] \times [v_0,v_0+\varepsilon]) \cap \mathcal{A} = \emptyset$. 

If the latter holds, by an analogous argument, we can also find an $\varepsilon > 0$ and a $u_1 \in (0,v_0)$ sufficiently close to $v_0$ such that $([u_1,v_0) \times [v_0,v_0+\varepsilon]) \cap \mathcal{A} = \emptyset$. Hence, for both cases, we define this coordinate patch to be $\mathcal{P}= [u_1,v_0) \times [v_0, v_0 + \varepsilon]$ and $\mathcal{P}\cap \mathcal{A} = \emptyset$. 
Figure \ref{FITFig1} below illustrates the existence of $\mathcal{P}$ in both cases. In view of Proposition \ref{SetupProp3}, we deduce that $\mathcal{P} \subset \mathcal{R}$. This allows us to carry out the estimates and lemmas established as above.

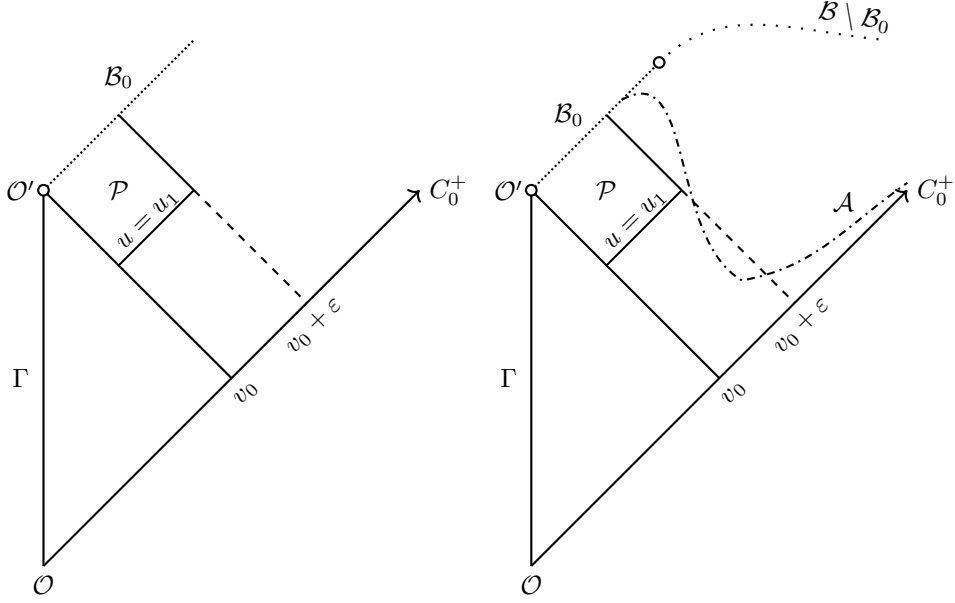
\begin{figure}[htbp]
\begin{tikzpicture}[
    dot/.style = {draw, fill = white, circle, inner sep = 0pt, minimum size = 4pt}]
\begin{scope}[thick]
\draw[thick] (0,0) node[anchor=north]{$\mathcal{O}$} --  (0,5);
\node[] at (-0.3,2.5) {$\Gamma$};
\node[] at (-0.3,5.0) {$\mathcal{O}'$};
\draw[->] (0,0) -- (5,5) node[anchor = west]{$C_0^+$};
\draw[thick, densely dotted] (0,5) to (2,7);
\node[] at (1,6.5) {$\mathcal{B}_0$};
\draw[thick] (0,5) --  (2.5,2.5);
\node[rotate = 45] at (2.7,2.3){$v_0$};
\draw[thick] (1,4) --  (2,5);
\draw[thick] (2,5) --  (1,6);
\node[] at (1,5){$\mathcal{P}$};
\draw[fill=white] (0,5) circle (2pt);
\node[rotate=45] at (1.4,4.6){$u = u_1$};
\draw[dashed] (2,5) -- (3.5,3.5);
\node[rotate = 45] at (3.6,3.2){$v_0 + \varepsilon$};
\end{scope}
\end{tikzpicture}
\begin{tikzpicture}[
    dot/.style = {draw, fill = white, circle, inner sep = 0pt, minimum size = 4pt}]
\begin{scope}[thick]
\draw[thick] (0,0) node[anchor=north]{$\mathcal{O}$} --  (0,5);
\node[] at (-0.3,2.5) {$\Gamma$};
\node[] at (-0.3,5.0) {$\mathcal{O}'$};
\draw[->] (0,0) -- (5,5) node[anchor = west]{$C_0^+$};
\draw[thick, densely dotted] (0,5) to (1.7,6.7);
\draw[thick, loosely dotted] (1.7,6.7) to [out=45,in=175] (4.7,7.0);
\draw[fill=white] (1.7,6.7) circle (2pt);
\node[] at (0.5,6) {$\mathcal{B}_0$};
\node[rotate=-10] at (4.3,7.3) {$\mathcal{B}\setminus\mathcal{B}_0$};
\draw[thick, dashdotted] (1.2,6.2) to [out=30,in=150]  (2.8,3.8) [out=10] to [in=215] (5,5.1);
\node[] at (4.15,4.85){$\mathcal{A}$};
\draw[thick] (0,5) --  (2.5,2.5);
\node[rotate = 45] at (2.7,2.3){$v_0$};
\draw[thick] (1,4) --  (2,5);
\draw[thick] (2,5) --  (1,6);
\node[] at (1,5){$\mathcal{P}$};
\node[rotate=45] at (1.4,4.6){$u = u_1$};
\draw[dashed] (2,5) -- (3.5,3.5);
\node[rotate=45] at (3.6,3.2){$v_0 + \varepsilon$};
\draw[fill=white] (0,5) circle (2pt);
\end{scope}
\end{tikzpicture}
\caption{The existence of the coordinate patch $\mathcal{P}$ in both cases.}
\label{FITFig1}
\end{figure}

\newpage

Similar to \eqref{ReducedMassRatio}, for $(u,v) \in \mathcal{P}$, we first define $\delta(u,v)$ and $\eta^*(u,v)$ as
\begin{equation}\label{deltauv}
    \delta(u,v) := \frac{r(u,v) - r(u,v_0)}{r(u,v_0)},
\end{equation}
\begin{equation}\label{eta*uv}
    \eta^*(u,v) := \frac{2\ml m(u,v)-m(u,v_0) \mr}{r(u,v)} - \frac{2}{r(u,v)} \int^v_{v_0} \frac{Q^2(\dv r)}{2r^2} (u,v') \mathrm{d}v'.
\end{equation}
By Theorem \ref{Trapped}, with $u_1$ taking the role of $\overline{u}$, considering $v_3$ as $v_0 + \varepsilon$, regarding $v_1$ as $v_0$, and treating $(u_0,v_2)$ as $(u,v) \in \mathcal{P}$, with $\delta(u,v) < 1$ we then have
\begin{equation}\label{FIT34}
\eta^*(u,v) \leq C_{24}(\chi_1) \frac{\delta(u,v)}{1+\delta(u,v)}\log\ml \frac{1}{\delta(u,v)}\mr + C_{23}(\chi_1)\frac{\delta(u,v) r(u,v_0)}{1-\sup_{{v \in [v_0,v_0+\varepsilon]}}\mu(u,v)}
\end{equation}
for all $(u,v) \in \widehat{\mathcal{P}}$ with $\widehat{\mathcal{P}}$ being a subset of $\mathcal{P}$. And via \eqref{TSF67}, $\widehat{\mathcal{P}} \subset \mathcal{P}$ is defined as \begin{equation}\label{FIT223}
\widehat{\mathcal{P}} := \left\{(u,v) \in \mathcal{P}: v - u \leq \frac{1}{C_1^2(u_1,v_0+\varepsilon;\chi_1)r(u_1,v_0+\varepsilon)} \right\}.
\end{equation}
Here we choose the parameter $\chi_1 \in \left(0,\frac{1}{3}\right]$ and note that the constants $C_{24}$ and $C_{23}$ are the only terms depending on $\chi_1$. For simplicity, we can pick $\chi_1 = \frac{1}{3}$ and denote 
\begin{equation}\label{FIT228}
C_{31} := C_1^2\ml u_1,v_0+\varepsilon;\frac{1}{3}\mr \cdot r(u_1,v_0+\varepsilon).
\end{equation}
As we have seen in previous sections, obtaining the correct sign for the function $\mu-Q^2/r^2$ is crucial. This motivates a further restriction on $\widehat{\mathcal{P}}$. Fix a parameter $\xi \in \ml 0,\frac{1}{2}\mr$ and choose $\Tilde{\mathcal{P}}$ as
\begin{equation}\label{FIT225}
\Tilde{\mathcal{P}} := \left\{(u,v) \in \mathcal{P}: v - v_0 \leq \frac{\xi}{2C_{31}}, u - v_0 \leq \frac{\xi}{2C_{31}} \right\}.
\end{equation}
Recall that by Proposition \ref{PropEstimates} via \eqref{Estimates17} we have
\begin{equation}\label{FIT224}
\frac{Q^2}{r^2}(u,v) \leq C_1^2 \ml u_1,v_0+\varepsilon; \frac{1}{3}\mr \cdot r(u_1,v_0+\varepsilon) \cdot (v - u) \cdot \mu(u,v).
\end{equation}
Consequently, for all $(u,v) \in \Tilde{\mathcal{P}}$ we obtain
\begin{equation}\label{FIT226}
\frac{Q^2}{r^2}(u,v) \leq \xi \mu(u,v) \leq \xi.
\end{equation}
In particular, we have
\begin{equation}\label{FIT227}
\ml \mu - \frac{Q^2}{r^2}\mr(u,v) \geq (1-\xi) \mu(u,v) \geq \frac{\mu(u,v)}{2} \geq 0.
\end{equation}
Motivated from \eqref{FIT225}, we also define $\tilde{\varepsilon} := \frac{\xi}{2C_{31}}$ representing the maximal width of $\tilde{\mathcal{P}}$. 

Within $\Tilde{\mathcal{P}}$ we then find a further refinement $\mathcal{P}'$. And the following lemma holds.
\begin{lemma}\label{FITEstLemma8} Along $v = v_0$, we define
\begin{equation}\label{FIT78}
    \gamma_s(u) := \int^u_{0} \frac{(- \du r)}{r} \frac{\mu}{1 - \mu} (u',v) \mathrm{d}u'
\end{equation}
and 
\begin{equation}\label{FIT79}
    \gamma_{s,0}(u) := \gamma_s(u,v_0).
\end{equation}
For each $\xi \in \left(0,\frac{1}{2}\right)$, there exists $u_2 \in [u_1,v_0)$ and in the corresponding refinement 
\begin{equation}\label{P'}
\mathcal{P}' = [u_2,v_0) \times [v_0,v_0+\tilde{\varepsilon}]
\end{equation}
of $\Tilde{\mathcal{P}}$, for any $u \in [u_2,v_0)$, we have
\begin{equation}\label{FIT77}
    \gamma_0(u) \leq \gamma_{s,0}(u) \leq (1+2\xi) \gamma_0(u).
\end{equation}
\end{lemma}
\remark{ The $s$ in the subscript of $\gamma_{s,0}(u)$ represents the corresponding $\gamma$ defined for the uncharged scalar field system. One can compare this to our definition of $\gamma$ in \eqref{FIT3}, for which the term $Q^2/r^2$ is dropped in \eqref{FIT78}.}

\begin{proof}
To prove \eqref{FIT77}, we first note that for any choice of $u_2 \in [u_1,v_0)$, inequality $\gamma_0 (u) \leq \gamma_{s,0}(u)$ follows immediately from definitions in \eqref{FIT3} and \eqref{FIT78}. For the upper bound of $\gamma_{s,0}(u)$, under the assumption of Theorem \ref{FITTheorem}, we have $\gamma_0(u) \rightarrow \infty$ as $u \rightarrow v_0^-$. This signals the existence of $u_2' \in [u_1,v_0)$ satisfying $0 < \gamma_0(u_2') \leq \gamma_{s,0}(u_2') < \infty$. Here the finiteness of $\gamma_{s,0}(u_2')$ comes from the fact that the only possible singularity for the integrand of $\gamma_{s,0}$ occurs at $\mathcal{O}'$. On the other hand, by \eqref{FIT227}, we note that
\begin{equation}\label{FIT80}
\int_{u_2'}^u \frac{(- \du r)}{r} \frac{\mu - Q^2/r^2}{1 - \mu} (u',v_0) \mathrm{d}u' \leq \int_{u_2'}^u \frac{(- \du r)}{r} \frac{\mu}{1 - \mu} (u',v_0) \mathrm{d}u' \leq \frac{1}{1-\xi} \int_{u_2'}^u \frac{(- \du r)}{r} \frac{\mu - Q^2/r^2}{1 - \mu} (u',v_0) \mathrm{d}u'.
\end{equation}
Consequently, as $u \rightarrow v_0^-$, since 
\begin{equation}\label{FIT81}
\int_{u_2'}^u \frac{(- \du r)}{r} \frac{\mu - Q^2/r^2}{1 - \mu} (u',v_0) \mathrm{d}u' \rightarrow \infty, 
\end{equation}
then we have
\begin{equation}
\int^u_{{u_2}'}\frac{(-\du r)}{r}\frac{\mu}{1-\mu}(u',v_0) \D u' \rightarrow \infty.
\end{equation}
We then decompose $\gamma_{s,0}(u)$ as
\begin{equation}\label{FIT100}
\gamma_{s,0}(u) = \int^u_{u_2'} \frac{(- \du r)}{r} \frac{\mu}{1 - \mu} (u',v_0) \mathrm{d}u' + \int^{u_2'}_{0} \frac{(- \du r)}{r} \frac{\mu}{1 - \mu} (u',v_0) \mathrm{d}u'.
\end{equation}
Observe that, for a fixed $u_2'$, the second integral is finite for a fixed $u_2'$. By \eqref{FIT80}, we further derive
\begin{equation}\label{FIT100A}
\gamma_{s,0}(u) \leq \frac{1}{1-\xi}\int^u_{u_2'} \frac{(- \du r)}{r} \frac{\mu - Q^2/r^2}{1 - \mu} (u',v_0) \mathrm{d}u' + \int^{u_2'}_{0} \frac{(- \du r)}{r} \frac{\mu}{1 - \mu} (u',v_0) \mathrm{d}u'.
\end{equation}
Since $\xi \in \ml 0,\frac{1}{2} \mr$, from \eqref{FIT81}, if we pick $u_2 \in [u_2',v_0)$ sufficiently close to $v_0$ such that
\begin{equation}
\frac{\xi(1-2\xi)}{1-\xi} \int^{u_2}_{u_2'} \frac{(- \du r)}{r} \frac{\mu - Q^2/r^2}{1 - \mu} (u',v_0) \mathrm{d}u' \geq \int^{u_2'}_{0} \frac{(- \du r)}{r} \frac{\mu}{1 - \mu} (u',v_0) \mathrm{d}u',
\end{equation}
then from \eqref{FIT100A}, for all $u \in [u_2,v_0)$, we have
\begin{equation}
\begin{aligned}
\gamma_{s,0}(u) &\leq \frac{1}{1-\xi}\int^u_{u_2'} \frac{(- \du r)}{r} \frac{\mu - Q^2/r^2}{1 - \mu} (u',v_0) \mathrm{d}u' + \int^{u_2'}_{0} \frac{(- \du r)}{r} \frac{\mu}{1 - \mu} (u',v_0) \mathrm{d}u' \\
&\leq \frac{1}{1-\xi}\int^u_{u_2'} \frac{(- \du r)}{r} \frac{\mu - Q^2/r^2}{1 - \mu} (u',v_0) \mathrm{d}u' + \frac{\xi(1-2\xi)}{1-\xi}\int^{u_2}_{u_2'} \frac{(- \du r)}{r} \frac{\mu}{1 - \mu} (u',v_0) \mathrm{d}u' \\
&\leq \frac{1}{1-\xi}\int^u_{u_2'} \frac{(- \du r)}{r} \frac{\mu - Q^2/r^2}{1 - \mu} (u',v_0) \mathrm{d}u' + \frac{\xi(1-2\xi)}{1-\xi}\int^{u}_{u_2'} \frac{(- \du r)}{r} \frac{\mu}{1 - \mu} (u',v_0) \mathrm{d}u' \\
&= (1+2\xi) \int^u_{u_2'} \frac{(- \du r)}{r} \frac{\mu - Q^2/r^2}{1 - \mu} (u',v_0) \mathrm{d}u' \\
&\leq (1+2\xi) \gamma_0(u),
\end{aligned}
\end{equation}
where the last inequality follows from $\gamma_0(u_2') > 0$.
\end{proof}

We continue the proof of Theorem \ref{FITTheorem}. Recalling the definition of $I$ in \eqref{FIT14}, since $I$ is complex-valued, we can decompose it into
\begin{equation}\label{FIT15}
        I(u) = X(u) + \ii Y(u)
\end{equation} 
with $X(u) = \re[I(u)]$ and $Y(u) = \im[I(u)]$. Let us further define the following:
\begin{equation}
\begin{aligned}
    X_- := \liminf_{u \rightarrow v_0^-} X(u), \qquad & \qquad X_+ := \limsup_{u \rightarrow v_0^-} X(u), \\ 
    Y_- := \liminf_{u \rightarrow v_0^-} Y(u), \qquad & \qquad Y_+ := \limsup_{u \rightarrow v_0^-} Y(u). \\ 
\end{aligned}
\end{equation}
\textit{A priori}, as the asymptotic behavior of $I$ is unknown, similar to \cite{christ4} for the uncharged case (dealing with only $X$), here we will have to consider all possible behaviors for $X$ and $Y$. For each separate $X$ or $Y$, denoted by $Z$, we have the following cases
\begin{enumerate}[(1)]
    \item $-\infty < Z_- = Z_+ < \infty$ (limit exists and is finite),
    \item $-\infty = Z_- = Z_+$ (limit exists and is equals to $- \infty$),
    \item $Z_- = Z_+ = \infty$ (limit exists and is equals to $+ \infty$),
    \item $-\infty < Z_- < Z_+ < \infty$ (limit does not exist, but the sequence is bounded),
    \item $-\infty < Z_- < Z_+ = \infty$ (limit does not exist),
    \item $-\infty = Z_- < Z_+ < \infty$ (limit does not exist),  
    \item $-\infty = Z_- < Z_+ = \infty$ (limit does not exist).
\end{enumerate}
For brevity, we denote case $(a,b)$ to refer to the case where $a$ corresponds to one of the seven scenarios as listed above for $X$, while $b$ corresponds to that for $Y$. For instance, case $(1,6)$ corresponds to the scenario, whereby the limit of $X(u)$ exists as $u \rightarrow v_0^-$, while the limit for $Y(u)$ does not exist, with $Y_- = -\infty$ and $Y_+ < \infty$.

Analogous to \cite{christ4} with 7 cases there, in this paper we characterize our $49$ cases into the following dichotomy
\begin{itemize}
\item[(I)] $I(u_n)$ is bounded for all possible sequences $\{u_n\}_n$ approaching to $v_0^-$. These correspond to cases $(1,1), (1,4), (4,1), (4,4)$.
\item[(II)] There exists a sequence $\{u_n\}_n$  such that $|I(u_n)| \rightarrow \infty$ as $u_n \rightarrow v_0$. These correspond to the remaining $45$ cases.
\end{itemize}
For \textit{Scenario (I)}, we proceed as follows. For case $(1,1)$, since the limit for $X$ and $Y$ both exist, we deduce that the limit for the complex-valued $I$ also exists. By the hypothesis \eqref{FITTheorem1}, we can set
\begin{equation}\label{FIT21}
    h := \mlm \ml \frac{r\dv \phi}{\dv r}\mr_0(0) - l\mrm \neq 0
\end{equation}
with
\begin{equation}\label{FIT20}
    l = \lim_{u \rightarrow v_0^-} I(u).
\end{equation}      
For $u_*$ sufficiently close to $v_0$ and for all $u \in [u_*,v_0)$, we then have
\begin{equation}\label{FIT22}
    |I(u)-l| \leq \frac{2h}{3}.
\end{equation}
Hence, for each $u \in [u_*,v_0)$, by applying \eqref{FIT13}, we have
\begin{equation}\label{FIT23}
\begin{aligned}
\ml \frac{r |\dv \phi|}{\dv r} \mr_0(u) 
=& \mlm \ml\frac{r \dv \phi}{\dv r} \mr_0(0) - l + l - I(u) \mrm e^{\gamma_0(u)}\\ 
\geq& \mlm \mlm \ml\frac{r \dv \phi}{\dv r} \mr_0(0) - l\mrm - |l - I(u)|  \mrm e^{\gamma_0(u)} \geq \frac{h}{3}e^{\gamma_0(u)}.
\end{aligned}
\end{equation}
For cases $(1,4), (4,1), (4,4)$, we denote
\begin{equation}\label{FIT24}
    h_X := X_+ - X_- \quad \text{ and } \quad h_Y := Y_+ - Y_-.
\end{equation}
By the definitions of $\limsup$ and $\liminf$, we must have the following
\begin{equation}\label{FIT25}
\begin{aligned}
&\max\left\{\mlm \re\ml \frac{r \dv \phi}{\dv r}\mr_0(0) - X_- \mrm, \mlm \re\ml \frac{r \dv \phi}{\dv r}\mr_0(0) - X_+ \mrm \right\}  \geq \frac{h_X}{2}, \\
&\max\left\{\mlm \im\ml \frac{r \dv \phi}{\dv r}\mr_0(0) - Y_- \mrm, \mlm \im\ml \frac{r \dv \phi}{\dv r}\mr_0(0) - Y_+ \mrm \right\} \geq \frac{h_Y}{2}.\\
\end{aligned}
\end{equation}
For each case, we denote $X_*$ and $Y_*$ to be the quantities ($X_* = X_+$ or $X_-$ and $Y_* = Y_+$ or $Y_-$) that the maximums in \eqref{FIT25} are attained. With those, we can find an increasing sequence $\{u_n\}$ with $u_n \rightarrow v_0$, and for $u_n$ sufficiently close to $v_0$ we have
\begin{equation}\label{FIT27}
\begin{aligned}
&\; \mlm \ml \frac{r \dv \phi}{\dv r}\mr_0(0) - I(u_n)\mrm \\
=&\; \mlm \re \ml \frac{r \dv \phi}{\dv r}\mr_0(0) - X_* + X_* - X(u_n) + \ii \ml \im \ml \frac{r \dv \phi}{\dv r}\mr_0(0) - Y_* + Y_* - Y(u_n)\mr \mrm \\
=&\; \sqrt{\mlm \re \ml \frac{r \dv \phi}{\dv r}\mr_0(0) - X_* + X_* - X(u_n)\mrm^2 + \mlm \im \ml \frac{r \dv \phi}{\dv r}\mr_0(0) - Y_* + Y_* - Y(u_n) \mrm^2} \\
\geq &\; \sqrt{\mlm \mlm \re \ml \frac{r \dv \phi}{\dv r}\mr_0(0) - X_*\mrm -  \mlm X_* - X(u_n)\mrm \mrm^2 + \mlm \mlm \im \ml \frac{r \dv \phi}{\dv r}\mr_0(0) - Y_* \mrm - \mlm Y_* - Y(u_n) \mrm \mrm^2} \\
\geq &\; \max \left\{ \mlm \mlm \re \ml \frac{r \dv \phi}{\dv r}\mr_0(0) - X_*\mrm -  \mlm X_* - X(u_n)\mrm \mrm, \mlm \mlm \im \ml \frac{r \dv \phi}{\dv r}\mr_0(0) - Y_* \mrm - \mlm Y_* - Y(u_n) \mrm \mrm \right\} \\
\geq &\; \max \left\{ \frac{h_X}{3}, \frac{h_Y}{3}\right\}= \frac{h}{3} > 0.
\end{aligned}
\end{equation}
Note that for the last step in \eqref{FIT27}, we use the fact that  $h:= \max\{h_X,h_Y\} > 0$. This is true since both $h_X$ and $h_Y$ cannot be simultaneously zero (otherwise, we fall back to case $(1,1)$). Thus, in the above scenarios, we can find an increasing sequence $\{u_n\}$ approaching $v_0$ such that
\begin{equation}\label{FIT28}
\ml \frac{r |\dv \phi|}{\dv r} \mr_0(u_n)  \geq \frac{h}{3}e^{\gamma_0(u_n)}.
\end{equation}
In addition, since $I$ is bounded in these cases, by \eqref{FIT13}, we have that $\ml \frac{r|\dv \phi|}{\dv r}\mr_0(u) e^{-\gamma_0(u)}$ is bounded. Henceforth, we can define the following quantity
\begin{equation}\label{FIT207}
b := \sup_{u\in [0,v_0)} \ml \frac{r|\dv \phi|}{\dv r}\mr_0(u)e^{-\gamma_0(u)} < +\infty.
\end{equation}

Given our sequence $u_n$ as mentioned above, for each $n$, we will pick a corresponding value of $v$ (denoted by $v(u_n)$) such that $(u_n, v(u_n))$ is in the refined region and a few geometric inequalities hold. We summarize these into a couple of lemmas. 

First, we have the following lemma.
\begin{lemma}\label{FITEstLemma4}
Consider a fixed $u \in [u_2, v_0)$ and recall region $\mathcal{P}'$ in \eqref{P'}. For $(u,v_1), (u,v_2) \in \mathcal{P}'$ with $v_1 < v_2$, we have
\begin{equation}\label{FIT45}
    \delta(u,v_1) \leq \delta(u,v_2).
\end{equation}
\end{lemma}
\begin{proof}
This follows from the fact that $\dv r \geq 0$ in $\mathcal{P}' \subset \mathcal{R}$.
\end{proof}

Next, by leveraging on Lemma \ref{FITEstLemma1}, we have the following improvement
\begin{lemma}\label{FITEstLemma3}
Given any sequence $u_n \in [u_2,v_0)$, for each $n$, suppose that there exists a corresponding value $v(u_n) \in [v_0,v_0 + \te]$ such that
\begin{equation}\label{FIT35a}
\delta\ml u,v(u_n)\mr \log \ml \frac{1}{\delta(u,v(u_n))}\mr \leq C_{33}e^{-2\gamma_0(u)}
\end{equation} with $C_{33}$ given in \eqref{FIT44} and $\delta(u,v(u_n)) \leq 1$ for all $u \in [u_2,u_n]$.
Then, for all $v \in [v_0,v(u_n)]$, we have
\begin{equation}\label{FIT36}
\frac{1}{1-\mu(u,v)} \leq \frac{3}{2}\frac{1}{1-\mu_0(u)} \leq \frac{3}{2}\frac{1}{1-\mu_0(0)}e^{\gamma_0(u)}.
\end{equation}
\end{lemma}
\begin{proof} 
We prove this lemma via a bootstrapping argument. Consider the following set
\begin{equation}\label{FIT37}
 S_1(u) = \left\{ v'' \in [v_0,v(u)]: \frac{1}{1-\mu(u,v')} \leq \frac{3}{2}\frac{1}{1-\mu_0(u)} \; \text{ for all }v'\in[v_0,v''] \right\}.
\end{equation}
It is clear that $S_1(u)$ is closed. Furthermore, by Lemma \ref{FITEstLemma1}, we have
\begin{equation}\label{FIT38}
    \frac{1}{1-\mu(u,v_0)} = \frac{1}{1-\mu_0(u)} \leq \frac{3}{2}\frac{1}{1-\mu_0(u)}.
\end{equation}
This indicates that $v_0 \in S_1(u)$ and thus $S_1(u)$ is non-empty. It remains to show that $S_1(u)$ is open. To do so, we rewrite $\eta^*$ in \eqref{eta*uv} as
\begin{equation}\label{FIT39A}
\eta^*(u,v) = \mu(u,v) - \mu_0(u)\frac{r(u,v_0)}{r(u,v)} - \frac{2}{r(u,v)}\int^v_{v_0} \frac{Q^2 \dv r}{2r^2}(u,v') \mathrm{d}v'.
\end{equation}
Let $v \in [v_0,v(u)]$. To obtain an estimate for $\mu(u,v)$, by rearranging the terms in \eqref{FIT39A} and by \eqref{FIT34}, \eqref{FIT226}, we derive
\begin{equation}\label{FIT39}
\begin{aligned}
\mu(u,v) \leq& \; C_{24} \frac{\delta(u,v)}{1+\delta(u,v)}\log\ml \frac{1}{\delta(u,v)}\mr + C_{23}\frac{\delta(u,v)r_0(u)}{1-\sup_{{v \in [v_0,v_0+\varepsilon]}}\mu(u,v)} \\
&+ \mu_0(u)\frac{r(u,v_0)}{r(u,v)} + \frac{2\xi}{r(u,v_0)}\int^v_{v_0} (\dv r)(u,v')\D v'\\
\leq &\; C_{24} \delta\ml u,v(u)\mr\log\ml \frac{1}{\delta\ml u,v(u)\mr}\mr + \frac{3 C_{23}}{2}\frac{r_0(0)}{1-\mu_0(0)}e^{\gamma_0(u)}\delta\ml u,v(u)\mr + \mu_0(u) + 2\xi\delta\ml u,v(u)\mr  \\
\leq &\; C_{32}\delta\ml u,v(u)\mr\log \ml \frac{1}{\delta\ml u,v(u)\mr}\mr e^{\gamma_0(u)} + \mu_0(u)
\end{aligned}
\end{equation}
with 
\begin{equation}\label{FIT41}
\begin{aligned}
C_{32} := C_{24} + \frac{3C_{23}}{2}\frac{r_0(0)}{1-\mu_0(0)} + 2\xi.
\end{aligned}
\end{equation}
Here we use $\log\ml \frac{1}{\delta(u,v)}\mr \geq 1$, $v \leq v(u)$, inequality \eqref{FIT226} and Lemma \ref{FITEstLemma4}. We can rewrite \eqref{FIT39} as 
\begin{equation}\label{FIT42}
\begin{aligned}
1-\mu(u,v) &\geq 1 - C_{32}\delta \ml u,v(u_n)\mr \log \ml\frac{1}{\delta \ml u,v(u_n)\mr} \mr e^{\gamma_0(u)} - \mu_0(u).\\
\end{aligned}
\end{equation}
If we further demand that
\begin{equation}\label{FIT43}
\begin{aligned}
\delta(u,v(u_n)) \log \ml \frac{1}{\delta(u,v(u_n))}\mr &\leq C_{33}e^{-2\gamma_0(u)},
\end{aligned}
\end{equation}
with
\begin{equation}\label{FIT44}
C_{33} := \frac{1 - \mu_0(0)}{8 C_{32}},
\end{equation}
we then get $1 - \mu(u,v) \geq \frac{7}{8}(1-\mu_0(u))$, which implies 
\begin{equation}\label{7.71 new}
    \frac{1}{1-\mu(u,v)} \leq \frac{8}{7}\frac{1}{1-\mu_0(u)}.
\end{equation}
With \eqref{7.71 new} we improve the estimate in \eqref{FIT37} and hence show that $S_1(u)$ is open. Thus, the estimate holds for all $v' \in [v_0,v(u_n)]$. Inequality \eqref{FIT36} then follows from \eqref{FIT5}.
\end{proof}

Lemma \ref{FITEstLemma4} describes the evolution of $\delta(u,v)$ along the $u$-direction. In contrast, the following lemma describes the evolution of $\delta(u,v)$ along the $v$-direction.
\begin{lemma}\label{FITEstLemma5}  Given any sequence $u_n \in [u_2,v_0)$, for each $n$, suppose that there exists a corresponding $v(u_n) \in [v_0,v_0 + \te]$, such that, for all $u \in [u_2,u_n]$, both
\begin{equation}\label{FIT35c}
\delta(u,v(u_n)) \leq \min\left\{ \oh,\frac{2(1 - \mu_0(0))e^{-\gamma_0(u)}}{3}, \frac{e^{-\frac{3}{2}(1+2\xi)\gamma_0(u)}}{2 C_{34}} \right\}
\end{equation} and \eqref{FIT35a} hold, with $C_{34}$ defined in \eqref{FIT71}. Then, for any $v \in (v_0,v(u_n)]$ and $u_3,u_4 \in [u_2,u_n]$ with $u_3 < u_4$, we have
\begin{equation}\label{FIT83}
\delta(u_3,v') \leq 4 e^{-\frac{1}{2}\log\ml \frac{r_0(u_3)}{r_0(u_4)}\mr + \frac{3}{2}(1+2\xi)\gamma_{0}(u_4)} \delta(u_4,v').
\end{equation}
\end{lemma}

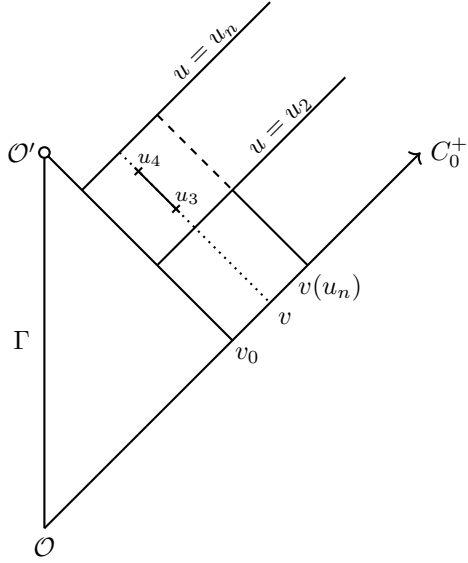
\begin{figure}[htbp]
\begin{tikzpicture}[
    dot/.style = {draw, fill = white, circle, inner sep = 0pt, minimum size = 4pt}]
\begin{scope}[thick]
\draw[thick] (0,0) node[anchor=north]{$\mathcal{O}$} --  (0,5);
\node[] at (-0.3,2.5) {$\Gamma$};
\node[] at (-0.3,5.0) {$\mathcal{O}'$};
\draw[->] (0,0) -- (5,5) node[anchor = west]{$C_0^+$};
\draw[thick] (0,5) --  (2.5,2.5);
\node[] at (2.7,2.3){$v_0$};
\draw[dashed] (1.5,5.5) -- (2.5,4.5);
\draw[thick] (2.5,4.5) -- (3.5,3.5);
\node[] at (3.8,3.2){$v(u_n)$};
\draw[dotted] (1,5) -- (3,3);
\node[] at (3.2,2.8){$v$};
\draw[thick] (0.5,4.5) -- (3,7);
\node[rotate=45] at (2.15,6.35){$u = u_n$};
\draw[thick] (1.5,3.5) -- (4,6);
\node[rotate=45] at (3.15,5.35){$u = u_2$};
\draw[thick] (1.75,4.25) -- (1.25,4.75);
\draw[thick] (1.7,4.2) -- (1.8,4.3);
\draw[thick] (1.7,4.3) -- (1.8,4.2);
\draw[thick] (1.2,4.7) -- (1.3,4.8);
\draw[thick] (1.2,4.8) -- (1.3,4.7);
\node[] at (1.4,4.9){\footnotesize $u_4$};
\node[] at (1.9,4.4){\footnotesize $u_3$};
\draw[fill=white] (0,5) circle (2pt);
\end{scope}
\end{tikzpicture}
\caption{The relative positions as described in Lemma \ref{FITEstLemma5}.}
\label{FITFig2}
\end{figure}

\begin{proof} By Lemma \ref{FITEstLemma4} and \eqref{FIT35c}, we have that $\delta(u,v) \leq \frac{1}{2}$ for all $u \in [u_3,u_4] \subset [u_2,u_n]$ and $v \in (v_0,v(u_n)]$. Consequently, by the inequality $\frac{1}{2} x \leq \log(1+x) \leq x$ for $0 \leq x \leq \frac{1}{2}$, for all $u \in [u_3,u_4]$ and $v \in (v_0,v(u_n)]$, we get 
\begin{equation}\label{FIT47}
    \frac{1}{2}\delta(u,v) \leq \log \ml \frac{r(u,v)}{r_0(u)}\mr \leq \delta(u,v).
\end{equation} 
Furthermore, by \eqref{Setup20}, we also have
\begin{equation}\label{FIT48A}
\dv \ml \frac{\du r}{r}\mr = \frac{\du r \dv r}{r^2} \ml \frac{\mu - Q^2/r^2}{1-\mu} - 1\mr.
\end{equation}
Consider any $(u,v) \in [u_2,u_n] \times (v_0,v(u_n)]$. By \eqref{FIT48A} and \eqref{FIT36} of Lemma \ref{FITEstLemma3}, we obtain
\begin{equation}\label{FIT48}
\begin{aligned}
\dv \ml \log \ml \frac{-\du r}{r}\mr\mr &= \frac{\dv r}{r}\ml \frac{\mu - Q^2/r^2}{1-\mu} - 1\mr \leq \frac{\dv r}{r} \ml \frac{\mu}{1-\mu} - 1\mr \\
&\quad\quad= \frac{\dv r}{r} \ml \frac{1}{1-\mu} - 2\mr \leq \frac{\dv r}{r}\ml \frac{3}{2}\frac{1}{1-\mu_0}-2\mr.
\end{aligned}
\end{equation}
Next, we perform an integration from $v = v_0$ to $v = v$ and derive
\begin{equation}\label{FIT49A}
\log \ml \frac{-\frac{\du r}{r}(u,v)}{\ml -\frac{\du r}{r} \mr_0(u)}\mr \leq \ml \frac{3}{2}\frac{1}{1-\mu_0(u)} - 2\mr \log \ml \frac{r(u,v)}{r_0(u)}\mr,
\end{equation}
which implies
\begin{equation}\label{FIT49}
\begin{aligned}
    \ml - \frac{\du r}{r}\mr(u,v) &\leq \ml -\frac{\du r}{r}\mr_0(u) e^{\ml \frac{3}{2}\frac{1}{1-\mu_0(u)} - 2\mr \log \ml \frac{r(u,v)}{r_0(u)}\mr}.
\end{aligned}
\end{equation}
By our assumption in \eqref{FIT35c}, for $u \in [u_2,u_n]$, we have 
\begin{equation}\label{FIT35d}
\delta(u,v(u_n)) \leq \frac{2(1 - \mu_0(0))}{3}e^{-\gamma_0(u)}.
\end{equation}
Hence, by \eqref{FIT45}, \eqref{FIT38}, \eqref{FIT47}, \eqref{FIT35d},  the exponent in the right hand side of \eqref{FIT49} can be bounded from above by
\begin{equation}\label{FIT51}
\begin{aligned}
\ml \frac{3}{2}\frac{1}{1-\mu_0(u)} - 2\mr \log \ml \frac{r(u,v)}{r_0(u)}\mr &\leq \frac{3}{2(1-\mu_0(u))} \delta(u,v) \leq \frac{3}{2(1-\mu_0(0))}e^{\gamma_0(u)} \delta(u,v(u_n)) \leq 1.
\end{aligned}
\end{equation} 
Combining \eqref{FIT49} and \eqref{FIT51} implies
\begin{equation}\label{FIT52}
\ml - \frac{\du r}{r}\mr(u,v) \leq 3\ml -\frac{\du r}{r}\mr_0(u).
\end{equation}
On the other hand, we can rewrite \eqref{FIT49} as
\begin{equation}\label{FIT53}
\ml \frac{- \du r}{r}\mr (u,v) - \ml \frac{- \du r}{r}\mr_0(u) \leq \ml e^{\ml \frac{3}{2}\frac{1}{1-\mu_0(u)} - 2\mr \log \ml \frac{r(u,v)}{r_0(u)}\mr} - 1 \mr \ml \frac{-\du r}{r}\mr_0.
\end{equation}
This form is of great use if we calculate
\begin{equation}\label{FIT54}
\begin{aligned}
-\du \log\ml \frac{r(u,v)}{r_0(u)}\mr  &= -\frac{\frac{\du r}{r_0}(u,v) - \frac{r(u,v)}{r_0(u)^2}(\du r)_0(u)}{\frac{r(u,v)}{r_0(u)}} =\frac{- \du r}{r}(u,v) - \ml\frac{- \du r}{r}\mr_0(u) \\
&\quad\quad\leq \ml e^{\ml \frac{3}{2}\frac{1}{1-\mu_0(u)} - 2\mr \log \ml \frac{r(u,v)}{r_0(u)}\mr} - 1 \mr \ml \frac{-\du r}{r}\mr_0.
\end{aligned}
\end{equation}
For each $(u,v) \in [u_2,u_n] \times (v_0,v(u_n)]$, we define
\begin{equation}\label{deltatilde}
\td(u,v) := \log \ml \frac{r(u,v)}{r_0(u)}\mr.
\end{equation}
By employing the inequality $e^x - 1 \leq x + (e-2)x^2$ for $ x \leq 1$, inequality \eqref{FIT54} yields 
\begin{equation}\label{FIT55A}
-\du \td \leq \ml \ml \frac{3}{2(1-\mu_0(u))} - 2 \mr \td + (e-2)\ml \frac{3}{2(1-\mu_0(u))} - 2 \mr^2\td^2 \mr \ml \frac{-\du r}{r}\mr_0 \\
\end{equation}
and
\begin{equation}\label{FIT55}
\begin{aligned}
\du \ml \frac{1}{\td}\mr &\leq \ml \ml \frac{3}{2(1-\mu_0(u))} - 2 \mr\ml \frac{1}{\td}\mr + (e-2)\ml \frac{3}{2(1-\mu_0(u))} - 2 \mr^2 \mr\ml \frac{-\du r}{r}\mr_0.
\end{aligned}
\end{equation}
By multiplying a suitable integrating factor on both sides of \eqref{FIT55}, we arrive at
\begin{equation}\label{FIT56}
\begin{aligned}
&\du \ml \frac{1}{\td(u,v)}e^{-\int^u_{0} \ml \frac{-\du r}{r}\mr_0(u') \ml \frac{3}{2(1-\mu_0(u'))} - 2\mr \mathrm{d}u'}\mr \\ 
\leq & \; (e-2)\ml \frac{3}{2(1-\mu_0(u))} - 2 \mr^2 \ml \frac{-\du r}{r}\mr_0(u) e^{-\int^u_{0} \ml \frac{-\du r}{r}\mr_0(u') \ml \frac{3}{2(1-\mu_0(u'))} - 2\mr \mathrm{d}u'}.
\end{aligned}
\end{equation}
By performing integration along $v \in (u_0,v(u_2)]$ from $u_3$ to $u_4$, for $u_2 \leq u_3 < u_4 \leq u_n$, we obtain
\begin{equation}\label{FIT59}
\begin{aligned}
&\frac{e^{-\int^{u_4}_{0} \ml \frac{-\du r}{r}\mr_0(u') \ml \frac{3}{2(1-\mu_0(u'))} - 2\mr \mathrm{d}u'}}{\td(u_4,v)} - \frac{e^{-\int^{u_3}_{0} \ml \frac{-\du r}{r}\mr_0(u') \ml \frac{3}{2(1-\mu_0(u'))} - 2\mr \mathrm{d}u'}}{\td(u_3,v)} \\
\leq & \; (e-2)\int^{u_4}_{u_3} \ml \frac{3}{2(1-\mu_0(u''))} - 2 \mr^2 \ml \frac{-\du r}{r}\mr_0(u'') e^{-\int^{u''}_{0} \ml \frac{-\du r}{r}\mr_0(u') \ml \frac{3}{2(1-\mu_0(u'))} - 2\mr \mathrm{d}u'} \mathrm{d}u''.
\end{aligned}
\end{equation}
As the exponent in \eqref{FIT59} has the following upper bound
\begin{equation}\label{FIT60}
\begin{aligned}
-\int^{u''}_{0} &\ml \frac{-\du r}{r}\mr_0(u') \ml \frac{3}{2(1-\mu_0(u'))} - 2\mr \mathrm{d}u' = -\int^{u''}_{0} \ml \frac{-\du r}{r}\mr_0(u') \ml \frac{3}{2}\frac{\mu_0(u')}{1-\mu_0(u')} - \frac{1}{2}\mr \mathrm{d}u' \\
&\leq -\int^{u''}_{0} \ml \frac{-\du r}{r}\mr_0(u') \ml \frac{3}{2}\frac{(\mu - Q^2/r^2)_0(u')}{1-\mu_0(u')} - \frac{1}{2}\mr \mathrm{d}u' \leq -\frac{3}{2}\gamma_0(u'') + \frac{1}{2}\log \ml \frac{r_0(0)}{r_0(u'')}\mr,
\end{aligned}
\end{equation}
we can shorten \eqref{FIT59} to
\begin{equation}\label{FIT61}
\begin{aligned}
&\frac{e^{-\int^{u_4}_{0} \ml \frac{-\du r}{r}\mr_0(u') \ml \frac{3}{2(1-\mu_0(u'))} - 2\mr \mathrm{d}u'}}{\td(u_4,v)} - \frac{e^{-\int^{u_3}_{0} \ml \frac{-\du r}{r}\mr_0(u') \ml \frac{3}{2(1-\mu_0(u'))} - 2\mr \mathrm{d}u'}}{\td(u_3,v)} \\
\leq & \; (e-2)\int^{u_4}_{u_3} \ml \frac{3}{2(1-\mu_0(u'))} - 2 \mr^2 \ml \frac{-\du r}{r}\mr_0(u') e^{-\frac{3}{2}\gamma_0(u') + \frac{1}{2}\log \ml \frac{r_0(0)}{r_0(u')}\mr} \mathrm{d}u'.
\end{aligned}
\end{equation}
To further simplify the upper bound in \eqref{FIT61}, we first observe that, for all $u' \in [u_2,v_0)$, we have
\begin{equation}\label{FIT62}
\ml \frac{3}{2}\frac{1}{1-\mu_0(u')}-2 \mr^2 \leq \frac{9}{4}\ml \frac{\mu_0(u')}{1-\mu_0(u')}\mr^2 + \frac{1}{4}.
\end{equation}
This implies that we can analyze the contribution of the integral from two separate terms as follows. First, we have
\begin{equation}\label{FIT63}
\begin{aligned}
&\int^{u_4}_{u_3} \ml \frac{-\du r}{r}\mr_0(u') e^{-\frac{3}{2}\gamma_0(u') + \frac{1}{2}\log \ml \frac{r_0(0)}{r_0(u')}\mr} \mathrm{d}u' \\
\leq & \; 2e^{-\frac{3}{2}\gamma_0(u_3)} \int^{u_4}_{u_3} \frac{1}{2}\ml \frac{-\du r}{r}\mr_0(u') e^{\frac{1}{2}\log \ml \frac{r_0(0)}{r_0(u')}\mr} \mathrm{d}u' \\
\leq & \; 2e^{-\frac{3}{2}\gamma_0(u_3)}e^{\frac{1}{2}\log\ml \frac{r_0(0)}{r_0(u_4)}\mr},
\end{aligned}
\end{equation}
since inequality \eqref{FIT227} with $\xi \in \ml 0,\oh\mr$ implies that $\gamma_0(u)$ is non-decreasing in $u$ for $u \in [u_2,v_0)$. Next, appealing to \eqref{FIT5} and \eqref{FIT227} along $v = v_0$, we obtain an estimate for the following integral
\begin{equation}\label{FIT64}
\begin{aligned}
&\int^{u_4}_{u_3} \ml \frac{-\du r}{r}\mr_0(u') \ml \frac{\mu_0(u')}{1-\mu_0(u')}\mr^2 e^{-\frac{3}{2}\gamma_0(u') + \frac{1}{2}\log \ml \frac{r_0(0)}{r_0(u')}\mr} \mathrm{d}u' \\
\leq & \; 2\int^{u_4}_{u_3} \ml \frac{-\du r}{r}\mr_0(u') \ml \frac{(\mu - Q^2/r^2)_0(u')}{1-\mu_0(u')}\mr \ml \frac{\mu_0(u')}{1-\mu_0(u')} \mr e^{-\frac{3}{2}\gamma_0(u') + \frac{1}{2}\log \ml \frac{r_0(0)}{r_0(u')}\mr} \mathrm{d}u' \\
\leq & \; \frac{2}{1-\mu_0(0)}\int^{u_4}_{u_3} \ml \frac{-\du r}{r}\mr_0(u') \ml \frac{(\mu - Q^2/r^2)_0(u')}{1-\mu_0(u')}\mr e^{-\frac{1}{2}\gamma_0(u') + \frac{1}{2}\log \ml \frac{r_0(0)}{r_0(u')}\mr} \mathrm{d}u' \\
= &  \; \frac{2}{1-\mu_0(0)} \int^{u_4}_{u_3} \du \gamma_0(u') e^{-\frac{1}{2}\gamma_0(u') + \frac{1}{2}\log \ml \frac{r_0(0)}{r_0(u')}\mr} \mathrm{d}u' \\
= &  -\frac{4}{1-\mu_0(0)} \int^{u_4}_{u_3} \du \ml e^{-\frac{1}{2}\gamma_0(u')} \mr e^{\frac{1}{2}\log \ml \frac{r_0(0)}{r_0(u')}\mr} \mathrm{d}u'.
\end{aligned}
\end{equation}
Integrating by parts, the last line of \eqref{FIT64} can be rewritten as
\begin{equation}\label{FIT67}
\begin{aligned}
&-\frac{4}{1-\mu_0(0)} \int^{u_4}_{u_3} \du \ml e^{-\frac{1}{2}\gamma_0(u')} \mr e^{\frac{1}{2}\log \ml \frac{r_0(0)}{r_0(u')}\mr} \mathrm{d}u' \\
=&  -\frac{4}{1-\mu_0(0)} \ml \ml e^{-\frac{1}{2}\gamma_0(u')} \mr e^{\frac{1}{2}\log \ml \frac{r_0(0)}{r_0(u')}\mr}|^{u' = u_4}_{u' = u_3}  - \int^{u_4}_{u_3} e^{-\frac{1}{2}\gamma_0(u')} \du \ml e^{\frac{1}{2}\log \ml \frac{r_0(0)}{r_0(u')}\mr}\mr \mathrm{d}u' \mr\\
=& \; \frac{4}{1-\mu_0(0)}\ml e^{-\frac{1}{2}\gamma_0(u_3) + \frac{1}{2}\log\ml \frac{r_0(0)}{r_0(u_3)}\mr} - e^{-\frac{1}{2}\gamma_0(u_4) + \frac{1}{2}\log\ml \frac{r_0(0)}{r_0(u_4)}\mr}\mr \\
& + \frac{4}{1-\mu_0(0)} \int^{u_4}_{u_3} e^{-\frac{1}{2}\gamma_0(u')} \du \ml e^{\frac{1}{2}\log \ml \frac{r_0(0)}{r_0(u')}\mr}\mr \mathrm{d}u'.
\end{aligned}
\end{equation}
Since $\du \ml e^{\frac{1}{2}\log \ml \frac{r_0(0)}{r_0(u')}\mr}\mr \geq 0$ and $\gamma_0(u')$ is non-decreasing in $u'$, the second term in the last line of \eqref{FIT67} has the following upper bound
\begin{equation}\label{FIT68}
\begin{aligned}
&\frac{4}{1-\mu_0(0)} \int^{u_4}_{u_3} e^{-\frac{1}{2}\gamma_0(u')} \du \ml e^{\frac{1}{2}\log \ml \frac{r_0(0)}{r_0(u')}\mr}\mr \mathrm{d}u' \\
\leq & \; \frac{4e^{-\frac{1}{2}\gamma_0(u_3)}}{1-\mu_0(0)} \int^{u_4}_{u_3} \du \ml e^{\frac{1}{2}\log \ml \frac{r_0(0)}{r_0(u')}\mr}\mr \mathrm{d}u' \\
= & \; \frac{4}{1-\mu_0(0)} \ml e^{-\frac{1}{2}\gamma_0(u_3) + \frac{1}{2}\log \ml \frac{r_0(0)}{r_0(u_4)}\mr} - e^{-\frac{1}{2}\gamma_0(u_3) + \frac{1}{2}\log \ml \frac{r_0(0)}{r_0(u_3)}\mr} \mr.
\end{aligned}
\end{equation}
Combining \eqref{FIT67} and \eqref{FIT68}, we observe precise cancellations and this gives
\begin{equation}\label{FIT69}
\begin{aligned}
-\frac{4}{1-\mu_0(0)} \int^{u_4}_{u_3} \du \ml e^{-\frac{1}{2}\gamma_0(u')} \mr e^{\frac{1}{2}\log \ml \frac{r_0(0)}{r_0(u')}\mr} \mathrm{d}u' \leq & \;  \frac{4}{1-\mu_0(0)} e^{\frac{1}{2}\log\ml \frac{r_0(0)}{r_0(u_4)}\mr}\ml e^{-\frac{1}{2}\gamma_0(u_3)} - e^{-\frac{1}{2}\gamma_0(u_4)} \mr.
\end{aligned}
\end{equation}
Back to the upper bound in \eqref{FIT61}, by combining \eqref{FIT62}, \eqref{FIT63}, \eqref{FIT69}, we hence deduce that
\begin{equation}\label{FIT70}
\begin{aligned}
&(e-2)\int^{u_4}_{u_3} \ml \frac{3}{2(1-\mu_0(u'))} - 2 \mr^2 \ml \frac{-\du r}{r}\mr_0(u') e^{-\frac{3}{2}\gamma_0(u') + \frac{1}{2}\log \ml \frac{r_0(0)}{r_0(u')}\mr} \mathrm{d}u' \\
\leq & \; (e-2)\ml \frac{1}{2}e^{-\frac{3}{2}\gamma_0(u_3)}e^{\frac{1}{2}\log\ml \frac{r_0(0)}{r_0(u_4)}\mr} + \frac{9}{1-\mu_0(0)} e^{\frac{1}{2}\log\ml \frac{r_0(0)}{r_0(u_4)}\mr}\ml e^{-\frac{1}{2}\gamma_0(u_3)} - e^{-\frac{1}{2}\gamma_0(u_4)} \mr\mr\\
\leq & \; \frac{19(e-2)}{2(1-\mu_0(0))}e^{\frac{1}{2}\log\ml\frac{r_0(0)}{r_0(u_4)}\mr} = C_{34} e^{\frac{1}{2}\log\ml\frac{r_0(0)}{r_0(u_4)}\mr}
\end{aligned}
\end{equation}
with
\begin{equation}\label{FIT71}
C_{34} := \frac{19(e-2)}{2(1-\mu_0(0))}.
\end{equation}
In addition, under the assumption \eqref{FIT35c}, together with \eqref{FIT77}, for all $u \in [u_2,u_n]$ and $v \in (v_0,v(u_n)]$, we derive
\begin{equation}\label{FIT73}
    \log\ml\frac{r(u,v)}{r_0(u)}\mr = \td(u,v) \leq \delta(u,v) \leq \delta(u,v(u_n)) \leq \frac{e^{-\frac{3}{2}(1+2\xi)\gamma_0(u)}}{2 C_{34}} \leq \frac{e^{-\frac{3}{2}\gamma_{s,0}(u)}}{2 C_{34}}.
\end{equation}
In particular, \eqref{FIT73} is true for $u = u_4$. For each $v \in (v_0,v(u_n)]$, by applying such an estimate to \eqref{FIT70}, we obtain
\begin{equation}\label{FIT74}
\begin{aligned}
(e-2)\int^{u_4}_{u_3} \ml \frac{3}{2(1-\mu_0(u'))} - 2 \mr^2 \ml \frac{-\du r}{r}\mr_0(u') e^{-\frac{3}{2}\gamma_0(u') + \frac{1}{2}\log \ml \frac{r_0(0)}{r_0(u')}\mr} \mathrm{d}u' \leq \;  \frac{e^{-\frac{3}{2}\gamma_{s,0}(u_4) + \frac{1}{2}\log\ml\frac{r_0(0)}{r_0(u_4)}\mr}}{2\td(u_4,v)}.
\end{aligned}   
\end{equation}

Next, we return to \eqref{FIT59} and proceed to simplify the terms on the left as follows
\begin{equation}\label{FIT75}
\begin{aligned}
&\frac{e^{-\int^{u_4}_{0} \ml \frac{-\du r}{r}\mr_0(u') \ml \frac{3}{2(1-\mu_0(u'))} - 2\mr \mathrm{d}u'}}{\td(u_4,v)} - \frac{e^{-\int^{u_3}_{0} \ml \frac{-\du r}{r}\mr_0(u') \ml \frac{3}{2(1-\mu_0(u'))} - 2\mr \mathrm{d}u'}}{\td(u_3,v)} \\
=& \; \frac{e^{-\int^{u_4}_{0} \ml \frac{-\du r}{r}\mr_0(u') \ml \frac{3}{2}\frac{\mu_0(u')}{(1-\mu_0(u'))} - \frac{1}{2}\mr \mathrm{d}u'}}{\td(u_4,v)} - \frac{e^{-\int^{u_3}_{0} \ml \frac{-\du r}{r}\mr_0(u') \ml \frac{3}{2}\frac{\mu_0(u')}{(1-\mu_0(u'))} - \frac{1}{2}\mr \mathrm{d}u'}}{\td(u_3,v)} \\
=& \; \frac{e^{- \frac{3}{2}\gamma_{s,0}(u_4) + \frac{1}{2}\log\ml\frac{r_0(0)}{r_0(u_4)}\mr}}{\td(u_4,v)} - \frac{e^{- \frac{3}{2}\gamma_{s,0}(u_3)+ \frac{1}{2}\log\ml\frac{r_0(0)}{r_0(u_3)}\mr}}{\td(u_3,v)}. \\
\end{aligned}
\end{equation}
By utilizing \eqref{FIT75} and the upper bound in \eqref{FIT74}, inequality \eqref{FIT59} can be rewritten as
$$\frac{e^{- \frac{3}{2}\gamma_{s,0}(u_4) + \frac{1}{2}\log\ml\frac{r_0(0)}{r_0(u_4)}\mr}}{\td(u_4,v)} - \frac{e^{- \frac{3}{2}\gamma_{s,0}(u_3)+ \frac{1}{2}\log\ml\frac{r_0(0)}{r_0(u_3)}\mr}}{\td(u_3,v)} \leq \frac{e^{-\frac{3}{2}\gamma_{s,0}(u_4) + \frac{1}{2}\log\ml\frac{r_0(0)}{r_0(u_4)}\mr}}{2\td(u_4,v)},$$
which implies 
$$\frac{e^{- \frac{3}{2}\gamma_{s,0}(u_4) + \frac{1}{2}\log\ml\frac{r_0(0)}{r_0(u_4)}\mr}}{2\td(u_4,v)} \leq \frac{e^{- \frac{3}{2}\gamma_{s,0}(u_3)+ \frac{1}{2}\log\ml\frac{r_0(0)}{r_0(u_3)}\mr}}{\td(u_3,v)}$$
and
\begin{equation}\label{FIT82}
\begin{aligned}
\td(u_3,v) &\leq \td{(u_4,v)} 2 e^{-\frac{1}{2}\log\ml \frac{r_0(u_3)}{r_0(u_4)}\mr-\frac{3}{2}(\gamma_{s,0}(u_3)-\gamma_{s,0}(u_4))} \\
&\leq \td{(u_4,v)} 2 e^{-\frac{1}{2}\log\ml \frac{r_0(u_3)}{r_0(u_4)}\mr + \frac{3}{2}\gamma_{s,0}(u_4)}.
\end{aligned}
\end{equation}
Together with \eqref{FIT47}, the definition of $\td$ in \eqref{deltatilde}, and \eqref{FIT77} in Lemma \ref{FITEstLemma8}, we therefore obtain \eqref{FIT83} and conclude the proof of Lemma \ref{FITEstLemma5}.
\end{proof} 
We also record some useful estimates for $|\phi|$ and $|Q\phi|$ in this setting
\begin{lemma}\label{FITEstLemma10}
For each $(u,v) \in \mathcal{P}'$ with $\chi \in \left(0,\frac{1}{3}\right)$, we have
\begin{equation}\label{FIT166}
|\phi|(u,v) \leq C_{35}  + \frac{1}{(4\pi)^{\frac{1}{2}}}\ml  \log \ml \frac{\frac{1 - \mu}{ \dv r}(u,v)}{\frac{1 - \mu}{ \dv r}(u_2,v)}\mr  \mr^{\frac{1}{2}}\ml \log \ml \frac{r(u_2,v)}{r(u,v)}\mr \mr^{\frac{1}{2}},
\end{equation}
and
\begin{equation}\label{FIT101}
|Q\phi|(u,v) \leq C_{5}(u_2,v_0+\te;\chi,2,\chi)r^{\frac{5}{4}-\frac{3\chi}{4}}(u,v').
\end{equation}
\end{lemma}
\begin{proof} The first inequality follows from \eqref{Estimates13} of Proposition \ref{PropEstimates} if we set 
\begin{equation}\label{FIT231}
C_{35} := \overline{P}(u_2,v_0 + \te).
\end{equation}
The second inequality can be deduced from \eqref{Estimates19} by choosing $\xi_1 = 2$ and $\xi_2 = \chi$ with $\chi \in \left(0,\frac{1}{3}\right]$.
\end{proof}

With the above preparations, we return to the proof of Theorem \ref{FITTheorem} for the cases under \textit{Scenario (I)}. 
We begin with defining the following quantities
\begin{equation}\label{FIT29}
    \Upxi(u,v) := \frac{r(u,v)}{r_0(u)}\ml \frac{r \dv \phi}{\dv r}\mr(u,v) - \ml \frac{r\dv \phi}{\dv r
    }\mr_0 (u) 
\end{equation}
and
\begin{equation}\label{FIT30}
    \psi(u,v) := e^{-\gamma_0(u)}\Upxi(u,v).
\end{equation}
From \eqref{FIT29} and \eqref{FIT30}, for any $(u,v) \in \mathcal{P}'$, we have
\begin{equation}\label{FIT31}
\mlm \frac{r(u,v)}{r_0(u)}\ml \frac{r \dv \phi}{\dv r}\mr(u,v) - \ml \frac{r\dv \phi}{\dv r
    }\mr_0 (u)  \mrm = |\psi(u,v)|e^{\gamma_0(u)}.
\end{equation}
Let $n \in \mathbb{N}$ be sufficiently large such that $u_n$ is sufficiently close to $v_0$. Pick $v(u_n)$ sufficiently close to $v_0$ such that inequalities \eqref{FIT35a}, \eqref{FIT35c} and
\begin{equation}\label{FIT127}
\sup_{v' \in [v_0,v(u_n)]}|\psi|(u_n,v') \leq \frac{h}{6}
\end{equation}
hold. By \eqref{FIT28}, for each $v' \in [v_0,v(u_n)]$, we deduce that
\begin{equation}\label{FIT128}
\begin{aligned}
\frac{r(u_n,v')}{r_0(u_n)}\frac{r|\dv \phi|}{\dv r}(u_n,v') &= \mlm \frac{r\dv \phi}{\dv r}(u_n,v')\ml \frac{r(u_n,v')}{r_0(u_n)}\mr - \ml \frac{r\dv \phi}{\dv r}\mr_0(u_n) + \ml \frac{r \dv \phi}{\dv r}\mr_0(u_n) \mrm  \\
&\geq \mlm \mlm \frac{r\dv \phi}{\dv r}(u_n,v')\ml \frac{r(u_n,v')}{r_0(u_n)}\mr - \ml \frac{r\dv \phi}{\dv r}\mr_0(u_n) \mrm - \mlm \ml \frac{r \dv \phi}{\dv r}\mr_0(u_n)\mrm\mrm \\
&= \mlm \mlm \ml \frac{r \dv \phi}{\dv r}\mr_0(u_n)\mrm - |\psi|(u_n,v')e^{\gamma_0(u_n)} \mrm \\
&\geq \mlm \frac{h}{3} - \frac{h}{6}\mrm e^{\gamma_0(u_n)} = \frac{h}{6}e^{\gamma_0(u_n)}.
\end{aligned}
\end{equation}
Observe that due to \eqref{Setup25}, for $(u,v) \in \mathcal{P}'$, quantity $\eta^*(u,v)$ in \eqref{ReducedMassRatio} can be re-written as
\begin{equation}\label{FIT130}
\begin{aligned}
\eta^*(u,v) &= \frac{2\int^v_{v_0} \dv m(u,v') \mathrm{d}v'}{r(u,v)} - \frac{2}{r(u,v)} \int^v_{v_0} \frac{Q^2(\dv r)}{2r^2} (u,v') \mathrm{d}v' \\
&= \frac{2}{r(u,v)}\ml \int^v_{v_0} \frac{2\pi r^2(1-\mu)|\dv \phi|^2}{\dv r}(u,v') \mathrm{d}v' \mr.
\end{aligned}
\end{equation}
Together with \eqref{FIT128}, for each $u_n$ and for all $v \in (v_0,v(u_n)]$, we have
\begin{equation}\label{FIT133}
\begin{aligned}
\eta^*(u_n,v) &=  \frac{2}{r(u_n,v)}\ml \int^v_{v_0} \frac{2\pi r^2(1-\mu)|\dv \phi|^2}{\dv r}(u_n,v') \mathrm{d}v' \mr  \\
&\geq \frac{8\pi \ml 1-\mu_0(0) \mr }{3r(u_n,v)}e^{-\gamma_0(u_n)}r^2_0(u_n)\ml \int^v_{v_0} (\dv r)\frac{1}{r^2(u_n,v')}
\frac{r^2(u_n,v')}{r^2_0(u_n)}\frac{ r^2|\dv \phi|^2}{(\dv r)^2}(u_n,v') \mathrm{d}v' \mr  \\
&\geq \frac{2\pi \ml 1-\mu_0(0) \mr h^2 }{27r(u_n,v)}e^{\gamma_0(u_n)}r^2_0(u_n)\ml \int^v_{v_0} \frac{(\dv r)(u_n,v')}{r^2(u_n,v')} \mathrm{d}v' \mr  \\
&= \frac{2\pi \ml 1-\mu_0(0) \mr h^2 }{27r(u_n,v)}e^{\gamma_0(u_n)}r^2_0(u_n)\ml \frac{1}{r_0(u_n)} - \frac{1}{r(u_n,v)}\mr  \\
&= \frac{2\pi \ml 1-\mu_0(0) \mr h^2 }{27}e^{\gamma_0(u_n)}\ml \frac{r_0(u_n)}{r(u_n,v)} \mr^2\ml \frac{r(u_n,v) - r_0(u_n)}{r_0(u_n)}\mr  \\
&= \frac{2\pi \ml 1-\mu_0(0) \mr h^2 }{27}e^{\gamma_0(u_n)}\ml \frac{1}{1+\delta(u_n,v)}\mr^2\ml \frac{r(u_n,v) - r_0(u_n)}{r_0(u_n)}\mr  \\
&\geq C_{37} \delta(u_n,v)e^{\gamma_0(u_n)}.  \\
\end{aligned}
\end{equation}
Here, we have used \eqref{FIT36} in Lemma \ref{FITEstLemma3}, $\delta(u_n,v) \leq \frac{1}{2}$, $\dv r > 0$, and we set
\begin{equation}\label{FIT134}
C_{37} := \frac{8\pi(1-\mu_0(0))h^2 }{243}.
\end{equation}

On the other hand, assuming the absence of a trapped surface at each $(u_n,v)$, from \eqref{FIT34}, we get
\begin{equation}\label{FIT131}
\begin{aligned}
\eta^*(u_n,v) &\leq C_{24} \frac{\delta(u_n,v)}{1+\delta(u_n,v)}\log\ml \frac{1}{\delta(u_n,v)}\mr + C_{23}\frac{\delta(u_n,v)r_0(u_n)}{1-\sup_{{v \in [v_0,v_0+\varepsilon]}}\mu(u_n,v)} \\
&\leq C_{24} \delta(u_n,v) \log\ml \frac{1}{\delta(u_n,v)}\mr + \frac{3 C_{23}}{2(1-\mu_0(0))} \delta(u_n,v) e^{\gamma_0(u_n)}r_0(u_n) \\
&\leq C_{36}\delta(u_n,v) \ml \log \ml \frac{1}{\delta(u_n,v)}\mr + e^{\gamma_0(u_n)} r_0(u_n) \mr
\end{aligned}
\end{equation}
with 
\begin{equation}\label{FIT132}
C_{36} = C_{24} + \frac{3 C_{23}}{2(1-\mu_0(0))}.
\end{equation}

By comparing \eqref{FIT133} with \eqref{FIT131} at $v = v(u_n)$, a contradiction would arise if
\begin{equation}\label{FIT135}
\log \ml \frac{1}{\delta(u_n,v(u_n))}\mr e^{-\gamma_0(u_n)} + r_0(u_n)  < \frac{C_{37}}{C_{36}}.
\end{equation}
By picking $n$ sufficiently large, we can always choose
\begin{equation}\label{FIT136}
r_0(u_n) \leq \frac{C_{37}}{4C_{36}}.
\end{equation}
At the same time, we can pick a corresponding $v(u_n)$ to be sufficiently close to $v_0$ such that\footnote{We are able to find such a value of $v(u_n)$ since $\delta(u_n,v) \rightarrow 0$ as $v \rightarrow v_0$.}
\begin{equation}\label{FIT137}
\delta(u_n,v(u_n)) = e^{-\frac{C_{37}}{4C_{36}}e^{\gamma_0(u_n)}}.
\end{equation}

Given \eqref{FIT127}, note that \eqref{FIT136} and \eqref{FIT137} imply inequality \eqref{FIT135}, and thus it leads to a contradiction. We then verify the assumption \eqref{FIT127} with the help of the following lemma.
\begin{lemma}\label{FITpsievolution}
Let $u_n \in [u_2,v_0)$ and $v(u_n) \in [v_0,v_0+\te]$ be chosen such that \eqref{FIT35a} and \eqref{FIT35c} hold. If we further demand that for all $u \in [u_2,u_n]$, we have
\begin{equation}\label{FIT35e}
\delta(u,v(u_n)) \leq \min \left\{ \ml \frac{\log(3)}{8C_{48}}\mr^2e^{-(9+6\xi)\gamma_0(u)}, \ml \frac{h}{1296C_{45}}\mr^2 e^{-\ml \frac{15}{2} + 3\xi\mr\gamma_0(u)}, \ml \frac{h}{1296 b C_{49}}\mr^2 e^{- \frac{19}{2}\gamma_0(u)} \right\}
\end{equation}
and
\begin{equation}\label{FIT252}
\sup_{v \in [v_0,v(u_n)]}|\psi|(u_2,v) \leq \frac{h}{54},
\end{equation}
then inequality \eqref{FIT127} holds, that is, 
$$\sup_{v' \in [v_0,v(u_n)]}|\psi|(u_n,v') \leq \frac{h}{6}.$$
\end{lemma}
\proof 
From \eqref{FIT30} and \eqref{FIT29}, we first compute
\begin{equation}\label{FIT143}
\begin{aligned}
\du \psi(u,v) &= \du \ml e^{-\gamma_0(u)}\Upxi(u,v) \mr = e^{-\gamma_0(u)}  \left( -\du \ml \gamma_0(u)\mr \Upxi(u,v) + \du \ml \Upxi(u,v)\mr \right). \\
\end{aligned}
\end{equation}
In computations that follow, we suppress the dependence of the quantities on $(u,v)$ and abbreviate terms without a subscript $()_0$ as being evaluated at $(u,v)$, while terms with that are evaluated at $(u,v_0)$. We can then calculate 
\begin{equation}\label{FIT144}
\begin{aligned}
\du(\gamma_0) &= \ml \frac{(-\du r)}{r}\frac{\mu - Q^2/r^2}{1-\mu}\mr_0
\end{aligned}
\end{equation}
and
\begin{equation}\label{FIT145}
\begin{aligned}
\du\Upxi &= \frac{r}{r_0}\du \ml \frac{r \dv \phi}{\dv r}\mr + \du \ml \frac{r}{r_0}\mr\frac{r \dv \phi}{\dv r} - \du \ml \frac{r \dv \phi}{\dv r}\mr_0.
\end{aligned}
\end{equation}
By \eqref{FIT8}, we can simplify \eqref{FIT145} to 
\begin{equation}\label{FIT147}
\begin{aligned}
\du \Upxi =&\; \frac{r}{r_0}\ml\ml -D_u \phi - \ii \e \frac{Q \phi (\du r)}{r(1-\mu)}\mr + \ml \frac{r \dv \phi}{\dv r}\mr \ml -\ii \e A_u + \frac{(-\du r)}{r} \frac{\mu - Q^2/r^2}{1 - \mu} \mr \mr + \frac{(\du r)r_0 - (\du r)_0 r}{r_0^2}\ml \frac{r \dv \phi}{\dv r}\mr \\
&- \ml \ml -D_u \phi - \ii \e \frac{Q \phi (\du r)}{r(1-\mu)}\mr_0 + \ml \frac{r \dv \phi}{\dv r}\mr_0 \ml -\ii \e A_u + \frac{(-\du r)}{r} \frac{\mu - Q^2/r^2}{1 - \mu} \mr_0\mr.
\end{aligned}
\end{equation}
Prior to substituting this back to \eqref{FIT143}, we first compute
\begin{equation}\label{FIT148}
\begin{aligned}
(-\du \gamma_0)\Upxi + \du \Upxi =\; & - \ml \frac{(-\du r)}{r}\frac{\mu - Q^2/r^2}{1-\mu}\mr_0 \ml \frac{r}{r_0}\ml \frac{r \dv \phi}{\dv r}\mr - \ml \frac{r \dv \phi}{\dv r}\mr_0\mr + \frac{(\du r)r_0 - (\du r)_0 r}{r_0^2}\ml \frac{r \dv \phi}{\dv r}\mr  \\
&+ \frac{r}{r_0}\ml\ml -D_u \phi - \ii \e \frac{Q \phi (\du r)}{r(1-\mu)}\mr + \ml \frac{r \dv \phi}{\dv r}\mr \ml - \ii \e A_u + \frac{(-\du r)}{r} \frac{\mu - Q^2/r^2}{1 - \mu} \mr \mr \\
&- \ml \ml -D_u \phi - \ii \e \frac{Q \phi (\du r)}{r(1-\mu)}\mr_0 + \ml \frac{r \dv \phi}{\dv r}\mr_0 \ml - \ii \e A_u + \frac{(-\du r)}{r} \frac{\mu - Q^2/r^2}{1 - \mu} \mr_0\mr  \\
=\; & - \ml \frac{(-\du r)}{r}\frac{\mu - Q^2/r^2}{1-\mu}\mr_0 \ml \frac{r}{r_0}\ml \frac{r \dv \phi}{\dv r}\mr\mr + \frac{(\du r)r_0 - (\du r)_0 r}{r_0^2}\ml \frac{r \dv \phi}{\dv r}\mr   \\
&+ \frac{r}{r_0}\ml\ml -D_u \phi - \ii \e \frac{Q \phi (\du r)}{r(1-\mu)}\mr + \ml \frac{r \dv \phi}{\dv r}\mr \ml -\ii \e A_u + \frac{(-\du r)}{r} \frac{\mu - Q^2/r^2}{1 - \mu} \mr \mr \\
&- \ml \ml -D_u \phi - \ii \e \frac{Q \phi (\du r)}{r(1-\mu)}\mr_0 + \ml \frac{r \dv \phi}{\dv r}\mr_0 \ml -\ii \e A_u \mr_0\mr.  \\
\end{aligned}
\end{equation}
Rearranging the expression above, we get
\begin{equation}\label{FIT149}
\begin{aligned}
(-\du \gamma_0)\Upxi + \du \Upxi 
= \; & \left\{ \frac{r}{r_0}\ml - D_u \phi - \ii \e \frac{Q \phi (\du r)}{r(1-\mu)}\mr + \ml D_u \phi + \ii \e \frac{Q \phi (\du r)}{r(1-\mu)}\mr_0 \right\} \\
&+ \ml \frac{r \dv \phi}{\dv r}\mr\frac{r}{r_0} \left\{ \du \gamma - \du \gamma_0 + \frac{\du r}{r} - \ml \frac{\du r}{r}\mr_0    \right\} \\
&- \ii \e \ml A_u \frac{r}{r_0}\ml \frac{r \dv \phi}{\dv r}\mr - (A_u)_0\ml \frac{r \dv \phi}{\dv r}\mr_0 \mr.
\end{aligned}
\end{equation}
Next, observe from \eqref{FIT29} and \eqref{FIT30} that
\begin{equation}
\begin{aligned}
\ml \frac{r \dv \phi}{\dv r} \mr \frac{r}{r_0} &=  \psi e^{\gamma_0} + \ml \frac{r \dv \phi}{\dv r}\mr_0.
\end{aligned}
\end{equation}
Thus, we can rewrite \eqref{FIT149} as 
\begin{equation}\label{FIT150}
\begin{aligned}
(-\du \gamma_0)\Upxi + \du \Upxi 
=\; & \left\{ \frac{r}{r_0}\ml - D_u \phi - \ii \e \frac{Q \phi (\du r)}{r(1-\mu)}\mr + \ml D_u \phi + \ii \e \frac{Q \phi (\du r)}{r(1-\mu)}\mr_0 \right\} \\
&+  \psi e^{\gamma_0}\left\{ \du \gamma - \du \gamma_0 + \frac{\du r}{r} - \ml \frac{\du r}{r}\mr_0    \right\} \\
&+  \ml \frac{r \dv \phi}{\dv r}\mr_0 \left\{ \du \gamma - \du \gamma_0 + \frac{\du r}{r} - \ml \frac{\du r}{r}\mr_0    \right\} \\
&-\ii \e A_u \psi e^{\gamma_0} - \ii \e \ml \frac{r \dv \phi}{\dv r}\mr_0 \ml A_u - (A_u)_0\mr. \\
\end{aligned}
\end{equation}
We now apply \eqref{FIT150} to \eqref{FIT143} and obtain
\begin{equation}\label{FIT151}
\begin{aligned}
\du \psi =\; & e^{-\gamma_0}\left\{ \frac{r}{r_0}\ml - D_u \phi - \ii \e \frac{Q \phi (\du r)}{r(1-\mu)}\mr + \ml D_u \phi + \ii \e \frac{Q \phi (\du r)}{r(1-\mu)}\mr_0 \right\} \\
&+  \psi \left\{ \du \gamma - \du \gamma_0 + \frac{\du r}{r} - \ml \frac{\du r}{r}\mr_0  - \ii \e A_u  \right\} \\
&+  e^{-\gamma_0}\ml \frac{r \dv \phi}{\dv r}\mr_0 \left\{ \du \gamma - \du \gamma_0 + \frac{\du r}{r} - \ml \frac{\du r}{r}\mr_0 - \ii \e A_u + \ii \e (A_u)_0    \right\}. \\
\end{aligned}
\end{equation}
By defining the following terms:
\begin{equation}\label{FIT152}
\begin{aligned}
J_4 &:=  \frac{r}{r_0}\ml - D_u \phi - \ii \e \frac{Q \phi (\du r)}{r(1-\mu)}\mr + \ml D_u \phi + \ii \e \frac{Q \phi (\du r)}{r(1-\mu)}\mr_0 , \\
J_5 &:=  \du \gamma - \du \gamma_0 + \frac{\du r}{r} - \ml \frac{\du r}{r}\mr_0, \\
J_6 &:= J_5 - \ii \e (A_u - (A_u)_0),
\end{aligned}
\end{equation}
the expression in \eqref{FIT151} can be rewritten as 
\begin{equation}\label{FIT153A}
\begin{aligned}
\du \psi &= e^{-\gamma_0} J_4 + \psi(J_5 - \ii \e A_u) + e^{-\gamma_0}\ml \frac{r \dv \phi}{\dv r}\mr_0 J_6. \\
\end{aligned}
\end{equation}
Through the use of a suitable integrating factor, we deduce that
\begin{equation}
\du(e^{-\int^u_{u_2} J_5(u',v) + \ii \e A_u(u',v) \D u'}\psi) =  e^{-\gamma_0-\int^u_{u_2} J_5(u',v) + \ii \e A_u(u',v) \D u'}\ml J_4 + J_6\ml \frac{r \dv \phi}{\dv r}\mr_0\mr.
\end{equation}
Upon integrating with respect to $u$ from $u_2$, we obtain
\begin{equation}\label{FIT154A}
\begin{aligned}
&e^{-\int^{u}_{u_2} J_5(u',v) - \ii \e A_u(u',v) \D u'} \psi(u,v) - \psi(u_2,v) \\
&= \int^{u}_{u_2} e^{-\gamma_0(u'')-\int^{u''}_{u_2} J_5(u',v) - \ii \e A_u(u',v) \D u'}\ml J_4(u'',v) + J_6\ml \frac{r \dv \phi}{\dv r}\mr_0(u'')\mr \D u''.
\end{aligned}
\end{equation}
This implies that
\begin{equation}\label{FIT154}
\begin{aligned}
|\psi|(u,v) &\leq e^{\int^u_{u_2} |J_5| (u',v) \D u'}\ml |\psi|(u_2,v) + \int^{u}_{u_2} e^{\int^{u''}_{u_2} |J_5|(u',v) \D u'}\ml |J_4|(u'',v) + |J_6|\ml \frac{r |\dv \phi|}{\dv r}\mr_0(u'',v)\mr \D u'' \mr.
\end{aligned}
\end{equation}
Here we use the fact that the term $\ii \e A_u$ is purely imaginary, and thus $e^{\ii \e A_u}$ constitutes an additional phase factor. We then proceed to estimate each of the terms $|J_4|, |J_5|, |J_6|$ in \eqref{FIT152}.
Starting from $|J_4|$, we can rewrite it as
\begin{equation}\label{FIT155}
\frac{r}{r_0}\ml - D_u \phi - \ii \e \frac{Q \phi (\du r)}{r(1-\mu)}\mr + \ml D_u \phi + \ii \e \frac{Q \phi (\du r)}{r(1-\mu)}\mr_0 = - \int^v_{v_0} \dv \ml \frac{r}{r_0} \ml D_u \phi + \ii \e \frac{Q \phi (\du r)}{r(1-\mu)} \mr \mr(u,v') \D v'.
\end{equation}
This inspires us to compute the integrand in \eqref{FIT155}, which yields
\begin{equation}\label{FIT156}
-\dv \ml \frac{r}{r_0} \ml D_u \phi + \ii \e \frac{Q \phi (\du r)}{r(1-\mu)} \mr \mr = - \frac{1}{r_0} \ml  \dv \ml r D_u \phi \mr +  \ii \e\dv \ml \frac{Q \phi (\du r)}{1-\mu}\mr \mr.
\end{equation}
Observe that by using \eqref{Setup19}, \eqref{Setup20}, \eqref{Setup27}, we have
\begin{equation}\label{FIT159}
\begin{aligned}
\dv \ml \frac{Q \phi (\du r)}{1-\mu}\mr =\; & \frac{\phi(\du r)(\dv Q)}{1-\mu} + \frac{Q (\dv \phi)(\du r)}{1-\mu} + \frac{Q \phi \dv \du r}{1-\mu} + \frac{Q\phi(\du r)(\dv \mu)}{(1-\mu)^2} \\
=\; &  (-4 \pi \e)\frac{\phi(\du r)}{1-\mu} r^2 \im{(\phi^{\dagger} \dv \phi)} + \frac{Q (\dv \phi)(\du r)}{1-\mu} + \frac{Q \phi(\du r)(\dv r)(\mu - Q^2/r^2)}{r(1-\mu)^2} \\
&+ \frac{Q\phi(\du r)}{(1-\mu)^2}\ml \frac{4 \pi r(1-\mu)|\dv \phi|^2}{\dv r} - \frac{\dv r}{r}\ml \mu - Q^2/r^2\mr\mr \\
=\; & (4 \pi \e)\frac{\phi(-\du r)r^2 \im{(\phi^{\dagger} \dv \phi)}}{1-\mu} + \frac{Q (\dv \phi)(\du r)}{1-\mu} +  4 \pi \frac{Q r |\dv \phi|^2(\du r) \phi}{\dv r(1-\mu)}.
\end{aligned}
\end{equation} 
Together with \eqref{Setup73}, equation \eqref{FIT159} implies that 
\begin{equation}\label{FIT160}
\begin{aligned}
& \dv \ml \frac{r}{r_0} \ml D_u \phi + \ii \e \frac{Q \phi (\du r)}{r(1-\mu)} \mr \mr \\
=\; & \frac{1}{r_0}\ml (-\du r) \dv \phi + \ii \e \frac{Q (\du r)(\dv r) \phi}{r(1-\mu)} +  (4 \pi \ii \e^2)\frac{\phi(-\du r)r^2 \im{(\phi^{\dagger} \dv \phi)}}{1-\mu} +\ii \e \frac{Q (\dv \phi)(\du r)}{1-\mu} +  4 \pi \ii \e \frac{Q r |\dv \phi|^2(\du r) \phi}{\dv r(1-\mu)}\mr.
\end{aligned}
\end{equation}
By \eqref{FIT155} and the fact that $\du \dv r \geq 0$,  we have the following estimate for $J_4$
\begin{equation}\label{FIT161}
|J_4|(u,v) \leq \frac{(-\du r)(u,v)}{r_0(u)}\ml |H_{1}| + |H_{2}| + |H_{3}| + |H_{4}| + |H_{5}|\mr(u,v),
\end{equation}
with
\begin{equation}\label{FIT162}
\begin{aligned}
|H_1|(u,v) &= \int^v_{v_0} |\dv \phi|(u,v') \D v', \\
|H_3|(u,v) &= 4 \pi \e^2 \int^v_{v_0}  \frac{r^2 |\phi|^2|\dv \phi|}{1-\mu}(u,v') \D v',  \\
|H_5|(u,v) &= 4 \pi \e \int^v_{v_0}  \frac{r |Q||\dv \phi|^2 |\phi|}{\dv r(1-\mu)}(u,v') \D v'. \\
\end{aligned} \quad 
\begin{aligned}
|H_2|(u,v) &= \e \int^v_{v_0}  \frac{|Q||\phi|(\dv r)}{r(1-\mu)}(u,v') \D v',\\
|H_4|(u,v) &= \e \int^v_{v_0} \frac{|Q||\dv \phi|}{1- \mu}(u,v') \D v', \\
\end{aligned}
\end{equation}

As \eqref{FIT35a} and \eqref{FIT35c} were verified, we thus have proved \eqref{FIT83}. Consequently, we have $\delta(u,v(u_n)) \leq \frac{1}{2}$ for each $u \in [u_2,u_n]$. This implies that \eqref{FIT34} is true and hence \eqref{FIT131} holds, not just for $u_n$, but also for each $u \in [u_2,u_n]$. This then yields that for each $v \in (v_0,v(u_n)]$, we have
\begin{equation}\label{FIT175}
\begin{aligned}
\eta^*(u,v) &\leq C_{36} \delta(u,v) \ml \log \ml \frac{1}{\delta(u,v)}\mr + e^{\gamma_0(u)}r_0(u)\mr \\
&\leq C_{38}\delta(u,v)^{\frac{1}{2}}e^{\gamma_0(u)}
\end{aligned}
\end{equation}
with
\begin{equation}\label{FIT270}
C_{38} := C_{36}\ml \frac{2}{e} + \frac{r_0(0)}{2}\mr.
\end{equation}
Here we have used  Lemma \ref{xlogx} and the fact that $\gamma_0(u_2)\geq 0$.

We start to bound $|H_1|$ to $|H_5|$ as follows. For $|H_1|$, using \eqref{FIT36}, \eqref{FIT130}, \eqref{FIT175}, we obtain
\begin{equation*}
\begin{aligned}
|H_1|(u,v) =& \int^v_{v_0}  |\dv \phi|(u,v') \D v'\\ 
\leq& \frac{1}{\sqrt{2\pi}} \ml \int^v_{v_0} \frac{2 \pi r^2(1-\mu)|\dv \phi|^2}{\dv r}(u,v') \D v'\mr^{\frac{1}{2}} \ml \int^v_{v_0} \frac{\dv r}{r^2(1-\mu)}(u,v') \D v'\mr^{\frac{1}{2}}
\end{aligned}
\end{equation*}
\begin{equation}\label{FIT163}
\begin{aligned}
&\leq \frac{\sqrt{3}}{\sqrt{8\pi} \ml 1-\mu_0(0) \mr^{\frac{1}{2}}} r(u,v)^{\frac{1}{2}} \eta^*(u,v)^{\frac{1}{2}} e^{\frac{1}{2}\gamma_0(u)}\ml \int^v_{v_0} \frac{\dv r}{r^2}(u,v') \D v'\mr^{\frac{1}{2}} \\
&\leq \frac{\sqrt{3}}{\sqrt{8\pi}\ml 1-\mu_0(0) \mr^{\frac{1}{2}}} \eta^*(u,v)^{\frac{1}{2}} e^{\frac{1}{2}\gamma_0(u)}\ml \frac{r(u,v)}{r(u,v_0)} - \frac{r(u,v)}{r(u,v)}\mr^{\frac{1}{2}} \\
&\leq \frac{\sqrt{3}C_{38}^{\frac{1}{2}}}{\sqrt{8\pi}\ml 1-\mu_0(0) \mr^{\frac{1}{2}}} e^{\gamma_0(u)}\delta(u,v)^{\frac{1}{2}}= C_{39} e^{\gamma_0(u)}\delta(u,v)^{\oh}
\end{aligned}
\end{equation}
with
\begin{equation}\label{FIT169}
C_{39} := \sqrt{\frac{3C_{38}}{8\sqrt{2}\pi(1-\mu_0(0))}}.
\end{equation}
Employing \eqref{FIT101} with $\chi \in \left(0,\frac{1}{3}\right]$, we can see that $|H_2|$ obeys
\begin{equation}\label{FIT164}
\begin{aligned}
|H_2|(u,v) &= \e \int^v_{v_0}  \frac{|Q||\phi|(\dv r)}{r(1-\mu)}(u,v') \D v' \leq \frac{3 C_{5} \e}{2(1-\mu_0(0))} e^{\gamma_0(u)} \int^v_{v_0} r^{\frac{5}{4}-\frac{3\chi}{4}} \dv r (u,v') \D v' \\
&\quad\quad\leq \frac{3 C_{5} \e r_0(0) r(u_2,v_0+\te)^{\frac{5}{4}-\frac{3\chi}{4}}}{2(1-\mu_0(0))} e^{\gamma_0(u)} \delta(u,v) = C_{40} e^{\gamma_0(u)}\delta(u,v)^{\frac{1}{2}}
\end{aligned}
\end{equation}
with
\begin{equation}\label{FIT170}
C_{40} := \frac{3 C_{5}(u_2,v_0+\te;\chi,2,\chi) \e r_0(0) r(u_2,v_0+\te)^{\frac{1}{4}-\frac{3\chi}{4}}}{2\sqrt{2}(1-\mu_0(0))}.
\end{equation}
To bound $|H_3|$, we first appeal to \eqref{FIT166} and use inequality $(a+b)^4 \leq 8(a^4 + b^4)$ to derive
\begin{equation}\label{FIT165}
\begin{aligned}
\int^v_{v_0} r^2|\phi|^4 \frac{\dv r}{1-\mu}(u,v') \D v' \leq \; & 8C_{35}^4\int^v_{v_0} r^2 \frac{\dv r}{1-\mu}(u,v') \D v' \\
&+  \frac{1}{2\pi^2}\int^v_{v_0} r^2 \frac{\dv r}{1-\mu}(u,v') \ml  \log \ml \frac{\frac{1 - \mu}{ \dv r}(u,v)}{\frac{1 - \mu}{ \dv r}(u_2,v)}\mr  \mr^{2}\ml \log \ml \frac{r(u_2,v)}{r(u,v)}\mr \mr^{2} \D v'. \\
\end{aligned}
\end{equation}
The first integrand on the right satisfies
\begin{equation}\label{FIT244}
\begin{aligned}
\int^v_{v_0} r^2\frac{\dv r}{1-\mu}(u,v') \D v' &\leq  \frac{3r^2(u_2,v_0 + \te)}{2(1-\mu_0(0))}e^{\gamma_0(u)}\int^v_{v_0} \dv r (u,v') \D v' \leq \frac{3r^2(u_2,v_0 + \te)r_0(0)}{2(1-\mu_0(0))}e^{\gamma_0(u)}\delta(u,v).
\end{aligned}
\end{equation}
By Lemma \ref{xlogx} and \eqref{FIT36}, the second integrated obeys
\begin{equation}\label{FIT245}
\begin{aligned}
& \int^v_{v_0} r^2 \frac{\dv r}{1-\mu}(u,v') \ml  \log \ml \frac{\frac{1 - \mu}{ \dv r}(u,v)}{\frac{1 - \mu}{ \dv r}(u_2,v)}\mr  \mr^{2}\ml \log \ml \frac{r(u_2,v)}{r(u,v)}\mr \mr^{2} \D v' \\
\leq \; &  \ml \frac{\bar{B}(u_2,v_0+\te)}{(1-\overline{\mu}_*)(u_2,v_0+\te)}\mr^{\frac{1}{2}}\frac{16 r(u_2,v_0+\te)^2}{e^4}\int^v_{v_0} \ml \frac{\dv r}{1-\mu}\mr^{\frac{1}{2}}(u,v')  \D v' \\
\leq \; &  \ml \frac{3\bar{B}(u_2,v_0+\te)\te}{2(1-\mu_0(0))(1-\overline{\mu}_*)(u_2,v_0+\te)}\mr^{\frac{1}{2}}\frac{16 r(u_2,v_0+\te)^2}{e^4} e^{\frac{1}{2}\gamma_0(u)} \ml \int^v_{v_0} \dv r(u,v')  \D v' \mr^{\frac{1}{2}} \\
\leq \; &  \ml \frac{3\bar{B}(u_2,v_0+\te)\te r_0(0)}{2(1-\mu_0(0))(1-\overline{\mu}_*)(u_2,v_0+\te)}\mr^{\frac{1}{2}}\frac{16 r(u_2,v_0+\te)^2}{e^4} e^{\frac{1}{2}\gamma_0(u)} \delta(u,v)^{\frac{1}{2}}.
\end{aligned}
\end{equation}
Combining \eqref{FIT165}, \eqref{FIT244}, \eqref{FIT245}, we then obtain
\begin{equation}\label{FIT246}
\int^v_{v_0} r^2|\phi|^4 \frac{\dv r}{1-\mu}(u,v') \D v' \leq C_{41}e^{\gamma_0(u)}\delta(u,v)^{\frac{1}{2}}
\end{equation}
with
\begin{equation}\label{FIT247}
C_{41} := \frac{12 C_{35}^4r^2(u_2,v_0 + \te)r_0(0)}{\sqrt{2}(1-\mu_0(0))} +  \ml \frac{3\bar{B}(u_2,v_0+\te)\te r_0(0)}{2(1-\mu_0(0))(1-\overline{\mu}_*)(u_2,v_0+\te)}\mr^{\frac{1}{2}}\frac{8 r(u_2,v_0+\te)^2}{\pi^2 e^4}.
\end{equation}
With these preparations, for $H_3$, we have
\begin{equation}\label{FIT168}
\begin{aligned}
|H_3|(u,v) &= 4 \pi \e^2 \int^v_{v_0}  \frac{r^2 |\phi|^2|\dv \phi|}{1-\mu}(u,v') \D v' \leq \frac{6 \pi \e^2}{(1-\mu_0(0))}e^{\gamma_0(u)} \int^v_{v_0}r^2|\phi|^2|\dv \phi|(u,v') \D v' \\
&\quad\leq \frac{6 \pi \e^2 e^{\gamma_0(u)}}{\sqrt{2\pi}(1-\mu_0(0))} \ml \int^v_{v_0}\frac{2\pi r^2(1-\mu)|\dv \phi|^2}{\dv r}(u,v') \D v' \mr^{\frac{1}{2}} \ml \int^v_{v_0} \frac{r^2 |\phi|^4 \dv r}{1-\mu}(u,v')\D v'\mr^{\frac{1}{2}} \\
&\quad\leq \frac{3 \sqrt{\pi} \e^2 C_{41}^{\frac{1}{2}} C_{38}^{\oh} r(u_2,v_0+\te)^{\frac{1}{2}}}{(1-\mu_0(0))}e^{\frac{7}{4}\gamma_0(u)} \delta(u,v)^{\frac{1}{2}} = C_{42} e^{\frac{7}{4}\gamma_0(u)}\delta(u,v)^{\frac{1}{2}}
\end{aligned}
\end{equation}
with
\begin{equation}\label{FIT171}
C_{42} := \frac{3 \sqrt{\pi} \e^2 C_{41}^{\frac{1}{2}} C_{38}^{\oh} r(u_2,v_0+\te)^{\frac{1}{2}}}{(1-\mu_0(0))}.
\end{equation}

Continuing on to $|H_4|$, by employing the estimate for $\int^v_{v_0} |\dv \phi|(u,v') \D v'$ that is already done in \eqref{FIT163} and applying \eqref{FIT226} with $\mu < 1$, we derive
\begin{equation}\label{FIT172}
\begin{aligned}
|H_4|(u,v) &= \e \int^v_{v_0} \frac{|Q||\dv \phi|}{1- \mu}(u,v') \D v' \leq \frac{\e \sqrt{\xi} r(u_2,v_0+\te)}{1-\mu_0(0)} \int^v_{v_0}|\dv \phi|(u,v') \D v' \\ 
&\quad\quad\leq \frac{\e \sqrt{\xi} r(u_2,v_0+\te) C_{39}}{1-\mu_0(0)}e^{2\gamma_0(u)}\delta(u,v)^{\frac{1}{2}} = C_{43} e^{2\gamma_0(u)}\delta(u,v)^{\frac{1}{2}}
\end{aligned}
\end{equation}
with 
\begin{equation}\label{FIT173}
 C_{43} := \frac{\e \sqrt{\xi} r(u_2,v_0+\te) C_{39}}{1-\mu_0(0)}.
\end{equation}
For $|H_5|$, using \eqref{FIT101}, we deduce 
\begin{equation}\label{FIT174}
\begin{aligned}
|H_5|(u,v) &= 4 \pi \e \int^v_{v_0}  \frac{r |Q||\dv \phi|^2 |\phi|}{\dv r(1-\mu)}(u,v') \D v' \\
&\leq \frac{9\e C_5 r^{\frac{1}{4}-\frac{3\chi}{4}}(u_2,v_0+\te)}{2(1-\mu_0(0))^2} e^{2\gamma_0(u)}\int^v_{v_0} \frac{2\pi r^2(1-\mu)|\dv \phi|^2}{\dv r}(u,v') \D v' \\
&\leq \frac{9\e C_5 C_{38} r^{\frac{5}{4}-\frac{3\chi}{4}}(u_2,v_0+\te) }{4(1-\mu_0(0))^2} e^{3\gamma_0(u)}\delta(u,v)^{\frac{1}{2}} = C_{44}e^{3\gamma_0(u)}\delta(u,v)^{\frac{1}{2}}
\end{aligned}
\end{equation}
with
\begin{equation}\label{FIT176}
C_{44} := \frac{9\e C_5 C_{38} r^{\frac{5}{4}-\frac{3\chi}{4}}(u_2,v_0+\te) }{4(1-\mu_0(0))^2} .
\end{equation}
Summarizing the above estimates for $|H_1|$ to $|H_5|$ and substituting them back to \eqref{FIT161}, we obtain
\begin{equation}\label{FIT177}
\begin{aligned}
|J_4|(u,v) &\leq (C_{39} + C_{40} + C_{42} + C_{43} + C_{44})\frac{(-\du r)(u,v)}{r_0(u)} e^{3\gamma_0(u)}\delta(u,v)^{\frac{1}{2}} \\
&=  (C_{39} + C_{40} + C_{42} + C_{43} + C_{44}) \ml -\frac{\du r}{r}\mr(u,v) \frac{r(u,v)}{r_0(u)} e^{3\gamma_0(u)}\delta(u,v)^{\frac{1}{2}} \\
&\leq C_{45} \ml \frac{-\du r}{r}\mr_0(u) e^{3\gamma_0(u)}\delta(u,v)^{\frac{1}{2}}
\end{aligned}
\end{equation}
with
\begin{equation}\label{FIT178}
C_{45} := \frac{9}{2}\ml C_{39} + C_{40} + C_{42} + C_{43} + C_{44} \mr.
\end{equation}
Here we use the fact that $\delta(u,v) \leq 1/2$ and the inequality \eqref{FIT52} in the proof of Lemma \ref{FITEstLemma5} holds.

We now proceed to estimate $|J_5|$. From the expression of $J_5$ in \eqref{FIT152}, we can rewrite it as
\begin{equation}\label{FIT179}
J_5(u,v) = \int^v_{v_0} \ml \dv \ml \frac{\du r}{r} \mr + \dv 
 \du \gamma \mr(u,v') \, \D v'.
\end{equation}
By \eqref{Setup20}, the first term in the integrand yields
\begin{equation}\label{FIT180}
\begin{aligned}
\dv \ml \frac{\du r}{r} \mr = \frac{\du r \dv r}{r^2}\ml \frac{2 \mu - 1 - Q^2/r^2}{1-\mu}\mr.
\end{aligned}
\end{equation}
For the second term in the integrand, observe that
\begin{equation}\label{FIT86}
\begin{aligned}
\partial_v\partial_u \gamma=\; &\dv \ml \frac{-\du r}{r} \frac{\mu - Q^2/r^2}{1-\mu}\mr \\
=\; & J_8 + J_{9} + J_{10}
\end{aligned}
\end{equation}
with 
\begin{equation}\label{FIT248}
J_8 := T_1, \quad J_9 := T_2, \quad \text{ and } \quad J_{10} := T_3,
\end{equation}
and $T_1, T_2,$ and $T_3$ defined in \eqref{Setup98}.
Henceforth, we have
\begin{equation}\label{FIT249}
\begin{aligned}
\dv \ml \frac{\du r}{r} \mr + \dv \du \gamma &= J_8 + J_{9} + J_{10} +  \frac{\du r \dv r}{r^2}\ml \frac{2 \mu - 1 - Q^2/r^2}{1-\mu}\mr \\
&= J_{8} + J_{9} + J_{11}
\end{aligned}
\end{equation}
with
\begin{equation}\label{FIT188}
J_{11} := \frac{\du r \dv r}{r^2}\ml -4 + \frac{3 - 5Q^2/r^2}{1-\mu}\mr.
\end{equation}
Returning back to \eqref{FIT179}, we deduce that
\begin{equation}\label{FIT187}
\begin{aligned}
|J_{5}|(u,v) \leq \int^v_{v_0} |J_8| + |J_{9}| + |J_{11}|(u,v')\, \D v'.
\end{aligned}
\end{equation}

To estimate $|J_8|, |J_{9}|, |J_{11}|$, we utilize similar strategies as above. Together with \eqref{FIT52}, for $|J_8|$ we derive
\begin{equation}\label{FIT189}
\begin{aligned}
\int^v_{v_0} |J_8|(u,v') \D v' &\leq 24 \pi \e \ml \frac{-\du r}{r}\mr_0(u) \int^v_{v_0} \frac{|Q||\phi||\dv \phi|}{1-\mu}(u,v') \D v' \\
&\leq \frac{24 \pi \e}{\sqrt{2\pi}} \ml \frac{-\du r}{r}\mr_0(u) \ml \int^v_{v_0} \frac{2\pi r^2(1-\mu)|\dv \phi|^2}{\dv r}(u,v') \D v'\mr^{\oh} \ml \int^v_{v_0} \frac{\dv r}{r^2}\frac{|Q|^2|\phi|^2}{(1-\mu)^3}(u,v') \mr^{\oh} \\
&\leq \frac{18\sqrt{3\pi} \e C_5 r^{\frac{3}{4}-\frac{\chi}{4}}(u_2,v_0+\te)}{\sqrt{2}(1-\mu_0(0))^{\frac{3}{2}}} \ml \frac{-\du r}{r}\mr_0(u) e^{\frac{3}{2}\gamma_0(u)} \eta^*(u,v)^{\oh} \ml \int^v_{v_0} \dv r(u,v') \mr^{\oh} \\
&\leq C_{46}  \ml \frac{-\du r}{r}\mr_0(u) e^{2\gamma_0(u)}\delta(u,v)^{\frac{1}{2}}
\end{aligned}
\end{equation}
with
\begin{equation}
C_{46} := \frac{18\sqrt{3\pi} \e C_5 r^{\frac{3}{4}-\frac{\chi}{4}}(u_2,v_0+\te) r_0(0)^{\frac{1}{2}}}{2^{\frac{3}{4}}(1-\mu_0(0))^{\frac{3}{2}}}.
\end{equation}
For $|J_{9}|$, using $Q^2/r^2 < 1$, we get
\begin{equation}\label{FIT190}
\begin{aligned}
\int^v_{v_0} |J_{9}|(u,v') \D v' &\leq 6 \ml \frac{-\du r}{r}\mr_0(u) \int^v_{v_0} \frac{2 \pi r|\dv \phi|^2}{\dv r(1-\mu)}(u,v') \D v' \\
&\leq \frac{27}{2(1-\mu_0(0))^2} \ml \frac{-\du r}{r}\mr_0(u) e^{2\gamma_0(u)}\frac{r(u,v)}{r_0(u)}\frac{1}{r(u,v)}\int^v_{v_0} \frac{2 \pi r^2(1-\mu)|\dv \phi|^2}{\dv r}(u,v') \D v' \\
&\leq \frac{81}{8(1-\mu_0(0))^2} \ml \frac{-\du r}{r}\mr_0(u) e^{2\gamma_0(u)} \eta^*(u,v) \\
&\leq C_{47} \ml \frac{-\du r}{r}\mr_0(u) e^{3\gamma_0(u)}\delta(u,v)^{\frac{1}{2}},
\end{aligned}
\end{equation}
with
\begin{equation}
C_{47} := \frac{81 C_{38}}{8(1-\mu_0(0))^2}.
\end{equation}
Similarly, we also obtain
\begin{equation}\label{FIT191}
\begin{aligned}
\int^v_{v_0} |J_{11}|(u,v') \D v' &\leq 3\ml -\frac{\du r}{r}\mr_0(u) \int^v_{v_0}  \frac{\dv r}{r}\ml 4 + \frac{3 + 5Q^2/r^2}{1-\mu}\mr(u,v') \D v' \\
&\leq \frac{36}{1-\mu_0(0)} \ml\frac{-\du r}{r}\mr_0(u)  e^{\gamma_0(u)}\int^v_{v_0}  \frac{\dv r}{r}(u,v') \, \D v' \\
&\leq \frac{36}{\sqrt{2}(1-\mu_0(0))} \ml\frac{-\du r}{r}\mr_0(u) 
  e^{\gamma_0(u)}\delta(u,v)^{\frac{1}{2}}. \\
\end{aligned}
\end{equation}

For $J_6$, we observe that it is the addition of $J_5$ and a term that can be estimated by applying 
$$\partial_v A_u=\frac{2Q\partial_u r \partial_v r}{r^2(1-\mu)} \quad \mbox{ and } \quad \frac{Q^2}{r^2} < 1$$ 
as below
\begin{equation}\label{FIT192}
\begin{aligned}
|-\ii \e(A_u(u,v) - (A_u)_0(u))| &= \mlm \e \int^v_{v_0} \dv A_u (u,v') \D v'\mrm \leq \int^v_{v_0} \frac{2\e |Q|(-\du r)(\dv r)}{r^2(1-\mu)}(u,v')\D v' \\
&\quad\quad\leq \frac{6\e }{(1-\mu_0(0))}\ml \frac{-\du r}{r}\mr_0(u) e^{\gamma_0(u)}\int^v_{v_0} \dv r (u,v')\D v' \\
&\quad\quad\leq  \frac{6\e r_0(0)}{\sqrt{2}(1-\mu_0(0))}\ml \frac{-\du r}{r}\mr_0(u) e^{\gamma_0(u)}\delta(u,v)^{\frac{1}{2}}.
\end{aligned}
\end{equation}
Combining the above estimates, we then obtain
\begin{equation}\label{FIT194}
\begin{aligned}
|J_5|(u,v) &\leq C_{48} \ml \frac{-\du r}{r}\mr_0(u) e^{3 \gamma_0(u)} \delta(u,v)^{\frac{1}{2}}, \\
|J_6|(u,v) &\leq C_{49} \ml \frac{-\du r}{r}\mr_0(u) e^{3 \gamma_0(u)} \delta(u,v)^{\frac{1}{2}}, \\
\end{aligned}
\end{equation}
where
\begin{equation}\label{FIT195}
\begin{aligned}
C_{48} &:= C_{46} + C_{47} + \frac{36}{\sqrt{2}\ml 1-\mu_0(0)\mr} \quad \mbox{ and } \quad
C_{49} &:= C_{48} + \frac{6\e r_0(0)}{\sqrt{2} \ml 1-\mu_0(0)\mr}. 
\end{aligned}
\end{equation}

With the desired estimates for $|J_4|, |J_5|, |J_6|$ as obtained above, we are now in shape to conclude. Applying the crucial inequality \eqref{FIT83}, for each $u \in [u_2,u_n]$ and $v \in (v_0,v(u_n)]$, we deduce that
\begin{equation}\label{FIT200}
\begin{aligned}
\int^u_{u_2} |J_5|(u',v) \D v' &\leq C_{48} \int^u_{u_2} \ml \frac{-\du r}{r}\mr_0(u') e^{3\gamma_0(u')}\delta(u',v)^{\frac{1}{2}}\D u' \\
&\leq 2C_{48} e^{\ml \frac{15}{4} + \frac{3}{2}\xi \mr\gamma_0(u)}\delta(u,v)^{\frac{1}{2}} \int^u_{u_2} \ml \frac{-\du r}{r}\mr_0(u') e^{-\frac{1}{4}\log\ml \frac{r_0(u')}{r_0(u)}\mr} \D u' \\
&\leq 8C_{48} e^{\ml \frac{15}{4} + \frac{3}{2}\xi \mr\gamma_0(u)}\delta(u,v)^{\frac{1}{2}}.
\end{aligned}
\end{equation}
In a similar fashion, we also derive
\begin{equation}\label{FIT197}
\int^u_{u_2} |J_4|(u',v)\D u' \leq  8C_{45} e^{\ml \frac{15}{4} + \frac{3}{2}\xi \mr\gamma_0(u)}\delta(u,v)^{\frac{1}{2}},
\end{equation}

\begin{equation}\label{FIT201}
\int^u_{u_2} |J_6|(u',v)\D u' \leq  8C_{49} e^{\ml \frac{15}{4} + \frac{3}{2}\xi \mr\gamma_0(u)}\delta(u,v)^{\frac{1}{2}}.
\end{equation}
In particular, for each given $u$ and a correspondingly given $u'' \in [u_2,u]$, by applying \eqref{FIT83} and our assumption \eqref{FIT35e}, we obtain
\begin{equation}\label{FIT250}
\begin{aligned}
\int_{u_2}^{u''} |J_5|(u',v) \D v' &\leq 8C_{48} e^{\ml \frac{15}{4} + \frac{3}{2}\xi \mr\gamma_0(u)}\delta(u'',v)^{\frac{1}{2}} \leq 8C_{48} e^{\ml \frac{9}{2} + 3\xi \mr\gamma_0(u)}\delta(u,v(u_n))^{\frac{1}{2}} \leq \log(3).
\end{aligned}
\end{equation}
Consequently, this yields
\begin{equation}\label{FIT198}
\int^u_{u_2} e^{\int^{u''}_{u_2}|J_5|(u',v)\D u'} |J_4|(u'',v) \D u'' \leq 24C_{45}e^{\ml \frac{15}{4} + \frac{3}{2}\xi \mr \gamma_0(u)}\delta(u,v)^{\frac{1}{2}},
\end{equation}
and by \eqref{FIT207} we have
\begin{equation}\label{FIT199}
\begin{aligned}
\int^u_{u_2} e^{\int^{u''}_{u_2}|J_5|(u',v)\D u'} |J_6|(u'',v)\ml \frac{r|\dv\phi|}{\dv r}\mr_0(u'') \D u'' &\leq 3be^{\gamma_0(u)} \int^u_{u_2} |J_6|(u',v)\D u'\\
&\leq 24b C_{49}e^{(\frac{19}{4}+\frac32\xi)\gamma_0(u)}\delta(u,v)^{\frac{1}{2}}.
\end{aligned}
\end{equation}
Collecting \eqref{FIT250}, \eqref{FIT198}, \eqref{FIT199} and going back to \eqref{FIT154}, we arrive at
\begin{equation}\label{FIT251}
|\psi|(u,v) \leq 3 \ml |\psi|(u_2,v) + 24C_{45}e^{\ml \frac{15}{4} + \frac{3}{2}\xi \mr \gamma_0(u)}\delta(u,v)^{\frac{1}{2}} +24b C_{49}e^{(\frac{19}{4}+\frac32\xi)\gamma_0(u)}\delta(u,v)^{\frac{1}{2}} \mr.
\end{equation}
Together with \eqref{FIT252}, for each $v \in [v_0,v(u_n)]$, we then conclude that \eqref{FIT251} implies
\begin{equation}\label{FIT253}
|\psi|(u_n,v) \leq 3 \ml \frac{h}{54} \times 3\mr = \frac{h}{6}.
\end{equation}
Taking the supremum over $[v_0,v(u_n)]$, we now verify \eqref{FIT127}.
\qed \\

Back to the proof of the main theorem, by assuming \eqref{FIT35e} and \eqref{FIT252}, Lemma \ref{FITpsievolution} implies that \eqref{FIT127} holds. Hence, it remains to check that such a choice of $v(u_n)$ in \eqref{FIT137} satisfies the assumptions \eqref{FIT35a}, \eqref{FIT35c}, \eqref{FIT35e}, \eqref{FIT252}. 
For \eqref{FIT35a}, \eqref{FIT35c}, \eqref{FIT35e}, at $u = u_n$ with $n$ sufficiently large, since $\gamma_0(u_n) \rightarrow \infty$, we have
\begin{equation}\label{FIT139}
\begin{aligned}
\delta(u_n,v(u_n))\log \ml \frac{1}{\delta(u_n,v(u_n))}\mr &= e^{-\frac{C_{37}}{4C_{36}}e^{\gamma_0(u_n)}}
\log\ml e^{\frac{C_{37}}{4C_{36}}e^{\gamma_0(u_n)}}\mr \leq C_{33}e^{-2\gamma_0(u_n)},
\end{aligned}
\end{equation}
\begin{equation}\label{FIT140}
\begin{aligned}
\delta(u_n,v(u_n)) &= e^{-\frac{C_{37}}{4C_{36}}e^{\gamma_0(u_n)}} \leq \min\left\{ \frac{1}{2},\frac{2(1 - \mu_0(0))e^{-\gamma_0(u_n)}}{3}, \frac{e^{-\frac{3}{2}(1+2\xi)\gamma_0(u_n)}}{2 C_{34}} \right\},
\end{aligned}
\end{equation}
\begin{equation}\label{FIT257}
\begin{aligned}
\delta(u_n,v(u_n)) &= e^{-\frac{C_{37}}{4C_{36}}e^{\gamma_0(u_n)}} \\
&\leq \min \left\{ \ml \frac{\log(3)}{8C_{48}}\mr^2e^{-(9+6\xi)\gamma_0(u_n)}, \ml \frac{h}{1296C_{45}}\mr^2 e^{-\ml \frac{15}{2} + 3\xi\mr\gamma_0(u_n)}, \ml \frac{h}{1296 b C_{49}}\mr^2 e^{- \frac{19}{2}\gamma_0(u_n)} \right\}.
\end{aligned}
\end{equation}
To verify that the above assumptions all hold for all $u \in [u_2,u_n]$, let $u_*$ be the smallest $u \in [u_2,u_n]$ with \eqref{FIT35a}, \eqref{FIT35c}, \eqref{FIT35e} satisfied for all $u \in [u_*,u_n]$. We can see that, if $u = u_* > u_2$, then one of the following three scenarios would happen
\begin{equation}\label{FIT243}
\delta(u_*,v(u_n))\log \ml \frac{1}{\delta(u_*,v(u_n))}\mr  = C_{33}e^{-2\gamma_0(u_*)},
\end{equation}
or
\begin{equation}\label{FIT141}
\delta(u_*,v(u_n)) = \min\left\{ \frac{1}{2},\frac{2(1 - \mu_0(0))e^{-\gamma_0(u_*)}}{3}, \frac{e^{-\frac{3}{2}(1+2\xi)\gamma_0(u_*)}}{2 C_{34}} \right\},
\end{equation}
or
\begin{equation}\label{FIT258}
\delta(u_*,v(u_n)) = \min \left\{ \ml \frac{\log(3)}{8C_{48}}\mr^2e^{-(9+6\xi)\gamma_0(u_*)}, \ml \frac{h}{1296C_{45}}\mr^2 e^{-\ml \frac{15}{2} + 3\xi\mr\gamma_0(u_*)}, \ml \frac{h}{1296 b C_{49}}\mr^2 e^{- \frac{19}{2}\gamma_0(u_*)} \right\}.
\end{equation}
By referring to the proof of Lemma \ref{FITEstLemma5}, we have that \eqref{FIT83} is true, as long as the assumption \eqref{FIT35c} holds for all $u \in [u_3,u_4]$. By setting $u_3 = u_*$ and $u_4 = u_n$, we can see that \eqref{FIT83} implies
\begin{equation}\label{FIT240}
\begin{aligned}
\delta(u_*,v(u_n))&\log \ml \frac{1}{\delta(u_*,v(u_n))}\mr \leq \frac{2}{e}\delta(u_*,v(u_n))^{\frac{1}{2}}
\\
&\leq \frac{4}{e}\delta(u_n,v(u_n))^{\frac{1}{2}} e^{-\frac{1}{4}\log \ml \frac{r_0(u_*)}{r_0(u_n)}\mr + \frac{3}{4}\ml 1 + 2\xi \mr\gamma_0(u_n)}= 4e^{-\frac{C_{37}}{8C_{36}}e^{\gamma_0(u_n)} + \frac{3}{4}\ml 1 + 2\xi \mr\gamma_0(u_n)-1}
\end{aligned}
\end{equation}
and 
\begin{equation}\label{FIT256}
\begin{aligned}
\delta(u_*,v(u_n)) &\leq \delta(u_n,v(u_n)) 4 e^{-\frac{1}{2}\log \ml \frac{r_0(u_*)}{r_0(u_n)}\mr + \frac{3}{2}\ml 1 + 2\xi \mr\gamma_0(u_n)} \\
&\leq 4e^{-\frac{C_{37}}{4C_{36}}e^{\gamma_0(u_n)} + \frac{3}{2}\ml 1 + 2\xi \mr\gamma_0(u_n)}.
\end{aligned}
\end{equation}
Hence, we observe that \eqref{FIT256} contradicts both \eqref{FIT243} and \eqref{FIT258}, while \eqref{FIT240} contradicts \eqref{FIT141} for sufficiently large $n$. We then deduce that $u_* = u_2$ necessarily, and hence the assumptions \eqref{FIT35a}, \eqref{FIT35c}, and \eqref{FIT35e} hold for all $u \in [u_2,u_n]$. 

It remains to verify that \eqref{FIT252} is true for a sufficiently large $n$. To do so, by the continuity of $|\psi|(u_2,\cdot)$, it suffices to prove that as $n \rightarrow \infty$, we have $v(u_n) \rightarrow v_0$. By \eqref{FIT83}, as $n \rightarrow \infty$, we indeed have
\begin{equation}\label{FIT254}
\delta(u_2,v(u_n)) \leq 4\delta(u_n,v(u_n))e^{\frac{3}{2}\ml 1 + 2\xi \mr\gamma_0(u_n)} = 4e^{-\frac{C_{37}}{4C_{36}}e^{\gamma_0(u_n)}+\frac{3}{2}(1+2\xi)\gamma_0(u_n)} \rightarrow 0.
\end{equation}
Henceforth, employing the continuity of $\delta(u_2,\cdot)$, applying Lemma \ref{FITEstLemma4} and the fact that $\delta(u_2,v_0) = 0$, we conclude that $v(u_n) \rightarrow v_0$ as $n \rightarrow \infty$, which thus verifies \eqref{FIT252}. This completes our discussion for cases under \textit{Scenario (I)}.

\vspace{5mm}

For \textit{Scenario (II)}, in each of the cases, we have 
\begin{equation}
\limsup_{u \rightarrow v_0^-}|I(u)| = +\infty.
\end{equation}
This implies that there exists an increasing sequence $u_n \rightarrow v_0^-$ such that $|I(u_n)| \rightarrow +\infty$ as $n \rightarrow \infty$. Following \eqref{FIT207}, we now define $b:[0,v_0) \rightarrow \mathbb{R}$ as
\begin{equation}\label{FIT208}
b(u) := \sup_{u'\in [0,u]} \ml \frac{r|\dv \phi|}{\dv r}\mr_0(u')e^{-\gamma_0(u')}.
\end{equation}
Note that $b(u)$ is a non-decreasing function and $b(u_n)\rightarrow+\infty$ as $u_n\rightarrow v_0^-$. The latter follows from \eqref{FIT13} and $|I(u_n)| \rightarrow +\infty$ as $u_n \rightarrow v_0^-$. Hence, there exists an increasing sequence $u_n \rightarrow v_0^-$ such that
\begin{equation}\label{FIT259}
b(u_n) = \ml \frac{r|\dv \phi|}{\dv r}\mr_0(u_n)e^{-\gamma_0(u_n)}.
\end{equation}
Parallel to the arguments presented in \textit{Scenario (I)}, we observe that all the corresponding estimates still hold as long as the inequalities in \eqref{FIT35a} and \eqref{FIT35c} are satisfied. By modifying \eqref{FIT35e} via replacing $h$ with $1$ and replacing $b$ with $b(u)$ such that for each $u \in [u_2,v_0)$ we have
\begin{equation}\label{FIT35f}
\delta \ml u,v(u_n)\mr \leq \min \left\{ \ml \frac{\log(3)}{8C_{48}}\mr^2e^{-(9+6\xi)\gamma_0(u)}, \ml \frac{1}{1296C_{45}}\mr^2 e^{-\ml \frac{15}{2} + 3\xi\mr\gamma_0(u)}, \ml \frac{1}{1296 b(u) C_{49}}\mr^2 e^{- \frac{19}{2}\gamma_0(u)} \right\},
\end{equation}
then we can derive
\begin{equation}\label{FIT260}
\sup_{v\in[v_0,v(u_n)]} |\psi|(u_2,v) \leq \frac{1}{54}.
\end{equation}
This further implies
\begin{equation}\label{FIT261}
\sup_{v\in[v_0,v(u_n)]} |\psi|(u_n,v) \leq \frac{1}{6}.
\end{equation}
Henceforth, by imposing \eqref{FIT35a}, \eqref{FIT35c}, \eqref{FIT35f}, it suffices to only check that \eqref{FIT252} holds. To do so, we return to \eqref{FIT128} and instead we deduce that
\begin{equation}\label{FIT212}
\begin{aligned}
\frac{r(u_n,v')}{r_0(u_n)}\frac{r|\dv \phi|}{\dv r}(u_n,v') \geq \ml b(u_n) - \frac{1}{6} \mr e^{\gamma_0(u_n)} \geq \frac{1}{2}b(u_n) e^{\gamma_0(u_n)}
\end{aligned}    
\end{equation}
for sufficiently large $n$ such that $b(u_n) > \frac{1}{3}$. On the other hand, for $\eta^*(u_n,v(u_n))$,  instead of \eqref{FIT133}, we now have
\begin{equation}\label{FIT214}
\begin{aligned}
\eta^*(u_n,v(u_n)) &=  \frac{2}{r \ml u_n,v(u_n) \mr}\ml \int^{v(u_n)}_{v_0} \frac{2\pi r^2(1-\mu)|\dv \phi|^2}{\dv r}(u_n,v') \mathrm{d}v' \mr  \\
&\geq \frac{8\pi \ml 1-\mu_0(0)\mr }{3r \ml u_n,v(u_n)\mr }e^{-\gamma_0(u_n)}r^2_0(u_n)\ml \int^{v(u_n)}_{v_0} (\dv r)\frac{1}{r^2(u_n,v')}
\frac{r^2(u_n,v')}{r^2_0(u_n)}\frac{ r^2|\dv \phi|^2}{(\dv r)^2}(u_n,v') \mathrm{d}v' \mr  \\
&\geq \frac{2\pi \ml 1-\mu_0(0) \mr b(u_n)^2 }{3r \ml u_n,v(u_n) \mr }e^{\gamma_0(u_n)}r^2_0(u_n)\ml \int^{v(u_n)}_{v_0} \frac{(\dv r)(u_n,v')}{r^2(u_n,v')} \mathrm{d}v' \mr  \\
&= \frac{2\pi \ml 1-\mu_0(0) \mr b(u_n)^2 }{3}e^{\gamma_0(u_n)}\ml \frac{r_0(u_n)}{r \ml u_n,v(u_n)\mr } \mr^2\ml \frac{r(u_n,v(u_n)) - r_0(u_n)}{r_0(u_n)}\mr  \\
&\geq C_{50} b(u_n)^2\delta \ml u_n,v(u_n) \mr e^{\gamma_0(u_n)}  \\
\end{aligned}
\end{equation}
with 
\begin{equation}\label{FIT215}
C_{50} = \frac{8\pi(1-\mu_0(0))}{27}.
\end{equation}
Comparing with \eqref{FIT131}, we observe that a contradiction arises if the following inequality holds
\begin{equation}\label{FIT216}
\log \ml \frac{1}{\delta(u_n,v(u_n))}\mr e^{-\gamma_0(u_n)} + r_0(u_n)  < \frac{C_{50}}{C_{36}} b(u_n)^2.
\end{equation}
For this to happen, we pick
\begin{itemize}[leftmargin = *]
\item $n$ sufficiently large such that \begin{equation}\label{FIT217}
r_0(u_n) \leq \frac{C_{50}}{4C_{36}}b(u_n)^2
\end{equation}
as $r_0(u_n) \rightarrow 0$ while $b(u_n) \rightarrow + \infty$, 
\item $v(u_n)$ sufficiently close to $v_0$ verifying
\begin{equation}\label{FIT218}
\delta(u_n,v(u_n)) = \frac{1}{2}e^{-b(u_n)^2\frac{C_{50}}{4C_{36}}e^{\gamma_0(u_n)}}.
\end{equation}
\end{itemize}
It remains to show that \eqref{FIT218} implies \eqref{FIT35a}, \eqref{FIT35c}, \eqref{FIT35f}, \eqref{FIT260}. By employing a similar argument as in \textit{Scenario (I)} presented above, it suffices to conduct the following modifications to \eqref{FIT240} and \eqref{FIT256} respectively, i.e.,
\begin{equation}\label{FIT262}
\delta(u_*,v(u_n))\log \ml \frac{1}{\delta(u_*,v(u_n))}\mr \leq 4e^{-\frac{C_{50}}{8C_{36}}b(u_n)e^{\gamma_0(u_n)}+\frac{3}{4}(1+2\xi)\gamma_0(u_n) - 1}
\end{equation}
and
\begin{equation}\label{FIT263}
\delta(u_*,v(u_n)) \leq 4e^{-\frac{C_{50}}{4C_{36}}b(u_n)e^{\gamma_0(u_n)}+\frac{3}{2}(1+2\xi)\gamma_0(u_n)}.
\end{equation}
With \eqref{FIT262} and \eqref{FIT263}, in an analogous manner as above, we can conclude that \eqref{FIT35a}, \eqref{FIT35c}, \eqref{FIT35f} hold. To verify \eqref{FIT260}, we use the fact that 
$$\delta(u_2,v(u_n)) \leq 4e^{-\frac{C_{50}}{4C_{36}}b(u_n)e^{\gamma_0(u_n)}+\frac{3}{2}(1+2\xi)\gamma_0(u_n)} \rightarrow 0.$$ 
This enables us that we can argue exactly as in \textit{Scenario (I)} to deduce that \eqref{FIT260} is satisfied. This then concludes the proof for the cases in \textit{Scenario (II)}. And we finish the proof of Theorem \ref{FITTheorem}.
\end{proof}

\section{Second Instability Theorem}\label{Second Instability Theorem}
We proceed to prove the second instability theorem for the charged system that we are working with. To do so, we consider the case not covered in Theorem \ref{FITTheorem}, i.e., with $I$ as defined in \eqref{FIT14} satisfying
\begin{equation}\label{SIT1A}
\lim_{u \rightarrow v_0^-}I(u) < + \infty
\end{equation}
and
\begin{equation}\label{SIT1}
\ml \frac{r \dv \phi}{\dv r} \mr_0 (0) = \lim_{u \rightarrow v_0^-} I(u).
\end{equation}
In \cite{christ4}, Christodoulou proved the corresponding second instability theorem for the spherically symmetric Einstein-scalar field system by asserting that he has a positive constant $p < 1$ such that
\begin{equation}\label{SIT2}
\limsup_{u \rightarrow v_0^-}\left\{ \mlm \ml \frac{r \dv \phi}{\dv r}\mr_0(0) - I(u)\mrm e^{\frac{1}{2}p\gamma_0(u)} \right\} \neq 0.
\end{equation}
However, for the charged case, it is not clear \textit{a priori} whether the above assertion still holds. Here we give an argument for the charged case without asserting this assumption. The theorem to be proven in this section is recorded below
\begin{theorem}\label{SITTheorem} (Second Instability Theorem.) Along $v = v_0$, assume that $\gamma_0(u)$ is unbounded as $u \rightarrow v_0^-$. Furthermore, we require that the limit for $I(u)$ exists and is given by \eqref{SIT1}. Under these assumptions, for $v \in [v_0,\infty]$, we define
\begin{equation}\label{SIT48}
h^2(v) := \frac{1}{v - v_0}\int_{v_0}^{v} |\psi|^2(0,v') \D v',
\end{equation}
with $\psi$ given in \eqref{FIT30}. Let $v(u)$ be the function given by
\begin{equation}\label{SIT166}
v(u) = v_0 + e^{-14\gamma_0(u) - \log \ml \frac{r_0(0)}{r_0(u)}\mr}
\end{equation}
for $u \in [u_7,v_0)$ with some $u_7 < v_0$ sufficiently close to $v_0$. Also denote $g(v(u))$ to be 
\begin{equation}\label{SIT167}
g(v(u)) :=  V(u) + e^{-\frac{1}{4}\gamma_0(u)} + e^{- \frac{1}{2}\log\ml\frac{r_0(0)}{r_0(u)}\mr}
\end{equation} 
with
\begin{equation}\label{SIT169}
V(u) := \mlm \ml \frac{r \dv \phi}{\dv r}\mr_0(0) - I(u)\mrm.
\end{equation}
Based on these requirements, if the following assumption holds
\begin{equation}\label{SIT168}
\limsup_{u \rightarrow v_0^-}\frac{h(v(u))}{g(v(u))} = \infty,
\end{equation} then the same conclusion as for Theorem \ref{FITTheorem} still holds.
\end{theorem}

Before we begin with the proof, we first state a useful calculus lemma.
\begin{lemma}\label{SITLemma1}
Set $X$ to be an arbitrary measurable subset in $\mathbb{R}$ and let functions $f:X \rightarrow [0,\infty)$, $g:X \rightarrow \mathbb{C}$, $h:X \rightarrow \mathbb{C}$ be arbitrary measurable functions defined on $X$. If we have
\begin{equation}\label{SIT3}
\int_X f|g|^2 \geq 4\int_X f|h|^2,
\end{equation}
then the following is true:
\begin{equation}\label{SIT4}
\int_X f|g+h|^2 \geq \frac{1}{4}\int_X f|g|^2 \geq \int_X f|h|^2.
\end{equation}
\end{lemma}
The proof of this lemma follows from a simple inequality
\begin{equation}\label{SIT5}
|g|^2 - 2 |h|^2 = |g+h -h|^2 - 2|h|^2 \leq (|g+h| + |h|)^2  - 2|h|^2 \leq 2 |g+h|^2
\end{equation}
and its integrated form
\begin{equation}\label{SIT6}
\begin{aligned}
\int_X f|g+h|^2 &\geq \frac{1}{2}\ml\int_X f|g|^2 - 2 \int_X f|h|^2\mr \geq \frac{1}{4}\int_X f|g|^2 \geq \int_X f|h|^2.
\end{aligned}
\end{equation}

\vspace{5mm}

We now start with the proof of Theorem \ref{SITTheorem}.
\proof 
Following the proof of Theorem \ref{FITTheorem} and arguing by means of contradiction, by constructing $\mathcal{P}'$ via Lemma \ref{FITEstLemma8} such that \eqref{FIT226} and \eqref{FIT77} hold, for each $u \in [u_2,v_0)$, we have
\begin{equation}\label{SIT11}
\ml \frac{r |\dv \phi|}{\dv r}\mr_0(u) = \mlm \ml \frac{r \dv \phi}{\dv r}\mr_0(0) - I(u)\mrm e^{\gamma_0(u)} = V(u)e^{\gamma_0(u)}
\end{equation}
with
\begin{equation}
\lim_{u \rightarrow v_0^-}V(u) = 0.
\end{equation}
In contrast to the scenario in the proof of Theorem \ref{FITTheorem} in which the limit only holds along a sequence approaching $v_0$, here \eqref{SIT11} holds for every $u \in [u_2,v_0)$. Hence, we can make a choice of $v(u) \in [v_0, v_0 + \te]$ for each fixed $u \in [u_2,v_0)$ as compared to along a sequence. Furthermore, for a given $\tu \in [u_2,v_0)$, $\tu$ plays the role of $u_n$ in the proof of Theorem \ref{FITTheorem}. Hence, the relevant assumptions \eqref{FIT35a} and \eqref{FIT35c} with $u_n$ replaced by $\tu$ imply that the corresponding estimates in Lemma \ref{FITEstLemma3} and  \ref{FITEstLemma5} hold with $u_n$ replaced by $\tu$.

Next, we prove a lemma that will be used in later estimates.
\begin{lemma}\label{SITLemma3}
Given $\tu \in [u_2,v_0)$, suppose that there exists a corresponding value of $v(\tu) \in [v_0,v_0 + \te]$ such that, for $\delta(u, v(\tu))$, conditions \eqref{FIT35a} and \eqref{FIT35c} hold, and a new requirement
\begin{equation}\label{SIT35c}
 \delta(u,v(\tu)) \leq \frac{e^{-\ml \frac{15}{2} + 3 \xi \mr\gamma_0(u)}}{64C_{53}^2}
\end{equation}
is satisfied for all $u \in [u_2,\tu]$ with $C_{53}$ defined in \eqref{SIT87}.
Then for each $u \in [u_2,\tu]$ and $v' \in [v_0,v(\tu)]$, we have
\begin{equation}\label{SIT81}
|\tgam(u,v') - \tgam_0(u)| \leq 1.
\end{equation}
Here, $\tgam(u,v') = \gamma(u,v';u_2)$ is defined in \eqref{FIT4a}.
\end{lemma}
\begin{proof}
To obtain \eqref{SIT81}, we consider $\dv \gamma$. Hence, we compute $\dv \ml \frac{-\du r}{r}\frac{\mu - Q^2/r^2}{1-\mu}\mr$. By \eqref{FIT86}, we have
\begin{equation}\label{SIT82}
\dv \ml \frac{-\du r}{r}\frac{\mu - Q^2/r^2}{1-\mu}\mr = J_8 + J_{9} + J_{10}
\end{equation}
with the expressions for $J_8, J_{9}, J_{10}$ available in \eqref{FIT248}. Hence, by Fubini's theorem and dominated convergence theorem \footnote{The use of Fubini's theorem and dominated convergence theorem are guaranteed by the estimates for $J_9, J_{10}, J_{11}$, which we will derive below.}, for each $u \in [u_2,\tu]$ and $v \in [v_0,v(\tu)]$, we obtain
\begin{equation}\label{SIT83}
|\tgam(u,v) - \tgam_0(u)| \leq \int^u_{u_2} \int^v_{v_0} |J_8|(u',v') + |J_{9}|(u',v') + |J_{10}|(u',v') \D v' \D u'.
\end{equation}
For the inner integral along $u = u'$ with $u' \in [u_2,u]$, estimates for $|J_{8}|$ and $|J_{9}|$ have been derived in \eqref{FIT189} and \eqref{FIT190}. In a similar fashion, we can also obtain an estimate for $|J_{10}|$. By \eqref{FIT226} for $\xi \in \left(0,\oh\right)$, we have $0 \leq \mu - \frac{2Q^2}{r^2}\leq 1$, and this further implies
\begin{equation}\label{SIT84}
\begin{aligned}
\int^v_{v_0} |J_{10}|(u',v') \D v' &\leq 9 \ml \frac{-\du r}{r} \mr_0(u') \int^v_{v_0}\frac{\dv r}{r} \frac{1}{1-\mu}(u',v') \D v' \\
&\leq \frac{27}{2(1-\mu_0(0))} \ml \frac{-\du r}{r} \mr_0(u') e^{\gamma_0(u')}\td(u',v) \\
&\leq C_{52} \ml \frac{-\du r}{r}\mr_0(u') e^{3\gamma_0(u')}\delta(u',v)^{\frac{1}{2}}
\end{aligned}
\end{equation}
with
\begin{equation}\label{SIT85}
C_{52} := \frac{27}{2\sqrt{2}(1-\mu_0(0))}.
\end{equation}
Together with \eqref{FIT189} and \eqref{FIT190}, we henceforth obtain
\begin{equation}\label{SIT86}
\int^v_{v_0} |J_8|(u',v') + |J_{9}|(u',v') + |J_{10}|(u',v') \D v' \leq C_{53}\ml \frac{-\du r}{r} \mr_0(u')e^{3\gamma_0(u')}\delta(u',v)^{\frac{1}{2}}
\end{equation}
with
\begin{equation}\label{SIT87}
C_{53} := C_{46} + C_{47} + C_{52}.
\end{equation}
Conducting another integration along a given incoming null curve, we appeal to \eqref{FIT83} of Lemma \ref{FITEstLemma5} and \eqref{SIT83} to obtain
\begin{equation}\label{SIT88}
\begin{aligned}
|\tgam(u,v) - \tgam_0(u)| &\leq \int^u_{u_2} \int^v_{v_0} |J_8|(u',v') + |J_{9}|(u',v') + |J_{10}|(u',v') \D v' \D u' \\
&\leq C_{53} e^{3\gamma_0(u)}\int^u_{u_2} \ml \frac{-\du r}{r}\mr_0(u') \delta(u',v)^{\frac{1}{2}}\D u' \\
&\leq 2C_{53} e^{ \ml \frac{15}{4} + \frac{3 \xi}{2} \mr \gamma_0(u)}\delta(u,v)^{\frac{1}{2}} \int^u_{u_2} \ml \frac{-\du r}{r}\mr_0(u') e^{-\frac{1}{4}\log \ml \frac{r_0(u')}{r_0(u)}\mr}\D u' \\
&\leq 8C_{53} e^{\ml \frac{15}{4} + \frac{3 \xi}{2} \mr\gamma_0(u)}\delta(u,v(\tu))^{\frac{1}{2}} \leq 1,
\end{aligned}
\end{equation}
in which the last inequality follows from our assumption \eqref{SIT35c}. 
\end{proof}

We are ready to demonstrate the relevant contradiction for proving Theorem \ref{SITTheorem}. Assuming the absence of a trapped surface, at $(\tu,v(\tu)) \in \mathcal{R}$, employing \eqref{FIT130} and  \eqref{FIT131}, we get
\begin{equation}\label{SIT104}
\begin{aligned}
\eta^*(\tu,v(\tu)) &= \frac{4\pi}{r(\tu,v(\tu))}\int^{v(\tu)}_{v_0} (1-\mu) \ml \frac{r^2|\dv \phi|^2}{\dv r}\mr(\tu,v') \D v' \\
&\leq C_{36}\delta(\tu,v(\tu))\ml \log \ml \frac{1}{\delta(\tu,v(\tu))}\mr + e^{\gamma_0(\tu)}r_0(\tu)\mr.
\end{aligned}
\end{equation}
Instead of using the usual upper bounds for $\delta(\tu,v(\tu))$ as in \eqref{FIT35a}, \eqref{FIT35c}, \eqref{SIT35c}, we derive an alternative upper bound for $\delta(\tu,v(\tu))$ as follows:
\begin{equation}\label{SIT135}
\begin{aligned}
\delta(\tu,v(\tu)) &= \frac{1}{r_0(\tu)}\int^{v(\tu)}_{v_0} (\dv r)(\tu,v')\D v' = \frac{1}{r_0(\tu)} \int^{v(\tu)}_{v_0} (\dv r)(0,v')e^{-\tgam(\tu,v') - \gamma(u_2,v')} \D v' \\
&\quad\quad\leq \frac{e C_{56}}{2r_0(\tu)} \int^{v(\tu)}_{v_0} e^{-\gamma_0(\tu) } \D v' \leq \frac{e C_{56} e^{-\gamma_0(\tu)}}{2r_0(\tu)}(v(\tu) - v_0).
\end{aligned}
\end{equation}
Here
\begin{equation}\label{SIT136}
C_{56}(u_2) := \max_{v'\in[v_0,v_0 + \te]}e^{|\gamma|(u_2,v')} < + \infty
\end{equation}
and we have utilized \eqref{FIT18a} for the evolution of $\dv r$ along $v = v'$ with $v' \in [v_0, v(\tu)] \subset [v_0,v_0 + \te]$. Note that $C_{56}(u_2)$ is finite due to the fact that $\gamma(u_2,\cdot)$ is continuous on $[v_0,v_0+\te]$. By the non-decreasing property of the mapping $x \mapsto x\log\ml \frac{1}{x}\mr$ for $x \in \left( 0, \frac{1}{e}\right]$, if we further require 
\begin{equation}\label{SIT138}
v(\tu) - v_0 \leq \frac{2r_0(\tu)}{e^2 C_{56}}e^{\gamma_0(\tu)},
\end{equation}
then by applying \eqref{SIT135} to \eqref{SIT104}, we derive 
\begin{equation}\label{SIT139}
\begin{aligned}
& \eta^*(\tu,v(\tu)) \\
\leq & \; C_{36}\ml \frac{eC_{56}e^{-\gamma_0(\tu)}}{2r_0(\tu)}(v(\tu) - v_0) \log \ml \frac{2r_0(\tu)}{eC_{56} e^{-\gamma_0(\tu)}(v(\tu) - v_0)}\mr + \frac{eC_{56}}{2}(v(\tu) - v_0) \mr \\
\leq & \; \frac{eC_{36}C_{56}}{2r_0(\tu)}(v(\tu) - v_0)\ml r_0(\tu) + e^{-\gamma_0(\tu)}\ml \log \ml \frac{1}{v(\tu) - v_0}\mr + \log \ml \frac{2r_0(0)}{eC_{56}}\mr  + 
 \log \ml\frac{r_0(\tu)}{r_0(0)} \mr + \gamma_0(\tu)\mr\mr.
\end{aligned}
\end{equation}
To obtain the contradiction, it is natural to obtain a lower bound for $\eta^*(\tu,v(\tu))$. To do so, observe by the definition of $\Upxi$ in \eqref{FIT29} that
\begin{equation}\label{SIT7}
\begin{aligned}
\ml \frac{r |\dv \phi|}{\dv r}\mr &= \frac{r_0}{r}\mlm \ml \frac{r \dv \phi}{\dv r}\mr_0 + \Upxi \mrm. \\
\end{aligned}
\end{equation}
This implies that, along $u = \tu$, we have 
\begin{equation}\label{SIT8}
\begin{aligned}
&\int_{v_0}^{v(\tu)} (1-\mu) \ml \frac{r^2 |\dv \phi|^2}{\dv r}\mr (\tu,v')  \D v'\\
=& \; \int_{v_0}^{v(\tu)} (1-\mu) \ml \frac{r_0}{r}\mr^2 \mlm \ml \frac{r \dv \phi}{\dv r}\mr_0 + \Upxi \mrm^2 (\dv r) (\tu,v')  \D v' \\
\geq&\; \frac{8}{27}(1-\mu_0(0))e^{-\gamma_0(\tu)}\int_{v_0}^{v(\tu)}  \mlm \ml \frac{r \dv \phi}{\dv r}\mr_0 + \Upxi \mrm^2 (\dv r) (\tu,v')  \D v' \\
=&\; \frac{8}{27}(1-\mu_0(0))e^{-\gamma_0(\tu)}\int_{v_0}^{v(\tu)}  \mlm \ml \frac{r \dv \phi}{\dv r}\mr_0 + \Upxi \mrm^2 (\dv r) (0,v')e^{-\tgam(\tu,v') - \gamma(u_2,v')}  \D v', \\
\end{aligned}
\end{equation}
where we have employed \eqref{FIT18a} and a similar argument as in \eqref{FIT133} to obtain a lower bound for $1-\mu$ and $r/r_0$. Additionally, by applying \eqref{Setup33} and \eqref{SIT81} of Lemma \ref{SITLemma3}, we obtain
\begin{equation}\label{SIT105}
\begin{aligned}
\int_{v_0}^{v(\tu)} (1-\mu) \ml \frac{r^2 |\dv \phi|^2}{\dv r}\mr (\tu,v')  \D v' \geq \frac{4 C_{54}}{27 e}(1-\mu_0(0))e^{-2\gamma_0(\tu)} \int_{v_0}^{v(\tu)}  \mlm \ml \frac{r \dv \phi}{\dv r}\mr_0 + \Upxi \mrm^2 (\tu,v') \D v'
\end{aligned}
\end{equation}
with
\begin{equation}\label{SIT106}
C_{54}(u_2) = \min_{v' \in [v_0,v_0+\te]} e^{-\gamma(u_2,v')} > 0.
\end{equation}
For the inequality in \eqref{SIT105}, we have used the fact that as $\gamma$ is a continuous function on the rectangular region $[0,u_2] \times [v_0,v_0 + \te] \subset \mathcal{R}$. 

Contrary to the proof of Theorem \ref{FITTheorem}, we instead require that the dominant term for the lower bound of $\eta^*$ is $\Upxi$ as we have an upper bound for $\ml \frac{r|\dv \phi|}{\dv r}\mr_0$ given in \eqref{SIT11}. Therefore, by Lemma \ref{SITLemma1}, if we can show that
\begin{equation}\label{SIT9}
\int_{v_0}^{v(\tu)} \mlm \Upxi \mrm^2  (\tu,v')  \D v' \geq 4 \int_{v_0}^{v(\tu)}  \mlm \frac{r \dv \phi}{\dv r} \mrm_0^2  (\tu,v')  \D v',
\end{equation}
then we have
\begin{equation}\label{SIT10}
\int_{v_0}^{v(\tu)}  \mlm \ml \frac{r \dv \phi}{\dv r}\mr_0 + \Upxi \mrm^2 (\tu,v') \D v' \geq \frac{1}{4}\int_{v_0}^{v(\tu)} \mlm \Upxi \mrm^2  (\tu,v')  \D v'.
\end{equation}
To show that \eqref{SIT9} holds, by applying \eqref{SIT11}, we observe that
\begin{equation}\label{SIT107}
4 \int_{v_0}^{v}  \mlm \frac{r \dv \phi}{\dv r} \mrm_0^2  (\tu,v')  \D v' \leq 4( v(\tu) - v_0 ) V^2(\tu)e^{2\gamma_0(\tu)}.
\end{equation}
Showing $4( v(\tu) - v_0 ) V^2(\tu)e^{2\gamma_0(\tu)} \leq \int_{v_0}^{v(\tu)}  \mlm \Upxi \mrm^2  (\tu,v')  \D v' $ requires more work. Together with \eqref{FIT30} and \eqref{FIT154}, for all $v \in [v_0,v(\tu)]$, we first write
\begin{equation}\label{SIT42A}
\begin{aligned}
\Upxi(\tu,v) &= e^{\gamma_0(\tu) + \int^{\tu}_{0} J_{5} (u',v) - \ii \e A_u(u',v) \D u'}  \ml  \psi(0,v) + J_{13}(\tu,v) \mr. 
\end{aligned}
\end{equation}
Since $J_5$ and $A_u$ are real-valued, we can write
\begin{equation}\label{SIT42}
\begin{aligned}
|\Upxi|^2(\tu,v) &= e^{2\gamma_0(\tu) + 2\int^{\tu}_{0} J_{5} (u',v) \D u'}  \mlm  \psi(0,v) + J_{13}(\tu,v) \mrm^2 \\
\end{aligned}
\end{equation}
with
\begin{equation}\label{SIT43}
J_{13}(\tu,v) := \int^{\tu}_{0} e^{-\gamma_0(u'')-\int^{u''}_{0} J_{5}(u',v) - \ii \e A_u(u',v) \D u'}\ml J_4(u'',v) + J_{6}(u'',v)\ml \frac{r \dv \phi}{\dv r}\mr_0(u'')\mr \D u''.
\end{equation}
Integrating \eqref{SIT42} along the incoming null curve $u = \tu$, we thus have
\begin{equation}\label{SIT44}
\begin{aligned}
\int_{v_0}^{v(\tu)} \mlm \Upxi \mrm^2 (\tu,v')  \D v' 
= e^{2\gamma_0(\tu)}\int^{v(\tu)}_{v_0} e^{2\int^{\tu}_{0} J_{5}(u',v') \D u'}|\psi(0,v') + J_{13}(\tu,v')|^2  \D v'.
\end{aligned}
\end{equation}
Next, for each $v' \in [v_0,v(\tu)]$, we consider the decomposition
\begin{equation}\label{SIT108}
\int^{\tu}_0 |J_5| (u',v') \D u' = \int^{\tu}_{u_2} |J_5| (u',v') \D u' + \int^{u_2}_0 |J_5| (u',v') \D u'.
\end{equation}
By the fact that $J_5$ is continuous on $[0,u_2]$, we have
\begin{equation}\label{SIT109}
C_{55}(u_2) := \sup_{v' \in [v_0,v_0 + \te]}\int^{u_2}_0 |J_5| (u',v') \D u' < + \infty.
\end{equation}
We then repeat the computations in Lemma \ref{FITpsievolution} up to \eqref{FIT250} to obtain
\begin{equation}\label{SIT110}
\int^{\tu}_{u_2} |J_5| (u',v') \D u' \leq \log(3).
\end{equation}
Consequently, this implies that
\begin{equation}\label{SIT111}
e^{2 \int^{\tu}_0 J_5(u',v') \D u'} \geq \frac{1}{9}e^{-2C_{55}}.
\end{equation}
Back to \eqref{SIT44}, we thus derive
\begin{equation}\label{SIT112}
\int_{v_0}^{v(\tu)} \mlm \Upxi \mrm^2 (\tu,v')  \D v' 
\geq \frac{e^{-2C_{55}}}{9}e^{2\gamma_0(\tu)}\int^{v(\tu)}_{v_0} |\psi(0,v') + J_{13}(\tu,v')|^2  \D v'.
\end{equation}
To provide a lower bound for $\int^{v(\tu)}_{v_0} |\psi(0,v') + J_{13}(u,v')|^2  \D v'$ in \eqref{SIT112}, we first attempt to show that
\begin{equation}\label{SIT113}
\int_{v_0}^{v(\tu)} |\psi|^2(0,v') \D v' \geq 4 \int_{v_0}^{v(\tu)} |J_{13}|^2(\tu,v') \D v' .
\end{equation}
Referring to the expression of $|J_{13}|$ in \eqref{SIT43}, together with \eqref{SIT109},\eqref{SIT110}, for each $v' \in [v_0,v(\tu)]$, we obtain
\begin{equation}\label{SIT114}
\begin{aligned}
|J_{13}|(\tu,v') \leq  3 e^{C_{55}}\int^{\tu}_{0} \mlm e^{-\gamma_0(u'')}\mrm \ml |J_4|(u'',v') + |J_{6}|(u'',v')\ml \frac{r |\dv \phi|}{\dv r}\mr_0(u'')\mr \D u''.
\end{aligned}
\end{equation} 
Note that $\gamma_0(u)$ is not necessarily non-negative for the charged case and thus necessitates the term $\mlm e^{-\gamma_0(u'')}\mrm$ to remain in \eqref{SIT114}. To estimate the integral on the right of \eqref{SIT114}, we first note that
\begin{equation}\label{SIT116}
C_{56} := \sup_{u \in [0,v_0)} V(u) < + \infty
\end{equation}
because $V(u) \rightarrow 0$ as $u \rightarrow v_0^-$ and $I(u)$ is continuous on any compact subsets of $[0,v_0)$. We continue to estimate
\begin{equation}\label{SIT115}
\begin{aligned}
&\int_{u_2}^{\tu} \mlm e^{-\gamma_0(u'')}\mrm \ml |J_4|(u'',v') + |J_{6}|(u'',v')\ml \frac{r |\dv \phi|}{\dv r}\mr_0(u'')\mr \D u'' \\
\leq & \; \int_{u_2}^{\tu} \ml |J_4|(u'',v') + |J_{6}|(u'',v')\ml \frac{r |\dv \phi|}{\dv r}\mr_0(u'')\mr \D u'' \\
\leq & \; 8C_{45}e^{\ml \frac{15}{4} + \frac{3}{2}\xi \mr\gamma_0(u)}\delta(\tu,v(\tu))^{\frac{1}{2}} + 8C_{49}C_{56}e^{\ml \frac{19}{4} + \frac{3}{2}\xi \mr\gamma_0(u)}\delta(\tu,v(\tu))^{\frac{1}{2}}.\\
\end{aligned}
\end{equation}
Here we have employed \eqref{SIT11} and similar arguments used to derive \eqref{FIT197} and \eqref{FIT201} in Lemma \ref{FITpsievolution}. 
To proceed, we impose
\begin{equation}\label{SIT117}
\delta(\tu,v(\tu)) \leq \min \left\{ \frac{e^{-\ml \frac{19}{2} + 3 \xi\mr\gamma_0(\tu)}}{256 C_{45}^2}, \frac{e^{-\ml \frac{23}{2} + 3 \xi\mr\gamma_0(\tu)}}{256 C_{49}^2C_{56}^2} \right\}.
\end{equation}
As a result, from \eqref{SIT115}, we obtain
\begin{equation}\label{SIT118}
\int_{u_2}^{\tu} \mlm e^{-\gamma_0(u'')}\mrm \ml |J_4|(u'',v') + |J_{6}|(u'',v')\ml \frac{r |\dv \phi|}{\dv r}\mr_0(u'')\mr \D u''  \leq e^{-\gamma_0(\tu)}.
\end{equation}
On the other hand, for the contribution of the integral from $u = 0$ to $u = u_2$, we have
\begin{equation}\label{SIT119}
\begin{aligned}
&\int_{0}^{u_2} \mlm e^{-\gamma_0(u'')}\mrm \ml |J_4|(u'',v') + |J_{6}|(u'',v')\ml \frac{r |\dv \phi|}{\dv r}\mr_0(u'')\mr \D u''  \\
\leq & \; C_{57} \int^{u_2}_0 |J_4|(u'',v') + |J_6|(u'',v') \D u''
\end{aligned}
\end{equation}
with
\begin{equation}\label{SIT120}
C_{57}(u_2) := \max_{u' \in [0,u_2]}\mlm e^{-\gamma_0(u')}\mrm \ml 1 + \max_{u' \in [0,u_2]} \ml \frac{r|\dv \phi|}{\dv r}\mr_0(u')\mr.
\end{equation}
To bound $|J_4|$, using \eqref{FIT152} and \eqref{FIT161}, we obtain 
\begin{equation}\label{SIT121}
|J_4|(u'',v') \leq \int^v_{v_0}|H_6|(u'',v'') \D v'' \leq C_{58}(v'-v_0)
\end{equation}
with
\begin{equation}\label{SIT122}
\begin{aligned}
H_6 :=&\; \frac{(-\du r) \dv \phi}{r_0} + \ii \e \frac{Q (\du r)(\dv r) \phi}{r_0 r(1-\mu)} +  (4 \pi \ii \e^2)\frac{\phi(-\du r)r^2 \im{(\phi^{\dagger} \dv \phi)}}{r_0(1-\mu)} \\
&+\ii \e \frac{Q (\dv \phi)(\du r)}{r_0(1-\mu)} +  4 \pi \ii \e \frac{Q r |\dv \phi|^2(\du r) \phi}{r_0(\dv r)(1-\mu)}
\end{aligned}
\end{equation}
being a continuous function on $[0,u_2]\times[v_0,v(\tu)]$. This implies that there exists a constant $C_{58} > 0$ such that $|H_6|(u'',v'') \leq C_{58}$ for all $(u'',v'') \in [0,u_2]\times[v_0,v(\tu)]$. Similarly, from \eqref{Setup30}, \eqref{FIT152}, \eqref{FIT249}, together with the bound on $|H_6|$ by $C_{58}$, we define
\begin{equation}\label{SIT124}
H_7 := J_9 + J_{10} + J_{12} - \ii\e\ml \frac{2Q \du r 
 \dv r}{r^2(1-\mu)}\mr
\end{equation}
and observe that since $H_{7}$ is a continuous function on $[0,u_2]\times[v_0,v(\tu)]$, there exists by a constant $C_{59} > 0$ such that $|H_7|(u'',v'') \leq C_{59}$ for all $(u'',v'') \in [0,u_2]\times[v_0,v(\tu)]$.  Hence, we have
\begin{equation}\label{SIT123}
|J_6|(u'',v') \leq \int^v_{v_0}|H_6| + |H_7|(u'',v'') \D v'' \leq (C_{58} + C_{59})(v'-v_0).
\end{equation}

Combining these estimates into \eqref{SIT114}, we derive the following estimate for $|J_{13}|$:
\begin{equation}\label{SIT125}
|J_{13}|(\tu,v') \leq 3e^{C_{55}}(e^{-\gamma_0(\tu)} + C_{57}(2C_{58} + C_{59})u_2(v' - v_0)) \\
\leq C_{60}(e^{-\gamma_0(\tu)} + (v' - v_0))
\end{equation}
with
\begin{equation}\label{SIT126}
C_{60}(u_2) := 3e^{C_{55}}\ml 1 + C_{57}(2C_{58} + C_{59})u_2 \mr.
\end{equation}
Consequently, this implies
\begin{equation}\label{SIT127}
\begin{aligned}
4 \int_{v_0}^{v(\tu)}|J_{13}|^2(\tu,v') \D v'  &\leq 8C_{60}^2 \int_{v_0}^{v(\tu)} e^{-2\gamma_0(\tu)} + (v'-v_0)^2 \D v' \\ 
&\leq 8C_{60}^2\ml e^{-2\gamma_0(\tu)} + (v(\tu) - v_0)^2 \mr (v(\tu) - v_0).
\end{aligned}
\end{equation}
Henceforth, by the definition of $h^2(v(\tu))$ in \eqref{SIT48}, verifying \eqref{SIT113} is equivalent to requiring 
\begin{equation}\label{SIT128}
h^2(v(\tu)) \geq 8C_{60}^2(e^{-2\gamma_0(\tu)} + (v(\tu) - v_0)^2).
\end{equation}
Now if \eqref{SIT128} is true, via \eqref{SIT112} and \eqref{SIT113}, we have
\begin{equation}\label{SIT129}
\int^{v(\tu)}_{v_0} |\Upxi|^2(\tu,v') \D v' \geq \frac{e^{-2C_{55}}}{36}e^{2\gamma_0(\tu)}\int^{v(\tu)}_{v_0} |\psi|^2(0,v') \D v'.
\end{equation}
With this, via \eqref{SIT107} and \eqref{SIT129}, validating \eqref{SIT9} is now the same as doing so for 
\begin{equation}\label{SIT130}
4(v(\tu)-v_0)V^2(\tu)e^{2\gamma_0(\tu)} \leq \frac{e^{-2C_{55}}}{36}e^{2\gamma_0(\tu)}\int^{v(\tu)}_{v_0} |\psi|^2(0,v') \D v'.
\end{equation}
Henceforth, we can instead require
\begin{equation}\label{SIT131}
h^2(v(\tu)) \geq 144 e^{2C_{55}} V^2(\tu).
\end{equation}
Recall that \eqref{SIT9} implies \eqref{SIT10}, which by \eqref{SIT129}, in turn implies that
\begin{equation}\label{SIT132}
\int_{v_0}^{v(\tu)}  \mlm \ml \frac{r \dv \phi}{\dv r}\mr_0 + \Upxi \mrm^2 (\tu,v') \D v' \geq \frac{e^{-2C_{55}}}{144}e^{2\gamma_0(\tu)}\int^{v(\tu)}_{v_0}|\psi|^2(0,v')\D v'.
\end{equation}
Hence by  \eqref{SIT8}, we have
\begin{equation}\label{SIT133}
\begin{aligned}
\eta^*(\tu,v(\tu)) &= \frac{4\pi}{r(\tu,v(\tu))}\int^{v(\tu)}_{v_0} (1-\mu) \ml \frac{r^2|\dv \phi|^2}{\dv r}\mr(\tu,v') \D v' \\
&\geq \frac{16\pi C_{54}(1-\mu_0(0))}{27e r(\tu,v(\tu))}e^{-2\gamma_0(\tu)}\int_{v_0}^{v(\tu)}  \mlm \ml \frac{r \dv \phi}{\dv r}\mr_0 + \Upxi \mrm^2 (\tu,v') \D v' \\
&\geq \frac{\pi C_{54}e^{-2C_{55}}(1-\mu_0(0))}{243e r(\tu,v(\tu))} \int^{v(\tu)}_{v_0}|\psi|^2(0,v')\D v' \\
&\geq \frac{2\pi C_{54}e^{-2C_{55}}(1-\mu_0(0))}{729 e r_0(\tu)} \int^{v(\tu)}_{v_0}|\psi|^2(0,v')\D v'.
\end{aligned}
\end{equation}\label{SIT134}
This would contradict \eqref{SIT139} if 
\begin{equation}\label{SIT140}
\begin{aligned}
& \frac{eC_{36}C_{56}}{2r_0(\tu)}(v(\tu) - v_0)\ml r_0(\tu) + e^{-\gamma_0(\tu)}\ml \log \ml \frac{1}{v(\tu) - v_0}\mr + \log \ml \frac{2r_0(0)}{eC_{56}}\mr  + 
 \log \ml\frac{r_0(\tu)}{r_0(0)} \mr + \gamma_0(\tu)\mr\mr \\
 \leq & \; \frac{\pi C_{54}e^{-2C_{55}}(1-\mu_0(0))}{729e r_0(\tu)} \int^{v(\tu)}_{v_0}|\psi|^2(0,v')\D v'.
\end{aligned}
\end{equation}
Equivalently, this is our last requirement for $h^2(v(\tu))$:
\begin{equation}\label{SIT141}
\begin{aligned}
h^2(v(\tu)) \geq  C_{61}\ml r_0(\tu) + e^{-\gamma_0(\tu)}\ml \log \ml \frac{1}{v(\tu) - v_0}\mr + \log \ml \frac{2r_0(0)}{eC_{56}}\mr  + 
 \log \ml\frac{r_0(\tu)}{r_0(0)} \mr + \gamma_0(\tu)\mr \mr
 \end{aligned}
\end{equation}
with
\begin{equation}\label{SIT142}
C_{61}(u_2) := \frac{729 e^2 C_{36} C_{56}}{2\pi C_{54}e^{-2C_{55}}(1-\mu_0(0))}.
\end{equation}

\vspace{5mm}

It now remains to check that each of the assumptions made above can be verified if we take $\tu$ sufficiently close to $v_0$. Collecting the above assumptions and using the fact that $\delta(u,v(\tu)) \leq \frac{1}{2}$ for each $u \in [u_2,\tu]$, here we list the requirements
\begin{enumerate}[(1)]
\item \eqref{FIT35a}: $$\delta(u,v(\tu)) \leq \frac{C_{33}^2e^2} {4}e^{-4\gamma_0(u)} \quad \mbox{for all} \,\, u \in [u_2,\tu],$$
\item \eqref{FIT35c}: $$\delta(u,v(\tu)) \leq \min\left\{ \frac{1}{2},\frac{2(1 - \mu_0(0))e^{-\gamma_0(u)}}{3}, \frac{e^{-\frac{9}{2}\gamma_0(u)}}{2 C_{34}} \right\} \quad \mbox{for all} \,\, u \in [u_2,\tu],$$
\item \eqref{SIT35c}:  
\begin{equation*}
 \delta(u,v(\tu)) \leq \frac{e^{-\ml \frac{15}{2} + 3 \xi \mr\gamma_0(u)}}{64C_{53}^2} \quad \mbox{for all} \,\, u \in [u_2,\tu],
\end{equation*}
 \item \eqref{SIT117}:
$$\delta(\tu,v(\tu)) \leq \min \left\{ \frac{e^{-\ml \frac{19}{2} + 3 \xi\mr\gamma_0(\tu)}}{256 C_{45}^2}, \frac{e^{-\ml \frac{23}{2} + 3 \xi\mr\gamma_0(\tu)}}{256 C_{49}^2C_{56}^2} \right\}, $$
\item \eqref{SIT138}: 
$$v(\tu) - v_0 \leq \frac{2r_0(\tu)}{e^2 C_{56}}e^{\gamma_0(\tu)} = \frac{2r_0(0)}{e^2 C_{56}}e^{\gamma_0(\tu) - \log \ml \frac{r_0(0)}{r_0(\tu)}\mr},$$
\item \eqref{SIT128}: $$h^2(v(\tu)) \geq 8C_{60}^2(e^{-2\gamma_0(\tu)} + (v(\tu) - v_0)^2),$$
\item \eqref{SIT131}: $$h^2(v(\tu)) \geq 144 e^{2C_{55}} V^2(\tu),$$
\item \eqref{SIT141}:
$$h^2(v(\tu)) \geq  C_{61}\ml r_0(\tu) + e^{-\gamma_0(\tu)}\ml \log \ml \frac{1}{v(\tu) - v_0}\mr + \log \ml \frac{2r_0(0)}{eC_{56}}\mr  + 
 \log \ml\frac{r_0(\tu)}{r_0(0)} \mr + \gamma_0(\tu)\mr \mr.$$
\end{enumerate}
Recall that each of the constants that appear above is only dependent on $u_2$, $\te$ and parameter $\xi \in \left(0,\oh\right)$.

Before we proceed with verifying these above assumptions, we first state a useful lemma.
\begin{lemma}\label{SITLemma4}
There exist positive constants $C_{66}$ and $C_{67}$ given in \eqref{SIT149} depending only on $u_2$ and $\te$, such that
\begin{equation}\label{SIT156}
C_{66}(v(\tu) - v_0) \leq \delta(u_2,v(\tu)) \leq C_{67}(v(\tu) - v_0).
\end{equation}
\end{lemma}

\begin{proof}
For each $(u,v) \in [0,u_2] \times [v_0,v(\tu)]$, we first compute 
\begin{equation}\label{SIT143A}
\begin{aligned}
\du \log \delta &= \frac{\frac{\du r}{r_0} - \frac{r}{r_0^2}(\du r_0)}{\frac{r - r_0}{r_0}} = r \frac{\frac{\du r}{r} - \frac{\du r_0}{r_0}}{r - r_0} \\
\end{aligned}
\end{equation}
and
\begin{equation}\label{SIT143}
\begin{aligned}
\du \log \delta(u,v) = r(u,v) \frac{\int^v_{v_0} \dv \ml \frac{\du r}{r}\mr(u,v') \D v'}{\int^v_{v_0} \dv r(u,v') \D v'} =  r(u,v) \frac{\int^v_{v_0}  \frac{\du r \dv r}{r^2} \ml \frac{2\mu - 1 - Q^2/r^2}{1-\mu}\mr  (u,v') \D v'}{\int^v_{v_0} \dv r(u,v') \D v'}. \\
\end{aligned}
\end{equation}
Here we have used \eqref{FIT180} for calculating $\dv \ml \frac{\du r}{r}\mr$. Next, note that in $[0,u_2]\times [v_0,v_0 + \te] \subset \mathcal{R}$, functions $\dv r$ and $\frac{\du r \dv r}{r^2} \ml \frac{2\mu - 1 - Q^2/r^2}{1-\mu}\mr$ are continuous. Since $[0,u_2]\times [v_0,v_0 + \te]$ is compact and $\mu < 1$, then for all $(u,v) \in [0,u_2]\times [v_0,v_0 + \te]$, there exist constants $C_{62}, C_{63}$, $C_{64}$, $C_{65}$, such that 
\begin{equation}\label{SIT144}
C_{62}(u_2) \leq \frac{\du r \dv r}{r^2} \ml \frac{2\mu - 1 - Q^2/r^2}{1-\mu}\mr(u,v) \leq C_{63}(u_2)
\end{equation}
and 
\begin{equation}\label{SIT145}
0 < C_{64}(u_2) \leq \dv r(u,v) \leq C_{65}(u_2).
\end{equation} 
Consequently, from \eqref{SIT143}, we have
$$C_{62}r(u,v) \frac{v - v_0}{\int^v_{v_0}\dv r(u,v')\D v'} \leq \du \log \delta(u,v) \leq C_{63}r(u,v) \frac{v - v_0}{\int^v_{v_0}\dv r(u,v')\D v'},$$
and
$$-\frac{|C_{62}|}{C_{64}}r(0, v_0+\te) \leq \du \log \delta(u,v)  \leq \frac{|C_{63}|}{C_{64}}r(0, v_0+\te).$$
These imply that
\begin{equation}\label{SIT147}
\begin{aligned}
-\frac{|C_{62}|}{C_{64}}r(0, v_0+\te)u_2 &\leq \log \ml \frac{\delta(u_2,v)}{\delta(0,v)}\mr   \leq \frac{|C_{63}|}{C_{64}}r(0, v_0+\te)u_2.
\end{aligned}
\end{equation}
By setting $v = v(\tu)$, we hence obtain
\begin{equation}\label{SIT148}
\begin{aligned}
C_{66}(v - v_0) &\leq \delta(u_2,v) \leq C_{67}(v - v_0)
\end{aligned}
\end{equation}
with
\begin{equation}\label{SIT149}
\begin{aligned}
C_{66} := e^{-\frac{|C_{62}|}{C_{64}}r(0, v_0+\te)u_2}v_0 > 0, \quad \text{ and } \quad 
C_{67} := e^{\frac{|C_{63}|}{C_{64}}r(0, v_0+\te)u_2}v_0 > 0.
\end{aligned}
\end{equation}
\end{proof}
We now start to verify the assumptions (1) to (8) as listed above. For a fixed $\tu$, since $v(\tu) - v_0 \rightarrow 0$ as $v(\tu) \rightarrow v_0$, we can pick $v(\tu)$ sufficiently close to $v_0$, such that it satisfies 
\begin{equation}\label{SIT150}
v(\tu) - v_0 = e^{-A\gamma_0(\tu) - B \log \ml \frac{r_0(0)}{r_0(\tu)}\mr}.
\end{equation}
We claim that by choosing $A$ and $B$ to be suitable constants, we can verify assumptions (1) to (5). First, note that for assumption (5) to hold, with $\tu$ sufficiently close to $v_0$, it suffices to pick $A > -1$ and $ B \geq 1$. To verify assumptions (1) to (4), we conduct as follows. For a given $\tu \in [u_2,v_0)$, we denote a set
\begin{equation}\label{SIT151}
S_2 := \left\{u \in [u_2,\tu]: \text{ Assumptions (1) to (4) hold for all } u' \in [u_2,u] \right\}.
\end{equation}
It is clear that the set $S_2$ is closed. To show that $S_2$ is non-empty, it suffices to verify that $u_2 \in S_2$. Observe by Lemma \ref{SITLemma4} and \eqref{SIT150} that 
\begin{equation}\label{SIT152}
\delta(u_2,v(\tu)) \leq C_{67}(v(\tu) - v_0) = C_{67} e^{-A\gamma_0(\tu) - B \log \ml \frac{r_0(0)}{r_0(\tu)}\mr}.
\end{equation}
Henceforth, assumptions (1) and (4) hold at $u = u_2$ if we have
\begin{equation}\label{SIT153}
\begin{aligned}
C_{67} e^{-A\gamma_0(\tu) - B\log \ml \frac{r_0(0)}{r_0(\tu)}\mr} &\leq  \frac{C_{33}^2e^2}{4}e^{-4\gamma_0(u_2)}, \\
C_{67} e^{-A\gamma_0(\tu) - B\log \ml \frac{r_0(0)}{r_0(\tu)}\mr} &\leq  \min\left\{ \frac{1}{2},\frac{2(1 - \mu_0(0))e^{-\gamma_0(u_2)}}{3}, \frac{e^{-\frac{9}{2}\gamma_0(u_2)}}{2 C_{34}} \right\}, \\
C_{67} e^{-A\gamma_0(\tu) - B\log \ml \frac{r_0(0)}{r_0(\tu)}\mr} &\leq  \frac{e^{-\ml \frac{15}{2} + 3 \xi \mr\gamma_0(u_2)}}{64C_{53}^2}, \\
C_{67} e^{-A\gamma_0(\tu) - B\log \ml \frac{r_0(0)}{r_0(\tu)}\mr} &\leq  \min \left\{ \frac{e^{-\ml \frac{19}{2} + 3 \xi\mr\gamma_0(u_2)}}{256 C_{45}^2}, \frac{e^{-\ml \frac{23}{2} + 3 \xi\mr\gamma_0(u_2)}}{256 C_{49}^2C_{56}^2} \right\}. \\
\end{aligned}
\end{equation}
These are true by picking $\tu$ sufficiently close to $v_0$, since the left sides of the inequalities depend on $\tu$ and they go to $0$ as $\tu \rightarrow v_0^-$, while the right sides of the inequalities do not depend on $\tu$. Hence, we prove $u_2 \in S_2$. We proceed to prove that $S_2$ is open. Pick an arbitrary $u_5 \in S_2$. We then have assumptions (1) to (4) hold for all $u \in [u_2,u_5]$. Next, choose an arbitrary $u_6 \in [u_2,u_5]$. Analogously to \eqref{SIT135}, we get
\begin{equation}\label{SIT157}
\begin{aligned}
\delta(u_6,v(\tu)) &= \frac{1}{r_0(u_6)}\int^{v(\tu)}_{v_0} (\dv r)(u_6,v')\D v' = \frac{1}{r_0(u_6)} \int^{v(\tu)}_{v_0} (\dv r)(0,v')e^{-\tgam(u_6,v') - \gamma(u_2,v')} \D v' \\
&\quad\quad\leq \frac{e C_{56} e^{-\gamma_0(u_6)}}{2r_0(u_6)}(v(\tu) - v_0) \leq \frac{e C_{56} e^{-\gamma_0(u_2)}}{2r_0(\tu)}(v(\tu) - v_0).
\end{aligned}
\end{equation}
Consequently, by \eqref{SIT150}, this implies
\begin{equation}\label{SIT158}
\begin{aligned}
\delta(u_6,v(\tu)) \leq \frac{e C_{56} e^{-\gamma_0(u_2)}}{2r_0(\tu)}(v(\tu) - v_0) = \frac{e C_{56} e^{-\gamma_0(u_2)}}{2r_0(0)}e^{-A \gamma_0(\tu) - ( B -1)\log \ml \frac{r_0(0)}{r_0(\tu)}\mr}.
\end{aligned}
\end{equation}
Since $u_6$ is arbitrary and the above estimate is uniform in $u_6$, we hence show that for all $u \in [u_2,u_5]$, we have
\begin{equation}\label{SIT159}
\begin{aligned}
\delta(u,v(\tu)) &\leq \frac{e C_{56} e^{-\gamma_0(u_2)}}{2r_0(0)}e^{-A \gamma_0(\tu) - ( B -1)\log \ml \frac{r_0(0)}{r_0(\tu)}\mr}.
\end{aligned}
\end{equation}
Henceforth, if we pick $A > \frac{23}{2} + 3 \xi$ and $B \geq 1$ corresponding to the most singular term among assumptions (1) to (4), then for $\tu$ sufficiently close to $v_0$, we now obtain better estimates for assumptions (1) to (4) for all $u \in [u_2,u_5]$. In other words, we show that $S_2$ is open since $u_5$ is an arbitrary element in $[u_2, \tu]$. This thus implies that $S_2 = [u_2,\tu]$. Therefore, assumptions (1) to (4) indeed hold for all $u \in [u_2,\tu]$. With $\xi \in \left(0,\frac{1}{2}\right)$, choosing $A=14$ and $B=1$ satisfies all the requirements.

It now remains to verify assumptions (6) to (8). Plugging in our choice 
$$v(\tu) - v_0=e^{-14\gamma_0(\tu) - \log \ml \frac{r_0(0)}{r_0(\tu)}\mr},$$ 
we obtain
\begin{equation}\label{SIT160}
8C_{60}\ml e^{-2\gamma_0(\tu)} + e^{-28\gamma_0(\tu) - 2\log \ml \frac{r_0(0)}{r_0(\tu)}\mr} \mr \leq 16C_{60}e^{-2\gamma_0(\tu)}
\end{equation}
and 
\begin{equation}\label{SIT161}
\begin{aligned}
&\; C_{61}\ml r_0(\tu) + e^{-\gamma_0(\tu)}\ml \log \ml \frac{1}{v(\tu) - v_0}\mr + \log \ml \frac{2r_0(0)}{eC_{56}}\mr  + 
 \log \ml\frac{r_0(\tu)}{r_0(0)} \mr + \gamma_0(\tu)\mr \mr \\
 = &\; C_{61}\ml r_0(\tu) + e^{-\gamma_0(\tu)}\ml\log \ml \frac{2r_0(0)}{eC_{56}}\mr  + 15\gamma_0(\tu) \mr \mr \\
\leq &\; C_{61}\ml r_0(0)e^{-\log \ml \frac{r_0(0)}{r_0(\tu)}\mr} + e^{-\frac{1}{2}\gamma_0(\tu)}\ml\log \ml \frac{2r_0(0)}{eC_{56}}\mr  + 15\mr \mr \\
\leq &\; C_{61}\ml r_0(0) + \log \ml \frac{2r_0(0)}{eC_{56}}\mr  + 15\mr\max\left\{e^{-\frac{1}{2}\gamma_0(\tu)}, e^{- \log \ml \frac{r_0(0)}{r_0(\tu)}\mr}\right\},
 \end{aligned}
\end{equation}
where in \eqref{SIT161} we use the fact that $\gamma_0(\tu)e^{-\frac{1}{2}\gamma_0(\tu)} \leq 1$ for $\gamma_0(\tu) \geq 0$. 

Hence, to fulfill assumptions (6) to (8), it suffices to impose
\begin{equation}\label{SIT162}
h^2(v(\tu)) \geq 16C_{60}e^{-2\gamma_0(\tu)}, \end{equation}
\begin{equation}\label{SIT163}
h^2(v(\tu)) \geq 144 e^{2C_{55}} V^2(\tu),
\end{equation}
and
\begin{equation}\label{SIT164}
h^2(v(\tu)) \geq C_{61}\ml r_0(0) + \log \ml \frac{2r_0(0)}{eC_{56}}\mr  + 15\mr\cdot\max\left\{e^{-\frac{1}{2}\gamma_0(\tu)}, e^{- \log \ml \frac{r_0(0)}{r_0(\tu)}\mr}\right\}.
\end{equation}
Observe that, for $\tu$ sufficiently close to $v_0$, the three competing vanishing terms are $V^2(\tu), e^{-\frac{1}{2}\gamma_0(\tu)}$ and $e^{-\log \ml \frac{r_0(0)}{r_0(\tu)}\mr}$. Henceforth, to verify assumptions (6) to (8), for a $\tu$ sufficiently close to $v_0$, it suffices to require that
\begin{equation}\label{SIT165}
\limsup_{\tu \rightarrow v_0} \frac{ h^2(v(\tu))}{ \ml V(\tu) + e^{-\frac{1}{4}\gamma_0(\tu)} + e^{-\frac{1}{2}\log\ml\frac{r_0(0)}{r_0(\tu)}\mr}\mr^2 } = \infty.
\end{equation}
This is exactly the condition \eqref{SIT168} stated in Theorem \ref{SITTheorem} and we thus conclude the proof of Theorem \ref{SITTheorem}.
\qed

\section{Exceptional Set and Weak Cosmic Censorship}\label{ExceptionalSetSection}

In this section, we prove that the subset of BV initial data, which lead to the incompleteness of future null infinity, are non-generic. 

As mentioned in the introduction, we prescribe quantity $\alpha$ along $C_0^+$ and  we have
\begin{equation}\label{Except1}
\alpha = \phi + \frac{\theta}{\dv r} \quad \mbox{ with } \quad \theta=r\partial_v\phi.
\end{equation}
Note that if both $\phi$ and $\theta$ are of bounded variation along $C_0^+$, so is $\alpha$. Aligned with \cite{christ2}, we define the total variation of $\phi$ along $C_0^+$ as
\begin{equation}\label{Except2}
TV_{\{0\} \times (0,\infty)}[\phi] = 2\int_0^{\infty}|\theta| (0,v') \frac{\D v'}{v'}.
\end{equation}
For notational simplicity, we denote 
\begin{equation}\label{vartheta}
    \vartheta(v) := \theta(0,v).
\end{equation}
In this paper, we consider initial data set of $\vartheta$ with $\vartheta \in BV([0,\infty),\D v) \cap L^1\ml (0,\infty), \frac{\D v}{v}\mr $.  The main conclusion of this section is

\begin{theorem}\label{ExceptionalSet}
(Exceptional Set.) For the spherically symmetric Einstein-Maxwell-charged scalar field system, within its initial data set $BV([0,\infty),\D v) \cap L^1\ml (0,\infty), \frac{\D v}{v}\mr$ prescribed along $C_0^+$, we have an exceptional set $\mathcal{E}$ and it satisfies the following properties:
\begin{itemize}
\item[-] for each $\vartheta\notin\mathcal{E}$, the evolutional spacetime arising from $\vartheta$ possesses a complete future null infinity; 

\item[-] for every $\vartheta\in\mathcal{E}$, there is a complex $2$-dimensional linear space $\Pi_{\vartheta}$ belonging to the initial data set and each element in $\Pi_{\vartheta} \setminus \{\vartheta \}$ does not belong to $\mathcal{E}$, i.e., 
\begin{equation}\label{Except3}
    \ml \Pi_{\vartheta} \setminus \{\vartheta \}\mr \bigcap \mathcal{E} = \emptyset.
\end{equation}
Moreover, if $\vartheta, \vartheta' \in \mathcal{E}$ and $\vartheta \neq \vartheta'$, then we have
\begin{equation}\label{Except4}
    \Pi_{\vartheta} \bigcap \Pi_{\vartheta'} = \emptyset.
\end{equation}
\end{itemize}
\end{theorem}

\begin{remark}
Theorem \ref{ExceptionalSet} implies that the exceptional set $\mathcal{E}$ has complex codimension $2$ instability in the $BV([0,\infty),\D v) \cap L^1\ml (0,\infty), \frac{\D v}{v}\mr$ initial data set. 
\end{remark}

\begin{proof}
We define the set $\mathcal{E}$ to be the union of initial datum along $C_0^+$, which does not satisfy the initial-data conditions in Theorem \ref{FITTheorem} and in Theorem \ref{SITTheorem}. Suppose that $\vartheta \in \mathcal{E}$. Then by Theorem \ref{FITTheorem} and Theorem \ref{SITTheorem}, we have
\begin{equation}\label{Except5}
\mlm I(u) - \vartheta(v_0) \mrm \rightarrow 0 \quad \mbox{ as } \quad u\rightarrow v_0^-
\end{equation}
and
\begin{equation}\label{Except6}
\limsup_{u \rightarrow v_0^-}\frac{h(v(u))}{g(v(u))} < + \infty
\end{equation}
with $v(u), h(v), g(v)$ defined in \eqref{SIT166}, \eqref{SIT48}, \eqref{SIT167}, respectively. We then choose two real-valued non-negative functions $f_1, f_2 \in L^1\ml (0,\infty), \frac{\D v}{v} \mr$. For $f_1$, we further require it to vanish on $[0,v_0)$, to be absolutely continuous on $[v_0, \infty)$, and to satisfy
\begin{equation}\label{Except7}
\lim_{v \rightarrow v_0^+}f_1(v) = 1.
\end{equation}
For $f_2$, we choose it to be absolutely continuous on $[0,\infty)$, vanishing on $[0,v_0]$, and with property
\begin{equation}\label{Except8}
\limsup_{v \rightarrow v_0^+}\frac{1}{g(v)^2}\ml \frac{1}{v-v_0}\int^v_{v_0}\ml \frac{v'}{v_0}\mr^2 f_2^2(v') \D v' \mr = \infty. \footnote{In particular, one can set $f_2(v)=0$ for $v\in[0,v_0)$, $f_2(v) = g(v)^{\oh}$ for $v\in[v_0, v_0+1]$, and $f_2$ to be absolutely continuous for $[v_0+1,\infty)$. To see this choice of $f_2$ satisfies the above requirements, it suffices to check that $g$ is differentiable on $(v_0,v_0+1]$, to show $\lim_{v\rightarrow v_0^+}f_2(v)=\lim_{v\rightarrow v_0^+}g(v)^{\frac12}=0$, and to prove property \eqref{Except8} holds. Differentiability of $g$ on $(v_0,v_0+1]$ follows from the fact that $v(u)$ is a differentiable function with respect to $u$ with $v'(u) > 0$. The condition $\lim_{v\rightarrow v_0^+}f_2(v)=0$ is easily verified since $\lim_{v\rightarrow v_0^+}g(v) = 0$. The property \eqref{Except8} comes from $\lim_{v\rightarrow v_0^+}g(v) = 0$ and the use of Lebesgue differentiation theorem.}
\end{equation}

With $\lambda_1, \lambda_2 \in \mathbb{C}$, we then consider the following initial data
\begin{equation}\label{Except9}
\tilde{\vartheta}_{(\lambda_1,\lambda_2)} = \vartheta + \lambda_1 f_1 + \lambda_2 f_2.
\end{equation}
Since $\tilde{\vartheta}_{(\lambda_1,\lambda_2)}$ and $\vartheta$ coincides for $v \in [0,v_0)$, the corresponding solutions coincide in $\mathcal{D}(0,v_0) \setminus \{v = v_0\}$. This thus permits a common definition of $\gamma_0$ and $I(u)$ as described in \eqref{FIT3} and \eqref{FIT14} for both $\tilde{\vartheta}_{(\lambda_1,\lambda_2)}$ and $\vartheta$. Applying \eqref{Except5} and \eqref{Except7}, taking limits as $v \rightarrow v_0^+$ while respecting \eqref{SIT166}, these yield
\begin{equation}\label{Except10}
\tilde{\vartheta}_{(\lambda_1,\lambda_2)}(v_0) = \lim_{u \rightarrow v_0^-} I(u) + \lambda_1.
\end{equation}
If $\lambda_1 \neq 0$, then we have $\tilde{\vartheta}_{(\lambda_1,\lambda_2)}(v_0) \neq \lim_{u \rightarrow 0^-} I(u)$. By Theorem \ref{FITTheorem}, we get $\tilde{\vartheta}_{(\lambda_1,\lambda_2)} \notin \mathcal{E}$.

\noindent If $\lambda_1 = 0$ and $\lambda_2 \neq 0$, we encounter the scenario
\begin{equation}\label{Except12}
\tilde{\vartheta}_{(\lambda_1,\lambda_2)}(v_0) = \lim_{u \rightarrow v_0^-} I(u) = \vartheta(v_0)
\end{equation}
and for $v \in (v_0,\infty)$ we have
\begin{equation}\label{Except13}
\tilde{\vartheta}_{(\lambda_1,\lambda_2)}(v) = \vartheta(v) + \lambda_2 f_2(v).
\end{equation}
Recalling \eqref{FIT31}, we then get

\begin{equation}\label{Except11}
\begin{aligned}
\tilde{h}_{(\lambda_1,\lambda_2)}(v)^2 =&\frac{2e^{-2\gamma_0(0)}}{v-v_0}\int^v_{v_0} \mlm \frac{v'}{v_0}\tilde{\vartheta}_{(\lambda_1,\lambda_2)}(v') - \tilde{\vartheta}_{(\lambda_1,\lambda_2)}(v_0)\mrm^2 \D v'\\
=&\frac{2e^{-2\gamma_0(0)}}{v-v_0}\int^v_{v_0} \mlm \lambda_2 \ml \frac{v'}{v_0}\mr f_2(v') - \ml \frac{v'}{v_0}\vartheta(v') - \vartheta(v_0)\mr\mrm^2 \D v'.
\end{aligned}
\end{equation}
By \eqref{SIT5}, we hence derive
\begin{equation}\label{Except16}
\begin{aligned}
\frac{\tilde{h}_{(\lambda_1,\lambda_2)}(v)^2}{g(v)^2} &= \frac{1}{g(v)^2}\frac{2e^{-2\gamma_0(0)}}{v-v_0}\int^v_{v_0} \mlm \frac{v'}{v_0}\tilde{\vartheta}_{(\lambda_1,\lambda_2)}(v') - \tilde{\vartheta}_{(\lambda_1,\lambda_2)}(v_0)\mrm^2 \D v' \\
&\geq \lambda_2 e^{-2\gamma_0(0)}\frac{1}{g(v)^2}\frac{1}{v-v_0}\int^v_{v_0} \ml \frac{v'}{v_0}\mr^2 f_2^2(v') \D v' - 2e^{-2\gamma_0(0)}\frac{1}{g(v)^2}\frac{1}{v-v_0}\int^v_{v_0} \mlm \frac{v'}{v_0}\vartheta(v') - \vartheta(v_0)\mrm^2 \D v' \\
&= \lambda_2 e^{-2\gamma_0(0)}\frac{1}{g(v)^2}\frac{1}{v-v_0}\int^v_{v_0} \ml \frac{v'}{v_0}\mr^2 f_2^2(v') \D v' - 2e^{-2\gamma_0(0)}\frac{h(v)^2}{g(v)^2}. \\
\end{aligned}
\end{equation}
Applying \eqref{Except6} and \eqref{Except8}, inequality \eqref{Except16} further implies that
\begin{equation}\label{Except17}
\limsup_{v \rightarrow v_0^+}\frac{\tilde{h}_{(\lambda_1,\lambda_2)}(v)^2}{g(v)^2} = \infty.
\end{equation}
Theorem \ref{SITTheorem} then applies to $\tilde{\vartheta}_{(\lambda_1,\lambda_2)}$ and hence we demonstrate $\tilde{\vartheta}_{(\lambda_1,\lambda_2)} \notin \mathcal{E}$. 

\vspace{5mm}

We then proceed to prove that the complex $2$-dimensional linear subspace $\Pi_{\vartheta}$ defined by
\begin{equation}\label{Except18}
\Pi_{\vartheta}:=\{ \vartheta_{(\lambda_1,\lambda_2)}: (\lambda_1,\lambda_2) \in \mathbb{C}^2 \}
\end{equation}
intersects $\mathcal{E}$ only at $(\lambda_1,\lambda_2) = (0,0)$, which corresponds to $\vartheta$ itself. Suppose that $\vartheta, \vartheta' \in \mathcal{E}$ and $\Pi_{\vartheta}, \Pi_{\vartheta'}$ intersect at
\begin{equation}\label{Except19}
\tilde{\vartheta}_{(\lambda_1,\lambda_2)} = \tilde{\vartheta}'_{(\lambda_1,\lambda_2)},
\end{equation}
then by the construction of $f_1, f_2$, we have
$\vartheta = \vartheta'$ on $[0,v_0)$ and hence we obtain the same functions $\gamma_0(u), I(u), g(v)$. Employing \eqref{Except5}, we also get
\begin{equation}\label{Except20}
\vartheta(v_0) = \vartheta'(v_0) = \lim_{u\rightarrow 0^-}I(u).
\end{equation}
By \eqref{Except9} and \eqref{Except19} with $v=v_0$, we then obtain
\begin{equation}\label{Except21}
\lambda_1 = \lambda_1'.
\end{equation}
Equation \eqref{Except19} also implies
\begin{equation}\label{Except22}
\vartheta - \vartheta' = (\lambda_2' - \lambda_2)f_2.
\end{equation}
For each $v \in [v_0,\infty)$, this gives
\begin{equation}\label{Except23}
\mlm \frac{v}{v_0} \vartheta(v) - \vartheta(0) - \ml \frac{v}{v_0}\vartheta'(v) - \vartheta'(0)\mr \mrm^2 = \ml \frac{v}{v_0} \mr^2 |\lambda_2' - \lambda_2|^2f_2^2(v).
\end{equation}
Recall the first line of \eqref{Except11}, for $\vartheta'\in\mathcal{E}$, we have
\begin{equation}\label{Except24}
h'(v) := \frac{2e^{-2\gamma_0(0)}}{v-v_0}\int^v_{v_0} \mlm \frac{v'}{v_0}\vartheta'_{(\lambda_1,\lambda_2)}(v') - \vartheta'_{(\lambda_1,\lambda_2)}(v_0)\mrm^2 \D v'.
\end{equation}
Together with \eqref{Except23}, it then implies that
\begin{equation}\label{Except25}
\begin{aligned}
|\lambda_2' - \lambda_2|^2 \frac{1}{g(v)^2}\frac{1}{v-v_0}\int^v_{v_0} \ml \frac{v'}{v_0} \mr^2 f_2^2(v') \D v' &\leq \frac{1}{g(v)^2}\frac{1}{v-v_0}\int^v_{v_0} \mlm \frac{v}{v_0} \vartheta(v) - \vartheta(0) - \ml \frac{v}{v_0}\vartheta'(v) - \vartheta'(0)\mr \mrm^2(v') \D v' \\
&\leq e^{2\gamma_0(0)}\ml \frac{h(v)^2}{g(v)^2} + \frac{h'(v)^2}{g(v)^2}\mr.
\end{aligned}
\end{equation}
Note that \eqref{Except6} implies
\begin{equation}\label{Except26}
\limsup_{v \rightarrow v_0^+} \frac{h(v)}{g(v)} < + \infty \quad \text{ and } \quad \limsup_{v \rightarrow v_0^+} \frac{h'(v)}{g(v)} < + \infty.
\end{equation}
Back to \eqref{Except25}, utilizing \eqref{Except8}, \eqref{Except26} and letting $v\rightarrow v_0^+$, we then conclude that $\lambda_2 = \lambda_2'$. With \eqref{Except22}, this implies that $\vartheta = \vartheta'$. Hence, we prove that $\Pi_{\vartheta} \bigcap \Pi_{\vartheta'} = \emptyset$ unless $\vartheta = \vartheta'$. 

\vspace{5mm}

It remains to prove that, once the initial-data conditions for Theorem \ref{FITTheorem} (first instability theorem) or Theorem \ref{SITTheorem} (second instability theorem) are satisfied and a trapped surface is formed, then the corresponding dynamical spacetime is blessed with a complete future null infinity. In \cite{dafermos} Dafermos identified six conditions: $\mathbf{A}', \mathbf{B}', \mathbf{\Gamma}', \mathbf{\Delta}', \mathbf{E}', \mathbf{\Sigma T}'$ and proved that, for spherically symmetric Einstein field equations, once these six conditions are fulfilled, then the existence of a trapped surface would imply the completeness of future null infinity. For spherically symmetric Einstein-Maxwell-charged scalar field system, condition $\mathbf{\Sigma T}'$ (a local extension principle at $r(u,v)>0$) was proved in Section of \cite{kommemi}. With asymptotically flat initial data, in our setting, conditions $\mathbf{A}', \mathbf{B}', \mathbf{\Delta}', \mathbf{E}'$ are also automatically satisfied. It remains to meet the condition $\mathbf{\Gamma}'$ (positive energy condition), i.e., $T_{uu}\geq 0, T_{vv}\geq 0, T_{uv}\geq 0$. Meanwhile, for our charged system, a direct calculation gives
\begin{equation}\label{WCSC1}
T_{uu} = |D_u\phi|^2\geq 0, \quad T_{vv} = |\dv \phi|^2\geq 0, \quad T_{uv} = \frac{Q^2}{16\pi r^2\Omega^2}\geq 0,
\end{equation}
which verifies condition $\mathbf{B}'$. Therefore, for the generic scenarios (when Theorem \ref{FITTheorem} or Theorem \ref{SITTheorem} is activated), we conclude the completeness of future null infinity. Hence the weak cosmic censorship for the spherically symmetric Einstein-Maxwell-charged scalar field system is proved.

\end{proof}

\appendix

\section{Alternative Proof of the Almost Scale-Critical Trapped-Surface-Formation Criteria for the Einstein-(Real) Scalar Field System}\label{Appendix A}

For the spherically symmetric Einstein-(real) scalar field system, by setting $Q \equiv 0$ in the system \eqref{Setup12} to \eqref{Setup20}, we get
\begin{equation}\label{Scalar1}
r \dv \du r + \dv r \du r = -\frac{\Omega^2}{4},
\end{equation}
\begin{equation}\label{Scalar3}
\du(\Omega^{-2}\du r) = - 4\pi r\Omega^{-2}|\du\phi|^2,
\end{equation}
\begin{equation}\label{Scalar4}
\dv(\Omega^{-2}\dv r) = - 4\pi r\Omega^{-2}|\dv \phi|^2,
\end{equation}
\begin{equation}\label{Scalar5}
r\du \dv \phi + \du r \dv \phi + \dv r \du \phi = 0.
\end{equation}
Define
$$m(u,v):=\frac{r}{2}\cdot(1+4\Omega^{-2}\partial_u r \partial_v r)(u,v) \quad \mbox{and} \quad \mu(u,v):=\frac{2m(u,v)}{r(u,v)}.$$
Analogously to \eqref{Setup20}, \eqref{Setup21}, \eqref{Setup25}, we deduce
\begin{equation}\label{Scalar6}
\dv \du r = \frac{1}{r}\ml \frac{\mu}{1-\mu}\mr \du r \dv r,
\end{equation}
\begin{equation}\label{Scalar7}
\du \ml \frac{1-\mu}{\dv r}\mr = \frac{4 \pi r (\du \phi)^2(1-\mu)}{(-\du r)(\dv r)},
\end{equation}
\begin{equation}\label{Scalar9}
\dv m = \frac{2\pi r^2(1-\mu)(\dv \phi)^2}{\dv r}.
\end{equation}
Furthermore, for any coordinate patch $[u_0,v_0) \times [v_1,v_2]$ with $u_0 < v_0$ and $v_1 < v_2$, we can define 
$$\delta(u) := \frac{r_2(u) - r_1(u)}{r_1(u)} \quad \mbox{and} \quad \eta(u) := \frac{2(m_2(u) - m_1(u))}{r_2(u)}$$
and let $\delta_0 := \delta(u_0)$ and $\eta_0 := \eta(u_0).$
Adopting the proof strategy in Theorem \ref{Trapped}, we give a shorter new proof of the below theorem, which was first achieved by Christodoulou in \cite{christ1}. 

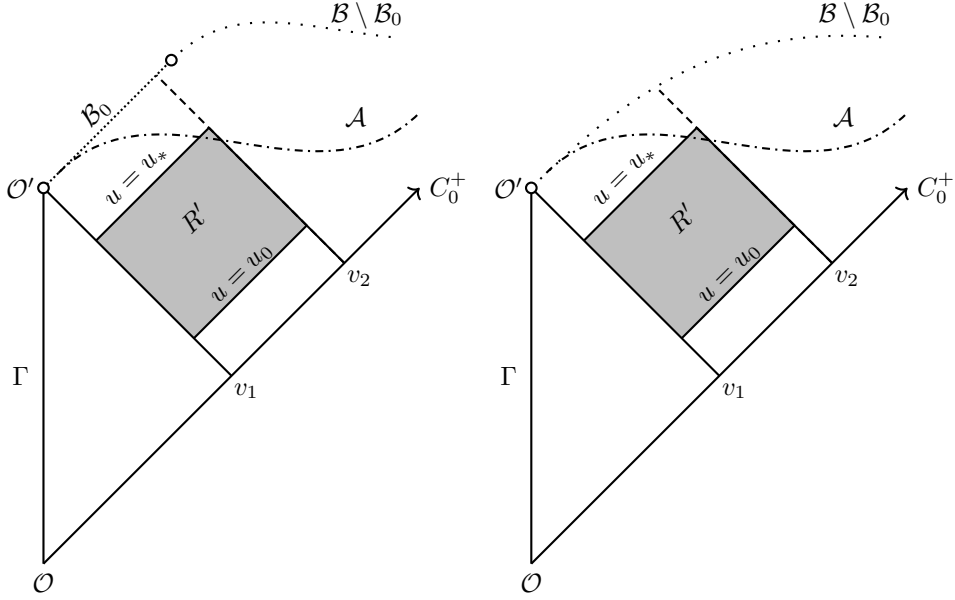
\begin{figure}[htbp]
\begin{tikzpicture}[
    dot/.style = {draw, fill = white, circle, inner sep = 0pt, minimum size = 4pt}]
\begin{scope}[thick]
\fill[gray!50] (0.7,4.3) to (2.2,5.8) to (3.5,4.5) to (2,3) to (0.7,4.3);
\draw[thick] (0,0) node[anchor=north]{$\mathcal{O}$} --  (0,5);
\node[] at (-0.3,2.5) {$\Gamma$};
\node[] at (-0.3,5.0) {$\mathcal{O}'$};
\draw[->] (0,0) -- (5,5) node[anchor = west]{$C_0^+$};
\draw[thick, densely dotted] (0,5) to (1.7,6.7);
\draw[thick, loosely dotted] (1.7,6.7) to [out=45,in=175] (4.7,7.0);
\node[rotate=45] at (0.7,6) {$\mathcal{B}_0$};
\node[rotate=0] at (4.3,7.3) {$\mathcal{B}\setminus\mathcal{B}_0$};
\draw[fill=white] (1.7,6.7) circle (2pt);
\draw[thick] (0,5) --  (2.5,2.5);
\node[] at (2.7,2.3){$v_1$};
\draw[dashdotted] (0.2,5.2) to[out = 45, in = -135] (5,6);
\node[] at (4.15,5.95){$\mathcal{A}$};
\node[rotate=45] at (2,4.6){$R'$};
\node[] at (4.2,3.8){$v_2$};
\draw[thick] (4,4) to (2.2,5.8);
\draw[thick] (2,3) to (3.5,4.5);
\node[rotate = 45] at (2.65,3.85) {$u = u_0$};
\draw[thick] (0.7,4.3) to (2.2,5.8);
\node[rotate=45] at (1.25,5.15) {$u = u_*$};
\draw[densely dashed] (4,4) to (1.5,6.5);
\draw[fill=white] (0,5) circle (2pt);
\end{scope}
\end{tikzpicture}
\begin{tikzpicture}[
    dot/.style = {draw, fill = white, circle, inner sep = 0pt, minimum size = 4pt}]
\begin{scope}[thick]
\fill[gray!50] (0.7,4.3) to (2.2,5.8) to (3.5,4.5) to (2,3) to (0.7,4.3);
\draw[thick] (0,0) node[anchor=north]{$\mathcal{O}$} --  (0,5);
\node[] at (-0.3,2.5) {$\Gamma$};
\node[] at (-0.3,5.0) {$\mathcal{O}'$};
\draw[->] (0,0) -- (5,5) node[anchor = west]{$C_0^+$};
\draw[thick, loosely dotted] (0,5) to [out=45,in=175] (4.7,7.0);
\node[rotate=0] at (4.3,7.3) {$\mathcal{B}\setminus\mathcal{B}_0$};
\draw[thick] (0,5) --  (2.5,2.5);
\node[] at (2.7,2.3){$v_1$};
\draw[dashdotted] (0.2,5.2) to[out = 45, in = -135] (5,6);
\node[] at (4.15,5.95){$\mathcal{A}$};
\node[rotate=45] at (2,4.6){$R'$};
\node[] at (4.2,3.8){$v_2$};
\draw[thick] (4,4) to (2.2,5.8);
\draw[thick] (2,3) to (3.5,4.5);
\node[rotate = 45] at (2.65,3.85) {$u = u_0$};
\draw[thick] (0.7,4.3) to (2.2,5.8);
\node[rotate=45] at (1.25,5.15) {$u = u_*$};
\draw[densely dashed] (4,4) to (1.7,6.3);
\draw[fill=white] (0,5) circle (2pt);
\end{scope}
\end{tikzpicture}
\caption{Diagram for the proof of Theorem \ref{ScalarTrapped}. The rectangle region $R'$ used in the proof below is indicated as above. The apparent horizon $\mathcal{A}$ is represented with a dotted curve. The scenario on the left corresponds to $|\mathcal{B}_0| \geq v_2 - v_1$, while the scenario on the right corresponds to $|\mathcal{B}_0| < v_2 - v_1$, and $\mathcal{B}_0$ is allowed to be an empty set.}
\label{TrappedFig2}
\end{figure}
{\color{black}
\begin{theorem}\label{ScalarTrapped}
Let $\mathcal{O}' = (v_1,v_1)$ be the future ending point of $\Gamma$ as portrayed in Figure \ref{TrappedFig2}. Let $v_2 > v_1$ and consider the future-directed incoming null curve along $v = v_2$. With $\delta_0 \in \ml 0, 1\mr$, along $u=u_0$ if we have
\begin{equation}\label{Scalar10}
\eta_0  > \frac{2\delta_0}{1+\delta_0} + \frac{\delta_0}{1+\delta_0}\log \ml \frac{1}{\delta_0}\mr,
\end{equation} then there exists $u_* \in (u_0,v_0)$ along $v = v_2$ in the future of $C_0^+$, such that a MOTS or a trapped surface is guaranteed to form along $[u_0,u_*] \times \{v_2\}$.
\end{theorem}
\begin{proof}
Let $u_* \in [u_0,v_1)$ be an arbitrary constant such that $(u_*,v_2) \in \mathcal{Q}$. Denote $R' := [u_0,u_*] \times [v_1,v_2]$ and define $x(u) := \frac{r_2(u)}{r_2(u_0)}$. We will show that condition \eqref{Scalar10} along $u=u_0$ guarantees the existence of a MOTS or a trapped surface in $R'$. Below, we present a mechanism in the regular (untrapped) region, that would lead to a MOTS or a trapped surface formation in $R'$. 

If there is already a MOTS or a trapped surface in $R'$, then we are done. Otherwise, assuming the absence of a MOTS or a trapped surface in $R'$, we proceed to compute
\begin{equation}\label{Scalar11}
\frac{\D \eta}{\D x} = \frac{\frac{\mathrm{d}\eta}{\mathrm{d}u}}{\frac{\mathrm{d}x}{\mathrm{d}u}} = \frac{\du \ml \frac{2(m_2 - m_1)}{r_2} \mr}{\frac{\du r_2}{r_2(u_0)}} = - \frac{\eta}{x} - \frac{2}{x(\du r_2)}\int^{v_2}_{v_1} -\du \dv m \; \D v'.
\end{equation}
By \eqref{Scalar9}, we have
\begin{equation}\label{Scalar12}
\begin{aligned}
- \du\dv m  =& -\du \ml \frac{2\pi r^2(1-\mu)(\dv \phi)^2}{\dv r} \mr = -2 \pi \ml \du\ml \frac{(1-\mu)}{\dv r}\mr r^2(\dv \phi)^2 + \ml \frac{2(1-\mu)}{\dv r}\mr r (\dv \phi) \du (r\dv \phi) \mr \\
&\quad\quad= 8\pi^2 r^3 (1-\mu) \ml  (\du r)(\dv r) \frac{(\du \phi)^2 (\dv \phi)^2}{(\du r)^2(\dv r)^2} - \frac{1}{2\pi r^2} (\du r)(\dv r) \frac{(\dv \phi)(\du \phi)}{(\dv r)(\du r)} \mr.
\end{aligned}
\end{equation}
With an elementary inequality $-x^2 + ax \leq \frac{a^2}{4}$, we then derive
\begin{equation}\label{Scalar13}
\begin{aligned}
\int^{v_2}_{v_1} -\du \dv m \; \D v' &= \int^{v_2}_{v_1} 8\pi^2 r^3 (1-\mu)(-\du r)(\dv r) \ml  -\frac{(\du \phi)^2 (\dv \phi)^2}{(\du r)^2(\dv r)^2} + \frac{1}{2\pi r^2} \frac{(\dv \phi)(\du \phi)}{(\dv r)(\du r)} \mr \; \D v' \\
&\leq \int^{v_2}_{v_1} 8\pi^2 r^3 (1-\mu)(-\du r)(\dv r) \ml \frac{1}{16\pi^2 r^4}\mr \; \D v' \\
&\leq \frac{(-\du r)_2}{2}\int^{v_2}_{v_1} \frac{\dv r}{r} \D v' = \frac{(-\du r)_2}{2}\log \ml 1 + \frac{r_2(u) - r_1(u)}{r_1(u)}\mr\leq \frac{(-\du r)_2}{2}\delta(u).
\end{aligned}
\end{equation}
In \eqref{Scalar13}, we use the fact that $\dv(-\du r) \geq 0$ in $R'$. Furthermore, for all $(u,v) \in R'$, we also have
\begin{equation}\label{Scalar14}
\dv r (u,v) \leq \dv r(u_0,v).
\end{equation}
This further implies that, for all $u \in [u_0,u_*]$, we get
\begin{equation}\label{Scalar15}
r_2(u) - r_1(u) \leq r_2(u_0) - r_1(u_0).
\end{equation}
For $\delta(u)$, we hence derive
\begin{equation}\label{Scalar16}
\begin{aligned}
\delta(u) = \frac{r_2}{r_1} - 1 &= \frac{r_2 - r_1}{r_2 - (r_2 - r_1)} \leq \frac{r_2(u_0)-r_1(u_0)}{r_2(u) - (r_2(u_0) - r_1(u_0))} \leq \frac{\delta(u_0)}{\frac{r_2(u)}{r_1(u_0)} - \delta(u_0)} = \frac{\delta_0}{x(u)(1+\delta_0) - \delta_0}.
\end{aligned}
\end{equation}
Returning to \eqref{Scalar11}, we obtain
\begin{equation}\label{Scalar17}
\begin{aligned}
\frac{\D \eta}{\D x} &= - \frac{\eta}{x} + \frac{\delta(u)}{x} \leq - \frac{\eta}{x} + \frac{\delta_0}{x(x(1+\delta_0)-\delta_0)},
\end{aligned}
\end{equation}
which simplifies to
\begin{equation}\label{Scalar18}
\begin{aligned}
\frac{\D}{\D x}(x \eta) &\leq  \frac{\delta_0}{(x(1+\delta_0)-\delta_0)}.
\end{aligned}
\end{equation}
Integrating from any given $x \in \ml x(u_*), 1 \mr$ to $1$ and using the fact that $\eta(x) \leq \mu_2(x)$, we arrive at
$$\eta_0 - x\eta(x) \leq \int^1_x \frac{\delta_0}{(s(1+\delta_0)-\delta_0)} \D s$$
and
\begin{equation}\label{Scalar19}
\begin{aligned}
\mu_2(x)  \geq \eta(x)  \geq \frac{1}{x}\ml \eta_0  - \frac{\delta_0}{1+\delta_0}\log\ml \frac{1}{x(1+\delta_0) - \delta_0}\mr \mr.
\end{aligned}
\end{equation}
Next, we consider the following two scenarios:
\begin{itemize}[leftmargin=*]
\item Scenario 1: $|\mathcal{B}_0| \geq v_2 - v_1$. This corresponds to the illustration on the left of Figure \ref{ScalarTrapped}. For this case, we are allowed to set $u_*$ to be arbitrarily close to $v_1$. Note that at $(u_*, v_2)$ we have
\begin{equation}\label{Scalar21a}
x(u_*) \leq \frac{r_1(u_*)}{r_2(u_0)} + \frac{\delta_0}{1 + \frac{r_2(u_0) - r_1(u_0)}{r_1(u_0)}} =  \frac{r_1(u_*)}{r_2(u_0)} + \frac{\delta_0}{1+\delta_0}.
\end{equation}
Employing \eqref{Scalar21a} and the fact that $r_1(u_*) \rightarrow 0$ as $u_* \rightarrow v_1^-$, we obtain $x(u_*) \leq \frac{p\delta_0}{1+\delta_0}$ for any $p > 1$.
\item Scenario 2: $|\mathcal{B}_0| < v_2 - v_1$. This case is illustrated in the right picture of Figure \ref{ScalarTrapped}. As seen from the picture, when discussing the point $(u_*, v_2)$, we may not be able to pick $u_*$ to be arbitrarily close to $v_1$ due to the spacelike nature of $\mathcal{B}\setminus \mathcal{B}_0$. Nonetheless, $r_2(u)\rightarrow 0$ as $(u,v_2)$ approaches any point on $\mathcal{B} \setminus \mathcal{B}_0$. Thus, if we choose $u_*$ along $v = v_2$ to be sufficiently close to $\mathcal{B} \setminus \mathcal{B}_0$, we then get $x(u_*) \leq \frac{p\delta_0}{1+\delta_0}$ for any $p > 1$.
\end{itemize}

\noindent In both cases, since $x(u)$ is a decreasing continuous function on $[u_0,u_*]$, there exists $\tilde{u}(k)$ such that $x(\tilde{u}(k)) = \frac{k\delta_0}{1+\delta_0}$ for any $k > 1$. We then evaluate \eqref{Scalar19} at this point $(\tilde{u}(k), v_2)$ and get
\begin{equation}\label{Scalar22}
\mu(\tilde{u}(k),v_2) \geq \frac{1}{k}\ml \eta_0 \ml \frac{1+\delta_0}{\delta_0}\mr - \log \ml \frac{1}{\delta_0}\mr + \log(k-1)\mr.
\end{equation}
Hence, we would have $\mu(\tilde{u}(k),v_2) \geq 1$ if 
\begin{equation}\label{Scalar23}
\eta_0\ml \frac{1+\delta_0}{\delta_0}\mr - \log\ml \frac{1}{\delta_0}\mr \geq k - \log\ml k - 1 \mr.
\end{equation}
Observe that, with $k>1$, the lower bound in \eqref{Scalar23} is attained at $k = 2$. Hence, we have $\mu(\tilde{u}(2),v_2) \geq 1$ if we request
\begin{equation}\label{Scalar24}
\eta_0 \geq \frac{2\delta_0}{1+\delta_0} + \frac{\delta_0}{1+\delta_0}\log \ml \frac{1}{\delta_0}\mr.
\end{equation}
It is exactly the requirement in \eqref{Scalar10} and this concludes the proof of Theorem \ref{ScalarTrapped}. 
\end{proof}

\section{\texorpdfstring{$C^1$ Extension Theorem with Doppler Exponent $\gamma$ for the Einstein-(Real) Scalar Field System}{C1 Extension Theorem with Doppler Exponent gamma for the Einstein-(Real) Scalar Field System.}}\label{Appendix B}

For the spherically symmetric Einstein-(real) scalar field system, we have the same set of equations from \eqref{Scalar1} to \eqref{Scalar9}. In addition, we have
\begin{equation}\label{Scalar25}
\du \mu = -\frac{4 \pi r(1-\mu)|\du \phi|^2}{-\du r} + \frac{-\du r}{r}\mu
\end{equation}
and denote
\begin{equation}\label{Scalar26}
\alpha := \frac{\dv (r\phi)}{\dv r}, \quad \quad \alpha' := \frac{\dv \alpha}{\dv r}.
\end{equation}
We define the \textit{Doppler exponent} in the (real) scalar field system to be
\begin{equation}\label{gammas}
\gamma(u,v) := \int^u_0 \frac{-\du r}{r}\frac{\mu}{1-\mu}(u',v) \D u'.
\end{equation}
In particular, its partial derivative with respect to $v$ is given by
\begin{equation}\label{B1}
\dv \gamma(u,v) = \int_{0}^{u} \frac{\dv r \du r}{r^2}\frac{1}{1-\mu}\ml 2\mu - 4 \pi \frac{r^2|\dv \phi|^2}{(\dv r)^2}\mr(u',v) \D u'
\end{equation}
if interchanging between the derivative and integral is allowed. Furthermore, we have the following \textit{Doppler identity}
\begin{equation}\label{B51}
(\dv r)(0,v) = (\dv r)(u,v)e^{\gamma(u,v)}.
\end{equation}
For the spherically symmetric Einstein-(real) scalar field system below, we proceed to prove
\begin{theorem}\label{DopplerScalar} ($C^1$ Extension Theorem with Doppler Exponent $\gamma$.) We prescribe $C^1$ initial data $\alpha$ along $C_0^+$. Let $v_* > 0$ and assume that for every $\hat{v} \in (0,v_*)$, a $C^1$ solution exists on $\mathcal{D}(0,\hat{v})$. If it obeys
\begin{equation}\label{B2}
\sup_{\mathcal{D}(0,v_*)} \gamma < + \infty,
\end{equation}
then a $C^1$ solution can be extended to $\mathcal{D}(0,\tilde{v})$ with $\tilde{v} > v_*$.
\end{theorem}
\noindent \textit{Proof.}  By the hypothesis, there is a positive constant $G$ such that
\begin{equation}\label{B3}
\sup_{\mathcal{D}(0,v_*)}\gamma(u,v;0) < G.
\end{equation}
For each $u \in [0,v_*)$, we define the quantity
\begin{equation}\label{B4}
B'(u) :=  \sup_{v \in [u,v_*]} |\alpha'|(u,v).
\end{equation}
We will proceed to show that this quantity remains bounded uniformly in $u$, which in turn implies
\begin{equation}\label{B5}
B' := \sup_{u \in [0,v_*)} \sup_{v \in [u,v_*]} |\alpha'|(u,v) = \sup_{\mathcal{D}(0,v_*)}|\alpha'| < + \infty
\end{equation}
and the solution extends as a $C^1$ solution up to $\mathcal{D}(0,v_*)$. 

\ul{Estimates for $B'(u)$.} To estimate $B'(u)$, we adopt the following strategy. We first compute $\du \alpha$ and get
\begin{equation}\label{B6.1}
\begin{aligned}
\du \alpha' = 2\frac{(-\du r)}{r}\frac{\mu}{1-\mu} \alpha' 
+ 3\frac{(-\du r)}{r^2}\frac{\mu}{1-\mu}\ml \frac{r \dv \phi}{\dv r}\mr + 4\pi \frac{(-\du r)}{r^2}\frac{1}{1-\mu}\ml \frac{r \dv \phi}{\dv r}\mr^3.\\
\end{aligned}
\end{equation}
Henceforth, for any $(u,v) \in \mathcal{D}(0,v_*)$, we obtain 
\begin{equation}\label{B7}
|\alpha'|(u,v) \leq E'(v_*) + \sum_{i=1}^{3} |I_i|(u',v) \D u' \quad \text{ with } \quad  
E'(v_*) := \sup_{\{0\} \times [0,v_*]}|\alpha'|
\end{equation}
and $|I_i|$ satisfying the following estimates:
\begin{equation}\label{I_1''}
|I_1|(u,v) \leq  2 \int^{u}_{0} (-\du r)|\alpha'|  \ml \frac{1}{r}\mr \ml \frac{\mu}{1-\mu}\mr(u',v) \mathrm{d}u', 
\end{equation}
\begin{equation}\label{I_2''}
|I_2|(u,v) \leq 3 \int^{u}_{0} (-\du r)\ml \frac{1}{r^2}\mr\ml \frac{\mu}{1-\mu}\mr\ml \frac{r|\dv \phi|}{\dv r}\mr (u',v) \mathrm{d}u',
\end{equation}
\begin{equation}\label{I_3''}
|I_3|(u,v) \leq  4\pi \int^{u}_{0} (-\du r)\ml \frac{1}{r^2}\mr\ml \frac{1}{1-\mu}\mr\ml \frac{r|\dv \phi|}{\dv r}\mr^3 (u',v) \mathrm{d}u'.
\end{equation}
We continue by obtaining key estimates for the following terms:
$$\frac{1}{1-\mu},\, |\alpha'|,\, \frac{r|\dv \phi|}{\dv r}, \, \ml \frac{r|\dv \phi|}{\dv r}\mr^3 \mbox{ for all } (u,v) \in \mathcal{D}(0,v_*).$$ 
This is done by employing the following method for each of these terms:
\begin{itemize}
\item[$\frac{1}{1-\mu}$:] By utilizing \eqref{Scalar25}, we compute
$$\du\ml \frac{1}{1-\mu(u,v)}\mr = \frac{\du \mu}{(1 - \mu)^2} = \frac{1}{1-\mu(u,v)} \ml -\frac{4\pi r|\du \phi|^2}{(-\du r)} + \frac{(-\du r)}{r}\frac{\mu}{1-\mu}\mr.$$
Together with $\du r < 0$, we obtain
\begin{equation}\label{B9}
\begin{aligned}
\frac{1}{1-\mu(u,v)} &= \frac{1}{1-\mu(0,v)}\exp\ml {\int^u_{0} \ml -\frac{4\pi r|\du \phi|^2}{(-\du r)} + \frac{(-\du r)}{r}\frac{\mu}{1-\mu}\mr (u',v) \mathrm{d}u'} \mr \\
&\leq \frac{1}{1-\mu(0,v)}\exp\ml {\int^u_{0} \ml  \frac{(-\du r)}{r}\frac{\mu}{1-\mu}\mr (u',v) \mathrm{d}u'} \mr.
\end{aligned}
\end{equation}
Employing the definition of $\gamma$ in \eqref{gammas} and $\mu_*$ in \eqref{Setup42}, we hence get
\begin{equation}\label{B11}
\frac{1}{1-\mu}(u,v) \leq \frac{1}{1-\mu_*}e^G.
\end{equation}
\item[$|\alpha'|$:] By the definition of $B'(u)$ in \eqref{B4}, for each $(u,v) \in \mathcal{D}(0,v_*)$, we have
\begin{equation}\label{B26}
|\alpha'|(u,v) \leq B'(u).
\end{equation}
\item[$\frac{r|\dv \phi|}{\dv r}$:] Applying the definition of $\alpha$ in \eqref{Scalar26} and $r\phi|_{\Gamma} = 0$, we can rewrite $\phi$ as
\begin{equation}\label{B12}
\phi(u,v) =  \frac{1}{r}\int^v_u (\alpha \dv r)(u,v') \mathrm{d}v',
\end{equation}
which in turn implies that
\begin{equation}\label{B13}
\frac{r \dv \phi}{\dv r}(u,v) = \alpha(u,v) - \frac{1}{r(u,v)}\int^v_u (\alpha \dv r)(u,v') \mathrm{d}v'.
\end{equation}
Using the fact that $r(u,u) = r|_{\Gamma} = 0$ and \eqref{B4}, we obtain
\begin{equation}\label{B14}
\begin{aligned}
\mlm\frac{r \dv \phi}{\dv r}\mrm(u,v) &\leq 
\mlm \alpha(u,v) - \frac{1}{r(u,v)} \int^v_u (\alpha \dv r)(u,v') \mathrm{d}v'\mrm \\
&\leq \frac{1}{r(u,v)}\int^v_u \mlm \alpha(u,v) - \alpha(u,v') \mrm \dv r(u,v') \mathrm{d}v' \\
&\leq \frac{1}{r(u,v)}\int^v_u \ml \int^v_{v'} \mlm \dv \alpha (u,v'')\mrm \mathrm{d}v'' \mr \dv r (u,v') \mathrm{d}v' \\
&\leq \frac{1}{r(u,v)}\int^v_u \ml \int^v_{v'} B'(u) (\dv r)(u,v'') \D v'' \mr (\dv r)(u,v')\D v'= \frac{B'(u)r(u,v)}{2}.
\end{aligned}
\end{equation}
\item[$\ml \frac{r|\dv \phi|}{\dv r}\mr^3:$] For this term, we proceed with a different strategy as follows. Observe that 
\begin{equation}\label{B15}
\begin{aligned}
\dv & \ml \frac{r \dv \phi}{\dv r}\mr^3 = \dv \ml \frac{\dv (r\phi)}{\dv r} - \phi \mr^3 = 3\ml \frac{r \dv \phi}{\dv r}\mr^2 \dv \ml \frac{\dv (r\phi)}{\dv r} - \phi\mr \\
&= 3\ml \frac{r \dv \phi}{\dv r}\mr^2 \ml \alpha ' \dv r - \dv \phi\mr = 3 \alpha' \frac{r^2(\dv \phi)^2}{\dv r} - \frac{3 \dv r}{r}\ml \frac{r \dv \phi}{\dv r}\mr^3.
\end{aligned}
\end{equation}
This then implies
\begin{equation}\label{B16}
\begin{aligned}
\dv  \ml r^3\ml \frac{r \dv \phi}{\dv r}\mr^3 \mr = 3r^2 \dv r \ml \frac{r \dv \phi}{\dv r}\mr^3 + r^3\ml3 \alpha' \frac{r^2(\dv \phi)^2}{\dv r} - \frac{3 \dv r}{r}\ml \frac{r \dv \phi}{\dv r}\mr^3 \mr = 3 \alpha' r^3 \ml \frac{r^2(\dv \phi)^2}{\dv r}\mr.
\end{aligned}
\end{equation}
Since $r^3\ml \frac{r \dv \phi}{\dv r}\mr^3 |_{\Gamma} = 0$, together with \eqref{B14}, equation \eqref{Scalar9}, $\dv r \geq 0, \dv m \geq 0$, we thus obtain
\begin{equation}\label{B17}
\begin{aligned}
r^3\ml \frac{r|\dv \phi|}{\dv r} \mr^3& (u,v) \leq \int_u^v \mlm \dv \ml r^3 \ml \frac{r\dv \phi}{\dv r} \mr^3 \mr \mrm(u,v') \D v' \leq 3B'(u)r^3(u,v)\int_u^v \frac{r^2|\dv \phi|^2}{\dv r}(u,v') \D v'\\
&\leq \frac{3}{2\pi}B'(u)r^3(u,v) \int^v_u \frac{\dv m}{1-\mu}(u,v') \D v' \leq \frac{3e^G}{\pi(1-\mu_*)} B'(u)r^4(u,v) \mu(u,v), \\
\end{aligned}
\end{equation}
which in turn implies
\begin{equation}\label{B18}
\ml \frac{r|\dv \phi|}{\dv r} \mr^3 (u,v) \leq C_3' B'(u)r(u,v) \mu(u,v)
\end{equation}
with the constant $C_3'$ given by
\begin{equation}\label{B19}
C_3' := \frac{3e^G}{\pi(1-\mu_*)}.
\end{equation}
\end{itemize}
Employing these estimates, we can now obtain an estimate for each of the terms $|I_i|$ as follows. For $I_1$, we have
\begin{equation}\label{B20}
\begin{aligned}
|I_1|(u,v) &\leq  2\int_{0}^u \frac{-\du r}{r} \frac{\mu}{1-\mu} (u',v) B'(u') \D u'.
\end{aligned}
\end{equation}
For $I_2$, by \eqref{B26}, we deduce that
\begin{equation}\label{B21}
\begin{aligned}
|I_2|(u,v) &\leq  3\int_{0}^u \frac{-\du r}{r^2} \frac{\mu}{1-\mu}\mlm \frac{r\dv \phi}{\dv r}\mrm(u',v) \D u' \leq  \frac{3}{2}\int_{0}^u \frac{-\du r}{r} \frac{\mu}{1-\mu}(u',v) B'(u') \D u'. \\
\end{aligned}
\end{equation}
Last but not least, for $I_3$, we obtain
\begin{equation}\label{B22}
\begin{aligned}
|I_3|(u,v) &\leq  4\pi\int_{0}^u \frac{-\du r}{r^2} \frac{1}{1-\mu}\mlm \frac{r\dv \phi}{\dv r}\mrm^3(u',v) \D u' \leq  4\pi C_3'\int_{0}^u \frac{-\du r}{r} \frac{\mu}{1-\mu}(u',v) B'(u') \D u'.
\end{aligned}
\end{equation}
Consequently, by substituting these estimates back to \eqref{B7}, we derive
\begin{equation}\label{B23}
|\alpha'|(u,v)\leq E'(v_*) + \ml \frac{7}{2} + 4\pi C_3'\mr \int_0^u \ml \frac{- \du r}{r} \frac{\mu}{1-\mu}\mr(u',v) B'(u') \D u'.
\end{equation}
For each $u$, we pick $\tilde{v} \in [u,v_*]$ such that the supremum for \eqref{B4} is attained at $(u,\tilde{v})$. Since \eqref{B23} is true for all $v \in [u,v_*]$, we deduce that
\begin{equation}\label{B25}
B'(u)\leq E'(v_*) + \ml \frac{7}{2} + 4\pi C_3'\mr \int_0^u \ml \frac{- \du r}{r} \frac{\mu}{1-\mu}\mr(u',\tilde{v}) B'(u') \D u'.
\end{equation}
Applying Gr\"onwall's inequality, we hence prove
\begin{equation}\label{B50}
\begin{aligned}
B'(u) \leq E'(v_*)\exp\ml \ml \frac{7}{2} + 4\pi C_3'\mr\int^u_0 \frac{-\du r}{r}\frac{\mu}{1-\mu}(u',\tilde{v})\D u'\mr \leq  E'(v_*)\exp\ml \ml \frac{7}{2} + 4\pi C_3'\mr G\mr.
\end{aligned}
\end{equation}
Observe that the upper bound obtained in \eqref{B50} is uniform in $u$ and $\tilde{v}$. Therefore, from \eqref{B5}, we conclude
\begin{equation}\label{B27}
B' \leq  E'(v_*)\exp\ml \ml \frac{7}{2} + 4\pi C_3'\mr G\mr < + \infty,
\end{equation}
which in turn implies that the solution extends as a $C^1$ solution up to $\mathcal{D}(0,v_*)$.

\ul{Extension Beyond Boundary.}

\noindent Next, we demonstrate that there exists some $\xi' > 0$, such that the solution extends as a $C^1$ solution to $\mathcal{D}(0,v_* +  \xi')$. To do so, we estimate the change in $\gamma$ across a rectangular region $\{(u,v): u \in [0,v_*], v\in [v_*,v_*+\xi]\}$ for some $\xi > 0$. An illustration of the set-up is in Figure \ref{SETFig1}.

For any given $u_1 \in [0,v_*)$, we consider the  rectangular strip given by
\begin{equation}\label{B28}
U_{\xi}(u_1) = [0,u_1) \times [v_*,v_*+\xi].
\end{equation}
Similar to \eqref{B5}, we define
\begin{equation}\label{B29}
\overline{B'_{\xi}}(u_1) := \sup_{U_{\xi}(u_1)} (|\alpha'|).
\end{equation}
In addition, we denote
\begin{equation}\label{B30}
\mathcal{D}(0,u_1,v_* + \xi) = \mathcal{D}(0,u_1,v_*) \cup U_{\xi}(u_1).
\end{equation}
From the first part of the proof, we have that the quantity $B'$ in \eqref{B5} is bounded, independent of $u_1$, and obeys the following estimates from \eqref{B27}:
\begin{equation}\label{B31}
B'(u) \leq E'(v_*)\exp\ml \ml \frac{7}{2}+4\pi C_3'(G)\mr G\mr =: C_{12}(G,v_*).
\end{equation}
We remark that $C_3'$ is defined in \eqref{B31} and the newly defined $C_{12}$ depends only on $v_*$ and $G$. Fix some $\hat{\xi} > 0$. We then have the following lemma:
\begin{lemma}\label{SETLemmaScalar}
For any $u_1 \in (0,v_*)$ and $\xi \in (0,\hat{\xi}]$, if $\gamma$ satisfies
\begin{equation}\label{B32}
\sup_{\mathcal{D}(0,u_1,v_* + \xi)} \gamma \leq 3G,
\end{equation}
then we have
\begin{equation}\label{B33}
\overline{B'_{\xi}}(u_1) \leq \max\{  C_{12}(G,v_*), C_{12}(3G,v_*+\hat{\xi}) \} =: B_{2}.
\end{equation}
\end{lemma}
The proof of this lemma follows in the same fashion as for deriving an estimate for $B'(u)$, noting that the bounds established in the region $U_\xi(u_1)$ would depend on the initial data along $\{0\} \times [v_*,v_* + \xi]$ and the uniform upper bound for $\gamma$ is now $3G$, independent of $u_1$ and $\xi$.

With the above lemma, we continue to prove that the $C^1$ solution extends beyond the domain of dependence $\mathcal{D}(0,u_1,v_*)$. First, we define
\begin{equation}\label{B34}
u_2 = \sup\left\{ u_1 \in [0,v_*): \sup_{\mathcal{D}(0,u_1,v_*+\xi)} \gamma \leq 3G\right\}.
\end{equation}
By the definition of $u_2$, we then have either $u_2 = v_*$ or
\begin{equation}\label{B35}
\sup_{\mathcal{D}(0,u_1,v_* + \xi)} \gamma = 3G
\end{equation}
is attained at some $u_2 \in [u_0,v_*)$. We now show that for $\xi$ sufficiently small, the latter will not happen. Suppose for a contradiction that the latter happens. Subsequently, set $u_1 = u_2$ as in the hypothesis of Lemma \ref{SETLemmaScalar}. Analogous to the proof of Lemma \ref{SETLemmaScalar}, in the region $U_\xi(u_2)$, we first obtain
\begin{equation}\label{B36}
\frac{|\dv \phi|}{\dv r}(u,v) \leq \frac{B_2}{2}.
\end{equation}
We can derive an estimate for $m(u,v)$ using \eqref{B36} as follows. By \eqref{Scalar9} and $m|_{\Gamma} = 0$, these imply 
\begin{equation}\label{B37}
\begin{aligned}
m(u,v) 
\leq 2\pi \int^v_u \frac{r^2|\dv \phi|^2}{(\dv r)^2}(\dv r)(u,v') \D v' \leq \frac{\pi B_2^2 }{2} \int^v_u r^{2}(\dv r)(u,v') \D v' = \frac{\pi B_2^2}{6}r^3(u,v)
\end{aligned}
\end{equation}
and hence
\begin{equation}\label{B38}
\mu(u,v) \leq \frac{\pi B_2^2}{3} r^2(u,v).
\end{equation}
In addition, from \eqref{B11} but with $G$ replaced by $3G$, we have
\begin{equation}\label{B39}
\frac{1}{1-\mu}(u,v) \leq \frac{1}{1-\mu_*}e^{3G}.
\end{equation}
With these estimates, we recall the derivative of $\gamma$ with respect to $v$ from \eqref{B1} 
\begin{equation}\label{B40}
\dv \gamma(u,v) = \int_{0}^{u} \frac{\dv r \du r}{r^2}\frac{1}{1-\mu}\ml 2\mu - 4 \pi \frac{r^2|\dv \phi|^2}{(\dv r)^2}\mr(u',v) \D u'.
\end{equation}
By \eqref{B38}, \eqref{B39} and the Doppler identity \eqref{B51},
we deduce that
\begin{equation}\label{B41}
\begin{aligned}
&\int_{0}^{u} \frac{\dv r \du r}{r^2}\frac{1}{1-\mu}\mlm 2\mu - 4 \pi \frac{r^2|\dv \phi|^2}{(\dv r)^2}\mrm(u',v) \D u' \\
\leq& \; \frac{5\pi}{3(1-\mu_*)}e^{3G} B_2^2 \int^u_0 (-\du r)(\dv r)(u',v) \D u' 
\leq \frac{5\pi}{6(1-\mu_*)}e^{3G} B_2^2 \int^u_0 (-\du r)(u',v) \D u' \\
&\quad\leq \; \frac{5\pi}{6(1-\mu_*)}e^{3G} B_2^2 r(0,v) = \frac{5\pi}{12(1-\mu_*)}e^{3G} B_2^2 v \leq \frac{5\pi}{12(1-\mu_*)}e^{3G} B_2^2 (v_* + \hat{\xi}). \\
\end{aligned}
\end{equation}
The estimate in \eqref{B41} implies that the expression in \eqref{B40} is valid by dominated convergence theorem. For each $(u,v) \in U_\xi(u_2)$, we then establish
\begin{equation}\label{B42}
\begin{aligned}
\gamma(u,v) \leq \gamma(u,v_*) + \int_v^{v_*}|\dv \gamma|(u,v') \D v' \leq G + \frac{5\pi}{12(1-\mu_*)}e^{3G} B_2^2 (v_* + \hat{\xi})\xi.
\end{aligned}
\end{equation}
If we pick 
\begin{equation}\label{B43}
\xi := \min \left\{  \hat{\xi}, \frac{12G(1-\mu_*)}{5\pi e^{3G} B_2^2 (v_* + \hat{\xi})}\right\},
\end{equation}
we arrive at
$\sup_{\mathcal{D}(0,u_1,v_* + \xi)} \gamma \leq 2G$, contradicting \eqref{B35}.

Thus, for each $u_1 \in (u_0,v_*)$, the hypothesis of Lemma \ref{SETLemmaScalar} is verified. This then implies that \eqref{B33} holds and hence  $\sup_{\mathcal{D}(u_0,u_1,v_*+\Delta v)}|\alpha'|$ remains bounded as $u_1 \rightarrow v_*$. Therefore, we conclude that the solution extends as a $C^1$ solution up to the closure of $\mathcal{D}(u_0,v_*,v_*+\xi)$. In particular, with $C^1$ initial data prescribed along $\{v_*\} \times [v_*,v_*+\xi]$, we show that the solution further extends as a $C^1$ solution to $\mathcal{D}(v_*,v_*+\xi')$ for some $\xi' \in (0,\xi]$. This completes the proof of Theorem \ref{DopplerScalar}.
\qed

\end{document}